\documentclass[12pt]{article}%
\usepackage{amsmath}
\usepackage{rotating}
\usepackage{version}
\usepackage{amsfonts}
\usepackage{amssymb}
\usepackage{graphicx}
\usepackage{float}
\usepackage{url}
\usepackage[round,authoryear]{natbib}
\usepackage{amsmath}
\usepackage{rotating}
\usepackage{version}
\usepackage{amsfonts}
\usepackage{amssymb}
\usepackage[hypertexnames=false]{hyperref}
\usepackage{ragged2e}
\usepackage{url}
\usepackage{multirow}
\usepackage{threeparttable}
\usepackage{color}
\usepackage{xr}
\usepackage[normalem]{ulem}%
\providecommand{\U}[1]{\protect\rule{.1in}{.1in}}
\newtheorem{theorem}{Theorem}

\newtheorem{assumption}{Assumption}

\newtheorem{lemma}{Lemma}

\newtheorem{proposition}{Proposition}
\newtheorem{remark}{Remark}

\newenvironment{proof}[1][Proof]{\noindent\textbf{#1.} }{\ \rule{0.5em}{0.5em}}
\makeatletter
\renewcommand*{\@fnsymbol}[1]{\ifcase#1\or$\dagger$\else\@arabic{\numexpr#1-1\relax}\fi}
\makeatother
\begin{document}

\title{How to Detect Network Dependence in Latent Factor Models? A Bias-Corrected CD
Test\thanks{We thank Peter C. B. Phillips (Editor), Guido Kuersteiner
(Co-Editor), and three anonymous referees for their thoughtful comments and
constructive suggestions. We also would like to thank Natalia Bailey, Alex
Chudik, Zhan Gao, Arturas Juodis, Ron Smith, Simon Reese and Qiankun Zhou for
helpful comments. The research for this paper was completed when Yimeng Xie
was a Ph.D. student at the University of Southern California. This paper has
been previously circulated under the title \textquotedblleft A Bias-Corrected
CD Test for Error Cross-Sectional Dependence in Panel Data Models with Latent
Factors". https://arxiv.org/abs/2109.00408}}
\author{M. Hashem Pesaran\\University of Cambridge, U.K. and University of Southern California, U.S.\\and\\Yimeng Xie\\Xiamen University, China}
\date{\today}
\maketitle

\begin{abstract}
In a recent paper \cite{juodis2022incidental} (JR) show that the application
of the CD test proposed by \cite{pesaran2004general} to residuals from panels
with latent factors results in over-rejection. They propose a randomized test
statistic to correct for over-rejection, and add a screening component to
achieve power. This paper considers the same problem but from a different
perspective, and shows that the standard CD test remains valid if the latent
factors are weak. A bias-corrected version, CD$^{\text{*}}$, is proposed which is shown to be asymptotically
standard normal under the null of error cross-sectional independence which has
power against network type alternatives. This result is shown to hold for pure
latent factor models as well as for panel regression models with latent
factors. The case where the errors are serially correlated is also considered.
Small sample properties of the CD$^{\text{*}}$ test are investigated by Monte Carlo experiments and are shown to have satisfactory small sample properties. In an empirical application, using the CD$^{\text{*}}$ test, it
is shown that there remains spatial error dependence in a panel data model for
real house price changes across 377 Metropolitan Statistical Areas in the
U.S., even after the effects of latent factors are filtered out.

\end{abstract}



\pagenumbering{gobble}%

\newpage%
\pagenumbering{arabic}%
%

\section{Introduction}

It is now quite standard to use latent multi-factor models to characterize and
explain cross-sectional dependence in panels when the cross section dimension
$(n)$ and the time series dimension $(T)$ are both large. However, due to
uncertainty regarding the nature of error cross-sectional dependence, it is
arguable whether error cross-sectional dependence is fully accounted for by
latent factors. Some of the factors could be semi-strong, and the errors might
have spatial or network features that are not necessarily captured by common
factors alone. \cite{chudik2011weak} provide an early discussion of the
different sources of cross-sectional dependence. Diagnostic tests of error
cross-sectional independence in panels are required to safeguard against
estimation bias and unreliable inference. See, for example,
\cite{bai2003inferential,bai2009panel},
\cite{phillips2003dynamic,phillips2007bias}, \cite{bai2006confidence},
\cite{pesaran2006estimation}, and \cite{pesaran2011large}. Such tests are also
helpful to researchers interested in network or spatial dependence once the
influence of common factors are filtered out. See, for example,
\cite{bailey2016two}, \cite{shi2017spatial}, \cite{aquaro2021estimation}, and
\cite{bai2021dynamic} amongst others.

One standard test for error cross-sectional independence is the CD test
proposed by \cite{pesaran2004general,pesaran2021general}, which has been
further developed. For example, \cite{hsiao2012diagnostic} apply the CD
test to panel data models with limited dependent variables, while
\cite{pesaran2015testing} uses it to test weak error cross-sectional
dependence in large panels. In a recent paper \cite{juodis2022incidental} (JR)
show that the application of the CD test to the residuals from panels with
latent factors is invalid and can result in over-rejection of the null of
error cross-sectional independence.\footnote{In the empirical finance
literature \cite{gagliardini2019diagnostic} propose a diagnostic criterion to
check if the errors from a (strong) factor model are weakly correlated or
contain missing strong factor(s). These authors do not propose a test of
cross-sectional error dependence but focus on the detection of potentially
omitted strong factors in asset pricing models.} They propose a randomized CD
test statistic as a solution. Their proposed test is constructed in two steps.
First, they multiply the residuals from panel regressions with independent
randomized weights to obtain their CD$_{W}$ statistic, which will have a zero
mean by construction. In this way they avoid the over-rejection problem of the
CD test, but by the very nature of the randomization process they recognize
that the CD$_{W}$ test will lack power. To overcome the problem of lack of
power, JR modify the CD$_{W}$ test statistic by adding to it a screening
component proposed by \cite{fan2015power} which is expected to tend to zero
with probability approaching one under the null hypothesis, but to diverge at
a reasonably fast rate under the alternative. This further modification of the
CD$_{W}$ test is denoted by CD$_{W+}$. Accordingly, it is presumed that the
CD$_{W+}$ test can overcome both over-rejection and the low power problems.
However, JR do not provide a formal proof establishing conditions under which
the screening component tends to zero under the null and diverges sufficiently
fast under alternatives, including spatial or network dependence type
alternatives. Also, our Monte Carlo simulations show that the CD$_{W+}$ test
tends to over-reject when the errors are non-Gaussian and $n>>T$, and lacks power under spatial and network alternatives, which is likely to be
particularly important in empirical applications.

In this paper we show that the standard CD test is in fact valid for testing
error cross-sectional independence in panel data models with weak latent
factors. However, when the latent factors are semi-strong or strong the use of
the CD test will result in over-rejection and will no longer be valid,
extending JR's results to panels with semi-strong latent factors. In short,
whilst the CD$_{W+}$ test is a useful and welcome addition to testing for
error cross-sectional independence, it would be interesting to develop a
modified version of the test that simultaneously deals with the over-rejection
problem and does not compromise power for a general class of alternatives. To
that end, firstly we study testing for error cross-sectional independence in a
pure latent factor model, and derive an explicit expression for the bias of
the CD test statistic in terms of factor loadings and error variances. We then
propose a bias-corrected version of the CD test statistic, denoted by
$CD^{\ast}$, which is shown to have $\mathcal{N}(0,1)$ asymptotic distribution
under the null hypothesis irrespective of whether the latent factors are weak
or strong. When the latent factors are weak the correction tends to zero, $CD$
and $CD^{\ast}$ will be asymptotically equivalent. However, $CD-CD^{\ast}$
diverges if at least one of the underlying latent factors is (semi) strong. We
show that under the null of cross-sectional independence, $CD^{\ast}$
converges to a standard normal distribution when $n$ and $T$ tend to infinity
so long as $n/T\rightarrow\kappa$, where $0<\kappa<\infty$, and a test based
on $CD^{\ast}$ will have the correct size asymptotically. In addition, it is
shown that the CD$^{\text{*}}$ test has power against spatial and network type
alternatives. In particular, we are the first to give a formal derivation of
the power function for CD tests against general spatial and network
alternatives, which can be applied equally to panel data models without latent
factors and therefore supplement earlier research on CD tests\footnote{For
instance, \cite{pesaran2004general}, which is the unpublished version of
\cite{pesaran2021general}, also discusses the power of the CD test against
spatial dependence in Section 8.2 of his paper, for a specific connection matrix with $n$
fixed as $T\rightarrow\infty$.}.

We then consider the application of the $CD^{\ast}$ to test error
cross-sectional independence in the case of panel regression models with
latent factors, discussed in \cite{pesaran2006estimation}. It is shown that
the asymptotic properties of the $CD^{\ast}$ in the case of pure latent factor
models also carry over to panel regression models with latent factors. We also
investigate the application of the $CD^{\ast}$ to panel data models with
serially correlated errors, and consider the method proposed by
\cite{baltagi2016testing} as well as using an autoregressive distributed lag
(ARDL) representation which transforms the model with serially correlated
errors to one without error serial correlation.

The finite sample performance of the CD$^{\text{*}}$ test is investigated by
Monte Carlo simulations in the case of pure latent factor models, panel
regression models with latent factors with and without error serial
correlation. It is found that the CD$^{\text{*}}$ test avoids the
over-rejection problem under the null and has power against spatial and
network alternatives, and has desirable small sample properties regardless of
whether the errors are Gaussian or not, under different combinations of $n$
and $T$. We also find that both adjustments for dealing with error serial
correlation considered in the paper give desirable small sample properties.
Finally, as compared to JR's CD$_{W^{+}}$ test, the proposed bias-corrected CD
test is better in controlling the size of the test and has much better power
properties against spatial or network alternatives.

The use of the CD$^{\text{*}}$ test is illustrated by an empirical application
in modeling real house price changes in the U.S. Because it is evident that
real house price changes are driven by macroeconomic trends which can be
modeled by latent factors, it is necessary to filter out these factors before
testing for spillover effect. By applying the CD$^{\text{*}}$ test to real
house price changes in the U.S. we are able to show significant existence of
weak cross-sectional dependence in addition to latent factors.

The rest of the paper is organized as follows. Section \ref{pure_factor} sets
out the latent factor model and its assumptions. Section \ref{sec:cd}
introduces the estimation of latent factors and the CD test. The
bias-corrected test, CD$^{\ast}$, is introduced in Section \ref{CD*} and its
asymptotic distributions are derived under the null and the alternative
hypotheses. The extension to more general panel regression models with
observed covariates as well as latent factors are discussed in Section
\ref{panel_model}. Adjustments to the CD$^{\text{*}}$ test for panels with
serially correlated errors are discussed in Section \ref{section_serial}.
Using Monte Carlo techniques, the small sample properties of CD, CD$^{\ast}$,
and CD$_{W+}$ tests are discussed in Section \ref{mc}. An empirical
illustration is provided in Section \ref{app}. Proofs of the propositions and
theorems are provided in an appendix. The auxiliary lemmas and the associated proofs are given in a supplement.

\textbf{Notations: }For the $n\times n$ matrix $\mathbf{A}=\left(
a_{ij}\right)  $, we denote its largest eigenvalue by $\mu_{max}\left(
\mathbf{A}\right)  $, its trace by $\mathrm{tr}\left(  \mathbf{A}\right)
=\sum_{i=1}^{n}a_{ii}$, its spectral norm by $\left\Vert \mathbf{A}\right\Vert
=\mu_{\max}^{1/2}\left(  \mathbf{A}^{\prime}\mathbf{A}\right)  $, its maximum
absolute column sum norm by $\left\Vert \mathbf{A}\right\Vert _{1}=\max_{1\leq
j\leq n}\left(  \sum_{i=1}^{n}\left\vert a_{ij}\right\vert \right)  $, and its
maximum absolute row sum norm by $\left\Vert \mathbf{A}\right\Vert _{\infty
}=\max_{1\leq i\leq n}\left(  \sum_{j=1}^{n}\left\vert a_{ij}\right\vert
\right)  $. We write $\mathbf{A}>\mathbf{0}$ when $\mathbf{A}$ is positive
definite. For matrices $\mathbf{B}=\left(  b_{ij}\right)  $ and $\mathbf{C}%
=\left(  c_{ij}\right)  $, $\mathbf{B\odot C=C\odot B}$ denote Hadamard
product with elements $b_{ij}c_{ij}$. $\rightarrow_{p}$ denotes convergence in
probability, $\rightarrow_{d}$ convergence in distribution, and
$\overset{a}{\thicksim}$ asymptotic equivalence in distribution. $O_{p}\left(
\cdot\right)  $ and $o_{p}\left(  \cdot\right)  $ denote the stochastic order
relations. In particular, $o_{p}(1)$ indicates terms that tend to zero in
probability as $(n,T)\rightarrow\infty$, such that $n/T\rightarrow\kappa$,
where $0<\kappa<\infty$. $C$ and $c$ will be used to denote finite large and
non-zero small positive numbers, respectively, that are bounded in $n$ and
$T$. They can take different values at different instances. If $\left\{
f_{n}\right\}  _{n=1}^{\infty}$ is any real sequence and $\left\{
g_{n}\right\}  _{n=1}^{\infty}$ is a sequence of positive real numbers, then
$f_{n}=O(g_{n})$, if there exists $C$ such that $\left\vert f_{n}\right\vert
/g_{n}\leq C$ for all $n$. $f_{n}=o(g_{n})$ if $f_{n}/g_{n}\rightarrow0$ as
$n\rightarrow\infty$. If $\left\{  f_{n}\right\}  _{n=1}^{\infty}$ and
$\left\{  g_{n}\right\}  _{n=1}^{\infty}$ are both positive sequences of real
numbers, then $f_{n}=\ominus\left(  g_{n}\right)  $ if there exists $n_{0}%
\geq1$ and positive finite constants $C_{0}$ and $C_{1}$, such that
$\inf_{n\geq n_{0}}\left(  f_{n}/g_{n}\right)  \geq C_{0},$ and $\sup_{n\geq
n_{0}}\left(  f_{n}/g_{n}\right)  \leq C_{1}$.

\section{The latent factor model\label{pure_factor}}

To simplify the exposition and to highlight the main issue of concern, namely
the presence of \textit{latent} (unobserved) factors in the panel regression
model, initially we focus on the approximate factor model, due to
\cite{chamberlain1983arbitrage}, and assume that for each unit $i=1,2,\ldots
,n$,
\begin{equation}
y_{it}=\boldsymbol{\gamma}_{i}^{^{\prime}}\mathbf{f}_{t}+u_{it},\text{ for
}t=1,2,...,T, \label{mod2}%
\end{equation}
where%
\begin{equation}
u_{it}/\sigma_{i}=\varepsilon_{it}\left(  \lambda_{T}\right)  =\varepsilon
_{it}+\lambda_{T}\sum_{j=1}^{n}w_{ij}\varepsilon_{jt},\text{ } \label{uit:p}%
\end{equation}
$\sup_{i}\sigma_{i}^{2}<C<\infty$ and $\inf_{i}\sigma_{i}^{2}>c>0$, $\left\{
w_{ij}:j=1,2,\ldots,n\right\}  $ represent the strengths of connections of
unit $i$ with the rest of units, $\mathbf{f}_{t}=\left(  f_{1t},f_{2t}%
,...,f_{m_{0}t}\right)  ^{\prime}$ is an $m_{0}\times1$ vector of latent
factors with $m_{0}$ fixed, and $\boldsymbol{\gamma}_{i}=\left(  \gamma
_{i1},\gamma_{i2},...,\gamma_{im_{0}}\right)  ^{\prime}$ is the vector of
associated factor loadings.

We make the following assumptions that are mostly standard in the analysis of
latent factor models.

\begin{assumption}
\label{ass:factors} (a) $\mathbf{f}_{t}$ is a covariance-stationary process
with zero means and the covariance matrix, $E\left(  \mathbf{f}_{t}%
\mathbf{f}_{t}^{\prime}\right)  =\mathbf{\Sigma}_{ff}>\mathbf{0}$. (b)
$T^{-1}\sum_{t=1}^{T}\left[  \Vert\mathbf{f}_{t}\Vert^{j}-E\left(
\Vert\mathbf{f}_{t}\Vert^{j}\right)  \right]  {\rightarrow}_{p}0$, for
$j=3,4$, as $T\rightarrow\infty$.\textbf{ }(c) There exists $T_{0}$ such that
for all $T>T_{0}$, $T^{-1}\sum_{t=1}^{T}\mathbf{f}_{t}\mathbf{f}_{t}^{\prime
}=T^{-1}\mathbf{F}^{\prime}\mathbf{F=\Sigma}_{T,ff}>\mathbf{0}$, and
$\mathbf{\Sigma}_{T,ff}\rightarrow_{p}E\left(  T^{-1}\mathbf{F}^{\prime
}\mathbf{F}\right)  =\mathbf{\Sigma}_{ff}>\mathbf{0}$, where $\mathbf{F=(f}%
_{1},\mathbf{f}_{2},...,\mathbf{f}_{T})^{\prime}$. (d) There exist constants
$r_{1}$, $C_{0}$ and $C_{1}>0$ such that
\begin{equation}
\sup_{j,t}\Pr\left(  \left\vert f_{jt}\right\vert >a\right)  \leq C_{0}%
\exp\left(  -C_{1}a^{r_{1}}\right),  \label{sup:f}%
\end{equation}
all $a>0$.
\end{assumption}

\begin{assumption}
\label{ass:errors} (a) $\varepsilon_{it}\sim IID\left(  0,1\right)  $ for all
$i$ and $t$, and there exist constants $r_{2}$, $C_{2}$ and $C_{3}>0$ such
that%
\begin{equation}
\sup_{i,t}\Pr\left(  \left\vert \varepsilon_{it}\right\vert >a\right)  \leq
C_{2}\exp\left(  -C_{3}a^{r_{2}}\right),  \label{sup:u}%
\end{equation}
for all $a>0$. (b) $\mu_{\max}\left(  \mathbf{V}_{\varepsilon T}\right)
=O_{p}(n/T)$, where $\mathbf{V}_{\varepsilon T}=T^{-1}\sum_{t=1}%
^{T}\boldsymbol{\varepsilon}_{\circ t}\boldsymbol{\varepsilon}_{\circ
t}^{\prime}$, and $\boldsymbol{\varepsilon}_{\circ t}=(\varepsilon
_{1t},\varepsilon_{2t},...,\varepsilon_{nt})^{\prime}$. (c) $\varepsilon_{it}$
is distributed independently of $\mathbf{f}_{t^{\prime}}$, for all $i,t$ and
$t^{\prime}$, and there exists $v_{0}>0$ such that for all $v=T-m_{0}>v_{0}$
\begin{equation}
\inf_{i}\left(  v^{-1}\boldsymbol{\varepsilon}_{i\circ}^{\prime}\mathbf{M}%
_{F}\boldsymbol{\varepsilon}_{i\circ}\right)  >c>0, \label{epsic}%
\end{equation}
where $\boldsymbol{\varepsilon}_{i\circ}=(\varepsilon_{i1},\varepsilon
_{i2},...,\varepsilon_{iT})^{\prime}$, and $\mathbf{M}_{F}=\mathbf{I}%
_{T}-\mathbf{F(F}^{\prime}\mathbf{F)}^{-1}\mathbf{F}^{\prime}$.
\end{assumption}

\begin{assumption}
\label{ass:loadings} The $m_{0}\times1$ vector of factor loadings
$\boldsymbol{\gamma}_{i}$ is bounded such that $\sup_{i}\left\Vert
\boldsymbol{\gamma}_{i}\right\Vert <C$, $n^{-1}\sum_{i=1}^{n}%
\boldsymbol{\gamma}_{i}\boldsymbol{\gamma}_{i}^{\prime}=\mathbf{\Sigma
}_{n,\gamma\gamma}\rightarrow\mathbf{\Sigma}_{\gamma\gamma}>\mathbf{0}$, and
\begin{equation}
1-\theta_{n}>0, \label{thetaCon}%
\end{equation}
for all $n>n_{0}$ and as $n\rightarrow\infty,$ where $\theta_{n}=1-n^{-1}%
\sum_{i=1}^{n}a_{i,n}^{2}$, with $
\text{ }a_{i,n}=1-\sigma_{i}\boldsymbol{\varphi}_{n}^{\prime}%
\boldsymbol{\gamma}_{i},
$ and $\boldsymbol{\varphi}_{n}=n^{-1}\sum_{i=1}^{n}\boldsymbol{\gamma}%
_{i}/\sigma_{i}.$
\end{assumption}

\begin{remark}
\label{rm:literature}The above assumptions relate closely to those made in the
literature on CD tests and high dimensional factor models. See, for example,
the assumptions in
\cite{pesaran2004general,pesaran2015testing,pesaran2021general}, and
assumptions in \cite{bai2003inferential}.\textbf{\ }Part (a) of Assumption
\ref{ass:factors} will be relaxed when we consider panel regression models
with observed regressors. The sub-exponential type conditions (\ref{sup:f})
and (\ref{sup:u}) are needed for bounding the probabilities across all $i$,
and are also adopted by \cite{fan2011high}, \cite{fan2013large} and
\cite{chudik2018one}.
\end{remark}

\begin{remark}
\label{rm:epsiC}Since $\mathbf{M}_{F}$ is an idempotent matrix with rank $v$,
then there exists the orthogonal transformation $\boldsymbol{\eta}_{i}%
=(\eta_{i1},\eta_{i2},...,\eta_{iv})^{\prime}=\mathbf{H}%
\boldsymbol{\varepsilon}_{i\circ}$ where $\mathbf{H}$ is a $v\times T$ matrix
such that $v^{-1}\boldsymbol{\varepsilon}_{i\circ}^{\prime}\mathbf{M}%
_{F}\boldsymbol{\varepsilon}_{i\circ}=v^{-1}\boldsymbol{\eta}_{i}^{\prime
}\boldsymbol{\eta}_{i}>0,$ and $E\left(  \eta_{it}^{2}\right)  =1$. See, for
example, \citet[p. 412]{durbin1950testing}. Therefore, there exists a finite
$v_{0}$ such that for all $v>v_{0}$ condition (\ref{epsic}) is met, and as
result we also have
\begin{equation}
E\left(  \frac{\boldsymbol{\varepsilon}_{i\circ}^{\prime}\mathbf{M}%
_{F}\boldsymbol{\varepsilon}_{i\circ}}{v}\right)  ^{-s}<\frac{1}{c^{s}%
}<C<\infty, \label{epsicA}%
\end{equation}
for all $i$ and any fixed $s>0$.
\end{remark}

The focus of this paper is on testing the null hypothesis of error
cross-sectional independence:
\begin{equation}
H_{0}:\lambda_{T}=0, \label{H0}%
\end{equation}
where $\lambda_{T}$ is defined by equation (\ref{uit:p}). For the analysis of
power we consider local alternatives:%
\begin{equation}
H_{1T}:\lambda_{T}=c_{\lambda}T^{-1/2}, \label{H1T}%
\end{equation}
with $c_{\lambda}\neq0$. We also introduce the following assumption on
$\mathbf{W}=\left(  w_{ij}\right)  $.

\begin{assumption}
\label{ass:ws}The connection matrix $\mathbf{W}$ has bounded maximum absolute
column and row sum norms:
\begin{equation}
\left\Vert \mathbf{W}\right\Vert _{1}=\sup_{j}\sum_{i=1}^{n}\left\vert
w_{ij}\right\vert <C,\text{ and }\left\Vert \mathbf{W}\right\Vert _{\infty
}=\sup_{i}\sum_{j=1}^{n}\left\vert w_{ij}\right\vert <C, \label{sum:wi}%
\end{equation}
and $w_{ii}=0$ for all $i$.
\end{assumption}

\begin{remark}
The connection matrix does not need to be symmetric. To see this, we can
consider a more generalized setup of idiosyncratic errors,
\begin{equation}
\frac{u_{it}}{\sigma_{i}}=\varepsilon_{it}+\lambda_{T}{\sum\limits_{j=1}^{n}%
}\frac{\rho_{i}\sigma_{j}\mathring{w}_{ij}}{\sigma_{i}}\varepsilon_{jt},
\label{uit:w}%
\end{equation}
where $\left\vert \lambda_{T}\right\vert <C$ and $\left\vert \rho
_{i}\right\vert <C$. It is clear that (\ref{uit:w}) reduces to $u_{it}$ in
(\ref{uit:p}) by letting $w_{ij}=\rho_{i}\sigma_{i}^{-1}\sigma_{j}\mathring
{w}_{ij}$. In this way, $\mathbf{W}$ need not be symmetric even if the
connection matrix $\left(  \mathring{w}_{ij}\right)  $ is symmetric. Under
local alternatives and Assumption \ref{ass:ws}, the specification of
(\ref{uit:p}) allows for a wide range of spatial and network dependence
characterized by the connection matrix. It is in accord with an early
discussion in \cite{chudik2011weak} that spatial dependence can be captured by
a weak factor model, so long as the number of weak factors tends to infinity
with $n$, which is ruled out in standard factor models where the number of
latent factors is assumed to be fixed.
\end{remark}

\begin{remark}
It is also easily seen that $\varepsilon_{it}\left(  \lambda_{T}\right)  $
defined by (\ref{uit:p}) is sub-exponential for any $\left\vert \lambda
_{T}\right\vert <C$. This follows since by Assumption \ref{ass:errors}
$\left\{  \varepsilon_{it}\right\}  $ are independently and identically
distributed sub-exponential processes, and by Assumption \ref{ass:ws}
$\sup_{i}\sum_{j=1}^{n}\left\vert w_{ij}\right\vert <C.$ For a proof see part
(b) of Theorem 5.5 in \cite{goldie1998subexponential} or Theorem 2.8.2 in
\cite{vershynin2018high}.
\end{remark}

\section{Estimation of latent factors and the CD test\label{sec:cd}}

Following the literature we use principal component (PC) analysis to estimate
the latent factors and their loadings. Let $\mathbf{Y=(y}_{1},\mathbf{y}%
_{2},\ldots,\mathbf{y}_{n}\mathbf{)}$ be the $T\times n$ matrix of
observations on $y_{it}$, where $\mathbf{y}_{i}=\left(  y_{i1},y_{i2}%
,\ldots,y_{iT}\right)  ^{^{\prime}}$ and denote the first $m_{0}$ largest
eigenvalues of $\mathbf{Y}^{^{\prime}}\mathbf{Y}$ by $(\hat{\rho}_{1}%
,\hat{\rho}_{2},...,\hat{\rho}_{m_{0}})$, and its associated $n\times m_{0}$
matrix of orthonormal eigenvectors by $\mathbf{\hat{Q}}$. The PC estimators of
factors $\mathbf{F}=\left(  \mathbf{f}_{1},\mathbf{f}_{2},\ldots
,\mathbf{f}_{T}\right)  ^{^{\prime}}$ and their loadings $\boldsymbol{\Gamma
}=\left(  \boldsymbol{\gamma}_{1},\boldsymbol{\gamma}_{2},\ldots
,\boldsymbol{\gamma}_{n}\right)  ^{^{\prime}}$ are then given by%
\begin{equation}
\mathbf{\mathbf{\hat{F}}}=\mathbf{\left(  \mathbf{\hat{f}}_{1},\mathbf{\hat
{f}}_{2},\ldots,\mathbf{\hat{f}}_{T}\right)  ^{\prime}}=\frac{1}{\sqrt{n}%
}\mathbf{\mathbf{Y\hat{Q}},}\text{ and }\boldsymbol{\hat{\Gamma}}=\left(
\boldsymbol{\hat{\gamma}}_{1},\boldsymbol{\hat{\gamma}}_{2},\ldots
,\boldsymbol{\hat{\gamma}}_{n}\right)  ^{^{\prime}}=\sqrt{n}\mathbf{\hat{Q}.}
\label{hatg}%
\end{equation}
By construction $n^{-1}\boldsymbol{\hat{\Gamma}}^{\prime}\boldsymbol{\hat
{\Gamma}}=\mathbf{I}_{m_{0}},$ and $T^{-1}\mathbf{\hat{F}}^{^{\prime}%
}\mathbf{\hat{F}=D}_{nT}$, where $\mathbf{D}_{nT}=(nT)^{-1}\mathrm{diag}%
(\hat{\rho}_{1},\hat{\rho}_{2},...,\hat{\rho}_{m_{0}})$. Under Assumptions
\ref{ass:factors}-\ref{ass:ws} the asymptotic results derived by
\cite{bai2003inferential} for PCs continue to apply here, and $u_{it}$ can be
consistently estimated by
\begin{equation}
\hat{u}_{it}=y_{it}-\boldsymbol{\hat{\gamma}}_{i}^{^{\prime}}\mathbf{\hat{f}%
}_{t}. \label{eit}%
\end{equation}

The CD test is based on the standardized residuals,
\begin{equation}
\tilde{\varepsilon}_{it,T}=\frac{\hat{u}_{it}}{\hat{\sigma}_{i,T}},
\label{egziit}%
\end{equation}
where $\hat{\sigma}_{i,T}=\left(  T^{-1}\sum_{t=1}^{T}\hat{u}_{it}^{2}\right)
^{1/2}=\left(  T^{-1}\mathbf{y}_{i}^{\prime}\mathbf{M}_{\hat{F}}\mathbf{y}%
_{i}\right)  ^{1/2}$, and $\mathbf{M}_{\hat{F}}=\mathbf{I}_{T}-\mathbf{\hat
{F}(\hat{F}}^{\prime}\mathbf{\hat{F})}^{-1}\mathbf{\hat{F}}^{\prime}$. Only
units with non-zero $\hat{\sigma}_{i,T}^{2}$ are included in the construction
of the CD test, namely
\begin{equation}
\inf_{i}\hat{\sigma}_{i,T}^{2}>c>0. \label{infhsigma2}%
\end{equation}
The standard CD test statistic based on the residuals, (\ref{eit}), is given
by%
\begin{equation}
CD=\sqrt{\frac{2T}{n(n-1)}}\left(  \sum_{i=1}^{n-1}\sum_{j=i+1}^{n}\hat{\rho
}_{ij,T}\right)  , \label{cd}%
\end{equation}
where $\hat{\rho}_{ij,T}=T^{-1}\sum_{t=1}^{T}\tilde{\varepsilon}_{it,T}%
\tilde{\varepsilon}_{jt,T}$.

\cite{juodis2022incidental} apply the CD test to a panel regression model with
latent factors, assuming that all the factors are strong. They show in that
case $CD=O_{p}\left(  \sqrt{T}\right)  $, and its use will lead to gross
over-rejection of the null of error cross-sectional independence. To deal with
the over-rejection problem, these authors propose a randomized CD test,
CD$_{W+}$. However, as shown in Section \ref{section.JR} of the supplement,
the CD$_{W+}$ test is likely to over-reject and tends to lack power against
spatial and network alternatives. See also Section
\ref{sec:simulation results} for Monte Carlo evidence on the small sample
performance of the CD$_{W+}$ test.

\section{The bias-corrected CD test\label{CD*}}

The main reason for the failure of the standard CD test in the case of latent
factor models lies in the fact that both the factors and their loadings are
unobserved and need to be estimated, and the differences between
$\boldsymbol{\hat{\gamma}}_{i}^{^{\prime}}\mathbf{\hat{f}}_{t}$ and
$\boldsymbol{\gamma}_{i}^{^{\prime}}\mathbf{f}_{t}$ do not tend to zero at a
sufficiently fast rate for the CD test to be valid. Since the errors from
estimation of $\boldsymbol{\gamma}_{i}^{^{\prime}}\mathbf{f}_{t}$ are included
in the residuals $\hat{u}_{it}$, the resultant CD statistic tends to
over-state the degree of underlying error cross-sectional dependence. This
problem also arises when latent factors are proxied by cross section averages,
as is the case when panel data models are estimated using correlated common
effect (CCE) estimators proposed by \cite{pesaran2006estimation}, which we
shall address below in Section \ref{panel_model}.

We propose a bias-corrected CD test statistic, which we denote by $CD^{\ast}$,
that \textit{directly} corrects the asymptotic bias of the CD test using the
estimates of the factor loadings and error variances. To obtain the expression
for the bias we first note under the null hypothesis of cross-sectional
independence, $CD=z_{nT}+o_{p}(1)$ and
\begin{equation}
z_{nT}=\frac{1}{\sqrt{T}}\sum_{t=1}^{T}\left(  \frac{\xi_{t,n}^{2}-1}{\sqrt
{2}}\right)  +o_{p}\left(  1\right)  , \label{znT}%
\end{equation}
where
\begin{equation}
\xi_{t,n}=\frac{1}{\sqrt{n}}\sum_{i=1}^{n}a_{i,n}\varepsilon_{it},\text{
}a_{i,n}=1-\sigma_{i}\boldsymbol{\varphi}_{n}^{\prime}\boldsymbol{\gamma}_{i},
\label{ksi.ain}%
\end{equation}
$\boldsymbol{\varphi}_{n}=n^{-1}\sum_{i=1}^{n}\boldsymbol{\gamma}_{i}%
/\sigma_{i}$,$\boldsymbol{\ }$which is established in the proof of Proposition
\ref{Proposition:CD(s)} in the Appendix.\textbf{ }Since $a_{i,n}$ are given
constants, then $E\left(  \xi_{t,n}\right)  =0$,
\begin{equation}
E\left(  \xi_{t,n}^{2}\right)  \equiv\omega_{n}^{2}=\frac{1}{n}\sum_{i=1}%
^{n}a_{i,n}^{2}=n^{-1}\sum_{i=1}^{n}\left(  1-\sigma_{i}\boldsymbol{\varphi
}_{n}^{\prime}\boldsymbol{\gamma}_{i}\right)  ^{2}, \label{EknT2}%
\end{equation}
and
\begin{equation}
Var\left(  \xi_{t,n}^{2}\right)  =2\left(  \frac{1}{n}\sum_{i=1}^{n}%
a_{i,n}^{2}\right)  ^{2}-\kappa_{2}\left(  \frac{1}{n^{2}}\sum_{i=1}%
^{n}a_{i,n}^{4}\right)  , \label{vknT2}%
\end{equation}
where $\kappa_{2}=E\left(  \varepsilon_{it}^{4}\right)  -3$. Clearly, when the
errors are Gaussian then $E\left(  \varepsilon_{it}^{4}\right)  =3$, and the
second term of $Var\left(  \xi_{t,n}^{2}\right)  $ defined by (\ref{vknT2}) is
exactly zero. But even for non-Gaussian errors the second term of $Var\left(
\xi_{t,n}^{2}\right)  $ is negligible when $n$ is sufficiently large. To see
this note that under Assumptions \ref{ass:errors} and \ref{ass:loadings}
\[
\frac{1}{n^{2}}\sum_{i=1}^{n}a_{i,n}^{4}=\frac{1}{n^{2}}\sum_{i=1}%
^{n}(1-\sigma_{i}\boldsymbol{\varphi}_{n}^{\prime}\boldsymbol{\gamma}_{i}%
)^{4}\leq\frac{C}{n},
\]
where $C$ is a positive constant. Since $\varepsilon_{it}$ (and henceforth
$\xi_{t,n}$) are assumed to be serially independent, then we can also compute
the mean and the variance of $z_{nT}$ as%
\begin{align*}
E\left(  z_{nT}\right)   &  =\frac{1}{\sqrt{T}}\sum_{t=1}^{T}\left(
\frac{\omega_{n}^{2}-1}{\sqrt{2}}\right)  =\sqrt{\frac{T}{2}}\left(
\omega_{n}^{2}-1\right)  ,\\
Var\left(  z_{nT}\right)   &  =\frac{1}{T}\sum_{t=1}^{T}Var\left(  \frac
{\xi_{t,n}^{2}}{\sqrt{2}}\right)  =\frac{Var\left(  \xi_{t,n}^{2}\right)  }%
{2}.
\end{align*}
The above expressions for $E\left(  z_{nT}\right)  $ give the source of the
asymptotic bias of $CD$ as $E\left(  z_{nT}\right)  $ rises with $\sqrt{T}$,
unless
\[
\lim_{n\rightarrow\infty}\omega_{n}^{2}=\lim_{n\rightarrow\infty}n^{-1}%
\sum_{i=1}^{n}\left(  1-\sigma_{i}\boldsymbol{\varphi}_{n}^{\prime
}\boldsymbol{\gamma}_{i}\right)  ^{2}=1.
\]
A bias-corrected version of $CD$ can be defined by%
\begin{equation}
CD^{\ast}(\theta_{n})=\frac{CD+\sqrt{\frac{T}{2}}\theta_{n}}{1-\theta_{n}},
\label{CD(s)}%
\end{equation}
where
\begin{equation}
\theta_{n}=1-n^{-1}\mathbf{a}_{n}^{\prime}\mathbf{a}_{n}\text{, }%
\mathbf{a}_{n}=\left(  a_{1,n},a_{2,n},\ldots,a_{n,n}\right)  ^{\prime},
\label{thetan}%
\end{equation}
and $1-\theta_{n}>0$ by condition (\ref{thetaCon}). Also upon using
(\ref{ksi.ain})%
\begin{equation}
\theta_{n}=2\left(  n^{-1}\sum_{i=1}^{n}\sigma_{i}\boldsymbol{\gamma}%
_{i}^{\prime}\right)  \boldsymbol{\varphi}_{n}-\boldsymbol{\varphi}%
_{n}^{\prime}\left(  \frac{1}{n}\sum_{i=1}^{n}\sigma_{i}^{2}\boldsymbol{\gamma
}_{i}\boldsymbol{\gamma}_{i}^{\prime}\right)  \boldsymbol{\varphi}_{n}\text{.
} \label{thetan1}%
\end{equation}
The main difference between $CD$ and $CD^{\ast}\left(  \theta_{n}\right)  $
depends on the magnitude of $\sqrt{T}\theta_{n}$, which in turn depends on the
strengths of the factor loadings. Following \cite{bailey2021measurement}, we
measure the strength of factor $j$ by $\alpha_{j},$ defined by the rate at
which the sum of absolute values of factor loadings rises with $n$, namely
\begin{equation}
\sum_{i=1}^{n}\left\vert \gamma_{ij}\right\vert =\ominus\left(  n^{\alpha_{j}%
}\right)  ,\text{ for }j=1,2,\ldots,m_{0}, \label{fs}%
\end{equation}
where $0\leq\alpha_{j}\leq1$. Using (\ref{thetan1}) it is now easily
established that $\theta_{n}$ $=\ominus\left(  n^{\alpha-1}\right)  $, where
$\alpha=max_{j=1,2,...,m_{0}}(\alpha_{j})$, and $\theta_{n}$ does not tend to
zero when there is at least one strong factor in the panel data
model.\footnote{For a proof see Section \ref{section.thetan} of the supplement.}. Therefore, based on (\ref{CD(s)}), the relationship between $CD$
and $CD^{\ast}\left(  \theta_{n}\right)  $ is essentially controlled by the
maximum factor strength $\alpha\ $as $\sqrt{T}\theta_{n}=O\left(
T^{1/2}n^{\alpha-1}\right)  $. Suppose now $T=\ominus\left(  n^{d}\right)  $
for some $d>0$, then $\sqrt{T}\theta_{n} =\ominus\left(  n^{\alpha+d/2-1}\right)  ,$ and the bias
correction becomes negligible if $\alpha<1-d/2$. Under the required relative
expansion rates of $n$ and $T$ entertained in this paper, we need to set
$d=1$, and for this choice the bias correction term, $\sqrt{T}\theta_{n}$,
becomes negligible if $\alpha<1/2$, and as a result $CD$ and
$CD^{\ast}\left(  \theta_{n}\right)  $ will be asymptotically equivalent. In fact,
the case of strong factors assumed in the PCA literature corresponds to
$\alpha_{j}=1$ for $j=1,2,\ldots,m_{0}$, which is also fulfilled by Assumption
\ref{ass:loadings} and used in our mathematical derivations.

The theoretical results for $CD^{\ast}(\theta_{n})$ are summarized in the
following proposition.

\begin{proposition}
\label{Proposition:CD(s)}Suppose that observations on $y_{it}$, for
$i=1,2,\ldots,n,$ and $t=1,2,\ldots,T$ are generated from the pure latent
factor model given by (\ref{mod2}) and (\ref{uit:p}), where the number of
factors, $m_{0}$, is known. Consider the statistic $CD^{\ast}\left(
\theta_{n}\right)  $ defined by (\ref{CD(s)}) and assume $(n,T)\rightarrow
\infty$, such that $n/T\rightarrow\kappa,$ and $0<\kappa<\infty$.

(a) Under the null hypothesis $H_{0}$, defined by (\ref{H0}), and supposing
that Assumptions \ref{ass:factors} to \ref{ass:loadings} hold, then
\begin{equation}
CD^{\ast}\left(  \theta_{n}\right)  \rightarrow_{d}\mathcal{N}(0,1).
\label{CDs:H0}%
\end{equation}

(b) Under local alternatives $H_{1T}$, defined by (\ref{H1T}), and supposing
that Assumptions \ref{ass:factors} to \ref{ass:ws} hold, then
\begin{equation}
CD^{\ast}\left(  \theta_{n}\right)  \rightarrow_{d}\mathcal{N}(\phi,1),
\label{CDs:H1T}%
\end{equation}
where $\phi=\lim_{n\rightarrow\infty}\phi_{n}$ and
\begin{equation}
\phi_{n}=\frac{\sqrt{2}c_{\lambda}}{1-\theta_{n}}n^{-1}\mathbf{a}_{n}^{\prime
}\mathbf{Wa}_{n}, \label{phi_nT}%
\end{equation}
$\mathbf{W}=\left(  w_{ij}\right)  $ is the connection matrix, $\mathbf{a}%
_{n}=\left(  a_{1,n},a_{2,n},\ldots,a_{n,n}\right)  ^{\prime}$, with $a_{i,n}$
and $\theta_{n}$ defined by (\ref{ksi.ain}) and (\ref{thetan}), respectively.
\end{proposition}

For a proof see the Appendix.\bigskip

The bias-corrected test statistic, $CD^{\ast}(\theta_{n}),$ depends on the
unknown parameter, $\theta_{n}$, which can be estimated by
\begin{equation}
\hat{\theta}_{nT}=1-\frac{1}{n}\sum_{i=1}^{n}\hat{a}_{i,nT}^{2} \label{htheta}%
\end{equation}
where
\begin{equation}
\hat{a}_{i,nT}=1-\hat{\sigma}_{i,T}\left(  \boldsymbol{\hat{\varphi}}%
_{nT}^{\prime}\boldsymbol{\hat{\gamma}}_{i}\right)  \text{, and\ }%
\boldsymbol{\hat{\varphi}}_{nT}=\frac{1}{n}\sum_{i=1}^{n}\boldsymbol{\hat
{\gamma}}_{i}/\hat{\sigma}_{i,T}. \label{aihat}%
\end{equation}
The following proposition establishes the probability order of the difference
between $\hat{\theta}_{nT}$ and $\theta_{n}$.

\begin{proposition}
\label{Proposition:theta}Suppose that observations on $y_{it}$, for
$i=1,2,\ldots,n,$ and $t=1,2,\ldots,T$ are generated from the pure latent
factor model given by (\ref{mod2}) and (\ref{uit:p}), where the number of
factors, $m_{0}$, is known, and $\lambda_{T}=c_{\lambda}T^{-1/2}$ with
$\left\vert c_{\lambda}\right\vert <\infty$. Consider the term $\theta_{n}$ in
the $CD^{\ast}\left(  \theta_{n}\right)  $ statistic given by (\ref{thetan})
and its estimator $\hat{\theta}_{nT}$ given by (\ref{htheta}). Let Assumptions
\ref{ass:factors} to \ref{ass:ws} hold and $(n,T)\rightarrow\infty$, such that
$n/T\rightarrow\kappa$, where $0<\kappa<\infty$. Then
\begin{equation}
\sqrt{T}\left(  \hat{\theta}_{nT}-\theta_{n}\right)  =o_{p}(1).
\label{diff theta}%
\end{equation}

\end{proposition}

\medskip For a proof see the Appendix.\bigskip

Consider now the following feasible version of $CD^{\ast}\left(  \theta
_{n}\right)  $,%

\begin{equation}
CD^{\ast}\left(  \hat{\theta}_{nT}\right)  =\frac{CD+\sqrt{\frac{T}{2}}%
\hat{\theta}_{nT}}{1-\hat{\theta}_{nT}}, \label{CD*a}%
\end{equation}
and note that in view of (\ref{CD(s)}) and (\ref{diff theta}) we have
\begin{align*}
CD^{\ast}\left(  \hat{\theta}_{nT}\right)   &  =\left(  \frac{1-\theta_{n}%
}{1-\hat{\theta}_{nT}}\right)  \frac{CD+\sqrt{\frac{T}{2}}\theta_{n}%
+\sqrt{\frac{T}{2}}\left(  \hat{\theta}_{nT}-\theta_{n}\right)  }{1-\theta
_{n}}\\
&  =\left(  \frac{1-\theta_{n}}{1-\hat{\theta}_{nT}}\right)  \left[  CD^{\ast
}\left(  \theta_{n}\right)  +o_{p}(1)\right]  .
\end{align*}
Also,
\[
\frac{1-\theta_{n}}{1-\hat{\theta}_{nT}}=1+\frac{\sqrt{T}\left(  \hat{\theta
}_{nT}-\theta_{n}\right)  }{\sqrt{T}\left(  1-\theta_{n}\right)  -\sqrt
{T}\left(  \hat{\theta}_{nT}-\theta_{n}\right)  }=1+o_{p}(1),
\]
and hence $CD^{\ast}\left(  \hat{\theta}_{nT}\right)  =CD^{\ast}(\theta
_{n})+o_{p}(1)$. We refer to $CD^{\ast}\left(  \hat{\theta}_{nT}\right)  $
simply as $CD^{\ast}$ and the test based on it as the CD$^{\text{*}}$ test.
The main result of the paper for pure latent factor models is summarized in
the following theorem.

\begin{theorem}
\label{Theorem:CDs} Suppose that observations on $y_{it}$, for $i=1,2,\ldots
,n,$ and $t=1,2,\ldots,T$ are generated from the pure latent factor model
given by (\ref{mod2}) and (\ref{uit:p}), where the number of factors, $m_{0}$,
is known. Consider the statistic $CD^{\ast}$ defined by (\ref{CD*a}), and
assume $(n,T)\rightarrow\infty$, such that $n/T\rightarrow\kappa,$ and
$0<\kappa<\infty$.

(a) Under the null hypothesis $H_{0}$, defined by (\ref{H0}), and supposing
that Assumptions \ref{ass:factors} to \ref{ass:loadings} hold, then
\[
CD^{\ast}\rightarrow_{d}\mathcal{N}(0,1).
\]

(b) Under local alternatives $H_{1T}$, defined by (\ref{H1T}), and supposing
that Assumptions \ref{ass:factors} to \ref{ass:ws} hold, then
\[
CD^{\ast}\rightarrow_{d}\mathcal{N}\left(  \phi,1\right)  ,
\]
where $\phi=\lim_{n\rightarrow\infty}\phi_{n}$, and $\phi_{n}$ is defined by
(\ref{phi_nT}).
\end{theorem}

For a proof see the Appendix.\bigskip

This theorem establishes the conditions under which the proposed CD$^{\ast}$
test has the correct size asymptotically. It also shows that the CD$^{\ast}$
test has power against network alternatives if the limit of $\phi_{n}$ defined
by (\ref{phi_nT}) is nonzero, namely so long as $\lim_{n\rightarrow\infty
}n^{-1}\mathbf{a}_{n}^{\prime}\mathbf{Wa}_{n}\neq0$. This condition is likely
to be satisfied if the connection matrix, $\mathbf{W}$, is not too sparse,
although it must be sufficiently sparse so that Assumption \ref{ass:ws} is
met. In the case where there are no latent factors, $\mathbf{a}_{n}%
=(1,1,...,1)^{\prime}$, it is sufficient that $n^{-1}\sum_{i=1}^{n}%
\sum_{j=1}^{n}w_{ij}\neq0$.

To our knowledge, this is the first paper to provide a formal derivation of the power
function of CD tests against spatial and network alternatives, which applies
equally to the CD test for panel data models without latent factors. Hence,
our derivation of the power function can be used to supplement earlier
research on CD tests.

As we shall see from the Monte Carlo results reported below, the CD$^{\ast}$
test performs well even if some of the latent factors happen to be weak with
$\alpha_{j}\in(0,1/2]$ or semi-strong with $\alpha_{j}\in\left(  1/2,1\right)
$. This is because when a factor is weak, it does not matter if its estimation
by PCA is not consistent at the standard rate of $\delta_{nT}=\min
(n^{1/2},T^{1/2})$, and its inclusion or exclusion from the analysis has no
material impact on the $CD^{\ast}$ statistics for $n$ and $T$ sufficiently
large. In view of this result, in the mathematical derivations it is
sufficient to consider the case of strong factors, and let the weak factors to
be absorbed in the error term.

However, it should be acknowledged that our derivations do not take account of
the case when one or more of the factors are semi-strong. Such an extension is
beyond the scope of the present paper, although recent studies by
\cite{bai2023approximate} and \cite{jiang2023revisiting} show that PCA
estimation is asymptotically valid for factor models so long as factor
strengths are all above $1/2$. It is therefore reasonable to conjecture that
the CD$^{\text{*}}$ test applied to PCA residuals will be asymptotically valid
even if some of the factors are semi-strong, namely if $1/2<\alpha_{j}<1$.

In practice, the true number of factors, $m_{0}$, is unknown. In cases where
the estimated number of factors, $\hat{m}$, is underestimated ($\hat{m}<m_{0}%
$), the CD$^{\text{*}}$ test has power against the missing strong factors.
However, the rejection of the null hypothesis by the CD$^{\text{*}}$ test does
not necessarily mean there are missing factors, since the rejection could be
due to network error dependence. It is, therefore, important for the
investigator to decide on the number of strong latent factors before the
implementation of the proposed CD$^{\text{*}}$ test. To that end, we refer the
reader to the information criterion approach advanced by
\cite{bai2002determining} and the eigenvalue ratio test of
\cite{ahn2013eigenvalue}, for example.

\section{The CD$^{\text{*}}$ test for panel regression models with interactive
effects\label{panel_model}}

Consider now the factor model (\ref{mod2}) augmented with observed regressors
\begin{equation}
y_{it}=\boldsymbol{\alpha}_{i}^{\prime}\mathbf{d}_{t}+\boldsymbol{\beta}%
_{i}^{\prime}\mathbf{x}_{it}+\boldsymbol{\gamma}_{i}^{^{\prime}}\mathbf{f}%
_{t}+u_{it}, \label{mod3}%
\end{equation}
where $\mathbf{d}_{t}$ is a $k_{d}\times1$ vector of observed common factors,
$\mathbf{x}_{it}$ is a $k_{x}\times1$ vector of unit-specific observed
covariates, $\boldsymbol{\alpha}_{i}=(\alpha_{i1},\alpha_{i2},...,\alpha
_{ik_{d}})^{\prime}$, and $\boldsymbol{\beta}_{i}=(\beta_{i1},\beta
_{i2},...,\beta_{ik_{x}})^{\prime}$ are their associated unknown coefficients.
To highlight the relevance of the CD test for this set up, model (\ref{mod3})
can be written alternatively as
\begin{equation}
y_{it}=\boldsymbol{\alpha}_{i}^{\prime}\mathbf{d}_{t}+\boldsymbol{\beta}%
_{i}^{\prime}\mathbf{x}_{it}+v_{it}\text{, } \label{mod1}%
\end{equation}
where the errors, $v_{it}$, follow the factor structure
\begin{equation}
v_{it}=\boldsymbol{\gamma}_{i}^{\prime}\mathbf{f}_{t}+u_{it}. \label{vit}%
\end{equation}
The CD test is applicable, without any modifications, to test the null
hypothesis that the errors of the panel regression model, $v_{it}$, are
cross-sectionally independent, so long as the regressors, $\mathbf{d}_{t}$ and
$\mathbf{x}_{it}$, are strictly exogenous with respect to $v_{it}$. When the
regressors are correlated with the errors, the least squares estimates of
$v_{it}$ become inconsistent and the standard CD test will fail. One important
example of endogeneity arises when both $y_{it}$ and $\mathbf{x}_{it}$ are
driven by the same latent factor(s). \cite{pesaran2006estimation} formalizes
this form of endogeneity by assuming that%
\begin{equation}
\mathbf{x}_{it}=\mathbf{A}_{i}^{^{\prime}}\mathbf{d}_{t}+\boldsymbol{\Gamma
}_{i}^{^{\prime}}\mathbf{f}_{t}+\boldsymbol{\varepsilon}_{xit}, \label{xitL}%
\end{equation}
where $\mathbf{A}_{i}$ and $\boldsymbol{\Gamma}_{i}$\ are $k_{d}\times k_{x}$
and $m_{0}\times k_{x}$ factor loading matrices and $\boldsymbol{\varepsilon
}_{xit}$ are distributed independently of $\mathbf{f}_{t}$. The system of
equations (\ref{mod1}), (\ref{vit}) and (\ref{xitL}) fully specify the
dependence of $\mathbf{x}_{it}$ and $v_{it}$, and allows consistent estimation
of $v_{it}$ which can then be used to test the hypothesis that $u_{it}$ are
cross-sectionally independent in the pure latent factor model (\ref{vit}). We
now show that the CD$^{\text{*}}$ test applied to these residuals will be
valid. To this end we make the following additional standard assumptions.

\begin{assumption}
\label{ass:xfactors}(a) The $k_{d}\times1$\ vector $\mathbf{d}_{t}$ is a
covariance stationary process, with absolute summable autocovariances and
$\mathbf{d}_{t}$ is distributed independently of $\mathbf{f}_{t^{^{\prime}}},$
for all $t$ and $t^{^{\prime}}$, such that $T^{-1}\mathbf{D}^{^{\prime}%
}\mathbf{F}=O_{p}\left(  T^{-1/2}\right)  $, where $\mathbf{D}=\left(
\mathbf{d}_{1},\mathbf{d}_{2},\ldots,\mathbf{d}_{T}\right)  ^{^{\prime}}$ and
$\mathbf{F}=\left(  \mathbf{f}_{1},\mathbf{f}_{2},\ldots,\mathbf{f}%
_{T}\right)  ^{^{\prime}}$ are matrices of observations on $\mathbf{d}_{t}$
and $\mathbf{f}_{t}$. (b) $\left(  \mathbf{d}_{t},\mathbf{f}_{t}\right)  $ is
distributed independently of $u_{is}$ and $\boldsymbol{\varepsilon}_{xis}$ for
all $i,t,s.$
\end{assumption}

\begin{assumption}
\label{ass:xfactor loadings}The unobserved factor loadings $\boldsymbol{\Gamma
}_{i}$ are bounded, i.e. $\left\Vert \boldsymbol{\Gamma}_{i}\,\right\Vert
_{2}<C$ for all $i.$
\end{assumption}

\begin{assumption}
\label{ass:xerrors}The individual-specific errors $\varepsilon_{it}$ in (\ref{uit:p}) and
$\boldsymbol{\varepsilon}_{xi,t^{\prime}}$ are distributed independently for
all $i,j,t$ and $t^{\prime},$ and $\boldsymbol{\varepsilon}_{xit}$ follows the
linear stationary process $\boldsymbol{\varepsilon}_{xit}=\sum_{l=0}^{\infty
}\mathbf{S}_{il}\boldsymbol{\eta}_{xi,t-l},$ where for each $i$,
$\boldsymbol{\eta}_{xit}$ is a $k_{x}\times1$ vector of serially uncorrelated
random variables with mean zero, the variance matrix $\mathbf{I}_{k_{x}}$, and
finite fourth-order cumulants. For each $i$, the coefficient matrices
$\mathbf{S}_{il}$ satisfy the condition
\[
Var\left(  \boldsymbol{\varepsilon}_{xit}\right)  =\sum_{l=0}^{\infty
}\mathbf{S}_{il}\mathbf{S}_{il}^{^{\prime}}=\boldsymbol{\Sigma}_{xi},
\]
where $\boldsymbol{\Sigma}_{xi}$ is a positive definite matrix, such that
$\sup_{i}\left\vert \left\vert \boldsymbol{\Sigma}_{xi}\right\vert \right\vert
_{2}<C$.
\end{assumption}

\begin{assumption}
\label{ass:rank condition}Let $\tilde{\boldsymbol{\Gamma}}=E\left(
\boldsymbol{\gamma}_{i},\boldsymbol{\Gamma}_{i}\right)  .$ We assume that
$Rank\left(  \tilde{\boldsymbol{\Gamma}}\right)  =m_{0}.$
\end{assumption}

\begin{assumption}
\label{ass:identification}Consider\ the cross-sectional averages of the
individual-specific variables, $\mathbf{z}_{it}=\left(  y_{it},\mathbf{x}%
_{it}^{^{\prime}}\right)  ^{^{\prime}}$ defined by $\bar{\mathbf{z}}%
_{t}=n^{-1}\sum_{i=1}^{n}\mathbf{z}_{it}$, and let $\mathbf{\bar{M}%
}=\mathbf{I}_{T}-\mathbf{\bar{H}}\left(  \mathbf{\bar{H}}^{^{\prime}%
}\mathbf{\bar{H}}\right)  ^{-1}\mathbf{\bar{H}}^{^{\prime}}$, and
$\mathbf{M}_{g}=\mathbf{I}_{T}-\mathbf{G}\left(  \mathbf{G}^{^{\prime}%
}\mathbf{G}\right)  ^{-1}\mathbf{G}^{^{\prime}}$, where $\mathbf{\bar{H}%
=}\left(  \mathbf{D},\mathbf{\bar{Z}}\right)  ,$ $\mathbf{G}=\left(
\mathbf{D,F}\right)  ,$ and $\mathbf{\bar{Z}=}\left(  \mathbf{\bar{z}}%
_{1},\mathbf{\bar{z}}_{2},\ldots,\mathbf{\bar{z}}_{T}\right)  ^{\prime}$ is
the $T\times\left(  k_{x}+1\right)  $ matrix of observations on the
cross-sectional averages. Let $\mathbf{X}_{i}=\mathbf{(x}_{i1},\mathbf{x}%
_{i2},...,\mathbf{x}_{iT})^{\prime}$, then the $k\times k$ matrices
$\mathbf{\hat{\Psi}}_{i,T}=T^{-1}\mathbf{X}_{i}^{^{\prime}}\mathbf{\bar{M}%
X}_{i}$ and $\mathbf{\Psi}_{ig}=T^{-1}\mathbf{X}_{i}^{^{\prime}}\mathbf{M}%
_{g}\mathbf{X}_{i}$ are non-singular, and $\mathbf{\hat{\Psi}}_{i,T}^{-1}$ and
$\mathbf{\Psi}_{ig}^{-1}$ have finite second-order moments for all $i.$

\begin{remark}
The above assumptions are standard in the panel data models with multi-factor
error structure. See, for example, \cite{pesaran2006estimation}. But in our
setup under Assumption \ref{ass:errors} we require the error term,
$\varepsilon_{it}$, to be serially independent, since our focus is on testing
$\varepsilon_{it}$ for cross-sectional independence, and this assumption is
needed for asymptotic normality of the bias-corrected CD test. Later in
Section \ref{section_serial}, we will consider models with serially correlated
errors and show that the bias-corrected CD test remains valid. Nevertheless,
we allow $\varepsilon_{xit}$, the errors in the $\mathbf{x}_{it}$ equations to
be serially correlated. Assumption \ref{ass:xfactors} separates the observed
and the latent factors, as in Assumption 11 of \cite{pesaran2011large}. This
assumption is required to obtain the probability order of estimated residuals
needed for computation of $CD^{\ast}$ statistic. A necessary condition for the
rank condition in Assumption \ref{ass:rank condition} to hold is $k_{x}\geq
m_{0}-1$.
\end{remark}
\end{assumption}

To estimate $v_{it}$ we first filter out the effects of observed covariates
using the CCE estimators proposed in \cite{pesaran2006estimation}, namely for
each $i$ we estimate $\boldsymbol{\beta}_{i}$ by%
\begin{equation}
\hat{\boldsymbol{\beta}}_{{CCE,i}}=\left(  \mathbf{X}_{i}^{^{\prime}%
}\mathbf{\bar{M}}\mathbf{X}_{i}\right)  ^{-1}\left(  \mathbf{X}_{i}^{^{\prime
}}\mathbf{\bar{M}y}_{i}\right)  , \label{b_CCE}%
\end{equation}
and following \cite{pesaran2011large}, estimate $\boldsymbol{\alpha}_{i}$ by%
\begin{equation}
\boldsymbol{\hat{\alpha}}_{CCE,i}=\left(  \mathbf{D}^{^{\prime}}%
\mathbf{D}\right)  ^{-1}\mathbf{D}^{^{\prime}}\left(  \mathbf{y}%
_{i}-\mathbf{X}_{i}\hat{\boldsymbol{\beta}}_{{CCE,i}}\right)  . \label{a_CCE}%
\end{equation}
Then we have the following estimator of $v_{it}$%
\begin{equation}
\hat{v}_{it}=y_{it}-\boldsymbol{\hat{\alpha}}_{CCE,i}^{^{\prime}}%
\mathbf{d}_{t}-\hat{\boldsymbol{\beta}}_{{CCE,i}}^{^{\prime}}\mathbf{x}_{it}.
\label{vhat_it0}%
\end{equation}
Using results in \cite{pesaran2011large} (p. 189) it follows that under
Assumptions \ref{ass:factors}-\ref{ass:identification}%
\begin{equation}
\hat{v}_{it}=v_{it}+O_{p}\left(  \frac{1}{n}\right)  +O_{p}\left(  \frac
{1}{\sqrt{T}}\right)  +O_{p}\left(  \frac{1}{\sqrt{nT}}\right)  .
\label{vhat_it}%
\end{equation}
Note when $\boldsymbol{\alpha}_{i}=\mathbf{0}$ and $\boldsymbol{\beta}%
_{i}=\mathbf{0}$, (\ref{mod1}) reduces to the pure latent factor model,
(\ref{mod2}), where PCA can be applied to $v_{it}=y_{it}$ directly. In the case of
panel regressions $\hat{v}_{it}$ can be used instead of $v_{it}$ to compute
the bias-corrected CD statistic given by (\ref{CD*a}). The errors involved
will become asymptotically negligible in view of the fast rate of convergence
of $\hat{v}_{it}$ to $v_{it}$, uniformly for each $i$ and $t$. Specifically,
as in the case of the pure latent factor model, we first compute $m_{0}$ PCs
of $\left\{  \hat{v}_{it};\text{ }i=1,\ldots,n;\text{ and }t=1,\ldots
,T\right\}  $ and the associated factor loadings, $(\boldsymbol{\hat{\gamma}%
}_{i},\mathbf{\hat{f}}_{t}),$ subject to the normalization $n^{-1}\sum
_{i=1}^{n}\boldsymbol{\hat{\gamma}}_{i}\boldsymbol{\hat{\gamma}}_{i}%
^{^{\prime}}=\mathbf{I}_{m_{0}}$. The residuals
\begin{equation}
\hat{u}_{it}=\hat{v}_{it}-\boldsymbol{\hat{\gamma}}_{i}^{^{\prime}%
}\mathbf{\hat{f}}_{t},\text{ for }i=1,\ldots,n\text{; and }t=1,\ldots,T,
\label{eit_x}%
\end{equation}
can then be used to compute the standard CD statistic, (\ref{cd}), and its
bias-corrected version, $CD^{\ast}$, using (\ref{CD*a}).

\begin{remark}
\label{CCERES}It is important to bear in mind that $\hat{u}_{it}$ is not the
same as the CCE residuals that result from running the panel regressions of
$y_{it}$ on $(\mathbf{d}_{t},\mathbf{x}_{it},\mathbf{\bar{z}}_{t})$. As shown
by \cite{juodis2022incidental}, the standard CD test applied to the CCE
residuals will result in over-rejection and is not recommended. In our
approach, we filter out the latent factors from $\hat{v}_{it}$ and use the
filtered residuals, $\hat{u}_{it}$, to compute the CD statistic and correct
it, as in CD$^{\ast}$, to allow for errors associated with estimation of
factors and their loadings.
\end{remark}

The following theorem extends Theorem \ref{Theorem:CDs} to panel
regression models with observed regressors.

\begin{theorem}
\label{Theorem:CDss} Suppose that observations on $y_{it}$, for $i=1,2,\ldots
,n,$ and $t=1,2,\ldots,T$ are generated from the panel regression model
defined by (\ref{mod1}),\ (\ref{vit}) and (\ref{xitL}), where the number of
latent factors in (\ref{vit}), $m_{0}$, is known. Consider the statistic
$CD^{\ast}$ given\ by (\ref{CD*a}) using the filtered residuals defined by
(\ref{eit_x}). Suppose that $(n,T)\rightarrow\infty$, such that
$n/T\rightarrow\kappa,$ and $0<\kappa<\infty$.

(a) Under the null hypothesis $H_{0}$, defined by (\ref{H0}), and supposing
that Assumptions \ref{ass:factors} to \ref{ass:loadings} and Assumptions
\ref{ass:xfactors} to \ref{ass:identification} hold, then
\[
CD^{\ast}\rightarrow_{d}\mathcal{N}(0,1).
\]

(b) Under local alternatives $H_{1T}$, defined by (\ref{H1T}), and supposing
that Assumptions \ref{ass:factors} to \ref{ass:identification} hold, then
\[
CD^{\ast}\rightarrow_{d}\mathcal{N}\left(  \phi,1\right)  ,
\]
where $\phi=\lim_{n\rightarrow\infty}\phi_{n}$ and $\phi_{n}$ defined by
(\ref{phi_nT}).
\end{theorem}

For a proof see the Appendix.

\section{ CD$^{\text{*}}$ tests for models with serially correlated
errors\label{section_serial}}

As shown by \cite{baltagi2016testing}, when the errors $u_{it}$ in
(\ref{mod3}) are serially correlated the variance of the standard CD test
statistic is not unity (even asymptotically) and the test is no longer valid.
The same also applies to the CD$^{\text{*}}$ test. To deal with this problem,
we propose two solutions which involve different ways of adjusting the
CD$^{\ast}$ test so that it will become applicable to panels with serially
correlated errors. The first method closely follows the variance adjustment
proposed by \cite{baltagi2016testing}, in which $CD^{\ast}$ is scaled by
$\varpi$ where
\begin{equation}
\varpi^{2}=\frac{2T}{n\left(  n-1\right)  }\sum_{i=2}^{n}\sum_{j=1}%
^{i-1}\boldsymbol{\tilde{\varepsilon}}_{i,T}^{\prime}\left(
\boldsymbol{\tilde{\varepsilon}}_{j,T}-\boldsymbol{\tilde{\varepsilon}%
}_{\left(  ij\right)  ,T}\right)  \boldsymbol{\tilde{\varepsilon}}%
_{j,T}^{\prime}\left(  \boldsymbol{\tilde{\varepsilon}}_{i,T}%
-\boldsymbol{\tilde{\varepsilon}}_{\left(  ij\right)  ,T}\right)  ,
\label{var_of_CD}%
\end{equation}
with $\boldsymbol{\tilde{\varepsilon}}_{i,T}=(\tilde{\varepsilon}%
_{i1,T},\tilde{\varepsilon}_{i2,T},\ldots,\tilde{\varepsilon}_{iT,T})^{\prime
}$, $\tilde{\varepsilon}_{it,T}$ defined in (\ref{egziit}) and
\[
\boldsymbol{\tilde{\varepsilon}}_{\left(  ij\right)  ,T}=\frac{1}{n-2}%
\sum_{1\leq\tau\neq i,j\leq n}\boldsymbol{\tilde{\varepsilon}}_{\tau,T}.
\]
The expression in (\ref{var_of_CD}) is the equivalent to that provided in
Theorem 3 of \cite{baltagi2016testing} but the factor of $2$ in
(\ref{var_of_CD}) is missing in their paper.\textbf{ }The same adjustment is
also applied to the CD$_{W+}$ test to allow for serially correlated errors.

Alternatively, following \cite{pesaran2004general}, we first transform the
panel regression model to eliminate the error serial correlation and then
apply the CD$^{\text{*}}$ test to the residuals of the transformed model. This
is possible so long as the error serial correlation can be approximated by a
finite order stationary autoregressive process. As a simple illustration
consider the pure latent factor model $y_{it}=\gamma_{i}f_{t}+u_{it},$ in
which factor $f_{t}$ and loading $\gamma_{i}$ are both latent, and the errors
$u_{it}$ are generated as $AR(1)$ processes, $u_{it}=\rho_{i}u_{it-1}%
+\epsilon_{it},$ where $\rho_{i}$ is the autoregression coefficient and
$\epsilon_{it}$ is serially independent, as well as being distributed
independently of $f_{t^{\prime}}$ for all $i$ and $t,t^{\prime}=1,2,\ldots,T$.
Testing the cross-sectional independence of $u_{it}$ is equivalent to testing
the cross-sectional independence of $\epsilon_{it}$ in the following
autoregressive distributed lag (ARDL) representation of $y_{it}$%
\[
y_{it}=\rho_{i}y_{i,t-1}+\gamma_{i}f_{t}-\rho_{i}\gamma_{i}f_{t-1}%
+\epsilon_{it},
\]
which can be written equivalently as a multi-factor AR panel regression model
\begin{equation}
y_{it}=\rho_{i}y_{i,t-1}+\boldsymbol{\mathring{\gamma}}_{i}^{\prime
}\mathbf{\mathring{f}}_{t}+\epsilon_{it}, \label{mod4_serial}%
\end{equation}
where $\mathbf{\mathring{f}}_{t}=(f_{t},f_{t-1})^{\prime}$, and
$\boldsymbol{\mathring{\gamma}}_{i}=\left(  \gamma_{i},-\rho_{i}\gamma
_{i}\right)  ^{\prime}$. Since $y_{i,t-1}$ is weakly exogenous, the
transformed model satisfies the setup of panel regression model (\ref{mod1})
with $\mathbf{\mathring{f}}_{t}$ viewed as a vector of latent variables with
the associated factor loadings, $\boldsymbol{\mathring{\gamma}}_{i}$. It
therefore follows that the CD$^{\text{*}}$ test can now be applied to test the
cross-sectional independence of $\epsilon_{it}$ in (\ref{mod4_serial}). We
refer to this test as the ARDL adjusted CD$^{\text{*}}$ test.

The same approach can also be used for panels with observed covariates. In
general, testing cross-sectional independence of $u_{it}$ in model
(\ref{mod3}) is equivalent to testing the cross-sectional independence of
$\epsilon_{it}$ in
\begin{equation}
y_{it}=\sum_{s=0}^{S}\boldsymbol{\alpha}_{i,s}^{\prime}\mathbf{d}_{t-s}%
+\sum_{s=1}^{S}\rho_{i,s}y_{it-s}+\sum_{s=0}^{S}\boldsymbol{\beta}%
_{i,s}^{\prime}\mathbf{x}_{it-s}+\mathbf{g}_{i}^{\prime}\mathbf{h}%
_{t}+\epsilon_{it}, \label{d_mod1}%
\end{equation}
where $\mathbf{h}_{t}$\ is an extended set of latent factors (that encompass
$\mathbf{f}_{t}$), and $\mathbf{g}_{i}$ are the associated factor loadings.
The number of lags $S$ is determined by the order of the AR specification
assumed for $u_{it}$ in (\ref{mod3}).

The variance adjustment is simpler to implement but it requires theoretical
justification in the context of panel data models with latent factors. The
ARDL adjustment is theoretically justified so long as the underlying errors
follow finite order AR processes. As we shall see both approaches work well in
dealing with serially correlated errors, at least in the context of the
limited MC designs that we are considering. Clearly, further theoretical and
Monte Carlo investigations are needed for a better understanding of the
relative merits of the two approaches.

\section{Small sample properties of CD$^{\ast}$ and CD$_{W^{+}}$
tests\label{mc}}

\subsection{Data generating process}

We consider the following data generating process
\begin{equation}
y_{it}=\text{a}_{i}+\sigma_{i}\left[  \beta_{i1}d_{t}+\beta_{i2}x_{it}%
+m_{0}^{-1/2}\boldsymbol{\gamma}_{i}^{\prime}\mathbf{f}_{t}+\varepsilon
_{it}\left(  \lambda\right)  \right]  ,\text{ }i=1,2,...,n;t=1,2,...,T,
\label{dgpy}%
\end{equation}
where $\varepsilon_{it}\left(  \lambda\right)  $ follows the first order
spatial autoregressive process, $SAR\left(  1\right)  $, such that%
\begin{equation}
\varepsilon_{it}\left(  \lambda\right)  =\lambda\sum_{j=1}^{n}w_{ij}%
\varepsilon_{jt}\left(  \lambda\right)  +c\text{ }\varepsilon_{it}.
\label{dgpe}%
\end{equation}
$\text{a}_{i}$ is a unit-specific effect$,$ $d_{t}$ is the observed common
factor, $x_{it}$ is the observed regressor that varies across $i$ and $t$,
$\mathbf{f}_{t}$ is the $m_{0}\times1$ vector of unobserved factors,
$\boldsymbol{\gamma}_{i}$ is the vector of associated factor loadings. The
scalar constants, $\sigma_{i}>0$, are generated as $\sigma_{i}^{2}%
=0.5+\frac{1}{2}\left(  s_{i}^{2}-1\right)  $, with $s_{i}^{2}\sim IID\chi
^{2}(2)$, which ensures that $E(\sigma_{i}^{2})=1$.

\subsubsection{DGP under the null hypothesis}

Under the null hypothesis, we set $\lambda=0$ and $c=1$, and consider both
serially independent errors and serially correlated errors, which are
generated by both Gaussian and non-Gaussian distributions:

\begin{itemize}
\item Serially independent errors: Gaussian errors, $\varepsilon_{it}\sim
IID\mathcal{N}(0,1)$; chi-squared distributed errors, $\varepsilon_{it}\sim
IID\left(  \frac{\chi^{2}(2)-2}{2}\right)  $.

\item Serially correlated errors: $\varepsilon_{it}=\rho_{\varepsilon
}\varepsilon_{it-1}+\sqrt{1-\rho_{\varepsilon}^{2}}$ $e_{\varepsilon it}$, for
$i=1,2,\ldots,n$ and $t=1,2,\ldots,T$, where $\rho_{\varepsilon}=0.5$ and
$e_{\varepsilon it}$ are generated as Gaussian errors, $e_{\varepsilon it}\sim
IID\mathcal{N}\left(  0,1\right)  $, or chi-squared distributed errors,
$e_{\varepsilon it}\thicksim IID\left(  \frac{\chi^{2}(2)-2}{2}\right)  $.
\end{itemize}

The focus of the experiments is on testing the null hypothesis that
$\varepsilon_{it}$ are cross-sectional independent, whilst allowing for the
presence of $m_{0}$ unobserved factors, $\mathbf{f}_{t}=(f_{1t},f_{2t}%
,...,f_{m_{0}t})^{\prime}$. We consider $m_{0}=1$ and $m_{0}=2$, and generate
the factor loadings $\boldsymbol{\gamma}_{i}=(\gamma_{i1},\gamma_{i2}%
)^{\prime}$ as:
\begin{align*}
\gamma_{i1}  &  \sim IID\mathcal{N}\left(  0.5,0.5\right)  \text{ for
}i=1,2,\ldots,\left[  n^{\alpha_{1}}\right]  ,\\
\gamma_{i2}  &  \sim IID\mathcal{N}\left(  1,1\right)  \text{ for
}i=1,2,\ldots,\left[  n^{\alpha_{2}}\right]  ,\\
\gamma_{ij}  &  =0\text{ for }i=\left[  n^{\alpha_{j}}\right]  +1,\left[
n^{\alpha_{j}}\right]  +2,....,n,\text{ and }j=1,2.
\end{align*}
In the one-factor case ($m_{0}=1$), we only include $f_{1t}$ as the latent
factor and denote its factor strength by $\alpha$. Three values of $\alpha$
are considered, namely $\alpha=1,2/3,1/2$, respectively representing strong,
semi-strong and weak factors. Similarly, in the two-factor case ($m_{0}=2$),
we include both $f_{1t}$ and $f_{2t}$ as the latent factors and consider the
following combinations of factor strengths: $(\alpha_{1},\alpha_{2})=\left[
(1,1),(1,2/3),(2/3,1/2)\right]  $. The intercepts a$_{i}$ are generated as
$IID\mathcal{N}(1,2)$ and fixed thereafter. The observed common factor is
generated as $d_{t}=\rho_{d}d_{t-1}+\sqrt{1-\rho_{d}^{2}}$ $v_{dt}$, with
$\rho_{d}=0.8,$ and $v_{dt}\thicksim IID\mathcal{N}(0,1)$, thus ensuring that
$E(d_{t})=0$ and $Var(d_{t})=1$. The observed unit-specific regressors,
$x_{it}$, for $i=1,2,\ldots,n$ are generated to have non-zero correlations with
the unobserved factors:
\begin{equation}
x_{it}=\gamma_{xi1}f_{1t}+\gamma_{xi2}f_{2t}+e_{xit}, \label{dgpx}%
\end{equation}
where $f_{jt}=r_{j}f_{j,t-1}+\sqrt{1-r_{j}^{2}}$ $v_{jt}$, with $r_{j}=0.9$
and $v_{jt}\sim IID\left(  \frac{\chi^{2}(2)-2}{2}\right)  $, for $j=1,2$. The
factor loadings in (\ref{dgpx}) are generated as $\gamma_{xi1}\sim IIDU\left(
0.25,0.75\right)  $ and $\gamma_{xi2}\sim IIDU\left(  0.1,0.5\right)  $. The
error term of (\ref{dgpx}) is generated as $e_{xit}=\rho_{i}e_{xi,t-1}%
+\sqrt{1-\rho_{i}^{2}}$ $v_{xit}\text{, }$where $\rho_{i}\sim IIDU(0,0.95)$
and $v_{xit}\thicksim IID\mathcal{N}(0,1)$.

We will examine the small sample properties of the CD and the bias-corrected
CD tests for both the pure latent factor model and for the panel regression
model which also includes observed covariates.

\begin{itemize}
\item In the case of the pure latent factor model we set $\beta_{i1}%
=\beta_{i2}=0$.

\item In the case of the panel regression model with latent factors, we allow
for heterogeneous slopes and generate the slopes of observed covariates,
$d_{t}$ and $x_{it}$, as $\beta_{i1}\sim IID\mathcal{N}(\mu_{\beta1}%
,\sigma_{\beta1}^{2}),$ and $\beta_{i2}\sim IID\mathcal{N}(\mu_{\beta2}%
,\sigma_{\beta2}^{2})$ where $\mu_{\beta1}=\mu_{\beta2}=0.5$ and
$\sigma_{\beta1}^{2}=\sigma_{\beta2}^{2}=0.25,$ respectively.
\end{itemize}

As our theoretical results show the null distributions of the CD and the
bias-corrected CD tests do not depend on a$_{i}$, $\beta_{i1}$ and $\beta
_{i2}$, it is therefore innocuous what values are chosen for these parameters.
Moreover, the average fit of the panel is controlled in terms of the limiting
value of the pooled R-squared defined by
\begin{equation}
PR_{nT}^{2}=1-\frac{(nT)^{-1}\sum_{i=1}^{n}\sum_{t=1}^{T}\sigma_{i}%
^{2}E\left(  \varepsilon_{it}^{2}\right)  }{(nT)^{-1}\sum_{i=1}^{n}\sum
_{t=1}^{T}Var\left(  y_{it}\right)  }. \label{PRnT}%
\end{equation}
Since the underlying processes, (\ref{dgpy}) and (\ref{dgpx}), are stationary
and $E\left(  \varepsilon_{it}^{2}\right)  =1$, we have%
\[
\lim_{T\rightarrow\infty}PR_{nT}^{2}=PR_{n}^{2}=\frac{n^{-1}\sum_{i=1}%
^{n}\sigma_{i}^{2}\left[  \beta_{i1}^{2}+\beta_{i2}^{2}Var\left(
x_{it}\right)  +m_{0}^{-1}\boldsymbol{\gamma}_{i}^{\prime}\boldsymbol{\gamma
}_{i}+2Cov\left(  x_{it},\boldsymbol{\gamma}_{i}^{^{\prime}}\mathbf{f}%
_{t}\right)  \right]  }{n^{-1}\sum_{i=1}^{n}Var\left(  y_{it}\right)  },
\]
where $\boldsymbol{\gamma}_{i}=\left(  \gamma_{i1},\gamma_{i2}\right)
^{\prime},$ $Var\left(  x_{it}\right)  =\boldsymbol{\gamma}_{xi}^{^{\prime}%
}\boldsymbol{\gamma}_{xi}+1,$ $Cov\left(  x_{it},\boldsymbol{\gamma}%
_{i}^{^{\prime}}\mathbf{f}_{t}\right)  =\boldsymbol{\gamma}_{xi}^{^{\prime}%
}\boldsymbol{\gamma}_{i},$ $\boldsymbol{\gamma}_{xi}=\left(  \gamma
_{xi1},\gamma_{xi2}\right)  ^{\prime}$, and
\[
Var\left(  y_{it}\right)  =\sigma_{i}^{2}\left[  \beta_{i1}^{2}+\beta_{i2}%
^{2}Var\left(  x_{it}\right)  +m_{0}^{-1}\boldsymbol{\gamma}_{i}^{\prime
}\boldsymbol{\gamma}_{i}+2m_{0}^{-1/2}Cov\left(  x_{it},\boldsymbol{\gamma
}_{i}^{^{\prime}}\mathbf{f}_{t}\right)  +1\right]  .
\]
Also since $\sigma_{i}^{2}$ and $\beta_{ij}$ are independently distributed and
$E(\sigma_{i}^{2})=1$, it then readily follows that $\lim_{n\rightarrow\infty
}PR_{n}^{2}=\eta^{2}/(1+\eta^{2})$, where%
\[
\eta^{2}=\mu_{\beta1}^{2}+\sigma_{\beta1}^{2}+\left(  \mu_{\beta2}^{2}%
+\sigma_{\beta2}^{2}\right)  \left[  1+E\left(  \boldsymbol{\gamma}%
_{xi}^{^{\prime}}\boldsymbol{\gamma}_{xi}\right)  \right]  +\frac{2\mu
_{\beta2}E\left(  \boldsymbol{\gamma}_{xi}^{^{\prime}}\boldsymbol{\gamma}%
_{i}\right)  }{\sqrt{m_{0}}}+\frac{E\left(  \boldsymbol{\gamma}_{i}^{^{\prime
}}\boldsymbol{\gamma}_{i}\right)  }{m_{0}}.
\]
By controlling the value of $\eta^{2}$ across the experiments we ensure that
the pooled R$^{2}$ in large samples is the same for all values of $\sigma
_{i}^{2}$. In particular, in the case of the pure latent model we have
$\eta^{2}=m_{0}^{-1}E\left(  \boldsymbol{\gamma}_{i}^{^{\prime}}%
\boldsymbol{\gamma}_{i}\right)  =O\left(  n^{\alpha-1}\right)  ,$ where
$\alpha=max(\alpha_{1},\alpha_{2})$.

\subsubsection{DGP under alternative hypotheses}

Under alternative hypotheses, using (\ref{dgpe}), we consider a spatial
alternative defined by
\begin{equation}
\boldsymbol{\varepsilon}_{\circ t}\left(  \lambda\right)  =c\left(
\lambda\right)  \left(  \mathbf{I}_{n}-\lambda\mathbf{W}\right)
^{-1}\boldsymbol{\varepsilon}_{\circ t}, \label{error}%
\end{equation}
where $\boldsymbol{\varepsilon}_{\circ t}\left(  \lambda\right)  =\left(
\varepsilon_{1t}\left(  \lambda\right)  ,\varepsilon_{2t}\left(
\lambda\right)  ,\ldots,\varepsilon_{nt}\left(  \lambda\right)  \right)
^{\prime},$ $\mathbf{W}=(w_{ij})$, and $\boldsymbol{\varepsilon}_{\circ
t}=(\varepsilon_{1t},\varepsilon_{2t},\ldots,\varepsilon_{nt})^{\prime}$. The
errors $\varepsilon_{it}$ are generated as described above. For the spatial
weights $w_{ij}$, we first set $w_{ij}^{0}=1$ if $j=i-2,i-1,i+1,i+2,$ and zero
otherwise. We then row normalize the weights such that $w_{ij}=\left(
\sum_{j=1}^{n}w_{ij}^{0}\right)  ^{-1}w_{ij}^{0}$. We also set $c\left(
\lambda\right)  ^{2}=n/$\textrm{$tr$}$\left[  \left(  \mathbf{I}_{n}%
-\lambda\mathbf{W}\right)  ^{-1}\left(  \mathbf{I}_{n}-\lambda\mathbf{W}%
\right)  ^{\prime-1}\right]  $, which ensures that $n^{-1}\sum_{i=1}%
^{n}Var(\varepsilon_{it}\left(  \lambda\right)  )=1$, for all values of
$\lambda$. In practice, only positive values of $\lambda$ are of interest, and
the power function need not be symmetric for all positive and negative values
of $\lambda$.

\subsection{CD, CD$^{\ast}$ and CD$_{W^{+}}$ tests}

All experiments are carried out for $n=100,200,500,1000$ and $T=100,200,500$,
and the number of replications is set to $2000$. Firstly we consider the DGPs
with serially independent errors. For the pure latent factor models, we
compute the filtered residuals as $\hat{v}_{it}=y_{it}-\hat{\text{a}}_{i}$,
where $\hat{\text{a}}_{i}=T^{-1}\sum_{t=1}^{T}y_{it}$. For the panel
regressions with latent factors, the filtered residuals are computed as%
\begin{equation}
\hat{v}_{it}=y_{it}-\hat{\text{a}}_{CCE,i}-\hat{\beta}_{CCE,i1}d_{t}%
-\hat{\beta}_{CCE,i2}x_{it}, \label{vit1}%
\end{equation}
where $\left(  \hat{\text{a}}_{CCE,i},\hat{\beta}_{CCE,i1},\hat{\beta
}_{CCE,i2}\right)  $ is the CCE estimator of a$_{i}$, $\beta_{i1}$ and
$\beta_{i2}$, as set out in \cite{pesaran2006estimation}. The residuals
$\{\hat{v}_{it};$ $i=1,2,\ldots,n;$ and $t=1,2,\ldots,T\}$, together with
their first $m$ PCs and the associated factor loadings, $(\boldsymbol{\hat
{\gamma}}_{i},\hat{\mathbf{f}}_{t}),$ are then used to compute the filtered
residuals, $\hat{u}_{it}=\hat{v}_{it}-\boldsymbol{\hat{\gamma}}_{i}^{\prime
}\hat{\mathbf{f}}_{t}$, to compute the CD test statistics, $CD$ and $CD^{\ast
}$, given by (\ref{cd}) and (\ref{CD*a}), respectively. For comparison, we
also consider the power enhanced version of the randomized CD test statistic
proposed by JR given by
\begin{equation}
CD_{W+}=CD_{W}+\Delta_{nT}, \label{CDW+}%
\end{equation}
where%
\begin{equation}
CD_{W}=\left(\frac{1}{nT}\sum_{i=1}^n\sum_{t=1}^T\hat{u}_{it}^2\right)^{-1}\left(\sqrt{\frac{2}{Tn(n-1)}}\sum_{t=1}^{T}\sum_{i=2}^{n}\sum_{j=1}%
^{i-1}  w_{i}\hat{u}_{it}    w_{j}\hat{u}_{jt}\right)  .
\label{CDW}%
\end{equation}
The weights $w_{i}$, for $i=1,2,...,n$ are independently drawn from a
Rademacher distribution and
\begin{equation}
\Delta_{nT}=\sum_{i=2}^{n}\sum_{j=1}^{i-1}\left\vert \hat{\rho}_{ij,T}%
\right\vert \mathbf{1}\left(  \left\vert \hat{\rho}_{ij,T}\right\vert
>2\sqrt{\frac{\ln(n)}{T}}\right)  , \label{Delta nT}%
\end{equation}
where $\hat{\rho}_{ij,T}=T^{-1}\sum_{t=1}^{T}\tilde{\varepsilon}_{it,T}%
\tilde{\varepsilon}_{jt,T}$, and $\tilde{\varepsilon}_{it,T}$ is defined by
(\ref{egziit}). As shown by JR, $CD_{W}$ has a zero mean by construction and avoids the over-rejection problem of the CD test, but it can also lack power by
the very nature of the randomization process. JR further suggest $CD_{W+}$ by
adding a screening component $\Delta_{nT}$ proposed by \cite{fan2015power},
which enhances the power of the test since $\Delta_{nT}$ converges to zero as
$n$ and $T\rightarrow\infty$ under the null hypotheses, but can diverge under
alternatives with a sufficient number of $(i,j)$ pairs with non-zero
correlations, $\rho_{ij}$.

As discussed in Section \ref{section_serial}, the CD$^{\text{*}}$ test is not
valid when the errors are serially correlated. In the simulations, we apply
the variance and ARDL adjustments to $CD$, $CD_{W+}$, and $CD^{\ast}$. The
variance adjusted versions are computed by scaling the original statistics by
the standard deviation of the CD statistics using the expression in
(\ref{var_of_CD}) with $\tilde{\varepsilon}_{it,T}=\hat{u}_{it}/\hat{\sigma}_{i,T}$, where $\hat{u}_{it}=\hat{v}_{it}-\boldsymbol{\hat{\gamma}%
}_{i}^{\prime}\hat{\mathbf{f}}_{t}$. The ARDL adjusted versions of $CD$,
$CD^{\ast}$, and $CD_{W+}$, are computed using the residuals from the
following dynamic panel data model with latent factors,%
\begin{equation}
y_{it}=a_{i}+\sum_{s=1}^{S}\rho_{is}y_{i,t-s}+\sum_{s=0}^{S}\beta
_{i1,s}d_{t-s}+\sum_{s=0}^{S}\beta_{i2,s}x_{i,t-s}+\mathbf{g}_{i}^{\prime
}\mathbf{h}_{t}+\epsilon_{it}. \label{dCDs_dgp2}%
\end{equation}
In the simulations we set $S=1$, but higher order values can also be
considered. The number of latent factors in $\mathbf{h}_{t}$ depends on $S$
and is given by $m_{h}=(S+1)m_{0}$. Accordingly, the number of selected PCs,
$\hat{m}$, should satisfy $\hat{m}\geq(S+1)m_{0}$. In the simulations if
$S=0$, we consider $\hat{m}=1$ and $2$ if $m_{0}=1$, and $\hat{m}=2$ and $4$
if $m_{0}=2$. But if $S=1$ we consider $\hat{m}=2$ and $4$ if $m_{0}=1$, and
$\hat{m}=4$ and $6$ if $m_{0}=2$. Seen from this perspective, the variance
adjustment approach to dealing with error serial correlation seems preferable
since it does not require specifying the lag order $S$.

\subsection{Simulation results\label{sec:simulation results}}

We first report the simulation results for the DGPs with normally distributed
errors, followed by the results based on DGPs with chi-squared distributed
errors. Next, we report simulation results for the DGPs with serially
correlated errors, using the variance and ARDL adjusted CD tests discussed in
Section \ref{section_serial}. Finally, to investigate the power of the
CD$^{\text{*}}$ test we consider the spatial $SAR(1)$ alternative with
$\lambda=0.25$. As to be expected the power rises very quickly as $\lambda$
deviates from $0$.\footnote{Simulated power functions are provided in the
supplement for $\lambda=\pm0.05$, $\pm0.1$, $\pm0.2$, $\pm0.3$,
$\pm0.4$, $\pm0.5$, $\pm0.6$, $\pm0.7$, $\pm0.8$, $\pm0.9$, $\pm0.95$.}

\subsubsection{Serially independent errors: normally distributed errors}

The simulation results for the DGPs with the errors following Gaussian
distribution are shown in Tables \ref{table.n.1f} to \ref{table.n.x2f.b}.
Tables \ref{table.n.1f} and \ref{table.n.1f.b} report the test results for the
latent factor model with one factor. Table \ref{table.n.1f} gives the results
for the case where the number of selected PCs, denoted by $\hat{m}$, is the
same as the true number of factors ($m_{0}=1$), while Table \ref{table.n.1f.b}
reports the results when $\hat{m}=2$. As to be expected the standard CD test
over-rejects when the factor is strong, namely when $\alpha=1$. By comparison,
the rejection frequencies of both CD$^{\ast}$ and CD$_{W+}$ tests under null
($\lambda=0)$ are generally around the nominal size of $5$ per cent. Under the
alternative (when $\lambda=0.25$), the CD$^{\ast}$ test has satisfactory power
properties with significantly high rejection frequencies even when the sample
size is small. But the CD$_{W+}$ test performs quite poorly under the spatial
alternative, especially when $T$ is small.

Tables \ref{table.n.2f} and \ref{table.n.2f.b} summarize the size and power
results for the latent factor model with $m_{0}=2,$ and reports the results
when $\hat{m}$ (the selected number of PCs) is set to $2$ (Table
\ref{table.n.2f}) and $4$ (Table \ref{table.n.2f.b}). The results are
qualitatively similar to the ones reported for the one factor model. The CD
test over-rejects if at least one of the factors is strong, and the empirical
sizes of CD$^{\ast}$ and CD$_{W+}$ tests are close to their nominal value of
$5$ per cent, although we now observe some mild over-rejection when $n=100$
and the selected number of PCs is 4. In terms of power, the CD$^{\ast}$ test
performs well, although there is some loss of power as the numbers of factors
and selected PCs rise. Similarly, the power of the CD$_{W^{+}}$ test is now
even lower and quite close to $5$ per cent when $T<500$ even if the number of
PCs is set to $m_{0}=2$.

Turning to panel regression models with latent factors estimated by CCE, the
associated simulation results are summarized in Tables \ref{table.n.x1f} to
\ref{table.n.x2f.b}. As can be seen, the results are very close to the ones
reported in Tables \ref{table.n.1f} to \ref{table.n.2f.b} for the latent
factor model, and are in line with the asymptotic result in (\ref{vhat_it})
that underlies the use of CCE approach to filter out the effects of observed
covariates, as well as latent factors.

\begin{table}[th!]
\caption{Size and power of tests of error cross-sectional dependence using one
PC $(\hat{m}=1)$ for the latent factor model with one factor $(m_{0}=1)$ and
serially independent Gaussian errors}%
\label{table.n.1f}%
\renewcommand{\thetable}{1A}
\par
\begin{center}
{\footnotesize
\begin{tabular}
[c]{crrrrrrrrrrrrrr}\hline\hline
\multicolumn{1}{l}{} &  &  &  & \multicolumn{3}{c}{$\alpha=1$} &  &
\multicolumn{3}{c}{$\alpha=2/3$} &  & \multicolumn{3}{c}{$\alpha=1/2$%
}\\\cline{5-7}\cline{9-11}\cline{13-15}%
\multicolumn{1}{c}{Tests} &  & $n\setminus T$ &  & 100 & 200 & 500 &  & 100 &
200 & 500 &  & 100 & 200 & 500\\\hline
&  &  &  &  &  &  &  &  &  &  &  &  &  & \\
\multicolumn{1}{l}{} &  & \multicolumn{13}{c}{Size ($H_{o}:\lambda=0$%
)}\\[0.5em]%
$CD$ &  & 100 &  & 64.7 & 88.1 & 97.5 &  & 5.8 & 9.7 & 22.5 &  & 5.3 & 5.9 &
9.4\\
&  & 200 &  & 67.7 & 92.3 & 99.4 &  & 5.3 & 7.1 & 14.2 &  & 5.9 & 5.4 & 7.0\\
&  & 500 &  & 71.0 & 95.2 & 100.0 &  & 5.1 & 4.2 & 8.6 &  & 6.2 & 4.9 & 4.3\\
&  & 1000 &  & 69.1 & 95.2 & 100.0 &  & 5.1 & 4.5 & 5.7 &  & 6.1 & 5.4 &
4.5\\[0.5em]%
$CD^{\ast}$ &  & 100 &  & 5.7 & 3.9 & 4.4 &  & 4.8 & 5.2 & 5.8 &  & 5.9 &
5.9 & 5.5\\
&  & 200 &  & 5.5 & 4.9 & 5.3 &  & 5.5 & 5.1 & 5.2 &  & 5.9 & 5.2 & 5.1\\
&  & 500 &  & 5.3 & 5.2 & 4.4 &  & 5.7 & 5.0 & 4.8 &  & 6.3 & 5.1 & 4.6\\
&  & 1000 &  & 4.4 & 5.3 & 5.1 &  & 5.5 & 4.8 & 4.6 &  & 6.1 & 5.5 &
4.9\\[0.5em]%
$CD_{W+}$ &  & 100 &  & 5.8 & 5.1 & 5.0 &  & 5.5 & 5.5 & 5.9 &  & 5.8 & 5.4 &
7.6\\
&  & 200 &  & 6.1 & 5.7 & 5.5 &  & 4.7 & 5.4 & 4.5 &  & 5.9 & 7.1 & 5.9\\
&  & 500 &  & 5.4 & 5.4 & 5.2 &  & 5.7 & 5.6 & 5.4 &  & 5.1 & 5.5 & 5.3\\
&  & 1000 &  & 5.1 & 4.6 & 4.8 &  & 4.6 & 5.9 & 5.3 &  & 4.7 & 5.9 & 6.0\\
&  &  &  &  &  &  &  &  &  &  &  &  &  & \\
\multicolumn{1}{l}{} &  & \multicolumn{13}{c}{Power ($H_{1}:\lambda=0.25$%
)}\\[0.5em]%
$CD$ &  & 100 &  & 23.8 & 37.5 & 55.7 &  & 68.9 & 86.7 & 97.7 &  & 81.0 &
93.8 & 99.4\\
&  & 200 &  & 16.0 & 31.0 & 50.6 &  & 75.4 & 93.1 & 99.8 &  & 84.9 & 97.6 &
100.0\\
&  & 500 &  & 10.9 & 22.9 & 46.3 &  & 82.3 & 97.5 & 100.0 &  & 89.6 & 98.9 &
100.0\\
&  & 1000 &  & 9.4 & 20.0 & 44.6 &  & 84.0 & 98.1 & 100.0 &  & 89.7 & 99.0 &
100.0\\[0.5em]%
$CD^{\ast}$ &  & 100 &  & 58.0 & 82.0 & 98.4 &  & 86.1 & 98.8 & 100.0 &  &
88.6 & 98.7 & 100.0\\
&  & 200 &  & 59.3 & 81.1 & 98.9 &  & 84.8 & 98.3 & 100.0 &  & 88.8 & 98.8 &
100.0\\
&  & 500 &  & 57.9 & 83.4 & 99.4 &  & 87.2 & 98.7 & 100.0 &  & 90.4 & 99.1 &
100.0\\
&  & 1000 &  & 60.1 & 84.0 & 99.4 &  & 86.7 & 99.2 & 100.0 &  & 90.1 & 99.2 &
100.0\\[0.5em]%
$CD_{W+}$ &  & 100 &  & 6.9 & 7.8 & 49.5 &  & 6.4 & 9.7 & 59.8 &  & 6.9 &
8.5 & 64.5\\
&  & 200 &  & 6.5 & 7.3 & 51.9 &  & 5.7 & 7.8 & 61.0 &  & 6.9 & 8.0 & 60.3\\
&  & 500 &  & 5.5 & 5.9 & 52.9 &  & 5.9 & 6.7 & 56.9 &  & 5.8 & 6.6 & 57.4\\
&  & 1000 &  & 4.9 & 5.6 & 51.7 &  & 4.6 & 6.2 & 54.2 &  & 5.6 & 5.6 &
52.8\\\hline
\end{tabular}
}
\end{center}
\par
{\footnotesize \textit{Notes}: The DGP is given by (\ref{dgpy}) with
$\beta_{i1}=\beta_{i2}=0$ and contains a single latent factor with different
factor strengths, $\alpha=1,$ $2/3$, and $1/2$. $\lambda$ denotes the spatial
autocorrelation coefficient defined by (\ref{error}). $m_{0}$ is the true
number of factors and $\hat{m}$ is the number of selected PCs used to compute
the different CD statistics. $CD$ denotes the standard test of error
cross-sectional dependence defined by (\ref{cd}), $CD^{\ast}$ is the
bias-corrected version defined by (\ref{CD*a}), and $CD_{W+}$ is the
power-enhanced randomized version defined by (\ref{CDW+}).}\end{table}

\begin{table}[th!]
\caption{Size and power of tests of error cross-sectional dependence using two
PCs $(\hat{m}=2)$ for the latent factor model with one factor $(m_{0}=1)$ and
serially independent Gaussian errors }%
\label{table.n.1f.b}%
\renewcommand{\thetable}{1B}
\par
\begin{center}
{\footnotesize
\begin{tabular}
[c]{crrrrrrrrrrrrrr}\hline\hline
\multicolumn{1}{l}{} &  &  &  & \multicolumn{3}{c}{$\alpha=1$} &  &
\multicolumn{3}{c}{$\alpha=2/3$} &  & \multicolumn{3}{c}{$\alpha=1/2$%
}\\\cline{5-7}\cline{9-11}\cline{13-15}%
\multicolumn{1}{c}{Tests} &  & $n\setminus T$ &  & 100 & 200 & 500 &  & 100 &
200 & 500 &  & 100 & 200 & 500\\\hline
&  &  &  &  &  &  &  &  &  &  &  &  &  & \\
\multicolumn{1}{l}{} &  & \multicolumn{13}{c}{Size ($H_{o}:\lambda=0$%
)}\\[0.5em]%
$CD$ &  & 100 &  & 65.2 & 87.9 & 97.6 &  & 5.9 & 10.2 & 21.8 &  & 5.6 & 6.2 &
8.8\\
&  & 200 &  & 68.3 & 91.7 & 99.5 &  & 5.1 & 7.6 & 14.6 &  & 5.8 & 5.1 & 7.1\\
&  & 500 &  & 70.7 & 94.9 & 100.0 &  & 5.0 & 4.5 & 8.2 &  & 6.3 & 5.0 & 5.0\\
&  & 1000 &  & 68.9 & 95.1 & 100.0 &  & 4.8 & 4.8 & 5.3 &  & 5.9 & 5.7 &
4.9\\[0.5em]%
$CD^{\ast}$ &  & 100 &  & 5.7 & 4.7 & 5.8 &  & 5.5 & 6.3 & 6.2 &  & 5.8 &
6.5 & 6.3\\
&  & 200 &  & 5.1 & 5.0 & 5.9 &  & 6.1 & 5.3 & 5.7 &  & 5.8 & 5.3 & 5.1\\
&  & 500 &  & 5.3 & 5.3 & 4.3 &  & 5.6 & 5.1 & 4.7 &  & 6.5 & 5.0 & 4.3\\
&  & 1000 &  & 4.3 & 5.1 & 5.0 &  & 5.8 & 5.1 & 5.0 &  & 6.3 & 5.8 &
4.6\\[0.5em]%
$CD_{W+}$ &  & 100 &  & 5.1 & 5.4 & 6.5 &  & 3.8 & 5.4 & 6.1 &  & 4.8 & 5.4 &
8.4\\
&  & 200 &  & 5.5 & 5.3 & 6.1 &  & 4.9 & 5.8 & 5.3 &  & 6.2 & 5.3 & 5.5\\
&  & 500 &  & 4.7 & 4.5 & 4.3 &  & 5.3 & 5.1 & 4.7 &  & 5.3 & 5.7 & 4.1\\
&  & 1000 &  & 3.9 & 4.6 & 4.9 &  & 4.9 & 5.7 & 4.6 &  & 5.9 & 4.6 & 4.8\\
&  &  &  &  &  &  &  &  &  &  &  &  &  & \\
\multicolumn{1}{l}{} &  & \multicolumn{13}{c}{Power ($H_{1}:\lambda=0.25$%
)}\\[0.5em]%
$CD$ &  & 100 &  & 26.6 & 44.2 & 62.7 &  & 59.4 & 73.8 & 85.8 &  & 70.5 &
83.8 & 92.2\\
&  & 200 &  & 17.3 & 34.5 & 56.9 &  & 70.5 & 89.3 & 98.3 &  & 81.1 & 95.2 &
99.4\\
&  & 500 &  & 12.0 & 25.2 & 49.7 &  & 80.2 & 96.8 & 100.0 &  & 88.2 & 98.3 &
100.0\\
&  & 1000 &  & 9.3 & 20.7 & 45.7 &  & 83.0 & 97.4 & 100.0 &  & 88.7 & 98.9 &
100.0\\[0.5em]%
$CD^{\ast}$ &  & 100 &  & 57.8 & 81.8 & 98.8 &  & 83.9 & 98.4 & 100.0 &  &
85.8 & 98.3 & 100.0\\
&  & 200 &  & 57.9 & 81.0 & 98.8 &  & 84.2 & 98.1 & 100.0 &  & 87.8 & 98.6 &
100.0\\
&  & 500 &  & 57.4 & 82.5 & 99.4 &  & 86.8 & 98.6 & 100.0 &  & 89.3 & 98.9 &
100.0\\
&  & 1000 &  & 59.7 & 83.6 & 99.4 &  & 86.2 & 99.1 & 100.0 &  & 89.4 & 99.2 &
100.0\\[0.5em]%
$CD_{W+}$ &  & 100 &  & 5.4 & 8.1 & 33.8 &  & 5.1 & 8.2 & 36.7 &  & 5.9 &
8.2 & 44.1\\
&  & 200 &  & 5.5 & 6.6 & 38.8 &  & 5.7 & 7.1 & 41.7 &  & 6.2 & 7.0 & 43.6\\
&  & 500 &  & 4.9 & 6.1 & 45.2 &  & 5.5 & 5.9 & 48.2 &  & 6.2 & 6.2 & 47.0\\
&  & 1000 &  & 3.7 & 5.2 & 43.9 &  & 5.3 & 6.6 & 48.3 &  & 6.0 & 6.1 &
47.3\\\hline
\end{tabular}
}
\end{center}
\par
{\footnotesize \textit{Notes}: See the notes to Table \ref{table.n.1f}.}\end{table}

\begin{table}[th!]
\caption{Size and power of tests of error cross-sectional dependence using two
PCs $(\hat{m}=2)$ for the latent factor model with two factors $(m_{0}=2)$ and
serially independent Gaussian errors }%
\label{table.n.2f}%
\renewcommand{\thetable}{2A}
\par
\begin{center}
{\footnotesize
\begin{tabular}
[c]{crrrrrrrrrrrrrr}\hline\hline
\multicolumn{1}{l}{} &  &  &  & \multicolumn{3}{c}{$\alpha_{1}=1,\alpha_{2}%
=1$} &  & \multicolumn{3}{c}{$\alpha_{1}=1,\alpha_{2}=2/3$} &  &
\multicolumn{3}{c}{$\alpha_{1}=2/3,\alpha_{2}=1/2$}\\\cline{5-7}%
\cline{9-11}\cline{13-15}%
\multicolumn{1}{c}{Tests} &  & $n\setminus T$ &  & 100 & 200 & 500 &  & 100 &
200 & 500 &  & 100 & 200 & 500\\\hline
&  &  &  &  &  &  &  &  &  &  &  &  &  & \\
\multicolumn{1}{l}{} &  & \multicolumn{13}{c}{Size ($H_{o}:\lambda=0$%
)}\\[0.5em]%
$CD$ &  & 100 &  & 99.9 & 100.0 & 100.0 &  & 98.3 & 99.9 & 100.0 &  & 8.3 &
15.4 & 40.7\\
&  & 200 &  & 100.0 & 100.0 & 100.0 &  & 99.4 & 100.0 & 100.0 &  & 6.6 & 8.6 &
24.1\\
&  & 500 &  & 100.0 & 100.0 & 100.0 &  & 99.7 & 100.0 & 100.0 &  & 5.7 & 5.8 &
13.0\\
&  & 1000 &  & 100.0 & 100.0 & 100.0 &  & 99.9 & 100.0 & 100.0 &  & 7.6 &
4.9 & 7.8\\[0.5em]%
$CD^{\ast}$ &  & 100 &  & 5.5 & 4.7 & 3.8 &  & 5.6 & 5.3 & 4.9 &  & 8.8 &
7.2 & 6.8\\
&  & 200 &  & 5.4 & 5.0 & 4.9 &  & 6.3 & 4.6 & 4.6 &  & 8.3 & 5.5 & 5.3\\
&  & 500 &  & 4.8 & 4.7 & 6.0 &  & 5.8 & 4.7 & 4.1 &  & 7.4 & 5.9 & 5.1\\
&  & 1000 &  & 5.1 & 4.5 & 4.8 &  & 5.4 & 5.4 & 4.7 &  & 8.3 & 6.2 &
5.2\\[0.5em]%
$CD_{W+}$ &  & 100 &  & 4.4 & 6.8 & 5.9 &  & 5.5 & 5.9 & 7.8 &  & 5.8 & 5.3 &
9.4\\
&  & 200 &  & 5.5 & 5.4 & 5.3 &  & 6.4 & 5.6 & 6.3 &  & 5.9 & 5.9 & 5.6\\
&  & 500 &  & 5.3 & 5.6 & 5.1 &  & 5.5 & 5.1 & 4.8 &  & 5.4 & 4.7 & 6.4\\
&  & 1000 &  & 4.9 & 4.6 & 4.6 &  & 5.1 & 4.3 & 5.4 &  & 6.6 & 5.4 & 5.1\\
&  &  &  &  &  &  &  &  &  &  &  &  &  & \\
\multicolumn{1}{l}{} &  & \multicolumn{13}{c}{Power ($H_{1}:\lambda=0.25$%
)}\\[0.5em]%
$CD$ &  & 100 &  & 99.4 & 100.0 & 100.0 &  & 90.3 & 98.0 & 99.4 &  & 56.7 &
65.4 & 81.3\\
&  & 200 &  & 99.8 & 100.0 & 100.0 &  & 93.1 & 99.2 & 100.0 &  & 71.5 & 84.3 &
97.8\\
&  & 500 &  & 99.9 & 100.0 & 100.0 &  & 93.4 & 99.8 & 100.0 &  & 83.1 & 95.4 &
100.0\\
&  & 1000 &  & 100.0 & 100.0 & 100.0 &  & 94.4 & 100.0 & 100.0 &  & 85.5 &
97.7 & 100.0\\[0.5em]%
$CD^{\ast}$ &  & 100 &  & 23.8 & 35.2 & 60.0 &  & 31.3 & 49.3 & 78.9 &  &
83.7 & 97.8 & 100.0\\
&  & 200 &  & 22.4 & 33.1 & 60.6 &  & 32.1 & 49.4 & 81.8 &  & 86.4 & 98.2 &
100.0\\
&  & 500 &  & 21.1 & 35.1 & 64.6 &  & 33.5 & 51.3 & 83.4 &  & 88.4 & 98.6 &
100.0\\
&  & 1000 &  & 23.5 & 34.2 & 63.1 &  & 36.2 & 52.8 & 85.2 &  & 88.6 & 98.7 &
100.0\\[0.5em]%
$CD_{W+}$ &  & 100 &  & 5.4 & 9.4 & 38.1 &  & 6.4 & 9.3 & 48.1 &  & 7.1 &
10.1 & 63.8\\
&  & 200 &  & 6.0 & 6.1 & 44.7 &  & 6.4 & 6.9 & 46.9 &  & 6.9 & 7.3 & 56.7\\
&  & 500 &  & 5.9 & 6.6 & 49.5 &  & 5.7 & 5.8 & 50.1 &  & 6.3 & 5.0 & 57.3\\
&  & 1000 &  & 5.0 & 5.5 & 50.1 &  & 4.9 & 5.2 & 50.1 &  & 6.7 & 6.3 &
53.4\\\hline
\end{tabular}
}
\end{center}
\par
{\footnotesize \textit{Notes}: The DGP is given by (\ref{dgpy}) with
$\beta_{i1}=\beta_{i2}=0$, and contains two latent factors with different
factor strengths, $(\alpha_{1},\alpha_{2})=(1,1)$, $(1,2/3)$, and $(2/3,1/2)$.
See also the notes to Table \ref{table.n.1f}.}\end{table}

\begin{table}[th!]
\caption{Size and power of tests of error cross-sectional dependence using
four PCs $(\hat{m}=4)$ for the latent factor model with two factors
$(m_{0}=2)$ and serially independent Gaussian errors }%
\label{table.n.2f.b}%
\renewcommand{\thetable}{2B}
\par
\begin{center}
{\footnotesize
\begin{tabular}
[c]{crrrrrrrrrrrrrr}\hline\hline
\multicolumn{1}{l}{} &  &  &  & \multicolumn{3}{c}{$\alpha_{1}=1,\alpha_{2}%
=1$} &  & \multicolumn{3}{c}{$\alpha_{1}=1,\alpha_{2}=2/3$} &  &
\multicolumn{3}{c}{$\alpha_{1}=2/3,\alpha_{2}=1/2$}\\\cline{5-7}%
\cline{9-11}\cline{13-15}%
\multicolumn{1}{c}{Tests} &  & $n\setminus T$ &  & 100 & 200 & 500 &  & 100 &
200 & 500 &  & 100 & 200 & 500\\\hline
&  &  &  &  &  &  &  &  &  &  &  &  &  & \\
\multicolumn{1}{l}{} &  & \multicolumn{13}{c}{Size ($H_{o}:\lambda=0$%
)}\\[0.5em]%
$CD$ &  & 100 &  & 99.8 & 100.0 & 100.0 &  & 98.8 & 100.0 & 100.0 &  & 7.4 &
15.9 & 40.8\\
&  & 200 &  & 100.0 & 100.0 & 100.0 &  & 99.6 & 100.0 & 100.0 &  & 5.7 & 8.9 &
24.3\\
&  & 500 &  & 100.0 & 100.0 & 100.0 &  & 100.0 & 100.0 & 100.0 &  & 5.9 &
5.7 & 12.8\\
&  & 1000 &  & 100.0 & 100.0 & 100.0 &  & 100.0 & 100.0 & 100.0 &  & 7.4 &
4.7 & 8.0\\[0.5em]%
$CD^{\ast}$ &  & 100 &  & 7.6 & 7.9 & 15.4 &  & 7.0 & 6.8 & 13.2 &  & 9.5 &
8.4 & 10.1\\
&  & 200 &  & 5.5 & 6.0 & 6.6 &  & 6.8 & 6.3 & 7.0 &  & 7.6 & 6.3 & 6.5\\
&  & 500 &  & 5.1 & 4.7 & 6.3 &  & 6.2 & 5.1 & 4.8 &  & 7.8 & 6.8 & 5.1\\
&  & 1000 &  & 5.8 & 5.0 & 4.7 &  & 5.4 & 5.0 & 4.8 &  & 8.6 & 5.7 &
5.7\\[0.5em]%
$CD_{W+}$ &  & 100 &  & 5.5 & 6.8 & 26.2 &  & 5.9 & 6.3 & 15.7 &  & 6.2 &
6.1 & 11.8\\
&  & 200 &  & 4.5 & 5.7 & 6.4 &  & 5.2 & 4.7 & 5.3 &  & 5.5 & 5.4 & 6.7\\
&  & 500 &  & 5.6 & 5.0 & 5.3 &  & 5.9 & 5.1 & 5.7 &  & 5.9 & 5.3 & 4.8\\
&  & 1000 &  & 4.6 & 5.8 & 4.9 &  & 5.3 & 4.9 & 4.6 &  & 5.9 & 5.5 & 5.9\\
&  &  &  &  &  &  &  &  &  &  &  &  &  & \\
\multicolumn{1}{l}{} &  & \multicolumn{13}{c}{Power ($H_{1}:\lambda=0.25$%
)}\\[0.5em]%
$CD$ &  & 100 &  & 99.4 & 100.0 & 100.0 &  & 93.6 & 98.8 & 99.8 &  & 39.9 &
43.0 & 55.4\\
&  & 200 &  & 99.8 & 100.0 & 100.0 &  & 95.4 & 99.5 & 100.0 &  & 61.7 & 73.3 &
87.8\\
&  & 500 &  & 99.9 & 100.0 & 100.0 &  & 94.6 & 99.9 & 100.0 &  & 79.2 & 93.4 &
99.6\\
&  & 1000 &  & 99.9 & 100.0 & 100.0 &  & 94.2 & 100.0 & 100.0 &  & 83.3 &
96.7 & 100.0\\[0.5em]%
$CD^{\ast}$ &  & 100 &  & 26.4 & 39.2 & 67.5 &  & 33.7 & 56.4 & 86.4 &  &
79.8 & 96.6 & 100.0\\
&  & 200 &  & 22.8 & 36.8 & 66.0 &  & 33.3 & 51.7 & 85.1 &  & 84.4 & 98.0 &
100.0\\
&  & 500 &  & 21.0 & 35.2 & 66.1 &  & 34.3 & 51.9 & 84.1 &  & 87.2 & 98.1 &
100.0\\
&  & 1000 &  & 23.4 & 34.5 & 63.4 &  & 35.9 & 52.4 & 84.9 &  & 87.3 & 98.7 &
100.0\\[0.5em]%
$CD_{W+}$ &  & 100 &  & 6.2 & 8.8 & 39.5 &  & 6.2 & 8.7 & 40.5 &  & 7.3 &
7.9 & 34.9\\
&  & 200 &  & 5.4 & 7.2 & 37.4 &  & 5.6 & 6.5 & 29.6 &  & 6.3 & 6.2 & 31.2\\
&  & 500 &  & 6.3 & 6.1 & 44.2 &  & 6.7 & 6.4 & 38.8 &  & 5.9 & 5.9 & 39.0\\
&  & 1000 &  & 5.0 & 6.6 & 45.6 &  & 4.8 & 5.0 & 42.4 &  & 6.1 & 5.7 &
43.7\\\hline
\end{tabular}
}
\end{center}
\par
{\footnotesize \textit{Notes}: See the notes to Table \ref{table.n.2f}.}\end{table}

\begin{table}[th!]
\caption{Size and power of tests of error cross-sectional dependence using one
PC $(\hat{m}=1)$ for the panel regression model with one latent factor
$(m_{0}=1)$ and serially independent Gaussian errors}%
\label{table.n.x1f}%
\renewcommand{\thetable}{3A} \hypertarget{table3}{}
\par
\begin{center}
{\footnotesize
\begin{tabular}
[c]{crrrrrrrrrrrrrr}\hline\hline
\multicolumn{1}{l}{} &  &  &  & \multicolumn{3}{c}{$\alpha=1$} &  &
\multicolumn{3}{c}{$\alpha=2/3$} &  & \multicolumn{3}{c}{$\alpha=1/2$%
}\\\cline{5-7}\cline{9-11}\cline{13-15}%
\multicolumn{1}{c}{Tests} &  & $n\setminus T$ &  & 100 & 200 & 500 &  & 100 &
200 & 500 &  & 100 & 200 & 500\\\hline
&  &  &  &  &  &  &  &  &  &  &  &  &  & \\
\multicolumn{1}{l}{} &  & \multicolumn{13}{c}{Size ($H_{o}:\lambda=0$%
)}\\[0.5em]%
$CD$ &  & 100 &  & 67.9 & 88.6 & 98.5 &  & 6.7 & 9.5 & 20.5 &  & 6.9 & 6.3 &
9.0\\
&  & 200 &  & 68.9 & 92.5 & 99.7 &  & 5.5 & 6.7 & 13.4 &  & 6.8 & 5.3 & 5.9\\
&  & 500 &  & 67.4 & 94.7 & 100.0 &  & 4.9 & 6.1 & 8.5 &  & 5.3 & 5.0 & 5.9\\
&  & 1000 &  & 69.0 & 95.1 & 100.0 &  & 5.9 & 4.3 & 6.6 &  & 6.8 & 5.8 &
5.6\\[0.5em]%
$CD^{\ast}$ &  & 100 &  & 5.1 & 5.6 & 4.6 &  & 7.4 & 5.5 & 6.1 &  & 7.8 &
6.1 & 6.2\\
&  & 200 &  & 5.5 & 6.0 & 4.2 &  & 6.2 & 5.8 & 4.6 &  & 6.8 & 5.7 & 4.9\\
&  & 500 &  & 5.0 & 5.0 & 4.5 &  & 5.5 & 5.7 & 5.5 &  & 5.5 & 5.3 & 5.5\\
&  & 1000 &  & 4.5 & 4.6 & 5.4 &  & 6.3 & 5.0 & 5.6 &  & 7.0 & 5.9 &
5.9\\[0.5em]%
$CD_{W+}$ &  & 100 &  & 5.4 & 5.3 & 6.4 &  & 5.3 & 5.7 & 6.2 &  & 4.5 & 5.9 &
6.4\\
&  & 200 &  & 5.8 & 5.1 & 5.2 &  & 5.3 & 4.9 & 5.3 &  & 6.1 & 6.4 & 4.8\\
&  & 500 &  & 5.1 & 5.4 & 5.9 &  & 5.6 & 4.5 & 5.1 &  & 5.6 & 6.4 & 5.7\\
&  & 1000 &  & 4.9 & 5.5 & 4.5 &  & 6.0 & 4.6 & 4.7 &  & 5.1 & 6.0 & 4.8\\
&  &  &  &  &  &  &  &  &  &  &  &  &  & \\
\multicolumn{1}{l}{} &  & \multicolumn{13}{c}{Power ($H_{1}:\lambda=0.25$%
)}\\[0.5em]%
$CD$ &  & 100 &  & 25.1 & 37.6 & 55.6 &  & 71.1 & 86.6 & 97.7 &  & 81.4 &
93.9 & 99.3\\
&  & 200 &  & 17.2 & 29.7 & 49.7 &  & 77.2 & 94.0 & 99.7 &  & 86.5 & 98.1 &
100.0\\
&  & 500 &  & 11.8 & 22.1 & 45.7 &  & 83.2 & 97.0 & 100.0 &  & 89.5 & 99.0 &
100.0\\
&  & 1000 &  & 10.1 & 19.7 & 44.8 &  & 85.8 & 98.3 & 100.0 &  & 89.8 & 99.0 &
100.0\\[0.5em]%
$CD^{\ast}$ &  & 100 &  & 57.5 & 81.3 & 98.2 &  & 85.9 & 98.0 & 100.0 &  &
88.9 & 98.8 & 100.0\\
&  & 200 &  & 58.9 & 83.0 & 99.1 &  & 85.8 & 98.3 & 100.0 &  & 89.2 & 99.2 &
100.0\\
&  & 500 &  & 59.0 & 83.7 & 99.3 &  & 87.4 & 98.6 & 100.0 &  & 90.6 & 99.4 &
100.0\\
&  & 1000 &  & 60.0 & 83.3 & 99.5 &  & 88.3 & 99.0 & 100.0 &  & 90.1 & 99.2 &
100.0\\[0.5em]%
$CD_{W+}$ &  & 100 &  & 5.9 & 8.6 & 46.4 &  & 5.7 & 8.9 & 60.1 &  & 6.2 &
8.8 & 64.6\\
&  & 200 &  & 6.0 & 6.6 & 51.8 &  & 6.0 & 6.3 & 60.2 &  & 6.2 & 7.6 & 59.3\\
&  & 500 &  & 5.4 & 6.7 & 55.0 &  & 5.9 & 6.1 & 57.4 &  & 6.2 & 7.0 & 58.0\\
&  & 1000 &  & 4.6 & 5.6 & 48.9 &  & 5.4 & 5.3 & 52.8 &  & 4.8 & 5.9 &
51.9\\\hline
\end{tabular}
}
\end{center}
\par
{\footnotesize \textit{Notes}: The DGP is given by (\ref{dgpy}) with
$\beta_{i1}$ and $\beta_{i2}$ both generated from normal distribution, and
contains a single latent factor with different factor strengths, $\alpha=1,$
$2/3$, and $1/2$. See also the notes to Table \ref{table.n.1f}.}\end{table}

\begin{table}[th!]
\caption{Size and power of tests of error cross-sectional dependence using two
PCs $(\hat{m}=2)$ for the panel regression model with one latent factor
$(m_{0}=1)$ and serially independent Gaussian errors}%
\label{table.n.x1f.b}%
\renewcommand{\thetable}{3B}
\par
\begin{center}
{\footnotesize
\begin{tabular}
[c]{crrrrrrrrrrrrrr}\hline\hline
\multicolumn{1}{l}{} &  &  &  & \multicolumn{3}{c}{$\alpha=1$} &  &
\multicolumn{3}{c}{$\alpha=2/3$} &  & \multicolumn{3}{c}{$\alpha=1/2$%
}\\\cline{5-7}\cline{9-11}\cline{13-15}%
\multicolumn{1}{c}{Tests} &  & $n\setminus T$ &  & 100 & 200 & 500 &  & 100 &
200 & 500 &  & 100 & 200 & 500\\\hline
&  &  &  &  &  &  &  &  &  &  &  &  &  & \\
\multicolumn{1}{l}{} &  & \multicolumn{13}{c}{Size ($H_{o}:\lambda=0$)}\\
$CD$ &  & 100 &  & 67.5 & 88.9 & 98.6 &  & 7.7 & 9.3 & 21.8 &  & 6.8 & 6.8 &
9.4\\
&  & 200 &  & 69.0 & 92.2 & 99.5 &  & 5.6 & 6.7 & 13.3 &  & 6.8 & 5.4 & 5.8\\
&  & 500 &  & 68.0 & 94.9 & 100.0 &  & 5.0 & 5.7 & 8.0 &  & 5.1 & 5.1 & 5.5\\
&  & 1000 &  & 69.9 & 95.0 & 100.0 &  & 5.9 & 4.2 & 6.7 &  & 6.7 & 5.7 &
5.7\\[0.5em]%
$CD^{\ast}$ &  & 100 &  & 5.6 & 6.5 & 6.4 &  & 7.9 & 5.8 & 7.1 &  & 8.0 &
6.6 & 6.7\\
&  & 200 &  & 5.8 & 5.7 & 4.9 &  & 6.2 & 6.9 & 5.7 &  & 7.0 & 5.6 & 4.7\\
&  & 500 &  & 5.5 & 5.0 & 4.4 &  & 5.6 & 5.3 & 5.3 &  & 5.6 & 5.3 & 5.5\\
&  & 1000 &  & 4.6 & 4.7 & 5.3 &  & 6.3 & 5.1 & 5.6 &  & 7.0 & 6.2 &
6.1\\[0.5em]%
$CD_{W+}$ &  & 100 &  & 6.1 & 5.8 & 6.6 &  & 6.9 & 5.9 & 6.4 &  & 5.5 & 6.5 &
7.4\\
&  & 200 &  & 5.9 & 4.8 & 5.6 &  & 5.2 & 6.1 & 5.6 &  & 5.3 & 5.4 & 4.4\\
&  & 500 &  & 5.2 & 4.9 & 5.3 &  & 5.1 & 5.1 & 5.7 &  & 5.7 & 5.1 & 5.3\\
&  & 1000 &  & 5.1 & 5.1 & 4.4 &  & 5.1 & 4.7 & 5.6 &  & 6.2 & 4.8 & 4.9\\
&  &  &  &  &  &  &  &  &  &  &  &  &  & \\
\multicolumn{1}{l}{} &  & \multicolumn{13}{c}{Power ($H_{1}:\lambda=0.25$)}\\
$CD$ &  & 100 &  & 28.3 & 43.6 & 63.0 &  & 59.0 & 75.0 & 86.3 &  & 69.5 &
84.5 & 90.9\\
&  & 200 &  & 18.7 & 32.3 & 55.7 &  & 71.5 & 89.4 & 98.3 &  & 80.8 & 96.1 &
99.5\\
&  & 500 &  & 12.4 & 23.7 & 47.6 &  & 80.9 & 96.0 & 100.0 &  & 87.6 & 98.6 &
100.0\\
&  & 1000 &  & 11.0 & 20.4 & 46.2 &  & 84.1 & 97.9 & 100.0 &  & 88.4 & 98.9 &
100.0\\[0.5em]%
$CD^{\ast}$ &  & 100 &  & 57.0 & 81.0 & 98.4 &  & 83.3 & 97.8 & 100.0 &  &
86.6 & 98.6 & 100.0\\
&  & 200 &  & 58.2 & 82.6 & 98.9 &  & 84.4 & 98.1 & 100.0 &  & 87.0 & 99.0 &
100.0\\
&  & 500 &  & 58.0 & 83.4 & 99.4 &  & 86.5 & 98.7 & 100.0 &  & 89.1 & 99.3 &
100.0\\
&  & 1000 &  & 58.5 & 83.0 & 99.4 &  & 87.1 & 98.9 & 100.0 &  & 89.4 & 99.2 &
100.0\\[0.5em]%
$CD_{W+}$ &  & 100 &  & 7.6 & 7.9 & 32.6 &  & 8.0 & 8.0 & 36.4 &  & 7.2 &
7.8 & 41.4\\
&  & 200 &  & 5.8 & 6.1 & 39.4 &  & 6.1 & 6.2 & 43.4 &  & 5.8 & 6.8 & 42.8\\
&  & 500 &  & 5.1 & 5.6 & 48.6 &  & 5.4 & 6.5 & 47.9 &  & 5.9 & 6.3 & 47.6\\
&  & 1000 &  & 5.1 & 5.7 & 47.7 &  & 5.9 & 4.6 & 49.9 &  & 5.6 & 5.3 &
48.8\\\hline
\end{tabular}
}
\end{center}
\par
{\footnotesize \textit{Notes}: See the notes to Table \ref{table.n.x1f}.}\end{table}

\begin{table}[th!]
\caption{Size and power of tests of error cross-sectional dependence using two
PCs $(\hat{m}=2)$ for the panel regression model with two latent factors
$(m_{0}=2)$ and serially independent Gaussian errors}%
\label{table.n.x2f}%
\renewcommand{\thetable}{4A}
\par
\begin{center}
{\footnotesize
\begin{tabular}
[c]{crrrrrrrrrrrrrr}\hline\hline
\multicolumn{1}{l}{} &  &  &  & \multicolumn{3}{c}{$\alpha_{1}=1,\alpha_{2}%
=1$} &  & \multicolumn{3}{c}{$\alpha_{1}=1,\alpha_{2}=2/3$} &  &
\multicolumn{3}{c}{$\alpha_{1}=2/3,\alpha_{2}=1/2$}\\\cline{5-7}%
\cline{9-11}\cline{13-15}%
\multicolumn{1}{c}{Tests} &  & $n\setminus T$ &  & 100 & 200 & 500 &  & 100 &
200 & 500 &  & 100 & 200 & 500\\\hline
&  &  &  &  &  &  &  &  &  &  &  &  &  & \\
\multicolumn{1}{l}{} &  & \multicolumn{13}{c}{Size ($H_{o}:\lambda=0$%
)}\\[0.5em]%
$CD$ &  & 100 &  & 100.0 & 100.0 & 100.0 &  & 97.8 & 100.0 & 100.0 &  & 7.4 &
13.9 & 40.1\\
&  & 200 &  & 100.0 & 100.0 & 100.0 &  & 99.3 & 100.0 & 100.0 &  & 5.3 & 8.2 &
23.8\\
&  & 500 &  & 100.0 & 100.0 & 100.0 &  & 99.4 & 100.0 & 100.0 &  & 6.2 & 5.4 &
11.0\\
&  & 1000 &  & 100.0 & 100.0 & 100.0 &  & 99.9 & 100.0 & 100.0 &  & 7.8 &
5.9 & 7.0\\[0.5em]%
$CD^{\ast}$ &  & 100 &  & 5.8 & 5.3 & 5.3 &  & 5.9 & 5.8 & 5.4 &  & 9.4 &
6.4 & 5.9\\
&  & 200 &  & 6.0 & 5.4 & 5.4 &  & 5.3 & 4.9 & 4.8 &  & 7.4 & 6.9 & 5.7\\
&  & 500 &  & 5.7 & 5.4 & 5.5 &  & 6.6 & 4.3 & 4.9 &  & 7.9 & 6.1 & 5.7\\
&  & 1000 &  & 5.1 & 5.4 & 5.7 &  & 5.1 & 6.3 & 4.8 &  & 9.4 & 7.0 &
5.8\\[0.5em]%
$CD_{W+}$ &  & 100 &  & 5.6 & 4.9 & 6.8 &  & 5.5 & 5.5 & 8.1 &  & 6.4 & 6.3 &
8.6\\
&  & 200 &  & 6.2 & 5.9 & 5.6 &  & 5.6 & 5.0 & 5.1 &  & 5.8 & 4.6 & 5.8\\
&  & 500 &  & 6.4 & 6.3 & 5.5 &  & 5.5 & 5.4 & 4.3 &  & 5.5 & 5.3 & 4.9\\
&  & 1000 &  & 5.2 & 5.1 & 4.7 &  & 6.6 & 5.1 & 4.6 &  & 6.0 & 6.4 & 4.3\\
&  &  &  &  &  &  &  &  &  &  &  &  &  & \\
\multicolumn{1}{l}{} &  & \multicolumn{13}{c}{Power ($H_{1}:\lambda=0.25$%
)}\\[0.5em]%
$CD$ &  & 100 &  & 99.2 & 99.9 & 100.0 &  & 89.6 & 96.9 & 99.5 &  & 56.1 &
66.1 & 81.6\\
&  & 200 &  & 99.5 & 100.0 & 100.0 &  & 91.8 & 99.2 & 100.0 &  & 71.6 & 86.3 &
98.1\\
&  & 500 &  & 100.0 & 100.0 & 100.0 &  & 92.2 & 99.8 & 100.0 &  & 83.2 &
95.3 & 100.0\\
&  & 1000 &  & 100.0 & 100.0 & 100.0 &  & 94.9 & 99.8 & 100.0 &  & 86.5 &
97.7 & 100.0\\[0.5em]%
$CD^{\ast}$ &  & 100 &  & 23.1 & 34.2 & 63.5 &  & 33.7 & 50.1 & 79.6 &  &
83.7 & 97.7 & 100.0\\
&  & 200 &  & 23.0 & 36.3 & 64.6 &  & 33.9 & 52.2 & 82.0 &  & 86.3 & 97.8 &
100.0\\
&  & 500 &  & 22.3 & 35.7 & 65.1 &  & 35.0 & 52.6 & 82.8 &  & 88.9 & 98.8 &
100.0\\
&  & 1000 &  & 23.5 & 34.6 & 64.8 &  & 35.2 & 53.6 & 85.2 &  & 89.9 & 99.1 &
100.0\\[0.5em]%
$CD_{W+}$ &  & 100 &  & 6.8 & 7.9 & 38.7 &  & 6.5 & 8.1 & 44.4 &  & 8.9 &
10.7 & 62.1\\
&  & 200 &  & 6.1 & 7.9 & 44.8 &  & 6.1 & 6.9 & 47.9 &  & 7.0 & 6.6 & 57.9\\
&  & 500 &  & 6.3 & 6.6 & 49.4 &  & 5.7 & 6.4 & 50.1 &  & 5.9 & 6.2 & 55.9\\
&  & 1000 &  & 5.3 & 5.6 & 49.9 &  & 6.8 & 6.3 & 49.8 &  & 6.2 & 7.2 &
49.9\\\hline
\end{tabular}
}
\end{center}
\par
{\footnotesize \textit{Notes}: The DGP is given by (\ref{dgpy}) with
$\beta_{i1}$ and $\beta_{i2}$ both generated from normal distribution, and
contains two latent factors with different factor strengths, $(\alpha
_{1},\alpha_{2})=(1,1),\text{ }(1,2/3),\text{ and }(2/3,1/2)$. See also the
notes to Table \ref{table.n.1f}.}\end{table}

\begin{table}[th!]
\caption{Size and power of tests of error cross-sectional dependence using
four PCs $(\hat{m}=4)$ for the panel regression model with two latent factors
$(m_{0}=2)$ and serially independent Gaussian errors}%
\label{table.n.x2f.b}%
\renewcommand{\thetable}{4B}
\par
\begin{center}
{\footnotesize
\begin{tabular}
[c]{crrrrrrrrrrrrrr}\hline\hline
\multicolumn{1}{l}{} &  &  &  & \multicolumn{3}{c}{$\alpha_{1}=1,\alpha_{2}%
=1$} &  & \multicolumn{3}{c}{$\alpha_{1}=1,\alpha_{2}=2/3$} &  &
\multicolumn{3}{c}{$\alpha_{1}=2/3,\alpha_{2}=1/2$}\\\cline{5-7}%
\cline{9-11}\cline{13-15}%
\multicolumn{1}{c}{Tests} &  & $n\setminus T$ &  & 100 & 200 & 500 &  & 100 &
200 & 500 &  & 100 & 200 & 500\\\hline
&  &  &  &  &  &  &  &  &  &  &  &  &  & \\
\multicolumn{1}{l}{} &  & \multicolumn{13}{c}{Size ($H_{o}:\lambda=0$%
)}\\[0.5em]%
$CD$ &  & 100 &  & 100.0 & 100.0 & 100.0 &  & 97.9 & 100.0 & 100.0 &  & 7.4 &
15.4 & 40.8\\
&  & 200 &  & 100.0 & 100.0 & 100.0 &  & 99.3 & 100.0 & 100.0 &  & 5.9 & 9.4 &
23.5\\
&  & 500 &  & 100.0 & 100.0 & 100.0 &  & 99.5 & 100.0 & 100.0 &  & 6.4 & 5.6 &
10.7\\
&  & 1000 &  & 100.0 & 100.0 & 100.0 &  & 99.8 & 100.0 & 100.0 &  & 7.0 &
6.3 & 7.0\\[0.5em]%
$CD^{\ast}$ &  & 100 &  & 7.9 & 8.1 & 16.0 &  & 7.4 & 8.6 & 13.7 &  & 10.3 &
9.6 & 10.4\\
&  & 200 &  & 6.1 & 6.6 & 7.9 &  & 6.6 & 6.3 & 6.6 &  & 8.6 & 7.7 & 6.7\\
&  & 500 &  & 5.7 & 5.7 & 6.2 &  & 7.0 & 4.8 & 4.7 &  & 8.0 & 6.6 & 5.8\\
&  & 1000 &  & 5.0 & 5.1 & 5.5 &  & 5.8 & 6.8 & 5.0 &  & 8.8 & 7.1 &
6.0\\[0.5em]%
$CD_{W+}$ &  & 100 &  & 6.3 & 7.1 & 25.2 &  & 5.6 & 6.3 & 17.0 &  & 6.3 &
5.8 & 10.1\\
&  & 200 &  & 6.1 & 5.3 & 6.1 &  & 5.7 & 5.9 & 5.4 &  & 6.6 & 5.8 & 6.0\\
&  & 500 &  & 5.9 & 5.7 & 4.7 &  & 6.1 & 6.3 & 5.3 &  & 5.9 & 6.1 & 5.3\\
&  & 1000 &  & 6.1 & 5.4 & 5.1 &  & 5.9 & 4.8 & 5.3 &  & 6.3 & 4.5 & 4.9\\
&  &  &  &  &  &  &  &  &  &  &  &  &  & \\
\multicolumn{1}{l}{} &  & \multicolumn{13}{c}{Power ($H_{1}:\lambda=0.25$%
)}\\[0.5em]%
$CD$ &  & 100 &  & 99.3 & 100.0 & 100.0 &  & 92.3 & 98.8 & 99.9 &  & 37.6 &
43.8 & 54.7\\
&  & 200 &  & 99.6 & 100.0 & 100.0 &  & 93.4 & 99.5 & 100.0 &  & 61.3 & 74.7 &
89.3\\
&  & 500 &  & 100.0 & 100.0 & 100.0 &  & 93.7 & 99.8 & 100.0 &  & 77.8 &
93.1 & 99.8\\
&  & 1000 &  & 100.0 & 100.0 & 100.0 &  & 94.8 & 99.8 & 100.0 &  & 84.3 &
96.8 & 100.0\\[0.5em]%
$CD^{\ast}$ &  & 100 &  & 27.6 & 43.1 & 77.2 &  & 37.2 & 56.5 & 86.6 &  &
79.2 & 95.9 & 100.0\\
&  & 200 &  & 26.1 & 41.3 & 70.8 &  & 34.6 & 54.7 & 85.9 &  & 82.7 & 97.2 &
100.0\\
&  & 500 &  & 22.7 & 36.4 & 67.4 &  & 34.9 & 51.9 & 83.4 &  & 87.3 & 98.3 &
100.0\\
&  & 1000 &  & 24.2 & 35.0 & 66.3 &  & 34.8 & 52.8 & 85.9 &  & 87.9 & 98.8 &
100.0\\[0.5em]%
$CD_{W+}$ &  & 100 &  & 6.7 & 7.5 & 47.8 &  & 5.1 & 8.6 & 42.3 &  & 7.0 &
7.4 & 33.4\\
&  & 200 &  & 5.9 & 6.4 & 31.8 &  & 6.3 & 7.4 & 28.5 &  & 7.7 & 6.9 & 28.2\\
&  & 500 &  & 6.2 & 6.1 & 38.5 &  & 6.3 & 6.4 & 38.8 &  & 6.5 & 7.3 & 39.0\\
&  & 1000 &  & 6.6 & 6.2 & 42.5 &  & 5.6 & 5.5 & 42.9 &  & 6.5 & 5.6 &
44.1\\\hline
\end{tabular}
}
\end{center}
\par
{\footnotesize \textit{Notes}: See the notes to Table \ref{table.n.x2f}.}\end{table}

\subsubsection{Serially independent errors: chi-squared distributed errors}

To save space, the simulation results for the DGPs with chi-squared errors are
provided in Tables \ref{table.chi2.1f} to \ref{table.chi2.x2f} in the
supplement. For the standard CD test and its biased-corrected
version, CD$^{\ast}$, as shown in Tables \ref{table.chi2.1f} and
\ref{table.chi2.2f}, the results are very similar to the ones with Gaussian
errors, suggesting that the CD$^{\text{*}}$ test is likely to be robust to
departures from Gaussianity. As with the experiments with Gaussian errors, the
standard CD test continues to over-reject unless $\alpha<2/3$, and the
CD$^{\text{*}}$ test has the correct size for all $n$ and $T$ combinations,
except when the number of selected PCs is large relative to $m_{0}$, and
$T=100$. The main difference between the results with and without Gaussian
errors is the tendency for the CD$_{W^{+}}$ test to over-reject when $n>T$,
which seems to be a universal feature of this test and holds for \textit{all}
choices of $m_{0}$ and the number of selected PCs, irrespective of whether the
factors are strong or weak. This could be due to the screening component of
the CD$_{W+}$ test not tending to zero sufficiently fast with $n$ and $T$.
Furthermore, the CD$^{\text{*}}$ test continues to have satisfactory power,
but the CD$_{W^{+}}$ test clearly lacks power against spatial or network
alternatives that are of primary interest.

Similar results are obtained for panel regression models with latent factors,
summarized in Tables \ref{table.chi2.x1f} and \ref{table.chi2.x2f} in the
supplement.

\subsubsection{Serially correlated errors}

To save space, the results for the DGPs with serially correlated errors are
summarized in Tables \ref{table.v.n.1f} to \ref{table.ardl.chi2.x2f} in the
supplement. Tables \ref{table.v.n.1f} to \ref{table.v.chi2.x2f} give
the simulation results for the variance adjusted CD tests, whilst Tables
\ref{table.ardl.n.1f} to \ref{table.ardl.chi2.x2f} provide the results for the
ARDL adjusted tests. Overall, the results corroborate our earlier findings
obtained for DGPs with serially independent errors. Both adjustments for
serial error correlation work well, with size and power of the adjusted
CD$^{\text{*}}$ tests being quite close to the results already reported for
DGPs with serially independent errors. It is also clear that without
adjustments for latent factors and error serial correlation, the standard CD
test will lead to large size distortions when the latent factors are strong.
But in line with our theoretical results, the standard CD test, when adjusted
for error serial correlation if needed, tends to have the correct size when
the latent factors are weak.

Comparing the two types of adjustments for error serial correlations (for pure
latent factor models as well as for panel regression models with latent
factors), the variance adjusted CD$^{\text{*}}$ test works particularly well,
and only shows mild over-rejection in the case where $T=100$ and $n>T$. In
contrast, the CD$_{W+}$ test with variance adjustment over-rejects for all
combinations of $n$ and $T$.

The ARDL adjusted version of the CD$^{\text{*}}$ test also works well when the
number of PCs is not too large, and tends to have the correct size for all
$(n,T)$ combinations and only shows slight over-rejection when $n=100$. The
CD$_{W+}$ test using ARDL adjustment does better in controlling for the size
when the errors are Gaussian, but tends to over-reject when the errors are
chi-squared distributed and $n>T$. Both adjusted versions of the CD$_{W+}$
test continue to lack power against spatial or network alternatives.

\section{Empirical application\label{app}}

It is well known that house price changes are spatially correlated, but it is
unclear if such correlations are mainly due to common factors (national or
regional) or arise from spatial spillover effects not related to the common
factors, a phenomenon also referred to as the ripple effect. See, for example,
\cite{holly2011spatial}, \cite{tsai2015spillover}, \cite{chiang2016ripple},
\cite{bailey2016two}, and \cite{aquaro2021estimation}. To test for the
presence of ripple effects the influence of common factors must first be
filtered out and this is often a challenging exercise due to the latent nature
of regional and national factors. Therefore, to find if there exist local
spillover effects, one needs to test for significant residual cross-sectional
dependence once the effects of common factors are filtered out.

We consider quarterly data on real house prices at the level of Metropolitan
Statistical Areas (MSAs) in the U.S. There are 381 MSAs, under the February
2013 definition provided by the U.S. Office of Management and Budget (OMB). We
use quarterly data on real house price changes compiled by
\cite{yang2021common} which covers $n=377$ MSAs from the contiguous United
States over the period 1975Q1-2014Q4 ($T=160$ quarters). To allow for possible
regional factors, we also follow \cite{bailey2016two} and start with the
Bureau of Economic Analysis eight regional classification, namely New England,
Mideast, Great Lakes, Plains, Southeast, Southwest, Rocky Mountain and Far
West. But due to the low number of MSAs in New England and Rocky Mountain
regions, we combine New England and Mideast, and Southwest and Rocky Mountain
as two regions. We end up with a six region classification ($R=6$), each
covering a reasonable number of MSAs.

We model house price changes and consider an extended factor model with
deterministic seasonal dummies to allow for seasonal movements in house
prices. \cite{bailey2016two} find evidence of regional factors in U.S. house
price changes which might not be picked up when using PCA. Given this finding,
our model includes observed regional and national factors, as well as latent
factors. Specifically, we suppose%
\begin{equation}
\pi_{irt}=a_{ir}+\sum_{j=1}^{3}\beta_{ir,j}\text{l}\left\{  q_{t}=j\right\}
+\delta_{ir,1}\bar{\pi}_{rt}+\delta_{ir,2}\bar{\pi}_{t}+\boldsymbol{\gamma
}_{ir}^{^{\prime}}\mathbf{f}_{t}+u_{irt}, \label{hp_mod2}%
\end{equation}
where $\pi_{irt}$ is the real house price change in MSA $i$ located in region
$r=1,2,\ldots,R$, l$\left\{  q_{t}=j\right\}  $ is the index for quarter $j$, and
$\mathbf{f}_{t}$ is the $m_{0}\times1$ vector of latent factors. $\bar{\pi
}_{rt}=n_{r}^{-1}\sum_{i=1}^{n_{r}}\pi_{irt}$, where $n=\sum_{r=1}^Rn_r$, and $n_r$ is the number of MSAs in region $r$, and $\bar{\pi}_{t}=n^{-1}%
\sum_{r=1}^{R}\sum_{i=1}^{n_{r}}\pi_{irt}$ are proxies for the regional and
national factors. To filter out the effects of seasonal dummies as well as
observed factors, we first run the least squares regression of $\pi_{irt}$ on
an intercept and $\left(  \text{l}\left\{  q_{t}=j\right\}  ,\bar{\pi}%
_{rt},\bar{\pi}_{t}\right)  $ for each $i$ to generate the residuals%
\begin{equation}
\hat{v}_{irt}=\pi_{irt}-\hat{a}_{ir}-\sum_{j=1}^{3}\hat{\beta}_{ir,j}%
\text{l}\left\{  q_{t}=j\right\}  -\hat{\delta}_{ir,1}\bar{\pi}_{rt}%
-\hat{\delta}_{ir,2}\bar{\pi}_{t}, \label{hp2}%
\end{equation}
and then apply PCA to $\left\{  \hat{v}_{irt}:i=1,2,\ldots,n_r,r=1,2,\ldots
,R,t=1,2,\ldots,T\right\}  $ to obtain $\boldsymbol{\hat{\gamma}}_{ir}$ and
$\mathbf{\hat{f}}_{t}$, yielding the residuals%
\begin{equation}
\hat{u}_{irt}=\pi_{irt}-\hat{a}_{ir}-\sum_{j=1}^{3}\hat{\beta}_{ir,j}%
\text{l}\left\{  q_{t}=j\right\}  -\hat{\delta}_{ir,1}\bar{\pi}_{rt}%
-\hat{\delta}_{ir,2}\bar{\pi}_{t}-\boldsymbol{\hat{\gamma}}_{ir}^{^{\prime}%
}\mathbf{\hat{f}}_{t}. \label{hp3}%
\end{equation}
For the case without adjusting for error serial correlation, the above
residuals are used to compute $CD$, $CD^{\ast}$ and $CD_{W+}$, given by
(\ref{cd}), (\ref{CD*a}) and (\ref{CDW+}). For the case with serially
correlated errors, the variance adjusted versions of the three CD statistics
are generated by scaling original test statistics using (\ref{var_of_CD}),
where $\tilde{\varepsilon}_{it,T}$ is replaced by the standardized residuals
generated from (\ref{hp3}), while the ARDL adjusted versions are computed
using the residuals from the following dynamic panel data model with latent
factors, $\mathbf{h}_{t}$:
\begin{align}
\pi_{irt}  &  =a_{ir}+\rho_{ir}\pi_{irt-1}+\sum_{j=1}^{3}\beta_{ir,j}%
\text{l}\left\{  q_{t}=j\right\}  +\delta_{ir,1}\bar{\pi}_{rt}+\delta
_{ir,2}\bar{\pi}_{t}\nonumber\\
&  +\sum_{j=1}^{3}\lambda_{ir,j}\text{l}\left\{  q_{t-1}=j\right\}
+\omega_{ir,1}\bar{\pi}_{rt-1}+\omega_{ir,2}\bar{\pi}_{t-1}+\mathbf{g}%
_{ir}^{\prime}\mathbf{h}_{t}+\epsilon_{irt}. \label{hp_mod4}%
\end{align}
To estimate the number of latent factors, $m$, we consider the information
criteria $IC_{P1}$ and $IC_{P2}$ proposed by \cite{bai2002determining}, and
the $ER$ and $GR$ criteria proposed by \cite{ahn2013eigenvalue}. Given the
spatial diversity of U.S. housing market, we set $m_{\max}=$ $10$, although
once we allow for national and regional factors we would expect $m_{0}$ and
its estimate, $\hat{m}$, to be relatively small.\textbf{ }The estimated number
of factors and the associated CD test statistics are summarized in Table
\ref{table.app.1}. The first four columns of the table report $\hat{m}$, $CD$,
$CD^{\ast}$ and $CD_{W+}$ statistics that are not adjusted for error serial
correlation, whilst the middle and the final four columns report the variance adjusted and ARDL adjusted versions of these statistics, respectively.

As can be seen in the case of no error serial correlations, there are large
differences in the number of factors selected by the different criteria, with
$IC_{p1}$ selecting the assumed maximum number of factors, $IC_{p2}$ selecting
$4$, and $ER$ and $GR$ both selecting $2$ factors. These estimates are not
affected when we allow for error serial correlations and consider variance
adjustment. 

CD, CD$^{\ast}$ and CD$_{W+}$ tests all reject the null hypothesis of
cross-sectional independence, irrespective of the choices of $\hat{m}$ and
whether we allow for error serial correlation. In view of the theoretical and
finite sample results reported in this paper, it is advisable to focus on the
CD$^{\text{*}}$ test results and recognize that the relatively large
magnitudes obtained for CD and CD$_{W+}$ test statistics could be due to their
tendencies to over-rejection in the presence of strong latent factors and
non-Gaussian errors. Focusing on $CD^{\ast}$, we find that even with $\hat
{m}=10$ the CD$^{\text{*}}$ test strongly rejects the null of cross-sectional
independence with $CD^{\ast}$ statistic of $25.5$, compared to 95 per cent
critical value of $1.96$. There is clear evidence that in addition to latent
factors, spatial modeling of the type carried out in \citet{bailey2016two} and
\cite{aquaro2021estimation} is likely to be necessary to account for the
remaining error cross-sectional dependence.

\begin{table}[t]
\caption{Tests of error cross-sectional dependence for the house price
application}%
\label{table.app.1}
\renewcommand{\thetable}{5} \renewcommand{\thetable}{5}
\par
\begin{center}
{\small
\begin{tabular}
[c]{lrrrrrrrrrrrrrrrr}\hline\hline
&  & \multicolumn{4}{c}{Not adjusted for} &  & \multicolumn{9}{c}{Adjusted for
error serial correlation} & \\\cline{8-16}
&  & \multicolumn{4}{c}{error serial correlation} &  &
\multicolumn{4}{c}{Variance adjusted} &  & \multicolumn{4}{c}{ARDL adjusted} &
\\\cline{3-6}\cline{8-11}\cline{13-16}
&  & $\hat{m}$ & $CD$ & $CD^{\ast}$ & $CD_{W+}$ &  & $\hat{m}$ & $CD$ &
$CD^{\ast}$ & $CD_{W+}$ &  & $\hat{m}$ & $CD$ & $CD^{\ast}$ & $CD_{W+}$ & \\
$IC_{p1}$ &  & 10 & 22.2 & 53.8 & 1080.9 &  & 10 & 10.5 & 25.5 & 512.2 &  &
10 & 59.1 & 64.7 & 1171.9 & \\
$IC_{p2}$ &  & 4 & 112.4 & 122.4 & 1444.8 &  & 4 & 50.7 & 55.3 & 652.3 &  &
2 & 64.7 & 66.3 & 1598.4 & \\
$ER$ &  & 2 & 107.9 & 117.9 & 1588.4 &  & 2 & 47.3 & 51.8 & 697.1 &  & 1 &
60.3 & 61.4 & 1766.8 & \\
$GR$ &  & 2 & 107.9 & 117.9 & 1588.4 &  & 2 & 47.3 & 51.8 & 697.1 &  & 1 &
60.3 & 61.4 & 1766.8 & \\\hline
\end{tabular}
}
\end{center}
\par
{\small
}
\par
{\small {\footnotesize \textit{Notes}: The test statistics that do not adjust
for serial correlation and that use variance adjustment are based on the panel
regression in (\ref{hp_mod2}). The test statistics that use ARDL adjustment
are based on the panel regression in (\ref{hp_mod4}). Both panel regressions
allow for seasonal dummies and national and regional effects. $IC_{p1}$ and
$IC_{p2}$ denote the two information criteria by \cite{bai2002determining},
while $ER$ and $GR$ refer to the two criteria proposed by
\cite{ahn2013eigenvalue}. The maximum number of latent factors, }}$m_{\max}$,
{\small {\footnotesize is set as }}$10${\small {\footnotesize . The number of
selected factors is denoted by $\hat{m}$. $CD$ denotes the standard test of
error cross-sectional dependence defined by (\ref{cd}), $CD^{\ast}$ is the
bias-corrected version defined by (\ref{CD*a}), and $CD_{W+}$ is the
power-enhanced randomized version defined by (\ref{CDW+}). }}\end{table}

\section{Concluding remarks\label{conclusion}}

This paper revisits the problem of testing error cross-sectional independence
in panel data models with latent factors. Starting with a pure latent
multi-factor model we show that the standard CD test proposed by
\cite{pesaran2004general} remains valid if the latent factors are weak, but
over-reject when one or more of the latent factors are strong. The
over-rejection of the CD test in the case of strong factors is also
established by \cite{juodis2022incidental}, who propose a randomized test
statistic to correct for over-rejection and add a screening component to
achieve power. However, as we show, JR's CD$_{W^{+}}$ test is not guaranteed
to have the correct size and need not be powerful against spatial or network
alternatives. Such alternatives are of particular interest in the analyses of
ripple effects in housing markets, and clustering of firms within industries
in capital or arbitrage asset pricing models. In fact, using Monte Carlo
experiments we show that under non-Gaussian errors the JR test continues to
over-reject when the cross section dimension ($n$) is larger than the time
dimension ($T$), and often has power close to size against spatial
alternatives. To overcome some of these shortcomings, we propose a simple
bias-corrected CD test statistic, labeled $CD^{\ast}$, which is shown to be
asymptotically $\mathcal{N}(0,1)$ under the null when $n$ and $T\rightarrow
\infty$ such that $n/T\rightarrow\kappa$, for a fixed constant $\kappa$. In
addition, the CD$^{\text{*}}$ test is shown to have power against network type
dependence. These results hold for pure latent factor models as well as for
panel regression models with latent factors. To deal with possible error
serial dependence, following \cite{baltagi2016testing}, we also consider a
variance adjusted version of $CD^{\ast}$, as well as an alternative ARDL
adjusted version that eliminates the error serial dependence before the
application of the CD$^{\text{*}}$ test procedure. Both of these approaches
are shown to perform well within the Monte Carlo set up of the paper.

\bigskip\bigskip

\clearpage

\renewcommand\theequation{A.\arabic{equation}} \setcounter{equation}{0}
\renewcommand\thesection{A.\arabic{section}} \setcounter{section}{0}


\begin{center}
{\large \textbf{APPENDIX}}
\end{center}

\bigskip

\noindent In this appendix we provide proofs of the propositions and and theorems.
The auxiliary lemmas and the associated proofs are given in the supplement.

\bigskip

\section{Proof of Proposition \ref{Proposition:CD(s)}}

Here we provide a proof for part (b) of Proposition \ref{Proposition:CD(s)}.
The proof for part (a) follows trivially by setting $\lambda_{T}=0$. To this
end we first note that the $CD$ statistic given by (\ref{cd}) can be written
as (for a proof see Lemma \ref{lm1} of the supplement)%
\[
CD=\left(  \sqrt{\frac{n}{n-1}}\right)  \frac{1}{\sqrt{2T}}\sum_{t=1}%
^{T}\left[  \left(  \frac{1}{\sqrt{n}}\sum_{i=1}^{n}\frac{\hat{u}_{it}}%
{\hat{\sigma}_{i,T}}\right)  ^{2}-1\right]  ,
\]
where $\hat{u}_{it}$ is defined by (\ref{eit}). Using Lemma \ref{lmcdtcd} of
the supplement we also note that
\begin{equation}
CD=\widetilde{CD}+o_{p}(1), \label{CDCDL}%
\end{equation}
where%
\begin{equation}
\widetilde{CD}=\left(  \sqrt{\frac{n}{n-1}}\right)  \frac{1}{\sqrt{T}}%
\sum_{t=1}^{T}\left[  \frac{\left(  \frac{1}{\sqrt{n}}\sum_{i=1}^{n}\frac
{\hat{u}_{it}}{\omega_{i,T}}\right)  ^{2}-1}{\sqrt{2}}\right]  ,
\label{CDtilde1}%
\end{equation}
with $\omega_{i,T}=\left(  T^{-1}\sigma_{i}^{2}\boldsymbol{\varepsilon
}_{i\circ}^{\prime}\mathbf{M}_{F}\boldsymbol{\varepsilon}_{i\circ}\right)
^{1/2}$, $\boldsymbol{\varepsilon}_{i\circ}=\left(  \varepsilon_{i1}%
,\varepsilon_{i2},\ldots,\varepsilon_{iT}\right)  ^{\prime}$ and
$\mathbf{M}_{F}=\mathbf{I}_{T}-\mathbf{F}\left(  \mathbf{F}^{\prime}%
\mathbf{F}\right)  ^{-1}\mathbf{F}^{\prime}$. Also letting
\begin{equation}
\varepsilon_{it}\left(  \lambda_{T}\right)  =\varepsilon_{it}+\lambda
_{T}\mathbf{w}_{i0}^{\prime}\boldsymbol{\varepsilon}_{\circ t}, \label{eit:p}%
\end{equation}
where $\lambda_{T}=c_{\lambda}T^{-1/2}$ with $c_{\lambda}\neq0$,
$\mathbf{w}_{i0}=\left(  w_{i1},w_{i2},\ldots,w_{in}\right)  ^{\prime}$ and
$\boldsymbol{\varepsilon}_{\circ t}=\left(  \varepsilon_{1t},\varepsilon
_{2t},\ldots,\varepsilon_{nt}\right)  ^{\prime}$, under (\ref{mod2}) $\hat
{u}_{it}$ can now be expressed as
\begin{equation}
\hat{u}_{it}=\sigma_{i}\varepsilon_{it}\left(  \lambda_{T}\right)
-\boldsymbol{\gamma}_{i}^{\prime}\left(  \mathbf{\hat{f}}_{t}-\mathbf{f}%
_{t}\right)  -\left(  \boldsymbol{\hat{\gamma}}_{i}-\boldsymbol{\gamma}%
_{i}\right)  ^{\prime}\mathbf{f}_{t}-\left(  \boldsymbol{\hat{\gamma}}%
_{i}-\boldsymbol{\gamma}_{i}\right)  ^{\prime}\left(  \mathbf{\hat{f}}%
_{t}-\mathbf{f}_{t}\right)  . \label{it1}%
\end{equation}
Let $\boldsymbol{\delta}_{i,T}=\boldsymbol{\gamma}_{i}/\omega_{i,T}$, and
$\hat{\boldsymbol{\delta}}_{i,T}=\hat{\boldsymbol{\gamma}}_{i}/\omega_{i,T}$.
Then%
\begin{equation}
\hat{u}_{it}/\omega_{i,T}=\sigma_{i}\varepsilon_{it}\left(  \lambda
_{T}\right)  /\omega_{i,T}-\boldsymbol{\delta}_{i,T}^{\prime}\left(
\mathbf{\hat{f}}_{t}-\mathbf{f}_{t}\right)  -\left(  \boldsymbol{\hat{\delta}%
}_{i,T}-\boldsymbol{\delta}_{i,T}\right)  ^{\prime}\mathbf{f}_{t}-\left(
\boldsymbol{\hat{\delta}}_{i,T}-\boldsymbol{\delta}_{i,T}\right)  ^{\prime
}\left(  \mathbf{\hat{f}}_{t}-\mathbf{f}_{t}\right)  . \label{eit2}%
\end{equation}
Also, subject to the normalization $n^{-1}\sum_{j=1}^{n}\boldsymbol{\hat
{\gamma}}_{j}\boldsymbol{\hat{\gamma}}_{j}^{\prime}=\mathbf{I}_{m_{0}}$ and
$n^{-1}\sum_{j=1}^{n}\boldsymbol{\gamma}_{j}\boldsymbol{\gamma}_{j}^{\prime
}=\mathbf{I}_{m_{0}}$ we have
\[
\mathbf{\hat{f}}_{t}=n^{-1}\sum_{j=1}^{n}\boldsymbol{\hat{\gamma}}_{j}%
y_{jt}=\left(  n^{-1}\sum_{j=1}^{n}\boldsymbol{\hat{\gamma}}_{j}%
\boldsymbol{\gamma}_{j}^{\prime}\right)  \mathbf{f}_{t}+n^{-1}\sum_{j=1}%
^{n}\boldsymbol{\hat{\gamma}}_{j}\sigma_{j}\varepsilon_{jt}\left(  \lambda
_{T}\right)  ,
\]
and hence%
\[
\mathbf{\hat{f}}_{t}-\mathbf{f}_{t}=\left[  n^{-1}\sum_{j=1}^{n}\left(
\boldsymbol{\hat{\gamma}}_{j}-\boldsymbol{\gamma}_{j}\right)
\boldsymbol{\gamma}_{j}^{\prime}\right]  \mathbf{f}_{t}+n^{-1}\sum_{j=1}%
^{n}\left(  \boldsymbol{\hat{\gamma}}_{j}-\boldsymbol{\gamma}_{j}\right)
\sigma_{j}\varepsilon_{jt}\left(  \lambda_{T}\right)  +n^{-1}\sum_{j=1}%
^{n}\boldsymbol{\gamma}_{j}\sigma_{j}\varepsilon_{jt}\left(  \lambda
_{T}\right)  .
\]
Using this result in (\ref{eit2}) we obtain%
\begin{align}
\hat{u}_{it}/\omega_{i,T}  &  =\sigma_{i}\varepsilon_{it}\left(  \lambda
_{T}\right)  /\omega_{i,T}-\boldsymbol{\delta}_{i,T}^{\prime}\left[
n^{-1}\sum_{j=1}^{n}\boldsymbol{\gamma}_{j}\sigma_{j}\varepsilon_{jt}\left(
\lambda_{T}\right)  \right] \nonumber\\
&  -\boldsymbol{\delta}_{i,T}^{\prime}\left[  n^{-1}\sum_{j=1}^{n}\left(
\boldsymbol{\hat{\gamma}}_{j}-\boldsymbol{\gamma}_{j}\right)
\boldsymbol{\gamma}_{j}^{\prime}\right]  \mathbf{f}_{t}-\boldsymbol{\delta
}_{i,T}^{\prime}\left[  n^{-1}\sum_{j=1}^{n}\left(  \boldsymbol{\hat{\gamma}%
}_{j}-\boldsymbol{\gamma}_{j}\right)  \sigma_{j}\varepsilon_{jt}\left(
\lambda_{T}\right)  \right] \nonumber\\
&  -\left(  \boldsymbol{\hat{\delta}}_{i,T}-\boldsymbol{\delta}_{i,T}\right)
^{\prime}\mathbf{f}_{t}-\left(  \boldsymbol{\hat{\delta}}_{i,T}%
-\boldsymbol{\delta}_{i,T}\right)  ^{\prime}\left(  \mathbf{\hat{f}}%
_{t}-\mathbf{f}_{t}\right)  , \label{eut3}%
\end{align}
and summing over $i$ yields
\begin{align*}
n^{-1/2}\sum_{i=1}^{n}\hat{u}_{it}/\omega_{i,T}  &  =n^{-1/2}\sum_{i=1}%
^{n}\sigma_{i}\varepsilon_{it}\left(  \lambda_{T}\right)  /\omega
_{i,T}-\boldsymbol{\varphi}_{nT}^{\prime}\left(  n^{-1/2}\sum_{i=1}%
^{n}\boldsymbol{\gamma}_{i}\sigma_{i}\varepsilon_{it}\left(  \lambda
_{T}\right)  \right) \\
&  -\boldsymbol{\varphi}_{nT}^{\prime}\left[  n^{-1/2}\sum_{i=1}^{n}\left(
\boldsymbol{\hat{\gamma}}_{i}-\boldsymbol{\gamma}_{i}\right)
\boldsymbol{\gamma}_{i}^{\prime}\right]  \mathbf{f}_{t}-\boldsymbol{\varphi
}_{nT}^{\prime}\left[  n^{-1/2}\sum_{i=1}^{n}\left(  \boldsymbol{\hat{\gamma}%
}_{i}-\boldsymbol{\gamma}_{i}\right)  \sigma_{i}\varepsilon_{it}\left(
\lambda_{T}\right)  \right] \\
&  -\left[  n^{-1/2}\sum_{i=1}^{n}\left(  \boldsymbol{\hat{\delta}}%
_{i,T}-\boldsymbol{\delta}_{i,T}\right)  ^{\prime}\right]  \mathbf{f}%
_{t}-\left[  n^{-1/2}\sum_{i=1}^{n}\left(  \boldsymbol{\hat{\delta}}%
_{i,T}-\boldsymbol{\delta}_{i,T}\right)  \right]  ^{\prime}\left(
\mathbf{\hat{f}}_{t}-\mathbf{f}_{t}\right)  ,
\end{align*}
where $\boldsymbol{\varphi}_{nT}=n^{-1}\sum_{i=1}^{n}\boldsymbol{\delta}%
_{i,T}$. Written more compactly
\begin{equation}
h_{t,nT}\left(  \lambda_{T}\right)  =n^{-1/2}\sum_{i=1}^{n}\hat{u}_{it}%
/\omega_{i,T}=\psi_{t,nT}\left(  \lambda_{T}\right)  -s_{t,nT}\left(
\lambda_{T}\right)  , \label{Avee}%
\end{equation}
where
\begin{align}
\psi_{t,nT}\left(  \lambda_{T}\right)   &  =\frac{1}{\sqrt{n}}\sum_{i=1}%
^{n}\frac{a_{i,nT}\sigma_{i}\varepsilon_{it}\left(  \lambda_{T}\right)
}{\omega_{i,T}},\text{ \ \ \ \ \ \ \ \ }a_{i,nT}=1-\omega_{i,T}%
\boldsymbol{\varphi}_{nT}^{\prime}\boldsymbol{\gamma}_{i},\label{epsinT1}\\
s_{t,nT}\left(  \lambda_{T}\right)   &  =\frac{1}{\sqrt{n}}\sum_{i=1}%
^{n}\left[  \boldsymbol{\varphi}_{nT}^{\prime}\left(  \boldsymbol{\hat{\gamma
}}_{i}-\boldsymbol{\gamma}_{i}\right)  \sigma_{i}\varepsilon_{it}\left(
\lambda_{T}\right)  +\left(  \boldsymbol{\hat{\delta}}_{i,T}%
-\boldsymbol{\delta}_{i,T}\right)  ^{\prime}\mathbf{\hat{f}}_{t}%
+\boldsymbol{\varphi}_{nT}^{\prime}\left(  \boldsymbol{\hat{\gamma}}%
_{i}-\boldsymbol{\gamma}_{i}\right)  \boldsymbol{\gamma}_{i}^{\prime
}\mathbf{f}_{t}\right]  . \label{snT}%
\end{align}
Further, let%
\begin{equation}
\xi_{t,n}\left(  \lambda_{T}\right)  =\frac{1}{\sqrt{n}}\sum_{i=1}^{n}%
a_{i,n}\varepsilon_{it}\left(  \lambda_{T}\right)  \text{, \ \ }%
a_{i,n}=1-\sigma_{i}\boldsymbol{\varphi}_{n}^{\prime}\boldsymbol{\gamma}_{i},
\label{egzi}%
\end{equation}
where $\boldsymbol{\varphi}_{n}=n^{-1}\sum_{i=1}^{n}\boldsymbol{\delta}_{i},$
and $\boldsymbol{\delta}_{i}=\boldsymbol{\gamma}_{i}/\sigma_{i}$. Then
$\psi_{t,nT}\left(  \lambda_{T}\right)  ,$ given by (\ref{epsinT1}), can be
written as
\begin{align*}
\psi_{t,nT}\left(  \lambda_{T}\right)   &  =\xi_{t,n}\left(  \lambda
_{T}\right)  +\frac{1}{\sqrt{n}}\sum_{i=1}^{n}\left(  1-\omega_{i,T}%
\boldsymbol{\varphi}_{nT}^{\prime}\boldsymbol{\gamma}_{i}\right)  \frac
{\sigma_{i}\varepsilon_{it}\left(  \lambda_{T}\right)  }{\omega_{i,T}}%
-\frac{1}{\sqrt{n}}\sum_{i=1}^{n}\left(  1-\sigma_{i}\boldsymbol{\varphi}%
_{n}^{\prime}\boldsymbol{\gamma}_{i}\right)  \varepsilon_{it}\left(
\lambda_{T}\right) \\
&  =\xi_{t,n}\left(  \lambda_{T}\right)  -\frac{1}{\sqrt{n}}\sum_{i=1}%
^{n}\boldsymbol{\varphi}_{nT}^{\prime}\boldsymbol{\gamma}_{i}\sigma
_{i}\varepsilon_{it}\left(  \lambda_{T}\right)  +\frac{1}{\sqrt{n}}\sum
_{i=1}^{n}\sigma_{i}\boldsymbol{\varphi}_{n}^{\prime}\boldsymbol{\gamma}%
_{i}\varepsilon_{it}\left(  \lambda_{T}\right)  +\frac{1}{\sqrt{n}}\sum
_{i=1}^{n}\left(  \frac{\sigma_{i}\varepsilon_{it}\left(  \lambda_{T}\right)
}{\omega_{i,T}}-\varepsilon_{it}\left(  \lambda_{T}\right)  \right) \\
&  =\xi_{t,n}\left(  \lambda_{T}\right)  +\frac{1}{\sqrt{n}}\sum_{i=1}%
^{n}\zeta_{it}\left(  \lambda_{T}\right)  -\left(  \boldsymbol{\varphi}%
_{nT}-\boldsymbol{\varphi}_{n}\right)  ^{\prime}\frac{1}{\sqrt{n}}\sum
_{i=1}^{n}\boldsymbol{\gamma}_{i}\sigma_{i}\varepsilon_{it}\left(  \lambda
_{T}\right)  ,
\end{align*}
where
\begin{equation}
\zeta_{it}\left(  \lambda_{T}\right)  =\left[  \frac{1}{\left(  T^{-1}%
\boldsymbol{\varepsilon}_{i\circ}^{\prime}\mathbf{M}_{F}%
\boldsymbol{\varepsilon}_{i\circ}\right)  ^{1/2}}-1\right]  \varepsilon
_{it}\left(  \lambda_{T}\right)  . \label{zeta}%
\end{equation}
Writing $\psi_{t,nT}\left(  \lambda_{T}\right)  $ more compactly we have
\begin{equation}
\psi_{t,nT}\left(  \lambda_{T}\right)  =\xi_{t,n}\left(  \lambda_{T}\right)
+\upsilon_{t,nT}\left(  \lambda_{T}\right)  -\left(  \boldsymbol{\varphi}%
_{nT}-\boldsymbol{\varphi}_{n}\right)  ^{\prime}\boldsymbol{\kappa}%
_{t,n}\left(  \lambda_{T}\right)  , \label{epsinT}%
\end{equation}
where%
\begin{align}
\xi_{t,n}\left(  \lambda_{T}\right)   &  =\frac{1}{\sqrt{n}}\sum_{i=1}%
^{n}a_{i,n}\varepsilon_{it}\left(  \lambda_{T}\right)  \text{, \ \ }%
a_{i,n}=1-\sigma_{i}\boldsymbol{\varphi}_{n}^{\prime}\boldsymbol{\gamma}%
_{i},\label{egzitn}\\
\boldsymbol{\kappa}_{t,n}\left(  \lambda_{T}\right)   &  =\frac{1}{\sqrt{n}%
}\sum_{i=1}^{n}\boldsymbol{\gamma}_{i}\sigma_{i}\varepsilon_{it}\left(
\lambda_{T}\right)  ,\label{kappa_t}\\
\upsilon_{t,nT}\left(  \lambda_{T}\right)   &  =\frac{1}{\sqrt{n}}\sum
_{i=1}^{n}\zeta_{it}\left(  \lambda_{T}\right)  . \label{nu_nt}%
\end{align}
Using (\ref{Avee}) in (\ref{CDtilde1}) and after some algebra we have (where
we have made the dependence of $\widetilde{CD}$ on $\lambda_{T}$ explicit)
\[
\widetilde{CD}(\lambda_{T})=\left(  \sqrt{\frac{n}{n-1}}\right)  \frac
{1}{\sqrt{T}}\sum_{t=1}^{T}\left(  \frac{\psi_{t,nT}^{2}\left(  \lambda
_{T}\right)  -1}{\sqrt{2}}\right)  +\left(  \sqrt{\frac{n}{n-1}}\right)
\left(  p_{nT}\left(  \lambda_{T}\right)  -q_{nT}\left(  \lambda_{T}\right)
\right)  ,
\]
where $\psi_{t,nT}\left(  \lambda_{T}\right)  $ is defined by (\ref{epsinT}),
\begin{equation}
p_{nT}\left(  \lambda_{T}\right)  =T^{-1/2}\sum_{t=1}^{T}s_{t,nT}^{2}\left(
\lambda_{T}\right)  , \label{pnT}%
\end{equation}
and%
\begin{equation}
q_{nT}\left(  \lambda_{T}\right)  =T^{-1/2}\sum_{t=1}^{T}\psi_{t,nT}\left(
\lambda_{T}\right)  s_{t,nT}\left(  \lambda_{T}\right)  . \label{qnT}%
\end{equation}
By Lemma \ref{lm5} of the supplement $p_{nT}\left(  \lambda
_{T}\right)  =o_{p}(1)$, and $q_{nT}\left(  \lambda_{T}\right)  =o_{p}(1)$.
Hence
\begin{equation}
\widetilde{CD}\left(  \lambda_{T}\right)  =\frac{1}{\sqrt{T}}\sum_{t=1}%
^{T}\left(  \frac{\psi_{t,nT}^{2}\left(  \lambda_{T}\right)  -1}{\sqrt{2}%
}\right)  +o_{p}(1). \label{CDtilteop_p}%
\end{equation}
Now consider $T^{-1/2}\sum_{t=1}^{T}\psi_{t,nT}^{2}\left(  \lambda_{T}\right)
$ and using (\ref{epsinT}) note that
\begin{align}
\frac{1}{\sqrt{T}}\sum_{t=1}^{T}\psi_{t,nT}^{2}\left(  \lambda_{T}\right)   &
=\left(  \frac{1}{\sqrt{T}}\sum_{t=1}^{T}\xi_{t,n}^{2}\left(  \lambda
_{T}\right)  \right)  \left(  1+\frac{2\sum_{t=1}^{T}\xi_{t,n}\left(
\lambda_{T}\right)  \upsilon_{t,nT}\left(  \lambda_{T}\right)  }{\sum
_{t=1}^{T}\xi_{t,n}^{2}\left(  \lambda_{T}\right)  }\right)  +\frac{1}%
{\sqrt{T}}\sum_{t=1}^{T}\upsilon_{t,nT}^{2}\left(  \lambda_{T}\right)
\nonumber\\
&  +\sqrt{T}\left(  \boldsymbol{\varphi}_{nT}-\boldsymbol{\varphi}_{n}\right)
^{\prime}\left(  \frac{\sum_{t=1}^{T}\boldsymbol{\kappa}_{t,n}\left(
\lambda_{T}\right)  \boldsymbol{\kappa}_{t,n}^{\prime}\left(  \lambda
_{T}\right)  }{T}\right)  \left(  \boldsymbol{\varphi}_{nT}%
-\boldsymbol{\varphi}_{n}\right) \nonumber\\
&  -2\sqrt{T}\left(  \boldsymbol{\varphi}_{nT}-\boldsymbol{\varphi}%
_{n}\right)  ^{\prime}\left(  \frac{1}{T}\sum_{t=1}^{T}\boldsymbol{\kappa
}_{t,n}\left(  \lambda_{T}\right)  \upsilon_{t,nT}\left(  \lambda_{T}\right)
\right) \nonumber\\
&  -2\sqrt{T}\left(  \boldsymbol{\varphi}_{nT}-\boldsymbol{\varphi}%
_{n}\right)  ^{\prime}\left(  \frac{1}{T}\sum_{t=1}^{T}\boldsymbol{\kappa
}_{t,n}\left(  \lambda_{T}\right)  \xi_{t,n}\left(  \lambda_{T}\right)
\right)  .\nonumber
\end{align}
By Lemma \ref{lm:gnT} of the supplement, it follows that
\begin{equation}
\frac{1}{\sqrt{T}}\sum_{t=1}^{T}\psi_{t,nT}^{2}\left(  \lambda_{T}\right)
=\left(  \frac{1}{\sqrt{T}}\sum_{t=1}^{T}\xi_{t,n}^{2}\left(  \lambda
_{T}\right)  \right)  \left(  1+\frac{2\sum_{t=1}^{T}\xi_{t,n}\left(
\lambda_{T}\right)  \upsilon_{t,nT}\left(  \lambda_{T}\right)  }{\sum
_{t=1}^{T}\xi_{t,n}^{2}\left(  \lambda_{T}\right)  }\right)  +o_{p}(1).
\label{sumepsi_p}%
\end{equation}
Consider now the bias-corrected version of $\widetilde{CD}\left(  \lambda
_{T}\right)  $ defined by%
\begin{equation}
\widetilde{CD}^{\ast}\left(  \lambda_{T}\right)  =\frac{\widetilde{CD}\left(
\lambda_{T}\right)  +\sqrt{\frac{T}{2}}\theta_{n}}{1-\theta_{n}}
\label{CDS1_p}%
\end{equation}
where $\theta_{n}=1-\frac{1}{n}\sum_{i=1}^{n}a_{i,n}^{2},$ and $a_{i,n}%
=1-\sigma_{i}\boldsymbol{\varphi}_{n}^{\prime}\boldsymbol{\gamma}_{i}$. Using
(\ref{CDtilteop_p}) in (\ref{CDS1_p}), we have
\begin{align*}
\widetilde{CD}^{\ast}\left(  \lambda_{T}\right)   &  =\frac{\frac{1}{\sqrt{T}%
}\sum_{t=1}^{T}\left(  \frac{\psi_{t,nT}^{2}\left(  \lambda_{T}\right)
-1}{\sqrt{2}}\right)  +\sqrt{\frac{T}{2}}\theta_{n}}{1-\theta_{n}}%
+o_{p}\left(  1\right) \\
&  =\frac{\frac{1}{\sqrt{T}}\sum_{t=1}^{T}\left(  \frac{\psi_{t,nT}^{2}\left(
\lambda_{T}\right)  }{\sqrt{2}}\right)  -\sqrt{\frac{T}{2}}\left(
1-\theta_{n}\right)  }{1-\theta_{n}}+o_{p}\left(  1\right)  .
\end{align*}
Now using (\ref{sumepsi_p}) in the above and after some re-arrangement of the
terms we obtain
\[
\widetilde{CD}^{\ast}\left(  \lambda_{T}\right)  =\frac{\left[  \frac{1}%
{\sqrt{T}}\sum_{t=1}^{T}\left(  \frac{\xi_{t,n}^{2}\left(  \lambda_{T}\right)
-\left(  1-\theta_{n}\right)  }{\sqrt{2}}\right)  \right]  \left(  1+\frac
{2}{\sqrt{T}}w_{nT}\left(  \lambda_{T}\right)  \right)  }{1-\theta_{n}}%
+\sqrt{2}w_{nT}\left(  \lambda_{T}\right)  +o_{p}\left(  1\right)  ,
\]
where
\[
w_{nT}\left(  \lambda_{T}\right)  =\frac{T^{-1/2}\sum_{t=1}^{T}\xi
_{t,n}\left(  \lambda_{T}\right)  \upsilon_{t,nT}\left(  \lambda_{T}\right)
}{T^{-1}\sum_{t=1}^{T}\xi_{t,n}^{2}\left(  \lambda_{T}\right)  }.
\]
By Lemma \ref{lm:wnT} of the supplement, $w_{nT}\left(  \lambda
_{T}\right)  =o_{p}\left(  1\right)  $. Hence\ \textbf{ }%
\[
\widetilde{CD}^{\ast}\left(  \lambda_{T}\right)  =\frac{\frac{1}{\sqrt{T}}%
\sum_{t=1}^{T}\left(  \frac{\xi_{t,n}^{2}\left(  \lambda_{T}\right)  -\left(
1-\theta_{n}\right)  }{\sqrt{2}}\right)  }{1-\theta_{n}}+o_{p}\left(
1\right)  .
\]
In particular, since $\varepsilon_{it}(\lambda_{T})=\varepsilon_{it}%
+\lambda_{T}\mathbf{w}_{i0}^{\prime}\boldsymbol{\varepsilon}_{\circ t}$, then
\begin{align*}
\xi_{t,n}(\lambda_{T})  &  =\frac{1}{\sqrt{n}}\sum_{i=1}^{n}a_{i,n}%
\varepsilon_{it}(\lambda_{T})\\
&  =\frac{1}{\sqrt{n}}\sum_{i=1}^{n}a_{i,n}\varepsilon_{it}+\lambda_{T}%
\frac{1}{\sqrt{n}}\sum_{i=1}^{n}a_{i,n}\mathbf{w}_{i0}^{\prime}%
\boldsymbol{\varepsilon}_{\circ t}\\
&  \equiv A_{nt}+\lambda_{T}B_{nt}.
\end{align*}
Using this result, $\widetilde{CD}^{\ast}(\lambda_{T})$ can be written as
\begin{align}
\widetilde{CD}^{\ast}\left(  \lambda_{T}\right)   &  =\frac{\frac{1}{\sqrt{T}%
}\sum_{t=1}^{T}\left(  \frac{\xi_{t,n}^{2}(\lambda_{T})-\left(  1-\theta
_{n}\right)  }{\sqrt{2}}\right)  }{1-\theta_{n}}+o_{p}\left(  1\right)
\nonumber\\
&  =\frac{\frac{1}{\sqrt{T}}\sum_{t=1}^{T}\left(  \frac{A_{nt}^{2}%
+2\lambda_{T}A_{nt}B_{nt}+\lambda_{T}^{2}B_{nt}^{2}-\left(  1-\theta
_{n}\right)  }{\sqrt{2}}\right)  }{1-\theta_{n}}+o_{p}\left(  1\right)
\nonumber\\
&  =\frac{\frac{1}{\sqrt{T}}\sum_{t=1}^{T}\left(  \frac{A_{nt}^{2}-\left(
1-\theta_{n}\right)  }{\sqrt{2}}\right)  }{1-\theta_{n}}+\frac{\sqrt{2}%
\lambda_{T}}{\left(  1-\theta_{n}\right)  \sqrt{T}}\sum_{t=1}^{T}A_{nt}%
B_{nt}+\frac{\frac{\lambda_{T}^{2}}{\sqrt{2}\sqrt{T}}\sum_{t=1}^{T}B_{nt}^{2}%
}{1-\theta_{n}}+o_{p}\left(  1\right) \nonumber\\
&  =\widetilde{CD}^{\ast}\left(  0\right)  +\phi_{nT}+g_{nT}+o_{p}\left(
1\right)  . \label{CDtstar}%
\end{align}
The second term of (\ref{CDtstar}) can be written as%
\begin{align*}
\phi_{nT}  &  =\frac{\sqrt{2}\lambda_{T}}{\left(  1-\theta_{n}\right)
\sqrt{T}}\sum_{t=1}^{T}A_{nt}B_{nt}\\
&  =\phi_{n}+\frac{\sqrt{2}c_{\lambda}}{\left(  1-\theta_{n}\right)  }%
T^{-1}\sum_{t=1}^{T}\left[  A_{nt}B_{nt}-E\left(  A_{nt}B_{nt}\right)
\right]
\end{align*}
where $\phi_{n}=\frac{\sqrt{2}c_{\lambda}}{\left(  1-\theta_{n}\right)
}T^{-1}\sum_{t=1}^{T}E\left(  A_{nt}B_{nt}\right)  $. Further
\begin{align*}
E\left(  A_{nt}B_{nt}\right)   &  =E\left(  \frac{1}{\sqrt{n}}\sum_{i=1}%
^{n}a_{i,n}\varepsilon_{it}\right)  \left(  \frac{1}{\sqrt{n}}\sum_{j=1}%
^{n}a_{j,n}\mathbf{w}_{j0}^{\prime}\boldsymbol{\varepsilon}_{\circ t}\right)
\\
&  =\frac{1}{n}\sum_{i=1}^{n}\sum_{j=1}^{n}a_{j,n}a_{i,n}E\left(
\mathbf{w}_{j0}^{\prime}\boldsymbol{\varepsilon}_{\circ t}\varepsilon
_{it}\right)  =\frac{1}{n}\sum_{i=1}^{n}\sum_{j=1}^{n}a_{j,n}a_{i,n}w_{ji}\\
&  =n^{-1}\mathbf{a}_{n}^{\prime}\mathbf{Wa}_{n}%
\end{align*}
with $\mathbf{a}_{n}=(a_{1,n},a_{2,n},...,a_{n,n})^{\prime}$. Hence%
\begin{equation}
\phi_{n}=\frac{\sqrt{2}c_{\lambda}}{\left(  1-\theta_{n}\right)  }\left(
T^{-1}\sum_{t=1}^{T}E\left(  A_{nt}B_{nt}\right)  \right)  =\frac{\sqrt
{2}c_{\lambda}}{\left(  1-\theta_{n}\right)  }n^{-1}\mathbf{a}_{n}^{\prime
}\mathbf{Wa}_{n}. \label{phi_n}%
\end{equation}
Under part (a) of Assumption \ref{ass:errors}, $\varepsilon_{it}\sim
IID\left(  0,1\right)  $ for all $i$ and $t$, with $E\left(  \varepsilon
_{it}^{8}\right)  <C$, it then follows that $A_{nt}B_{nt}-E\left(
A_{nt}B_{nt}\right)  $ will be serially independent with zero means and finite
variances and by weak law of large numbers $T^{-1}\sum_{t=1}^{T}\left[
A_{nt}B_{nt}-E\left(  A_{nt}B_{nt}\right)  \right]  =o_{p}(1)$. Therefore%
\begin{equation}
\phi_{nT}=\phi_{n}+o_{p}(1). \label{phi_nT2}%
\end{equation}
Similarly, for the third term of (\ref{CDtstar}) we first note that
\begin{align*}
T^{-1}\sum_{t=1}^{T}B_{nt}^{2}  &  \rightarrow_{p}E\left(  B_{nt}^{2}\right)
=\frac{1}{n}\sum_{i=1}^{n}\sum_{j=1}^{n}a_{i,n}a_{j,n}\mathbf{w}_{i0}^{\prime
}E\left(  \boldsymbol{\varepsilon}_{\circ t}\boldsymbol{\varepsilon}_{\circ
t}^{\prime}\right)  \mathbf{w}_{j0}\\
&  =\frac{1}{n}\sum_{i=1}^{n}\sum_{j=1}^{n}a_{i,n}a_{j,n}\mathbf{w}%
_{i0}^{\prime}\mathbf{w}_{j0}=\frac{1}{n}\sum_{s=1}^{n}\sum_{i=1}^{n}%
\sum_{j=1}^{n}a_{i,n}w_{is}a_{j,n}w_{js}\\
&  =\frac{1}{n}\sum_{s=1}^{n}\left(  \sum_{i=1}^{n}a_{i,n}w_{is}\right)
^{2}<C,
\end{align*}
and hence
\begin{equation}
g_{nT}=\frac{\frac{\lambda_{T}^{2}}{\sqrt{2}\sqrt{T}}\sum_{t=1}^{T}B_{nt}^{2}%
}{1-\theta_{n}}=\frac{c_{\lambda}^{2}}{\sqrt{2T}\left(  1-\theta_{n}\right)
}\left(  T^{-1}\sum_{t=1}^{T}B_{nt}^{2}\right)  =O_{p}\left(  \frac{1}%
{\sqrt{T}}\right)  . \label{gnT}%
\end{equation}
Using (\ref{phi_nT2}) and (\ref{gnT}) in (\ref{CDtstar}) now yields
\begin{equation}
\widetilde{CD}^{\ast}\left(  \lambda_{T}\right)  =\widetilde{CD}^{\ast}\left(
0\right)  +\phi_{n}+o_{p}\left(  1\right)  . \label{CDCD0}%
\end{equation}
Consider the first term of (\ref{CDtstar}), $\widetilde{CD}^{\ast}\left(
0\right)  $, and note that $A_{nt}$ can be written as $\,A_{nt}=n^{-1/2}%
\mathbf{a}_{n}^{\prime}\boldsymbol{\varepsilon}_{\circ t}$, where
$\boldsymbol{a}_{n}^{\prime}=(a_{1,n},a_{2,n},...,a_{n,n})$, and we have
(recall that $a_{i}=1-\sigma_{i}\boldsymbol{\varphi}_{n}^{\prime
}\boldsymbol{\gamma}_{i}$)%
\[
E\left(  A_{nt}^{2}\right)  =n^{-1}E\left(  \boldsymbol{\varepsilon}_{\circ
t}^{\prime}\mathbf{a}_{n}\mathbf{a}_{n}^{\prime}\boldsymbol{\varepsilon
}_{\circ t}\right)  =n^{-1}\mathbf{a}_{n}^{\prime}\mathbf{a}_{n}=n^{-1}%
\sum_{i=1}^{n}a_{i,n}^{2}=\frac{1}{n}\sum_{i=1}^{n}\left(  1-\sigma
_{i}\boldsymbol{\varphi}_{n}^{\prime}\boldsymbol{\gamma}_{i}\right)
^{2}=1-\theta_{n}>0,
\]
and (using result (S.7) of Lemma 6 in \cite{pesaran2024testing})%
\begin{align*}
E\left(  A_{nt}^{4}\right)   &  =n^{2}E\left[  (\boldsymbol{\varepsilon
}_{\circ t}^{\prime}\mathbf{a}_{n}\mathbf{a}_{n}^{\prime}%
\boldsymbol{\varepsilon}_{\circ t})^{2}\right]  =\kappa_{2}n^{-2}%
\mathrm{tr}\left[  \mathbf{A\odot A}\right]  +n^{-2}\left[  \mathrm{tr}\left(
\mathbf{A}\right)  \right]  ^{2}+2n^{-2}\mathrm{tr}\left(  \mathbf{A}%
^{2}\right) \\
&  =\kappa_{2}n^{-2}\sum_{i=1}^{n}a_{i,n}^{4}+3\left(  n^{-1}\sum_{i=1}%
^{n}a_{i,n}^{2}\right)  ^{2},
\end{align*}
where $\kappa_{2}=E\left(  \varepsilon_{it}^{4}\right)  -3$ and $\mathbf{A}%
=\mathbf{a}_{n}\mathbf{a}_{n}^{\prime}$. Hence
\[
Var\left(  A_{nt}^{2}\right)  =2\left(  \frac{1}{n}\sum_{i=1}^{n}a_{i,n}%
^{2}\right)  ^{2}-\frac{\kappa_{2}}{n}\left(  \frac{1}{n}\sum_{i=1}^{n}%
a_{i,n}^{4}\right)  .
\]
Furthermore, since $\sup_{i}\left\vert a_{i,n}\right\vert =\sup_{i}\left\vert
1-\sigma_{i}\boldsymbol{\varphi}_{n}^{\prime}\boldsymbol{\gamma}%
_{i}\right\vert <1+\left(  \sup_{i}\sigma_{i}\right)  \left(  \sup
_{i}\left\Vert \boldsymbol{\gamma}_{i}\right\Vert \right)  \left\Vert
\boldsymbol{\varphi}_{n}\right\Vert <C$, then $n^{-2}\sum_{i=1}^{n}a_{i,n}%
^{4}=O(n^{-1})$, and $Var\left(  A_{nt}^{2}\right)  =2\left(  1-\theta
_{n}\right)  ^{2}+O(n^{-1})$. Using the above results it now readily follows
that,
\begin{equation}
\widetilde{CD}^{\ast}(0)=\frac{\frac{1}{\sqrt{T}}\sum_{t=1}^{T}\left(
\frac{A_{nt}^{2}-\left(  1-\theta_{n}\right)  }{\sqrt{2}}\right)  }%
{1-\theta_{n}}=\frac{1}{\sqrt{T}}\sum_{t=1}^{T}\left(  \frac{A_{nt}%
^{2}-E\left(  A_{nt}^{2}\right)  }{\sqrt{Var\left(  A_{nt}^{2}\right)
+O(n^{-1})}}\right)  . \label{CD*0}%
\end{equation}
Since under part (a) of Assumption \ref{ass:errors}, $\varepsilon_{it}\sim
IID\left(  0,1\right)  $ for all $i$ and $t$, then $A_{nt}=\frac{1}{\sqrt{n}%
}\sum_{i=1}^{n}a_{i,n}\varepsilon_{it}$, and $A_{nt}^{2}$ are also
independently distributed over $t$ with finite second order moments. Then by
Lindeberg-L\'{e}vy central limit theorem it follows that $\widetilde{CD}%
^{\ast}(0)\rightarrow_{d}\mathcal{N}(0,1),$ as $n$ and $T\rightarrow\infty$.
Using this result in (\ref{CDCD0}) we further have $\widetilde{CD}^{\ast
}\left(  \lambda_{T}\right)  \rightarrow_{d}\mathcal{N}\left(  \phi,1\right)
$, where $\phi=\lim_{n\rightarrow\infty}\phi_{n}$. By Lemma \ref{lmcdtcd} of
the supplement we have $CD=\widetilde{CD}+o_{p}(1)$, then it follows
\begin{align*}
CD^{\ast}(\theta_{n})  &  =\frac{CD+\sqrt{\frac{T}{2}}\theta_{n}}{1-\theta
_{n}}=\frac{\widetilde{CD}+\sqrt{\frac{T}{2}}\theta_{n}}{1-\theta_{n}}%
+o_{p}(1)\\
&  =\widetilde{CD}^{\ast}\left(  \lambda_{T}\right)  +o_{p}(1),
\end{align*}
where the final line holds by (\ref{CDS1_p}). Now result (\ref{CDs:H1T}) is established as required.

\section{Proof of Proposition \ref{Proposition:theta}}

Note that $\theta_{n}$ define by (\ref{thetan}) can be written as $\theta
_{n}=2g_{n}-\boldsymbol{\varphi}_{n}^{\prime}\mathbf{H}_{n}\boldsymbol{\varphi
}_{n}$, where $g_{n}=n^{-1}\sum_{i=1}^{n}\sigma_{i}\boldsymbol{\varphi}%
_{n}^{\prime}\boldsymbol{\gamma}_{i}$, $\mathbf{H}_{n}=n^{-1}\sum_{i=1}%
^{n}\sigma_{i}\left(  \boldsymbol{\gamma}_{i}\boldsymbol{\gamma}_{i}^{\prime
}\right)  $, $\boldsymbol{\varphi}_{n}=n^{-1}\sum_{i=1}^{n}\boldsymbol{\delta
}_{i},$ and $\boldsymbol{\delta}_{i}=\boldsymbol{\gamma}_{i}/\sigma_{i}$.
Similarly using (\ref{htheta}) we have $\hat{\theta}_{nT}=2\hat{g}%
_{nT}-\boldsymbol{\hat{\varphi}}_{nT}^{\prime}\mathbf{\hat{H}}_{nT}%
\boldsymbol{\hat{\varphi}}_{nT}$, where $\hat{g}_{nT}=n^{-1}\sum_{i=1}^{n}%
\hat{\sigma}_{i,T}\boldsymbol{\hat{\varphi}}_{nT}^{\prime}\boldsymbol{\hat
{\gamma}}_{i},$ $\mathbf{\hat{H}}_{nT}=n^{-1}\sum_{i=1}^{n}\hat{\sigma}%
_{i,T}^{2}\left(  \boldsymbol{\hat{\gamma}}_{i}\boldsymbol{\hat{\gamma}}%
_{i}^{\prime}\right)  ,$ $\boldsymbol{\hat{\varphi}}_{nT}=n^{-1}\sum_{i=1}%
^{n}\boldsymbol{\hat{\delta}}_{i,nT},$ and\ $\boldsymbol{\hat{\delta}}%
_{i,nT}=\boldsymbol{\hat{\gamma}}_{i}/\hat{\sigma}_{i,T}$. Then%
\begin{equation}
\sqrt{T}\left(  \hat{\theta}_{nT}-\theta_{n}\right)  =2\sqrt{T}\left(  \hat
{g}_{nT}-g_{n}\right)  -\sqrt{T}\left(  \boldsymbol{\hat{\varphi}}%
_{nT}^{\prime}\mathbf{\hat{H}}_{nT}\boldsymbol{\hat{\varphi}}_{nT}%
-\boldsymbol{\varphi}_{n}^{\prime}\mathbf{H}_{n}\boldsymbol{\varphi}%
_{n}\right)  . \label{corollary01}%
\end{equation}
Consider the first term of the above
\begin{align}
\sqrt{T}\left(  \hat{g}_{nT}-g_{n}\right)   &  =\sqrt{T}\left(
\boldsymbol{\hat{\varphi}}_{nT}-\boldsymbol{\varphi}_{n}\right)  \left(
\frac{1}{n}\sum_{i=1}^{n}\sigma_{i}\boldsymbol{\gamma}_{i}\right) \nonumber\\
&  +\sqrt{T}\left[  \boldsymbol{\hat{\varphi}}_{nT}^{\prime}\left(  \frac
{1}{n}\sum_{i=1}^{n}\hat{\sigma}_{i,T}\boldsymbol{\hat{\gamma}}_{i}-\frac
{1}{n}\sum_{i=1}^{n}\sigma_{i}\boldsymbol{\gamma}_{i}\right)  \right]  ,
\label{corollary1a}%
\end{align}
and since $\sigma_{i}$ and $\boldsymbol{\gamma}_{i}$ are bounded then
$n^{-1}\sum_{i=1}^{n}\sigma_{i}\boldsymbol{\gamma}_{i}=O(1)$. Also by
(\ref{phinhatLemma}) of Lemma \ref{PHIn} in the supplement we have
$\sqrt{T}\left(  \boldsymbol{\hat{\varphi}}_{nT}-\boldsymbol{\varphi}%
_{n}\right)  =o_{p}(1)$, and hence the first term of the above is $o_{p}(1)$.
To establish the probability order of the second term of (\ref{corollary1a}),
we first note that
\begin{align}
\sqrt{T}\left[  \boldsymbol{\hat{\varphi}}_{nT}^{\prime}\left(  \frac{1}%
{n}\sum_{i=1}^{n}\hat{\sigma}_{i,T}\boldsymbol{\hat{\gamma}}_{i}-\frac{1}%
{n}\sum_{i=1}^{n}\sigma_{i}\boldsymbol{\gamma}_{i}\right)  \right]   &
=\sqrt{T}\left[  \boldsymbol{\varphi}_{n}^{\prime}\left(  \frac{1}{n}%
\sum_{i=1}^{n}\hat{\sigma}_{i,T}\boldsymbol{\hat{\gamma}}_{i}-\frac{1}{n}%
\sum_{i=1}^{n}\sigma_{i}\boldsymbol{\gamma}_{i}\right)  \right] \nonumber\\
&  +\sqrt{T}\left[  \left(  \boldsymbol{\hat{\varphi}}_{nT}%
-\boldsymbol{\varphi}_{n}\right)  ^{^{\prime}}\left(  \frac{1}{n}\sum
_{i=1}^{n}\hat{\sigma}_{i,T}\boldsymbol{\hat{\gamma}}_{i}-\frac{1}{n}%
\sum_{i=1}^{n}\sigma_{i}\boldsymbol{\gamma}_{i}\right)  \right]  .
\label{corollary1e}%
\end{align}
But by (\ref{bound phin}) and (\ref{lc1}) of the supplement,
$\boldsymbol{\varphi}_{n}=O_{p}(1)$ and $n^{-1}\sum_{i=1}^{n}\left(
\hat{\sigma}_{i,T}\boldsymbol{\hat{\gamma}}_{i}-\sigma_{i}\boldsymbol{\gamma
}_{i}\right)  =O_{p}\left(  \ln\left(  n\right)  /T\right)  ,$ which also
establishes that the second term of (\ref{corollary1e}) is $o_{p}(1)$.
Therefore overall we have%
\begin{equation}
\sqrt{T}\left(  \hat{g}_{nT}-g_{n}\right)  =o_{p}(1). \label{corollary1c}%
\end{equation}
Consider now the second term of (\ref{corollary01}) and note that
\begin{align}
\sqrt{T}\left(  \boldsymbol{\hat{\varphi}}_{nT}^{\prime}\mathbf{\hat{H}}%
_{nT}\boldsymbol{\hat{\varphi}}_{nT}-\boldsymbol{\varphi}_{n}^{\prime
}\mathbf{H}_{n}\boldsymbol{\varphi}_{n}\right)   &  =\sqrt{T}\left(
\boldsymbol{\hat{\varphi}}_{nT}-\boldsymbol{\varphi}_{n}\right)  ^{\prime
}\mathbf{\hat{H}}_{nT}\left(  \boldsymbol{\hat{\varphi}}_{nT}%
-\boldsymbol{\varphi}_{n}\right) \nonumber\\
&  +2\sqrt{T}\left(  \boldsymbol{\hat{\varphi}}_{nT}-\boldsymbol{\varphi}%
_{n}\right)  ^{\prime}\mathbf{\hat{H}}_{nT}\boldsymbol{\varphi}_{n}+\sqrt
{T}\boldsymbol{\varphi}_{n}^{\prime}\left(  \mathbf{\hat{H}}_{nT}%
-\mathbf{H}_{n}\right)  \boldsymbol{\varphi}_{n}, \label{corollary1b}%
\end{align}
where $\mathbf{\hat{H}}_{nT}=\frac{1}{n}\sum_{i=1}^{n}\hat{\sigma}_{i,T}%
^{2}\left(  \boldsymbol{\hat{\gamma}}_{i}\boldsymbol{\hat{\gamma}}_{i}%
^{\prime}\right)  $, and
\begin{align}
\sqrt{T}\left(  \mathbf{\hat{H}}_{nT}-\mathbf{H}_{n}\right)   &  =\frac
{\sqrt{T}}{n}\sum_{i=1}^{n}\sigma_{i}^{2}\left(  \boldsymbol{\hat{\gamma}}%
_{i}\boldsymbol{\hat{\gamma}}_{i}^{\prime}-\boldsymbol{\gamma}_{i}%
\boldsymbol{\gamma}_{i}^{\prime}\right)  +\frac{\sqrt{T}}{n}\sum_{i=1}%
^{n}\left(  \hat{\sigma}_{i,T}^{2}-\sigma_{i}^{2}\right)  \left(
\boldsymbol{\hat{\gamma}}_{i}\boldsymbol{\hat{\gamma}}_{i}^{\prime
}-\boldsymbol{\gamma}_{i}\boldsymbol{\gamma}_{i}^{\prime}\right) \nonumber\\
&  +\frac{\sqrt{T}}{n}\sum_{i=1}^{n}\left(  \hat{\sigma}_{i,T}^{2}%
-\omega_{i,T}^{2}\right)  \boldsymbol{\gamma}_{i}\boldsymbol{\gamma}%
_{i}^{\prime}+\frac{\sqrt{T}}{n}\sum_{i=1}^{n}\left(  \omega_{i,T}^{2}%
-\sigma_{i}^{2}\right)  \boldsymbol{\gamma}_{i}\boldsymbol{\gamma}_{i}%
^{\prime}\nonumber\\
&  =\sum_{j=1}^{4}\mathbf{D}_{j,nT}.\nonumber
\end{align}
The first two terms of (\ref{corollary1b}) are $o_{p}(1)$, since $\left\Vert
\boldsymbol{\varphi}_{n}\right\Vert <C$,  $\sqrt{T}\left(  \boldsymbol{\hat
{\varphi}}_{nT}-\boldsymbol{\varphi}_{n}\right)  =o_{p}(1)$, and $n^{-1}%
\sum_{i=1}^{n}\hat{\sigma}_{i,T}^{2}\left(  \boldsymbol{\hat{\gamma}}%
_{i}\boldsymbol{\hat{\gamma}}_{i}^{\prime}\right)  =O_{p}(1)$. To establish
the probability order of the third term of (\ref{corollary1b}), since
$\left\Vert \boldsymbol{\varphi}_{n}\right\Vert <C$ it is sufficient to
consider the four terms of $\sqrt{T}\left(  \mathbf{\hat{H}}_{nT}%
-\mathbf{H}_{n}\right)  $. It is clear that $\mathbf{D}_{2,nT}$ is dominated
by $\mathbf{D}_{1,nT\ }$ and by (\ref{lc4}) of Lemma \ref{lc} of the
supplement,
\[
\mathbf{D}_{1,nT\ }=\frac{\sqrt{T}}{n}\sum_{i=1}^{n}\sigma_{i}^{2}\left(
\boldsymbol{\hat{\gamma}}_{i}\boldsymbol{\hat{\gamma}}_{i}^{\prime
}-\boldsymbol{\gamma}_{i}\boldsymbol{\gamma}_{i}^{\prime}\right)
=O_{p}\left(  \sqrt{\frac{\ln\left(  n\right)  }{n}}\right)  =o_{p}(1).
\]
Using (\ref{Ews1}) of Lemma \ref{lm4} of the supplement and replacing
$b_{ni}$ with $\gamma_{ij}\gamma_{ij^{^{\prime}}}$ for $j,j^{\prime
}=1,2,\ldots,m_{0}$, it then follows that
\[
\mathbf{D}_{3,nT}=\frac{\sqrt{T}}{n}\sum_{i=1}^{n}\left(  \hat{\sigma}%
_{i,T}^{2}-\omega_{i,T}^{2}\right)  \boldsymbol{\gamma}_{i}\boldsymbol{\gamma
}_{i}^{\prime}=O_{p}\left(  \frac{\ln\left(  n\right)  }{\sqrt{T}}\right)
=o_{p}(1).
\]
Finally, denote the $\left(  j,j^{\prime}\right)  $ element of $\mathbf{D}%
_{4,nT}$ by $d_{4,nT}(j,j^{\prime})$ and note that
\[
d_{4,nT}(j,j^{\prime})=\frac{1}{n}\sum_{i=1}^{n}\left(  \sigma_{i}^{2}%
\gamma_{ij}\gamma_{ij^{\prime}}\right)  \sqrt{T}\left(  \frac
{\boldsymbol{\varepsilon}_{i\circ}^{^{\prime}}\mathbf{M}_{F}%
\boldsymbol{\varepsilon}_{i\circ}}{T}-1\right)  \text{, for }j,j^{\prime
}=1,2,...,m_{0}.
\]
But under Assumptions \ref{ass:errors} and \ref{ass:loadings}, $\left\vert
\sigma_{i}^{2}\gamma_{ij}\gamma_{ij^{\prime}}\right\vert <C$, and $\sqrt
{T}\left(  T^{-1}\boldsymbol{\varepsilon}_{i\circ}^{^{\prime}}\mathbf{M}%
_{F}\boldsymbol{\varepsilon}_{i\circ}-1\right)  $, for $i=1,2,...,n$ are
identically and independently distributed across $i$, with mean $1/\sqrt{T}$
and a finite variance\footnote{The mean and variance of $\sqrt{T}\left(
T^{-1}\boldsymbol{\varepsilon}_{i\circ}^{^{\prime}}\mathbf{M}_{F}%
\boldsymbol{\varepsilon}_{i\circ}-1\right)  $ can be obtained using
(\ref{mom1}) and (\ref{mom2}) in Lemma \ref{Qmom} of the supplement.}. Then by standard law of large numbers, for each $(j,j^{\prime})$,
$d_{4,nT}(j,j^{\prime})\rightarrow_{p}0$, as $n$ and $T\rightarrow\infty$, and
hence we also have $\mathbf{D}_{4,nT}=o_{p}(1)$. Overall, $\mathbf{\hat{H}%
}_{nT}-\mathbf{H}_{n}=o_{p}(1)$, and we have $\sqrt{T}\left(  \boldsymbol{\hat
{\varphi}}_{nT}^{\prime}\mathbf{\hat{H}}_{nT}\boldsymbol{\hat{\varphi}}%
_{nT}-\boldsymbol{\varphi}_{n}^{\prime}\mathbf{H}_{n}\boldsymbol{\varphi}%
_{n}\right)  =o_{p}(1)$. Using this result and (\ref{corollary1c}) in
(\ref{corollary01}) now yields $\sqrt{T}\left(  \hat{\theta}_{nT}-\theta
_{n}\right)  =o_{p}(1)$, as required.

\section{Proof of Theorem \ref{Theorem:CDs}}

Recall from (\ref{CD*a}) that $CD^{\ast}$ is given by%
\[
CD^{\ast}=\frac{CD+\sqrt{\frac{T}{2}}\hat{\theta}_{nT}}{1-\hat{\theta}_{nT}},
\]
where $\hat{\theta}_{nT}=1-\frac{1}{n}\sum_{i=1}^{n}\hat{a}_{i,n}^{2}$,
$\hat{a}_{i,n}=1-\hat{\sigma}_{i,T}\left(  \boldsymbol{\varphi}_{nT}^{\prime
}\boldsymbol{\hat{\gamma}}_{i}\right)  ,$ and $\boldsymbol{\hat{\varphi}}%
_{nT}=n^{-1}\sum_{i=1}^{n}\boldsymbol{\hat{\gamma}}_{i}/\hat{\sigma}_{i,T}$,
subject to the normalization $n^{-1}\sum_{i=1}^{n}\boldsymbol{\hat{\gamma}%
}_{i}\boldsymbol{\hat{\gamma}}_{i}^{^{\prime}}=\mathbf{I}_{m_{0}}$. By result
(\ref{diff theta}) of Proposition \ref{Proposition:theta}, $\sqrt{T}\left(
\hat{\theta}_{nT}-\theta_{n}\right)  =o_{p}(1)$, and hence
\begin{align*}
CD^{\ast}  &  =\frac{CD+\sqrt{\frac{T}{2}}\theta_{n}+\sqrt{\frac{T}{2}}\left(
\hat{\theta}_{nT}-\theta_{n}\right)  }{1-\theta_{n}-\left(  \hat{\theta}%
_{nT}-\theta_{n}\right)  }+o_{p}(1)\\
&  =\frac{CD+\sqrt{\frac{T}{2}}\theta_{n}}{1-\theta_{n}}+o_{p}\left(
1\right)  =CD^{\ast}(\theta_{n})+o_{p}\left(  1\right)  .
\end{align*}
Theorem \ref{Theorem:CDs} is then established by following Proposition
\ref{Proposition:CD(s)}.

\section{Proof of Theorem \ref{Theorem:CDss}}

Let $v_{it}=y_{it}-\boldsymbol{\alpha}_{i}^{\prime}\mathbf{d}_{t}%
-\boldsymbol{\beta}_{i}^{^{\prime}}\mathbf{x}_{it}$, and $u_{it}%
=y_{it}-\boldsymbol{\alpha}_{i}^{\prime}\mathbf{d}_{t}-\boldsymbol{\beta}%
_{i}^{^{\prime}}\mathbf{x}_{it}-\boldsymbol{\gamma}_{i}^{^{\prime}}%
\mathbf{f}_{t}=v_{it}-\boldsymbol{\gamma}_{i}^{^{\prime}}\mathbf{f}_{t}$, and
consider the following two optimization problems%
\begin{align}
&  \min_{\boldsymbol{\Gamma,F}}\frac{1}{nT}\sum_{i=1}^{n}\sum_{t=1}^{T}\left(
v_{it}-\boldsymbol{\gamma}_{i}^{^{\prime}}\mathbf{f}_{t}\right)
^{2},\label{Aproblem}\\
&  \min_{\boldsymbol{\Gamma,F}}\frac{1}{nT}\sum_{i=1}^{n}\sum_{t=1}^{T}\left(
\hat{v}_{it}-\boldsymbol{\gamma}_{i}^{^{\prime}}\mathbf{f}_{t}\right)  ^{2},
\label{Bproblem}%
\end{align}
where%
\begin{align}
\hat{v}_{it}  &  =y_{it}-\boldsymbol{\hat{\alpha}}_{CCE,i}^{\prime}%
\mathbf{d}_{t}-\hat{\boldsymbol{\beta}}_{CCE,i}^{\prime}\mathbf{x}%
_{it}\nonumber\\
&  =y_{it}-\boldsymbol{\alpha}_{i}^{\prime}\mathbf{d}_{t}-\boldsymbol{\beta
}_{i}^{^{\prime}}\mathbf{x}_{it}-\left(  \boldsymbol{\hat{\alpha}}%
_{CCE,i}-\boldsymbol{\alpha}_{i}\right)  ^{\prime}\mathbf{d}_{t}-\left(
\hat{\boldsymbol{\beta}}_{CCE,i}-\boldsymbol{\beta}_{i}\right)  ^{^{\prime}%
}\mathbf{x}_{it}\nonumber\\
&  =v_{it}-\left(  \boldsymbol{\hat{\alpha}}_{CCE,i}-\boldsymbol{\alpha}%
_{i}\right)  ^{\prime}\mathbf{d}_{t}-\left(  \hat{\boldsymbol{\beta}}%
_{CCE,i}-\boldsymbol{\beta}_{i}\right)  ^{^{\prime}}\mathbf{x}_{it}.
\label{vithat}%
\end{align}
We need to show that solving problem (\ref{Bproblem}) is asymptotically
equivalent to solving problem (\ref{Aproblem}). First, using the results in
\cite{pesaran2011large} and the fact that $\mathbf{d}_{t}$\ and $\mathbf{x}%
_{it}$ are (stochastically) bounded\footnote{See equation (31) in
\cite{pesaran2011large}.},%
\begin{align}
\mathbf{d}_{t}^{\prime}\left(  \boldsymbol{\hat{\alpha}}_{CCE,i}%
-\boldsymbol{\alpha}_{i}\right)   &  =O_{p}\left(  \frac{1}{\sqrt{T}}\right)
+O_{p}\left(  \frac{1}{n}\right)  +O_{p}\left(  \frac{1}{\sqrt{nT}}\right)
,\label{diff_alpha}\\
\mathbf{x}_{it}^{^{\prime}}\left(  \hat{\boldsymbol{\beta}}_{CCE,i}%
-\boldsymbol{\beta}_{i}\right)   &  =O_{p}\left(  \frac{1}{\sqrt{T}}\right)
+O_{p}\left(  \frac{1}{n}\right)  +O_{p}\left(  \frac{1}{\sqrt{nT}}\right)  ,
\label{diff_beta}%
\end{align}
then rewrite the criterion for (\ref{Bproblem}) with (\ref{vithat}),%
\begin{align}
&  \frac{1}{nT}\sum_{i=1}^{n}\sum_{t=1}^{T}\left(  \hat{v}_{it}%
-\boldsymbol{\gamma}_{i}^{^{\prime}}\mathbf{f}_{t}\right)  ^{2}\nonumber\\
&  =\frac{1}{nT}\sum_{i=1}^{n}\sum_{t=1}^{T}\left(  v_{it}-\boldsymbol{\gamma
}_{i}^{^{\prime}}\mathbf{f}_{t}-\left(  \boldsymbol{\hat{\alpha}}%
_{CCE,i}-\boldsymbol{\alpha}_{i}\right)  ^{\prime}\mathbf{d}_{t}-\left(
\hat{\boldsymbol{\beta}}_{CCE,i}-\boldsymbol{\beta}_{i}\right)  ^{^{\prime}%
}\mathbf{x}_{it}\right)  ^{2}\nonumber\\
&  =\frac{1}{nT}\sum_{i=1}^{n}\sum_{t=1}^{T}\left(  v_{it}-\boldsymbol{\gamma
}_{i}^{^{\prime}}\mathbf{f}_{t}\right)  ^{2}+\frac{1}{nT}\sum_{i=1}^{n}%
\sum_{t=1}^{T}\left(  \boldsymbol{\hat{\alpha}}_{CCE,i}-\boldsymbol{\alpha
}_{i}\right)  ^{\prime}\mathbf{d}_{t}\mathbf{d}_{t}^{\prime}\left(
\boldsymbol{\hat{\alpha}}_{CCE,i}-\boldsymbol{\alpha}_{i}\right) \nonumber\\
&  +\frac{1}{nT}\sum_{i=1}^{n}\sum_{t=1}^{T}\left(  \hat{\boldsymbol{\beta}%
}_{CCE,i}-\boldsymbol{\beta}_{i}\right)  ^{^{\prime}}\mathbf{x}_{it}%
\mathbf{x}_{it}^{^{\prime}}\left(  \hat{\boldsymbol{\beta}}_{CCE,i}%
-\boldsymbol{\beta}_{i}\right)  -\frac{2}{nT}\sum_{i=1}^{n}\sum_{t=1}%
^{T}\left(  v_{it}-\boldsymbol{\gamma}_{i}^{^{\prime}}\mathbf{f}_{t}\right)
\left(  \boldsymbol{\hat{\alpha}}_{CCE,i}-\boldsymbol{\alpha}_{i}\right)
^{\prime}\mathbf{d}_{t}\nonumber\\
&  -\frac{2}{nT}\sum_{i=1}^{n}\sum_{t=1}^{T}\left(  v_{it}-\boldsymbol{\gamma
}_{i}^{^{\prime}}\mathbf{f}_{t}\right)  \left(  \hat{\boldsymbol{\beta}%
}_{CCE,i}-\boldsymbol{\beta}_{i}\right)  ^{^{\prime}}\mathbf{x}_{it}-\frac
{2}{nT}\sum_{i=1}^{n}\sum_{t=1}^{T}\left(  \boldsymbol{\hat{\alpha}}%
_{CCE,i}-\boldsymbol{\alpha}_{i}\right)  ^{\prime}\mathbf{d}_{t}%
\mathbf{x}_{it}^{^{\prime}}\left(  \hat{\boldsymbol{\beta}}_{CCE,i}%
-\boldsymbol{\beta}_{i}\right) \nonumber\\
&  =\sum_{j=1}^{6}A_{j,nT}. \label{ABC}%
\end{align}
Therefore, using (\ref{diff_alpha}) and (\ref{diff_beta}), then
$A_{2,nT}=O_{p}\left(  \frac{1}{\sqrt{T}}\right)  +O_{p}\left(  \frac{1}%
{n}\right)  +O_{p}\left(  \frac{1}{\sqrt{nT}}\right)  $ and $A_{3,nT}%
=O_{p}\left(  \frac{1}{\sqrt{T}}\right)  +O_{p}\left(  \frac{1}{n}\right)
+O_{p}\left(  \frac{1}{\sqrt{nT}}\right)  $. Also, consider the fourth term of
(\ref{ABC}) and note that by Cauchy-Schwarz inequality,
\begin{align*}
\left\vert A_{4,nT}\right\vert  &  =2\left\vert \frac{1}{nT}\sum_{i=1}^{n}%
\sum_{t=1}^{T}u_{it}\left(  \boldsymbol{\hat{\alpha}}_{CCE,i}%
-\boldsymbol{\alpha}_{i}\right)  ^{\prime}\mathbf{d}_{t}\right\vert \\
&  \leq2\left(  \frac{1}{nT}\sum_{i=1}^{n}\sum_{t=1}^{T}u_{it}^{2}\right)
^{1/2}\left(  \frac{1}{nT}\sum_{i=1}^{n}\sum_{t=1}^{T}\left\Vert \left(
\boldsymbol{\hat{\alpha}}_{CCE,i}-\boldsymbol{\alpha}_{i}\right)  ^{\prime
}\mathbf{d}_{t}\right\Vert ^{2}\right)  ^{1/2},
\end{align*}
where $u_{it}=v_{it}-\boldsymbol{\gamma}_{i}^{^{\prime}}\mathbf{f}_{t}%
=y_{it}-\boldsymbol{\alpha}_{i}^{\prime}\mathbf{d}_{t}-\boldsymbol{\beta}%
_{i}^{^{\prime}}\mathbf{x}_{it}-\boldsymbol{\gamma}_{i}^{^{\prime}}%
\mathbf{f}_{t}$, so given (\ref{diff_alpha}) we have $A_{4,nT}=O_{p}\left(
\frac{1}{\sqrt{T}}\right)  +O_{p}\left(  \frac{1}{n}\right)  +O_{p}\left(
\frac{1}{\sqrt{nT}}\right)  $. Similarly, we can show $A_{5,nT}$ and
$A_{6,nT}$ share the same probability order as $A_{4,nT}$. Since in both
optimization problems $\boldsymbol{\gamma}_{i}$ and $\mathbf{f}_{t}$ are only
identified up to $m_{0}\times m_{0}$ rotation matrices, it follows that%
\begin{align*}
\min_{\boldsymbol{\Gamma,F}}\frac{1}{nT}\sum_{i=1}^{n}\sum_{t=1}^{T}\left(
v_{it}-\boldsymbol{\gamma}_{i}^{^{\prime}}\mathbf{f}_{t}\right)  ^{2}  &
\equiv\min_{\boldsymbol{\Gamma,F}}\frac{1}{nT}\sum_{i=1}^{n}\sum_{t=1}%
^{T}\left(  \hat{v}_{it}-\boldsymbol{\gamma}_{i}^{^{\prime}}\mathbf{f}%
_{t}\right)  ^{2}\\
&  +O_{p}\left(  \frac{1}{\sqrt{T}}\right)  +O_{p}\left(  \frac{1}{n}\right)
+O_{p}\left(  \frac{1}{\sqrt{nT}}\right)  .
\end{align*}
Hence, PCs based on $\hat{v}_{it}$ are asymptotically equivalent to those
based on $v_{it}$. The remaining proof of Theorem \ref{Theorem:CDss} 
follows from the proof of Theorem \ref{Theorem:CDs}.

{
\footnotesize
\bibliographystyle{apalike}

}

\clearpage\newpage{\footnotesize
\normalsize
}

\begin{center}
\bigskip

{\large Supplementary Material}

\bigskip

{\large \textbf{How to Detect Network Dependence in Latent Factor Models? A
Bias-Corrected CD Test}}

{\Large {\large \bigskip} }

{\large by }

{\Large \bigskip}

{\large M. Hashem Pesaran and Yimeng Xie}

{\Large \bigskip}

\today

\end{center}

\bigskip%

\renewcommand\theequation{S.\arabic{equation}}
\setcounter{equation}{0}
\renewcommand\thesection{S\arabic{section}}
\setcounter{section}{0}
\renewcommand\thepage{S\arabic{page}}
\setcounter{page}{1}
\setcounter{theorem}{1}
\renewcommand{\thefootnote}{S\arabic{footnote}}
\setcounter{footnote}{0}
\renewcommand{\thetable}{S.\arabic{table}}
\setcounter{table}{0}
\renewcommand\thefigure{S.\arabic{figure}}
\setcounter{figure}{0}
\renewcommand\thelemma{S.\arabic{lemma}}
\setcounter{lemma}{0}
\renewcommand\theremark{S.\arabic{remark}}
\setcounter{remark}{0}%

This supplement is in four sections. Section \ref{section.s2} states
and establishes the auxiliary lemmas used in the proofs of propositions and
theorems in the paper. Section \ref{section.thetan} derives the order of
$\theta_{n}$, defined by (\ref{thetan}) in the paper, in terms of the factor
strengths. Section \ref{section.JR} considers the CD$_{W+}$ test proposed by
Juodis and Reese (2022), and discusses some of its properties. Section
\ref{section.s3} reports simulation results for the experiments discussed in
Section \ref{mc} of the main paper.

\section{Statement and proofs of the lemmas\label{section.s2}}

This section provides auxiliary lemmas and the associated proofs, which are
required to establish the main results of the paper.

\begin{lemma}
\label{lm1}The CD statistic defined by (\ref{cd}) can be written equivalently
as,%
\begin{equation}
CD=\left(  \sqrt{\frac{n}{n-1}}\right)  \frac{1}{\sqrt{2T}}\sum_{t=1}%
^{T}\left[  \left(  \frac{1}{\sqrt{n}}\sum_{i=1}^{n}\frac{\hat{u}_{it}}%
{\hat{\sigma}_{i,T}}\right)  ^{2}-1\right]  . \label{CDlm1}%
\end{equation}

\end{lemma}

\begin{proof}
Using $\hat{\rho}_{ij,T}=\left(  \frac{1}{T}\sum_{t=1}^{T}\hat{u}_{it}\hat
{u}_{jt}\right)  /\hat{\sigma}_{i,T}\hat{\sigma}_{j,T}$ in (\ref{cd}) we
have:
\begin{equation}
CD=\sqrt{\frac{2T}{n(n-1)}}\sum_{i=1}^{n-1}\sum_{j=i+1}^{n}\frac{\frac{1}%
{T}\sum_{t=1}^{T}\hat{u}_{it}\hat{u}_{jt}}{\hat{\sigma}_{i,T}\hat{\sigma
}_{j,T}}=\sqrt{\frac{2T}{n(n-1)}}\frac{1}{T}\sum_{t=1}^{T}\left(  \sum
_{i=1}^{n-1}\sum_{j=i+1}^{n}\left(  \frac{\hat{u}_{it}}{\hat{\sigma}_{i,T}%
}\right)  \left(  \frac{\hat{u}_{jt}}{\hat{\sigma}_{j,T}}\right)  \right)  .
\label{lm1.1}%
\end{equation}
Further, we note that%
\[
\frac{1}{n}\sum_{i=1}^{n-1}\sum_{j=i+1}^{n}\left(  \frac{\hat{u}_{it}}%
{\hat{\sigma}_{i,T}}\right)  \left(  \frac{\hat{u}_{jt}}{\hat{\sigma}_{j,T}%
}\right)  =\frac{1}{2}\left[  \left(  \frac{1}{\sqrt{n}}\sum_{i=1}^{n}%
\frac{\hat{u}_{it}}{\hat{\sigma}_{i,T}}\right)  ^{2}-\frac{1}{n}\sum_{i=1}%
^{n}\left(  \frac{\hat{u}_{it}}{\hat{\sigma}_{i,T}}\right)  ^{2}\right]  .
\]
Then using this result in (\ref{lm1.1}), and after some algebra, we have
\begin{align*}
CD  &  =\sqrt{\frac{2Tn^{2}}{n(n-1)}}\frac{1}{2T}\sum_{t=1}^{T}\left[  \left(
\frac{1}{\sqrt{n}}\sum_{i=1}^{n}\frac{\hat{u}_{it}}{\hat{\sigma}_{i,T}%
}\right)  ^{2}-\frac{1}{n}\sum_{i=1}^{n}\left(  \frac{\hat{u}_{it}}%
{\hat{\sigma}_{i,T}}\right)  ^{2}\right] \\
&  =\sqrt{\frac{2Tn^{2}}{n(n-1)}}\frac{1}{2}\left[  \frac{1}{T}\sum_{t=1}%
^{T}\left(  \frac{1}{\sqrt{n}}\sum_{i=1}^{n}\frac{\hat{u}_{it}}{\hat{\sigma
}_{i,T}}\right)  ^{2}-\frac{1}{n}\sum_{i=1}^{n}\frac{1}{T}\sum_{t=1}%
^{T}\left(  \frac{\hat{u}_{it}}{\hat{\sigma}_{i,T}}\right)  ^{2}\right] \\
&  =\left(  \sqrt{\frac{n}{n-1}}\right)  \frac{1}{\sqrt{2T}}\sum_{t=1}%
^{T}\left[  \left(  \frac{1}{\sqrt{n}}\sum_{i=1}^{n}\frac{\hat{u}_{it}}%
{\hat{\sigma}_{i,T}}\right)  ^{2}-1\right]  ,
\end{align*}
as required.
\end{proof}

\begin{lemma}
\label{lm2} Consider the latent factor model given by (\ref{mod2}) and
(\ref{uit:p}). The latent factors, $\mathbf{f}_{t}$, and their loadings,
$\boldsymbol{\gamma}_{i}$, are estimated by principal components,
$\mathbf{\hat{f}}_{t}$ and $\boldsymbol{\hat{\gamma}}_{i}$, given by
(\ref{hatg}). Suppose that Assumptions \ref{ass:factors}-\ref{ass:ws} hold and
$\left(  n,T\right)  \rightarrow\infty$, such that $n/T\rightarrow\kappa\,$
for $0<\kappa<\infty$. Then
\begin{align}
\left\Vert \mathbf{\hat{F}}-\mathbf{F}\right\Vert _{F}  &  =O_{p}\left(
\frac{\sqrt{T}}{\delta_{nT}}\right)  ,\label{lm2.1}\\
\left\Vert \boldsymbol{\hat{\Gamma}}-\boldsymbol{\Gamma}\right\Vert _{F}  &
=O_{p}\left(  \frac{\sqrt{n}}{\delta_{nT}}\right)  ,\label{lm2.2}\\
\left\Vert \mathbf{U}\left(  \lambda_{T}\right)  ^{^{\prime}}\mathbf{(\hat
{F}-F)}\right\Vert _{F}  &  =O_{p}\left(  \frac{\sqrt{nT}}{\delta_{nT}%
}\right)  ,\label{lm2.3}\\
\left\Vert \boldsymbol{\Gamma}^{\prime}(\boldsymbol{\hat{\Gamma}%
}-\boldsymbol{\Gamma})\right\Vert _{F}  &  =O_{p}\left(  \frac{n}{\delta_{nT}%
}\right)  ,\label{lm2.4}\\
\left\Vert \mathbf{F}^{\prime}(\mathbf{\hat{F}}-\mathbf{F})\right\Vert _{F}
&  =O_{p}\left(  \frac{T}{\delta_{nT}}\right)  ,\label{lm2.5}\\
\left(  \mathbf{\hat{F}-F}\right)  ^{^{\prime}}\mathbf{F}  &  =O_{p}\left(
\frac{T}{\delta_{nT}^{2}}\right)  ,\label{baib2}\\
\left(  \mathbf{\hat{F}-F}\right)  ^{^{\prime}}\mathbf{\hat{F}}  &
=O_{p}\left(  \frac{T}{\delta_{nT}^{2}}\right)  ,\label{baib3}\\
(\boldsymbol{\hat{\Gamma}}-\boldsymbol{\Gamma})^{\prime}\mathbf{u}_{\circ
t}\left(  \lambda_{T}\right)   &  =O_{p}\left(  \frac{n}{\delta_{nT}^{2}%
}\right)  , \label{baib11}%
\end{align}
where $\mathbf{F}=\left(  \mathbf{f}_{1},\mathbf{f}_{2},\ldots,\mathbf{f}%
_{T}\right)  ^{\prime}$, $\mathbf{\hat{F}}=\left(  \mathbf{\hat{f}}%
_{1},\mathbf{\hat{f}}_{2},\ldots,\mathbf{\hat{f}}_{T}\right)  ^{\prime}$,
$\boldsymbol{\Gamma}=\left(  \boldsymbol{\gamma}_{1},\boldsymbol{\gamma}%
_{2},...,\boldsymbol{\gamma}_{n}\right)  ^{\prime}$, $\boldsymbol{\hat{\Gamma
}}=\left(  \boldsymbol{\hat{\gamma}}_{1},\boldsymbol{\hat{\gamma}}%
_{2},...,\boldsymbol{\hat{\gamma}}_{n}\right)  ^{\prime}$, $\mathbf{U}\left(
\lambda_{T}\right)  =$ $(\mathbf{u}_{\circ1}\left(  \lambda_{T}\right)
,\mathbf{u}_{\circ2}\left(  \lambda_{T}\right)  ,...,\mathbf{u}_{\circ
T}\left(  \lambda_{T}\right)  )^{\prime}$, $\mathbf{u}_{\circ t}\left(
\lambda_{T}\right)  =(\sigma_{1}\varepsilon_{1t}\left(  \lambda_{T}\right)
,\sigma_{2}\varepsilon_{2t}\left(  \lambda_{T}\right)  ,...,\sigma
_{n}\varepsilon_{nt}\left(  \lambda_{T}\right)  )^{\prime}$, and
$\varepsilon_{it}\left(  \lambda_{T}\right)  =\varepsilon_{it}+\lambda_{T}%
\sum_{j=1}^{n}w_{ij}\varepsilon_{jt}$.
\end{lemma}

\begin{proof}
Since Assumptions \ref{ass:factors}-\ref{ass:ws} are a sub-set of assumptions
made by Bai (2003), so results (\ref{lm2.1}) to (\ref{lm2.4}), (\ref{baib2})
and (\ref{baib3}) follow directly from Lemmas B.1, B.2 and B.3, and Theorems 1
and 2 of Bai (2003). Results (\ref{lm2.5}) and (\ref{baib11}) can be
established analogously.
\end{proof}

\begin{lemma}
\label{lmsup0}Suppose that Assumptions \ref{ass:factors}-\ref{ass:loadings}
hold and $\left(  n,T\right)  \rightarrow\infty$, such that $n/T\rightarrow
\kappa\,,$ for $0<\kappa<\infty$. Then%
\begin{align}
\sup_{i}\left(  T^{-1}\left\Vert \boldsymbol{\varepsilon}_{i\circ}\right\Vert
^{2}\right)   &  =O_{p}\left(  1\right)  ,\label{sup_u2_v1}\\
\sup_{i}\left\Vert \frac{\mathbf{F}^{\prime}\boldsymbol{\varepsilon}_{i\circ}%
}{T}\right\Vert  &  =O_{p}\left(  \sqrt{\frac{\ln(n)}{T}}\right)
,\label{sup_fu_v1}\\
\sup_{t}\left\Vert \frac{\boldsymbol{\Gamma}^{\prime}\boldsymbol{\varepsilon
}_{\circ t}}{n}\right\Vert  &  =O_{p}\left(  \sqrt{\frac{\ln(T)}{n}}\right)
,\label{sup_gu_v1}\\
\sup_{i}\left\Vert \frac{1}{nT}\sum_{t=1}^{T}\sum_{j=1}^{n}\sigma
_{j}\boldsymbol{\gamma}_{j}\varepsilon_{it}\varepsilon_{jt}\right\Vert  &
=O_{p}\left(  \sqrt{\frac{\ln\left(  n\right)  }{nT}}\right)  ,
\label{sup_gujui_v1}%
\end{align}
where $\boldsymbol{\varepsilon}_{i\circ}=\left(  \varepsilon_{i1}%
,\varepsilon_{i2},\ldots,\varepsilon_{iT}\right)  ^{\prime}$ and
$\boldsymbol{\varepsilon}_{\circ t}=\left(  \varepsilon_{1t},\varepsilon
_{2t},\ldots,\varepsilon_{nt}\right)  ^{\prime}$.
\end{lemma}

\begin{proof}
Consider (\ref{sup_u2_v1}) and note
\[
T^{-1}\left\Vert \boldsymbol{\varepsilon}_{i\circ}\right\Vert ^{2}=\frac{1}%
{T}\sum_{t=1}^{T}\left[  \varepsilon_{it}^{2}-E\left(  \varepsilon_{it}%
^{2}\right)  \right]  +\frac{1}{T}\sum_{t=1}^{T}E\left(  \varepsilon_{it}%
^{2}\right)  =\frac{1}{T}\sum_{t=1}^{T}z_{it}+1,
\]
where $z_{it}=\varepsilon_{it}^{2}-E\left(  \varepsilon_{it}^{2}\right)  $.
Then
\begin{equation}
\sup_{i}\left(  T^{-1}\left\Vert \boldsymbol{\varepsilon}_{i\circ}\right\Vert
^{2}\right)  \leq\sup_{i}\left\vert \frac{1}{T}\sum_{t=1}^{T}z_{it}\right\vert
+1. \label{u2_part}%
\end{equation}
To establish the probability of the first term, consider the filtration
$\mathcal{I}_{it}^{\left(  1\right)  }=\{\varepsilon_{i\tau}:\tau
=t-1,t-2,\ldots\}$ and, given the serial independence of $\varepsilon_{it}$,
note that $E\left(  z_{it}|\mathcal{I}_{i,t-1}^{\left(  1\right)  }\right)
=0$, so $z_{it}$ is a martingale difference process with respect to
$\mathcal{I}_{i,t-1}^{\left(  1\right)  }$. In addition, $Var\left(
z_{it}\right)  =Var\left(  \varepsilon_{it}^{2}\right)  =E\left(
\varepsilon_{it}^{4}\right)  -\left(  E\varepsilon_{it}^{2}\right)  ^{2}$
which is bounded by assumption. Also, as $\varepsilon_{it}$ is sub-exponential
by part (a) of Assumption \ref{ass:errors}, $\varepsilon_{it}^{2}$ (and hence
$z_{it}$) is sub-exponential, and there exist positive constants $C_{4}$,
$C_{5}$ and $r_{3}$ such that
\[
\sup_{i}\Pr\left(  \left\vert z_{it}\right\vert >a\right)  \leq C_{4}%
\exp\left(  -C_{5}a^{r_{3}}\right)  ,\text{ for all }a>0.\text{ }%
\]
Then by Lemma A3 in the online theory supplement of Chudik et al. (2018), for
$\varsigma_{T}=\ominus\left(  T^{\mu}\right)  $ and $0<\mu<\left(
r_{3}+1\right)  /\left(  r_{3}+2\right)  $, there exists a positive constant
$C_{6}$ such that$\Pr\left(  \left\vert \sum_{t=1}^{T}z_{it}\right\vert
>\varsigma_{T}\right)  \leq\exp\left(  -C_{6}T^{-1}\varsigma_{T}^{2}\right)
$, and if $\mu>\left(  r_{3}+1\right)  /\left(  r_{3}+2\right)  $ there exists
a positive constant $C_{7}$ such that $\Pr\left(  \left\vert \sum_{t=1}%
^{T}z_{it}\right\vert >\varsigma_{T}\right)  \leq\exp\left(  -C_{7}\left(
\varsigma_{T}\right)  ^{\frac{r_{3}}{r_{3}+1}}\right)  . $ By Boole's
inequality, we have%
\[
\Pr\left(  \sup_{i}\left\vert \sum_{t=1}^{T}z_{it}\right\vert >\varsigma
_{T}\right)  \leq\exp\left(  \ln\left(  n\right)  -C_{6}T^{-1}\varsigma
_{T}^{2}\right)  ,\text{ if }0<\mu<\left(  r_{3}+1\right)  /\left(
r_{3}+2\right)  ,
\]
\[
\Pr\left(  \sup_{i}\left\vert \sum_{t=1}^{T}z_{it}\right\vert >\varsigma
_{T}\right)  \leq\exp\left(  \ln\left(  n\right)  -C_{7}\left(  \varsigma
_{T}\right)  ^{\frac{r_{3}}{r_{3}+1}}\right)  ,\text{ if }\mu>\left(
r_{3}+1\right)  /\left(  r_{3}+2\right)  .
\]
Let $\varsigma_{T}=C_{8}\sqrt{T\ln\left(  n\right)  }$ where $C_{8}$ is a
finite but sufficiently large constant. Then for $0<\mu<\left(  r_{3}%
+1\right)  /\left(  r_{3}+2\right)  $, we have%
\begin{align*}
\Pr\left(  \sup_{i}\left\vert \frac{1}{T}\sum_{t=1}^{T}z_{it}\right\vert
>C_{8}\sqrt{\frac{\ln\left(  n\right)  }{T}}\right)   &  =\Pr\left(  \sup
_{i}\left\vert \sum_{t=1}^{T}\varepsilon_{it}-E\left(  \varepsilon_{it}%
^{2}\right)  \right\vert >C_{8}\sqrt{T\ln\left(  n\right)  }\right) \\
&  \leq\exp\left[  \ln\left(  n\right)  -C_{6}T^{-1}C_{8}^{2}T\ln\left(
n\right)  \right] \\
&  =\exp\left[  \ln\left(  n\right)  -C_{6}C_{8}^{2}\ln\left(  n\right)
\right]  ,
\end{align*}
which is $o\left(  1\right)  $ given $C_{8}$ is sufficiently large. Also for
$\mu\geq\left(  r_{3}+1\right)  /\left(  r_{3}+2\right)  $, we have%
\begin{align*}
\Pr\left(  \sup_{i}\left\vert \frac{1}{T}\sum_{t=1}^{T}\varepsilon
_{it}\right\vert >C_{8}\sqrt{\frac{\ln\left(  n\right)  }{T}}\right)   &
=\Pr\left(  \sup_{i}\left\vert \sum_{t=1}^{T}\varepsilon_{it}\right\vert
>C_{8}\sqrt{T\ln\left(  n\right)  }\right) \\
&  \leq\exp\left[  \ln\left(  n\right)  -C_{7}\left(  C_{8}\sqrt{T\ln\left(
n\right)  }\right)  ^{\frac{r_{3}}{r_{3}+1}}\right]  ,
\end{align*}
which is also $o\left(  1\right)  $ as$\ n$ and $T$ are of the same order of
magnitude and sufficiently we have%
\[
\frac{\ln\left(  n\right)  }{\left[  \sqrt{n\ln\left(  n\right)  }\right]
^{\frac{r_{3}}{r_{3}+1}}}=\frac{\left[  \ln\left(  n\right)  \right]
^{\frac{r_{3}+2}{2\left(  r_{3}+1\right)  }}}{n^{\frac{r_{3}}{2\left(
r_{3}+1\right)  }}}=\left[  \frac{\left(  \ln\left(  n\right)  \right)
^{\frac{r_{3}+2}{r_{3}}}}{n}\right]  ^{\frac{r_{3}}{2\left(  r_{3}+1\right)
}}\rightarrow0.
\]
Therefore, $\sup_{i}\left\vert T^{-1}\sum_{t=1}^{T}z_{it}\right\vert
=O_{p}\left(  \sqrt{\frac{\ln\left(  n\right)  }{T}}\right)  $ and
(\ref{sup_u2}) follows from (\ref{u2_part}). Next, consider (\ref{sup_fu_v1})
and note by Assumptions \ref{ass:factors} and \ref{ass:errors}, $\mathbf{f}%
_{t}$ is independent from $\varepsilon_{it^{\prime}}$ for all $t,t^{\prime
}=1,2,\ldots,T$, also $\varepsilon_{it}$ is serially independent, then for a
suitable choice of $\mathcal{I}_{i,t-1}^{\left(  2\right)  }=\{\mathbf{f}%
_{\tau}\varepsilon_{i\tau}:$ $\tau=t-1,t-2,\ldots\}$ and $i=1,2,\ldots,n$,
$E\left(  \mathbf{f}_{t}\varepsilon_{it}|\mathcal{I}_{i,t-1}^{\left(
2\right)  }\right)  =E\left(  \mathbf{f}_{t}|\mathcal{I}_{i,t-1}^{\left(
2\right)  }\right)  E\left(  \varepsilon_{it}\right)  =\mathbf{0}$ so
$\mathbf{f}_{t}\varepsilon_{it}$ is a martingale difference sequence with
respect to the filtration $\mathcal{I}_{i,t-1}^{\left(  2\right)  }$. In
addition, $E\left(  \mathbf{f}_{t}\varepsilon_{it}\right)  =E\left(
\mathbf{f}_{t}\right)  E\left(  \varepsilon_{it}\right)  =\mathbf{0}$ and
$Var\left(  \mathbf{f}_{t}\varepsilon_{it}\right)  =E\left(  \mathbf{f}%
_{t}\mathbf{f}_{t}^{\prime}\right)  E\left(  \varepsilon_{it}^{2}\right)  $,
which is bounded by Assumptions \ref{ass:factors} and \ref{ass:errors}. Also
by assumptions both $\mathbf{f}_{t}$ and $\varepsilon_{it}$ are
sub-exponential, then it also follows that $\mathbf{f}_{t}\varepsilon_{it}$ is
sub-exponential. Hence, the method of proof used above can also be applied to
$\left\Vert \sum_{t=1}^{T}\mathbf{f}_{t}\varepsilon_{it}\right\Vert $, and
result (\ref{sup_fu_v1}) follows. Similarly, (\ref{sup_gu_v1}) can be
established by the symmetry of the standard factor models in
$\boldsymbol{\gamma}_{i}$ and $\mathbf{f}_{t}$. Now consider
(\ref{sup_gujui_v1}), and note that we have the following decomposition,
\begin{align*}
\mathbf{q}_{i,nT}  &  =\frac{1}{nT}\sum_{t=1}^{T}\sum_{j=1}^{n}\sigma
_{j}\boldsymbol{\gamma}_{j}\varepsilon_{it}\varepsilon_{jt}\\
&  =\frac{1}{nT}\sum_{t=1}^{T}\sum_{j=1}^{n}\left[  \sigma_{j}%
\boldsymbol{\gamma}_{j}\varepsilon_{it}\varepsilon_{jt}-\sigma_{j}%
\boldsymbol{\gamma}_{j}E\left(  \varepsilon_{it}\varepsilon_{jt}\right)
\right]  +\frac{1}{nT}\sum_{t=1}^{T}\sum_{j=1}^{n}\sigma_{j}\boldsymbol{\gamma
}_{j}E\left(  \varepsilon_{it}\varepsilon_{jt}\right)  =\mathbf{q}%
_{i,nT}^{(a)}+\mathbf{q}_{i,nT}^{(b)}.
\end{align*}
Since $E(\varepsilon_{it}\varepsilon_{jt})=0$ if $i\neq j$, and $E(\varepsilon
_{it}\varepsilon_{jt})=1$, if $i=j$, then $\mathbf{q}_{i,nT}^{(b)}=\frac
{1}{nT}\sum_{t=1}^{T}\sum_{j=1}^{n}\sigma_{j}\boldsymbol{\gamma}_{j}E\left(
\varepsilon_{it}\varepsilon_{jt}\right)  =n^{-1}\sigma_{j}\boldsymbol{\gamma
}_{j}, $ and $\sup_{i}\left\Vert \mathbf{q}_{i,nT}^{(b)}\right\Vert
=O(n^{-1})$. Consider now the first term and note that $\mathbf{q}%
_{i,nT}^{(a)}=\frac{1}{nT}\sum_{t=1}^{T}\sum_{j=1}^{n}\mathbf{s}_{i,jt}, $
where $\mathbf{s}_{i,jt}=\sigma_{j}\boldsymbol{\gamma}_{j}\left[
\varepsilon_{it}\varepsilon_{jt}-E\left(  \varepsilon_{it}\varepsilon
_{jt}\right)  \right]  $. Since by assumption $\varepsilon_{it}$ are
independently distributed over all $i$ and $t$, then $E\left(  \mathbf{s}%
_{i,jt}\left\vert \mathcal{I}_{i,t-1}^{\left(  3\right)  }\right.  \right)
=\mathbf{0}$, where $\mathcal{I}_{i,t-1}^{\left(  3\right)  }=\{\varepsilon
_{i\tau}\varepsilon_{j\tau}:$ $j=1,2,...,n$ and $\tau=t-1,t-2,...\}$. Hence
$\mathbf{s}_{i,jt}$ is a martingale difference process with respect to the
filtration, $\mathcal{I}_{i,t-1}^{\left(  3\right)  }$. The variance of
$\mathbf{s}_{i,jt}$ is $\left(  \sigma_{j}^{2}\boldsymbol{\gamma}%
_{j}\boldsymbol{\gamma}_{j}^{\prime}\right)  Var(\varepsilon_{it}%
\varepsilon_{jt})$ where $Var(\varepsilon_{it}\varepsilon_{jt})=1$ if $i\neq
j$, and $Var(\varepsilon_{it}\varepsilon_{jt})=Var(\varepsilon_{it}%
^{2})=E(\varepsilon_{it}^{4})-1$ if $i=j$, so that by Assumption
\ref{ass:errors}, $\left\Vert Var(\mathbf{s}_{i,jt})\right\Vert <C$. Also,
since by assumption $\varepsilon_{it}$ is sub-exponential, then it follows
that $\mathbf{s}_{i,jt}$ is also sub-exponential, and the above method of
proof can be applied to all elements of $\mathbf{q}_{i,nT}^{(a)}$.
Specifically $\sup_{i}\left\Vert \mathbf{q}_{i,nT}^{(a)}\right\Vert
=O_{p}\left(  \sqrt{\frac{\ln(n)}{nT}}\right)  $, $\sup_{i}\left\Vert
\mathbf{q}_{i,nT}\right\Vert =O_{p}\left(  \sqrt{\frac{\ln(n)}{nT}}\right)
+O(n^{-1})=O_{p}\left(  \sqrt{\frac{\ln(n)}{nT}}\right)  $, and result
(\ref{sup_gujui_v1}) follows, as required.\textbf{ }
\end{proof}

\begin{lemma}
\label{lmsup}Consider the latent factor model given by (\ref{mod2}) and
(\ref{uit:p}). The latent factors, $\mathbf{f}_{t}$, and their loadings,
$\boldsymbol{\gamma}_{i}$, are estimated by principal components,
$\mathbf{\hat{f}}_{t}$ and $\boldsymbol{\hat{\gamma}}_{i}$, given by
(\ref{hatg}). Suppose that Assumptions \ref{ass:factors}-\ref{ass:ws} hold and
$\left(  n,T\right)  \rightarrow\infty$, such that $n/T\rightarrow\kappa\,,$
for $0<\kappa<\infty$. Then
\begin{align}
\sup_{i}\left(  T^{-1}\left\Vert \boldsymbol{\varepsilon}_{i\circ}\left(
\lambda_{T}\right)  \right\Vert ^{2}\right)   &  =O_{p}\left(  1\right)
,\label{sup_u2}\\
\sup_{i}\left\Vert \frac{\mathbf{F}^{\prime}\boldsymbol{\varepsilon}_{i\circ
}\left(  \lambda_{T}\right)  }{T}\right\Vert  &  =O_{p}\left(  \sqrt{\frac
{\ln(n)}{T}}\right)  ,\label{sup_fu}\\
\sup_{t}\left\Vert \frac{\boldsymbol{\Gamma}^{\prime}\boldsymbol{\varepsilon
}_{\circ t}\left(  \lambda_{T}\right)  }{n}\right\Vert  &  =O_{p}\left(
\sqrt{\frac{\ln(T)}{n}}\right)  ,\label{sup_gu}\\
\sup_{i}\left\Vert \frac{1}{nT}\sum_{t=1}^{T}\sum_{j=1}^{n}\sigma
_{j}\boldsymbol{\gamma}_{j}\varepsilon_{it}\left(  \lambda_{T}\right)
\varepsilon_{jt}\left(  \lambda_{T}\right)  \right\Vert  &  =O_{p}\left(
\sqrt{\frac{\ln\left(  n\right)  }{nT}}\right)  ,\label{sup_gujui}\\
\sup_{i}\left\Vert \boldsymbol{\hat{\gamma}}_{i}-\boldsymbol{\gamma}%
_{i}\right\Vert  &  =O_{p}\left(  \sqrt{\frac{\ln(n)}{T}}\right)
,\label{sup_|hatg-g|}\\
\sup_{t}\left\Vert \mathbf{\hat{f}}_{t}-\mathbf{f}_{t}\right\Vert  &
=O_{p}\left(  \sqrt{\frac{\ln(T)}{n}}\right)  , \label{sup_|hatf-f|}%
\end{align}
where $\boldsymbol{\varepsilon}_{i\circ}\left(  \lambda_{T}\right)  =\left(
\varepsilon_{i1}\left(  \lambda_{T}\right)  ,\varepsilon_{i2}\left(
\lambda_{T}\right)  ,\ldots,\varepsilon_{iT}\left(  \lambda_{T}\right)
\right)  ^{\prime}$ and $\boldsymbol{\varepsilon}_{\circ t}\left(  \lambda
_{T}\right)  =\left(  \varepsilon_{1t}\left(  \lambda_{T}\right)
,\varepsilon_{2t}\left(  \lambda_{T}\right)  ,\ldots,\varepsilon_{nt}\left(
\lambda_{T}\right)  \right)  ^{\prime}$.
\end{lemma}

\begin{proof}
Consider (\ref{sup_u2}) and note by definition
\begin{equation}
\varepsilon_{it}\left(  \lambda_{T}\right)  =\varepsilon_{it}+\lambda
_{T}\mathbf{w}_{i0}^{\prime}\boldsymbol{\varepsilon}_{\circ t}, \label{epit}%
\end{equation}
where $\mathbf{w}_{i0}=\left(  w_{i1},w_{i2},\ldots,w_{in}\right)  ^{\prime}$
and $\boldsymbol{\varepsilon}_{\circ t}=\left(  \varepsilon_{1t}%
,\varepsilon_{2t},\ldots,\varepsilon_{nt}\right)  ^{\prime}$. Then%
\[
\varepsilon_{it}^{2}\left(  \lambda_{T}\right)  =\varepsilon_{it}^{2}%
+2\lambda_{T}\mathbf{w}_{i0}^{\prime}\boldsymbol{\varepsilon}_{\circ
t}\varepsilon_{it}+\lambda_{T}^{2}\mathbf{w}_{i0}^{\prime}%
\boldsymbol{\varepsilon}_{\circ t}\boldsymbol{\varepsilon}_{\circ t}^{\prime
}\mathbf{w}_{i0},
\]
and hence%
\[
\frac{1}{T}\sum_{t=1}^{T}\varepsilon_{it}^{2}\left(  \lambda_{T}\right)
=\frac{1}{T}\sum_{t=1}^{T}\varepsilon_{it}^{2}+2\lambda_{T}\mathbf{w}%
_{i0}^{\prime}\left(  \frac{1}{T}\sum_{t=1}^{T}\boldsymbol{\varepsilon}_{\circ
t}\varepsilon_{it}\right)  +\lambda_{T}^{2}\mathbf{w}_{i0}^{\prime}%
\mathbf{V}_{\varepsilon T}\mathbf{w}_{i0},
\]
where $\mathbf{V}_{\varepsilon T}=T^{-1}\sum_{t=1}^{T}\boldsymbol{\varepsilon
}_{\circ t}\boldsymbol{\varepsilon}_{\circ t}^{\prime}$ and $\left\Vert
\mathbf{V}_{\varepsilon T}\right\Vert =O_{p}\left(  1\right)  $ by part (b) of
Assumption \ref{ass:errors}. It follows%
\[
\left\vert \frac{1}{T}\sum_{t=1}^{T}\varepsilon_{it}^{2}\left(  \lambda
_{T}\right)  \right\vert \leq\left\vert \frac{1}{T}\sum_{t=1}^{T}%
\varepsilon_{it}^{2}\right\vert +2\left\vert \lambda_{T}\right\vert \left\Vert
\mathbf{w}_{i0}\right\Vert \left\Vert \frac{1}{T}\sum_{t=1}^{T}%
\boldsymbol{\varepsilon}_{\circ t}\varepsilon_{it}\right\Vert +\lambda_{T}%
^{2}\left\Vert \mathbf{w}_{i0}\right\Vert ^{2}\left\Vert \mathbf{V}%
_{\varepsilon T}\right\Vert .
\]
Denote $\mathbf{e}_{i}$ as $n\times1$ selection vector with $1$ on its
$i^{th}$ element and zeros elsewhere, and note that
\[
\sup_{i}\left\Vert \frac{1}{T}\sum_{t=1}^{T}\boldsymbol{\varepsilon}_{\circ
t}\varepsilon_{it}\right\Vert =\sup_{i}\left\Vert \frac{1}{T}\sum_{t=1}%
^{T}\boldsymbol{\varepsilon}_{\circ t}\boldsymbol{\varepsilon}_{\circ
t}^{\prime}\mathbf{e}_{i}\right\Vert \leq\left\Vert \frac{1}{T}\sum_{t=1}%
^{T}\boldsymbol{\varepsilon}_{\circ t}\boldsymbol{\varepsilon}_{\circ
t}^{\prime}\right\Vert \left(  \sup_{i}\left\Vert \mathbf{e}_{i}\right\Vert
\right)  =\left\Vert \mathbf{V}_{\varepsilon T}\right\Vert .
\]
Using this result we now have (recalling that $\lambda_{T}=c_{\lambda}%
T^{-1/2})$
\begin{align*}
\sup_{i}\left\vert \frac{1}{T}\sum_{t=1}^{T}\varepsilon_{it}^{2}\left(
\lambda_{T}\right)  \right\vert  &  \leq\sup_{i}\left\vert \frac{1}{T}%
\sum_{t=1}^{T}\varepsilon_{it}^{2}\right\vert +2c_{\lambda}T^{-1/2}\left\Vert
\mathbf{V}_{\varepsilon T}\right\Vert \left(  \sup_{i}\left\Vert
\mathbf{w}_{i0}\right\Vert \right)  +\\
&  c_{\lambda}^{2}T^{-1}\left\Vert \mathbf{V}_{\varepsilon T}\right\Vert
\left(  \sup_{i}\left\Vert \mathbf{w}_{i0}\right\Vert ^{2}\right)  .
\end{align*}
Therefore, since $\sup_{i}\left\Vert \mathbf{w}_{i0}\right\Vert <C$ and by
assumption $\left\Vert \mathbf{V}_{\varepsilon T}\right\Vert =O_{p}(1)$, then
\begin{equation}
\sup_{i}\left\vert \frac{1}{T}\sum_{t=1}^{T}\varepsilon_{it}^{2}\left(
\lambda_{T}\right)  \right\vert =\sup_{i}\left\vert \frac{1}{T}\sum_{t=1}%
^{T}\varepsilon_{it}^{2}\right\vert +O_{p}\left(  \frac{1}{\sqrt{T}}\right)  .
\label{sup:ave:eij}%
\end{equation}
Result (\ref{sup_u2}) now follows from (\ref{sup_u2_v1}). Similarly, to
establish (\ref{sup_fu}) note that
\[
\frac{\mathbf{F}^{\prime}\boldsymbol{\varepsilon}_{i\circ}\left(  \lambda
_{T}\right)  }{T}=\frac{1}{T}\sum_{t=1}^{T}\mathbf{f}_{t}\varepsilon
_{it}\left(  \lambda_{T}\right)  =\frac{1}{T}\sum_{t=1}^{T}\mathbf{f}%
_{t}\varepsilon_{it}+\frac{\lambda_{T}}{T}\sum_{t=1}^{T}\mathbf{f}%
_{t}\boldsymbol{\varepsilon}_{\circ t}^{\prime}\mathbf{w}_{i0},
\]
and%
\[
\left\Vert \frac{\mathbf{F}^{\prime}\boldsymbol{\varepsilon}_{i\circ}\left(
\lambda_{T}\right)  }{T}\right\Vert \leq\left\Vert \frac{1}{T}\sum_{t=1}%
^{T}\mathbf{f}_{t}\varepsilon_{it}\right\Vert +\left\vert \lambda
_{T}\right\vert \left\Vert \mathbf{w}_{i0}\right\Vert \left\Vert \frac{1}%
{T}\sum_{t=1}^{T}\mathbf{f}_{t}\boldsymbol{\varepsilon}_{\circ t}^{\prime
}\right\Vert .
\]
Applying the supremum operator to both sides yields
\[
\sup_{i}\left\Vert \frac{\mathbf{F}^{\prime}\boldsymbol{\varepsilon}_{i\circ
}\left(  \lambda_{T}\right)  }{T}\right\Vert \leq\sup_{i}\left\Vert \frac
{1}{T}\sum_{t=1}^{T}\mathbf{f}_{t}\varepsilon_{it}\right\Vert +\left\vert
\lambda_{T}\right\vert \left(  \sup_{i}\left\Vert \mathbf{w}_{i0}\right\Vert
\right)  \left\Vert \frac{1}{T}\sum_{t=1}^{T}\mathbf{f}_{t}%
\boldsymbol{\varepsilon}_{\circ t}^{\prime}\right\Vert .
\]
Also
\begin{align*}
E\left\Vert \frac{1}{T}\sum_{t=1}^{T}\mathbf{f}_{t}\boldsymbol{\varepsilon
}_{\circ t}^{\prime}\right\Vert ^{2}  &  \leq E\left\Vert \frac{1}{T}%
\sum_{t=1}^{T}\mathbf{f}_{t}\boldsymbol{\varepsilon}_{\circ t}^{\prime
}\right\Vert _{F}^{2}=\mathrm{tr}\left[  E\left[  \left(  \frac{1}{T}%
\sum_{t=1}^{T}\mathbf{f}_{t}\boldsymbol{\varepsilon}_{\circ t}^{\prime
}\right)  \left(  \frac{1}{T}\sum_{t=1}^{T}\mathbf{f}_{t}%
\boldsymbol{\varepsilon}_{\circ t}^{\prime}\right)  ^{\prime}\right]  \right]
\\
&  =\mathrm{tr}\left(  \frac{1}{T^{2}}\sum_{t=1}^{T}\sum_{t^{\prime}=1}%
^{T}E\left(  \mathbf{f}_{t}\boldsymbol{\varepsilon}_{\circ t}^{\prime
}\boldsymbol{\varepsilon}_{\circ t^{\prime}}\mathbf{f}_{t^{\prime}}^{\prime
}\right)  \right)  =\frac{1}{T^{2}}\mathrm{tr}\left(  \sum_{t=1}^{T}E\left(
\mathbf{f}_{t}\mathbf{f}_{t}^{\prime}\right)  E\left(  \boldsymbol{\varepsilon
}_{\circ t}^{\prime}\boldsymbol{\varepsilon}_{\circ t}\right)  \right) \\
&  =\frac{n}{T}\mathrm{tr}\left(  \mathbf{\Sigma}_{ff}\right)  =O\left(
\frac{nm_{0}}{T}\right)  ,
\end{align*}
which establishes that $T^{-1}\sum_{t=1}^{T}\mathbf{f}_{t}%
\boldsymbol{\varepsilon}_{\circ t}^{\prime}=O_{p}\left(  1\right)  $ since by
assumption $n$ and $T$ have the same orders of magnitudes. Given $\lambda
_{T}=c_{\lambda}T^{-1/2}$ and $\sup_{i}\left\Vert \mathbf{w}_{i0}\right\Vert
<C$, we now have%
\[
\sup_{i}\left\Vert \frac{\mathbf{F}^{\prime}\boldsymbol{\varepsilon}_{i\circ
}\left(  \lambda_{T}\right)  }{T}\right\Vert =\sup_{i}\left\Vert \frac{1}%
{T}\sum_{t=1}^{T}\mathbf{f}_{t}\varepsilon_{it}\right\Vert +O_{p}\left(
\frac{1}{\sqrt{T}}\right)  ,
\]
and (\ref{sup_fu}) follows using (\ref{sup_fu_v1}). Similarly, (\ref{sup_gu})
can be established using result (\ref{sup_gu_v1}). Next, consider
(\ref{sup_gujui}) and using the definition of $\varepsilon_{it}\left(
\lambda_{T}\right)  $ in (\ref{epit}) yields
\begin{align*}
\frac{1}{nT}\sum_{t=1}^{T}\sum_{j=1}^{n}\sigma_{j}\boldsymbol{\gamma}%
_{j}\varepsilon_{it}\left(  \lambda_{T}\right)  \varepsilon_{jt}\left(
\lambda_{T}\right)   &  =\frac{1}{nT}\sum_{t=1}^{T}\sum_{j=1}^{n}\sigma
_{j}\boldsymbol{\gamma}_{j}\varepsilon_{it}\varepsilon_{jt}+\frac{\lambda_{T}%
}{nT}\sum_{t=1}^{T}\sum_{j=1}^{n}\sigma_{j}\boldsymbol{\gamma}_{j}%
\mathbf{w}_{j0}^{\prime}\boldsymbol{\varepsilon}_{\circ t}\varepsilon_{it}+\\
&  \frac{\lambda_{T}}{nT}\sum_{t=1}^{T}\sum_{j=1}^{n}\sigma_{j}%
\boldsymbol{\gamma}_{j}\mathbf{w}_{i0}^{\prime}\boldsymbol{\varepsilon}_{\circ
t}\varepsilon_{jt}+\frac{\lambda_{T}^{2}}{n}\sum_{j=1}^{n}\sigma
_{j}\boldsymbol{\gamma}_{j}\mathbf{w}_{i0}^{\prime}\left(  \frac{1}{T}%
\sum_{t=1}^{T}\boldsymbol{\varepsilon}_{\circ t}\boldsymbol{\varepsilon
}_{\circ t}^{\prime}\right)  \mathbf{w}_{j0},
\end{align*}
which implies
\begin{align*}
\left\Vert \frac{1}{nT}\sum_{t=1}^{T}\sum_{j=1}^{n}\sigma_{j}%
\boldsymbol{\gamma}_{j}\varepsilon_{it}\left(  \lambda_{T}\right)
\varepsilon_{jt}\left(  \lambda_{T}\right)  \right\Vert  &  \leq\left\Vert
\frac{1}{nT}\sum_{t=1}^{T}\sum_{j=1}^{n}\sigma_{j}\boldsymbol{\gamma}%
_{j}\varepsilon_{it}\varepsilon_{jt}\right\Vert +\left\Vert \frac{\lambda_{T}%
}{nT}\sum_{t=1}^{T}\sum_{j=1}^{n}\sigma_{j}\boldsymbol{\gamma}_{j}%
\mathbf{w}_{j0}^{\prime}\boldsymbol{\varepsilon}_{\circ t}\varepsilon
_{it}\right\Vert +\\
&  \left\Vert \frac{\lambda_{T}}{nT}\sum_{t=1}^{T}\sum_{j=1}^{n}\sigma
_{j}\boldsymbol{\gamma}_{j}\mathbf{w}_{i0}^{\prime}\boldsymbol{\varepsilon
}_{\circ t}\varepsilon_{jt}\right\Vert +\left\Vert \frac{\lambda_{T}^{2}}%
{n}\sum_{j=1}^{n}\sigma_{j}\boldsymbol{\gamma}_{j}\mathbf{w}_{i0}^{\prime
}\left(  \frac{1}{T}\sum_{t=1}^{T}\boldsymbol{\varepsilon}_{\circ
t}\boldsymbol{\varepsilon}_{\circ t}^{\prime}\right)  \mathbf{w}%
_{j0}\right\Vert .
\end{align*}
Taking the supremum on both sides of this inequality yields%
\begin{align}
&  \sup_{i}\left\Vert \frac{1}{nT}\sum_{t=1}^{T}\sum_{j=1}^{n}\sigma
_{j}\boldsymbol{\gamma}_{j}\varepsilon_{it}\left(  \lambda_{T}\right)
\varepsilon_{jt}\left(  \lambda_{T}\right)  \right\Vert \nonumber\\
&  \leq\sup_{i}\left\Vert \frac{1}{nT}\sum_{t=1}^{T}\sum_{j=1}^{n}\sigma
_{j}\boldsymbol{\gamma}_{j}\varepsilon_{it}\varepsilon_{jt}\right\Vert
+\sup_{i}\left\Vert \frac{\lambda_{T}}{nT}\sum_{t=1}^{T}\sum_{j=1}^{n}%
\sigma_{j}\boldsymbol{\gamma}_{j}\mathbf{w}_{j0}^{\prime}%
\boldsymbol{\varepsilon}_{\circ t}\varepsilon_{it}\right\Vert +\nonumber\\
&  \sup_{i}\left\Vert \frac{\lambda_{T}}{nT}\sum_{t=1}^{T}\sum_{j=1}^{n}%
\sigma_{j}\boldsymbol{\gamma}_{j}\mathbf{w}_{i0}^{\prime}%
\boldsymbol{\varepsilon}_{\circ t}\varepsilon_{jt}\right\Vert +\sup
_{i}\left\Vert \frac{\lambda_{T}^{2}}{n}\sum_{j=1}^{n}\sigma_{j}%
\boldsymbol{\gamma}_{j}\mathbf{w}_{i0}^{\prime}\left(  \frac{1}{T}\sum
_{t=1}^{T}\boldsymbol{\varepsilon}_{\circ t}\boldsymbol{\varepsilon}_{\circ
t}^{\prime}\right)  \mathbf{w}_{j0}\right\Vert . \label{gujui}%
\end{align}
Note that $\varepsilon_{it}=\boldsymbol{\varepsilon}_{\circ t}^{\prime
}\mathbf{e}_{i}$ where $\mathbf{e}_{i}$ is an $n\times1$ selection vector with
$1$ on its $i^{th}$ element and zero elsewhere. Then the second term of the
above can be bounded as
\begin{align*}
\sup_{i}\left\Vert \frac{\lambda_{T}}{nT}\sum_{t=1}^{T}\sum_{j=1}^{n}%
\sigma_{j}\boldsymbol{\gamma}_{j}\mathbf{w}_{j0}^{\prime}%
\boldsymbol{\varepsilon}_{\circ t}\varepsilon_{it}\right\Vert  &  =\sup
_{i}\left\Vert \frac{\lambda_{T}}{n}\sum_{j=1}^{n}\sigma_{j}\boldsymbol{\gamma
}_{j}\mathbf{w}_{j0}^{\prime}\left(  \frac{1}{T}\sum_{t=1}^{T}%
\boldsymbol{\varepsilon}_{\circ t}\boldsymbol{\varepsilon}_{\circ t}^{\prime
}\right)  \mathbf{e}_{i}\right\Vert \\
&  \leq\left\vert \lambda_{T}\right\vert \left\Vert \frac{1}{n}\sum_{j=1}%
^{n}\sigma_{j}\boldsymbol{\gamma}_{j}\mathbf{w}_{j0}^{\prime}\left(  \frac
{1}{T}\sum_{t=1}^{T}\boldsymbol{\varepsilon}_{\circ t}\boldsymbol{\varepsilon
}_{\circ t}^{\prime}\right)  \right\Vert \left(  \sup_{i}\left\Vert
\mathbf{e}_{i}\right\Vert \right) \\
&  \leq\left\vert \lambda_{T}\right\vert \left\Vert \frac{1}{n}\sum_{j=1}%
^{n}\sigma_{j}\boldsymbol{\gamma}_{j}\mathbf{w}_{j0}^{\prime}\right\Vert
\left\Vert \mathbf{V}_{\varepsilon T}\right\Vert .
\end{align*}
Since $\sup_{i}\sigma_{i}^{2}<C$, and by Assumptions \ref{ass:loadings}%
-\ref{ass:ws}, $\sup_{s,i}\gamma_{si}^{2}<C$ and $\sup_{i}\sum_{j=1}%
^{n}\left\vert w_{ji}\right\vert <C$, then it follows
\begin{align*}
\left\Vert \frac{1}{n}\sum_{j=1}^{n}\sigma_{j}\boldsymbol{\gamma}%
_{j}\mathbf{w}_{j0}^{\prime}\right\Vert ^{2}  &  \leq\left\Vert \frac{1}%
{n}\sum_{j=1}^{n}\sigma_{j}\boldsymbol{\gamma}_{j}\mathbf{w}_{j0}^{\prime
}\right\Vert _{F}^{2}=\frac{1}{n^{2}}\sum_{s=1}^{m_{0}}\sum_{i=1}^{n}\left(
\sum_{j=1}^{n}\sigma_{j}\gamma_{sj}w_{ji}\right)  ^{2}\\
&  \leq\frac{1}{n^{2}}\sum_{s=1}^{m_{0}}\sum_{i=1}^{n}\left(  \sum_{j=1}%
^{n}\left\vert \sigma_{j}\gamma_{sj}\right\vert \left\vert w_{ji}\right\vert
\right)  ^{2}\\
&  \leq\left(  \sup_{i}\sigma_{i}^{2}\right)  \left(  \sup_{s,i}\gamma
_{si}^{2}\right)  \left[  \frac{1}{n}\sum_{s=1}^{m_{0}}\left(  \frac{1}{n}%
\sum_{i=1}^{n}\left(  \sum_{j=1}^{n}\left\vert w_{ji}\right\vert \right)
^{2}\right)  \right] \\
&  =O\left(  \frac{1}{n}\right)  .
\end{align*}
In addition, $\lambda_{T}=c_{\lambda}T^{-1/2}$ and $\left\Vert \mathbf{V}%
_{\varepsilon T}\right\Vert =O_{p}\left(  n/T\right)  $ with $0<n/T<C$. Hence,
we have $\sup_{i}\left\Vert \frac{\lambda_{T}}{nT}\sum_{t=1}^{T}\sum_{j=1}%
^{n}\sigma_{j}\boldsymbol{\gamma}_{j}\mathbf{w}_{j0}^{\prime}%
\boldsymbol{\varepsilon}_{\circ t}\varepsilon_{it}\right\Vert =O_{p}\left(
\left(  nT\right)  ^{-1/2}\right)  $. Similarly the third term of
(\ref{gujui}) is also $O_{p}\left(  \left(  nT\right)  ^{-1/2}\right)  $. For
the fourth term of (\ref{gujui}),
\begin{align*}
\sup_{i}\left\Vert \frac{\lambda_{T}^{2}}{n}\sum_{j=1}^{n}\sigma
_{j}\boldsymbol{\gamma}_{j}\mathbf{w}_{i0}^{\prime}\left(  \frac{1}{T}%
\sum_{t=1}^{T}\boldsymbol{\varepsilon}_{\circ t}\boldsymbol{\varepsilon
}_{\circ t}^{\prime}\right)  \mathbf{w}_{j0}\right\Vert  &  \leq\lambda
_{T}^{2}\left(  \sup_{i}\left\Vert \mathbf{w}_{i0}\right\Vert \right)  \left(
\frac{1}{n}\sum_{j=1}^{n}\sigma_{j}\left\Vert \boldsymbol{\gamma}%
_{j}\right\Vert \left\Vert \mathbf{w}_{j0}\right\Vert \right)  \left\Vert
\mathbf{V}_{\varepsilon T}\right\Vert \\
&  =O_{p}\left(  \frac{1}{T}\right)  .
\end{align*}
Using the above results in (\ref{gujui}) we now have
\[
\sup_{i}\left\Vert \frac{1}{nT}\sum_{t=1}^{T}\sum_{j=1}^{n}\sigma
_{j}\boldsymbol{\gamma}_{j}\varepsilon_{it}\left(  \lambda_{T}\right)
\varepsilon_{jt}\left(  \lambda_{T}\right)  \right\Vert \leq\sup_{i}\left\Vert
\frac{1}{nT}\sum_{t=1}^{T}\sum_{j=1}^{n}\sigma_{j}\boldsymbol{\gamma}%
_{j}\varepsilon_{it}\varepsilon_{jt}\right\Vert +O_{p}\left(  \sqrt{\frac
{\ln\left(  n\right)  }{nT}}\right)  ,
\]
and (\ref{sup_gujui}) follows using (\ref{sup_gujui_v1}) to establish the
order of the first term of the above. To establish \ref{sup_|hatg-g|}) note
that by definition of $\boldsymbol{\hat{\gamma}}_{i}$,
\begin{align*}
\boldsymbol{\hat{\gamma}}_{i}-\boldsymbol{\gamma}_{i}  &  =\left(
\mathbf{\hat{F}}^{\prime}\mathbf{\hat{F}}\right)  ^{-1}\mathbf{\hat{F}%
}^{\prime}\left(  \mathbf{F}\boldsymbol{\gamma}_{i}+\sigma_{i}%
\boldsymbol{\varepsilon}_{i\circ}\left(  \lambda_{T}\right)  \right)
-\boldsymbol{\gamma}_{i}=\left(  \mathbf{\hat{F}}^{\prime}\mathbf{\hat{F}%
}\right)  ^{-1}\mathbf{\hat{F}}^{\prime}\left[  \left(  \mathbf{F-\hat{F}%
}\right)  \boldsymbol{\gamma}_{i}+\mathbf{\hat{F}}\boldsymbol{\gamma}%
_{i}+\sigma_{i}\boldsymbol{\varepsilon}_{i\circ}\left(  \lambda_{T}\right)
\right]  -\boldsymbol{\gamma}_{i}\\
&  =\left(  \frac{\mathbf{\hat{F}}^{\prime}\mathbf{\hat{F}}}{T}\right)
^{-1}\left(  \frac{\mathbf{\hat{F}}-\mathbf{F+F}}{T}\right)  ^{\prime}\left[
\left(  \mathbf{F-\hat{F}}\right)  \boldsymbol{\gamma}_{i}+\sigma
_{i}\boldsymbol{\varepsilon}_{i\circ}\left(  \lambda_{T}\right)  \right] \\
&  =\left(  \frac{\mathbf{\hat{F}}^{\prime}\mathbf{\hat{F}}}{T}\right)
^{-1}\left[  \frac{\left(  \mathbf{\hat{F}}-\mathbf{F}\right)  ^{\prime
}\left(  \mathbf{F-\hat{F}}\right)  }{T}\right]  \boldsymbol{\gamma}%
_{i}+\left(  \frac{\mathbf{\hat{F}}^{\prime}\mathbf{\hat{F}}}{T}\right)
^{-1}\left[  \frac{\mathbf{F}^{\prime}\left(  \mathbf{F-\hat{F}}\right)  }%
{T}\right]  \boldsymbol{\gamma}_{i}\\
&  +\left(  \frac{\mathbf{\hat{F}}^{\prime}\mathbf{\hat{F}}}{T}\right)
^{-1}\left[  \frac{\sigma_{i}\left(  \mathbf{\hat{F}}-\mathbf{F}\right)
^{\prime}\boldsymbol{\varepsilon}_{i\circ}\left(  \lambda_{T}\right)  }%
{T}\right]  +\left(  \frac{\mathbf{\hat{F}}^{\prime}\mathbf{\hat{F}}}%
{T}\right)  ^{-1}\left(  \frac{\sigma_{i}\mathbf{F}^{\prime}%
\boldsymbol{\varepsilon}_{i\circ}\left(  \lambda_{T}\right)  }{T}\right)
=\sum_{j=1}^{4}\mathbf{a}_{j,iT},
\end{align*}
and%
\begin{equation}
\left\Vert \boldsymbol{\hat{\gamma}}_{i}-\boldsymbol{\gamma}_{i}\right\Vert
\leq\sum_{j=1}^{4}\left\Vert \mathbf{a}_{j,iT}\right\Vert . \label{|hg-g|}%
\end{equation}
Firstly we have%
\begin{align*}
\left\Vert \mathbf{a}_{1,iT}\right\Vert  &  \leq\left\Vert \left(
\frac{\mathbf{\hat{F}}^{\prime}\mathbf{\hat{F}}}{T}\right)  ^{-1}\right\Vert
\left\Vert \left[  \frac{\left(  \mathbf{\hat{F}}-\mathbf{F}\right)  ^{\prime
}\left(  \mathbf{F-\hat{F}}\right)  }{T}\right]  \right\Vert \left\Vert
\boldsymbol{\gamma}_{i}\right\Vert \\
&  \leq\left\Vert \left(  \frac{\mathbf{\hat{F}}^{\prime}\mathbf{\hat{F}}}%
{T}\right)  ^{-1}\right\Vert \left(  \frac{\left\Vert \mathbf{\hat{F}%
}-\mathbf{F}\right\Vert ^{2}}{T}\right)  \left\Vert \boldsymbol{\gamma}%
_{i}\right\Vert ,
\end{align*}
which implies%
\[
\sup_{i}\left\Vert \mathbf{a}_{1,iT}\right\Vert \leq\left\Vert \left(
\frac{\mathbf{\hat{F}}^{\prime}\mathbf{\hat{F}}}{T}\right)  ^{-1}\right\Vert
\left(  \frac{\left\Vert \mathbf{\hat{F}}-\mathbf{F}\right\Vert ^{2}}%
{T}\right)  \sup_{i}\left\Vert \boldsymbol{\gamma}_{i}\right\Vert .
\]
Using (\ref{hfhf}) in Lemma \ref{lm4} and (\ref{lm2.1}) in Lemma \ref{lm2}, we
note that%
\begin{equation}
\frac{\mathbf{\hat{F}}^{^{\prime}}\mathbf{\hat{F}}}{T}=O_{p}(1)\text{,
}\left(  \frac{\mathbf{\hat{F}}^{^{\prime}}\mathbf{\hat{F}}}{T}\right)
^{-1}=O_{p}(1)\text{, and }T^{-1}\left\Vert \mathbf{\hat{F}}-\mathbf{F}%
\right\Vert _{F}^{2}=O_{p}\left(  \frac{1}{\delta_{nT}^{2}}\right)  .
\label{FhatFhat}%
\end{equation}
Using this result and $\sup_{i}\left\Vert \boldsymbol{\gamma}_{i}\right\Vert
<C$, we obtain
\begin{equation}
\sup_{i}\left\Vert \mathbf{a}_{1,iT}\right\Vert =O_{p}\left(  \frac{1}%
{\delta_{nT}^{2}}\right)  . \label{sup_g_a1}%
\end{equation}
Similarly, $\left\Vert \mathbf{a}_{2,iT}\right\Vert \leq\left\Vert \left(
\frac{\mathbf{\hat{F}}^{\prime}\mathbf{\hat{F}}}{T}\right)  ^{-1}\right\Vert
\left\Vert \frac{\mathbf{F}^{\prime}\left(  \mathbf{F-\hat{F}}\right)  }%
{T}\right\Vert \left\Vert \boldsymbol{\gamma}_{i}\right\Vert ,$ so using
(\ref{baib2}) and (\ref{FhatFhat}) it yields
\begin{equation}
\sup_{i}\left\Vert \mathbf{a}_{2,iT}\right\Vert \leq\left\Vert \left(
\frac{\mathbf{\hat{F}}^{\prime}\mathbf{\hat{F}}}{T}\right)  ^{-1}\right\Vert
\left(  \frac{\left\Vert \mathbf{F}^{\prime}\left(  \mathbf{F-\hat{F}}\right)
\right\Vert }{T}\right)  \sup_{i}\left\Vert \boldsymbol{\gamma}_{i}\right\Vert
=O_{p}\left(  \frac{1}{\delta_{nT}^{2}}\right)  . \label{sup_g_a2}%
\end{equation}
Regarding $\mathbf{a}_{3,iT}$, by Cauchy-Schwarz inequality we have
\[
\left\Vert \mathbf{a}_{3,iT}\right\Vert \leq\left\Vert \left(  \frac
{\mathbf{\hat{F}}^{\prime}\mathbf{\hat{F}}}{T}\right)  ^{-1}\right\Vert
\left\Vert \frac{\sigma_{i}\left(  \mathbf{\hat{F}}-\mathbf{F}\right)
^{\prime}\boldsymbol{\varepsilon}_{i\circ}\left(  \lambda_{T}\right)  }%
{T}\right\Vert \leq\left\Vert \left(  \frac{\mathbf{\hat{F}}^{\prime
}\mathbf{\hat{F}}}{T}\right)  ^{-1}\right\Vert \left(  \frac{\left\Vert
\mathbf{\hat{F}}-\mathbf{F}\right\Vert ^{2}}{T}\right)  ^{1/2}\left(
\frac{\sigma_{i}^{2}\left\Vert \boldsymbol{\varepsilon}_{i\circ}\left(
\lambda_{T}\right)  \right\Vert ^{2}}{T}\right)  ^{1/2},
\]
and therefore%
\[
\sup_{i}\left\Vert \mathbf{a}_{3,iT}\right\Vert \leq\left(  \sup_{i}\sigma
_{i}\right)  \left\Vert \left(  \frac{\mathbf{\hat{F}}^{\prime}\mathbf{\hat
{F}}}{T}\right)  ^{-1}\right\Vert \left(  \frac{\left\Vert \mathbf{\hat{F}%
}-\mathbf{F}\right\Vert ^{2}}{T}\right)  ^{1/2}\left[  \sup_{i}\left(
\frac{\left\Vert \boldsymbol{\varepsilon}_{i\circ}\left(  \lambda_{T}\right)
\right\Vert ^{2}}{T}\right)  \right]  ^{1/2}.
\]
Now using (\ref{sup_u2}) and (\ref{FhatFhat}) it follows that
\begin{equation}
\sup_{i}\left\Vert \mathbf{a}_{3,iT}\right\Vert =O_{p}\left(  \frac{1}%
{\delta_{nT}}\right)  . \label{sup_g_a3}%
\end{equation}
Next, note
\[
\left\Vert \mathbf{a}_{4,iT}\right\Vert \leq\left\Vert \left(  \frac
{\mathbf{\hat{F}}^{\prime}\mathbf{\hat{F}}}{T}\right)  ^{-1}\right\Vert
\left\Vert \frac{\sigma_{i}\mathbf{F}^{\prime}\boldsymbol{\varepsilon}%
_{i\circ}\left(  \lambda_{T}\right)  }{T}\right\Vert
\]
then by (\ref{sup_fu}) we also have%
\begin{equation}
\sup_{i}\left\Vert \mathbf{a}_{4,iT}\right\Vert \leq\left(  \sup_{i}\sigma
_{i}\right)  \left\Vert \left(  \frac{\mathbf{\hat{F}}^{\prime}\mathbf{\hat
{F}}}{T}\right)  ^{-1}\right\Vert \left(  \sup_{i}\left\Vert \frac
{\mathbf{F}^{\prime}\boldsymbol{\varepsilon}_{i\circ}\left(  \lambda
_{T}\right)  }{T}\right\Vert \right)  =O_{p}\left(  \sqrt{\frac{\ln(n)}{T}%
}\right)  . \label{sup_g_a4}%
\end{equation}
Hence using (\ref{sup_g_a1})-(\ref{sup_g_a4}) in (\ref{|hg-g|}) we have%
\[
\sup_{i}\left\Vert \boldsymbol{\hat{\gamma}}_{i}-\boldsymbol{\gamma}%
_{i}\right\Vert \leq\sum_{j=1}^{4}\sup_{i}\left\Vert \mathbf{a}_{j,iT}%
\right\Vert =O_{p}\left(  \sqrt{\frac{\ln(n)}{T}}\right)  ,
\]
as required. Result (\ref{sup_|hatf-f|}) follows by symmetry.
\end{proof}

\begin{lemma}
\label{lm_ave:rho} Consider $\varepsilon_{it}\left(  \lambda_{T}\right)
=\varepsilon_{it}+\lambda_{T}\mathbf{w}_{i0}^{\prime}\boldsymbol{\varepsilon
}_{\circ t}$, where $\boldsymbol{\varepsilon}_{\circ t}=(\varepsilon
_{1t},\varepsilon_{2t},...,\varepsilon_{nt})^{\prime}$, $\varepsilon_{it}\sim
IID\left(  0,1\right)  $ for all $i$ and $t$, $\mathbf{w}_{i0}=\left(
w_{i1},w_{i2},\ldots,w_{in}\right)  ^{\prime}$, and $\mathbf{W}=\left(
w_{ij}\right)  $ satisfy the bounded conditions $\left\Vert \mathbf{W}%
\right\Vert _{1}=\sup_{j}\sum_{i=1}^{n}\left\vert w_{ij}\right\vert <C,$ and
$\left\Vert \mathbf{W}\right\Vert _{\infty}=\sup_{i}\sum_{j=1}^{n}\left\vert
w_{ij}\right\vert <C$. Then for all $\left\vert \lambda_{T}\right\vert <C$ we
have
\[
\sup_{j}\sum_{i=1}^{n}\left\vert E\left[  \varepsilon_{it}\left(  \lambda
_{T}\right)  \varepsilon_{jt}\left(  \lambda_{T}\right)  \right]  \right\vert
<C,\text{ }\sup_{i}\sum_{j=1}^{n}\left\vert E\left[  \varepsilon_{it}\left(
\lambda_{T}\right)  \varepsilon_{jt}\left(  \lambda_{T}\right)  \right]
\right\vert <C,
\]
and
\begin{equation}
n^{-1}\sum_{i=1}^{n}\sum_{j=1}^{n}\left\vert E\left[  \varepsilon_{it}\left(
\lambda_{T}\right)  \varepsilon_{jt}\left(  \lambda_{T}\right)  \right]
\right\vert <C. \label{ave:rho}%
\end{equation}

\end{lemma}

\begin{proof}
Let $\boldsymbol{\varepsilon}_{\circ t}(\lambda_{T})=$
$\boldsymbol{\varepsilon}_{\circ t}+\lambda_{T}\mathbf{W}%
\boldsymbol{\varepsilon}_{\circ t}$, where $\mathbf{W}^{\prime}=(\mathbf{w}%
_{10},\mathbf{w}_{20},...,\mathbf{w}_{n0})$. Then
\[
\boldsymbol{\varepsilon}_{\circ t}(\lambda_{T})\boldsymbol{\varepsilon}_{\circ
t}^{\prime}(\lambda_{T})=\boldsymbol{\varepsilon}_{\circ t}%
\boldsymbol{\varepsilon}_{\circ t}^{\prime}+\lambda_{T}^{2}\mathbf{W}%
\boldsymbol{\varepsilon}_{\circ t}\boldsymbol{\varepsilon}_{\circ t}^{\prime
}\mathbf{W}^{\prime}+\lambda_{T}\mathbf{W}\boldsymbol{\varepsilon}_{\circ
t}\boldsymbol{\varepsilon}_{\circ t}^{\prime}+\lambda_{T}%
\boldsymbol{\varepsilon}_{\circ t}\boldsymbol{\varepsilon}_{\circ t}^{\prime
}\mathbf{W}^{\prime},
\]
and%
\[
\mathbf{V}_{\varepsilon}(\lambda_{T})=E\left[  \boldsymbol{\varepsilon}_{\circ
t}(\lambda_{T})\boldsymbol{\varepsilon}_{\circ t}^{\prime}(\lambda
_{T})\right]  =\mathbf{I}_{n}+\lambda_{T}\left(  \mathbf{W+W}^{\prime}\right)
+\lambda_{T}^{2}\mathbf{WW}^{\prime}.
\]
Consider the maximum absolute column sum norm of $\mathbf{V}_{\varepsilon
}(\lambda_{T})$ and note that%
\[
\left\Vert \mathbf{V}_{\varepsilon}(\lambda_{T})\right\Vert _{1}=\sup_{j}%
\sum_{i=1}^{n}\left\vert E\left[  \varepsilon_{it}\left(  \lambda_{T}\right)
\varepsilon_{jt}\left(  \lambda_{T}\right)  \right]  \right\vert <1+\left\vert
\lambda_{T}\right\vert \left(  \left\Vert \mathbf{W}\right\Vert _{1}%
+\left\Vert \mathbf{W}\right\Vert _{\infty}\right)  +\lambda_{T}^{2}\left\Vert
\mathbf{W}\right\Vert _{1}\left\Vert \mathbf{W}\right\Vert _{\infty}<C\text{.
}%
\]
Similarly for the maximum absolute row sum norm of $\mathbf{V}_{\varepsilon
}(\lambda_{T})$
\[
\left\Vert \mathbf{V}_{\varepsilon}(\lambda_{T})\right\Vert _{\infty}=\sup
_{i}\sum_{j=1}^{n}\left\vert E\left[  \varepsilon_{it}\left(  \lambda
_{T}\right)  \varepsilon_{jt}\left(  \lambda_{T}\right)  \right]  \right\vert
<1+\left\vert \lambda_{T}\right\vert \left(  \left\Vert \mathbf{W}\right\Vert
_{\infty}+\left\Vert \mathbf{W}\right\Vert _{1}\right)  +\lambda_{T}%
^{2}\left\Vert \mathbf{W}\right\Vert _{\infty}\left\Vert \mathbf{W}\right\Vert
_{1}<C\text{,}%
\]
and result (\ref{ave:rho}) follows.
\end{proof}

\begin{lemma}
\label{lm_supbaib1} Consider the latent factor model given by (\ref{mod2}) and
(\ref{uit:p}). Suppose that Assumptions \ref{ass:factors}-\ref{ass:ws} hold
and $\left(  n,T\right)  \rightarrow\infty$, such that $n/T\rightarrow
\kappa\,,$ for $0<\kappa<\infty$. Then for the estimator of factors, we have
\begin{equation}
\sup_{i}\left\Vert \frac{\left(  \mathbf{\hat{F}-F}\right)  ^{\prime
}\boldsymbol{\varepsilon}_{i\circ}\left(  \lambda_{T}\right)  }{T}\right\Vert
=O_{p}\left(  \sqrt{\frac{\ln\left(  n\right)  }{nT}}\right)  .
\label{sup_baib1}%
\end{equation}

\end{lemma}

\begin{proof}
By (A.1) of Bai (2003) we note that%
\[
\mathbf{\hat{f}}_{t}-\mathbf{f}_{t}=\frac{1}{T}\sum_{t^{\prime}=1}%
^{T}\mathbf{\hat{f}}_{t^{\prime}}\eta_{tt^{\prime}}+\frac{1}{T}\sum
_{t^{\prime}=1}^{T}\mathbf{\hat{f}}_{t^{\prime}}\zeta_{tt^{\prime}}+\frac
{1}{T}\sum_{t^{\prime}=1}^{T}\mathbf{\hat{f}}_{t^{\prime}}\varkappa
_{tt^{\prime}}+\frac{1}{T}\sum_{t^{\prime}=1}^{T}\mathbf{\hat{f}}_{t^{\prime}%
}\xi_{tt^{\prime}}%
\]
where $\eta_{tt^{\prime}}=n^{-1}\sum_{i=1}^{n}\sigma_{i}^{2}E\left(
\varepsilon_{it}\left(  \lambda_{T}\right)  \varepsilon_{it^{\prime}}\left(
\lambda_{T}\right)  \right)  $, $\zeta_{tt^{\prime}}=n^{-1}\sum_{i=1}%
^{n}\sigma_{i}^{2}\varepsilon_{it}^{2}\left(  \lambda_{T}\right)
-\eta_{tt^{\prime}}$, $\varkappa_{tt^{\prime}}=n^{-1}\sum_{i=1}^{n}\sigma
_{i}\mathbf{f}_{t^{\prime}}^{\prime}\boldsymbol{\gamma}_{i}^{\prime
}\varepsilon_{it}\left(  \lambda_{T}\right)  $, and $\xi_{tt^{\prime}}%
=n^{-1}\sum_{i=1}^{n}\sigma_{i}\mathbf{f}_{t}^{\prime}\boldsymbol{\gamma}%
_{i}^{\prime}\varepsilon_{it^{\prime}}\left(  \lambda_{T}\right)  $. Hence
\begin{align*}
\frac{1}{T}\left(  \mathbf{\hat{F}-F}\right)  ^{\prime}\boldsymbol{\varepsilon
}_{i\circ}\left(  \lambda_{T}\right)   &  =\frac{1}{T}\sum_{t=1}^{T}\left(
\mathbf{\hat{f}}_{t}-\mathbf{f}_{t}\right)  \varepsilon_{it}\left(
\lambda_{T}\right) \\
&  =\frac{1}{T^{2}}\sum_{t=1}^{T}\sum_{t^{\prime}=1}^{T}\mathbf{\hat{f}%
}_{t^{\prime}}\eta_{tt^{\prime}}\varepsilon_{it}\left(  \lambda_{T}\right)
+\frac{1}{T^{2}}\sum_{t=1}^{T}\sum_{t^{\prime}=1}^{T}\mathbf{\hat{f}%
}_{t^{\prime}}\zeta_{tt^{\prime}}\varepsilon_{it}\left(  \lambda_{T}\right) \\
&  +\frac{1}{T^{2}}\sum_{t=1}^{T}\sum_{t^{\prime}=1}^{T}\mathbf{\hat{f}%
}_{t^{\prime}}\varkappa_{tt^{\prime}}\varepsilon_{it}\left(  \lambda
_{T}\right)  +\frac{1}{T^{2}}\sum_{t=1}^{T}\sum_{t^{\prime}=1}^{T}%
\mathbf{\hat{f}}_{t^{\prime}}\xi_{tt^{\prime}}\varepsilon_{it}\left(
\lambda_{T}\right) \\
&  =\sum_{j=1}^{4}\mathbf{b}_{j,iT},
\end{align*}
and%
\begin{equation}
\sup_{i}\left\Vert \frac{\left(  \mathbf{\hat{F}-F}\right)  ^{\prime
}\boldsymbol{\varepsilon}_{i\circ}\left(  \lambda_{T}\right)  }{T}\right\Vert
\leq\sum_{j=1}^{4}\sup_{i}\left\Vert \mathbf{b}_{j,iT}\right\Vert .
\label{|(hf-f)u|}%
\end{equation}
Firstly, consider $\mathbf{b}_{1,iT}$ and note that%
\begin{align*}
\mathbf{b}_{1,iT}  &  =\frac{1}{T^{2}}\sum_{t=1}^{T}\sum_{t^{\prime}=1}%
^{T}\left(  \mathbf{\hat{f}}_{t^{\prime}}-\mathbf{f}_{t^{\prime}}\right)
\eta_{tt^{\prime}}\varepsilon_{it}\left(  \lambda_{T}\right)  +\frac{1}{T^{2}%
}\sum_{t=1}^{T}\sum_{t^{\prime}=1}^{T}\mathbf{f}_{t^{\prime}}\eta_{tt^{\prime
}}\varepsilon_{it}\left(  \lambda_{T}\right)  =\mathbf{b}_{1,1,iT}%
+\mathbf{b}_{1,2,iT},
\end{align*}
so $\left\Vert \mathbf{b}_{1,iT}\right\Vert \leq\left\Vert \mathbf{b}%
_{1,1,iT}\right\Vert +\left\Vert \mathbf{b}_{1,2,iT}\right\Vert . $ Note for
the first term on the right hand side,
\begin{align*}
\left\Vert \mathbf{b}_{1,1,iT}\right\Vert  &  =\left\Vert \frac{1}{T}%
\sum_{t^{\prime}=1}^{T}\left(  \mathbf{\hat{f}}_{t^{\prime}}-\mathbf{f}%
_{t^{\prime}}\right)  \left(  \frac{1}{T}\sum_{t=1}^{T}\eta_{tt^{\prime}%
}\varepsilon_{it}\left(  \lambda_{T}\right)  \right)  \right\Vert \\
&  \leq\left(  \frac{1}{T}\sum_{t^{\prime}=1}^{T}\left\Vert \mathbf{\hat{f}%
}_{t^{\prime}}-\mathbf{f}_{t^{\prime}}\right\Vert ^{2}\right)  ^{1/2}\left[
\frac{1}{T}\sum_{t^{\prime}=1}^{T}\left(  \frac{1}{T}\sum_{t=1}^{T}%
\eta_{tt^{\prime}}\varepsilon_{it}\left(  \lambda_{T}\right)  \right)
^{2}\right]  ^{1/2}\\
&  \leq\frac{1}{\sqrt{T}}\left(  \frac{1}{T}\sum_{t^{\prime}=1}^{T}\left\Vert
\mathbf{\hat{f}}_{t^{\prime}}-\mathbf{f}_{t^{\prime}}\right\Vert ^{2}\right)
^{1/2}\left(  \frac{1}{T}\sum_{t^{\prime}=1}^{T}\sum_{t=1}^{T}\eta
_{tt^{\prime}}^{2}\right)  ^{1/2}\left(  \frac{1}{T}\sum_{t=1}^{T}%
\varepsilon_{it}^{2}\left(  \lambda_{T}\right)  \right)  ^{1/2},
\end{align*}
where the second and third lines hold by Cauchy-Schwarz inequality. By
definition of $\eta_{tt^{\prime}}$ and serial independence of $\varepsilon
_{it}\left(  \lambda_{T}\right)  ,$ $\eta_{tt^{\prime}}=\bar{\sigma}_{n}^{2}$
for $t=t^{\prime}$ but $0$ otherwise, where $\bar{\sigma}_{n}^{2}=n^{-1}%
\sum_{i=1}^{n}\sigma_{i}^{2}E\left(  \varepsilon_{it}^{2}\left(  \lambda
_{T}\right)  \right)  $. Under assumption on weight $\left\{  w_{ij}\right\}
$, $E\left(  \varepsilon_{it}^{2}\left(  \lambda_{T}\right)  \right)
=1+\lambda_{T}^{2}\left(  \sum_{j=1}^{n}w_{ij}^{2}\right)  <C $, so it follows that%
\begin{equation}
\frac{1}{T}\sum_{t^{\prime}=1}^{T}\sum_{t=1}^{T}\eta_{tt^{\prime}}^{2}=\left(
\bar{\sigma}_{n}^{2}\right)  ^{2}<C. \label{average_sigma2}%
\end{equation}
Given results (\ref{lm2.1}), (\ref{sup_u2}) and (\ref{average_sigma2}), we
further obtain
\begin{align*}
\sup_{i}\left\Vert \mathbf{b}_{1,1,iT}\right\Vert  &  \leq\frac{1}{\sqrt{T}%
}\left(  \frac{1}{T}\sum_{t^{\prime}=1}^{T}\left\Vert \mathbf{\hat{f}%
}_{t^{\prime}}-\mathbf{f}_{t^{\prime}}\right\Vert ^{2}\right)  ^{1/2}\left(
\frac{1}{T}\sum_{t^{\prime}=1}^{T}\sum_{t=1}^{T}\eta_{tt^{\prime}}^{2}\right)
^{1/2}\left(  \sup_{i}\frac{1}{T}\sum_{t=1}^{T}\varepsilon_{it}^{2}\left(
\lambda_{T}\right)  \right)  ^{1/2}\\
&  =O_{p}\left(  T^{-1/2}\delta_{nT}^{-1}\right)  .
\end{align*}
Now consider $\mathbf{b}_{1,2,iT}$. Using properties of $\eta_{tt^{\prime}}$
we have%
\[
\frac{1}{T^{2}}\sum_{t=1}^{T}\sum_{t^{\prime}=1}^{T}\mathbf{f}_{t^{\prime}%
}\eta_{tt^{\prime}}\varepsilon_{it}\left(  \lambda_{T}\right)  =\frac{1}%
{T^{2}}\sum_{t=1}^{T}\mathbf{f}_{t}\eta_{tt}\varepsilon_{it}\left(
\lambda_{T}\right)  +\frac{1}{T^{2}}\sum_{t=1}^{T}\sum_{t^{\prime}\neq
t^{\prime}}^{T}\mathbf{f}_{t^{\prime}}\eta_{tt^{\prime}}\varepsilon
_{it}\left(  \lambda_{T}\right)  =\bar{\sigma}_{n}^{2}\left(  \frac{1}{T^{2}%
}\sum_{t=1}^{T}\mathbf{f}_{t}\varepsilon_{it}\left(  \lambda_{T}\right)
\right)  ,
\]
and therefore%
\[
\left\Vert \mathbf{b}_{1,2,iT}\right\Vert =\left\Vert \frac{1}{T^{2}}%
\sum_{t=1}^{T}\sum_{t^{\prime}=1}^{T}\mathbf{f}_{t^{\prime}}\eta_{tt^{\prime}%
}\varepsilon_{it}\left(  \lambda_{T}\right)  \right\Vert \leq\frac{1}{T}%
\bar{\sigma}_{n}^{2}\left(  \left\Vert \frac{1}{T}\sum_{t=1}^{T}\mathbf{f}%
_{t}\varepsilon_{it}\left(  \lambda_{T}\right)  \right\Vert \right)  .
\]
Then using (\ref{sup_fu}) we have%
\[
\sup_{i}\left\Vert \mathbf{b}_{1,2,iT}\right\Vert \leq\frac{1}{T}\bar{\sigma
}_{n}^{2}\sup_{i}\left\Vert \frac{1}{T}\sum_{t=1}^{T}\mathbf{f}_{t}%
\varepsilon_{it}\left(  \lambda_{T}\right)  \right\Vert =O_{p}\left(  \frac
{1}{T}\sqrt{\frac{\ln\left(  n\right)  }{T}}\right)  .
\]
Hence, combining the probability orders of $\sup_{i}\left\Vert \mathbf{b}%
_{1,1,iT}\right\Vert $ and $\sup_{i}\left\Vert \mathbf{b}_{1,2,iT}\right\Vert
$, we have%
\begin{equation}
\sup_{i}\left\Vert \mathbf{b}_{1,nT}\right\Vert \leq\sup_{i}\left\Vert
\mathbf{b}_{1,1,iT}\right\Vert +\sup_{i}\left\Vert \mathbf{b}_{1,2,iT}%
\right\Vert =O_{p}\left(  \frac{1}{T^{1/2}\delta_{nT}}\right)  .
\label{sup:b1nT}%
\end{equation}
Next, consider $\mathbf{b}_{2,iT}$ in (\ref{|(hf-f)u|}), which can be written
as
\begin{align*}
\mathbf{b}_{2,iT}  &  =\frac{1}{T^{2}}\sum_{t=1}^{T}\sum_{t^{\prime}=1}%
^{T}\left(  \mathbf{\hat{f}}_{t^{\prime}}-\mathbf{f}_{t^{\prime}}\right)
\zeta_{tt^{\prime}}\varepsilon_{it}\left(  \lambda_{T}\right)  +\frac{1}%
{T^{2}}\sum_{t=1}^{T}\sum_{t^{\prime}=1}^{T}\mathbf{f}_{t^{\prime}}%
\zeta_{tt^{\prime}}\varepsilon_{it}\left(  \lambda_{T}\right)  =\mathbf{b}%
_{2,1,iT}+\mathbf{b}_{2,2,iT}.
\end{align*}
For the first term, we can apply Cauchy-Schwarz inequality to obtain
\begin{align*}
\left\Vert \mathbf{b}_{2,1,iT}\right\Vert  &  =\left\Vert \frac{1}{T^{2}}%
\sum_{t^{\prime}=1}^{T}\left(  \mathbf{\hat{f}}_{t^{\prime}}-\mathbf{f}%
_{t^{\prime}}\right)  \left(  \sum_{t=1}^{T}\zeta_{tt^{\prime}}\varepsilon
_{it}\left(  \lambda_{T}\right)  \right)  \right\Vert \\
&  \leq\left(  \frac{1}{T}\sum_{t^{\prime}=1}^{T}\left\Vert \mathbf{\hat{f}%
}_{t^{\prime}}-\mathbf{f}_{t^{\prime}}\right\Vert ^{2}\right)  ^{1/2}\left[
\frac{1}{T^{3}}\sum_{t^{\prime}=1}^{T}\left(  \sum_{t=1}^{T}\zeta_{tt^{\prime
}}\varepsilon_{it}\left(  \lambda_{T}\right)  \right)  ^{2}\right]  ^{1/2}\\
&  \leq\left(  \frac{1}{T}\sum_{t^{\prime}=1}^{T}\left\Vert \mathbf{\hat{f}%
}_{t^{\prime}}-\mathbf{f}_{t^{\prime}}\right\Vert ^{2}\right)  ^{1/2}\left(
\frac{1}{T^{2}}\sum_{t^{\prime}=1}^{T}\sum_{t=1}^{T}\zeta_{tt^{\prime}}%
^{2}\right)  ^{1/2}\left(  \frac{1}{T}\sum_{t=1}^{T}\varepsilon_{it}%
^{2}\left(  \lambda_{T}\right)  \right)  ^{1/2}.
\end{align*}
Since
\[
\zeta_{tt^{\prime}}=n^{-1}\sum_{i=1}^{n}\sigma_{i}^{2}\varepsilon_{it}%
^{2}\left(  \lambda_{T}\right)  -\eta_{tt^{\prime}}=\frac{1}{n}\sum_{j=1}%
^{n}\sigma_{j}^{2}\left[  \varepsilon_{jt}\left(  \lambda_{T}\right)
\varepsilon_{jt^{\prime}}\left(  \lambda_{T}\right)  -E\left(  \varepsilon
_{jt}\left(  \lambda_{T}\right)  \varepsilon_{jt^{\prime}}\left(  \lambda
_{T}\right)  \right)  \right]  ,
\]
then it follows that
\begin{align}
E\left(  \frac{1}{T^{2}}\sum_{t^{\prime}=1}^{T}\sum_{t=1}^{T}\zeta
_{tt^{\prime}}^{2}\right)   &  =\frac{1}{T^{2}}\sum_{t^{\prime}=1}^{T}%
\sum_{t=1}^{T}E\left(  \zeta_{tt^{\prime}}^{2}\right) \nonumber\\
&  =\frac{1}{T^{2}n^{2}}\sum_{t^{\prime}=1}^{T}\sum_{t=1}^{T}E\left(
\sum_{j=1}^{n}\sigma_{j}^{2}\left[  \varepsilon_{jt}\left(  \lambda
_{T}\right)  \varepsilon_{jt^{\prime}}\left(  \lambda_{T}\right)  -E\left(
\varepsilon_{jt}\left(  \lambda_{T}\right)  \varepsilon_{jt^{\prime}}\left(
\lambda_{T}\right)  \right)  \right]  \right)  ^{2}\nonumber\\
&  =\frac{1}{T^{2}n^{2}}\sum_{t=1}^{T}\sum_{t^{\prime}=1}^{T}\sum_{j=1}%
^{n}\sum_{j^{\prime}=1}^{n}\sigma_{j}^{2}\sigma_{j^{\prime}}^{2}E\left(
\varepsilon_{jt}\left(  \lambda_{T}\right)  \varepsilon_{jt^{\prime}}\left(
\lambda_{T}\right)  \varepsilon_{j^{\prime}t}\left(  \lambda_{T}\right)
\varepsilon_{j^{\prime}t^{\prime}}\left(  \lambda_{T}\right)  \right)
\nonumber\\
&  -\frac{1}{T^{2}n^{2}}\sum_{t=1}^{T}\sum_{t^{\prime}=1}^{T}\sum_{j=1}%
^{n}\sum_{j^{\prime}=1}^{n}\sigma_{j}^{2}\sigma_{j^{\prime}}^{2}E\left(
\varepsilon_{jt}\left(  \lambda_{T}\right)  \varepsilon_{jt^{\prime}}\left(
\lambda_{T}\right)  \right)  E\left(  \varepsilon_{j^{\prime}t}\left(
\lambda_{T}\right)  \varepsilon_{j^{\prime}t^{\prime}}\left(  \lambda
_{T}\right)  \right)  . \label{E:zeta2}%
\end{align}
For the first term of (\ref{E:zeta2}), given the serial independence of
$\varepsilon_{it}\left(  \lambda_{T}\right)  $ and note $E\left(
\varepsilon_{it}\left(  \lambda_{T}\right)  \right)  =E\left(  \varepsilon
_{it}+\lambda_{T}\mathbf{w}_{i0}^{\prime}\boldsymbol{\varepsilon}_{\circ
t}\right)  =0$, some algebra yields
\begin{align}
&  \frac{1}{T^{2}n^{2}}\sum_{t=1}^{T}\sum_{t^{\prime}=1}^{T}\sum_{j=1}^{n}%
\sum_{j^{\prime}=1}^{n}\sigma_{j}^{2}\sigma_{j^{\prime}}^{2}E\left(
\varepsilon_{jt}\left(  \lambda_{T}\right)  \varepsilon_{jt^{\prime}}\left(
\lambda_{T}\right)  \varepsilon_{j^{\prime}t}\left(  \lambda_{T}\right)
\varepsilon_{j^{\prime}t^{\prime}}\left(  \lambda_{T}\right)  \right)
\nonumber\\
&  =\frac{1}{T^{2}n^{2}}\sum_{t=1}^{T}\sum_{j=1}^{n}\sigma_{j}^{4}E\left(
\varepsilon_{jt}^{4}\left(  \lambda_{T}\right)  \right)  +\frac{1}{T^{2}n^{2}%
}\sum_{t=1}^{T}\sum_{t^{\prime}\neq t}^{T}\sum_{j=1}^{n}\sigma_{j}^{4}E\left(
\varepsilon_{jt}^{2}\left(  \lambda_{T}\right)  \right)  E\left(
\varepsilon_{jt^{\prime}}^{2}\left(  \lambda_{T}\right)  \right)  +\nonumber\\
&  \frac{1}{T^{2}n^{2}}\sum_{t=1}^{T}\sum_{j=1}^{n}\sum_{j^{\prime}\neq j}%
^{n}\sigma_{j}^{2}\sigma_{j^{\prime}}^{2}E\left(  \varepsilon_{jt}^{2}\left(
\lambda_{T}\right)  \varepsilon_{j^{\prime}t}^{2}\left(  \lambda_{T}\right)
\right)  +\nonumber\\
&  \frac{1}{T^{2}n^{2}}\sum_{t=1}^{T}\sum_{t^{\prime}\neq t}^{T}\sum_{j=1}%
^{n}\sum_{j^{\prime}\neq j}^{n}\sigma_{j}^{2}\sigma_{j^{\prime}}^{2}E\left(
\varepsilon_{jt}\left(  \lambda_{T}\right)  \varepsilon_{jt^{\prime}}\left(
\lambda_{T}\right)  \varepsilon_{j^{\prime}t}\left(  \lambda_{T}\right)
\varepsilon_{j^{\prime}t^{\prime}}\left(  \lambda_{T}\right)  \right)  .
\label{E:ujt}%
\end{align}
To show the order of (\ref{E:ujt}), we note $\varepsilon_{it}=\mathbf{e}%
_{i}^{\prime}\boldsymbol{\varepsilon}_{\circ t}$ where $\mathbf{e}_{i}$ is an
$n\times1$ selection vector with $1$ on its $i^{th}$ element and zero
elsewhere, then%
\begin{align*}
\varepsilon_{jt}^{4}\left(  \lambda_{T}\right)   &  =\left(  \varepsilon
_{it}+\lambda_{T}\mathbf{w}_{i0}^{\prime}\boldsymbol{\varepsilon}_{\circ
t}\right)  ^{4}=\left[  \left(  \mathbf{e}_{i}+\lambda_{T}\mathbf{w}%
_{i0}\right)  ^{\prime}\boldsymbol{\varepsilon}_{\circ t}\right]  ^{4}\\
&  =\left[  \boldsymbol{\varepsilon}_{\circ t}^{\prime}\left(  \mathbf{e}%
_{i}+\lambda_{T}\mathbf{w}_{i0}\right)  \left(  \mathbf{e}_{i}+\lambda
_{T}\mathbf{w}_{i0}\right)  ^{\prime}\boldsymbol{\varepsilon}_{\circ
t}\right]  ^{2}=\left(  \boldsymbol{\varepsilon}_{\circ t}^{\prime}%
\mathbf{A}_{i}\boldsymbol{\varepsilon}_{\circ t}\right)  ^{2},
\end{align*}
where $\mathbf{A}_{i}=\left(  \mathbf{e}_{i}+\lambda_{T}\mathbf{w}%
_{i0}\right)  \left(  \mathbf{e}_{i}+\lambda_{T}\mathbf{w}_{i0}\right)
^{\prime}.$ Using result (S.7) of Lemma 6 in Pesaran and Yamagata (2024) and
noting $w_{ii}=0$ for all $i$, we have
\begin{align}
E\left(  \varepsilon_{jt}^{4}\left(  \lambda_{T}\right)  \right)   &
=E\left(  \boldsymbol{\varepsilon}_{\circ t}^{\prime}\mathbf{A}_{i}%
\boldsymbol{\varepsilon}_{\circ t}\right)  ^{2} =\kappa_{2}\mathrm{tr}\left[
\left(  \mathbf{A}_{i}\odot\mathbf{A}_{i}\right)  \right]  +\left[
\mathrm{tr}\left(  \mathbf{A}_{i}\right)  \right]  ^{2}+2\mathrm{tr}\left(
\mathbf{A}_{i}^{2}\right)  ,\nonumber
\end{align}
where $\kappa_{2}=E(\varepsilon_{it}^{4})-3$. Also, by condition
(\ref{sum:wi}) we have
\[
\sum_{j=1}^{n}w_{ij}^{2}\leq\left(  \sum_{j=1}^{n}\left\vert w_{ij}\right\vert
\right)  ^{2}<C,\sum_{j=1}^{n}w_{ij}^{4}\leq\left(  \sum_{j=1}^{n}\left\vert
w_{ij}\right\vert \right)  ^{4}<C,
\]
using which yields
\begin{equation}
E\left(  \varepsilon_{jt}^{4}\left(  \lambda_{T}\right)  \right)  =\kappa
_{2}\left(  1+\lambda_{T}^{4}\sum_{j=1}^{n}w_{ij}^{4}\right)  +3\left(
1+\lambda_{T}^{2}\sum_{j=1}^{n}w_{ij}^{2}\right)  =O\left(  1\right)  .
\label{E:bit4}%
\end{equation}
It therefore follows that $E\left(  \varepsilon_{jt}^{2}\left(  \lambda
_{T}\right)  \right)  $ and $E\left(  \varepsilon_{jt}^{4}\left(  \lambda
_{T}\right)  \right)  $ are bounded so that the first two terms of
(\ref{E:ujt}) satisfy
\begin{align*}
\frac{1}{T^{2}n^{2}}\sum_{t=1}^{T}\sum_{j=1}^{n}\sigma_{j}^{4}E\left(
\varepsilon_{jt}^{4}\left(  \lambda_{T}\right)  \right)   &  =O\left(
\frac{1}{nT}\right)  ,\\
\frac{1}{T^{2}n^{2}}\sum_{t=1}^{T}\sum_{t^{\prime}\neq t}^{T}\sum_{j=1}%
^{n}\sigma_{j}^{4}E\left(  \varepsilon_{jt}^{2}\left(  \lambda_{T}\right)
\right)  E\left(  \varepsilon_{jt^{\prime}}^{2}\left(  \lambda_{T}\right)
\right)   &  =O\left(  \frac{1}{n}\right)  .
\end{align*}
In addition, the third term of (\ref{E:ujt}) can be bounded using
Cauchy-Schwarz inequality,
\begin{align*}
&  \frac{1}{T^{2}n^{2}}\sum_{t=1}^{T}\sum_{j=1}^{n}\sum_{j^{\prime}\neq j}%
^{n}\sigma_{j}^{2}\sigma_{j^{\prime}}^{2}E\left(  \varepsilon_{jt}^{2}\left(
\lambda_{T}\right)  \varepsilon_{j^{\prime}t}^{2}\left(  \lambda_{T}\right)
\right) \\
&  =\frac{1}{T^{2}n^{2}}\sum_{t=1}^{T}\sum_{j=1}^{n}\sum_{j^{\prime}\neq
j}^{n}\sigma_{j}^{2}\sigma_{j^{\prime}}^{2}\left[  E\left(  \varepsilon
_{jt}^{4}\left(  \lambda_{T}\right)  \right)  \right]  ^{1/2}\times\left[
E\left(  \varepsilon_{j^{\prime}t}^{4}\left(  \lambda_{T}\right)  \right)
\right]  ^{1/2}=O\left(  \frac{1}{T}\right)  .
\end{align*}
The fourth term of (\ref{E:ujt}) can be expanded based on the serial
independence of $\varepsilon_{it}\left(  \lambda_{T}\right)  $, so that%
\begin{align*}
&  \frac{1}{T^{2}n^{2}}\sum_{t=1}^{T}\sum_{t^{\prime}\neq t}^{T}\sum_{j=1}%
^{n}\sum_{j^{\prime}\neq j}^{n}\sigma_{j}^{2}\sigma_{j^{\prime}}^{2}E\left(
\varepsilon_{jt}\left(  \lambda_{T}\right)  \varepsilon_{jt^{\prime}}\left(
\lambda_{T}\right)  \varepsilon_{j^{\prime}t}\left(  \lambda_{T}\right)
\varepsilon_{j^{\prime}t^{\prime}}\left(  \lambda_{T}\right)  \right) \\
&  =\frac{1}{T^{2}n^{2}}\sum_{j=1}^{n}\sum_{j^{\prime}\neq j}^{n}\sigma
_{j}^{2}\sigma_{j^{\prime}}^{2}\left[  \sum_{t=1}^{T}E\left(  \varepsilon
_{jt}\left(  \lambda_{T}\right)  \varepsilon_{j^{\prime}t}\left(  \lambda
_{T}\right)  \right)  \right]  \left[  \sum_{t^{\prime}\neq t}^{T}E\left(
\varepsilon_{jt^{\prime}}\left(  \lambda_{T}\right)  \varepsilon_{j^{\prime
}t^{\prime}}\left(  \lambda_{T}\right)  \right)  \right]  .
\end{align*}
Also note by definition of $\varepsilon_{jt}\left(  \lambda_{T}\right)  $, for
$j\neq j^{\prime}$,%
\[
E\left(  \varepsilon_{jt}\left(  \lambda_{T}\right)  \varepsilon_{j^{\prime}%
t}\left(  \lambda_{T}\right)  \right)  =\lambda_{T}g_{jj^{\prime}}+\lambda
_{T}^{2}\mathbf{w}_{j0}^{\prime}\mathbf{w}_{j^{\prime}0},
\]
where $g_{jj^{\prime}}=w_{jj^{\prime}}+w_{j^{\prime}j}$ and $\mathbf{w}%
_{j0}=\left(  w_{j1},w_{j2},\ldots,w_{jn}\right)  ^{\prime}$, so that%
\[
\left[  \sum_{t=1}^{T}E\left(  \varepsilon_{jt}\left(  \lambda_{T}\right)
\varepsilon_{j^{\prime}t}\left(  \lambda_{T}\right)  \right)  \right]  \left[
\sum_{t^{\prime}\neq t}^{T}E\left(  \varepsilon_{jt^{\prime}}\left(
\lambda_{T}\right)  \varepsilon_{j^{\prime}t^{\prime}}\left(  \lambda
_{T}\right)  \right)  \right]  =T\left(  T-1\right)  \left(  \lambda
_{T}g_{jj^{\prime}}+\lambda_{T}^{2}\mathbf{w}_{j0}^{\prime}\mathbf{w}%
_{j^{\prime}0}\right)  ^{2}.
\]
Then it follows
\begin{align*}
&  \frac{1}{T^{2}n^{2}}\sum_{t=1}^{T}\sum_{t^{\prime}\neq t}^{T}\sum_{j=1}%
^{n}\sum_{j^{\prime}\neq j}^{n}\sigma_{j}^{2}\sigma_{j^{\prime}}^{2}E\left(
\varepsilon_{jt}\left(  \lambda_{T}\right)  \varepsilon_{jt^{\prime}}\left(
\lambda_{T}\right)  \varepsilon_{j^{\prime}t}\left(  \lambda_{T}\right)
\varepsilon_{j^{\prime}t^{\prime}}\left(  \lambda_{T}\right)  \right) \\
&  =\frac{T\left(  T-1\right)  }{T^{2}n^{2}}\sum_{j=1}^{n}\sum_{j^{\prime}\neq
j}^{n}\sigma_{j}^{2}\sigma_{j^{\prime}}^{2}\left(  \lambda_{T}g_{jj^{\prime}%
}+\lambda_{T}^{2}\mathbf{w}_{j0}^{\prime}\mathbf{w}_{j^{\prime}0}\right)
^{2}\\
&  \leq\frac{2T\left(  T-1\right)  }{T^{2}n^{2}}\sum_{j=1}^{n}\sum_{j^{\prime
}\neq j}^{n}\sigma_{j}^{2}\sigma_{j^{\prime}}^{2}\left(  \lambda
_{T}g_{jj^{\prime}}\right)  ^{2}+\frac{2T\left(  T-1\right)  }{T^{2}n^{2}}%
\sum_{j=1}^{n}\sum_{j^{\prime}\neq j}^{n}\sigma_{j}^{2}\sigma_{j^{\prime}}%
^{2}\left(  \lambda_{T}^{2}\mathbf{w}_{j0}^{\prime}\mathbf{w}_{j^{\prime}%
0}\right)  ^{2}.
\end{align*}
Since $\sum_{j^{\prime}=1}^{n}\left\vert w_{jj^{\prime}}\right\vert ^{2}%
\leq\left(  \sum_{j^{\prime}=1}^{n}\left\vert w_{jj^{\prime}}\right\vert
\right)  ^{2}\leq\left\Vert \mathbf{w}_{j0}\right\Vert ^{2}<C$ \ by
(\ref{sum:wi}), $\lambda_{T}=c_{\lambda}T^{-1/2}$ and $\left\vert
\mathbf{w}_{j0}^{\prime}\mathbf{w}_{j^{\prime}0}\right\vert \leq\left\Vert
\mathbf{w}_{j0}\right\Vert \left\Vert \mathbf{w}_{j^{\prime}0}\right\Vert <C$,
we further have%
\begin{align*}
\frac{2T\left(  T-1\right)  }{T^{2}n^{2}}\sum_{j=1}^{n}\sum_{j^{\prime}\neq
j}^{n}\sigma_{j}^{2}\sigma_{j^{\prime}}^{2}\left(  \lambda_{T}g_{jj^{\prime}%
}\right)  ^{2}  &  =\frac{2T\left(  T-1\right)  \lambda_{T}^{2}}{T^{2}n^{2}%
}\sum_{j=1}^{n}\sum_{j^{\prime}\neq j}^{n}\sigma_{j}^{2}\sigma_{j^{\prime}%
}^{2}\left(  w_{jj^{\prime}}+w_{j^{\prime}j}\right)  ^{2}\\
&  \leq\frac{4T\left(  T-1\right)  \lambda_{T}^{2}}{T^{2}n^{2}}\sum_{j=1}%
^{n}\sum_{j^{\prime}\neq j}^{n}\sigma_{j}^{2}\sigma_{j^{\prime}}^{2}\left(
\left\vert w_{jj^{\prime}}\right\vert ^{2}+\left\vert w_{j^{\prime}%
j}\right\vert ^{2}\right) \\
&  \leq\frac{4T\left(  T-1\right)  \lambda_{T}^{2}}{T^{2}n^{2}}\left(
\sup_{i}\sigma_{i}^{4}\right)  \sum_{j=1}^{n}\sum_{j^{\prime}\neq j}%
^{n}\left(  \left\vert w_{jj^{\prime}}\right\vert ^{2}+\left\vert
w_{j^{\prime}j}\right\vert ^{2}\right) \\
&  =O\left(  \frac{1}{nT}\right)  ,
\end{align*}
and%
\begin{align*}
\frac{2T\left(  T-1\right)  }{T^{2}n^{2}}\sum_{j=1}^{n}\sum_{j^{\prime}\neq
j}^{n}\sigma_{j}^{2}\sigma_{j^{\prime}}^{2}\left(  \lambda_{T}^{2}%
\mathbf{w}_{j0}^{\prime}\mathbf{w}_{j^{\prime}0}\right)  ^{2}  &
=\frac{2T\left(  T-1\right)  \lambda_{T}^{4}}{T^{2}n^{2}}\sum_{j=1}^{n}%
\sum_{j^{\prime}\neq j}^{n}\sigma_{j}^{2}\sigma_{j^{\prime}}^{2}\left(
\mathbf{w}_{j0}^{\prime}\mathbf{w}_{j^{\prime}0}\right)  ^{2}\\
&  \leq\frac{2T\left(  T-1\right)  \lambda_{T}^{4}}{T^{2}}\left(  \sup
_{i}\sigma_{i}^{4}\right)  \left(  \sup_{j,j^{\prime}}\left\vert
\mathbf{w}_{j0}^{\prime}\mathbf{w}_{j^{\prime}0}\right\vert \right)  ^{2}\\
&  =O_{p}\left(  \frac{1}{T^{2}}\right)  .
\end{align*}
Overall, using (\ref{E:ujt}) we are able to show the first term of
(\ref{E:zeta2}) is $O\left(  T^{-1}\right)  $ as $n,T\rightarrow\infty$ such
that $n/T=\kappa$ where $0<\kappa<\infty$. Compared to that, the second term
of (\ref{E:zeta2}) satisfies%
\begin{align*}
&  \frac{1}{T^{2}n^{2}}\sum_{t=1}^{T}\sum_{t^{\prime}=1}^{T}\sum_{j=1}^{n}%
\sum_{j^{\prime}=1}^{n}\sigma_{j}^{2}\sigma_{j^{\prime}}^{2}E\left(
\varepsilon_{jt}\left(  \lambda_{T}\right)  \varepsilon_{jt^{\prime}}\left(
\lambda_{T}\right)  \right)  E\left(  \varepsilon_{j^{\prime}t}\left(
\lambda_{T}\right)  \varepsilon_{j^{\prime}t^{\prime}}\left(  \lambda
_{T}\right)  \right) \\
&  =\frac{1}{T^{2}n^{2}}\sum_{t=1}^{T}\sum_{j=1}^{n}\sum_{j^{\prime}=1}%
^{n}\sigma_{j}^{2}\sigma_{j^{\prime}}^{2}E\left(  \varepsilon_{jt}^{2}\left(
\lambda_{T}\right)  \right)  E\left(  \varepsilon_{j^{\prime}t}^{2}\left(
\lambda_{T}\right)  \right)  =O\left(  \frac{1}{T}\right)  .
\end{align*}
As we have shown the orders of the two terms in (\ref{E:zeta2}), it follows
that $E\left(  \frac{1}{T^{2}}\sum_{t^{\prime}=1}^{T}\sum_{t=1}^{T}%
\zeta_{tt^{\prime}}^{2}\right)  =O\left(  \frac{1}{T}\right)  $, so by Markov
inequality $T^{-2}\sum_{t^{\prime}=1}^{T}\sum_{t=1}^{T}\zeta_{tt^{\prime}}%
^{2}=O_{p}\left(  \frac{1}{T}\right)  $. Using this result, (\ref{lm2.1}), and
(\ref{sup_u2}), it follows that
\begin{align*}
\sup_{i}\left\Vert \mathbf{b}_{2,1,iT}\right\Vert  &  \leq\left(  \frac{1}%
{T}\sum_{t^{\prime}=1}^{T}\left\Vert \mathbf{\hat{f}}_{t^{\prime}}%
-\mathbf{f}_{t^{\prime}}\right\Vert ^{2}\right)  ^{1/2}\left(  \frac{1}{T^{2}%
}\sum_{t^{\prime}=1}^{T}\sum_{t=1}^{T}\zeta_{tt^{\prime}}^{2}\right)
^{1/2}\left(  \sup_{i}\frac{1}{T}\sum_{t=1}^{T}\varepsilon_{it}^{2}\left(
\lambda_{T}\right)  \right)  ^{1/2}\\
&  =O_{p}\left(  \frac{1}{\delta_{nT}}\right)  \times O_{p}\left(  \frac
{1}{\sqrt{n}}\right)  =O_{p}\left(  \frac{1}{\sqrt{n}\delta_{nT}}\right)  .
\end{align*}
Note that $\mathbf{b}_{2,2,iT}$ can be written as
\[
\mathbf{b}_{2,2,iT}=\frac{1}{\sqrt{nT}}\frac{1}{T}\sum_{t=1}^{T}\mathbf{z}%
_{t}\varepsilon_{it}\left(  \lambda_{T}\right)  ,
\]
where
\[
\mathbf{z}_{t}=\frac{1}{\sqrt{nT}}\sum_{t^{\prime}=1}^{T}\sum_{k=1}^{n}%
\sigma_{k}\mathbf{f}_{t^{\prime}}\left[  \varepsilon_{kt^{\prime}}\left(
\lambda_{T}\right)  \varepsilon_{kt}\left(  \lambda_{T}\right)  -E\left(
\varepsilon_{kt^{\prime}}\left(  \lambda_{T}\right)  \varepsilon_{kt}\left(
\lambda_{T}\right)  \right)  \right]  ,
\]
and $\left\Vert \mathbf{z}_{t}\right\Vert ^{2}=O_{p}\left(  1\right)  $. Then
by Cauchy-Schwarz inequality,%
\[
\left\Vert \mathbf{b}_{2,2,iT}\right\Vert =\frac{1}{\sqrt{nT}}\left\Vert
\frac{1}{T}\sum_{t=1}^{T}\mathbf{z}_{t}\varepsilon_{it}\left(  \lambda
_{T}\right)  \right\Vert \leq\frac{1}{\sqrt{nT}}\left(  \frac{1}{T}\sum
_{t=1}^{T}\left\Vert \mathbf{z}_{t}\right\Vert ^{2}\right)  ^{1/2}\left(
\frac{1}{T}\sum_{t=1}^{T}\varepsilon_{it}^{2}\left(  \lambda_{T}\right)
\right)  ^{1/2},
\]
and in view of (\ref{sup_u2}) we have
\[
\sup_{i}\left\Vert \mathbf{b}_{2,2,iT}\right\Vert \leq\frac{1}{\sqrt{nT}%
}\left(  \frac{1}{T}\sum_{t=1}^{T}\left\Vert \mathbf{z}_{t}\right\Vert
^{2}\right)  ^{1/2}\sup_{i}\left(  \frac{1}{T}\sum_{t=1}^{T}\varepsilon
_{it}^{2}\left(  \lambda_{T}\right)  \right)  ^{1/2}=O_{p}\left(  \frac
{1}{\sqrt{nT}}\right)  .
\]
Hence,%
\[
\sup_{i}\left\Vert \mathbf{b}_{2,iT}\right\Vert \leq\sup_{i}\left\Vert
\mathbf{b}_{2,1,iT}\right\Vert +\sup_{i}\left\Vert \mathbf{b}_{2,2,iT}%
\right\Vert =O_{p}\left(  \frac{1}{\sqrt{n}\delta_{nT}}\right)  +O_{p}\left(
\frac{1}{\sqrt{nT}}\right)  .
\]
Now consider $\mathbf{b}_{3,iT}$ in (\ref{|(hf-f)u|}) and note that%
\begin{align*}
\mathbf{b}_{3,iT}  &  =\frac{1}{T^{2}}\sum_{t=1}^{T}\sum_{t^{\prime}=1}%
^{T}\left(  \mathbf{\hat{f}}_{t^{\prime}}-\mathbf{f}_{t^{\prime}}\right)
\varkappa_{tt^{\prime}}\varepsilon_{it}\left(  \lambda_{T}\right)  +\frac
{1}{T^{2}}\sum_{t=1}^{T}\sum_{t^{\prime}=1}^{T}\mathbf{f}_{t^{\prime}%
}\varkappa_{tt^{\prime}}\varepsilon_{it}\left(  \lambda_{T}\right)
=\mathbf{b}_{3,1,iT}+\mathbf{b}_{3,2,iT}.
\end{align*}
To bound the first term, note that
\begin{align*}
\left\Vert \mathbf{b}_{3,1,iT}\right\Vert  &  =\left\Vert \frac{1}{T^{2}}%
\sum_{t^{\prime}=1}^{T}\left(  \mathbf{\hat{f}}_{t^{\prime}}-\mathbf{f}%
_{t^{\prime}}\right)  \left(  \sum_{t=1}^{T}\varkappa_{tt^{\prime}}%
\varepsilon_{it}\left(  \lambda_{T}\right)  \right)  \right\Vert \\
&  \leq\left(  \frac{1}{T}\sum_{t^{\prime}=1}^{T}\left\Vert \mathbf{\hat{f}%
}_{t^{\prime}}-\mathbf{f}_{t^{\prime}}\right\Vert ^{2}\right)  ^{1/2}\left[
\frac{1}{T^{3}}\sum_{t^{\prime}=1}^{T}\left(  \sum_{t=1}^{T}\varkappa
_{tt^{\prime}}\varepsilon_{it}\left(  \lambda_{T}\right)  \right)
^{2}\right]  ^{1/2}\\
&  \leq\left(  \frac{1}{T}\sum_{t^{\prime}=1}^{T}\left\Vert \mathbf{\hat{f}%
}_{t^{\prime}}-\mathbf{f}_{t^{\prime}}\right\Vert ^{2}\right)  ^{1/2}\left(
\frac{1}{T^{2}}\sum_{t^{\prime}=1}^{T}\sum_{t=1}^{T}\varkappa_{tt^{\prime}%
}^{2}\right)  ^{1/2}\left(  \frac{1}{T}\sum_{t=1}^{T}\varepsilon_{it}%
^{2}\left(  \lambda_{T}\right)  \right)  ^{1/2}.
\end{align*}
By definition of $\varkappa_{tt^{\prime}}$ and given (\ref{ave:rho}), it
follows that
\begin{align*}
E\left(  \frac{1}{T^{2}}\sum_{t=1}^{T}\sum_{t^{\prime}=1}^{T}\varkappa
_{tt^{\prime}}^{2}\right)   &  =\frac{1}{T^{2}}\sum_{t^{\prime}=1}^{T}%
\sum_{t=1}^{T}E\left(  \frac{1}{n}\sum_{j=1}^{n}\sigma_{j}\mathbf{f}%
_{t^{\prime}}^{\prime}\boldsymbol{\gamma}_{j}\varepsilon_{jt}\left(
\lambda_{T}\right)  \right)  ^{2}\\
&  =\frac{1}{T^{2}n^{2}}\sum_{t^{\prime}=1}^{T}\sum_{t=1}^{T}\sum_{j=1}%
^{n}\sum_{j^{\prime}=1}^{n}\sigma_{j}\sigma_{j^{\prime}}E\left(
\mathbf{f}_{t^{\prime}}^{\prime}\boldsymbol{\gamma}_{j}\mathbf{f}_{t^{\prime}%
}\boldsymbol{\gamma}_{j^{\prime}}\right)  E\left(  \varepsilon_{jt}\left(
\lambda_{T}\right)  \varepsilon_{j^{\prime}t}\left(  \lambda_{T}\right)
\right) \\
&  \leq\left(  \sup_{i}\left\Vert \boldsymbol{\gamma}_{i}\right\Vert
^{2}\right)  \left(  \sup_{t}E\left\Vert \mathbf{f}_{t}\right\Vert
^{2}\right)  \frac{1}{T^{2}n^{2}}\sum_{t^{\prime}=1}^{T}\sum_{t=1}^{T}%
\sum_{j=1}^{n}\sum_{j^{\prime}=1}^{n}\sigma_{j}\sigma_{j^{\prime}}E\left(
\varepsilon_{jt}\left(  \lambda_{T}\right)  \varepsilon_{j^{\prime}t}\left(
\lambda_{T}\right)  \right) \\
&  \leq\left(  \sup_{i}\left\Vert \boldsymbol{\gamma}_{i}\right\Vert
^{2}\right)  \left(  \sup_{t}E\left\Vert \mathbf{f}_{t}\right\Vert
^{2}\right)  \left(  \sup_{i}\sigma_{i}^{2}\right)  \frac{1}{nT}\sum_{t=1}%
^{T}\left(  \frac{1}{n}\sum_{j=1}^{n}\sum_{j^{\prime}=1}^{n}\left\vert
E\left(  \varepsilon_{jt}\left(  \lambda_{T}\right)  \varepsilon_{j^{\prime}%
t}\left(  \lambda_{T}\right)  \right)  \right\vert \right) \\
&  =O\left(  \frac{1}{n}\right)  ,
\end{align*}
using which and Markov inequality yields $T^{-2}\sum_{t^{\prime}=1}^{T}%
\sum_{t=1}^{T}\varkappa_{tt^{\prime}}^{2}=O_{p}\left(  n^{-1}\right)  $. Then
given this result and (\ref{lm2.1}), (\ref{sup_u2}), it follows%
\begin{align*}
\sup_{i}\left\Vert \mathbf{b}_{3,1,iT}\right\Vert  &  =\left(  \frac{1}{T}%
\sum_{t^{\prime}=1}^{T}\left\Vert \mathbf{\hat{f}}_{t^{\prime}}-\mathbf{f}%
_{t^{\prime}}\right\Vert ^{2}\right)  ^{1/2}\left(  \frac{1}{T^{2}}%
\sum_{t^{\prime}=1}^{T}\sum_{t=1}^{T}\varkappa_{tt^{\prime}}^{2}\right)
^{1/2}\left(  \sup_{i}\frac{1}{T}\sum_{t=1}^{T}\varepsilon_{it}^{2}\left(
\lambda_{T}\right)  \right)  ^{1/2}\\
&  =O_{p}\left(  \frac{1}{\delta_{nT}}\right)  \times O_{p}\left(  \frac
{1}{\sqrt{n}}\right)  =O_{p}\left(  \frac{1}{\sqrt{n}\delta_{nT}}\right)  .
\end{align*}
Next, we consider $\left\Vert \mathbf{b}_{3,2,iT}\right\Vert $ and observe
that%
\[
\frac{1}{T^{2}}\sum_{t=1}^{T}\sum_{t^{\prime}=1}^{T}\mathbf{f}_{t^{\prime}%
}\varkappa_{tt^{\prime}}\varepsilon_{it}\left(  \lambda_{T}\right)  =\left(
\frac{1}{T}\sum_{t^{\prime}=1}^{T}\mathbf{f}_{t^{\prime}}\mathbf{f}%
_{t^{\prime}}^{\prime}\right)  \left(  \frac{1}{Tn}\sum_{t=1}^{T}\sum
_{j=1}^{n}\sigma_{j}\boldsymbol{\gamma}_{j}\varepsilon_{jt}\left(  \lambda
_{T}\right)  \varepsilon_{it}\left(  \lambda_{T}\right)  \right)  ,
\]
using which yields
\[
\left\Vert \mathbf{b}_{3,2,iT}\right\Vert =\left\Vert \frac{1}{T^{2}}%
\sum_{t=1}^{T}\sum_{t^{\prime}=1}^{T}\mathbf{f}_{t^{\prime}}\varkappa
_{tt^{\prime}}\varepsilon_{it}\left(  \lambda_{T}\right)  \right\Vert
\leq\left(  \frac{1}{T}\sum_{t^{\prime}=1}^{T}\left\Vert \mathbf{f}%
_{t^{\prime}}\mathbf{f}_{t^{\prime}}^{\prime}\right\Vert \right)  \left(
\left\Vert \frac{1}{nT}\sum_{t=1}^{T}\sum_{j=1}^{n}\sigma_{j}%
\boldsymbol{\gamma}_{j}\varepsilon_{jt}\left(  \lambda_{T}\right)
\varepsilon_{it}\left(  \lambda_{T}\right)  \right\Vert \right)  .
\]
Then given (\ref{sup_gujui}),
\[
\sup_{i}\left\Vert \mathbf{b}_{3,2,iT}\right\Vert \leq\left(  \frac{1}{T}%
\sum_{t^{\prime}=1}^{T}\left\Vert \mathbf{f}_{t^{\prime}}\mathbf{f}%
_{t^{\prime}}^{\prime}\right\Vert \right)  \left(  \sup_{i}\left\Vert \frac
{1}{nT}\sum_{t=1}^{T}\sum_{j=1}^{n}\sigma_{j}\boldsymbol{\gamma}%
_{j}\varepsilon_{jt}\left(  \lambda_{T}\right)  \varepsilon_{it}\left(
\lambda_{T}\right)  \right\Vert \right)  =O_{p}\left(  \sqrt{\frac{\ln\left(
n\right)  }{nT}}\right)  .
\]
Hence, $\sup_{i}\left\Vert \mathbf{b}_{3,iT}\right\Vert \leq\sup_{i}\left\Vert
\mathbf{b}_{3,1,iT}\right\Vert +\sup_{i}\left\Vert \mathbf{b}_{3,2,iT}%
\right\Vert =O_{p}\left(  \sqrt{\frac{\ln\left(  n\right)  }{nT}}\right)  . $
Similarly, the probability order of $\sup_{i}\left\Vert \mathbf{b}%
_{4,iT}\right\Vert $ can also be shown to be $O_{p}\left(  \sqrt{\ln\left(
n\right)  /\left(  nT\right)  }\right)  $. Overall, result (\ref{sup_baib1})
follows as we can use (\ref{|(hf-f)u|}) to show%
\begin{align*}
\sup_{i}\left\Vert \frac{\left(  \mathbf{\hat{F}-F}\right)  ^{\prime
}\boldsymbol{\varepsilon}_{i\circ}\left(  \lambda_{T}\right)  }{T}\right\Vert
&  \leq\sup_{i}\left(  \sum_{j=1}^{4}\left\Vert \mathbf{b}_{j,iT}\right\Vert
\right)  \leq\sum_{j=1}^{4}\sup_{i}\left\Vert \mathbf{b}_{j,iT}\right\Vert \\
&  =O_{p}\left(  \frac{1}{\sqrt{T}\delta_{nT}}\right)  +O_{p}\left(  \frac
{1}{\sqrt{n}\delta_{nT}}\right)  +O_{p}\left(  \frac{1}{\sqrt{nT}}\right)
+O_{p}\left(  \sqrt{\frac{\ln\left(  n\right)  }{nT}}\right) \\
&  =O_{p}\left(  \sqrt{\frac{\ln\left(  n\right)  }{nT}}\right)  .
\end{align*}

\end{proof}

\begin{lemma}
\label{lm:sup(ub)}Denote $\boldsymbol{\varepsilon}_{i\circ}=\left(
\varepsilon_{i1},\varepsilon_{i2},\ldots,\varepsilon_{iT}\right)  ^{\prime}$
and $\mathbf{b}_{i}=\left(  b_{i1},b_{i2},\ldots,b_{iT}\right)  ^{\prime}$
with $b_{it}=\mathbf{w}_{i0}^{\prime}\boldsymbol{\varepsilon}_{\circ t}$,
$\mathbf{w}_{i0}=\left(  w_{i1},w_{is},\ldots,w_{in}\right)  ^{\prime}$ and
$\boldsymbol{\varepsilon}_{\circ t}=\left(  \varepsilon_{1t},\varepsilon
_{2t},\ldots,\varepsilon_{nt}\right)  ^{\prime}$. Suppose that Assumptions
\ref{ass:errors} and \ref{ass:ws} hold. Then as $\left(  n,T\right)
\rightarrow\infty$, such that $n/T\rightarrow\kappa\,,$ for $0<\kappa<\infty
$,
\begin{align}
\sup_{i}\left\vert \frac{\boldsymbol{\varepsilon}_{i\circ}^{\prime}%
\mathbf{b}_{i}}{T}\right\vert  &  =O_{p}\left(  \sqrt{\frac{\ln\left(
n\right)  }{T}}\right)  ,\label{sup_i(ub)}\\
\sup_{i}\left\vert \frac{\mathbf{b}_{i}^{\prime}\mathbf{b}_{i}}{T}\right\vert
&  =O_{p}\left(  1\right)  . \label{sup_i(b2)}%
\end{align}

\end{lemma}

\begin{proof}
Consider (\ref{sup_i(ub)}) and denote $\mathcal{I}_{i,t-1}^{\left(
\varepsilon b\right)  }=\left\{  \varepsilon_{i\tau}b_{i\tau}:\tau
=t-1,t-2,\ldots\right\}  $ and $i=1,2,\ldots,n$. By assumption, $\varepsilon
_{it}$ is cross-sectionally independent, and is independent from $b_{it}$ as
$w_{ii}=0$ for all $i$. Then $E\left(  \varepsilon_{it}b_{it}|\mathcal{I}%
_{i,t-1}^{\left(  \varepsilon b\right)  }\right)  =$ $E\left(  \varepsilon
_{it}b_{it}\right)  =0$. Also $Var\left(  \varepsilon_{it}b_{it}\right)
=Var\left(  \varepsilon_{it}\right)  Var\left(  b_{it}\right)  =E\left(
b_{it}^{2}\right)  , $ which is bounded as by condition (\ref{sum:wi}),%
\begin{equation}
E\left(  b_{it}^{2}\right)  =E\left(  \mathbf{w}_{i0}^{\prime}%
\boldsymbol{\varepsilon}_{\circ t}\right)  ^{2}=\mathbf{w}_{i0}^{\prime
}E\left(  \boldsymbol{\varepsilon}_{\circ t}\boldsymbol{\varepsilon}_{\circ
t}^{\prime}\right)  \mathbf{w}_{i0}=\sum_{s=1}^{n}w_{is}^{2}<C. \label{E:bit2}%
\end{equation}
In addition, since $\varepsilon_{it}\left(  \lambda_{T}\right)  =$
$\varepsilon_{it}+\lambda_{T}b_{it}$ is sub-exponential for any $\left\vert
\lambda_{T}\right\vert <C$, then it follows that $\varepsilon_{it}$ and
$b_{it}$ are both sub-exponential, and hence $\varepsilon_{it}b_{it}$ is also
sub-exponential. Hence, result (\ref{sup_i(ub)}) can be established by
applying the method of proof used for result (\ref{sup_u2_v1}). Now consider
(\ref{sup_i(b2)}) and note
\[
\frac{\mathbf{b}_{i}^{\prime}\mathbf{b}_{i}}{T}=\frac{1}{T}\sum_{t=1}%
^{T}\left[  b_{it}^{2}-E\left(  b_{it}^{2}\right)  \right]  +\frac{1}{T}%
\sum_{t=1}^{T}E\left(  b_{it}^{2}\right)  ,
\]
which further implies%
\[
\sup_{i}\left\vert \frac{\mathbf{b}_{i}^{\prime}\mathbf{b}_{i}}{T}\right\vert
\leq\sup_{i}\left\vert \frac{1}{T}\sum_{t=1}^{T}\left[  b_{it}^{2}-E\left(
b_{it}^{2}\right)  \right]  \right\vert +\sup_{i}\left(  \frac{1}{T}\sum
_{t=1}^{T}E\left(  b_{it}^{2}\right)  \right)  .
\]
Denote $\mathcal{I}_{i,t-1}^{\left(  b\right)  }=\left\{  b_{i\tau}%
:\tau=t-1,t-2,\ldots\right\}  $ and $i=1,2,\ldots,n$, then $E\left(
b_{it}^{2}-E\left(  b_{it}^{2}\right)  |\mathcal{I}_{i,t-1}^{\left(  b\right)
}\right)  =$ $E\left(  b_{it}^{2}-E\left(  b_{it}^{2}\right)  \right)  =0$.
Besides, $Var\left(  b_{it}^{2}\right)  =E\left(  b_{it}^{4}\right)  -\left[
E\left(  b_{it}^{2}\right)  \right]  ^{2}$ is bounded as (\ref{E:bit4}) shows
$E\left(  b_{it}^{4}\right)  =E\left(  \mathbf{w}_{i0}^{\prime}%
\boldsymbol{\varepsilon}_{\circ t}\right)  ^{4}=O\left(  1\right)  $. Also
$b_{it}^{2}-E\left(  b_{it}^{2}\right)  $ is sub-exponential given that it is
already established that $b_{it}$ is sub-exponential. Therefore, it follows
that $\sup_{i}\left\vert \frac{1}{T}\sum_{t=1}^{T}\left[  b_{it}^{2}-E\left(
b_{it}^{2}\right)  \right]  \right\vert =O_{p}\left(  \sqrt{\frac{\ln\left(
n\right)  }{T}}\right)  $ by applying the method of proof used for result
(\ref{sup_u2}). Also by (\ref{E:bit2}), we have $\sup_{i}\left(  \frac{1}%
{T}\sum_{t=1}^{T}E\left(  b_{it}^{2}\right)  \right)  =O\left(  1\right)  $.
Result (\ref{sup_i(b2)}) now follows straightforwardly.
\end{proof}

\begin{lemma}
\label{lm4}Consider the latent factor model given by (\ref{mod2}) and
(\ref{uit:p}). Let $\hat{\sigma}_{i,T}=\left(  T^{-1}\mathbf{e}_{i}^{\prime
}\mathbf{e}_{i}\right)  ^{1/2}$, where $\mathbf{e}_{i}=\mathbf{M}_{\hat{F}%
}\mathbf{y}_{i}$, $\mathbf{M}_{\hat{F}}=\mathbf{I}_{T}-\mathbf{\hat{F}(\hat
{F}}^{\prime}\mathbf{\hat{F})}^{-1}\mathbf{\hat{F}}^{\prime}$, $\mathbf{y}%
_{i}=(y_{i1},y_{i2},...,y_{iT})^{\prime}$, and $\mathbf{\hat{F}}$ is given by
(\ref{hatg}). Also let $\omega_{i,T}=\left(  T^{-1}\sigma_{i}^{2}%
\boldsymbol{\varepsilon}_{i\circ}^{^{\prime}}\mathbf{M}_{F}%
\boldsymbol{\varepsilon}_{i\circ}\right)  ^{1/2},$ where $\mathbf{M}%
_{F}=\mathbf{I}_{T}-\mathbf{F(F}^{\prime}\mathbf{F)}^{-1}\mathbf{F}^{\prime}$.
Suppose that Assumptions \ref{ass:factors}-\ref{ass:ws} hold and $\left(
n,T\right)  \rightarrow\infty$, such that $n/T\rightarrow\kappa\,,$ for
$0<\kappa<\infty$. Then
\begin{align}
\sup_{i}\left\vert \hat{\sigma}_{i,T}^{2}-\omega_{i,T}^{2}\right\vert  &
=O_{p}\left(  \frac{\ln\left(  n\right)  }{T}\right)  ,
\label{sup:|sgm2-ome2|}\\
\sup_{i}\left\vert \hat{\sigma}_{i,T}-\omega_{i,T}\right\vert  &
=O_{p}\left(  \frac{\ln\left(  n\right)  }{T}\right)  ,\label{sup:|sgm-ome|}\\
\sup_{i}\left\vert \frac{1}{\hat{\sigma}_{i,T}}-\frac{1}{\omega_{i,T}%
}\right\vert  &  =O_{p}\left(  \frac{\ln\left(  n\right)  }{T}\right)  ,
\label{sup:|sgm(-1)-ome(-1)|}%
\end{align}
and%
\begin{align}
\frac{1}{n}\sum_{i=1}^{n}\left\vert \hat{\sigma}_{i,T}^{2}-\omega_{i,T}%
^{2}\right\vert  &  =O_{p}\left(  \frac{\ln\left(  n\right)  }{T}\right)
,\label{Ews1}\\
\frac{1}{n}\sum_{i=1}^{n}\left\vert \hat{\sigma}_{i,T}-\omega_{i,T}%
\right\vert  &  =O_{p}\left(  \frac{\ln\left(  n\right)  }{T}\right)
,\label{Ews}\\
\frac{1}{n}\sum_{i=1}^{n}\left\vert \frac{1}{\hat{\sigma}_{i,T}}-\frac
{1}{\omega_{i,T}}\right\vert  &  =O_{p}\left(  \frac{\ln\left(  n\right)  }%
{T}\right)  . \label{Ewin}%
\end{align}

\end{lemma}

\begin{proof}
Note that by (\ref{mod2}), $\mathbf{y}_{i}=\mathbf{F}\boldsymbol{\gamma}%
_{i}+\sigma_{i}\boldsymbol{\varepsilon}_{i\circ}\left(  \lambda_{T}\right)  $,
where $\boldsymbol{\varepsilon}_{i\circ}\left(  \lambda_{T}\right)  =\left(
\varepsilon_{i1}\left(  \lambda_{T}\right)  ,\varepsilon_{i2}\left(
\lambda_{T}\right)  ,\ldots,\varepsilon_{iT}\left(  \lambda_{T}\right)
\right)  ^{\prime}$, which in turn implies
\begin{align*}
\mathbf{e}_{i}  &  =\mathbf{M}_{\hat{F}}\mathbf{y}_{i}=\mathbf{M}_{\hat{F}%
}\left(  \mathbf{F}\boldsymbol{\gamma}_{i}+\sigma_{i}\boldsymbol{\varepsilon
}_{i\circ}\left(  \lambda_{T}\right)  \right) \\
&  =\sigma_{i}\mathbf{M}_{F}\boldsymbol{\varepsilon}_{i\circ}\left(
\lambda_{T}\right)  +\sigma_{i}\left(  \mathbf{M}_{\hat{F}}-\mathbf{M}%
_{F}\right)  \boldsymbol{\varepsilon}_{i\circ}\left(  \lambda_{T}\right)
+\mathbf{M}_{\hat{F}}\mathbf{F}\boldsymbol{\gamma}_{i}.
\end{align*}
Then $\hat{\sigma}_{i,T}^{2}$ can be decomposed as%
\begin{align}
\hat{\sigma}_{i,T}^{2}-\omega_{i,T}^{2}  &  =\frac{\boldsymbol{\gamma}%
_{i}^{^{\prime}}\mathbf{F}^{^{\prime}}\mathbf{M}_{\hat{F}}\mathbf{F}%
\boldsymbol{\gamma}_{i}}{T}+\frac{\sigma_{i}^{2}\boldsymbol{\varepsilon
}_{i\circ}\left(  \lambda_{T}\right)  ^{\prime}\left(  \mathbf{M}_{\hat{F}%
}-\mathbf{M}_{F}\right)  \left(  \mathbf{M}_{\hat{F}}-\mathbf{M}_{F}\right)
\boldsymbol{\varepsilon}_{i\circ}\left(  \lambda_{T}\right)  }{T}\nonumber\\
&  +\frac{2\sigma_{i}^{2}\boldsymbol{\varepsilon}_{i\circ}\left(  \lambda
_{T}\right)  ^{\prime}\mathbf{M}_{F}\left(  \mathbf{M}_{\hat{F}}%
-\mathbf{M}_{F}\right)  \boldsymbol{\varepsilon}_{i\circ}\left(  \lambda
_{T}\right)  }{T}+\frac{2\sigma_{i}^{2}\boldsymbol{\varepsilon}_{i\circ
}\left(  \lambda_{T}\right)  ^{\prime}\mathbf{M}_{F}\mathbf{M}_{\hat{F}%
}\mathbf{F}\boldsymbol{\gamma}_{i}}{T}\nonumber\\
&  +\frac{2\sigma_{i}\boldsymbol{\varepsilon}_{i\circ}\left(  \lambda
_{T}\right)  ^{\prime}\left(  \mathbf{M}_{\hat{F}}-\mathbf{M}_{F}\right)
\mathbf{M}_{\hat{F}}\mathbf{F}\boldsymbol{\gamma}_{i}}{T}+\frac{\sigma_{i}%
^{2}\left(  \boldsymbol{\varepsilon}_{i\circ}\left(  \lambda_{T}\right)
+\boldsymbol{\varepsilon}_{i\circ}\right)  ^{\prime}\mathbf{M}_{F}\left(
\boldsymbol{\varepsilon}_{i\circ}\left(  \lambda_{T}\right)
-\boldsymbol{\varepsilon}_{i\circ}\right)  }{T}\nonumber\\
&  =\sum_{j=1}^{6}B_{j,iT}, \label{sum B}%
\end{align}
and $\left\vert \hat{\sigma}_{i,T}^{2}-\omega_{i,T}^{2}\right\vert \leq
\sum_{j=1}^{6}\left\vert B_{j,iT}\right\vert $. Starting with $B_{1,iT}$, note
that%
\[
\left\vert B_{1,iT}\right\vert =\left\vert \frac{\boldsymbol{\gamma}%
_{i}^{^{\prime}}\mathbf{F}^{^{\prime}}\mathbf{M}_{\hat{F}}\mathbf{F}%
\boldsymbol{\gamma}_{i}}{T}\right\vert =\left\vert \frac{\boldsymbol{\gamma
}_{i}^{^{\prime}}\left(  \mathbf{F-\hat{F}}\right)  ^{^{\prime}}%
\mathbf{M}_{\hat{F}}\left(  \mathbf{F-\hat{F}}\right)  \boldsymbol{\gamma}%
_{i}}{T}\right\vert \leq\left\Vert \boldsymbol{\gamma}_{i}\right\Vert
^{2}\left\Vert \mathbf{M}_{\hat{F}}\right\Vert \left(  \frac{\left\Vert
\mathbf{\hat{F}-F}\right\Vert ^{2}}{T}\right)  .
\]
Also $\sup_{i}\left\Vert \boldsymbol{\gamma}_{i}\right\Vert <C$ and
$\left\Vert \mathbf{M}_{\hat{F}}\right\Vert =1$, and using (\ref{lm2.1}) of
Lemma \ref{lm2} we have
\[
\sup_{i}\left\vert B_{1,iT}\right\vert \leq\left(  \sup_{i}\left\Vert
\boldsymbol{\gamma}_{i}\right\Vert ^{2}\right)  \left\Vert \mathbf{M}_{\hat
{F}}\right\Vert \left(  \frac{\left\Vert \mathbf{\hat{F}-F}\right\Vert ^{2}%
}{T}\right)  =O_{p}\left(  \frac{1}{\delta_{nT}^{2}}\right)  .
\]
To establish the probability orders of the remaining terms of (\ref{Ews1}), we
first observe that
\[
\frac{\mathbf{\hat{F}}^{^{\prime}}\mathbf{\hat{F}}}{T}-\frac{\mathbf{F}%
^{^{\prime}}\mathbf{F}}{T}=\frac{\left(  \mathbf{\hat{F}-F}\right)
^{^{\prime}}\left(  \mathbf{\hat{F}-F}\right)  }{T}+\frac{\left(
\mathbf{\hat{F}-F}\right)  ^{^{\prime}}\mathbf{F}}{T}+\frac{\mathbf{F}%
^{^{\prime}}\left(  \mathbf{\hat{F}-F}\right)  }{T}.
\]
Using results (\ref{lm2.1}) and (\ref{baib2}) it follows that
\begin{equation}
\left\Vert \frac{\mathbf{\hat{F}}^{^{\prime}}\mathbf{\hat{F}}}{T}%
-\frac{\mathbf{F}^{^{\prime}}\mathbf{F}}{T}\right\Vert \leq\frac{\left\Vert
\mathbf{\hat{F}-F}\right\Vert ^{2}}{T}+\frac{\left\Vert \left(  \mathbf{\hat
{F}-F}\right)  ^{\prime}\mathbf{F}\right\Vert }{T}+\frac{\left\Vert
\mathbf{F}^{\prime}\left(  \mathbf{\hat{F}-F}\right)  \right\Vert }{T}%
=O_{p}\left(  \frac{1}{\delta_{nT}^{2}}\right)  . \label{F'F1}%
\end{equation}
By assumption $T^{-1}\mathbf{F}^{^{\prime}}\mathbf{F}$ is a positive definite
matrix, then
\[
\left\Vert \left(  \frac{\mathbf{\hat{F}}^{^{\prime}}\mathbf{\hat{F}}}%
{T}\right)  ^{-1}-\left(  \frac{\mathbf{F}^{^{\prime}}\mathbf{F}}{T}\right)
^{-1}\right\Vert \leq\left\Vert \left(  \frac{\mathbf{\hat{F}}^{^{\prime}%
}\mathbf{\hat{F}}}{T}\right)  ^{-1}\right\Vert \left\Vert \frac{\mathbf{\hat
{F}}^{^{\prime}}\mathbf{\hat{F}}}{T}-\frac{\mathbf{F}^{^{\prime}}\mathbf{F}%
}{T}\right\Vert \left\Vert \left(  \frac{\mathbf{F}^{^{\prime}}\mathbf{F}}%
{T}\right)  ^{-1}\right\Vert =O_{p}\left(  \frac{1}{\delta_{nT}^{2}}\right)
,
\]
so it follows that
\begin{equation}
\frac{\mathbf{\hat{F}}^{^{\prime}}\mathbf{\hat{F}}}{T}=\frac{\mathbf{F}%
^{^{\prime}}\mathbf{F}}{T}+O_{p}\left(  \frac{1}{\delta_{nT}^{2}}\right)
,\text{ and }\left(  \frac{\mathbf{\hat{F}}^{^{\prime}}\mathbf{\hat{F}}}%
{T}\right)  ^{-1}=\left(  \frac{\mathbf{F}^{^{\prime}}\mathbf{F}}{T}\right)
^{-1}+O_{p}\left(  \frac{1}{\delta_{nT}^{2}}\right)  . \label{hfhf}%
\end{equation}
Now we consider $B_{2,iT},$ and note that%
\begin{align}
B_{2,iT}  &  =T^{-1}\sigma_{i}^{2}\boldsymbol{\varepsilon}_{i\circ}\left(
\lambda_{T}\right)  ^{\prime}\left[  \mathbf{\hat{F}}\left(  \mathbf{\hat{F}%
}^{^{\prime}}\mathbf{\hat{F}}\right)  ^{-1}\mathbf{\hat{F}}^{^{\prime}%
}-\mathbf{F}\left(  \mathbf{F}^{^{\prime}}\mathbf{F}\right)  ^{-1}%
\mathbf{F}^{^{\prime}}\right]  \boldsymbol{\varepsilon}_{i\circ}\left(
\lambda_{T}\right)  +\nonumber\\
&  T^{-1}\sigma_{i}^{2}\boldsymbol{\varepsilon}_{i\circ}\left(  \lambda
_{T}\right)  ^{\prime}\left(  \mathbf{I}_{m_{0}}-\mathbf{\hat{F}}\left(
\mathbf{\hat{F}}^{^{\prime}}\mathbf{\hat{F}}\right)  ^{-1}\mathbf{\hat{F}%
}^{^{\prime}}\right)  \mathbf{F}\left(  \mathbf{F}^{^{\prime}}\mathbf{F}%
\right)  ^{-1}\mathbf{F}^{^{\prime}}\boldsymbol{\varepsilon}_{i\circ}\left(
\lambda_{T}\right)  +\nonumber\\
&  T^{-1}\sigma_{i}^{2}\boldsymbol{\varepsilon}_{i\circ}\left(  \lambda
_{T}\right)  ^{\prime}\mathbf{F}\left(  \mathbf{F}^{^{\prime}}\mathbf{F}%
\right)  ^{-1}\mathbf{F}^{^{\prime}}\left(  \mathbf{I}_{m_{0}}-\mathbf{\hat
{F}}\left(  \mathbf{\hat{F}}^{^{\prime}}\mathbf{\hat{F}}\right)
^{-1}\mathbf{\hat{F}}^{^{\prime}}\right)  \boldsymbol{\varepsilon}_{i\circ
}\left(  \lambda_{T}\right)  . \label{B2}%
\end{align}
We further note
\begin{align*}
\mathbf{\hat{F}}\left(  \mathbf{\hat{F}}^{^{\prime}}\mathbf{\hat{F}}\right)
^{-1}\mathbf{\hat{F}}^{^{\prime}}-\mathbf{F}\left(  \mathbf{F}^{^{\prime}%
}\mathbf{F}\right)  ^{-1}\mathbf{F}^{^{\prime}}  &  =\left(  \mathbf{\hat
{F}-F}\right)  \left(  \mathbf{\hat{F}}^{^{\prime}}\mathbf{\hat{F}}\right)
^{-1}\left(  \mathbf{\hat{F}-F}\right)  ^{^{\prime}}+\left[  \mathbf{F}\left(
\mathbf{\hat{F}}^{^{\prime}}\mathbf{\hat{F}}\right)  ^{-1}\mathbf{F}%
^{^{\prime}}-\mathbf{F}\left(  \mathbf{F}^{^{\prime}}\mathbf{F}\right)
^{-1}\mathbf{F}^{^{\prime}}\right] \\
&  +\mathbf{F}\left(  \mathbf{\hat{F}}^{^{\prime}}\mathbf{\hat{F}}\right)
^{-1}\left(  \mathbf{\hat{F}-F}\right)  ^{^{\prime}}+\left(  \mathbf{\hat
{F}-F}\right)  \left(  \mathbf{\hat{F}}^{^{\prime}}\mathbf{\hat{F}}\right)
^{-1}\mathbf{F}^{\prime}.
\end{align*}
Using (\ref{B2}), we have $B_{2,iT}=\sum_{j=1}^{6}B_{2,j,iT}$, where%
\begin{align*}
B_{2,1,iT}  &  =T^{-1}\sigma_{i}^{2}\boldsymbol{\varepsilon}_{i\circ}\left(
\lambda_{T}\right)  ^{\prime}\left(  \mathbf{\hat{F}-F}\right)  \left(
\mathbf{\hat{F}}^{^{\prime}}\mathbf{\hat{F}}\right)  ^{-1}\left(
\mathbf{\hat{F}-F}\right)  ^{\prime}\boldsymbol{\varepsilon}_{i\circ}\left(
\lambda_{T}\right)  ,\\
B_{2,2,iT}  &  =T^{-1}\sigma_{i}^{2}\boldsymbol{\varepsilon}_{i\circ}\left(
\lambda_{T}\right)  ^{\prime}\left[  \mathbf{F}\left(  \mathbf{\hat{F}%
}^{^{\prime}}\mathbf{\hat{F}}\right)  ^{-1}\mathbf{F}^{^{\prime}}%
-\mathbf{F}\left(  \mathbf{F}^{^{\prime}}\mathbf{F}\right)  ^{-1}%
\mathbf{F}^{^{\prime}}\right]  \boldsymbol{\varepsilon}_{i\circ}\left(
\lambda_{T}\right)  ,\\
B_{2,3,iT}  &  =B_{2,4,iT}=T^{-1}\sigma_{i}^{2}\boldsymbol{\varepsilon
}_{i\circ}\left(  \lambda_{T}\right)  ^{\prime}\mathbf{F}\left(
\mathbf{\hat{F}}^{^{\prime}}\mathbf{\hat{F}}\right)  ^{-1}\left(
\mathbf{\hat{F}-F}\right)  ^{\prime}\boldsymbol{\varepsilon}_{i\circ}\left(
\lambda_{T}\right)  ,\\
B_{2,5,iT}  &  =B_{2,6,iT}=T^{-1}\sigma_{i}^{2}\boldsymbol{\varepsilon
}_{i\circ}\left(  \lambda_{T}\right)  ^{\prime}\left(  \mathbf{I}_{m_{0}%
}-\mathbf{\hat{F}}\left(  \mathbf{\hat{F}}^{^{\prime}}\mathbf{\hat{F}}\right)
^{-1}\mathbf{\hat{F}}^{^{\prime}}\right)  \mathbf{F}\left(  \mathbf{F}%
^{^{\prime}}\mathbf{F}\right)  ^{-1}\mathbf{F}^{\prime}\boldsymbol{\varepsilon
}_{i\circ}\left(  \lambda_{T}\right)  ,
\end{align*}
and $B_{2,iT}\leq\sum_{j=1}^{6}\left\vert B_{2,j,iT}\right\vert $. Starting
with the first term we note that
\begin{align*}
\left\vert B_{2,1,iT}\right\vert  &  =\left\vert T^{-1}\sigma_{i}%
^{2}\boldsymbol{\varepsilon}_{i\circ}\left(  \lambda_{T}\right)  ^{\prime
}\left(  \mathbf{\hat{F}-F}\right)  \left(  \mathbf{\hat{F}}^{^{\prime}%
}\mathbf{\hat{F}}\right)  ^{-1}\left(  \mathbf{\hat{F}-F}\right)  ^{\prime
}\boldsymbol{\varepsilon}_{i\circ}\left(  \lambda_{T}\right)  \right\vert \\
&  \leq\sigma_{i}^{2}\left(  \frac{1}{T}\left\Vert \boldsymbol{\varepsilon
}_{i\circ}\left(  \lambda_{T}\right)  \right\Vert ^{2}\right)  \left(
\frac{\left\Vert \mathbf{\hat{F}-F}\right\Vert ^{2}}{T}\right)  \left\Vert
\left(  \frac{\mathbf{\hat{F}}^{^{\prime}}\mathbf{\hat{F}}}{T}\right)
^{-1}\right\Vert .
\end{align*}
By (\ref{lm2.1}) and (\ref{sup_u2}) we have%
\[
\sup_{i}\left\vert B_{2,1,iT}\right\vert \leq\left(  \sup_{i}\sigma_{i}%
^{2}\right)  \left(  \sup_{i}\frac{1}{T}\left\Vert \boldsymbol{\varepsilon
}_{i\circ}\left(  \lambda_{T}\right)  \right\Vert ^{2}\right)  \left(
\frac{\left\Vert \mathbf{\hat{F}-F}\right\Vert ^{2}}{T}\right)  \left\Vert
\left(  \frac{\mathbf{\hat{F}}^{^{\prime}}\mathbf{\hat{F}}}{T}\right)
^{-1}\right\Vert =O_{p}\left(  \frac{1}{\delta_{nT}^{2}}\right)  .
\]
Next, consider $B_{2,2,iT}$ and note that%
\begin{align*}
\left\vert B_{2,2,iT}\right\vert  &  =\left\vert \frac{1}{T}\sigma_{i}%
^{2}\boldsymbol{\varepsilon}_{i\circ}\left(  \lambda_{T}\right)  ^{\prime
}\left[  \mathbf{F}\left(  \mathbf{\hat{F}}^{^{\prime}}\mathbf{\hat{F}%
}\right)  ^{-1}\mathbf{F}^{^{\prime}}-\mathbf{F}\left(  \mathbf{F}^{^{\prime}%
}\mathbf{F}\right)  ^{-1}\mathbf{F}^{^{\prime}}\right]
\boldsymbol{\varepsilon}_{i\circ}\left(  \lambda_{T}\right)  \right\vert \\
&  =\sigma_{i}^{2}\left\vert \left(  \frac{\boldsymbol{\varepsilon}_{i\circ
}\left(  \lambda_{T}\right)  ^{\prime}\mathbf{F}}{T}\right)  \left[  \left(
\frac{\mathbf{\hat{F}}^{^{\prime}}\mathbf{\hat{F}}}{T}\right)  ^{-1}-\left(
\frac{\mathbf{F}^{^{\prime}}\mathbf{F}}{T}\right)  ^{-1}\right]  \left(
\frac{\mathbf{F}^{^{\prime}}\boldsymbol{\varepsilon}_{i\circ}\left(
\lambda_{T}\right)  }{T}\right)  \right\vert \\
&  \leq\sigma_{i}^{2}\left(  \left\Vert \frac{\boldsymbol{\varepsilon}%
_{i\circ}\left(  \lambda_{T}\right)  ^{\prime}\mathbf{F}}{T}\right\Vert
^{2}\right)  \left\Vert \left(  \frac{\mathbf{\hat{F}}^{^{\prime}}%
\mathbf{\hat{F}}}{T}\right)  ^{-1}-\left(  \frac{\mathbf{F}^{^{\prime}%
}\mathbf{F}}{T}\right)  ^{-1}\right\Vert .
\end{align*}
Using (\ref{sup_fu}) (from Lemma \ref{lmsup}) and (\ref{hfhf})), we further
have%
\begin{align*}
\sup_{i}\left\vert B_{2,2,iT}\right\vert  &  \leq\left(  \sup_{i}\sigma
_{i}^{2}\right)  \left(  \sup_{i}\left\Vert \frac{\boldsymbol{\varepsilon
}_{i\circ}\left(  \lambda_{T}\right)  ^{\prime}\mathbf{F}}{T}\right\Vert
^{2}\right)  \left\Vert \left(  \frac{\mathbf{\hat{F}}^{^{\prime}}%
\mathbf{\hat{F}}}{T}\right)  ^{-1}-\left(  \frac{\mathbf{F}^{^{\prime}%
}\mathbf{F}}{T}\right)  ^{-1}\right\Vert \\
&  =O_{p}\left(  \frac{\ln\left(  n\right)  }{T}\right)  \times O_{p}\left(
\frac{1}{\delta_{nT}^{2}}\right)  =O_{p}\left(  \frac{\ln\left(  n\right)
}{T\delta_{nT}^{2}}\right)  .
\end{align*}
Next, consider
\begin{align*}
\left\vert B_{2,3,iT}\right\vert  &  =\sigma_{i}^{2}\left\vert \left(
\frac{\boldsymbol{\varepsilon}_{i\circ}\left(  \lambda_{T}\right)  ^{\prime
}\mathbf{F}}{T}\right)  \left(  \frac{\mathbf{\hat{F}}^{^{\prime}}%
\mathbf{\hat{F}}}{T}\right)  ^{-1}\left[  \frac{\left(  \mathbf{\hat{F}%
-F}\right)  ^{\prime}\boldsymbol{\varepsilon}_{i\circ}\left(  \lambda
_{T}\right)  }{T}\right]  \right\vert \\
&  \leq\sigma_{i}^{2}\left\Vert \frac{\boldsymbol{\varepsilon}_{i\circ}\left(
\lambda_{T}\right)  ^{\prime}\mathbf{F}}{T}\right\Vert \left\Vert \left(
\frac{\mathbf{\hat{F}}^{^{\prime}}\mathbf{\hat{F}}}{T}\right)  ^{-1}%
\right\Vert \left\Vert \frac{\left(  \mathbf{\hat{F}-F}\right)  ^{\prime
}\boldsymbol{\varepsilon}_{i\circ}\left(  \lambda_{T}\right)  }{T}\right\Vert
,
\end{align*}
and using (\ref{sup_fu}), (\ref{sup_baib1}) and (\ref{hfhf}) now yields%
\begin{align*}
\sup_{i}\left\vert B_{2,3,iT}\right\vert  &  \leq\left\Vert \left(
\frac{\mathbf{\hat{F}}^{^{\prime}}\mathbf{\hat{F}}}{T}\right)  ^{-1}%
\right\Vert \left(  \sup_{i}\left\Vert \frac{\boldsymbol{\varepsilon}_{i\circ
}\left(  \lambda_{T}\right)  ^{\prime}\mathbf{F}}{T}\right\Vert \right)
\left(  \sup_{i}\left\Vert \frac{\left(  \mathbf{\hat{F}-F}\right)  ^{\prime
}\boldsymbol{\varepsilon}_{i\circ}\left(  \lambda_{T}\right)  }{T}\right\Vert
\right) \\
&  =O_{p}\left(  \sqrt{\frac{\ln\left(  n\right)  }{T}}\right)  \times
O_{p}\left(  \sqrt{\frac{\ln\left(  n\right)  }{nT}}\right)  .
\end{align*}
Similarly,
\begin{align}
\left\vert B_{2,5,iT}\right\vert  &  =\left\vert T^{-1}\sigma_{i}%
^{2}\boldsymbol{\varepsilon}_{i\circ}\left(  \lambda_{T}\right)  ^{\prime
}\left(  \mathbf{I}_{m_{0}}-\mathbf{\hat{F}}\left(  \mathbf{\hat{F}}%
^{^{\prime}}\mathbf{\hat{F}}\right)  ^{-1}\mathbf{\hat{F}}^{^{\prime}}\right)
\mathbf{F}\left(  \mathbf{F}^{^{\prime}}\mathbf{F}\right)  ^{-1}%
\mathbf{F}^{\prime}\boldsymbol{\varepsilon}_{i\circ}\left(  \lambda
_{T}\right)  \right\vert \nonumber\\
&  \leq\sigma_{i}^{2}\left\vert \frac{\boldsymbol{\varepsilon}_{i\circ}\left(
\lambda_{T}\right)  ^{\prime}\mathbf{P}_{F}\boldsymbol{\varepsilon}_{i\circ
}\left(  \lambda_{T}\right)  }{T}\right\vert +\sigma_{i}^{2}\left\vert
\frac{\boldsymbol{\varepsilon}_{i\circ}^{\prime}\left(  \lambda_{T}\right)
\mathbf{P}_{F}\left(  \mathbf{P}_{\hat{F}}-\mathbf{P}_{F}\right)
\boldsymbol{\varepsilon}_{i\circ}\left(  \lambda_{T}\right)  }{T}\right\vert ,
\label{B2,5,nT}%
\end{align}
where $\mathbf{P}_{F}=\mathbf{F}\left(  \mathbf{F}^{^{\prime}}\mathbf{F}%
\right)  ^{-1}\mathbf{F}^{^{\prime}}$ and $\mathbf{P}_{\hat{F}}=\mathbf{\hat
{F}}\left(  \mathbf{\hat{F}}^{^{\prime}}\mathbf{\hat{F}}\right)
^{-1}\mathbf{\hat{F}}^{^{\prime}}$. Further%
\[
\left\vert \frac{\boldsymbol{\varepsilon}_{i\circ}\left(  \lambda_{T}\right)
^{\prime}\mathbf{P}_{F}\boldsymbol{\varepsilon}_{i\circ}\left(  \lambda
\right)  }{T}\right\vert =\left(  \frac{\boldsymbol{\varepsilon}_{i\circ
}\left(  \lambda_{T}\right)  ^{\prime}\mathbf{F}}{T}\right)  \left(
\frac{\mathbf{F}^{^{\prime}}\mathbf{F}}{T}\right)  ^{-1}\left(  \frac
{\mathbf{F}^{^{\prime}}\boldsymbol{\varepsilon}_{i\circ}\left(  \lambda
_{T}\right)  }{T}\right)  \leq\left\Vert \left(  \frac{\mathbf{F}^{^{\prime}%
}\mathbf{F}}{T}\right)  ^{-1}\right\Vert \left\Vert \frac
{\boldsymbol{\varepsilon}_{i\circ}\left(  \lambda_{T}\right)  ^{\prime
}\mathbf{F}}{T}\right\Vert ^{2},
\]
and using (\ref{sup_fu}) it follows that
\[
\sup_{i}\left\vert \frac{\boldsymbol{\varepsilon}_{i\circ}\left(  \lambda
_{T}\right)  ^{\prime}\mathbf{P}_{F}\boldsymbol{\varepsilon}_{i\circ}\left(
\lambda_{T}\right)  }{T}\right\vert \leq\left\Vert \left(  \frac
{\mathbf{F}^{^{\prime}}\mathbf{F}}{T}\right)  ^{-1}\right\Vert \left(
\sup_{i}\left\Vert \frac{\boldsymbol{\varepsilon}_{i\circ}\left(  \lambda
_{T}\right)  ^{\prime}\mathbf{F}}{T}\right\Vert ^{2}\right)  =O_{p}\left(
\frac{\ln\left(  n\right)  }{T}\right)  .
\]
Also, the probability order of the first term in (\ref{B2,5,nT}) dominates the
second term, and we have
\[
\sup_{i}\left\vert B_{2,5,iT}\right\vert \leq\left(  \sup_{i}\sigma_{i}%
^{2}\right)  \left(  \sup_{i}\left\vert \frac{\boldsymbol{\varepsilon}%
_{i\circ}\left(  \lambda_{T}\right)  ^{\prime}\mathbf{P}_{F}%
\boldsymbol{\varepsilon}_{i\circ}\left(  \lambda_{T}\right)  }{T}\right\vert
\right)  =O_{p}\left(  \frac{\ln\left(  n\right)  }{T}\right)  .
\]
Overall, using the above results we obtain%
\[
\sup_{i}\left\vert B_{2,iT}\right\vert \leq\sup_{i}\left(  \sum_{j=1}%
^{6}\left\vert B_{2,j,iT}\right\vert \right)  \leq\sum_{j=1}^{6}\sup
_{i}\left\vert B_{2,j,iT}\right\vert =O_{p}\left(  \frac{\ln\left(  n\right)
}{T}\right)  .
\]
Next, consider $B_{3,iT}$ which can be rewritten as%
\begin{align*}
B_{3,iT}  &  =\frac{2\sigma_{i}^{2}\boldsymbol{\varepsilon}_{i\circ}\left(
\lambda_{T}\right)  ^{\prime}\left(  \mathbf{I}_{m_{0}}-\mathbf{P}_{F}\right)
\left(  \mathbf{P}_{F}-\mathbf{P}_{\hat{F}}\right)  \boldsymbol{\varepsilon
}_{i\circ}\left(  \lambda_{T}\right)  }{T}=\frac{2\sigma_{i}^{2}%
\boldsymbol{\varepsilon}_{i\circ}\left(  \lambda_{T}\right)  ^{\prime}\left(
\mathbf{P}_{F}-\mathbf{P}_{\hat{F}}-\mathbf{P}_{F}+\mathbf{P}_{F}%
\mathbf{P}_{\hat{F}}\right)  \boldsymbol{\varepsilon}_{i\circ}\left(
\lambda_{T}\right)  }{T}\\
&  =\frac{2\sigma_{i}^{2}\boldsymbol{\varepsilon}_{i\circ}\left(  \lambda
_{T}\right)  ^{\prime}\mathbf{P}_{\hat{F}}\boldsymbol{\varepsilon}_{i\circ
}\left(  \lambda_{T}\right)  }{T}+\frac{2\sigma_{i}^{2}\boldsymbol{\varepsilon
}_{i\circ}\left(  \lambda_{T}\right)  ^{\prime}\mathbf{P}_{F}\mathbf{P}%
_{\hat{F}}\boldsymbol{\varepsilon}_{i\circ}\left(  \lambda_{T}\right)  }%
{T}=B_{3,1,nT}+B_{3,2,nT}.
\end{align*}
Note that
\begin{align*}
\left\vert B_{3,1,iT}\right\vert  &  =2\left\vert \sigma_{i}^{2}\left(
\frac{\boldsymbol{\varepsilon}_{i\circ}\left(  \lambda_{T}\right)  ^{\prime
}\mathbf{\hat{F}}}{T}\right)  \left(  \frac{\mathbf{\hat{F}}^{\prime
}\mathbf{\hat{F}}}{T}\right)  ^{-1}\left(  \frac{\mathbf{\hat{F}}^{\prime
}\boldsymbol{\varepsilon}_{i\circ}\left(  \lambda_{T}\right)  }{T}\right)
\right\vert \leq2\sigma_{i}^{2}\left\Vert \frac{\boldsymbol{\varepsilon
}_{i\circ}\left(  \lambda_{T}\right)  ^{\prime}\mathbf{\hat{F}}}{T}\right\Vert
^{2}\left\Vert \left(  \frac{\mathbf{\hat{F}}^{\prime}\mathbf{\hat{F}}}%
{T}\right)  ^{-1}\right\Vert ^{2}\\
&  \leq4\sigma_{i}^{2}\left(  \left\Vert \frac{\boldsymbol{\varepsilon
}_{i\circ}\left(  \lambda_{T}\right)  ^{\prime}\mathbf{F}}{T}\right\Vert
^{2}+\left\Vert \frac{\boldsymbol{\varepsilon}_{i\circ}\left(  \lambda
_{T}\right)  ^{\prime}\left(  \mathbf{\hat{F}}-\mathbf{F}\right)  }%
{T}\right\Vert ^{2}\right)  \left\Vert \left(  \frac{\mathbf{\hat{F}}^{\prime
}\mathbf{\hat{F}}}{T}\right)  ^{-1}\right\Vert ^{2}.
\end{align*}
Then using (\ref{sup_fu}) and (\ref{sup_baib1}), we have%
\begin{align*}
\sup_{i}\left\vert B_{3,1,iT}\right\vert  &  \leq4\left(  \sup_{i}\sigma
_{i}^{2}\right)  \sup_{i}\left(  \left\Vert \frac{\boldsymbol{\varepsilon
}_{i\circ}\left(  \lambda_{T}\right)  ^{\prime}\mathbf{F}}{T}\right\Vert
^{2}+\left\Vert \frac{\boldsymbol{\varepsilon}_{i\circ}\left(  \lambda
_{T}\right)  ^{\prime}\left(  \mathbf{\hat{F}}-\mathbf{F}\right)  }%
{T}\right\Vert ^{2}\right)  \left\Vert \left(  \frac{\mathbf{\hat{F}}^{\prime
}\mathbf{\hat{F}}}{T}\right)  ^{-1}\right\Vert ^{2}\\
&  \leq4\left(  \sup_{i}\sigma_{i}^{2}\right)  \left(  \sup_{i}\left\Vert
\frac{\boldsymbol{\varepsilon}_{i\circ}\left(  \lambda_{T}\right)  ^{\prime
}\mathbf{F}}{T}\right\Vert ^{2}+\sup_{i}\left\Vert \frac
{\boldsymbol{\varepsilon}_{i\circ}\left(  \lambda_{T}\right)  ^{\prime}\left(
\mathbf{\hat{F}}-\mathbf{F}\right)  }{T}\right\Vert ^{2}\right)
\times\left\Vert \left(  \frac{\mathbf{\hat{F}}^{\prime}\mathbf{\hat{F}}}%
{T}\right)  ^{-1}\right\Vert ^{2}\\
&  =O_{p}\left(  \frac{\ln\left(  n\right)  }{T}\right)  .
\end{align*}
Also,
\begin{align*}
\left\vert B_{3,2,iT}\right\vert  &  =2\sigma_{i}^{2}\left\vert \left(
\frac{\boldsymbol{\varepsilon}_{i\circ}\left(  \lambda_{T}\right)  ^{\prime
}\mathbf{F}}{T}\right)  \left(  \frac{\mathbf{F}^{\prime}\mathbf{F}}%
{T}\right)  ^{-1}\left(  \frac{\mathbf{F}^{\prime}\mathbf{\hat{F}}}{T}\right)
\left(  \frac{\mathbf{\hat{F}}^{\prime}\mathbf{\hat{F}}}{T}\right)
^{-1}\left(  \frac{\mathbf{\hat{F}}^{\prime}\boldsymbol{\varepsilon}_{i\circ
}\left(  \lambda_{T}\right)  }{T}\right)  \right\vert \\
&  \leq2\sigma_{i}^{2}\left\vert \left(  \frac{\boldsymbol{\varepsilon
}_{i\circ}\left(  \lambda_{T}\right)  ^{\prime}\mathbf{F}}{T}\right)  \left(
\frac{\mathbf{F}^{\prime}\mathbf{F}}{T}\right)  ^{-1}\left(  \frac
{\mathbf{F}^{\prime}\mathbf{\hat{F}}}{T}\right)  \left(  \frac{\mathbf{\hat
{F}}^{\prime}\mathbf{\hat{F}}}{T}\right)  ^{-1}\left[  \frac{\mathbf{F}%
^{\prime}\boldsymbol{\varepsilon}_{i\circ}\left(  \lambda_{T}\right)  }%
{T}+\frac{\left(  \mathbf{\hat{F}-F}\right)  ^{\prime}\boldsymbol{\varepsilon
}_{i\circ}\left(  \lambda_{T}\right)  }{T}\right]  \right\vert \\
&  \leq2\sigma_{i}^{2}\left\Vert \frac{\boldsymbol{\varepsilon}_{i\circ
}\left(  \lambda_{T}\right)  ^{\prime}\mathbf{F}}{T}\right\Vert \left(
\left\Vert \frac{\mathbf{F}^{\prime}\boldsymbol{\varepsilon}_{i\circ}\left(
\lambda_{T}\right)  }{T}\right\Vert +\left\Vert \frac{\left(  \mathbf{\hat
{F}-F}\right)  ^{\prime}\boldsymbol{\varepsilon}_{i\circ}\left(  \lambda
_{T}\right)  }{T}\right\Vert \right)  \left\Vert \left(  \frac{\mathbf{F}%
^{\prime}\mathbf{F}}{T}\right)  ^{-1}\right\Vert \times\\
&  \left\Vert \left(  \frac{\mathbf{\hat{F}}^{\prime}\mathbf{\hat{F}}}%
{T}\right)  ^{-1}\right\Vert \left\Vert \frac{\mathbf{F}^{\prime}%
\mathbf{\hat{F}}}{T}\right\Vert .
\end{align*}
By (\ref{sup_fu}) and (\ref{sup_baib1}) we obtain%
\begin{align}
&  \sup_{i}\sigma_{i}^{2}\left\Vert \frac{\boldsymbol{\varepsilon}_{i\circ
}\left(  \lambda_{T}\right)  ^{\prime}\mathbf{F}}{T}\right\Vert \left(
\left\Vert \frac{\mathbf{F}^{\prime}\boldsymbol{\varepsilon}_{i\circ}\left(
\lambda_{T}\right)  }{T}\right\Vert +\left\Vert \frac{\left(  \mathbf{\hat
{F}-F}\right)  ^{\prime}\boldsymbol{\varepsilon}_{i\circ}\left(  \lambda
_{T}\right)  }{T}\right\Vert \right) \nonumber\\
&  \leq\left(  \sup_{i}\sigma_{i}^{2}\right)  \left[  \sup_{i}\left\Vert
\frac{\boldsymbol{\varepsilon}_{i\circ}\left(  \lambda_{T}\right)  ^{\prime
}\mathbf{F}}{T}\right\Vert ^{2}+\sup_{i}\left\Vert \frac
{\boldsymbol{\varepsilon}_{i\circ}\left(  \lambda_{T}\right)  ^{\prime
}\mathbf{F}}{T}\right\Vert \left(  \sup_{i}\left\Vert \frac{\left(
\mathbf{\hat{F}-F}\right)  ^{\prime}\boldsymbol{\varepsilon}_{i\circ}\left(
\lambda_{T}\right)  }{T}\right\Vert \right)  \right] \nonumber\\
&  =O_{p}\left(  \frac{\ln\left(  n\right)  }{T}\right)  +O_{p}\left(
\sqrt{\frac{\ln\left(  n\right)  }{T}}\right)  \times O_{p}\left(  \sqrt
{\frac{\ln\left(  n\right)  }{nT}}\right)  =O_{p}\left(  \frac{\ln\left(
n\right)  }{T}\right)  . \label{B3,2,1}%
\end{align}
Further, by (\ref{baib2}),
\begin{equation}
\left\Vert \frac{\mathbf{F}^{\prime}\mathbf{\hat{F}}}{T}\right\Vert
\leq\left\Vert \frac{\mathbf{F}^{\prime}\mathbf{F}}{T}\right\Vert +\left\Vert
\frac{\mathbf{F}^{\prime}\left(  \mathbf{\hat{F}-F}\right)  }{T}\right\Vert
=O_{p}\left(  1\right)  . \label{B3,2,2}%
\end{equation}
Using (\ref{B3,2,1}), (\ref{B3,2,2}), and (\ref{hfhf}) now yields
\begin{align*}
\sup_{i}\left\vert B_{3,2,iT}\right\vert  &  \leq2\left[  \sup_{i}\sigma
_{i}^{2}\left\Vert \frac{\boldsymbol{\varepsilon}_{i\circ}\left(  \lambda
_{T}\right)  ^{\prime}\mathbf{F}}{T}\right\Vert \left(  \left\Vert
\frac{\mathbf{F}^{\prime}\boldsymbol{\varepsilon}_{i\circ}\left(  \lambda
_{T}\right)  }{T}\right\Vert +\left\Vert \frac{\left(  \mathbf{\hat{F}%
-F}\right)  ^{\prime}\boldsymbol{\varepsilon}_{i\circ}\left(  \lambda
_{T}\right)  }{T}\right\Vert \right)  \right]  \times\\
&  \left\Vert \left(  \frac{\mathbf{F}^{\prime}\mathbf{F}}{T}\right)
^{-1}\right\Vert \left\Vert \left(  \frac{\mathbf{\hat{F}}^{\prime
}\mathbf{\hat{F}}}{T}\right)  ^{-1}\right\Vert \left\Vert \frac{\mathbf{F}%
^{\prime}\mathbf{\hat{F}}}{T}\right\Vert =O_{p}\left(  \frac{\ln\left(
n\right)  }{T}\right)  .
\end{align*}
Hence,
\[
\sup_{i}\left\vert B_{3,iT}\right\vert \leq\sup_{i}\left(  \left\vert
B_{3,1,iT}\right\vert +\left\vert B_{3,2,iT}\right\vert \right)  \leq\sup
_{i}\left\vert B_{3,1,iT}\right\vert +\sup_{i}\left\vert B_{3,2,iT}\right\vert
=O_{p}\left(  \frac{\ln\left(  n\right)  }{T}\right)  .
\]
Now consider $B_{4,iT}$, and note that%
\begin{align*}
B_{4,iT}  &  =\frac{2\sigma_{i}\boldsymbol{\gamma}_{i}^{\prime}\left(
\mathbf{F-\hat{F}}\right)  ^{\prime}\mathbf{M}_{\hat{F}}\mathbf{M}%
_{F}\boldsymbol{\varepsilon}_{i\circ}\left(  \lambda_{T}\right)  }{T}\\
&  =-\frac{2\sigma_{i}\boldsymbol{\gamma}_{i}^{\prime}\left(  \mathbf{\hat
{F}-F}\right)  ^{\prime}\boldsymbol{\varepsilon}_{i\circ}\left(  \lambda
_{T}\right)  }{T}+\frac{2\sigma_{i}\boldsymbol{\gamma}_{i}^{\prime}\left(
\mathbf{\hat{F}-F}\right)  ^{\prime}\mathbf{\hat{F}}\left(  \mathbf{\hat{F}%
}^{\prime}\mathbf{\hat{F}}\right)  ^{-1}\mathbf{\hat{F}}^{\prime
}\boldsymbol{\varepsilon}_{i\circ}\left(  \lambda_{T}\right)  }{T}\\
&  -\frac{2\sigma_{i}\boldsymbol{\gamma}_{i}^{\prime}\left(  \mathbf{\hat
{F}-F}\right)  ^{\prime}\mathbf{M}_{\hat{F}}\left(  \mathbf{\hat{F}-F}\right)
\left(  \mathbf{F}^{\prime}\mathbf{F}\right)  ^{-1}\mathbf{F}^{\prime
}\boldsymbol{\varepsilon}_{i\circ}\left(  \lambda_{T}\right)  }{T}=\sum
_{j=1}^{3}B_{4,j,1T}.
\end{align*}
For the first term of the above equation, we have%
\[
\left\vert B_{4,1,iT}\right\vert \leq2\left\vert \frac{\sigma_{i}%
\boldsymbol{\gamma}_{i}^{\prime}\left(  \mathbf{\hat{F}-F}\right)  ^{\prime
}\boldsymbol{\varepsilon}_{i\circ}\left(  \lambda_{T}\right)  }{T}\right\vert
\leq2\sigma_{i}\left\Vert \boldsymbol{\gamma}_{i}\right\Vert \left\Vert
\frac{\left(  \mathbf{\hat{F}-F}\right)  ^{\prime}\boldsymbol{\varepsilon
}_{i\circ}\left(  \lambda_{T}\right)  }{T}\right\Vert ,
\]
where $\sigma_{i}$\ and $\boldsymbol{\gamma}_{i}$ are bounded. Then using
(\ref{sup_baib1}) it follows that%
\[
\sup_{i}\left\vert B_{4,1,iT}\right\vert \leq2\left(  \sup_{i}\sigma
_{i}\right)  \left(  \sup_{i}\left\Vert \boldsymbol{\gamma}_{i}\right\Vert
\right)  \left(  \sup_{i}\left\Vert \frac{\left(  \mathbf{\hat{F}-F}\right)
^{\prime}\boldsymbol{\varepsilon}_{i\circ}\left(  \lambda_{T}\right)  }%
{T}\right\Vert \right)  =O_{p}\left(  \sqrt{\frac{\ln\left(  n\right)  }{nT}%
}\right)  .
\]
Similarly,%
\begin{align*}
\left\vert B_{4,2,iT}\right\vert  &  =\left\vert \frac{2\sigma_{i}%
\boldsymbol{\gamma}_{i}^{\prime}\left(  \mathbf{\hat{F}-F}\right)  ^{\prime
}\mathbf{\hat{F}}\left(  \mathbf{\hat{F}}^{\prime}\mathbf{\hat{F}}\right)
^{-1}\mathbf{\hat{F}}^{\prime}\boldsymbol{\varepsilon}_{i\circ}\left(
\lambda_{T}\right)  }{T}\right\vert \\
&  \leq2\sigma_{i}\left\Vert \boldsymbol{\gamma}_{i}\right\Vert \left\Vert
\frac{\left(  \mathbf{\hat{F}-F}\right)  ^{\prime}\mathbf{\hat{F}}}%
{T}\right\Vert \left\Vert \left(  \frac{\mathbf{\hat{F}}^{^{\prime}%
}\mathbf{\hat{F}}}{T}\right)  ^{-1}\right\Vert \left(  \frac{\left\Vert
\mathbf{F}^{\prime}\boldsymbol{\varepsilon}_{i\circ}\left(  \lambda
_{T}\right)  \right\Vert }{T}+\frac{\left\Vert \left(  \mathbf{\hat{F}%
-F}\right)  ^{\prime}\boldsymbol{\varepsilon}_{i\circ}\left(  \lambda
_{T}\right)  \right\Vert }{T}\right)  ,
\end{align*}
and using (\ref{sup_fu}) and (\ref{sup_baib1}) we have%
\begin{align*}
\sup_{i}\left\vert B_{4,2,iT}\right\vert  &  \leq2\left(  \sup_{i}\sigma
_{i}\right)  \left(  \sup_{i}\left\Vert \boldsymbol{\gamma}_{i}\right\Vert
\right)  \left(  \sup_{i}\frac{\left\Vert \mathbf{F}^{\prime}%
\boldsymbol{\varepsilon}_{i\circ}\left(  \lambda_{T}\right)  \right\Vert }%
{T}+\sup_{i}\frac{\left\Vert \left(  \mathbf{\hat{F}-F}\right)  ^{\prime
}\boldsymbol{\varepsilon}_{i\circ}\left(  \lambda_{T}\right)  \right\Vert }%
{T}\right) \\
&  \times\left\Vert \frac{\left(  \mathbf{\hat{F}-F}\right)  ^{\prime
}\mathbf{\hat{F}}}{T}\right\Vert \left\Vert \left(  \frac{\mathbf{\hat{F}%
}^{^{\prime}}\mathbf{\hat{F}}}{T}\right)  ^{-1}\right\Vert =O_{p}\left(
\sqrt{\frac{\ln\left(  n\right)  }{T}}\right)  \times O_{p}\left(  \frac
{1}{\delta_{nT}^{2}}\right)  .
\end{align*}
Moreover,%
\begin{align*}
\left\vert B_{4,3,iT}\right\vert  &  =\left\vert \frac{2\sigma_{i}%
\boldsymbol{\gamma}_{i}^{\prime}\left(  \mathbf{\hat{F}-F}\right)  ^{\prime
}\mathbf{M}_{\hat{F}}\left(  \mathbf{\hat{F}-F}\right)  \left(  \mathbf{F}%
^{\prime}\mathbf{F}\right)  ^{-1}\mathbf{F}^{\prime}\boldsymbol{\varepsilon
}_{i\circ}\left(  \lambda_{T}\right)  }{T}\right\vert \\
&  \leq2\sigma_{i}\left\Vert \boldsymbol{\gamma}_{i}\right\Vert \left(
\frac{\left\Vert \mathbf{\hat{F}-F}\right\Vert ^{2}}{T}\right)  \left\Vert
\left(  \frac{\mathbf{F}^{\prime}\mathbf{F}}{T}\right)  ^{-1}\right\Vert
\left\Vert \frac{\mathbf{F}^{\prime}\boldsymbol{\varepsilon}_{i\circ}\left(
\lambda_{T}\right)  }{T}\right\Vert ,
\end{align*}
then taking the supremum,%
\begin{align*}
\sup_{i}\left\vert B_{4,3,iT}\right\vert  &  \leq2\left(  \sup_{i}\sigma
_{i}\right)  \left(  \sup_{i}\left\Vert \boldsymbol{\gamma}_{i}\right\Vert
\right)  \left(  \sup_{i}\left\Vert \frac{\mathbf{F}^{\prime}%
\boldsymbol{\varepsilon}_{i\circ}\left(  \lambda_{T}\right)  }{T}\right\Vert
\right)  \left(  \frac{\left\Vert \mathbf{\hat{F}-F}\right\Vert ^{2}}%
{T}\right)  \left\Vert \left(  \frac{\mathbf{F}^{\prime}\mathbf{F}}{T}\right)
^{-1}\right\Vert \\
&  =O_{p}\left(  \sqrt{\frac{\ln\left(  n\right)  }{T}}\right)  \times
O_{p}\left(  \frac{1}{\delta_{nT}^{2}}\right)  .
\end{align*}
Hence, it follows that
\[
\sup_{i}\left\vert B_{4,iT}\right\vert \leq\sup_{i}\sum_{j=1}^{3}\left\vert
B_{4,j,iT}\right\vert \leq\sum_{j=1}^{3}\sup_{i}\left\vert B_{4,j,iT}%
\right\vert =O_{p}\left(  \sqrt{\frac{\ln\left(  n\right)  }{nT}}\right)  .
\]
Similarly, $\sup_{i}\left\vert B_{5,iT}\right\vert =O_{p}\left(  \sqrt
{\ln\left(  n\right)  /\left(  nT\right)  }\right)  $. For the final term
$B_{6,iT}$, note that%
\begin{align*}
B_{6,iT}  &  =\frac{\sigma_{i}^{2}\left(  \boldsymbol{\varepsilon}_{i\circ
}\left(  \lambda_{T}\right)  +\boldsymbol{\varepsilon}_{i\circ}\right)
^{\prime}\left(  \boldsymbol{\varepsilon}_{i\circ}\left(  \lambda_{T}\right)
-\boldsymbol{\varepsilon}_{i\circ}\right)  }{T}-\frac{\sigma_{i}^{2}\left(
\boldsymbol{\varepsilon}_{i\circ}\left(  \lambda_{T}\right)
+\boldsymbol{\varepsilon}_{i\circ}\right)  ^{\prime}\mathbf{F}}{T}\left(
\frac{\mathbf{F}^{\prime}\mathbf{F}}{T}\right)  ^{-1}\frac{\mathbf{F}^{\prime
}\left(  \boldsymbol{\varepsilon}_{i\circ}\left(  \lambda_{T}\right)
-\boldsymbol{\varepsilon}_{i\circ}\right)  }{T}\\
&  =B_{6,1,iT}+B_{6,2,iT}.
\end{align*}
Since $\boldsymbol{\varepsilon}_{i\circ}\left(  \lambda_{T}\right)
=\boldsymbol{\varepsilon}_{i\circ}+\lambda_{T}\mathbf{b}_{i}$ where
$\boldsymbol{\varepsilon}_{i\circ}=\left(  \varepsilon_{i1},\varepsilon
_{i2},\ldots,\varepsilon_{iT}\right)  ^{\prime}$, $\mathbf{b}_{i}=\left(
b_{i1},b_{i2},\ldots,b_{iT}\right)  ^{\prime}$, $b_{it}=\mathbf{w}%
_{i0}^{\prime}\boldsymbol{\varepsilon}_{\circ t}$, $\mathbf{w}_{i0}=\left(
w_{i1},w_{is},\ldots,w_{in}\right)  ^{\prime}$ and $\boldsymbol{\varepsilon
}_{\circ t}=\left(  \varepsilon_{1t},\varepsilon_{2t},\ldots,\varepsilon
_{nt}\right)  ^{\prime}$, then%
\[
B_{6,1,iT}=\sigma_{i}^{2}\left(  \frac{2\lambda_{T}\boldsymbol{\varepsilon
}_{i\circ}^{\prime}\mathbf{b}_{i}}{T}+\frac{\lambda_{T}^{2}\mathbf{b}%
_{i}^{\prime}\mathbf{b}_{i}}{T}\right)  .
\]
Further, using results (\ref{sup_i(ub)}) and (\ref{sup_i(b2)}) we have%
\[
\sup_{i}\left\vert B_{6,1,iT}\right\vert \leq\left(  \sup_{i}\sigma_{i}%
^{2}\right)  \left(  2\left\vert \lambda_{T}\right\vert \sup_{i}\left\vert
\frac{\boldsymbol{\varepsilon}_{i\circ}^{\prime}\mathbf{b}_{i}}{T}\right\vert
+\lambda_{T}^{2}\sup_{i}\left\vert \frac{\mathbf{b}_{i}^{\prime}\mathbf{b}%
_{i}}{T}\right\vert \right)  =O_{p}\left(  \frac{\sqrt{\ln\left(  n\right)  }%
}{T}\right)  +O_{p}\left(  \frac{1}{T}\right)  .
\]
Now consider $B_{6,2,iT}$ and note that%
\[
\left\vert B_{6,2,iT}\right\vert \leq\sigma_{i}^{2}\left\Vert \frac{\left(
\boldsymbol{\varepsilon}_{i\circ}\left(  \lambda_{T}\right)
+\boldsymbol{\varepsilon}_{i\circ}\right)  ^{\prime}\mathbf{F}}{T}\right\Vert
\left\Vert \left(  \frac{\mathbf{F}^{\prime}\mathbf{F}}{T}\right)
^{-1}\right\Vert \left\Vert \frac{\mathbf{F}^{\prime}\left(
\boldsymbol{\varepsilon}_{i\circ}\left(  \lambda_{T}\right)
-\boldsymbol{\varepsilon}_{i\circ}\right)  }{T}\right\Vert .
\]
Further, using (\ref{sup_fu}) we have%
\begin{align*}
\sup_{i}\left\Vert \frac{\left(  \boldsymbol{\varepsilon}_{i\circ}\left(
\lambda_{T}\right)  +\boldsymbol{\varepsilon}_{i\circ}\right)  ^{\prime
}\mathbf{F}}{T}\right\Vert  &  \leq\sup_{i}\left\Vert \frac
{\boldsymbol{\varepsilon}_{i\circ}\left(  \lambda_{T}\right)  ^{\prime
}\mathbf{F}}{T}\right\Vert +\sup_{i}\left\Vert \frac{\boldsymbol{\varepsilon
}_{i\circ}^{\prime}\mathbf{F}}{T}\right\Vert =O_{p}\left(  \sqrt{\ln\left(
n\right)  /T}\right)  ,\\
\sup_{i}\left\Vert \frac{\mathbf{F}^{\prime}\left(  \boldsymbol{\varepsilon
}_{i\circ}\left(  \lambda_{T}\right)  -\boldsymbol{\varepsilon}_{i\circ
}\right)  }{T}\right\Vert  &  \leq\sup_{i}\left\Vert \frac
{\boldsymbol{\varepsilon}_{i\circ}\left(  \lambda_{T}\right)  ^{\prime
}\mathbf{F}}{T}\right\Vert +\sup_{i}\left\Vert \frac{\boldsymbol{\varepsilon
}_{i\circ}^{\prime}\mathbf{F}}{T}\right\Vert =O_{p}\left(  \sqrt{\ln\left(
n\right)  /T}\right)  .
\end{align*}
Hence, it follows that%
\begin{align*}
\sup_{i}\left\vert B_{6,2,iT}\right\vert  &  \leq\left\Vert \left(
\frac{\mathbf{F}^{\prime}\mathbf{F}}{T}\right)  ^{-1}\right\Vert \left(
\sup_{i}\sigma_{i}^{2}\right)  \left(  \sup_{i}\left\Vert \frac{\left(
\boldsymbol{\varepsilon}_{i\circ}\left(  \lambda_{T}\right)
+\boldsymbol{\varepsilon}_{i\circ}\right)  ^{\prime}\mathbf{F}}{T}\right\Vert
\right) \\
&  \times\left(  \sup_{i}\left\Vert \frac{\mathbf{F}^{\prime}\left(
\boldsymbol{\varepsilon}_{i\circ}\left(  \lambda_{T}\right)
-\boldsymbol{\varepsilon}_{i\circ}\right)  }{T}\right\Vert \right)
=O_{p}\left(  \frac{\ln\left(  n\right)  }{T}\right)  .
\end{align*}
Overall,
\[
\sup_{i}\left\vert B_{6,iT}\right\vert \leq\sup_{i}\left\vert B_{6,1,iT}%
\right\vert +\sup_{i}\left\vert B_{6,2,iT}\right\vert =O_{p}\left(  \frac
{\ln\left(  n\right)  }{T}\right)  .
\]
Using the above results of $B_{1,iT}$ to $B_{6,iT}$, and noting that $n$ and
$T$ are assumed to be of the same order of magnitude, we obtain
\[
\sup_{i}\left\vert \hat{\sigma}_{i,T}^{2}-\omega_{i,T}^{2}\right\vert \leq
\sup_{i}\sum_{j=1}^{6}\left\vert B_{j,iT}\right\vert \leq\sum_{j=1}^{6}%
\sup_{i}\left\vert B_{j,iT}\right\vert =O_{p}\left(  \frac{\ln\left(
n\right)  }{T}\right)  ,
\]
so (\ref{sup:|sgm2-ome2|}) is established. To prove (\ref{sup:|sgm-ome|}),
note that
\[
\sup_{i}\left\vert \hat{\sigma}_{i,T}-\omega_{i,T}\right\vert =\sup_{i}%
\frac{\left\vert \hat{\sigma}_{i,T}^{2}-\omega_{i,T}^{2}\right\vert }%
{\hat{\sigma}_{i,T}+\omega_{i,T}}=\sup_{i}\left(  \frac{1}{\hat{\sigma}%
_{i,T}+\omega_{i,T}}\right)  \sup_{i}\left\vert \hat{\sigma}_{i,T}^{2}%
-\omega_{i,T}^{2}\right\vert .
\]
Since $\hat{\sigma}_{i,T}^{2}>0$ (by construction)
\[
\sup_{i}\left\vert \hat{\sigma}_{i,T}-\omega_{i,T}\right\vert \leq\sup
_{i}\left(  \frac{1}{\omega_{i,T}}\right)  \sup_{i}\left\vert \hat{\sigma
}_{i,T}^{2}-\omega_{i,T}^{2}\right\vert .
\]
By definition of $\omega_{i,T}$, we have $\omega_{i,T}^{-2}=\sigma_{i}%
^{-2}\left(  T^{-1}\boldsymbol{\varepsilon}_{i\circ}^{\prime}\mathbf{M}%
_{F}\boldsymbol{\varepsilon}_{i\circ}\right)  ^{-1}$, so that
\begin{align*}
\sup_{i}\left(  \frac{1}{\omega_{i,T}}\right)   &  =\sup_{i}\left(  \frac
{1}{\omega_{i,T}^{2}}\right)  ^{1/2}\leq\left(  \sup_{i}\frac{1}{\sigma
_{i}^{2}}\right)  ^{1/2}\left(  \sup_{i}\frac{1}{T^{-1}\boldsymbol{\varepsilon
}_{i\circ}^{\prime}\mathbf{M}_{F}\boldsymbol{\varepsilon}_{i\circ}}\right)
^{1/2}\\
&  =\left(  \frac{1}{\inf_{i}\sigma_{i}^{2}}\right)  ^{1/2}\left(  \frac
{1}{\inf_{i}\left(  T^{-1}\boldsymbol{\varepsilon}_{i\circ}^{\prime}%
\mathbf{M}_{F}\boldsymbol{\varepsilon}_{i\circ}\right)  }\right)  ^{1/2}.
\end{align*}
Since $\inf_{i}\sigma_{i}^{2}>c>0$ and by condition (\ref{epsic}) in
Assumption \ref{ass:errors}, $\inf_{i}T^{-1}\boldsymbol{\varepsilon}_{i\circ
}^{\prime}\mathbf{M}_{F}\boldsymbol{\varepsilon}_{i\circ}>c>0$ as
$T\rightarrow\infty$, then it readily follows $\sup_{i}\omega_{i,T}%
^{-1}<C<\infty$, which together with (\ref{sup:|sgm2-ome2|}) now establishes
result (\ref{sup:|sgm-ome|}). Similarly, note that%
\[
\sup_{i}\left\vert \frac{1}{\hat{\sigma}_{i,T}}-\frac{1}{\omega_{i,T}%
}\right\vert \leq\sup_{i}\left(  \frac{1}{\hat{\sigma}_{i,T}}\right)  \sup
_{i}\left(  \frac{1}{\omega_{i,T}}\right)  \sup_{i}\left\vert \hat{\sigma
}_{i,T}-\omega_{i,T}\right\vert ,
\]
where $\sup_{i}\left(  \frac{1}{\hat{\sigma}_{i,T}}\right)  <C$, by
construction. Hence, result (\ref{sup:|sgm(-1)-ome(-1)|}) can be established
using (\ref{sup:|sgm-ome|}). Finally, results (\ref{Ews1}),(\ref{Ews}) and
(\ref{Ewin}) follow using (\ref{sup:|sgm2-ome2|}), (\ref{sup:|sgm-ome|}) and
(\ref{sup:|sgm(-1)-ome(-1)|}), respectively.
\end{proof}

\begin{lemma}
\label{lmhgg}Consider the latent factor model given by (\ref{mod2}) and
(\ref{uit:p}). The latent factors, $\mathbf{f}_{t}$, and their loadings,
$\boldsymbol{\gamma}_{i}$, are estimated by principal components,
$\mathbf{\hat{f}}_{t}$ and $\boldsymbol{\hat{\gamma}}_{i}$, given by
(\ref{hatg}). Suppose that Assumptions \ref{ass:factors}-\ref{ass:ws} hold and
$\left(  n,T\right)  \rightarrow\infty$, such that $n/T\rightarrow\kappa\,,$
for $0<\kappa<\infty$. Then
\begin{align}
\mathbf{d}_{1,nT}  &  =\frac{1}{n}\sum_{i=1}^{n}b_{in}(\boldsymbol{\hat
{\gamma}}_{i}-\boldsymbol{\gamma}_{i})=O_{p}\left(  \sqrt{\frac{\ln\left(
n\right)  }{nT}}\right)  ,\label{hgg}\\
\mathbf{d}_{2,nT}  &  =\frac{1}{n}\sum_{i=1}^{n}(\boldsymbol{\hat{\delta}%
}_{i,T}-\boldsymbol{\delta}_{i,T})=O_{p}\left(  \sqrt{\frac{\ln\left(
n\right)  }{nT}}\right)  ,\label{hdd}\\
\mathbf{d}_{3,nT}  &  =\frac{1}{n}\sum_{i=1}^{n}b_{in}\left(  \boldsymbol{\hat
{\gamma}}_{i}-\boldsymbol{\gamma}_{i}\right)  \boldsymbol{\gamma}_{i}^{\prime
}=O_{p}\left(  \sqrt{\frac{\ln\left(  n\right)  }{nT}}\right)  ,\label{hggg}\\
\mathbf{d}_{4,nT}  &  =\frac{1}{n}\sum_{i=1}^{n}\left(  \omega_{i,T}%
-\sigma_{i}\right)  (\boldsymbol{\hat{\gamma}}_{i}-\boldsymbol{\gamma}%
_{i})=O_{p}\left(  \frac{1}{\delta_{nT}^{2}}\right)  ,\label{hgg_b}\\
\mathbf{d}_{5,nT}  &  =\frac{1}{n}\sum_{i=1}^{n}\left(  \frac{1}{\omega_{i,T}%
}-\frac{1}{\sigma_{i}}\right)  (\boldsymbol{\hat{\gamma}}_{i}%
-\boldsymbol{\gamma}_{i})=O_{p}\left(  \frac{1}{\delta_{nT}^{2}}\right)
,\label{hgg_a}\\
\mathbf{d}_{6,nT}  &  =\frac{1}{n}\sum_{i=1}^{n}\left(  \hat{\sigma}%
_{i,T}-\omega_{i,T}\right)  (\boldsymbol{\hat{\gamma}}_{i}-\boldsymbol{\gamma
}_{i})=O_{p}\left[  \left(  \frac{\ln\left(  n\right)  }{T}\right)
^{3/2}\right]  ,\label{hgg_c}\\
\mathbf{d}_{7,nT}  &  =\frac{1}{n}\sum_{i=1}^{n}\left(  \frac{1}{\hat{\sigma
}_{i,T}}-\frac{1}{\omega_{i,T}}\right)  (\boldsymbol{\hat{\gamma}}%
_{i}-\boldsymbol{\gamma}_{i})=O_{p}\left[  \left(  \frac{\ln\left(  n\right)
}{T}\right)  ^{3/2}\right]  , \label{hgg_d}%
\end{align}
where $\left\{  b_{in}\right\}  _{i=1}^{n}$ is a sequence of fixed values
bounded in $n$, such that $n^{-1}\sum_{i=1}^{n}b_{in}^{2}=O(1)$,
$\boldsymbol{\delta}_{i,T}=\boldsymbol{\gamma}_{i}/\omega_{i,T},$
$\boldsymbol{\hat{\delta}}_{i,T}=\hat{\boldsymbol{\gamma}}_{i}/\omega_{i,T},$
and $\omega_{i,T}=\left(  T^{-1}\sigma_{i}^{2}\boldsymbol{\varepsilon}%
_{i\circ}^{\prime}\mathbf{M}_{F}\boldsymbol{\varepsilon}_{i\circ}\right)
^{1/2}.$
\end{lemma}

\begin{proof}
Note that in general%
\begin{align}
\boldsymbol{\hat{\gamma}}_{i}-\boldsymbol{\gamma}_{i}  &  =\left(
\frac{\mathbf{\hat{F}}^{\prime}\mathbf{\hat{F}}}{T}\right)  ^{-1}\left(
\frac{\mathbf{\hat{F}}^{\prime}\mathbf{F}\boldsymbol{\gamma}_{i}}{T}%
+\frac{\sigma_{i}\mathbf{\hat{F}}^{\prime}\boldsymbol{\varepsilon}_{i\circ
}\left(  \lambda_{T}\right)  }{T}\right)  -\left(  \frac{\mathbf{\hat{F}%
}^{\prime}\mathbf{\hat{F}}}{T}\right)  ^{-1}\left(  \frac{\mathbf{\hat{F}%
}^{\prime}\mathbf{\hat{F}}}{T}\right)  \boldsymbol{\gamma}_{i}\nonumber\\
&  =-\left(  \frac{\mathbf{\hat{F}}^{\prime}\mathbf{\hat{F}}}{T}\right)
^{-1}\left[  \frac{\mathbf{\hat{F}}^{\prime}\left(  \mathbf{\hat{F}-F}\right)
\boldsymbol{\gamma}_{i}}{T}\right]  +\left(  \frac{\mathbf{\hat{F}}^{\prime
}\mathbf{\hat{F}}}{T}\right)  ^{-1}\left(  \frac{\sigma_{i}\mathbf{\hat{F}%
}^{\prime}\boldsymbol{\varepsilon}_{i\circ}\left(  \lambda_{T}\right)  }%
{T}\right)  , \label{hgg0}%
\end{align}
and we have%
\begin{align}
\mathbf{d}_{1,nT}  &  =\frac{1}{n}\sum_{i=1}^{n}b_{in}\left(  \boldsymbol{\hat
{\gamma}}_{i}-\boldsymbol{\gamma}_{i}\right) \nonumber\\
&  =-\left(  \frac{\mathbf{\hat{F}}^{\prime}\mathbf{\hat{F}}}{T}\right)
^{-1}\left[  \frac{\mathbf{\hat{F}}^{\prime}\left(  \mathbf{\hat{F}-F}\right)
}{T}\right]  \left(  \frac{1}{n}\sum_{i=1}^{n}b_{in}\boldsymbol{\gamma}%
_{i}\right)  +\left(  \frac{\mathbf{\hat{F}}^{\prime}\mathbf{\hat{F}}}%
{T}\right)  ^{-1}T^{-1}\left(  \frac{1}{n}\sum_{i=1}^{n}b_{in}\sigma
_{i}\mathbf{\hat{F}}^{\prime}\boldsymbol{\varepsilon}_{i\circ}\left(
\lambda_{T}\right)  \right)  . \label{hgg1}%
\end{align}
Since by assumption $\left\Vert \boldsymbol{\gamma}_{i}\right\Vert <C$, we
have
\[
\left\Vert \frac{1}{n}\sum_{i=1}^{n}b_{in}\boldsymbol{\gamma}_{i}\right\Vert
\leq\left(  \frac{1}{n}\sum_{i=1}^{n}b_{in}^{2}\right)  ^{1/2}\left\Vert
\frac{1}{n}\sum_{i=1}^{n}\boldsymbol{\gamma}_{i}\boldsymbol{\gamma}%
_{i}^{^{\prime}}\right\Vert ^{1/2}\leq\left(  \frac{1}{n}\sum_{i=1}^{n}%
b_{in}^{2}\right)  ^{1/2}\left(  \frac{1}{n}\sum_{i=1}^{n}\left\Vert
\boldsymbol{\gamma}_{i}\right\Vert ^{2}\right)  ^{1/2}<C.
\]
Also by results (\ref{FhatFhat}) and (\ref{baib2}), the first term of
(\ref{hgg1}) is $O_{p}\left(  \delta_{nT}^{-2}\right)  $. For the second term
of (\ref{hgg1}), since $\left(  T^{-1}\mathbf{\hat{F}}^{\prime}\mathbf{\hat
{F}}\right)  ^{-1}=O_{p}(1)$, we note that
\[
T^{-1}\left(  \mathbf{\hat{F}}-\mathbf{F+F}\right)  ^{\prime}\left(  \frac
{1}{n}\sum_{i=1}^{n}b_{in}\sigma_{i}\boldsymbol{\varepsilon}_{i\circ}\left(
\lambda_{T}\right)  \right)  =\frac{1}{n}\sum_{i=1}^{n}b_{in}\left(
\frac{\mathbf{\hat{F}}-\mathbf{F}}{T}\right)  ^{\prime}\sigma_{i}%
\boldsymbol{\varepsilon}_{i\circ}\left(  \lambda_{T}\right)  +\frac{1}{Tn}%
\sum_{i=1}^{n}b_{in}\mathbf{F}^{^{\prime}}\sigma_{i}\boldsymbol{\varepsilon
}_{i\circ}\left(  \lambda_{T}\right)  .
\]
Using result (\ref{sup_baib1}), we have
\begin{align}
&  \left\Vert \frac{1}{n}\sum_{i=1}^{n}b_{in}\left(  \frac{\mathbf{\hat{F}%
}-\mathbf{F}}{T}\right)  ^{\prime}\sigma_{i}\boldsymbol{\varepsilon}_{i\circ
}\left(  \lambda_{T}\right)  \right\Vert \nonumber\\
&  \leq\frac{1}{n}\sum_{i=1}^{n}\left\vert b_{in}\sigma_{i}\right\vert
\left\Vert \left(  \frac{\mathbf{\hat{F}}-\mathbf{F}}{T}\right)  ^{\prime
}\boldsymbol{\varepsilon}_{i\circ}\left(  \lambda_{T}\right)  \right\Vert
\leq\left(  \sup_{i}\left\vert b_{in}\right\vert \right)  \left(  \sup
_{i}\sigma_{i}\right)  \left(  \sup_{i}\left\Vert \frac{\left(  \mathbf{\hat
{F}}-\mathbf{F}\right)  ^{\prime}\boldsymbol{\varepsilon}_{i\circ}\left(
\lambda_{T}\right)  }{T}\right\Vert \right) \nonumber\\
&  =O_{p}\left(  \sqrt{\frac{\ln\left(  n\right)  }{nT}}\right)  .
\label{FhatD}%
\end{align}
Under part (a) of Assumption \ref{ass:errors} and by the serial independence
of $\varepsilon_{it}$,
\begin{align*}
E\left\Vert \frac{1}{\sqrt{nT}}\sum_{i=1}^{n}\sum_{t=1}^{T}b_{in}%
\mathbf{f}_{t}\sigma_{i}\boldsymbol{\varepsilon}_{i\circ}\left(  \lambda
_{T}\right)  \right\Vert ^{2}  &  =\frac{1}{nT}\sum_{i=1}^{n}\sum_{j=1}%
^{n}\sum_{t=1}^{T}b_{in}b_{jn}\sigma_{i}\sigma_{j}E\left\Vert \mathbf{f}%
_{t}\right\Vert ^{2}E\left(  \varepsilon_{it}\left(  \lambda_{T}\right)
\varepsilon_{jt}\left(  \lambda_{T}\right)  \right) \\
&  \leq E\left\Vert \mathbf{f}_{t}\right\Vert ^{2}\left(  \sup_{i}b_{in}%
^{2}\right)  \left(  \sup_{i}\sigma_{i}^{2}\right)  \left[  \frac{1}{nT}%
\sum_{t=1}^{T}\sum_{i=1}^{n}\sum_{j=1}^{n}\left\vert E\left(  \varepsilon
_{it}\left(  \lambda_{T}\right)  \varepsilon_{jt}\left(  \lambda_{T}\right)
\right)  \right\vert \right]
\end{align*}
which is $O\left(  1\right)  $ based on (\ref{ave:rho}) and the boundedness of
$E\left\Vert \mathbf{f}_{t}\right\Vert ^{2}$, $b_{in}^{2}$ and $\sigma_{i}%
^{2}$ required by assumptions. So it follows
\begin{equation}
\frac{1}{nT}\sum_{i=1}^{n}b_{in}\mathbf{F}^{^{\prime}}\sigma_{i}%
\boldsymbol{\varepsilon}_{i\circ}\left(  \lambda_{T}\right)  =\frac{1}%
{\sqrt{nT}}\left(  \frac{1}{\sqrt{nT}}\sum_{i=1}^{n}\sum_{t=1}^{T}%
b_{in}\mathbf{f}_{t}\sigma_{i}\boldsymbol{\varepsilon}_{i\circ}\left(
\lambda_{T}\right)  \right)  =O_{p}\left(  \frac{1}{\sqrt{nT}}\right)  .
\label{Fun}%
\end{equation}
Result (\ref{hgg}) now follows using (\ref{FhatD}) and (\ref{Fun}) in
(\ref{hgg1}), and noting that by assumption $n$ and $T$ are of the same order.
Consider now (\ref{hdd}), which can be written as
\begin{align*}
\mathbf{d}_{2,nT}  &  =\frac{1}{n}\sum_{i=1}^{n}\left(  \frac{\boldsymbol{\hat
{\gamma}}_{i}}{\omega_{i,T}}-\frac{\boldsymbol{\gamma}_{i}}{\omega_{i,T}%
}\right) \\
&  =\frac{1}{n}\sum_{i=1}^{n}\left(  \frac{\boldsymbol{\hat{\gamma}}%
_{i}-\boldsymbol{\gamma}_{i}}{\sigma_{i}}\right)  \left(  1-\frac{\omega
_{i,T}-\sigma_{i}}{\omega_{i,T}}\right) \\
&  =\frac{1}{n}\sum_{i=1}^{n}\left(  \frac{\boldsymbol{\hat{\gamma}}%
_{i}-\boldsymbol{\gamma}_{i}}{\sigma_{i}}\right)  -\frac{1}{n}\sum_{i=1}%
^{n}\left(  \frac{\boldsymbol{\hat{\gamma}}_{i}-\boldsymbol{\gamma}_{i}%
}{\sigma_{i}}\right)  \left(  1-\frac{T}{\boldsymbol{\varepsilon}_{i\circ
}^{\prime}\mathbf{M}_{F}\boldsymbol{\varepsilon}_{i\circ}}\right)  .
\end{align*}
The first term of the above has the same form as (\ref{hgg}), and becomes
identical to it if we replace $a_{i}$ in (\ref{hgg}) with $1/\sigma_{i}$,
since by assumption $\inf_{i}(\sigma_{i})>c$. Hence, the order of the first
term is $O_{p}(\sqrt{\ln\left(  n\right)  /\left(  nT\right)  })$. Also the
second term is dominated by the first term, since $1-\left(  T^{-1}%
\boldsymbol{\varepsilon}_{i\circ}^{\prime}\mathbf{M}_{F}%
\boldsymbol{\varepsilon}_{i\circ}\right)  ^{-1}=O_{p}(T^{-1})$ based on result
(\ref{Es}). Therefore, (\ref{hdd}) is established as required. For
(\ref{hggg}), note that
\begin{align}
\mathbf{d}_{3,nT}  &  =\frac{1}{n}\sum_{i=1}^{n}b_{in}\left[  -\left(
\mathbf{\hat{F}}^{\prime}\mathbf{\hat{F}}\right)  ^{-1}\mathbf{\hat{F}%
}^{\prime}\left(  \mathbf{\hat{F}-F}\right)  \boldsymbol{\gamma}_{i}%
+\sigma_{i}\left(  \mathbf{\hat{F}}^{\prime}\mathbf{\hat{F}}\right)
^{-1}\mathbf{\hat{F}}^{\prime}\boldsymbol{\varepsilon}_{i\circ}\left(
\lambda_{T}\right)  \right]  \boldsymbol{\gamma}_{i}^{\prime}\nonumber\\
&  =-\left(  \frac{\mathbf{\hat{F}}^{\prime}\mathbf{\hat{F}}}{T}\right)
^{-1}\frac{\mathbf{\hat{F}}^{\prime}\left(  \mathbf{\hat{F}-F}\right)  }%
{T}\left(  \frac{1}{n}\sum_{i=1}^{n}b_{in}\boldsymbol{\gamma}_{i}%
\boldsymbol{\gamma}_{i}^{\prime}\right)  +\left(  \frac{\mathbf{\hat{F}%
}^{\prime}\mathbf{\hat{F}}}{T}\right)  ^{-1}T^{-1}\mathbf{\hat{F}}^{\prime
}\left(  \frac{1}{n}\sum_{i=1}^{n}b_{in}\sigma_{i}\boldsymbol{\varepsilon
}_{i\circ}\left(  \lambda_{T}\right)  \boldsymbol{\gamma}_{i}^{\prime}\right)
\nonumber\\
&  =-\left(  \frac{\mathbf{\hat{F}}^{\prime}\mathbf{\hat{F}}}{T}\right)
^{-1}\frac{\mathbf{\hat{F}}^{\prime}\left(  \mathbf{\hat{F}-F}\right)  }%
{T}\left(  \frac{1}{n}\sum_{i=1}^{n}b_{in}\boldsymbol{\gamma}_{i}%
\boldsymbol{\gamma}_{i}^{\prime}\right)  +\left(  \frac{\mathbf{\hat{F}%
}^{\prime}\mathbf{\hat{F}}}{T}\right)  ^{-1}\left(  \frac{1}{n}\sum_{i=1}%
^{n}T^{-1}b_{in}\sigma_{i}\left(  \mathbf{\hat{F}-F}\right)  ^{\prime
}\boldsymbol{\varepsilon}_{i\circ}\left(  \lambda_{T}\right)
\boldsymbol{\gamma}_{i}^{\prime}\right) \nonumber\\
&  +\left(  \frac{\mathbf{\hat{F}}^{\prime}\mathbf{\hat{F}}}{T}\right)
^{-1}\left(  \frac{1}{n}\sum_{i=1}^{n}T^{-1}b_{in}\sigma_{i}\mathbf{F}%
^{\prime}\boldsymbol{\varepsilon}_{i\circ}\left(  \lambda_{T}\right)
\boldsymbol{\gamma}_{i}^{\prime}\right)  . \label{hagg4}%
\end{align}
Recall that $\left(  T^{-1}\mathbf{\hat{F}}^{\prime}\mathbf{\hat{F}}\right)
^{-1}=O_{p}(1)$, and $n^{-1}\sum_{i=1}^{n}\boldsymbol{\gamma}_{i}%
\boldsymbol{\gamma}_{i}^{\prime}=O_{p}(1)$. Also note that $b_{in}$ is bounded
in $n$. Then using (\ref{baib3}) it follows that ($n$ and $T$ being of the
same order)
\[
\left(  \frac{\mathbf{\hat{F}}^{\prime}\mathbf{\hat{F}}}{T}\right)  ^{-1}%
\frac{\mathbf{\hat{F}}^{\prime}\left(  \mathbf{\hat{F}-F}\right)  }{T}\left(
n^{-1}\sum_{i=1}^{n}b_{in}\boldsymbol{\gamma}_{i}\boldsymbol{\gamma}%
_{i}^{\prime}\right)  =O_{p}\left(  \frac{1}{\min(n,T)}\right)  =O_{p}%
(\delta_{nT\ }^{-2}).
\]
Similarly, using (\ref{sup_baib1})
\[
\left(  \frac{\mathbf{\hat{F}}^{\prime}\mathbf{\hat{F}}}{T}\right)
^{-1}\left(  n^{-1}\sum_{i=1}^{n}T^{-1}\left(  \mathbf{\hat{F}-F}\right)
^{\prime}b_{in}\sigma_{i}\boldsymbol{\varepsilon}_{i\circ}\left(  \lambda
_{T}\right)  \boldsymbol{\gamma}_{i}^{\prime}\right)  =O_{p}\left(
\sqrt{\frac{\ln\left(  n\right)  }{nT}}\right)  .
\]
The last term of ( \ref{hagg4})\ can be written as $\left(  nT\right)
^{-1/2}\left(  T^{-1}\mathbf{\hat{F}}^{\prime}\mathbf{\hat{F}}\right)
^{-1}\left(  n^{-1/2}T^{-1/2}\sum_{i=1}^{n}b_{in}\sigma_{i}\mathbf{F}^{\prime
}\boldsymbol{\varepsilon}_{i\circ}\left(  \lambda_{T}\right)
\boldsymbol{\gamma}_{i}^{\prime}\right)  ,$ where $n^{-1/2}T^{-1/2}\sum
_{i=1}^{n}b_{in}\sigma_{i}\mathbf{F}^{\prime}\boldsymbol{\varepsilon}_{i\circ
}\left(  \lambda_{T}\right)  \boldsymbol{\gamma}_{i}^{\prime}$ is an
$m_{0}\times m_{0}$ matrix with its $(j,j^{\prime})$ element given by
$n^{-1/2}T^{-1/2}\sum_{i=1}^{n}\sum_{t=1}^{T}$ $b_{in}\sigma_{i}%
f_{jt}\varepsilon_{it}\left(  \lambda_{T}\right)  \gamma_{ij^{\prime}}$ for
$j,j^{\prime}=1,2,\ldots,m_{0}$. It can be further shown
\begin{align*}
&  E\left(  n^{-1/2}T^{-1/2}\sum_{i=1}^{n}\sum_{t=1}^{T}b_{in}\sigma_{i}%
f_{jt}\gamma_{ij^{\prime}}\varepsilon_{it}\left(  \lambda_{T}\right)  \right)
^{2}\\
&  =\frac{1}{nT}\sum_{i=1}^{n}\sum_{i^{\prime}=1}^{n}\sum_{t=1}^{T}%
b_{in}b_{i^{\prime}n}\sigma_{i}\sigma_{i^{\prime}}\gamma_{ij^{\prime}}%
\gamma_{i^{\prime}j^{\prime}}E\left(  f_{jt}^{2}\right)  E\left(
\varepsilon_{it}\left(  \lambda_{T}\right)  \varepsilon_{i^{\prime}t}\left(
\lambda_{T}\right)  \right) \\
&  \leq\frac{1}{nT}\sum_{i=1}^{n}\sum_{i^{\prime}=1}^{n}\sum_{t=1}%
^{T}\left\vert b_{in}b_{i^{\prime}n}\sigma_{i}\sigma_{i^{\prime}}%
\gamma_{ij^{\prime}}\gamma_{i^{\prime}j^{\prime}}\right\vert E\left(
f_{jt}^{2}\right)  \left\vert E\left(  \varepsilon_{it}\left(  \lambda
_{T}\right)  \varepsilon_{i^{\prime}t}\left(  \lambda_{T}\right)  \right)
\right\vert \\
&  \leq\left(  \sup_{i}b_{in}^{2}\right)  \left(  \sup_{i}\gamma_{ij^{\prime}%
}^{2}\right)  \left(  \sup_{i}\sigma_{i}^{2}\right)  E\left(  f_{jt}%
^{2}\right)  \left[  \frac{1}{nT}\sum_{i=1}^{n}\sum_{i^{\prime}=1}^{n}%
\sum_{t=1}^{T}\left\vert E\left(  \varepsilon_{it}\left(  \lambda_{T}\right)
\varepsilon_{i^{\prime}t}\left(  \lambda_{T}\right)  \right)  \right\vert
\right]
\end{align*}
which is $O\left(  1\right)  $ based on (\ref{ave:rho}) and the boundedness of
$b_{in}^{2}$, $\gamma_{ij^{\prime}}^{2}$, $\sigma_{i}^{2}$ and $E\left(
f_{jt}^{2}\right)  $. Consequentially, $n^{-1/2}T^{-1/2}\sum_{i=1}^{n}%
b_{in}\sigma_{i}\mathbf{F}^{\prime}\boldsymbol{\varepsilon}_{i\circ}\left(
\lambda_{T}\right)  \boldsymbol{\gamma}_{i}^{\prime}=O_{p}\left(  1\right)  $
and the last term of (\ref{hagg4}) are also $O_{p}(\delta_{nT\ }^{-2})$. Thus
result (\ref{hggg}) is established, as required. To prove (\ref{hgg_b}) we
first write it as%
\[
\mathbf{d}_{4,nT}=\left(  \sqrt{\frac{1}{T}}\right)  \frac{1}{n}\sum_{i=1}%
^{n}q_{iT}\left(  \boldsymbol{\hat{\gamma}}_{i}-\boldsymbol{\gamma}%
_{i}\right)  ,
\]
where
\[
q_{iT}=\sqrt{T}\left(  \omega_{i,T}-\sigma_{i}\right)  =\sigma_{i}\sqrt
{T}\left[  \left(  \frac{\boldsymbol{\varepsilon}_{i\circ}^{\prime}%
\mathbf{M}_{F}\boldsymbol{\varepsilon}_{i\circ}}{T}\right)  ^{1/2}-1\right]
,
\]
and conditional on $\mathbf{F}$ and $\sigma_{i}$, $q_{iT}$ are independently
distributed across $i$. Using results in Lemma \ref{Emom} it is easily seen
that $E\left(  q_{iT}\right)  =O(T^{-1/2})$ and $Var\left(  q_{iT}\right)
=O(1)$, and hence $n^{-1}\sum_{i=1}^{n}q_{iT}^{2}=O_{p}(1)$. Also by
Cauchy-Schwarz inequality we have
\[
\left\Vert \mathbf{d}_{4,nT}\right\Vert \leq\left(  \sqrt{\frac{1}{T}}\right)
\left(  n^{-1}\sum_{i=1}^{n}q_{iT}^{2}\right)  ^{1/2}\left(  n^{-1/2}%
\left\Vert \mathbf{\hat{\Gamma}}-\mathbf{\Gamma}\right\Vert \right)  ,
\]
where $T^{-1}n=\ominus(1)$, and by (\ref{lm2.2}) $n^{-1/2}\left\Vert
\mathbf{\hat{\Gamma}}-\mathbf{\Gamma}\right\Vert =O_{p}(\delta_{nT}^{-1})$,
and (\ref{hgg_b}) is established. Result (\ref{hgg_a}) follows similarly, with
$q_{iT}$ re-defined as $q_{iT}=\sigma_{i}^{-1}\sqrt{T}\left[  \left(
\frac{\boldsymbol{\varepsilon}_{i\circ}^{\prime}\mathbf{M}_{F}%
\boldsymbol{\varepsilon}_{i\circ}}{T}\right)  ^{-1/2}-1\right]  $, and noting
that $\sup_{i}(1/\sigma_{i}^{2})<C$, and using results in Lemma \ref{Emom}.
Result (\ref{hgg_c}) is established as
\begin{align*}
\left\Vert \mathbf{d}_{6,nT}\right\Vert  &  =\left\Vert \frac{1}{n}\sum
_{i=1}^{n}\left(  \hat{\sigma}_{i,T}-\omega_{i,T}\right)  \left(
\boldsymbol{\hat{\gamma}}_{i}-\boldsymbol{\gamma}_{i}\right)  \right\Vert
\leq\frac{1}{n}\sum_{i=1}^{n}\left\Vert \left(  \hat{\sigma}_{i,T}%
-\omega_{i,T}\right)  \left(  \boldsymbol{\hat{\gamma}}_{i}-\boldsymbol{\gamma
}_{i}\right)  \right\Vert \\
&  \leq\left(  \frac{1}{n}\sum_{i=1}^{n}\left\vert \hat{\sigma}_{i,T}%
-\omega_{i,T}\right\vert \right)  \left(  \sup_{i}\left\Vert \boldsymbol{\hat
{\gamma}}_{i}-\boldsymbol{\gamma}_{i}\right\Vert \right)  =O_{p}\left[
\left(  \frac{\ln\left(  n\right)  }{T}\right)  ^{3/2}\right]
\end{align*}
where by (\ref{sup_|hatg-g|}) $\sup_{i}\left\Vert \boldsymbol{\hat{\gamma}%
}_{i}-\boldsymbol{\gamma}_{i}\right\Vert =O_{p}\left(  \sqrt{\ln\left(
n\right)  /T}\right)  $, and by (\ref{Ews}) $n^{-1}\sum_{i=1}^{n}\left\vert
\hat{\sigma}_{i,T}-\omega_{i,T}\right\vert =O_{p}(\ln\left(  n\right)  /T)$.
Similarly by (\ref{lm2.2}) and (\ref{Ewin}) we have%
\begin{align*}
\left\Vert \mathbf{d}_{7,nT}\right\Vert  &  =\left\Vert \frac{1}{n}\sum
_{i=1}^{n}\left(  \frac{1}{\hat{\sigma}_{i,T}}-\frac{1}{\omega_{i,T}}\right)
\left(  \boldsymbol{\hat{\gamma}}_{i}-\boldsymbol{\gamma}_{i}\right)
\right\Vert \leq\frac{1}{n}\sum_{i=1}^{n}\left\Vert \left(  \frac{1}%
{\hat{\sigma}_{i,T}}-\frac{1}{\omega_{i,T}}\right)  \left(  \boldsymbol{\hat
{\gamma}}_{i}-\boldsymbol{\gamma}_{i}\right)  \right\Vert \\
&  \leq\left(  \frac{1}{n}\sum_{i=1}^{n}\left\vert \frac{1}{\hat{\sigma}%
_{i,T}}-\frac{1}{\omega_{i,T}}\right\vert \right)  \left(  \sup_{i}\left\Vert
\boldsymbol{\hat{\gamma}}_{i}-\boldsymbol{\gamma}_{i}\right\Vert \right)
=O_{p}\left[  \left(  \frac{\ln\left(  n\right)  }{T}\right)  ^{3/2}\right]  .
\end{align*}

\end{proof}

\begin{lemma}
\label{lm5}Consider the latent factor model given by (\ref{mod2}) and
(\ref{uit:p}). Suppose that Assumptions \ref{ass:factors}-\ref{ass:ws} hold
and $\left(  n,T\right)  \rightarrow\infty$, such that $n/T\rightarrow
\kappa\,,$ for $0<\kappa<\infty$. Then
\begin{align}
p_{nT}\left(  \lambda_{T}\right)   &  =\frac{1}{\sqrt{T}}\sum_{t=1}%
^{T}s_{t,nT}^{2}\left(  \lambda_{T}\right)  =o_{p}(1),\label{pnTS}\\
q_{nT}\left(  \lambda_{T}\right)   &  =\frac{1}{\sqrt{T}}\sum_{t=1}^{T}%
\psi_{t,nT}\left(  \lambda_{T}\right)  s_{t,nT}\left(  \lambda_{T}\right)
=o_{p}(1), \label{qnTS}%
\end{align}
where $\psi_{t,nT}\left(  \lambda_{T}\right)  $ and $s_{t,nT}\left(
\lambda_{T}\right)  $ are defined by (\ref{epsinT1}) and (\ref{snT}), respectively.
\end{lemma}

\begin{proof}
Using (\ref{snT}), recall that
\begin{align}
s_{t,nT}\left(  \lambda_{T}\right)   &  =\boldsymbol{\varphi}_{nT}^{\prime
}\left[  n^{-1/2}\sum_{i=1}^{n}\left(  \boldsymbol{\hat{\gamma}}%
_{i}-\boldsymbol{\gamma}_{i}\right)  \sigma_{i}\varepsilon_{it}\left(
\lambda_{T}\right)  \right]  +\boldsymbol{\varphi}_{nT}^{\prime}\left[
n^{-1/2}\sum_{i=1}^{n}\left(  \boldsymbol{\hat{\gamma}}_{i}-\boldsymbol{\gamma
}_{i}\right)  \boldsymbol{\gamma}_{i}^{\prime}\right]  \mathbf{f}%
_{t}\nonumber\\
&  +\left[  n^{-1/2}\sum_{i=1}^{n}\left(  \boldsymbol{\hat{\delta}}%
_{i,T}-\boldsymbol{\delta}_{i,T}\right)  ^{\prime}\right]  \mathbf{f}%
_{t}+\left[  n^{-1/2}\sum_{i=1}^{n}\left(  \boldsymbol{\hat{\delta}}%
_{i,T}-\boldsymbol{\delta}_{i,T}\right)  ^{^{\prime}}\right]  \left(
\mathbf{\hat{f}}_{t}-\mathbf{f}_{t}\right)  .\nonumber\\
&  \label{sntTx}%
\end{align}
We also note that using (\ref{epsinT}), $\psi_{t,nT}\left(  \lambda
_{T}\right)  $ can be written as%
\begin{equation}
\psi_{t,nT}\left(  \lambda_{T}\right)  =\xi_{t,n}\left(  \lambda_{T}\right)
-\left(  \boldsymbol{\varphi}_{nT}-\boldsymbol{\varphi}_{n}\right)  ^{\prime
}\boldsymbol{\kappa}_{t,n}\left(  \lambda_{T}\right)  +\upsilon_{t,nT}\left(
\lambda_{T}\right)  \label{egzintT}%
\end{equation}
where
\begin{align}
\xi_{t,n}\left(  \lambda_{T}\right)   &  =\frac{1}{\sqrt{n}}\sum_{i=1}%
^{n}a_{i,n}\varepsilon_{it}\left(  \lambda_{T}\right)  ,\text{ }%
a_{i,n}=1-\sigma_{i}\boldsymbol{\varphi}_{n}^{\prime}\boldsymbol{\gamma}%
_{i},\label{egziS}\\
\boldsymbol{\kappa}_{t,n}\left(  \lambda_{T}\right)   &  =\frac{1}{\sqrt{n}%
}\sum_{i=1}^{n}\boldsymbol{\gamma}_{i}\sigma_{i}\varepsilon_{it}\left(
\lambda_{T}\right)  ,\label{kapaS}\\
\upsilon_{t,nT}\left(  \lambda_{T}\right)   &  =\frac{1}{\sqrt{n}}\sum
_{i=1}^{n}\left[  \frac{1}{\left(  \boldsymbol{\varepsilon}_{i\circ}^{\prime
}\mathbf{M}_{F}\boldsymbol{\varepsilon}_{i\circ}/T\right)  ^{1/2}}-1\right]
\varepsilon_{it}\left(  \lambda_{T}\right)  . \label{vS}%
\end{align}
After squaring $s_{t,nT}\left(  \lambda_{T}\right)  $, we end up with
$p_{nT}\left(  \lambda_{T}\right)  =\sum_{j=1}^{10}A_{j,nT}\left(  \lambda
_{T}\right)  $, composed of four squared terms and six cross product terms.
For the first square term we have
\[
A_{1,nT}\left(  \lambda_{T}\right)  =\sqrt{T}\boldsymbol{\varphi}_{nT}%
^{\prime}\left(  \frac{1}{T}\sum_{t=1}^{T}\mathbf{b}_{t,n}\left(  \lambda
_{T}\right)  \mathbf{b}_{t,n}^{^{\prime}}\left(  \lambda_{T}\right)  \right)
\boldsymbol{\varphi}_{nT},
\]
where $\mathbf{b}_{t,n}\left(  \lambda_{T}\right)  =n^{-1/2}\sum_{i=1}%
^{n}\left(  \boldsymbol{\hat{\gamma}}_{i}-\boldsymbol{\gamma}_{i}\right)
\sigma_{i}\varepsilon_{it}\left(  \lambda_{T}\right)  $. Let $\mathbf{u}%
_{\circ t}\left(  \lambda_{T}\right)  =\left(  \sigma_{1}\varepsilon
_{1t}\left(  \lambda_{T}\right)  ,\sigma_{2}\varepsilon_{2t}\left(
\lambda_{T}\right)  ,\ldots,\sigma_{n}\varepsilon_{nt}\left(  \lambda
_{T}\right)  \right)  ^{\prime}$ so that $\mathbf{b}_{t,n}\left(  \lambda
_{T}\right)  =n^{-1/2}\left(  \boldsymbol{\hat{\Gamma}}-\boldsymbol{\Gamma
}\right)  ^{^{\prime}}\mathbf{u}_{\circ t}\left(  \lambda_{T}\right)  $. Then%
\begin{equation}
\left\vert A_{1,nT}\left(  \lambda_{T}\right)  \right\vert \leq\frac{\sqrt{T}%
}{n}\left\Vert \boldsymbol{\varphi}_{nT}\right\Vert ^{2}\left\Vert
\boldsymbol{\hat{\Gamma}}-\boldsymbol{\Gamma}\right\Vert ^{2}\left\Vert
\mathbf{V}_{T}\left(  \lambda_{T}\right)  \right\Vert , \label{A1}%
\end{equation}
where $\mathbf{V}_{T}\left(  \lambda_{T}\right)  =T^{-1}\sum_{t=1}%
^{T}\mathbf{u}_{\circ t}\left(  \lambda_{T}\right)  \mathbf{u}_{\circ
t}^{^{\prime}}\left(  \lambda_{T}\right)  .$ Since $\boldsymbol{\varphi}%
_{n}=n^{-1}\sum_{i=1}^{n}\boldsymbol{\gamma}_{i}/\sigma_{i}=O(1)$ by
Assumption \ref{ass:loadings}, and $\boldsymbol{\varphi}_{nT}%
=\boldsymbol{\varphi}_{n}+o_{p}\left(  1\right)  $ by result (\ref{phinLemma})
in Lemma \ref{PHIn}, then
\begin{equation}
\boldsymbol{\varphi}_{nT}=O_{p}\left(  1\right)  . \label{varphi_nT}%
\end{equation}
Further, using (\ref{uit:p}) in the paper, $\mathbf{u}_{\circ t}\left(
\lambda_{T}\right)  =\mathbf{D}_{0}\left(  \boldsymbol{\varepsilon}_{\circ
t}+\lambda_{T}\mathbf{W}\boldsymbol{\varepsilon}_{\circ t}\right)  $, where
$\mathbf{D}_{0}=diag\left(  \sigma_{1},\sigma_{2},\ldots,\sigma_{n}\right)  $,
$\boldsymbol{\varepsilon}_{\circ t}=\left(  \varepsilon_{1t},\varepsilon
_{2t},\ldots,\varepsilon_{nT}\right)  ^{\prime}$, and $\mathbf{W}=\left(
w_{ij}\right)  $. Therefore%
\[
\mathbf{u}_{\circ t}\left(  \lambda_{T}\right)  \mathbf{u}_{\circ t}^{\prime
}\left(  \lambda_{T}\right)  =\mathbf{D}_{0}\left(  \boldsymbol{\varepsilon
}_{\circ t}\boldsymbol{\varepsilon}_{\circ t}^{\prime}+\lambda_{T}%
^{2}\mathbf{W}\boldsymbol{\varepsilon}_{\circ t}\boldsymbol{\varepsilon
}_{\circ t}^{\prime}\mathbf{W}^{\prime}+\lambda_{T}\boldsymbol{\varepsilon
}_{\circ t}\boldsymbol{\varepsilon}_{\circ t}^{\prime}\mathbf{W}^{\prime
}+\lambda_{T}\mathbf{W}\boldsymbol{\varepsilon}_{\circ t}%
\boldsymbol{\varepsilon}_{\circ t}^{\prime}\right)  \mathbf{D}_{0},
\]
and%
\[
\mathbf{V}_{T}\left(  \lambda_{T}\right)  =\mathbf{D}_{0}\left(
\mathbf{V}_{\varepsilon T}+\lambda_{T}^{2}\mathbf{W\mathbf{V}}_{\varepsilon
T}\mathbf{W}^{\prime}+\lambda_{T}\mathbf{V}_{\varepsilon T}\mathbf{W}^{\prime
}+\lambda_{T}\mathbf{W\mathbf{V}}_{\varepsilon T}\right)  \mathbf{D}_{0},
\]
where $\mathbf{V}_{\varepsilon T}=T^{-1}\sum_{t=1}^{T}\boldsymbol{\varepsilon
}_{\circ t}\boldsymbol{\varepsilon}_{\circ t}^{\prime}$. It follows that
\[
\left\Vert \mathbf{V}_{T}\left(  \lambda_{T}\right)  \right\Vert
\leq\left\Vert \mathbf{\mathbf{V}}_{\varepsilon T}\right\Vert \left\Vert
\mathbf{D}_{0}\right\Vert ^{2}\left(  \left\Vert \mathbf{I}_{n}\right\Vert
+\lambda_{T}^{2}\left\Vert \mathbf{W}\right\Vert \left\Vert \mathbf{W}%
^{\prime}\right\Vert +\left\vert \lambda_{T}\right\vert \left\Vert
\mathbf{W}^{\prime}\right\Vert +\left\vert \lambda_{T}\right\vert \left\Vert
\mathbf{W}\right\Vert \right)  .
\]
Note that $\left\Vert \mathbf{D}_{0}\right\Vert $ and $\left\Vert
\mathbf{W}\right\Vert $ are both bounded, and by part (b) of Assumption
\ref{ass:errors} we have $\left\Vert \mathbf{V}_{\varepsilon T}\right\Vert
=\mu_{max}\left(  \mathbf{V}_{\varepsilon T}\right)  =O_{p}\left(  \frac{n}%
{T}\right)  $. Hence,
\begin{equation}
\left\Vert \mathbf{V}_{T}\left(  \lambda_{T}\right)  \right\Vert =O_{p}\left(
\frac{n}{T}\right)  . \label{|VT|}%
\end{equation}
Since $n$ and $T$ are of the same order of magnitude, then using results
(\ref{lm2.2}), (\ref{varphi_nT}), and (\ref{|VT|}) in (\ref{A1}) yields
\[
\left\vert A_{1,nT}\left(  \lambda_{T}\right)  \right\vert =\frac{\sqrt{T}}%
{n}O_{p}\left(  \frac{n}{\delta_{nT}^{2}}\right)  =O_{p}\left(  \frac
{1}{\delta_{nT}}\right)  .
\]
For the second squared term we have
\[
A_{2,nT}\left(  \lambda_{T}\right)  =\sqrt{T}\left[  n^{-1/2}\sum_{i=1}%
^{n}\left(  \boldsymbol{\hat{\delta}}_{i,T}-\boldsymbol{\delta}_{i,T}\right)
^{\prime}\right]  \left(  T^{-1}\mathbf{F}^{\prime}\mathbf{F}\right)  \left[
n^{-1/2}\sum_{i=1}^{n}\left(  \boldsymbol{\hat{\delta}}_{i,T}%
-\boldsymbol{\delta}_{i,T}\right)  \right]  .
\]
where $T^{-1}\mathbf{F}^{\prime}\mathbf{F}=O_{p}(1)$ and using (\ref{hdd})
$n^{-1/2}\sum_{i=1}^{n}\left(  \boldsymbol{\hat{\delta}}_{i,T}%
-\boldsymbol{\delta}_{i,T}\right)  =O_{p}\left(  \sqrt{\ln\left(  n\right)
/T}\right)  $. Hence, $A_{2,nT}\left(  \lambda_{T}\right)  =$ $O_{p}\left(
\ln\left(  n\right)  /\sqrt{T}\right)  =o_{p}(1)$. Similarly,%
\begin{align*}
A_{3,nT}\left(  \lambda_{T}\right)   &  =\sqrt{T}\boldsymbol{\varphi}%
_{nT}^{\prime}\left[  n^{-1/2}\sum_{i=1}^{n}\left(  \boldsymbol{\hat{\gamma}%
}_{i}-\boldsymbol{\gamma}_{i}\right)  \boldsymbol{\gamma}_{i}^{\prime}\right]
\left(  T^{-1}\sum_{t=1}^{T}\mathbf{f}_{t}\mathbf{f}_{t}^{\prime}\right)
\left[  n^{-1/2}\sum_{i=1}^{n}\boldsymbol{\gamma}_{i}\left(  \boldsymbol{\hat
{\gamma}}-\boldsymbol{\gamma}_{i}\right)  ^{\prime}\right]
\boldsymbol{\varphi}_{nT}\\
&  =\sqrt{T}\boldsymbol{\varphi}_{nT}^{\prime}\left[  n^{-1/2}\sum_{i=1}%
^{n}\left(  \boldsymbol{\hat{\gamma}}_{i}-\boldsymbol{\gamma}_{i}\right)
\boldsymbol{\gamma}_{i}^{\prime}\right]  \left(  T^{-1}\mathbf{F}^{\prime
}\mathbf{F}\right)  \left[  n^{-1/2}\sum_{i=1}^{n}\boldsymbol{\gamma}%
_{i}\left(  \boldsymbol{\hat{\gamma}}-\boldsymbol{\gamma}_{i}\right)
^{\prime}\right]  \boldsymbol{\varphi}_{nT},
\end{align*}
where $\Vert\boldsymbol{\varphi}_{nT}\Vert$ is bounded by (\ref{varphi_nT}),
and by (\ref{hggg}) $n^{-1/2}\sum_{i=1}^{n}\left(  \boldsymbol{\hat{\gamma}%
}_{i}-\boldsymbol{\gamma}_{i}\right)  \boldsymbol{\gamma}_{i}^{\prime}%
=O_{p}\left(  \sqrt{\ln\left(  n\right)  /T}\right)  $. Hence, $A_{3,nT}%
\left(  \lambda_{T}\right)  =$ $O_{p}\left(  \ln\left(  n\right)  /\sqrt
{T}\right)  =o_{p}(1)$. For the final squared term,
\[
A_{4,nT}\left(  \lambda_{T}\right)  =\sqrt{T}\left[  \frac{1}{\sqrt{n}}%
\sum_{i=1}^{n}\left(  \boldsymbol{\hat{\delta}}_{i,T}-\boldsymbol{\delta
}_{i,T}\right)  \right]  ^{\prime}\left(  \frac{\left\Vert \hat{\mathbf{F}%
}-\mathbf{F}\right\Vert ^{2}}{T}\right)  \left[  \frac{1}{\sqrt{n}}\sum
_{i=1}^{n}\left(  \boldsymbol{\hat{\delta}}_{i,T}-\boldsymbol{\delta}%
_{i,T}\right)  \right]  ,
\]
By results (\ref{lm2.1}) and (\ref{hdd}), it follows $A_{4,nT}\left(
\lambda_{T}\right)  =\sqrt{T}O_{p}\left(  \delta_{nT\ }^{-2}\right)  \times
O_{p}\left(  \ln\left(  n\right)  /T\right)  =o_{p}(1)$. The probability
orders of the cross product terms of $p_{nT}\left(  \lambda_{T}\right)  $,
namely $A_{5,NT}\left(  \lambda_{T}\right)  ,\ldots,A_{10,NT}\left(
\lambda_{T}\right)  $, are also easily seen to be $o_{p}(1)$, by application
of the Cauchy-Schwarz inequality to the product pairs of the terms
$A_{1,NT}\left(  \lambda_{T}\right)  ,A_{2,NT}\left(  \lambda_{T}\right)
,A_{3,NT}\left(  \lambda_{T}\right)  ,$ and $A_{4,NT}\left(  \lambda
_{T}\right)  $. Thus, overall $p_{nT}\left(  \lambda_{T}\right)  =o_{p}(1)$,
as required. Consider now $q_{nT}\left(  \lambda_{T}\right)  $ and note that
it can be written as (using (\ref{egzintT}) in (\ref{qnTS}))
\begin{align*}
q_{nT}\left(  \lambda_{T}\right)   &  =\frac{1}{\sqrt{T}}\sum_{t=1}%
^{T}s_{t,nT}\left(  \lambda_{T}\right)  \xi_{t,n}\left(  \lambda_{T}\right)
-\left(  \boldsymbol{\varphi}_{nT}-\boldsymbol{\varphi}_{n}\right)  ^{\prime
}\frac{1}{\sqrt{T}}\sum_{t=1}^{T}s_{t,nT}\left(  \lambda_{T}\right)
\boldsymbol{\kappa}_{t,n}\left(  \lambda_{T}\right) \\
&  +\frac{1}{\sqrt{T}}\sum_{t=1}^{T}s_{t,nT}\left(  \lambda_{T}\right)
\upsilon_{t,nT}\left(  \lambda_{T}\right)  ,
\end{align*}
where $\xi_{t,n}\left(  \lambda_{T}\right)  $, $\upsilon_{t,nT}\left(
\lambda_{T}\right)  $ and $\boldsymbol{\kappa}_{t,n}\left(  \lambda
_{T}\right)  \mathbf{\ }$are given by (\ref{egziS}), (\ref{kapaS}) and
(\ref{vS}), respectively. The first term of the above can be written as%
\begin{align*}
\frac{1}{\sqrt{T}}\sum_{t=1}^{T}s_{t,nT}\left(  \lambda_{T}\right)  \xi
_{t,n}\left(  \lambda_{T}\right)   &  =\boldsymbol{\varphi}_{nT}^{\prime
}\left[  \frac{1}{\sqrt{n}}\sum_{i=1}^{n}\left(  \boldsymbol{\hat{\gamma}}%
_{i}-\boldsymbol{\gamma}_{i}\right)  \frac{1}{\sqrt{T}}\sum_{t=1}^{T}\xi
_{t,n}\left(  \lambda_{T}\right)  \sigma_{i}\varepsilon_{it}\left(
\lambda_{T}\right)  \right] \\
&  +\boldsymbol{\varphi}_{nT}^{\prime}\left[  n^{-1/2}\sum_{i=1}^{n}\left(
\boldsymbol{\hat{\gamma}}_{i}-\boldsymbol{\gamma}_{i}\right)
\boldsymbol{\gamma}_{i}^{\prime}\right]  \left(  \frac{1}{\sqrt{T}}\sum
_{t=1}^{T}\xi_{t,n}\left(  \lambda_{T}\right)  \mathbf{f}_{t}\right) \\
&  +\left[  n^{-1/2}\sum_{i=1}^{n}\left(  \boldsymbol{\hat{\delta}}%
_{i,T}-\boldsymbol{\delta}_{i,T}\right)  ^{\prime}\right]  \left(  \frac
{1}{\sqrt{T}}\sum_{t=1}^{T}\xi_{t,n}\left(  \lambda_{T}\right)  \mathbf{f}%
_{t}\right) \\
&  +\left[  n^{-1/2}\sum_{i=1}^{n}\left(  \boldsymbol{\hat{\delta}}%
_{i,T}-\boldsymbol{\delta}_{i,T}\right)  ^{^{\prime}}\right]  \left[  \frac
{1}{\sqrt{T}}\sum_{t=1}^{T}\xi_{t,n}\left(  \lambda_{T}\right)  \left(
\mathbf{\hat{f}}_{t}-\mathbf{f}_{t}\right)  \right] \\
&  =\sum_{j=1}^{4}B_{j,nT}\left(  \lambda_{T}\right)  .
\end{align*}
Using (\ref{egziS}), $B_{1,nT}\left(  \lambda_{T}\right)  $ can be written as
\[
B_{1,nT}\left(  \lambda_{T}\right)  =\boldsymbol{\varphi}_{nT}^{\prime}\left[
\frac{\sqrt{T}}{\sqrt{n}}\sum_{i=1}^{n}\left(  \boldsymbol{\hat{\gamma}}%
_{i}-\boldsymbol{\gamma}_{i}\right)  \frac{1}{\sqrt{n}}\sum_{j=1}^{n}%
a_{j,n}\left(  \frac{1}{T}\sum_{t=1}^{T}\sigma_{i}\varepsilon_{jt}\left(
\lambda_{T}\right)  \varepsilon_{it}\left(  \lambda_{T}\right)  \right)  ,
\right]
\]
where $a_{i,n}=1-\sigma_{i}\boldsymbol{\varphi}_{n}^{\prime}\boldsymbol{\gamma
}_{i}$ and $\boldsymbol{\varphi}_{nT}=O_{p}(1)$. Also since $\varepsilon
_{it}\left(  \lambda_{T}\right)  $ are independently distributed over $t$ and
weakly cross-sectionally dependent, and $n$ and $T$ are of the same order,
then
\[
B_{1,nT}\left(  \lambda_{T}\right)  =O_{p}\left(  \frac{1}{\sqrt{n}}\sum
_{i=1}^{n}a_{i,n}\sigma_{i}\left(  \boldsymbol{\hat{\gamma}}_{i}%
-\boldsymbol{\gamma}_{i}\right)  \right)  .
\]
Further, letting $b_{in}=a_{i,n}\sigma_{i}$, it follows from (\ref{hgg}) that
$n^{-1/2}\sum_{i=1}^{n}a_{i,n}\sigma_{i}\left(  \boldsymbol{\hat{\gamma}}%
_{i}-\boldsymbol{\gamma}_{i}\right)  =O_{p}(\sqrt{\ln\left(  n\right)  /T})$,
which in turn establishes that $B_{1,nT}\left(  \lambda_{T}\right)  =o_{p}%
(1)$. Similarly, using (\ref{egziS}), $B_{2,nT}\left(  \lambda_{T}\right)  $
can be written as%
\[
B_{2,nT}\left(  \lambda_{T}\right)  =\boldsymbol{\varphi}_{nT}^{\prime}\left[
n^{-1/2}\sum_{i=1}^{n}\left(  \boldsymbol{\hat{\gamma}}_{i}-\boldsymbol{\gamma
}_{i}\right)  \boldsymbol{\gamma}_{i}^{\prime}\right]  \left(  \frac{1}%
{\sqrt{nT}}\sum_{j=1}^{n}\sum_{t=1}^{T}a_{j,n}\mathbf{f}_{t}\varepsilon
_{jt}\left(  \lambda_{T}\right)  \right)  ,
\]
where $\boldsymbol{\varphi}_{nT}=O_{p}(1)$. Under parts (a) and (c) of
Assumption \ref{ass:errors}
$
\frac{1}{\sqrt{nT}}\sum_{j=1}^{n}\sum_{t=1}^{T}a_{j,n}\mathbf{f}%
_{t}\varepsilon_{jt}\left(  \lambda_{T}\right)  =O_{p}(1).
$
Using this result together with (\ref{hggg}) it follows that $B_{2,nT}\left(
\lambda_{T}\right)  =o_{p}(1)$. Similarly, using (\ref{hdd}) we can establish
that $B_{3,nT}\left(  \lambda_{T}\right)  =o_{p}(1)$. The final term,
$B_{4,nT}\left(  \lambda_{T}\right)  $, is dominated by the third term and is
also $o_{p}(1)$. Thus overall, $T^{-1/2}\sum_{t=1}^{T}s_{t,nT}\left(
\lambda_{T}\right)  \xi_{t,n}\left(  \lambda_{T}\right)  =o_{p}(1)$. Using the
same line of reasoning, it is also readily established that $T^{-1/2}%
\sum_{t=1}^{T}s_{t,nT}\left(  \lambda_{T}\right)  \boldsymbol{\kappa}%
_{t,n}\left(  \lambda_{T}\right)  =o_{p}(1)$, considering that,
$\boldsymbol{\kappa}_{t,n}\left(  \lambda_{T}\right)  $ $=n^{-1/2}\sum
_{i=1}^{n}\boldsymbol{\gamma}_{i}\sigma_{i}\varepsilon_{it}\left(  \lambda
_{T}\right)  $ has the same format as $\xi_{t,n}\left(  \lambda_{T}\right)  $,
and in addition by (\ref{phinLemma}) $\boldsymbol{\varphi}_{nT}%
-\boldsymbol{\varphi}_{n}=O_{p}(n^{-1/2}T^{-1/2})+O_{p}(T^{-1})$. Finally, the
last term of $q_{nT}$ is given by%
\begin{align*}
\frac{1}{\sqrt{T}}\sum_{t=1}^{T}s_{t,nT}\left(  \lambda_{T}\right)
\upsilon_{t,nT}\left(  \lambda_{T}\right)   &  =\frac{1}{\sqrt{T}}\sum
_{t=1}^{T}\upsilon_{t,nT}\left(  \lambda_{T}\right)  \boldsymbol{\varphi}%
_{nT}^{\prime}\left[  n^{-1/2}\sum_{i=1}^{n}\left(  \boldsymbol{\hat{\gamma}%
}_{i}-\boldsymbol{\gamma}_{i}\right)  \sigma_{i}\varepsilon_{it}\left(
\lambda_{T}\right)  \right] \\
&  +\frac{1}{\sqrt{T}}\sum_{t=1}^{T}\upsilon_{t,nT}\left(  \lambda_{T}\right)
\boldsymbol{\varphi}_{nT}^{\prime}\left[  n^{-1/2}\sum_{i=1}^{n}\left(
\boldsymbol{\hat{\gamma}}_{i}-\boldsymbol{\gamma}_{i}\right)
\boldsymbol{\gamma}_{i}^{\prime}\right]  \mathbf{f}_{t}\\
&  +\frac{1}{\sqrt{T}}\sum_{t=1}^{T}\upsilon_{t,nT}\left(  \lambda_{T}\right)
\left[  n^{-1/2}\sum_{i=1}^{n}\left(  \boldsymbol{\hat{\delta}}_{i,T}%
-\boldsymbol{\delta}_{i,T}\right)  ^{\prime}\right]  \mathbf{f}_{t}\\
&  +\frac{1}{\sqrt{T}}\sum_{t=1}^{T}\upsilon_{t,nT}\left(  \lambda_{T}\right)
\left[  n^{-1/2}\sum_{i=1}^{n}\left(  \boldsymbol{\hat{\delta}}_{i,T}%
-\boldsymbol{\delta}_{i,T}\right)  ^{^{\prime}}\right]  \left(  \mathbf{\hat
{f}}_{t}-\mathbf{f}_{t}\right) \\
&  =\sum_{j=1}^{4}C_{j,nT}\left(  \lambda_{T}\right)  .
\end{align*}
Using (\ref{vS}) we have
\begin{align*}
C_{1,nT}\left(  \lambda_{T}\right)   &  =\frac{1}{\sqrt{T}}\sum_{t=1}^{T}%
\frac{1}{\sqrt{n}}\sum_{j=1}^{n}\left(  \frac{1}{\left(
\boldsymbol{\varepsilon}_{j\circ}^{\prime}\mathbf{M}_{F}%
\boldsymbol{\varepsilon}_{j\circ}/T\right)  ^{1/2}}-1\right)  \varepsilon
_{jt}\left(  \lambda_{T}\right)  \boldsymbol{\varphi}_{nT}^{\prime}\left[
n^{-1/2}\sum_{i=1}^{n}\left(  \boldsymbol{\hat{\gamma}}_{i}-\boldsymbol{\gamma
}_{i}\right)  \sigma_{i}\varepsilon_{it}\left(  \lambda_{T}\right)  \right] \\
&  =\sqrt{\frac{T}{n}}\boldsymbol{\varphi}_{nT}^{\prime}\sum_{j=1}^{n}\left[
n^{-1/2}\sum_{i=1}^{n}\left(  \boldsymbol{\hat{\gamma}}_{i}-\boldsymbol{\gamma
}_{i}\right)  \frac{1}{T}\sum_{t=1}^{T}\sigma_{i}\left(  \frac{1}{\left(
\boldsymbol{\varepsilon}_{j\circ}^{\prime}\mathbf{M}_{F}%
\boldsymbol{\varepsilon}_{j\circ}/T\right)  ^{1/2}}-1\right)  \varepsilon
_{it}\left(  \lambda_{T}\right)  \varepsilon_{jt}\left(  \lambda_{T}\right)
\right]  .
\end{align*}
Since $\varepsilon_{it}\left(  \lambda_{T}\right)  $ is distributed
independently over $t$ and weakly cross-sectionally dependent, then
\[
\frac{1}{T}\sum_{t=1}^{T}\sigma_{i}\left(  \frac{1}{\left(
\boldsymbol{\varepsilon}_{j\circ}^{\prime}\mathbf{M}_{F}%
\boldsymbol{\varepsilon}_{j\circ}/T\right)  ^{1/2}}-1\right)  \varepsilon
_{it}\left(  \lambda_{T}\right)  \varepsilon_{jt}\left(  \lambda_{T}\right)
\rightarrow_{p}0\text{, if }i\neq j\text{, }%
\]
and%
\begin{align*}
&  \frac{1}{T}\sum_{t=1}^{T}\sigma_{i}\left(  \frac{1}{\left(
\boldsymbol{\varepsilon}_{j\circ}^{\prime}\mathbf{M}_{F}%
\boldsymbol{\varepsilon}_{j\circ}/T\right)  ^{1/2}}-1\right)  \varepsilon
_{it}\left(  \lambda_{T}\right)  \varepsilon_{jt}\left(  \lambda_{T}\right) \\
&  \rightarrow_{p}\lim_{T\rightarrow\infty}\frac{1}{T}\sum_{t=1}^{T}\sigma
_{i}E\left\{  \left[  \left(  \frac{\boldsymbol{\varepsilon}_{i\circ}^{\prime
}\mathbf{M}_{F}\boldsymbol{\varepsilon}_{i\circ}}{T}\right)  ^{-1/2}-1\right]
\varepsilon_{it}^{2}\left(  \lambda_{T}\right)  \right\}  \text{, if }i=j.
\end{align*}
Also, by definition $\varepsilon_{it}\left(  \lambda_{T}\right)
=\varepsilon_{it}+\lambda_{T}b_{it}$, where $b_{it}=\mathbf{w}_{i0}^{\prime
}\boldsymbol{\varepsilon}_{\circ t}$, $w_{ii}=0$, and $\varepsilon_{it}$ is
independent from $b_{it}$, and therefore
\begin{align*}
&  E\left\{  \left[  \left(  \frac{\boldsymbol{\varepsilon}_{i\circ}^{\prime
}\mathbf{M}_{F}\boldsymbol{\varepsilon}_{i\circ}}{T}\right)  ^{-1/2}-1\right]
\varepsilon_{it}^{2}\left(  \lambda_{T}\right)  \right\} \\
&  =E\left\{  \left[  \left(  \frac{\boldsymbol{\varepsilon}_{i\circ}^{\prime
}\mathbf{M}_{F}\boldsymbol{\varepsilon}_{i\circ}}{T}\right)  ^{-1/2}-1\right]
\varepsilon_{it}^{2}\right\}  +E\left[  \left(  \frac{\boldsymbol{\varepsilon
}_{i\circ}^{\prime}\mathbf{M}_{F}\boldsymbol{\varepsilon}_{i\circ}}{T}\right)
^{-1/2}-1\right]  \lambda_{T}^{2}E\left(  b_{it}^{2}\right) \\
&  =O\left(  \frac{1}{T}\right)  +O\left(  \frac{1}{T^{2}}\right)  ,
\end{align*}
where the last line holds by (\ref{E:bit2}), (\ref{Es}) and (\ref{E7}).
Moreover, by (\ref{hgg}) $n^{-1/2}\sum_{i=1}^{n}\left(  \boldsymbol{\hat
{\gamma}}_{i}-\boldsymbol{\gamma}_{i}\right)  =O_{p}(\sqrt{\ln\left(
n\right)  /T})$. As $n$ and $T$ being of the same order, it then follows that
$C_{1,nT}\left(  \lambda_{T}\right)  =o_{p}(1)$. Similarly to $B_{2,nT}\left(
\lambda_{T}\right)  $, we have
\begin{align*}
C_{2,nT}\left(  \lambda_{T}\right)   &  =\boldsymbol{\varphi}_{nT}^{\prime
}\left[  n^{-1/2}\sum_{i=1}^{n}\left(  \boldsymbol{\hat{\gamma}}%
_{i}-\boldsymbol{\gamma}_{i}\right)  \boldsymbol{\gamma}_{i}^{\prime}\right]
\left[  \frac{1}{\sqrt{nT}}\sum_{j=1}^{n}\sum_{t=1}^{T}\left(  \frac
{1}{\left(  \boldsymbol{\varepsilon}_{j\circ}^{\prime}\mathbf{M}%
_{F}\boldsymbol{\varepsilon}_{j\circ}/T\right)  ^{1/2}}-1\right)
\varepsilon_{jt}\left(  \lambda_{T}\right)  \mathbf{f}_{t}\right] \\
&  =O_{p}\left(  \sqrt{\frac{\ln\left(  n\right)  }{T}}\right)  O_{p}%
(1)=o_{p}(1).
\end{align*}
The same line of reasoning as used for $B_{3,nT}\left(  \lambda_{T}\right)  $
and $B_{4,nT}\left(  \lambda_{T}\right)  $ can be used to establish
$C_{j,nT}\left(  \lambda_{T}\right)  =o_{p}(1)$ for $j=3$ and $4.$ Hence,
$T^{-1/2}\sum_{t=1}^{T}s_{t,nT}\left(  \lambda_{T}\right)  \upsilon
_{t,nT}\left(  \lambda_{T}\right)  =o_{p}(1)$, and overall we have
$q_{nT}\left(  \lambda_{T}\right)  =o_{p}(1)$, as required.
\end{proof}

\begin{lemma}
\label{PHIn}Consider the latent factor model given by (\ref{mod2}) and
(\ref{uit:p}). Suppose that Assumptions \ref{ass:factors}-\ref{ass:ws} hold
and $\left(  n,T\right)  \rightarrow\infty$, such that $n/T\rightarrow
\kappa\,,$ for $0<\kappa<\infty$. Then%
\begin{equation}
\sqrt{T}\left(  \boldsymbol{\varphi}_{n}-\boldsymbol{\varphi}_{nT}\right)
=O_{p}\left(  n^{-1/2}\right)  +O_{p}\left(  T^{-1/2}\right)  ,
\label{phinLemma}%
\end{equation}%
\begin{equation}
\sqrt{T}\left(  \boldsymbol{\hat{\varphi}}_{nT}-\boldsymbol{\varphi}%
_{n}\right)  =o_{p}(1), \label{phinhatLemma}%
\end{equation}
where $\boldsymbol{\varphi}_{n}=n^{-1}\sum_{i=1}^{n}\boldsymbol{\gamma}%
_{i}/\sigma_{i},$ $\boldsymbol{\varphi}_{nT}=n^{-1}\sum_{i=1}^{n}%
\boldsymbol{\gamma}_{i}/\omega_{i,T}$ with $\omega_{i,T}=\left(  T^{-1}%
\sigma_{i}^{2}\boldsymbol{\varepsilon}_{i\circ}^{\prime}\mathbf{M}%
_{F}\boldsymbol{\varepsilon}_{i\circ}\right)  ^{1/2}$, $\boldsymbol{\hat
{\varphi}}_{nT}=n^{-1}\sum_{i=1}^{n}\boldsymbol{\hat{\gamma}}_{i}/\hat{\sigma
}_{i,T}$, $\hat{\sigma}_{i,T}=\left(  T^{-1}\mathbf{y}_{i}^{\prime}%
\mathbf{M}_{\hat{F}}\mathbf{y}_{i}\right)  ^{1/2}$ and $\boldsymbol{\hat
{\gamma}}_{i}$ and $\mathbf{\hat{F}}$ are the principal component estimators
of $\boldsymbol{\gamma}_{i}$ and $\mathbf{F}$.
\end{lemma}

\begin{proof}
First note that
\begin{align*}
\sqrt{T}\left(  \boldsymbol{\varphi}_{n}-\boldsymbol{\varphi}_{nT}\right)   &
=\frac{\sqrt{T}}{n}\sum_{i=1}^{n}\frac{\boldsymbol{\gamma}_{i}}{\sigma_{i}%
}\left\{  \left(  1-\frac{\sigma_{i}}{\omega_{i,T}}\right)  -\left[
1-E\left(  \frac{\sigma_{i}}{\omega_{i,T}}\right)  \right]  \right\}
+\frac{\sqrt{T}}{n}\sum_{i=1}^{n}\frac{\boldsymbol{\gamma}_{i}}{\sigma_{i}%
}\left[  1-E\left(  \frac{\sigma_{i}}{\omega_{i,T}}\right)  \right]  ,\\
&  =\mathbf{g}_{1,nT}+\mathbf{g}_{2,nT},
\end{align*}
where%
\begin{align*}
\mathbf{g}_{1,nT}  &  =-\frac{1}{n}\sum_{i=1}^{n}\sqrt{T}\left[  \frac
{\sigma_{i}}{\omega_{i,T}}-E\left(  \frac{\sigma_{i}}{\omega_{i,T}}\right)
\right]  \frac{\boldsymbol{\gamma}_{i}}{\sigma_{i}},\\
\mathbf{g}_{2,nT}  &  =\frac{\sqrt{T}}{n}\sum_{i=1}^{n}\left[  1-E\left(
\frac{\sigma_{i}}{\omega_{i,T}}\right)  \right]  \frac{\boldsymbol{\gamma}%
_{i}}{\sigma_{i}}.
\end{align*}
Since $\sigma_{i}/\omega_{i,T}=\left(  T^{-1}\boldsymbol{\varepsilon}_{i\circ
}^{\prime}\mathbf{M}_{F}\boldsymbol{\varepsilon}_{i\circ}\right)  ^{-1/2}$,
$\left\Vert \boldsymbol{\gamma}_{i}\right\Vert <C$, $\sigma_{i}<C$, then using
result (\ref{Es}) we have $E\left(  \frac{\sigma_{i}}{\omega_{i,T}}\right)
=1+O\left(  T^{-1}\right)  $, and $\mathbf{g}_{2,nT}=$ $O\left(
T^{-1/2}\right)  $. The first term can be written as $\mathbf{g}_{1,nT}%
=n^{-1}\sum_{i=1}^{n}\boldsymbol{\gamma}_{i}\sigma_{i}^{-1}\chi_{i,T}$, where
$\chi_{i,T}=-\sqrt{T}\left[  \sigma_{i}/\omega_{i,T}-E\left(  \sigma
_{i}/\omega_{i,T}\right)  \right]  $. Conditional on $\mathbf{F}$ and
$\sigma_{i}$, $\chi_{i,T}$ are distributed independently over $i$ with mean
zero and bounded variances:\footnote{When $\varepsilon_{it}$ are normally
distributed we have the exact result $E\left(  \frac{T}%
{\boldsymbol{\varepsilon}_{i\circ}^{\prime}\mathbf{M}_{F}%
\boldsymbol{\varepsilon}_{i\circ}}\right)  =T/(T-m_{0}-2).$}
\[
Var\left(  \chi_{i,T}\right)  =T\left[  E\left(  \frac{T}%
{\boldsymbol{\varepsilon}_{i\circ}^{\prime}\mathbf{M}_{F}%
\boldsymbol{\varepsilon}_{i\circ}}\right)  -\left[  E\left(  \frac{\sigma_{i}%
}{\omega_{i,T}}\right)  \right]  ^{2}\right]  =T\left[  1+O\left(  \frac{1}%
{T}\right)  -\left[  1+O\left(  \frac{1}{T}\right)  \right]  ^{2}\right]
=O(1).
\]
Hence, $\mathbf{g}_{1,nT}=O_{p}\left(  n^{-1/2}\right)  $, and the desired
result (\ref{phinLemma}) follows. Consider now (\ref{phinhatLemma}) and note
that it can be decomposed as
\begin{equation}
\sqrt{T}\left(  \boldsymbol{\hat{\varphi}}_{nT}-\boldsymbol{\varphi}%
_{n}\right)  =\sqrt{T}\left(  \boldsymbol{\varphi}_{nT}-\boldsymbol{\varphi
}_{n}\right)  +\sqrt{T}\left(  \boldsymbol{\hat{\varphi}}_{nT}%
-\boldsymbol{\varphi}_{nT}\right)  , \label{Tphi}%
\end{equation}
where it is already established that the first term is $o_{p}(1)$. Consider
now the second term of (\ref{Tphi}) and note that it can be written as
\begin{align*}
\sqrt{T}\left(  \boldsymbol{\hat{\varphi}}_{nT}-\boldsymbol{\varphi}%
_{nT}\right)   &  =\frac{\sqrt{T}}{n}\sum_{i=1}^{n}\left(  \frac
{\boldsymbol{\hat{\gamma}}_{i}}{\hat{\sigma}_{i,T}}\boldsymbol{-}%
\frac{\boldsymbol{\gamma}_{i}}{\omega_{i,T}}\right) \\
&  =\frac{\sqrt{T}}{n}\sum_{i=1}^{n}\boldsymbol{\gamma}_{i}\left(  \frac
{1}{\hat{\sigma}_{i,T}}\boldsymbol{-}\frac{1}{\omega_{i,T}}\right)
+\frac{\sqrt{T}}{n}\sum_{i=1}^{n}\frac{1}{\sigma_{i}}\left(  \boldsymbol{\hat
{\gamma}}_{i}-\boldsymbol{\gamma}_{i}\right)  +\\
&  \frac{\sqrt{T}}{n}\sum_{i=1}^{n}\left(  \frac{1}{\omega_{i,T}}-\frac
{1}{\sigma_{i}}\right)  \left(  \boldsymbol{\hat{\gamma}}_{i}%
-\boldsymbol{\gamma}_{i}\right)  +\frac{\sqrt{T}}{n}\sum_{i=1}^{n}\left(
\frac{1}{\hat{\sigma}_{i,T}}-\frac{1}{\omega_{i,T}}\right)  \left(
\boldsymbol{\hat{\gamma}}_{i}-\boldsymbol{\gamma}_{i}\right)  .
\end{align*}
Now using (\ref{Ewin}) of Lemma \ref{lm4}\ we have%
\[
\frac{\sqrt{T}}{n}\sum_{i=1}^{n}\boldsymbol{\gamma}_{i}\left(  \frac{1}%
{\hat{\sigma}_{i,T}}-\frac{1}{\omega_{i,T}}\right)  =O_{p}\left(  \frac
{\ln\left(  n\right)  }{\sqrt{T}}\right)  .
\]
Using this result as well as (\ref{hgg}), (\ref{hgg_a}) and (\ref{hgg_d}), we
have $\sqrt{T}\left(  \boldsymbol{\hat{\varphi}}_{nT}-\boldsymbol{\varphi
}_{nT}\right)  =o_{p}(1)$ as required.
\end{proof}

\begin{lemma}
\label{lmQ} Suppose $\mathbf{M}_{F}=\mathbf{I}_{T}-\mathbf{F(F}^{\prime
}\mathbf{F)}^{-1}\mathbf{F}^{\prime}$, where $\mathbf{F}$ is a $T\times m_{0}$
matrix, and $\boldsymbol{\tau}_{T}$ is a $T\times1$ vector of ones. Then%
\begin{align}
\mathrm{tr}\left(  \mathbf{M}_{F}\right)   &  =v,\text{ }\mathrm{tr}\left(
\mathbf{M}_{F}\mathbf{\odot M}_{F}\right)  =O\left(  v\right)  ,\label{lmQ1}\\
\mathrm{tr}\left(  \mathbf{M}_{F}\mathbf{\odot M}_{F}\mathbf{\odot M}%
_{F}\right)   &  =O\left(  v\right)  ,\mathrm{tr}\left(  \mathbf{M}%
_{F}\mathbf{\odot M}_{F}\mathbf{\odot M}_{F}\mathbf{\odot M}_{F}\right)
=O\left(  v\right)  ,\label{lmQ2}\\
\mathrm{tr}\left[  \left(  \mathbf{I}_{T}\mathbf{\odot M}_{F}\right)
\mathbf{M}_{F}\right]   &  =O\left(  v\right)  ,\mathrm{tr}\left[  \left(
\mathbf{M}_{F}\mathbf{\odot M}_{F}\right)  \mathbf{M}_{F}\right]  =O\left(
v\right)  , \label{lmQ3}%
\end{align}%
\begin{align}
\boldsymbol{\tau}_{T}^{\prime}\left(  \mathbf{M}_{F}\mathbf{\odot M}%
_{F}\right)  \boldsymbol{\tau}_{T}  &  =O\left(  v\right)  ,\boldsymbol{\tau
}_{T}^{\prime}\left(  \mathbf{M}_{F}\mathbf{\odot M}_{F}\mathbf{\odot M}%
_{F}\right)  \boldsymbol{\tau}_{T}=O\left(  v\right)  ,\label{lmQ4}\\
\boldsymbol{\tau}_{T}^{\prime}\left(  \mathbf{M}_{F}\mathbf{\odot M}%
_{F}\mathbf{\odot M}_{F}\mathbf{\odot M}_{F}\right)  \boldsymbol{\tau}_{T}  &
=O\left(  v\right)  ,\boldsymbol{\tau}_{T}^{\prime}\left(  \mathbf{I}%
_{T}\mathbf{\odot M}_{F}\right)  \left(  \mathbf{M}_{F}\mathbf{\odot M}%
_{F}\mathbf{\odot M}_{F}\right)  \boldsymbol{\tau}_{T}=O\left(  v\right)
,\label{lmQ5}\\
\boldsymbol{\tau}_{T}^{\prime}\left(  \mathbf{I}_{T}\mathbf{\odot M}%
_{F}\right)  \left(  \mathbf{M}_{F}\mathbf{\odot M}_{F}\right)  \left(
\mathbf{I}_{T}\mathbf{\odot M}_{F}\right)  \boldsymbol{\tau}_{T}  &  =O\left(
v\right)  ,\boldsymbol{\tau}_{T}^{\prime}\left(  \mathbf{I}_{T}\mathbf{\odot
M}_{F}\right)  \left(  \mathbf{I}_{T}\mathbf{\odot M}_{F}\right)
\boldsymbol{\tau}_{T}=O\left(  v\right)  , \label{lmQ6}%
\end{align}%
\begin{align}
\boldsymbol{\tau}_{T}^{\prime}\left(  \mathbf{I}_{T}\mathbf{\odot M}%
_{F}\right)  \mathbf{M}_{F}\left(  \mathbf{I}_{T}\mathbf{\odot M}_{F}\right)
\boldsymbol{\tau}_{T}  &  =O\left(  v^{3/2}\right)  ,\boldsymbol{\tau}%
_{T}^{\prime}\left(  \mathbf{M}_{F}\mathbf{\odot M}_{F}\right)  \mathbf{M}%
_{F}\left(  \mathbf{I}_{T}\mathbf{\odot M}_{F}\right)  \boldsymbol{\tau}%
_{T}=O\left(  v^{3/2}\right)  ,\label{lmQ7}\\
\boldsymbol{\tau}_{T}^{\prime}\left(  \mathbf{I}_{T}\mathbf{\odot M}%
_{F}\right)  \mathbf{M}_{F}\left(  \mathbf{I}_{T}\mathbf{\odot M}%
_{F}\mathbf{\odot M}_{F}\right)  \boldsymbol{\tau}_{T}  &  =O\left(
v^{3/2}\right)  , \label{lmQ8}%
\end{align}
where $v=T-m_{0}$.
\end{lemma}

\begin{proof}
See Lemma 10 of Pesaran and Yamagata (2024).
\end{proof}

\begin{lemma}
\label{Qmom}Suppose the $T\times1$ vector $\boldsymbol{\varepsilon}%
=\mathbf{(}\varepsilon_{1},\varepsilon_{2},...,\varepsilon_{T})^{\prime}$ is
$\boldsymbol{\varepsilon}\thicksim IID(\mathbf{0},\mathbf{I}_{T})$, $\sup
_{t}E\left(  \left\vert \varepsilon_{t}\right\vert ^{8+s}\right)  <C$ for some
small $s>0$, and $\mathbf{M}_{F}=\mathbf{I}_{T}-\mathbf{F(F}^{\prime
}\mathbf{F)}^{-1}\mathbf{F}^{\prime}$, where the $T\times m_{0}$ matrix
$\mathbf{F}$ is distributed independently of $\boldsymbol{\varepsilon}$,
then\textbf{ }%
\begin{equation}
E\left(  \frac{\boldsymbol{\varepsilon}^{\prime}\mathbf{M}_{F}%
\boldsymbol{\varepsilon}}{v}\right)  =1,\text{ } \label{mom1}%
\end{equation}%
\begin{equation}
E\left[  \left(  \frac{\boldsymbol{\varepsilon}^{\prime}\mathbf{M}%
_{F}\boldsymbol{\varepsilon}}{v}\right)  ^{2}\right]  =1+O\left(  \frac{1}%
{v}\right)  , \label{mom2}%
\end{equation}%
\begin{equation}
E\left(  q_{v}\right)  =0\mathbf{,}\text{ }E\left(  q_{v}^{4}\right)
=O(1)\mathbf{,} \label{qv}%
\end{equation}
where $v=T-m_{0}$, and $q_{v}=\sqrt{v}\left(  \frac{\boldsymbol{\varepsilon
}^{\prime}\mathbf{M}_{F}\boldsymbol{\varepsilon}}{v}-1\right)  $.
\end{lemma}

\begin{proof}
Denote $\boldsymbol{\tau}_{T}$ as a $T\times1$ vector of ones and suppose
$\gamma_{1}=E\left(  \varepsilon_{t}^{3}\right)  $, $\gamma_{2}=E\left(
\varepsilon_{t}^{4}\right)  -3$, $\gamma_{3}=E\left(  \varepsilon_{t}%
^{5}\right)  -10\gamma_{1}$, $\gamma_{4}=E\left(  \varepsilon_{t}^{6}\right)
-15\gamma_{2}-10\gamma_{1}^{2}-15$, $\gamma_{6}=E\left(  \varepsilon_{t}%
^{8}\right)  -28\gamma_{4}-56\gamma_{3}\gamma_{1}-35\gamma_{2}^{2}%
-210\gamma_{2}-280\gamma_{1}^{2}-105$, which are all bounded as it is assumed
$\sup_{t}E\left(  \left\vert \varepsilon_{t}\right\vert ^{8+\epsilon}\right)
<C$. Since $\boldsymbol{\varepsilon}\thicksim IID(\mathbf{0},\mathbf{I}_{T})$
and $\mathbf{M}_{F}=\left(  m_{tt^{\prime}}\right)  $ is an idempotent matrix
then results (S.6) to (S.9) of Lemma 6 in Pesaran and Yamagata (2024) apply
and we have (since $\mathrm{tr}\left(  \mathbf{M}_{F}^{s}\right)
=\mathrm{tr}\left(  \mathbf{M}_{F}\right)  =v$, for $s=1,2,\ldots$)
\begin{equation}
E\left(  \boldsymbol{\varepsilon}^{\prime}\mathbf{M}_{F}%
\boldsymbol{\varepsilon}\right)  =v, \label{mom1_a}%
\end{equation}%
\begin{equation}
E\left[  \left(  \boldsymbol{\varepsilon}^{\prime}\mathbf{M}_{F}%
\boldsymbol{\varepsilon}\right)  ^{2}\right]  =v^{2}+2v+\gamma_{2}%
\mathrm{tr}\left(  \mathbf{M}_{F}\mathbf{\odot M}_{F}\right)  , \label{mom2_a}%
\end{equation}%
\begin{align}
E\left[  \left(  \boldsymbol{\varepsilon}^{\prime}\mathbf{M}_{F}%
\boldsymbol{\varepsilon}\right)  ^{3}\right]   &  =v^{3}+6v^{2}+8v+\gamma
_{4}\mathrm{tr}\left(  \mathbf{M}_{F}\mathbf{\odot M}_{F}\mathbf{\odot M}%
_{F}\right)  +3\gamma_{2}\left(  v+4\right)  \mathrm{tr}\left(  \mathbf{M}%
_{F}\mathbf{\odot M}_{F}\right) \nonumber\\
&  +6\gamma_{1}^{2}\left[  \boldsymbol{\tau}_{T}^{\prime}\left(
\mathbf{I}_{T}\mathbf{\odot M}_{F}\right)  \mathbf{M}_{F}\left(
\mathbf{I}_{T}\mathbf{\odot M}_{F}\right)  \boldsymbol{\tau}_{T}\right]
+4\gamma_{1}^{2}\left[  \boldsymbol{\tau}_{T}^{\prime}\left(  \mathbf{M}%
_{F}\mathbf{\odot M}_{F}\mathbf{\odot M}_{F}\right)  \boldsymbol{\tau}%
_{T}\right]  \label{mom3_a}%
\end{align}%
\begin{equation}
E\left[  \left(  \boldsymbol{\varepsilon}^{\prime}\mathbf{M}_{F}%
\boldsymbol{\varepsilon}\right)  ^{4}\right]  =v^{4}+12v^{3}+44v^{2}%
+48v+\gamma_{2}f_{\gamma_{2}}+\gamma_{4}f_{\gamma_{4}}+\gamma_{6}f_{\gamma
_{6}}+\gamma_{1}^{2}f_{\gamma_{1}^{2}}+\gamma_{2}^{2}f_{\gamma_{2}^{2}}%
+\gamma_{1}\gamma_{3}f_{\gamma_{1}\gamma_{3}}, \label{mom4_a}%
\end{equation}
where%
\begin{align*}
f_{\gamma_{2}}  &  =\left(  6v^{2}+48v\right)  \mathrm{tr}\left(
\mathbf{M}_{F}\mathbf{\odot M}_{F}\right)  +12\left[  \boldsymbol{\tau}%
_{T}^{\prime}\left(  \mathbf{M}_{F}\mathbf{\odot M}_{F}\right)
\boldsymbol{\tau}_{T}\mathrm{tr}\left(  \mathbf{M}_{F}\mathbf{\odot M}%
_{F}\right)  \right] \\
&  +96\mathrm{tr}\left[  \left(  \mathbf{I}_{T}\mathbf{\odot M}_{F}\right)
\mathbf{M}_{F}\right]  +48\boldsymbol{\tau}_{T}^{\prime}\left(  \mathbf{I}%
_{T}\mathbf{\odot M}_{F}\right)  \left(  \mathbf{I}_{T}\mathbf{\odot M}%
_{F}\right)  \boldsymbol{\tau}_{T},
\end{align*}%
\[
f_{\gamma_{4}}=\left(  4v+24\right)  \mathrm{tr}\left(  \mathbf{M}%
_{F}\mathbf{\odot M}_{F}\mathbf{\odot M}_{F}\right)  ,
\]%
\[
f_{\gamma_{6}}=\mathrm{tr}\left(  \mathbf{M}_{F}\mathbf{\odot M}%
_{F}\mathbf{\odot M}_{F}\mathbf{\odot M}_{F}\right)  ,
\]%
\begin{align*}
f_{\gamma_{1}^{2}}  &  =24v\boldsymbol{\tau}_{T}^{\prime}\left(
\mathbf{I}_{T}\mathbf{\odot M}_{F}\right)  \mathbf{M}_{F}\left(
\mathbf{I}_{T}\mathbf{\odot M}_{F}\right)  \boldsymbol{\tau}_{T}%
+48\boldsymbol{\tau}_{T}^{\prime}\left(  \mathbf{I}_{T}\mathbf{\odot M}%
_{F}\right)  \mathbf{M}_{F}\left(  \mathbf{I}_{T}\mathbf{\odot M}_{F}\right)
\boldsymbol{\tau}_{T}\\
&  +16v\boldsymbol{\tau}_{T}^{\prime}\left(  \mathbf{M}_{F}\mathbf{\odot
M}_{F}\mathbf{\odot M}_{F}\right)  \boldsymbol{\tau}_{T}+96\boldsymbol{\tau
}_{T}^{\prime}\left(  \mathbf{M}_{F}\mathbf{\odot M}_{F}\right)
\mathbf{M}_{F}\left(  \mathbf{I}_{T}\mathbf{\odot M}_{F}\right)
\boldsymbol{\tau}_{T}\\
&  +96\mathrm{tr}\left[  \mathbf{M}_{F}\left(  \mathbf{M}_{F}\mathbf{\odot
M}_{F}\right)  \mathbf{M}_{F}\right]  ,
\end{align*}%
\begin{align*}
f_{\gamma_{2}^{2}}  &  =3\left[  \mathrm{tr}\left(  \mathbf{M}_{F}%
\mathbf{\odot M}_{F}\right)  \right]  ^{2}+24\boldsymbol{\tau}_{T}^{\prime
}\left(  \mathbf{I}_{T}\mathbf{\odot M}_{F}\right)  \left(  \mathbf{M}%
_{F}\mathbf{\odot M}_{F}\right)  \left(  \mathbf{I}_{T}\mathbf{\odot M}%
_{F}\right)  \boldsymbol{\tau}_{T}\\
&  +8\boldsymbol{\tau}_{T}^{\prime}\left(  \mathbf{M}_{F}\mathbf{\odot M}%
_{F}\mathbf{\odot M}_{F}\mathbf{\odot M}_{F}\right)  \boldsymbol{\tau}_{T},
\end{align*}%
\[
f_{\gamma_{1}\gamma_{3}}=24\boldsymbol{\tau}_{T}^{\prime}\left(
\mathbf{I}_{T}\mathbf{\odot M}_{F}\right)  \mathbf{M}_{F}\left(
\mathbf{I}_{T}\mathbf{\odot M}_{F}\mathbf{\odot M}_{F}\right)
\boldsymbol{\tau}_{T}+32\boldsymbol{\tau}_{T}^{\prime}\left(  \mathbf{I}%
_{T}\mathbf{\odot M}_{F}\right)  \left(  \mathbf{M}_{F}\mathbf{\odot M}%
_{F}\mathbf{\odot M}_{F}\right)  \boldsymbol{\tau}_{T}.
\]
Result (\ref{mom1}) follows from (\ref{mom1_a}). To establish (\ref{mom2}),
using (\ref{mom2_a}), we first note that%
\[
E\left(  \frac{\boldsymbol{\varepsilon}^{\prime}\mathbf{M}_{F}%
\boldsymbol{\varepsilon}}{v}\right)  ^{2}=1+\frac{2}{v}+\gamma_{2}%
\frac{\mathrm{tr}\left(  \mathbf{M}_{F}\mathbf{\odot M}_{F}\right)  }{v^{2}}.
\]
But by (\ref{lmQ1}) $\mathrm{tr}\left(  \mathbf{M}_{F}\mathbf{\odot M}%
_{F}\right)  =O\left(  v\right)  $ and by assumption $\gamma_{2}$ is bounded.
Hence $E\left(  \frac{\boldsymbol{\varepsilon}^{\prime}\mathbf{M}%
_{F}\boldsymbol{\varepsilon}}{v}\right)  ^{2}=1+O(\frac{1}{v})$, as required.
To prove (\ref{qv}), noting that $E\left(  \boldsymbol{\varepsilon}^{\prime
}\mathbf{M}_{F}\boldsymbol{\varepsilon}\right)  =v$ then $E(q_{v})=0$. Also
\[
q_{v}^{4}=v^{2}\left(  \frac{\boldsymbol{\varepsilon}^{\prime}\mathbf{M}%
_{F}\boldsymbol{\varepsilon}}{v}-1\right)  ^{4}=v^{2}\left[  \left(
\frac{\boldsymbol{\varepsilon}^{\prime}\mathbf{M}_{F}\boldsymbol{\varepsilon}%
}{v}\right)  ^{4}-4\left(  \frac{\boldsymbol{\varepsilon}^{\prime}%
\mathbf{M}_{F}\boldsymbol{\varepsilon}}{v}\right)  ^{3}+6\left(
\frac{\boldsymbol{\varepsilon}^{\prime}\mathbf{M}_{F}\boldsymbol{\varepsilon}%
}{v}\right)  ^{2}-4\left(  \frac{\boldsymbol{\varepsilon}^{\prime}%
\mathbf{M}_{F}\boldsymbol{\varepsilon}}{v}\right)  +1\right]  ,
\]
and taking expectation yields
\begin{align*}
E\left(  q_{v}^{4}\right)   &  =v^{2}\left[  E\left[  \left(  \frac
{\boldsymbol{\varepsilon}^{\prime}\mathbf{M}_{F}\boldsymbol{\varepsilon}}%
{v}\right)  ^{4}\right]  -4E\left[  \left(  \frac{\boldsymbol{\varepsilon
}^{\prime}\mathbf{M}_{F}\boldsymbol{\varepsilon}}{v}\right)  ^{3}\right]
+6E\left[  \left(  \frac{\boldsymbol{\varepsilon}^{\prime}\mathbf{M}%
_{F}\boldsymbol{\varepsilon}}{v}\right)  ^{2}\right]  -4E\left(
\frac{\boldsymbol{\varepsilon}^{\prime}\mathbf{M}_{F}\boldsymbol{\varepsilon}%
}{v}\right)  +1\right] \\
&  =\frac{1}{v^{2}}E\left[  \left(  \boldsymbol{\varepsilon}^{\prime
}\mathbf{M}_{F}\boldsymbol{\varepsilon}\right)  ^{4}\right]  -\frac{4}%
{v}E\left[  \left(  \boldsymbol{\varepsilon}^{\prime}\mathbf{M}_{F}%
\boldsymbol{\varepsilon}\right)  ^{3}\right]  +6E\left[  \left(
\boldsymbol{\varepsilon}^{\prime}\mathbf{M}_{F}\boldsymbol{\varepsilon
}\right)  ^{2}\right]  -4vE\left(  \boldsymbol{\varepsilon}^{\prime}%
\mathbf{M}_{F}\boldsymbol{\varepsilon}\right)  +v^{2}.
\end{align*}
Now using the results in (\ref{mom1_a})-(\ref{mom4_a}), and after some
algebra, we obtain%
\begin{align}
E\left(  q_{v}^{4}\right)   &  =\frac{1}{v^{2}}\left[
\begin{array}
[c]{c}%
12v^{2}+48v+12\gamma_{2}\left[  \boldsymbol{\tau}_{T}^{\prime}\left(
\mathbf{M}_{F}\mathbf{\odot M}_{F}\right)  \boldsymbol{\tau}_{T}\right]
\mathrm{tr}\left(  \mathbf{M}_{F}\mathbf{\odot M}_{F}\right) \\
+96\gamma_{2}\mathrm{tr}\left[  \left(  \mathbf{I}_{T}\mathbf{\odot M}%
_{F}\right)  \mathbf{M}_{F}\right]  +48\gamma_{2}\left[  \boldsymbol{\tau}%
_{T}^{\prime}\left(  \mathbf{I}_{T}\mathbf{\odot M}_{F}\right)  \left(
\mathbf{I}_{T}\mathbf{\odot M}_{F}\right)  \boldsymbol{\tau}_{T}\right] \\
+\left(  4\gamma_{4}v+24\gamma_{4}\right)  \mathrm{tr}\left(  \mathbf{M}%
_{F}\mathbf{\odot M}_{F}\mathbf{\odot M}_{F}\right)  +\gamma_{6}%
\mathrm{tr}\left(  \mathbf{M}_{F}\mathbf{\odot M}_{F}\mathbf{\odot M}%
_{F}\mathbf{\odot M}_{F}\right) \\
+48\gamma_{1}^{2}\left[  \boldsymbol{\tau}_{T}^{\prime}\left(  \mathbf{I}%
_{T}\mathbf{\odot M}_{F}\right)  \mathbf{M}_{F}\left(  \mathbf{I}%
_{T}\mathbf{\odot M}_{F}\right)  \boldsymbol{\tau}_{T}\right]  +96\gamma
_{1}^{2}\left[  \boldsymbol{\tau}_{T}^{\prime}\left(  \mathbf{M}%
_{F}\mathbf{\odot M}_{F}\right)  \mathbf{M}_{F}\left(  \mathbf{I}%
_{T}\mathbf{\odot M}_{F}\right)  \boldsymbol{\tau}_{T}\right] \\
+96\gamma_{1}^{2}\mathrm{tr}\left[  \left(  \mathbf{M}_{F}\mathbf{\odot M}%
_{F}\right)  \mathbf{M}_{F}\right]  +3\gamma_{2}^{2}\left[  \mathrm{tr}\left(
\mathbf{M}_{F}\mathbf{\odot M}_{F}\right)  \right]  ^{2}\\
+24\gamma_{2}^{2}\left[  \boldsymbol{\tau}_{T}^{\prime}\left(  \mathbf{I}%
_{T}\mathbf{\odot M}_{F}\right)  \left(  \mathbf{M}_{F}\mathbf{\odot M}%
_{F}\right)  \left(  \mathbf{I}_{T}\mathbf{\odot M}_{F}\right)
\boldsymbol{\tau}_{T}\right]  +8\gamma_{2}^{2}\left[  \boldsymbol{\tau}%
_{T}^{\prime}\left(  \mathbf{M}_{F}\mathbf{\odot M}_{F}\mathbf{\odot M}%
_{F}\mathbf{\odot M}_{F}\right)  \boldsymbol{\tau}_{T}\right] \\
+24\gamma_{1}\gamma_{3}\left[  \boldsymbol{\tau}_{T}^{\prime}\left(
\mathbf{I}_{T}\mathbf{\odot M}_{F}\right)  \mathbf{M}_{F}\left(
\mathbf{I}_{T}\mathbf{\odot M}_{F}\mathbf{\odot M}_{F}\right)
\boldsymbol{\tau}_{T}\right] \\
+32\gamma_{1}\gamma_{3}\left[  \boldsymbol{\tau}_{T}^{\prime}\left(
\mathbf{I}_{T}\mathbf{\odot M}_{F}\right)  \left(  \mathbf{M}_{F}\mathbf{\odot
M}_{F}\mathbf{\odot M}_{F}\right)  \boldsymbol{\tau}_{T}\right]
\end{array}
\right] \nonumber\\
&  =\sum_{s=1}^{15}a_{s,v}. \label{E qv4}%
\end{align}
Further noting that $\gamma_{1},\gamma_{2}$,$\gamma_{3}$,$\gamma_{4}$%
,$\gamma_{5}$, and $\gamma_{6}$ are all bounded, then using the results
(\ref{lmQ1})-(\ref{lmQ8}) we have%
\[
a_{1,v}=12,a_{2,v}=\frac{48}{v},
\]%
\[
a_{3,v}=\frac{12\gamma_{2}\left[  \boldsymbol{\tau}_{T}^{\prime}\left(
\mathbf{M}_{F}\mathbf{\odot M}_{F}\right)  \boldsymbol{\tau}_{T}\right]
\mathrm{tr}\left(  \mathbf{M}_{F}\mathbf{\odot M}_{F}\right)  }{v^{2}%
}=O\left(  1\right)  ,
\]%
\[
a_{4,v}=\frac{96\gamma_{2}\mathrm{tr}\left[  \left(  \mathbf{I}_{T}%
\mathbf{\odot M}_{F}\right)  \mathbf{M}_{F}\right]  }{v^{2}}=O\left(  \frac
{1}{v}\right)  ,
\]%
\[
a_{5,v}=\frac{48\gamma_{2}\left[  \boldsymbol{\tau}_{T}^{\prime}\left(
\mathbf{I}_{T}\mathbf{\odot M}_{F}\right)  \left(  \mathbf{I}_{T}\mathbf{\odot
M}_{F}\right)  \boldsymbol{\tau}_{T}\right]  }{v^{2}}=O\left(  \frac{1}%
{v}\right)  ,
\]%
\[
a_{6,v}=\frac{\left(  4\gamma_{4}v+24\gamma_{4}\right)  \mathrm{tr}\left(
\mathbf{M}_{F}\mathbf{\odot M}_{F}\mathbf{\odot M}_{F}\right)  }{v^{2}%
}=O\left(  1\right)  ,
\]%
\[
a_{7,v}=\frac{\gamma_{6}\mathrm{tr}\left(  \mathbf{M}_{F}\mathbf{\odot M}%
_{F}\mathbf{\odot M}_{F}\mathbf{\odot M}_{F}\right)  }{v^{2}}=O\left(
\frac{1}{v}\right)  ,
\]%
\[
a_{8,v}=\frac{48\gamma_{1}^{2}\left[  \boldsymbol{\tau}_{T}^{\prime}\left(
\mathbf{I}_{T}\mathbf{\odot M}_{F}\right)  \mathbf{M}_{F}\left(
\mathbf{I}_{T}\mathbf{\odot M}_{F}\right)  \boldsymbol{\tau}_{T}\right]
}{v^{2}}=O\left(  \frac{1}{\sqrt{v}}\right)  ,
\]%
\[
a_{9,v}=\frac{96\gamma_{1}^{2}\boldsymbol{\tau}_{T}^{\prime}\left(
\mathbf{M}_{F}\mathbf{\odot M}_{F}\right)  \mathbf{M}_{F}\left(
\mathbf{I}_{T}\mathbf{\odot M}_{F}\right)  \boldsymbol{\tau}_{T}}{v^{2}%
}=O\left(  \frac{1}{\sqrt{v}}\right)  ,
\]%
\[
a_{10,v}=\frac{96\gamma_{1}^{2}\mathrm{tr}\left[  \left(  \mathbf{M}%
_{F}\mathbf{\odot M}_{F}\right)  \mathbf{M}_{F}\right]  }{v^{2}}=O\left(
\frac{1}{v}\right)  ,
\]%
\[
a_{11,v}=\frac{3\gamma_{2}^{2}\left[  \mathrm{tr}\left(  \mathbf{M}%
_{F}\mathbf{\odot M}_{F}\right)  \right]  ^{2}}{v^{2}}=O\left(  1\right)  ,
\]%
\[
a_{12,v}=\frac{24\gamma_{2}^{2}\left[  \boldsymbol{\tau}_{T}^{\prime}\left(
\mathbf{I}_{T}\mathbf{\odot M}_{F}\right)  \left(  \mathbf{M}_{F}\mathbf{\odot
M}_{F}\right)  \left(  \mathbf{I}_{T}\mathbf{\odot M}_{F}\right)
\boldsymbol{\tau}_{T}\right]  }{v^{2}}=O\left(  \frac{1}{v}\right)  ,
\]%
\[
a_{13,v}=\frac{8\gamma_{2}^{2}\left[  \boldsymbol{\tau}_{T}^{\prime}\left(
\mathbf{M}_{F}\mathbf{\odot M}_{F}\mathbf{\odot M}_{F}\mathbf{\odot M}%
_{F}\right)  \boldsymbol{\tau}_{T}\right]  }{v^{2}}=O\left(  \frac{1}%
{v}\right)  ,
\]%
\[
a_{14,v}=\frac{24\gamma_{1}\gamma_{3}\left[  \boldsymbol{\tau}_{T}^{\prime
}\left(  \mathbf{I}_{T}\mathbf{\odot M}_{F}\right)  \mathbf{M}_{F}\left(
\mathbf{I}_{T}\mathbf{\odot M}_{F}\mathbf{\odot M}_{F}\right)
\boldsymbol{\tau}_{T}\right]  }{v^{2}}=O\left(  \frac{1}{\sqrt{v}}\right)  ,
\]%
\[
a_{15,v}=\frac{32\gamma_{1}\gamma_{3}\left[  \boldsymbol{\tau}_{T}^{\prime
}\left(  \mathbf{I}_{T}\mathbf{\odot M}_{F}\right)  \left(  \mathbf{M}%
_{F}\mathbf{\odot M}_{F}\mathbf{\odot M}_{F}\right)  \boldsymbol{\tau}%
_{T}\right]  }{v^{2}}=O\left(  \frac{1}{v}\right)  .
\]
Using these results in (\ref{E qv4}) it now follows that $E\left(  q_{v}%
^{4}\right)  =O\left(  1\right)  $, as required.
\end{proof}

\begin{lemma}
\label{Emom} Suppose the $T\times1$ vector $\boldsymbol{\varepsilon
}=\mathbf{(}\varepsilon_{1},\varepsilon_{2},...,\varepsilon_{T})^{\prime}$ is
$\boldsymbol{\varepsilon}\thicksim IID(\mathbf{0},\mathbf{I}_{T})$, $\sup
_{t}E\left(  \left\vert \varepsilon_{t}\right\vert ^{8+s}\right)  <C$ for some
small $s>0$, and $\mathbf{M}_{F}=\mathbf{I}_{T}-\mathbf{F(F}^{\prime
}\mathbf{F)}^{-1}\mathbf{F}^{\prime}$, where the $T\times m_{0}$ matrix
$\mathbf{F}$ is distributed independently of $\boldsymbol{\varepsilon}$.
Suppose there exists a finite integer $v_{0}$ such that for all $v>v_{0}$,
\begin{equation}
\frac{\boldsymbol{\varepsilon}^{\prime}\mathbf{M}_{F}\boldsymbol{\varepsilon}%
}{v}>c>0. \label{Lbound}%
\end{equation}
Then for $v>v_{0}$,
\begin{equation}
E\left[  \left(  \frac{\boldsymbol{\varepsilon}^{\prime}\mathbf{M}%
_{F}\boldsymbol{\varepsilon}}{v}\right)  ^{1/2}\right]  =1+O\left(  \frac
{1}{v}\right)  ,\text{ }E\left[  \left(  \frac{\boldsymbol{\varepsilon
}^{\prime}\mathbf{M}_{F}\boldsymbol{\varepsilon}}{v}\right)  ^{-s/2}\right]
=1+O\left(  \frac{1}{v}\right)  ,\text{ for }s=1,2,3,4, \label{Es}%
\end{equation}%
\begin{equation}
E\left[  \varepsilon_{t}^{2}\left(  \frac{\boldsymbol{\varepsilon}^{\prime
}\mathbf{M}_{F}\boldsymbol{\varepsilon}}{v}\right)  \right]  =1+O\left(
\frac{1}{v}\right)  , \label{E5}%
\end{equation}%
\begin{equation}
E\left[  \varepsilon_{t}^{2}\left(  \frac{\boldsymbol{\varepsilon}^{\prime
}\mathbf{M}_{F}\boldsymbol{\varepsilon}}{v}\right)  ^{-1}\right]  =1+O\left(
\frac{1}{v}\right)  , \label{E6}%
\end{equation}%
\begin{equation}
E\left[  \varepsilon_{t}^{2}\left(  \frac{\boldsymbol{\varepsilon}^{\prime
}\mathbf{M}_{F}\boldsymbol{\varepsilon}}{v}\right)  ^{-1/2}\right]
=1+O\left(  \frac{1}{v}\right)  , \label{E7}%
\end{equation}%
\begin{equation}
E\left[  \varepsilon_{t}\varepsilon_{t^{\prime}}\left(  \frac
{\boldsymbol{\varepsilon}^{\prime}\mathbf{M}_{F}\boldsymbol{\varepsilon}}%
{v}\right)  ^{-1/2}\right]  =O\left(  \frac{1}{v}\right)  ,\text{ for }t\neq
t^{\prime}, \label{E8}%
\end{equation}%
\begin{equation}
E\left[  \varepsilon_{t}\left(  \frac{\boldsymbol{\varepsilon}^{\prime
}\mathbf{M}_{F}\boldsymbol{\varepsilon}}{v}\right)  ^{-1/2}\right]  =O\left(
\frac{1}{v}\right)  . \label{E9}%
\end{equation}

\end{lemma}

\begin{proof}
To establish the results in (\ref{Es}) we first note that%
\begin{equation}
\frac{\boldsymbol{\varepsilon}^{\prime}\mathbf{M}_{F}\boldsymbol{\varepsilon}%
}{v}=1+\frac{1}{\sqrt{v}}q_{v} \label{eMeT^-1}%
\end{equation}
where
\begin{equation}
q_{v}=\sqrt{v}\left(  \frac{\boldsymbol{\varepsilon}^{\prime}\mathbf{M}%
_{F}\boldsymbol{\varepsilon}}{v}-1\right)  \text{.} \label{eMeT^-1/2}%
\end{equation}
Applying the Taylor Theorem to $\left(  v^{-1}\boldsymbol{\varepsilon}%
^{\prime}\mathbf{M}_{F}\boldsymbol{\varepsilon}\right)  ^{1/2}$ we have
\begin{equation}
\left(  \frac{\boldsymbol{\varepsilon}^{\prime}\mathbf{M}_{F}%
\boldsymbol{\varepsilon}}{v}\right)  ^{1/2}=1+\frac{1}{2}\frac{q_{v}}{\sqrt
{v}}-\frac{1}{8v}R_{v}, \label{shalf}%
\end{equation}
where%
\[
R_{v}=\left(  1+\frac{\bar{q}_{v}}{\sqrt{v}}\right)  ^{^{-3/2}}q_{v}^{2},
\]
and $\bar{q}_{v}$ lies on the interval between $0$ and $q_{v}$. Since
$E\left(  q_{v}\right)  =0$ as shown by result (\ref{qv}) of Lemma \ref{Qmom},
taking expectations of both sides of (\ref{shalf}) yields \ \
\begin{equation}
E\left(  \frac{\boldsymbol{\varepsilon}^{\prime}\mathbf{M}_{F}%
\boldsymbol{\varepsilon}}{v}\right)  ^{1/2}=1-\frac{1}{8v}E\left(
R_{v}\right)  . \label{Es_a}%
\end{equation}
It is, therefore, sufficient to show that $E\left(  \left\vert R_{v}%
\right\vert \right)  <C$. By Cauchy-Schwarz inequality we have
\begin{equation}
E\left(  \left\vert R_{v}\right\vert \right)  \leq\left[  E\left(  \left\vert
1+\frac{\bar{q}_{v}}{\sqrt{v}}\right\vert ^{-3}\right)  \right]  ^{1/2}\left[
E\left(  q_{v}^{4}\right)  \right]  ^{1/2}. \label{Rv}%
\end{equation}
Consider $\left\vert 1+\frac{\bar{q}_{v}}{\sqrt{v}}\right\vert $\ and
distinguish the cases (a) $q_{v}\geq0$ or equivalently if $\frac
{\boldsymbol{\varepsilon}^{\prime}\mathbf{M}_{F}\boldsymbol{\varepsilon}}%
{v}\geq1$ and (b) $q_{v}<0$ or equivalently if $\frac{\boldsymbol{\varepsilon
}^{\prime}\mathbf{M}_{F}\boldsymbol{\varepsilon}}{v}<1$. Under (a)$\ 0\leq$
$\bar{q}_{v}<q_{v}$, we have$\left\vert 1+\frac{\bar{q}_{v}}{\sqrt{v}%
}\right\vert \geq1$. Under (b) $q_{v}<\bar{q}_{v}<0$, we have $\left\vert
1+\frac{\bar{q}_{v}}{\sqrt{v}}\right\vert >\left\vert 1+\frac{q_{v}}{\sqrt{v}%
}\right\vert =\frac{\boldsymbol{\varepsilon}^{\prime}\mathbf{M}_{F}%
\boldsymbol{\varepsilon}}{v}$, and under condition (\ref{Lbound}) $\left\vert
1+\frac{\bar{q}_{v}}{\sqrt{v}}\right\vert >c>0$. Hence, irrespective of
whether $q_{v}\geq0$ or not,
\begin{equation}
\left\vert 1+\frac{\bar{q}_{v}}{\sqrt{v}}\right\vert >c>0, \label{Lb}%
\end{equation}
and we have $E\left(  \left\vert 1+\frac{\bar{q}_{v}}{\sqrt{v}}\right\vert
^{-3}\right)  <C$. Also it is established that $E\left(  q_{v}^{4}\right)
=O(1)$ by result (\ref{qv}) of Lemma \ref{Qmom}. Using these results in
(\ref{Rv}) it follows that $E\left(  \left\vert R_{v}\right\vert \right)  <C$,
and given (\ref{Es_a}) we can show$E\left[  \left(  \boldsymbol{\varepsilon
}^{\prime}\mathbf{M}_{F}\boldsymbol{\varepsilon}/v\right)  ^{1/2}\right]
=1+O\left(  v^{-1}\right)  . $ The other results in (\ref{Es}) can also be
established similarly. Result (\ref{E5}) follows (S.7) in Lemma 6 of Pesaran
and Yamagata (2024) by setting\textbf{\ }$\varepsilon_{t}^{2}%
=\boldsymbol{\varepsilon}^{\prime}\mathbf{A}_{1}\boldsymbol{\varepsilon}$,
where $\mathbf{A}_{1}$ has only one none-zero element on its diagonal. Result
(\ref{E6}) can be established using a result due to Lieberman (1994) (see
Lemmas 5 and 21 in the online supplement of Pesaran and Yamagata (2024)). To
establish (\ref{E7}), note that\ by applying the Taylor Theorem to $\left(
v^{-1}\boldsymbol{\varepsilon}^{\prime}\mathbf{M}_{F}\boldsymbol{\varepsilon
}\right)  ^{-1/2}$,
\[
\varepsilon_{t}^{2}\left(  \frac{\boldsymbol{\varepsilon}^{\prime}%
\mathbf{M}_{F}\boldsymbol{\varepsilon}}{v}\right)  ^{^{-1/2}}=\varepsilon
_{t}^{2}-\frac{1}{2}\frac{q_{v}\varepsilon_{t}^{2}}{\sqrt{v}}+\frac{3}%
{8v}R_{e,v}%
\]
where $R_{e,v}=\left(  1+\frac{\bar{q}_{v}}{\sqrt{v}}\right)  ^{-5/2}q_{v}%
^{2}\varepsilon_{t}^{2}, $ and taking expectations yields
\begin{equation}
E\left[  \varepsilon_{t}^{2}\left(  \frac{\boldsymbol{\varepsilon}^{\prime
}\mathbf{M}_{F}\boldsymbol{\varepsilon}}{v}\right)  ^{^{-1/2}}\right]
=E\left(  \varepsilon_{t}^{2}\right)  -\frac{1}{2}\frac{E\left(
q_{v}\varepsilon_{t}^{2}\right)  }{\sqrt{v}}+\frac{3}{8v}E\left(
R_{e,v}\right)  . \label{eshalf}%
\end{equation}
$E(\varepsilon_{t}^{2})=1$, and using (\ref{E5}) we have
\begin{equation}
E\left[  \varepsilon_{t}^{2}\left(  \frac{q_{v}}{\sqrt{v}}\right)  \right]
=E\left[  \varepsilon_{t}^{2}\left(  \frac{\boldsymbol{\varepsilon}^{\prime
}\mathbf{M}_{F}\boldsymbol{\varepsilon}}{v}-1\right)  \right]  =O\left(
\frac{1}{v}\right)  . \label{Ee2qv}%
\end{equation}
By Cauchy-Schwarz inequality
\begin{align}
E\left(  \left\vert R_{e,v}\right\vert \right)   &  \leq\left[  E\left[
\varepsilon_{t}^{4}\left(  \left\vert 1+\frac{\bar{q}_{v}}{\sqrt{v}%
}\right\vert ^{-5}\right)  \right]  \right]  ^{1/2}\left[  E\left(  q_{v}%
^{4}\right)  \right]  ^{1/2}\nonumber\\
&  \leq\left[  E\left(  \left\vert 1+\frac{\bar{q}_{v}}{\sqrt{v}}\right\vert
^{-10}\right)  \right]  ^{1/4}\left[  E\left(  \varepsilon_{t}^{8}\right)
\right]  ^{1/4}\left[  E\left(  q_{v}^{4}\right)  \right]  ^{1/2}. \label{Rev}%
\end{align}
Given (\ref{Lb}), it is easily seen that $E\left(  \left\vert 1+\frac{\bar
{q}_{v}}{\sqrt{v}}\right\vert ^{-10}\right)  <C$. Also $E\left(
\varepsilon_{t}^{8}\right)  <C$ by assumption and $E\left(  q_{v}^{4}\right)
<C$ by (\ref{qv}). Hence, using these results in (\ref{Rev}) it follows
$E\left(  \left\vert R_{e,v}\right\vert \right)  <C$, which completes the
proof of (\ref{E7}). To establish (\ref{E8}), using (\ref{shalf}) note that
for $t\neq t^{\prime}$,\textbf{ }%
\[
E\left[  \varepsilon_{t}\varepsilon_{t^{\prime}}\left(  \frac
{\boldsymbol{\varepsilon}^{\prime}\mathbf{M}_{F}\boldsymbol{\varepsilon}}%
{v}\right)  ^{-1/2}\right]  =E\left(  \varepsilon_{t}\varepsilon_{t^{\prime}%
}\right)  -\frac{1}{2}\frac{E\left(  q_{v}\varepsilon_{t}\varepsilon
_{t^{\prime}}\right)  }{\sqrt{v}}+\frac{3}{8v}E\left(  R_{tt^{\prime}%
,v}\right)
\]
where $R_{tt^{\prime},v}=\left(  1+\frac{\bar{q}_{v}}{\sqrt{v}}\right)
^{-5/2}q_{v}^{2}\varepsilon_{t}\varepsilon_{t^{\prime}}. $ Note that $E\left(
\varepsilon_{t}\varepsilon_{t^{\prime}}\right)  =0$ for $t\neq t^{\prime}$ by
serial independence of $\varepsilon_{t}$. In addition, using definition of
$q_{v}$ in (\ref{eMeT^-1/2}) yields%
\begin{align*}
\frac{E\left(  q_{v}\varepsilon_{t}\varepsilon_{t^{\prime}}\right)  }{\sqrt
{v}}  &  =E\left[  \left(  \frac{\boldsymbol{\varepsilon}^{\prime}%
\mathbf{M}_{F}\boldsymbol{\varepsilon}}{v}-1\right)  \varepsilon
_{t}\varepsilon_{t^{\prime}}\right]  =E\left[  \left(  \frac
{\boldsymbol{\varepsilon}^{\prime}\mathbf{M}_{F}\boldsymbol{\varepsilon}}%
{v}\right)  \varepsilon_{t}\varepsilon_{t^{\prime}}\right]  =\frac{1}%
{2v}E\left[  \left(  \boldsymbol{\varepsilon}^{\prime}\mathbf{A}%
\boldsymbol{\varepsilon}\right)  \left(  \boldsymbol{\varepsilon}^{\prime
}\mathbf{B}\boldsymbol{\varepsilon}\right)  \right]  ,
\end{align*}
where $\mathbf{A}=\mathbf{M}_{F}$ and $\mathbf{B}=\mathbf{(}b_{tt^{\prime}})$
with $b_{tt^{\prime}}$ and $b_{t^{\prime}t}$ ($t\neq t^{\prime}$) being the
only non-zero elements. Now using (S.7) of Lemma 6 in Pesaran and Yamagata
(2024) it follows that $E\left(  q_{v}\varepsilon_{t}\varepsilon_{t^{\prime}%
}\right)  =0$. Also by Cauchy-Schwarz inequality ,
\begin{align*}
E\left(  \left\vert R_{tt^{\prime},v}\right\vert \right)   &  \leq\left[
E\left[  \varepsilon_{t}^{2}\varepsilon_{t^{\prime}}^{2}\left(  \left\vert
1+\frac{\bar{q}_{v}}{\sqrt{v}}\right\vert ^{-5}\right)  \right]  \right]
^{1/2}\left[  E\left(  q_{v}^{4}\right)  \right]  ^{1/2}\\
&  \leq\left[  E\left(  \left\vert 1+\frac{\bar{q}_{v}}{\sqrt{v}}\right\vert
^{-10}\right)  \right]  ^{1/4}\left[  E\left(  \varepsilon_{t}^{4}\right)
E\left(  \varepsilon_{t^{\prime}}^{4}\right)  \right]  ^{1/4}\left[  E\left(
q_{v}^{4}\right)  \right]  ^{1/2}<C,
\end{align*}
where the final inequality follows using the same line of argument used to
bound $E\left(  \left\vert R_{e,v}\right\vert \right)  $ in (\ref{Rev}).
Overall, result (\ref{E8}) is established. Finally consider (\ref{E9}) and
note that\
\[
E\left[  \varepsilon_{t}\left(  \frac{\boldsymbol{\varepsilon}^{\prime
}\mathbf{M}_{F}\boldsymbol{\varepsilon}}{v}\right)  ^{-1/2}\right]  =E\left(
\varepsilon_{t}\right)  -\frac{1}{2}\frac{E\left(  q_{v}\varepsilon
_{t}\right)  }{\sqrt{v}}+\frac{3}{8v}E\left(  R_{t,v}\right)  ,
\]
where $R_{t,v}=\left(  1+\frac{\bar{q}_{v}}{\sqrt{v}}\right)  ^{-5/2}q_{v}%
^{2}\varepsilon_{t}. $ We have $E\left(  \varepsilon_{t}\right)  =0$ and
\[
\frac{E\left(  q_{v}\varepsilon_{t}\right)  }{\sqrt{v}}=E\left(
\frac{\boldsymbol{\varepsilon}^{\prime}\mathbf{M}_{F}\boldsymbol{\varepsilon
}\varepsilon_{t}}{v}-\varepsilon_{t}\right)  =E\left(  \frac
{\boldsymbol{\varepsilon}^{\prime}\mathbf{M}_{F}\boldsymbol{\varepsilon
}\varepsilon_{t}}{v}\right)  .
\]
Denote $\left\{  m_{jj^{\prime}}:j,j^{\prime}=1,2,\ldots,T\right\}  $ as the
element of $\mathbf{M}_{F}$, such that based on part (c) of Assumption
\ref{ass:errors} $m_{jj^{\prime}}$ is independent from $\varepsilon_{t}$ and
$E\left(  m_{jj^{\prime}}\right)  =O\left(  1\right)  $. Then it follows that%
\begin{align*}
E\left(  \frac{\boldsymbol{\varepsilon}^{\prime}\mathbf{M}_{F}%
\boldsymbol{\varepsilon}\varepsilon_{t}}{v}\right)   &  =\frac{1}{v}\sum
_{j=1}^{v}\sum_{j^{\prime}=1}^{v}E\left(  \varepsilon_{j}m_{jj^{\prime}%
}\varepsilon_{j^{\prime}}\varepsilon_{t}\right)  =\frac{1}{v}\sum_{j=1}%
^{v}\sum_{j^{\prime}=1}^{v}E\left(  \varepsilon_{j}\varepsilon_{j^{\prime}%
}\varepsilon_{t}\right)  E\left(  m_{jj^{\prime}}\right)  =\frac{1}{v}E\left(
\varepsilon_{t}^{3}\right)  E\left(  m_{tt}\right)
\end{align*}
which is $O\left(  v^{-1}\right)  $ where the second equation holds due to the
independence of $\varepsilon_{t}$ and $m_{jj^{\prime}},$ while the third
equation holds due to the serial independence of $\varepsilon_{t}$. Besides by
Cauchy-Schwarz inequality, \
\begin{align*}
E\left(  \left\vert R_{t,v}\right\vert \right)   &  \leq\left[  E\left[
\varepsilon_{t}^{2}\left(  \left\vert 1+\frac{\bar{q}_{v}}{\sqrt{v}%
}\right\vert ^{-5}\right)  \right]  \right]  ^{1/2}\left[  E\left(  q_{v}%
^{4}\right)  \right]  ^{1/2}\\
&  \leq\left[  E\left(  \left\vert 1+\frac{\bar{q}_{v}}{\sqrt{v}}\right\vert
^{-10}\right)  \right]  ^{1/4}\left[  E\left(  \varepsilon_{t}^{4}\right)
\right]  ^{1/4}\left[  E\left(  q_{v}^{4}\right)  \right]  ^{1/2}<C,
\end{align*}
where the last inequality holds again using the same line of argument used to
bound $E\left(  \left\vert R_{e,v}\right\vert \right)  $ in (\ref{Rev}).
Overall, result (\ref{E9}) is established.
\end{proof}

\begin{lemma}
\label{lmcdtcd}\textbf{ }Consider the latent factor model given by
(\ref{mod2}) and (\ref{uit:p}). $\widetilde{CD}$ and $CD$ statistics are
defined by (\ref{CDtilde1}) and (\ref{CDlm1}). Suppose that Assumptions
\ref{ass:factors}-\ref{ass:ws} hold and $\left(  n,T\right)  \rightarrow
\infty$, such that $n/T\rightarrow\kappa\,,$ for $0<\kappa<\infty$. Then
\begin{equation}
CD=\widetilde{CD}+o_{p}\left(  1\right)  . \label{lemgap}%
\end{equation}

\end{lemma}

\begin{proof}
Using (\ref{CDtilde1}) and (\ref{CDlm1}) we first note that%
\begin{equation}
\left(  \sqrt{\frac{2\left(  n-1\right)  }{n}}\right)  \left(
CD-\widetilde{CD}\right)  =\frac{1}{\sqrt{T}}\sum_{t=1}^{T}\left[  \left(
\frac{1}{\sqrt{n}}\sum_{i=1}^{n}\frac{\hat{u}_{it}}{\hat{\sigma}_{i,T}%
}\right)  ^{2}-\left(  \frac{1}{\sqrt{n}}\sum_{i=1}^{n}\frac{\hat{u}_{it}%
}{\omega_{i,T}}\right)  ^{2}\right]  . \label{g0}%
\end{equation}
Also note that
\begin{equation}
\frac{1}{\sqrt{n}}\sum_{i=1}^{n}\frac{\hat{u}_{it}}{\hat{\sigma}_{i,T}%
}=h_{t,nT\ }+g_{t,nT} \label{g1}%
\end{equation}
where (also see (\ref{Avee}))%
\[
h_{t,nT\ }=\frac{1}{\sqrt{n}}\sum_{i=1}^{n}\frac{\hat{u}_{it}}{\omega_{i,T}%
}=\frac{\mathbf{c}_{nT}^{\prime}\mathbf{\hat{u}}_{\circ t}}{\sqrt{n}}\text{,
and }g_{t,nT}=\frac{1}{\sqrt{n}}\sum_{i=1}^{n}\hat{u}_{it}\left(  \frac
{1}{\hat{\sigma}_{i,T}}-\frac{1}{\omega_{i,T}}\right)  =\frac{\mathbf{d}%
_{nT}^{\prime}\mathbf{\hat{u}}_{\circ t}}{\sqrt{n}}%
\]
$\mathbf{\hat{u}}_{\circ t}=(\hat{u}_{1t},\hat{u}_{2t},...,\hat{u}%
_{nt})^{\prime}$, $\mathbf{c}_{nT}=(\omega_{1,T\ }^{-1},\omega_{2,T}%
^{-1},...,\omega_{n,T}^{-1})^{\prime}$, $\mathbf{d}_{nT}=(d_{1T}%
,d_{2T},...,d_{nT})^{\prime}\,$, and $d_{iT}=\hat{\sigma}_{i,T}^{-1}%
-\omega_{i,T}^{-1}$. Then squaring both sides of (\ref{g1}) and using the
result in (\ref{g0}) we have%
\begin{align}
\left(  \sqrt{\frac{2\left(  n-1\right)  }{n}}\right)  \left(
CD-\widetilde{CD}\right)   &  =\frac{1}{\sqrt{T}}\sum_{t=1}^{T}g_{t,nT}%
^{2}+\frac{2}{\sqrt{T}}\sum_{t=1}^{T}h_{t,nT\ }g_{t,nT}\nonumber\\
&  =\sqrt{\frac{T}{n}}\left(  \frac{1}{\sqrt{n}}\mathbf{d}_{nT}^{\prime
}\mathbf{\hat{V}}_{T}\mathbf{d}_{nT}+\frac{2}{\sqrt{n}}\mathbf{c}_{nT}%
^{\prime}\mathbf{\hat{V}}_{T}\mathbf{d}_{nT}\right)  , \label{gapCD}%
\end{align}
where $\mathbf{\hat{V}}_{T}=T^{-1}\sum_{t=1}^{T}\mathbf{\hat{u}}_{\circ
t}\mathbf{\hat{u}}_{\circ t}^{\prime}$. Now using (\ref{it1}), the error
vector $\mathbf{\hat{u}}_{\circ t}$ can be written as
\[
\mathbf{\hat{u}}_{\circ t}=\mathbf{u}_{\circ t}\left(  \lambda_{T}\right)
-\mathbf{\Gamma}\left(  \mathbf{\hat{f}}_{t}-\mathbf{f}_{t}\right)  -\left(
\mathbf{\hat{\Gamma}}-\mathbf{\Gamma}\right)  \mathbf{f}_{t}-\left(
\mathbf{\hat{\Gamma}}-\mathbf{\Gamma}\right)  \left(  \mathbf{\hat{f}}%
_{t}-\mathbf{f}_{t}\right)  .
\]
where $\mathbf{u}_{\circ t}\left(  \lambda_{T}\right)  =\left(  \sigma
_{1}\varepsilon_{1t}\left(  \lambda_{T}\right)  ,\sigma_{2}\varepsilon
_{2t}\left(  \lambda_{T}\right)  ,\ldots,\sigma_{n}\varepsilon_{nt}\left(
\lambda_{T}\right)  \right)  ^{\prime}$. Using this expression we now have%
\begin{align*}
\mathbf{\hat{V}}_{T}  &  =T^{-1}\sum_{t=1}^{T}\mathbf{u}_{\circ t}\left(
\lambda_{T}\right)  \mathbf{u}_{\circ t}^{\prime}\left(  \lambda_{T}\right)
+\mathbf{\Gamma}\left[  T^{-1}\sum_{t=1}^{T}\left(  \mathbf{\hat{f}}%
_{t}-\mathbf{f}_{t}\right)  \left(  \mathbf{\hat{f}}_{t}-\mathbf{f}%
_{t}\right)  ^{\prime}\right]  \mathbf{\Gamma}^{\prime}\\
&  +\left(  \mathbf{\hat{\Gamma}-\Gamma}\right)  \left(  T^{-1}\sum_{t=1}%
^{T}\mathbf{f}_{t}\mathbf{f}_{t}^{\prime}\right)  \left(  \mathbf{\hat{\Gamma
}-\Gamma}\right)  ^{\prime}+\left(  \mathbf{\hat{\Gamma}-\Gamma}\right)
\left[  T^{-1}\sum_{t=1}^{T}\left(  \mathbf{\hat{f}}_{t}-\mathbf{f}%
_{t}\right)  \left(  \mathbf{\hat{f}}_{t}-\mathbf{f}_{t}\right)  ^{\prime
}\right]  \left(  \mathbf{\hat{\Gamma}-\Gamma}\right)  ^{\prime}%
\mathbf{\Gamma}^{\prime}\\
&  -\left[  T^{-1}\sum_{t=1}^{T}\mathbf{u}_{\circ t}\left(  \lambda
_{T}\right)  \left(  \mathbf{\hat{f}}_{t}-\mathbf{f}_{t}\right)  ^{\prime
}\right]  -\left[  T^{-1}\sum_{t=1}^{T}\mathbf{u}_{\circ t}\left(  \lambda
_{T}\right)  \mathbf{f}_{t}^{\prime}\right]  \left(  \mathbf{\hat{\Gamma
}-\Gamma}\right)  ^{\prime}\\
&  -\left[  T^{-1}\sum_{t=1}^{T}\mathbf{u}_{\circ t}\left(  \lambda
_{T}\right)  \left(  \mathbf{\hat{f}}_{t}-\mathbf{f}_{t}\right)  ^{\prime
}\right]  \left(  \mathbf{\hat{\Gamma}-\Gamma}\right)  ^{\prime}%
+\mathbf{\Gamma}\left[  T^{-1}\sum_{t=1}^{T}\left(  \mathbf{\hat{f}}%
_{t}-\mathbf{f}_{t}\right)  \mathbf{f}_{t}^{\prime}\right]  \left(
\mathbf{\hat{\Gamma}-\Gamma}\right)  ^{\prime}\\
&  +\mathbf{\Gamma}\left[  T^{-1}\sum_{t=1}^{T}\left(  \mathbf{\hat{f}}%
_{t}-\mathbf{f}_{t}\right)  \left(  \mathbf{\hat{f}}_{t}-\mathbf{f}%
_{t}\right)  ^{\prime}\right]  \left(  \mathbf{\hat{\Gamma}-\Gamma}\right)
^{\prime}+\left(  \mathbf{\hat{\Gamma}-\Gamma}\right)  \left[  T^{-1}%
\sum_{t=1}^{T}\mathbf{f}_{t}\left(  \mathbf{\hat{f}}_{t}-\mathbf{f}%
_{t}\right)  ^{\prime}\right]  \left(  \mathbf{\hat{\Gamma}-\Gamma}\right)
^{\prime},
\end{align*}
or in matrix forms
\begin{align*}
\mathbf{\hat{V}}_{T}  &  =\mathbf{V}_{T}\left(  \lambda_{T}\right)
+\mathbf{\Gamma}\left[  T^{-1}\left(  \mathbf{\hat{F}-F}\right)  ^{\prime
}\left(  \mathbf{\hat{F}-F}\right)  \right]  \mathbf{\Gamma}^{\prime}+\left(
\mathbf{\hat{\Gamma}-\Gamma}\right)  \mathbf{\Sigma}_{T,ff}\left(
\mathbf{\hat{\Gamma}-\Gamma}\right)  ^{\prime}\\
&  +\left(  \mathbf{\hat{\Gamma}-\Gamma}\right)  \left[  T^{-1}\left(
\mathbf{\hat{F}-F}\right)  ^{\prime}\left(  \mathbf{\hat{F}-F}\right)
\right]  \left(  \mathbf{\hat{\Gamma}-\Gamma}\right)  ^{\prime}-T^{-1}%
\mathbf{U}^{\prime}\left(  \lambda_{T}\right)  \left(  \mathbf{\hat{F}%
-F}\right)  \mathbf{\Gamma}^{\prime}\\
&  -T^{-1}\mathbf{U}^{\prime}\left(  \lambda_{T}\right)  \mathbf{F}\left(
\mathbf{\hat{\Gamma}-\Gamma}\right)  ^{\prime}-T^{-1}\mathbf{U}^{\prime
}\left(  \lambda_{T}\right)  \left(  \mathbf{\hat{F}-F}\right)  \left(
\mathbf{\hat{\Gamma}-\Gamma}\right)  ^{\prime}+\mathbf{\Gamma}\left[
T^{-1}\mathbf{F}^{\prime}\left(  \mathbf{\hat{F}-F}\right)  \right]  \left(
\mathbf{\hat{\Gamma}-\Gamma}\right)  ^{\prime}\\
&  +\mathbf{\Gamma}\left[  T^{-1}\left(  \mathbf{\hat{F}-F}\right)  ^{\prime
}\left(  \mathbf{\hat{F}-F}\right)  \right]  \left(  \mathbf{\hat{\Gamma
}-\Gamma}\right)  ^{\prime}+\left(  \mathbf{\hat{\Gamma}-\Gamma}\right)
\left[  T^{-1}\mathbf{F}^{\prime}\left(  \mathbf{\hat{F}-F}\right)  \right]
\left(  \mathbf{\hat{\Gamma}-\Gamma}\right)  ^{\prime},%
\end{align*}
where $\mathbf{V}_{T}\left(  \lambda_{T}\right)  =T^{-1}\sum_{t=1}%
^{T}\mathbf{u}_{\circ t}\left(  \lambda_{T}\right)  \mathbf{u}_{\circ
t}^{^{\prime}}\left(  \lambda_{T}\right)  $ and $\mathbf{U}\left(  \lambda
_{T}\right)  =\left(  \mathbf{u}_{\circ1}\left(  \lambda_{T}\right)
,\mathbf{u}_{\circ2}\left(  \lambda_{T}\right)  ,\ldots,\mathbf{u}_{\circ
T}\left(  \lambda_{T}\right)  \right)  ^{\prime}$. By (\ref{|VT|}) we have
$\left\Vert \mathbf{V}_{T}\left(  \lambda_{T}\right)  \right\Vert =\mu_{\max
}\left(  \mathbf{V}_{T}\left(  \lambda_{T}\right)  \right)  =O_{p}(\frac{n}%
{T})$. By results in Lemma \ref{lm2} all other terms of the $\mathbf{\hat{V}%
}_{T}$ are either $O_{p}(1)$ or of lower order, and we also have
$\mathbf{\hat{V}}_{T}=O_{p}\left(  1\right)  $ since $n$ and $T$ are of the
same order. Consider now the terms in (\ref{gapCD}) and note that%
\[
\left(  \sqrt{\frac{2\left(  n-1\right)  }{n}}\right)  \left\vert
CD-\widetilde{CD}\right\vert <C\left\Vert \mathbf{\hat{V}}_{T}\right\Vert
\left[  \left(  \frac{1}{\sqrt{n}}\left\Vert \mathbf{d}_{nT}\right\Vert
^{2}\right)  +\left(  \frac{2}{\sqrt{n}}\left\Vert \mathbf{c}_{nT}\right\Vert
\right)  \left\Vert \mathbf{d}_{nT}\right\Vert \right]  .
\]
where $\frac{1}{\sqrt{n}}\left\Vert \mathbf{c}_{nT}\right\Vert =\left(
n^{-1}\sum_{i=1}^{n}\omega_{i,T}^{-2}\right)  ^{1/2}$ and $\left\Vert
\mathbf{d}_{nT}\right\Vert =\left(  \sum_{i=1}^{n}\left(  \hat{\sigma}%
_{i,T}^{-1}-\omega_{i,T}^{-1}\right)  ^{2}\right)  ^{1/2}$. By part (c) of
Assumption \ref{ass:errors} $E\left(  \omega_{i,T}^{-2}\right)  <C<\infty$,
such that $\frac{1}{n}E\left\Vert \mathbf{c}_{nT}\right\Vert ^{2}=n^{-1}%
\sum_{i=1}^{n}E\left(  \omega_{i,T}^{-2}\right)  <C$ and hence $n^{-1/2}%
\left\Vert \mathbf{c}_{nT}\right\Vert =O_{p}\left(  1\right)  $. Also by
(\ref{sup:|sgm(-1)-ome(-1)|}) we have \textbf{ }%
\[
\sup_{i}\left\vert \hat{\sigma}_{i,T}^{-1}-\omega_{i,T}^{-1}\right\vert
^{2}=\left(  \sup_{i}\left\vert \hat{\sigma}_{i,T}^{-1}-\omega_{i,T}%
^{-1}\right\vert \right)  ^{2}=O_{p}\left[  \left(  \frac{\ln\left(  n\right)
}{T}\right)  ^{2}\right]  ,
\]
therefore%
\[
\left\Vert \mathbf{d}_{nT}\right\Vert ^{2}=\sum_{i=1}^{n}\left(  \hat{\sigma
}_{i,T}^{-1}-\omega_{i,T}^{-1}\right)  ^{2}\leq n\sup_{i}\left(  \hat{\sigma
}_{i,T}^{-1}-\omega_{i,T}^{-1}\right)  ^{2}=O_{p}\left[  n\left(  \frac
{\ln\left(  n\right)  }{T}\right)  ^{2}\right]  =o_{p}\left(  1\right)  ,
\]
recalling that $n$ and $T$ are of the same order. Hence, $\left\vert
CD-\widetilde{CD}\right\vert =o_{p}(1)$, as required.
\end{proof}

\begin{lemma}
\label{lc}Suppose the data are generated by the latent factor model given by
(\ref{mod2}) and (\ref{uit:p}). The latent factors, $\mathbf{f}_{t}$, and
their loadings, $\boldsymbol{\gamma}_{i}$, are estimated by principal
components, $\mathbf{\hat{f}}_{t}$ and $\boldsymbol{\hat{\gamma}}_{i}$, given
by (\ref{hatg}). Suppose that Assumptions \ref{ass:factors}-\ref{ass:ws} hold
and $\left(  n,T\right)  \rightarrow\infty$, such that $n/T\rightarrow
\kappa\,,$ for $0<\kappa<\infty$. Then
\begin{align}
\frac{1}{n}\sum_{i=1}^{n}\frac{\boldsymbol{\hat{\gamma}}_{i}%
-\boldsymbol{\gamma}_{i}}{\sigma_{i}}  &  =O_{p}\left(  \sqrt{\frac{\ln\left(
n\right)  }{nT}}\right)  ,\label{lc2}\\
\frac{1}{n}\sum_{i=1}^{n}\left(  \boldsymbol{\hat{\gamma}}_{i}%
-\boldsymbol{\gamma}_{i}\right)  \sigma_{i}  &  =O_{p}\left(  \sqrt{\frac
{\ln\left(  n\right)  }{nT}}\right)  ,\label{lc3}\\
\frac{1}{n}\sum_{i=1}^{n}\sigma_{i}^{2}\left(  \boldsymbol{\hat{\gamma}}%
_{i}\boldsymbol{\hat{\gamma}}_{i}^{\prime}-\boldsymbol{\gamma}_{i}%
\boldsymbol{\gamma}_{i}^{\prime}\right)   &  =O_{p}\left(  \sqrt{\frac
{\ln\left(  n\right)  }{nT}}\right)  ,\label{lc4}\\
\frac{1}{n}\sum_{i=1}^{n}\hat{\sigma}_{i,T}\boldsymbol{\hat{\gamma}}_{i}%
-\frac{1}{n}\sum_{i=1}^{n}\sigma_{i}\boldsymbol{\gamma}_{i}  &  =O_{p}\left(
\frac{\ln\left(  n\right)  }{T}\right)  . \label{lc1}%
\end{align}

\end{lemma}

\begin{proof}
Results (\ref{lc2}) and (\ref{lc3}) follow directly from (\ref{hgg}) by
setting $b_{in}=\sigma_{i}^{-1}$ and $b_{in}=\sigma_{i}$, respectively. To
prove (\ref{lc4}), note that%
\begin{align}
\frac{1}{n}\sum_{i=1}^{n}\sigma_{i}^{2}\left(  \boldsymbol{\hat{\gamma}}%
_{i}\boldsymbol{\hat{\gamma}}_{i}^{\prime}-\boldsymbol{\gamma}_{i}%
\boldsymbol{\gamma}_{i}^{\prime}\right)   &  =\frac{1}{n}\sum_{i=1}^{n}%
\sigma_{i}^{2}\left(  \boldsymbol{\hat{\gamma}}_{i}-\boldsymbol{\gamma}%
_{i}\right)  \left(  \boldsymbol{\hat{\gamma}}_{i}-\boldsymbol{\gamma}%
_{i}\right)  ^{\prime}\nonumber\\
&  +\frac{1}{n}\sum_{i=1}^{n}\sigma_{i}^{2}\left(  \boldsymbol{\hat{\gamma}%
}_{i}-\boldsymbol{\gamma}_{i}\right)  \boldsymbol{\gamma}_{i}^{\prime}%
+\frac{1}{n}\sum_{i=1}^{n}\sigma_{i}^{2}\boldsymbol{\gamma}_{i}\left(
\boldsymbol{\hat{\gamma}}_{i}-\boldsymbol{\gamma}_{i}\right)  ^{\prime}.
\label{lc4.1}%
\end{align}
Since $\sigma_{i}^{2}$ is bounded, then
\begin{align*}
\left\Vert \frac{1}{n}\sum_{i=1}^{n}\sigma_{i}^{2}\left(  \boldsymbol{\hat
{\gamma}}_{i}-\boldsymbol{\gamma}_{i}\right)  \left(  \boldsymbol{\hat{\gamma
}}_{i}-\boldsymbol{\gamma}_{i}\right)  ^{\prime}\right\Vert  &  \leq\left(
\sup_{i}\sigma_{i}^{2}\right)  \left\Vert \frac{1}{n}\sum_{i=1}^{n}\left(
\boldsymbol{\hat{\gamma}}_{i}-\boldsymbol{\gamma}_{i}\right)  \left(
\boldsymbol{\hat{\gamma}}_{i}-\boldsymbol{\gamma}_{i}\right)  ^{\prime
}\right\Vert \\
&  \leq\left(  \sup_{i}\sigma_{i}^{2}\right)  \left(  \frac{1}{n}\sum
_{i=1}^{n}\left\Vert \boldsymbol{\hat{\gamma}}_{i}-\boldsymbol{\gamma}%
_{i}\right\Vert ^{2}\right)  ,
\end{align*}
and using (\ref{lm2.2}) it follows that
\[
\frac{1}{n}\sum_{i=1}^{n}\sigma_{i}^{2}\left(  \boldsymbol{\hat{\gamma}}%
_{i}-\boldsymbol{\gamma}_{i}\right)  \left(  \boldsymbol{\hat{\gamma}}%
_{i}-\boldsymbol{\gamma}_{i}\right)  ^{\prime}=O_{p}\left(  \frac{1}%
{\delta_{nT}^{2}}\right)  .
\]
Also, using (\ref{hggg}) and setting $b_{in}=\sigma_{i}$, we have%
\[
\left\Vert \frac{1}{n}\sum_{i=1}^{n}\sigma_{i}^{2}\left(  \boldsymbol{\hat
{\gamma}}_{i}-\boldsymbol{\gamma}_{i}\right)  \boldsymbol{\gamma}_{i}^{\prime
}\right\Vert =\left\Vert \frac{1}{n}\sum_{i=1}^{n}\sigma_{i}^{2}%
\boldsymbol{\gamma}_{i}\left(  \boldsymbol{\hat{\gamma}}_{i}%
-\boldsymbol{\gamma}_{i}\right)  ^{\prime}\right\Vert =O_{p}\left(
\sqrt{\frac{\ln\left(  n\right)  }{nT}}\right)  ,
\]
then (\ref{lc4}) is established based on (\ref{lc4.1}). Finally, consider
(\ref{lc1}) and note that
\begin{align}
&  \frac{1}{n}\sum_{i=1}^{n}\hat{\sigma}_{i,T}\boldsymbol{\hat{\gamma}}%
_{i}-\frac{1}{n}\sum_{i=1}^{n}\sigma_{i}\boldsymbol{\gamma}_{i}\nonumber\\
&  =\frac{1}{n}\sum_{i=1}^{n}\left[  \left(  \hat{\sigma}_{i,T}-\omega
_{i,T}\right)  +\omega_{i,T}\right]  \left(  \boldsymbol{\hat{\gamma}}%
_{i}-\boldsymbol{\gamma}_{i}+\boldsymbol{\gamma}_{i}\right)  -\frac{1}{n}%
\sum_{i=1}^{n}\sigma_{i}\boldsymbol{\gamma}_{i}\nonumber\\
&  =\frac{1}{n}\sum_{i=1}^{n}\boldsymbol{\gamma}_{i}\left(  \omega
_{i,T}-\sigma_{i}\right)  +\frac{1}{n}\sum_{i=1}^{n}\boldsymbol{\gamma}%
_{i}\left(  \hat{\sigma}_{i,T}-\omega_{i,T}\right)  +\frac{1}{n}\sum_{i=1}%
^{n}\sigma_{i}\left(  \boldsymbol{\hat{\gamma}}_{i}-\boldsymbol{\gamma}%
_{i}\right) \nonumber\\
&  +\frac{1}{n}\sum_{i=1}^{n}\left(  \omega_{i,T}-\sigma_{i}\right)  \left(
\boldsymbol{\hat{\gamma}}_{i}-\boldsymbol{\gamma}_{i}\right)  +\frac{1}{n}%
\sum_{i=1}^{n}\left(  \hat{\sigma}_{i,T}-\omega_{i,T}\right)  \left(
\boldsymbol{\hat{\gamma}}_{i}-\boldsymbol{\gamma}_{i}\right) \nonumber\\
&  =\mathbf{A}_{1,nT}+\mathbf{A}_{2,nT}+\mathbf{A}_{3,nT}+\mathbf{A}%
_{4,nT}+\mathbf{A}_{5,nT}. \label{lc1.1}%
\end{align}
Recall also that under Assumptions \ref{ass:errors} and \ref{ass:loadings}
$\sigma_{i}$ and $\boldsymbol{\gamma}_{i}$ are bounded and $\omega_{i,T}%
^{2}=T^{-1}\sigma_{i}^{2}\boldsymbol{\varepsilon}_{i\circ}^{^{\prime}%
}\mathbf{M}_{F}\boldsymbol{\varepsilon}_{i\circ}$, for $i=1,2,...,n$ are
distributed independently across $i,$ and from $\sigma_{i}$ and
$\boldsymbol{\gamma}_{i}$. Starting with $\mathbf{A}_{1,nT}$, by (\ref{Es}) we
have
\[
E\left(  \sqrt{nT}\mathbf{A}_{1,nT}\right)  =\frac{\sqrt{nT}}{n}\sum_{i=1}%
^{n}(\boldsymbol{\gamma}_{i}\sigma_{i})E\left[  \left(  \frac
{\boldsymbol{\varepsilon}_{i\circ}^{^{\prime}}\mathbf{M}_{F}%
\boldsymbol{\varepsilon}_{i\circ}}{T}\right)  ^{1/2}-1\right]  =O\left(
\frac{\sqrt{nT}}{T}\right)  .
\]
Since $n$ and $T$ are assumed to be of the same order then $E\left(  \sqrt
{nT}\mathbf{A}_{1,nT}\right)  =O\left(  1\right)  $. Also, using result
(\ref{Es})%
\[
Var\left(  \sqrt{nT}\mathbf{A}_{1,nT}\right)  =\frac{nT}{n^{2}}\sum_{i=1}%
^{n}\left(  \sigma_{i}^{2}\boldsymbol{\gamma}_{i}\boldsymbol{\gamma}%
_{i}^{\prime}\right)  Var\left[  \left(  \frac{\boldsymbol{\varepsilon
}_{i\circ}^{^{\prime}}\mathbf{M}_{F}\boldsymbol{\varepsilon}_{i\circ}}%
{T}\right)  ^{1/2}\right]  =O\left(  1\right)  .
\]
Therefore, $\sqrt{nT}\mathbf{A}_{1,nT}=O_{p}(1)$ and it follows that
$\mathbf{A}_{1,nT}=O_{p}\left[  \left(  nT\right)  ^{-1/2}\right]  $. Further,
using (\ref{Ews}) and setting $b_{in}=\gamma_{ij}$, for $j=1,2,...,m_{0}$, it
follows that
\[
\mathbf{A}_{2,nT}=\frac{1}{n}\sum_{i=1}^{n}\boldsymbol{\gamma}_{i}\left(
\hat{\sigma}_{i,T}-\omega_{i,T}\right)  =O_{p}\left(  \frac{\ln\left(
n\right)  }{T}\right)  .
\]
Since $\mathbf{A}_{3,nT}$ is the same as the result in (\ref{lc3}), which is
already established, then $\mathbf{A}_{3,nT}=O_{p}\left(  \sqrt{\ln\left(
n\right)  /\left(  nT\right)  }\right)  $. Using result (\ref{hgg_b}) it
follows that
\[
\mathbf{A}_{4,nT}=\frac{1}{n}\sum_{i=1}^{n}\left(  \omega_{i,T}-\sigma
_{i}\right)  \left(  \boldsymbol{\hat{\gamma}}_{i}-\boldsymbol{\gamma}%
_{i}\right)  =O_{p}\left(  \frac{1}{\delta_{nT}^{2}}\right)  .
\]
Using result (\ref{hgg_c}) we have%
\[
\mathbf{A}_{5,nT}=\frac{1}{n}\sum_{i=1}^{n}\left(  \hat{\sigma}_{i,T}%
-\omega_{i,T}\right)  \left(  \boldsymbol{\hat{\gamma}}_{i}-\boldsymbol{\gamma
}_{i}\right)  =O_{p}\left[  \left(  \frac{\ln\left(  n\right)  }{T}\right)
^{3/2}\right]  .
\]
Result (\ref{lc1}) now follows straightforwardly based on (\ref{lc1.1}).
\end{proof}

\begin{lemma}
\label{lm:gnT}\textbf{ }Suppose the data are generated by the latent factor
model given by (\ref{mod2}) and (\ref{uit:p}). Denote $\boldsymbol{\varphi
}_{n}=n^{-1}\sum_{i=1}^{n}\boldsymbol{\gamma}_{i}/\sigma_{i}$,
$\boldsymbol{\varphi}_{nT}=n^{-1}\sum_{i=1}^{n}\boldsymbol{\gamma}_{i}%
/\omega_{i,T}$ with $\omega_{i,T}=\left(  T^{-1}\sigma_{i}^{2}%
\boldsymbol{\varepsilon}_{i\circ}^{\prime}\mathbf{M}_{F}%
\boldsymbol{\varepsilon}_{i\circ}\right)  ^{1/2}$ where
$\boldsymbol{\varepsilon}_{i\circ}=\left(  \varepsilon_{i1},\varepsilon
_{i2},\ldots,\varepsilon_{iT}\right)  ^{\prime}$ and $\mathbf{M}%
_{F}=\mathbf{I}_{T}-\mathbf{F}\left(  \mathbf{F}^{\prime}\mathbf{F}\right)
^{-1}\mathbf{F}^{\prime}$. Suppose that Assumptions \ref{ass:factors}%
-\ref{ass:ws} hold and $\left(  n,T\right)  \rightarrow\infty$, such that
$n/T\rightarrow\kappa\,,$ for $0<\kappa<\infty$. Then%
\begin{align*}
g_{2,nT}\left(  \lambda_{T}\right)   &  =\frac{1}{\sqrt{T}}\sum_{t=1}%
^{T}\upsilon_{t,nT}^{2}\left(  \lambda_{T}\right)  =o_{p}\left(  1\right)  ,\\
g_{3,nT}\left(  \lambda_{T}\right)   &  =\sqrt{T}\left(  \boldsymbol{\varphi
}_{nT}-\boldsymbol{\varphi}_{n}\right)  ^{\prime}\left(  \frac{\sum_{t=1}%
^{T}\boldsymbol{\kappa}_{t,n}\left(  \lambda_{T}\right)  \boldsymbol{\kappa
}_{t,n}^{\prime}\left(  \lambda_{T}\right)  }{T}\right)  \left(
\boldsymbol{\varphi}_{nT}-\boldsymbol{\varphi}_{n}\right)  =o_{p}\left(
1\right)  ,\\
g_{4,nT}\left(  \lambda_{T}\right)   &  =\sqrt{T}\left(  \boldsymbol{\varphi
}_{nT}-\boldsymbol{\varphi}_{n}\right)  ^{\prime}\left(  \frac{1}{T}\sum
_{t=1}^{T}\boldsymbol{\kappa}_{t,n}\left(  \lambda_{T}\right)  \upsilon
_{t,nT}\left(  \lambda_{T}\right)  \right)  =o_{p}\left(  1\right)  ,\\
g_{5,nT}\left(  \lambda_{T}\right)   &  =\sqrt{T}\left(  \boldsymbol{\varphi
}_{nT}-\boldsymbol{\varphi}_{n}\right)  ^{\prime}\left(  \frac{1}{T}\sum
_{t=1}^{T}\boldsymbol{\kappa}_{t,n}\left(  \lambda_{T}\right)  \xi
_{t,n}\left(  \lambda_{T}\right)  \right)  =o_{p}\left(  1\right)  ,
\end{align*}
where%
\begin{align*}
\upsilon_{t,nT}\left(  \lambda_{T}\right)   &  =\frac{1}{\sqrt{n}}\sum
_{i=1}^{n}\left(  \frac{1}{\left(  \boldsymbol{\varepsilon}_{i\circ}^{\prime
}\mathbf{M}_{F}\boldsymbol{\varepsilon}_{i\circ}/T\right)  ^{1/2}}-1\right)
\varepsilon_{it}\left(  \lambda_{T}\right)  ,\\
\boldsymbol{\kappa}_{t,n}\left(  \lambda_{T}\right)   &  =\frac{1}{\sqrt{n}%
}\sum_{i=1}^{n}\boldsymbol{\gamma}_{i}\sigma_{i}\varepsilon_{it}\left(
\lambda_{T}\right)  ,\\
\xi_{t,n}\left(  \lambda_{T}\right)   &  =\frac{1}{\sqrt{n}}\sum_{i=1}%
^{n}a_{i,n}\varepsilon_{it}\left(  \lambda_{T}\right)  ,\text{ }%
a_{i,n}=1-\sigma_{i}\boldsymbol{\varphi}_{n}^{\prime}\boldsymbol{\gamma}_{i}.
\end{align*}

\end{lemma}

\begin{proof}
Starting with the $g_{2,nT}\left(  \lambda_{T}\right)  $, note that%
\begin{equation}
g_{2,nT}\left(  \lambda_{T}\right)  =\sqrt{T}\left(  \frac{1}{T}\sum_{t=1}%
^{T}\upsilon_{t,nT}^{2}\left(  \lambda_{T}\right)  \right)  \leq\sqrt
{T}\left(  \sup_{t}\upsilon_{t,nT}^{2}\left(  \lambda_{T}\right)  \right)  .
\label{g2nT}%
\end{equation}
Meanwhile, we have\textbf{ }%
\[
\left\vert \upsilon_{t,nT}\left(  \lambda_{T}\right)  \right\vert \leq\frac
{1}{\sqrt{n}}\sum_{i=1}^{n}\left\vert \left(  \frac{1}{\left(
\boldsymbol{\varepsilon}_{i\circ}^{\prime}\mathbf{M}_{F}%
\boldsymbol{\varepsilon}_{i\circ}/T\right)  ^{1/2}}-1\right)  \right\vert
\left\vert \varepsilon_{it}\left(  \lambda_{T}\right)  \right\vert ,
\]
and
\begin{equation}
\sup_{t}\left\vert \upsilon_{t,nT}\left(  \lambda_{T}\right)  \right\vert
\leq\left[  \frac{1}{\sqrt{n}}\sum_{i=1}^{n}\left\vert \left(  \frac
{1}{\left(  \boldsymbol{\varepsilon}_{i\circ}^{\prime}\mathbf{M}%
_{F}\boldsymbol{\varepsilon}_{i\circ}/T\right)  ^{1/2}}-1\right)  \right\vert
\right]  \left(  \sup_{i,t}\left\vert \varepsilon_{it}\left(  \lambda
_{T}\right)  \right\vert \right)  . \label{sup:vtn}%
\end{equation}
Consider the first term of the product in (\ref{sup:vtn}) and note since
$\boldsymbol{\varepsilon}_{i\circ}$ is cross-sectionally independent
conditional then
\[
E\left[  \frac{1}{\sqrt{n}}\sum_{i=1}^{n}\left\vert \left(  \frac{1}{\left(
\boldsymbol{\varepsilon}_{i\circ}^{\prime}\mathbf{M}_{F}%
\boldsymbol{\varepsilon}_{i\circ}/T\right)  ^{1/2}}-1\right)  \right\vert
\right]  ^{2}=\frac{1}{n}\sum_{i=1}^{n}E\left(  \frac{1}{\left(
\boldsymbol{\varepsilon}_{i\circ}^{\prime}\mathbf{M}_{F}%
\boldsymbol{\varepsilon}_{i\circ}/T\right)  ^{1/2}}-1\right)  ^{2}.
\]
Also, using (\ref{Es}) we obtain
\begin{align*}
E\left(  \frac{1}{\left(  \boldsymbol{\varepsilon}_{i\circ}^{\prime}%
\mathbf{M}_{F}\boldsymbol{\varepsilon}_{i\circ}/T\right)  ^{1/2}}-1\right)
^{2}  &  =E\left(  \frac{1}{\boldsymbol{\varepsilon}_{i\circ}^{\prime
}\mathbf{M}_{F}\boldsymbol{\varepsilon}_{i\circ}/T}\right)  -2E\left(
\frac{1}{\left(  \boldsymbol{\varepsilon}_{i\circ}^{\prime}\mathbf{M}%
_{F}\boldsymbol{\varepsilon}_{i\circ}/T\right)  ^{1/2}}\right)  +1\\
&  =\left[  E\left(  \frac{1}{\boldsymbol{\varepsilon}_{i\circ}^{\prime
}\mathbf{M}_{F}\boldsymbol{\varepsilon}_{i\circ}/T}\right)  -1\right]
-2\left[  E\left(  \frac{1}{\left(  \boldsymbol{\varepsilon}_{i\circ}^{\prime
}\mathbf{M}_{F}\boldsymbol{\varepsilon}_{i\circ}/T\right)  ^{1/2}}\right)
-1\right] \\
&  =O\left(  \frac{1}{T}\right)  .
\end{align*}
Therefore by Markov inequality, we have%
\begin{equation}
\frac{1}{\sqrt{n}}\sum_{i=1}^{n}\left\vert \left(  \frac{1}{\left(
\boldsymbol{\varepsilon}_{i\circ}^{\prime}\mathbf{M}_{F}%
\boldsymbol{\varepsilon}_{i\circ}/T\right)  ^{1/2}}-1\right)  \right\vert
=O_{p}\left(  \frac{1}{\sqrt{T}}\right)  . \label{ave:zMz1}%
\end{equation}
Now consider the second term of of the product in (\ref{sup:vtn}) and note
there exist $C_{\varepsilon,1}$, $C_{\varepsilon,2}$ and $r_{\varepsilon}>0$
such that
\[
\Pr\left(  \sup_{i,t}\left\vert \varepsilon_{it}\left(  \lambda_{T}\right)
\right\vert >a_{\varepsilon}\right)  \leq nT\sup_{i,t}\Pr\left(  \left\vert
\varepsilon_{it}\left(  \lambda_{T}\right)  \right\vert >a_{\varepsilon
}\right)  \leq nTC_{\varepsilon,1}\exp\left(  -C_{\varepsilon,2}\left(
a_{\varepsilon}\right)  ^{r_{\varepsilon}}\right)  .
\]
Hence, by letting $a_{\varepsilon}=\ominus\left(  \ln\left(  nT\right)
\right)  $, we have%
\[
\Pr\left(  \sup_{i,t}\left\vert \varepsilon_{it}\left(  \lambda_{T}\right)
\right\vert >a_{\varepsilon}\right)  \leq C_{\varepsilon,1}\exp\left(
\ln\left(  nT\right)  -C_{\varepsilon,2}\left(  a_{\varepsilon}\right)
^{r_{\varepsilon}}\right)  =O\left(  1\right)  ,
\]
which further implies $\sup_{i,t}\left\vert \varepsilon_{it}\left(
\lambda_{T}\right)  \right\vert =O_{p}\left(  \ln\left(  nT\right)  \right)
$. Invoking this result and (\ref{ave:zMz1}) in (\ref{sup:vtn}) now yields%
\begin{equation}
\sup_{t}\left\vert \upsilon_{t,nT}\left(  \lambda_{T}\right)  \right\vert
=O_{p}\left(  \frac{\ln\left(  nT\right)  }{\sqrt{T}}\right)  . \label{sup:v}%
\end{equation}
Then consider (\ref{g2nT}) and it follows%
\begin{equation}
g_{2,nT}\left(  \lambda_{T}\right)  \leq\sqrt{T}\left(  \sup_{t}%
\upsilon_{t,nT}^{2}\left(  \lambda_{T}\right)  \right)  \leq\sqrt{T}\left(
\sup_{t}\left\vert \upsilon_{t,nT}\left(  \lambda_{T}\right)  \right\vert
\right)  ^{2}=O_{p}\left(  \frac{\left[  \ln\left(  nT\right)  \right]  ^{2}%
}{\sqrt{T}}\right)  =o_{p}\left(  1\right)  , \label{g2_p}%
\end{equation}
as $n$ and $T$ are of the same order of magnitude. Consider $g_{3,nT}\left(
\lambda_{T}\right)  $ and note that the result of Lemma \ref{PHIn} holds such
that \textbf{ }
\begin{equation}
\sqrt{T}\left(  \boldsymbol{\varphi}_{n}-\boldsymbol{\varphi}_{nT}\right)
=O_{p}\left(  n^{-1/2}\right)  +O_{p}\left(  T^{-1/2}\right)  . \label{phin_p}%
\end{equation}
Furthermore, we have
\begin{align*}
E\left\Vert \boldsymbol{\kappa}_{t,n}\left(  \lambda_{T}\right)  \right\Vert
^{2}  &  =\frac{1}{n}\sum_{i=1}^{n}\sum_{j=1}^{n}\boldsymbol{\gamma}%
_{i}^{\prime}\boldsymbol{\gamma}_{j}\sigma_{i}\sigma_{j}E\left(
\varepsilon_{it}\left(  \lambda_{T}\right)  \varepsilon_{jt}\left(
\lambda_{T}\right)  \right) \\
&  \leq\frac{1}{n}\sum_{i=1}^{n}\sum_{j=1}^{n}\left\Vert \boldsymbol{\gamma
}_{i}\right\Vert \left\Vert \boldsymbol{\gamma}_{j}\right\Vert \left\vert
\sigma_{i}\sigma_{j}\right\vert \left\vert E\left(  \varepsilon_{it}\left(
\lambda_{T}\right)  \varepsilon_{jt}\left(  \lambda_{T}\right)  \right)
\right\vert \\
&  \leq\left(  \sup_{i}\left\Vert \boldsymbol{\gamma}_{i}\right\Vert
^{2}\right)  \left(  \sup_{i}\sigma_{i}^{2}\right)  \left[  \frac{1}{n}%
\sum_{i=1}^{n}\sum_{j=1}^{n}\left\vert E\left(  \varepsilon_{it}\left(
\lambda_{T}\right)  \varepsilon_{jt}\left(  \lambda_{T}\right)  \right)
\right\vert \right]  .
\end{align*}
Given (\ref{ave:rho}) and boundedness of $\boldsymbol{\gamma}_{i}$ and
$\sigma_{i}$, it follows
\begin{equation}
E\left\Vert \boldsymbol{\kappa}_{t,n}\left(  \lambda_{T}\right)  \right\Vert
^{2}\leq\left(  \sup_{i}\left\Vert \boldsymbol{\gamma}_{i}\right\Vert
^{2}\right)  \left(  \sup_{i}\sigma_{i}^{2}\right)  \left[  \frac{1}{n}%
\sum_{i=1}^{n}\sum_{j=1}^{n}\left\vert E\left(  \varepsilon_{it}\left(
\lambda_{T}\right)  \varepsilon_{jt}\left(  \lambda_{T}\right)  \right)
\right\vert \right]  =O\left(  1\right)  \label{E:kappa2}%
\end{equation}
and $\boldsymbol{\kappa}_{t,n}\left(  \lambda_{T}\right)  $ is therefore
$O_{p}\left(  1\right)  $. Also, under Assumption \ref{ass:errors}
$\boldsymbol{\kappa}_{t,n}\left(  \lambda_{T}\right)  $ is serially
independent and we have $T^{-1}\sum_{t=1}^{T}\boldsymbol{\kappa}_{t,n}\left(
\lambda_{T}\right)  \boldsymbol{\kappa}_{t,n}^{\prime}\left(  \lambda
_{T}\right)  =O_{p}\left(  1\right)  $. Using this result together with
(\ref{phin_p}) we then have
\begin{equation}
g_{3,nT}\left(  \lambda_{T}\right)  =o_{p}(1). \label{g3_p}%
\end{equation}
Next consider $g_{4,nT}\left(  \lambda_{T}\right)  $ and note that by
Cauchy-Schwarz inequality%
\[
\left\Vert \frac{1}{T}\sum_{t=1}^{T}\boldsymbol{\kappa}_{t,n}\left(
\lambda_{T}\right)  \upsilon_{t,nT}\left(  \lambda_{T}\right)  \right\Vert
\leq\left(  \frac{1}{T}\sum_{t=1}^{T}\left\Vert \boldsymbol{\kappa}%
_{t,n}\left(  \lambda_{T}\right)  \right\Vert ^{2}\right)  ^{1/2}\left(
\frac{1}{T}\sum_{t=1}^{T}\upsilon_{t,nT}^{2}\left(  \lambda_{T}\right)
\right)  ^{1/2}=O_{p}\left(  \frac{\ln\left(  nT\right)  }{\sqrt{T}}\right)
\]
where the equation holds by (\ref{g2_p})\ and (\ref{E:kappa2}). Then using the
above results it also follows that%
\begin{equation}
g_{4,nT}\left(  \lambda_{T}\right)  =\sqrt{T}\left(  \boldsymbol{\varphi}%
_{nT}-\boldsymbol{\varphi}_{n}\right)  ^{\prime}\left(  \frac{1}{T}\sum
_{t=1}^{T}\boldsymbol{\kappa}_{t,n}\left(  \lambda_{T}\right)  \upsilon
_{t,nT}\left(  \lambda_{T}\right)  \right)  =o_{p}(1). \label{g4_p}%
\end{equation}
Similarly, note $\xi_{t,n}\left(  \lambda_{T}\right)  =n^{-1/2}\sum_{i=1}%
^{n}a_{i,n}\varepsilon_{it}\left(  \lambda_{T}\right)  $ and by Cauchy-Schwarz
inequality,%
\[
\frac{1}{T}\sum_{t=1}^{T}\boldsymbol{\kappa}_{t,n}\left(  \lambda_{T}\right)
\xi_{t,n}\left(  \lambda_{T}\right)  \leq\left(  \frac{1}{T}\sum_{t=1}%
^{T}\left\Vert \boldsymbol{\kappa}_{t,n}\left(  \lambda_{T}\right)
\right\Vert ^{2}\right)  ^{1/2}\left(  \frac{1}{T}\sum_{t=1}^{T}\xi_{t,n}%
^{2}\left(  \lambda_{T}\right)  \right)  ^{1/2}%
\]
where given (\ref{ave:rho}), $\sup_{i}a_{i,n}^{2}<C$ and $\sup_{i}\sigma
_{i}^{-2}<C$, it further follows that
\begin{align}
E\left(  \xi_{t,n}^{2}\left(  \lambda_{T}\right)  \right)   &  =\frac{1}%
{n}\sum_{i=1}^{n}\sum_{j=1}^{n}a_{i,n}a_{j,n}E\left(  \varepsilon_{it}\left(
\lambda_{T}\right)  \varepsilon_{jt}\left(  \lambda_{T}\right)  \right)
\nonumber\\
&  <\left(  \sup_{i}a_{i,n}^{2}\right)  \frac{1}{n}\sum_{i=1}^{n}\sum
_{j=1}^{n}\left\vert E\left(  \varepsilon_{it}\left(  \lambda_{T}\right)
\varepsilon_{jt}\left(  \lambda_{T}\right)  \right)  \right\vert <C.
\label{E:ksi2}%
\end{align}
Hence, $T^{-1}\sum_{t=1}^{T}\boldsymbol{\kappa}_{t,n}\left(  \lambda
_{T}\right)  \xi_{t,n}\left(  \lambda_{T}\right)  =O_{p}\left(  1\right)  $,
and again using (\ref{phin_p}) it follows that
\begin{equation}
g_{5,nT}\left(  \lambda_{T}\right)  =\sqrt{T}\left(  \boldsymbol{\varphi}%
_{nT}-\boldsymbol{\varphi}_{n}\right)  ^{\prime}\left(  \frac{1}{T}\sum
_{t=1}^{T}\boldsymbol{\kappa}_{t,n}\left(  \lambda_{T}\right)  \xi
_{t,n}\left(  \lambda_{T}\right)  \right)  =o_{p}(1). \label{g5_p}%
\end{equation}

\end{proof}

\begin{lemma}
\label{lm:wnT}\textbf{ }Suppose the data are generated by the latent factor
model given by (\ref{mod2}) and (\ref{uit:p}). Further denote
\[
w_{nT}\left(  \lambda_{T}\right)  =\frac{T^{-1/2}\sum_{t=1}^{T}\xi
_{t,n}\left(  \lambda_{T}\right)  \upsilon_{t,nT}\left(  \lambda_{T}\right)
}{T^{-1}\sum_{t=1}^{T}\xi_{t,n}^{2}\left(  \lambda_{T}\right)  },
\]
where%
\begin{align*}
\xi_{t,n}\left(  \lambda_{T}\right)   &  =\frac{1}{\sqrt{n}}\sum_{i=1}%
^{n}a_{i,n}\varepsilon_{it}\left(  \lambda_{T}\right)  ,a_{i,n}=1-\sigma
_{i}\boldsymbol{\varphi}_{n}^{\prime}\boldsymbol{\gamma}_{i},\\
\upsilon_{t,nT}\left(  \lambda_{T}\right)   &  =\frac{1}{\sqrt{n}}\sum
_{i=1}^{n}\zeta_{it}\varepsilon_{it}\left(  \lambda_{T}\right)  ,\zeta
_{it}=\frac{1}{\left(  T^{-1}\boldsymbol{\varepsilon}_{i\circ}^{\prime
}\mathbf{M}_{F}\boldsymbol{\varepsilon}_{i\circ}\right)  ^{1/2}}-1,
\end{align*}
with $\boldsymbol{\varphi}_{n}=n^{-1}\sum_{i=1}^{n}\boldsymbol{\gamma}%
_{i}/\sigma_{i}$, $\boldsymbol{\varepsilon}_{i\circ}=\left(  \varepsilon
_{i1},\varepsilon_{i2},\ldots,\varepsilon_{iT}\right)  ^{\prime}$,
$\mathbf{M}_{F}=\mathbf{I}_{T}-\mathbf{F}\left(  \mathbf{F}^{\prime}%
\mathbf{F}\right)  ^{-1}\mathbf{F}^{\prime}$. Suppose that Assumptions
\ref{ass:factors}-\ref{ass:ws} hold and $(n,T)\rightarrow\infty$, such that
$n/T\rightarrow\kappa,$ and $0<\kappa<\infty$. Then $w_{nT}\left(  \lambda
_{T}\right)  =o_{p}\left(  1\right)  $.
\end{lemma}

\begin{proof}
Consider first the denominator of $w_{nT}\left(  \lambda_{T}\right)  $ and
using (\ref{E:ksi2}) note that
\[
E\left(  \frac{1}{T}\sum_{t=1}^{T}\xi_{t,n}^{2}\left(  \lambda_{T}\right)
\right)  =\frac{1}{T}\sum_{t=1}^{T}E\left(  \xi_{t,n}^{2}\left(  \lambda
_{T}\right)  \right)  <C.
\]
Hence, by Markov inequality it is obvious that the denominator of
$w_{nT}\left(  \lambda_{T}\right)  $ is $O_{p}\left(  1\right)  $.\textbf{
}Consider now the numerator of $w_{nT}\left(  \lambda_{T}\right)  $, which is
denoted as $r_{nT}\left(  \lambda_{T}\right)  $. For simplicity, let%
\[
\xi_{t,n}=\frac{1}{\sqrt{n}}\sum_{i=1}^{n}a_{i,n}\varepsilon_{it}%
,\upsilon_{t,nT}=\frac{1}{\sqrt{n}}\sum_{i=1}^{n}\zeta_{it}\varepsilon_{it},
\]
then%
\begin{align*}
r_{nT}\left(  \lambda_{T}\right)   &  =\frac{1}{\sqrt{T}}\sum_{t=1}^{T}%
\xi_{t,n}\upsilon_{t,nT}+\frac{1}{\sqrt{T}}\sum_{t=1}^{T}\left[  \xi
_{t,n}\left(  \lambda_{T}\right)  -\xi_{t,n}\right]  \upsilon_{t,nT}\left(
\lambda_{T}\right) \\
&  +\frac{1}{\sqrt{T}}\sum_{t=1}^{T}\xi_{t,n}\left(  \lambda_{T}\right)
\left[  \upsilon_{t,nT}\left(  \lambda_{T}\right)  -\upsilon_{t,nT}\right]
-\frac{1}{\sqrt{T}}\sum_{t=1}^{T}\left[  \xi_{t,n}\left(  \lambda_{T}\right)
-\xi_{t,n}\right]  \left[  \upsilon_{t,nT}\left(  \lambda_{T}\right)
-\upsilon_{t,nT}\right] \\
&  =r_{1,nT}+\sum_{j=2}^{4}r_{j,nT}\left(  \lambda_{T}\right)  .
\end{align*}
To bound $r_{1,nT}$, we firstly note $\xi_{t,n}\upsilon_{t,nT}=\frac{1}{n}%
\sum_{i=1}^{n}\sum_{j=1}^{n}a_{in}\varepsilon_{it}\zeta_{jt},$ where for
$i\neq j$ $\varepsilon_{it}$ and $\zeta_{jt}$ are distributed independently by
parts (a) and (c) of Assumption \ref{ass:errors}. Therefore,
\[
E\left(  \xi_{t,n}\upsilon_{t,nT}\right)  =\frac{1}{n}\sum_{i=1}^{n}%
a_{in}E\left[  \varepsilon_{it}^{2}\left(  \frac{1}{\left(  T^{-1}%
\boldsymbol{\varepsilon}_{i\circ}^{\prime}\mathbf{M}_{F}%
\boldsymbol{\varepsilon}_{i\circ}\right)  ^{1/2}}-1\right)  \right]  .
\]
Furthermore, since $a_{in}$ is bounded, and by result (\ref{E7})
\[
E\left[  \varepsilon_{it}^{2}\left(  \frac{1}{\left(  T^{-1}%
\boldsymbol{\varepsilon}_{i\circ}^{\prime}\mathbf{M}_{F}%
\boldsymbol{\varepsilon}_{i\circ}\right)  ^{1/2}}-1\right)  \right]  =O\left(
\frac{1}{T}\right)  .
\]
Then $E\left(  \xi_{t,n}\upsilon_{t,nT}\right)  =O\left(  T^{-1}\right)  $,
and it follows that
\begin{align}
E\left(  r_{1,nT}\right)   &  =\frac{1}{\sqrt{T}}\sum_{t=1}^{T}E\left(
\xi_{t,n}\upsilon_{t,nT}\right)  =O\left(  \frac{1}{\sqrt{T}}\right)
,\label{EwnT}\\
Var\left(  r_{1,nT}\right)   &  =E\left(  r_{1,nT}^{2}\right)  -\left[
E\left(  r_{1,nT}\right)  \right]  ^{2}=\frac{1}{T}\sum_{t=1}^{T}%
\sum_{t^{\prime}=1}^{T}E\left(  \xi_{t,n}\upsilon_{t,nT}\xi_{t^{\prime}%
,n}v_{t^{\prime},nT}\right)  +O\left(  T^{-1}\right)  . \label{Varr_nT}%
\end{align}
Consider now the first term of $Var\left(  r_{1,nT}\right)  $, and using
(\ref{zeta}) we have
\begin{align*}
&  E\left(  \xi_{t,n}\upsilon_{t,nT}\xi_{t^{\prime},n}v_{t^{\prime},nT}\right)
\\
&  \equiv E\left[  \left(  \frac{1}{n}\sum_{i=1}^{n}\sum_{j=1}^{n}a_{in}%
a_{jn}\varepsilon_{it}\varepsilon_{jt^{\prime}}\right)  \left(  \frac{1}%
{n}\sum_{r=1}^{n}\sum_{s=1}^{n}\zeta_{rt}\zeta_{st^{\prime}}\right)  \right]
\\
&  =\frac{1}{n^{2}}\sum_{i=1}^{n}\sum_{j=1}^{n}\sum_{r=1}^{n}\sum_{s=1}%
^{n}a_{in}a_{jn}E\left\{  \left(  \frac{1}{\left(  T^{-1}%
\boldsymbol{\varepsilon}_{r\circ}^{\prime}\mathbf{M}_{F}%
\boldsymbol{\varepsilon}_{r\circ}\right)  ^{1/2}}-1\right)  \left(  \frac
{1}{\left(  T^{-1}\boldsymbol{\varepsilon}_{s\circ}^{\prime}\mathbf{M}%
_{F}\boldsymbol{\varepsilon}_{s\circ}\right)  ^{1/2}}-1\right)  \varepsilon
_{it}\varepsilon_{jt^{\prime}}\varepsilon_{rt}\varepsilon_{st^{\prime}%
}\right\}  .
\end{align*}
Since by part (a) of Assumption \ref{ass:errors}, $\varepsilon_{it}^{\prime}s$
are cross-sectionally independent and after some algebra we have%
\begin{equation}
\frac{1}{T}\sum_{t=1}^{T}\sum_{t^{\prime}=1}^{T}E\left(  \xi_{t,n}%
\upsilon_{t,nT}\xi_{t^{\prime},n}v_{t^{\prime},nT}\right)  =\sum_{l=1}%
^{6}D_{l,nT}, \label{egzi*v}%
\end{equation}
where%
\[
D_{1,nT}=\frac{1}{Tn^{2}}\sum_{t=1}^{T}\sum_{t^{\prime}=1}^{T}\sum_{i=1}%
^{n}a_{in}^{2}E\left[  \varepsilon_{it}^{2}\varepsilon_{it^{\prime}}%
^{2}\left(  \frac{1}{\left(  T^{-1}\boldsymbol{\varepsilon}_{i\circ}^{\prime
}\mathbf{M}_{F}\boldsymbol{\varepsilon}_{i\circ}\right)  ^{1/2}}-1\right)
^{2}\right]  ,
\]%
\begin{align*}
D_{2,nT}  &  =\frac{1}{Tn^{2}}\sum_{t=1}^{T}\sum_{t^{\prime}=1}^{T}\sum
_{i=1}^{n}\sum_{r\neq i}^{n}a_{in}^{2}E\left[  \varepsilon_{it}\varepsilon
_{it^{\prime}}^{2}\left(  \frac{1}{\left(  T^{-1}\boldsymbol{\varepsilon
}_{i\circ}^{\prime}\mathbf{M}_{F}\boldsymbol{\varepsilon}_{i\circ}\right)
^{1/2}}-1\right)  \right]  \times\\
&  E\left[  \left(  \frac{1}{\left(  T^{-1}\boldsymbol{\varepsilon}_{r\circ
}^{\prime}\mathbf{M}_{F}\boldsymbol{\varepsilon}_{r\circ}\right)  ^{1/2}%
}-1\right)  \varepsilon_{rt}\right]  ,
\end{align*}%
\begin{align*}
D_{3,nT}  &  =\frac{1}{Tn^{2}}\sum_{t=1}^{T}\sum_{t^{\prime}=1}^{T}\sum
_{i=1}^{n}\sum_{s\neq i}^{n}a_{in}^{2}E\left[  \varepsilon_{it}^{2}%
\varepsilon_{it^{\prime}}\left(  \frac{1}{\left(  T^{-1}%
\boldsymbol{\varepsilon}_{i\circ}^{\prime}\mathbf{M}_{F}%
\boldsymbol{\varepsilon}_{i\circ}\right)  ^{1/2}}-1\right)  \right]  \times\\
&  E\left[  \left(  \frac{1}{\left(  T^{-1}\boldsymbol{\varepsilon}_{s\circ
}^{\prime}\mathbf{M}_{F}\boldsymbol{\varepsilon}_{s\circ}\right)  ^{1/2}%
}-1\right)  \varepsilon_{st^{\prime}}\right]  ,
\end{align*}%
\[
D_{4,nT}=\frac{1}{Tn^{2}}\sum_{t=1}^{T}\sum_{t^{\prime}=1}^{T}\sum_{i=1}%
^{n}\sum_{r\neq i}^{n}a_{in}^{2}E\left(  \varepsilon_{it}\varepsilon
_{it^{\prime}}\right)  E\left[  \left(  \frac{1}{\left(  T^{-1}%
\boldsymbol{\varepsilon}_{r\circ}^{\prime}\mathbf{M}_{F}%
\boldsymbol{\varepsilon}_{r\circ}\right)  ^{1/2}}-1\right)  ^{2}%
\varepsilon_{rt}\varepsilon_{rt^{\prime}}\right]  ,
\]%
\begin{align*}
D_{5,nT}  &  =\frac{1}{Tn^{2}}\sum_{t=1}^{T}\sum_{t^{\prime}=1}^{T}\sum
_{i=1}^{n}\sum_{j\neq i}^{n}a_{in}a_{jn}E\left[  \varepsilon_{it}^{2}\left(
\frac{1}{\left(  T^{-1}\boldsymbol{\varepsilon}_{i\circ}^{\prime}%
\mathbf{M}_{F}\boldsymbol{\varepsilon}_{i\circ}\right)  ^{1/2}}-1\right)
\right]  \times\\
&  E\left[  \varepsilon_{jt^{\prime}}^{2}\left(  \frac{1}{\left(
T^{-1}\boldsymbol{\varepsilon}_{j\circ}^{\prime}\mathbf{M}_{F}%
\boldsymbol{\varepsilon}_{j\circ}\right)  ^{1/2}}-1\right)  \right]  ,
\end{align*}
and%
\begin{align*}
D_{6,nT}  &  =\frac{1}{Tn^{2}}\sum_{t=1}^{T}\sum_{t^{\prime}=1}^{T}\sum
_{i=1}^{n}\sum_{j\neq i}^{n}a_{in}a_{jn}E\left[  \varepsilon_{it}%
\varepsilon_{it^{\prime}}\left(  \frac{1}{\left(  T^{-1}%
\boldsymbol{\varepsilon}_{i\circ}^{\prime}\mathbf{M}_{F}%
\boldsymbol{\varepsilon}_{i\circ}\right)  ^{1/2}}-1\right)  \right]  \times\\
&  E\left[  \varepsilon_{jt}\varepsilon_{jt^{\prime}}\left(  \frac{1}{\left(
T^{-1}\boldsymbol{\varepsilon}_{j\circ}^{\prime}\mathbf{M}_{F}%
\boldsymbol{\varepsilon}_{j\circ}\right)  ^{1/2}}-1\right)  \right]  .
\end{align*}
For $D_{1,nT}$ we have
\begin{align*}
D_{1,nT}  &  =\frac{1}{n^{2}}\sum_{i=1}^{n}a_{in}^{2}\frac{1}{T}\sum_{t=1}%
^{T}E\left[  \varepsilon_{it}^{4}\left(  \frac{1}{\left(  T^{-1}%
\boldsymbol{\varepsilon}_{i\circ}^{\prime}\mathbf{M}_{F}%
\boldsymbol{\varepsilon}_{i\circ}\right)  ^{1/2}}-1\right)  ^{2}\right] \\
+  &  \frac{1}{n^{2}}\sum_{i=1}^{n}a_{in}^{2}\frac{1}{T}\sum_{t^{\prime}\neq
t}^{T}E\left[  \varepsilon_{it}^{2}\varepsilon_{it^{\prime}}^{2}\left(
\frac{1}{\left(  T^{-1}\boldsymbol{\varepsilon}_{i\circ}^{\prime}%
\mathbf{M}_{F}\boldsymbol{\varepsilon}_{i\circ}\right)  ^{1/2}}-1\right)
^{2}\right]  =D_{1,1,nT}+D_{1,2,nT}.
\end{align*}
It is clear that $D_{1,1,nT}=O(n^{-1})$. Also by Cauchy-Schwarz inequality%
\[
E\left[  \varepsilon_{it}^{2}\varepsilon_{it^{\prime}}^{2}\left(  \frac
{1}{\left(  T^{-1}\boldsymbol{\varepsilon}_{i\circ}^{\prime}\mathbf{M}%
_{F}\boldsymbol{\varepsilon}_{i\circ}\right)  ^{1/2}}-1\right)  ^{2}\right]
\leq\left[  E\left(  \varepsilon_{it}^{4}\varepsilon_{it^{\prime}}^{4}\right)
\right]  ^{1/2}\left[  E\left(  \frac{1}{\left(  T^{-1}\boldsymbol{\varepsilon
}_{i\circ}^{\prime}\mathbf{M}_{F}\boldsymbol{\varepsilon}_{i\circ}\right)
^{1/2}}-1\right)  ^{4}\right]  ^{1/2},
\]
where by part (a) of Assumption \ref{ass:errors} we have $E\left(
\varepsilon_{it}^{4}\varepsilon_{it^{\prime}}^{4}\right)  =E\left(
\varepsilon_{it}^{4}\right)  E\left(  \varepsilon_{it^{\prime}}^{4}\right)
=O\left(  1\right)  $. Further, in view of (\ref{Es})%
\begin{align*}
E\left(  \frac{1}{\left(  T^{-1}\boldsymbol{\varepsilon}_{i\circ}^{\prime
}\mathbf{M}_{F}\boldsymbol{\varepsilon}_{i\circ}\right)  ^{1/2}}-1\right)
^{4}  &  =E\left[  \frac{1}{\left(  T^{-1}\boldsymbol{\varepsilon}_{i\circ
}^{\prime}\mathbf{M}_{F}\boldsymbol{\varepsilon}_{i\circ}\right)  ^{2}%
}\right]  -4E\left[  \frac{1}{\left(  T^{-1}\boldsymbol{\varepsilon}_{i\circ
}^{\prime}\mathbf{M}_{F}\boldsymbol{\varepsilon}_{i\circ}\right)  ^{3/2}%
}\right]  +6E\left[  \frac{1}{T^{-1}\boldsymbol{\varepsilon}_{i\circ}^{\prime
}\mathbf{M}_{F}\boldsymbol{\varepsilon}_{i\circ}}\right] \\
&  -4E\left[  \frac{1}{\left(  T^{-1}\boldsymbol{\varepsilon}_{i\circ}%
^{\prime}\mathbf{M}_{F}\boldsymbol{\varepsilon}_{i\circ}\right)  ^{1/2}%
}\right]  +1=O\left(  \frac{1}{T}\right)  .
\end{align*}
Using this result it follows that $D_{1,2,nT}=$ $O\left(  \sqrt{T}%
n^{-1}\right)  ,$ and overall we have
\begin{equation}
D_{1,nT}=O\left(  \sqrt{T}n^{-1}\right)  . \label{egzi*v_1}%
\end{equation}
Consider now $D_{2,nT}$ and note
\begin{align*}
\left\vert D_{2,nT}\right\vert  &  \leq\frac{1}{Tn^{2}}\sum_{t=1}^{T}%
\sum_{t^{\prime}=1}^{T}\sum_{i=1}^{n}\sum_{r\neq i}^{n}a_{in}^{2}\left\vert
E\left[  \varepsilon_{it}\varepsilon_{it^{\prime}}^{2}\left(  \frac{1}{\left(
T^{-1}\boldsymbol{\varepsilon}_{i\circ}^{\prime}\mathbf{M}_{F}%
\boldsymbol{\varepsilon}_{i\circ}\right)  ^{1/2}}-1\right)  \right]
\right\vert \times\\
&  \left\vert E\left[  \left(  \frac{1}{\left(  T^{-1}\boldsymbol{\varepsilon
}_{r\circ}^{\prime}\mathbf{M}_{F}\boldsymbol{\varepsilon}_{r\circ}\right)
^{1/2}}-1\right)  \varepsilon_{rt}\right]  \right\vert .
\end{align*}
By (\ref{E9}) we have%
\begin{equation}
E\left[  \left(  \frac{1}{\left(  T^{-1}\boldsymbol{\varepsilon}_{r\circ
}^{\prime}\mathbf{M}_{F}\boldsymbol{\varepsilon}_{r\circ}\right)  ^{1/2}%
}-1\right)  \varepsilon_{rt}\right]  =E\left[  \frac{\varepsilon_{rt}}{\left(
T^{-1}\boldsymbol{\varepsilon}_{r\circ}^{\prime}\mathbf{M}_{F}%
\boldsymbol{\varepsilon}_{r\circ}\right)  ^{1/2}}\right]  =O\left(  \frac
{1}{T}\right)  . \label{Eegzi}%
\end{equation}
Under Assumption \ref{ass:errors} and given results (\ref{E6}) and (\ref{E7}),
it follows%
\begin{align}
E\left[  \varepsilon_{it}^{2}\left(  \frac{1}{\left(  T^{-1}%
\boldsymbol{\varepsilon}_{i\circ}^{\prime}\mathbf{M}_{F}%
\boldsymbol{\varepsilon}_{i\circ}\right)  ^{1/2}}-1\right)  ^{2}\right]   &
=E\left[  \frac{\varepsilon_{it}^{2}}{\left(  T^{-1}\boldsymbol{\varepsilon
}_{i\circ}^{\prime}\mathbf{M}_{F}\boldsymbol{\varepsilon}_{i\circ}\right)
}\right]  -2E\left[  \frac{\varepsilon_{it}^{2}}{\left(  T^{-1}%
\boldsymbol{\varepsilon}_{i\circ}^{\prime}\mathbf{M}_{F}%
\boldsymbol{\varepsilon}_{i\circ}\right)  ^{1/2}}\right]  +1\nonumber\\
&  =\left[  E\left[  \frac{\varepsilon_{it}^{2}}{\left(  T^{-1}%
\boldsymbol{\varepsilon}_{i\circ}^{\prime}\mathbf{M}_{F}%
\boldsymbol{\varepsilon}_{i\circ}\right)  }\right]  -1\right]  -2\left[
E\left[  \frac{\varepsilon_{it}^{2}}{\left(  T^{-1}\boldsymbol{\varepsilon
}_{i\circ}^{\prime}\mathbf{M}_{F}\boldsymbol{\varepsilon}_{i\circ}\right)
^{1/2}}\right]  -1\right] \nonumber\\
&  =O\left(  \frac{1}{T}\right)  . \label{Eegzi^2}%
\end{align}
Further, by Cauchy-Schwarz inequality,%
\[
\left\vert E\left[  \varepsilon_{it}\varepsilon_{it^{\prime}}^{2}\left(
\frac{1}{\left(  T^{-1}\boldsymbol{\varepsilon}_{i\circ}^{\prime}%
\mathbf{M}_{F}\boldsymbol{\varepsilon}_{i\circ}\right)  ^{1/2}}-1\right)
\right]  \right\vert \leq\left[  E\left(  \varepsilon_{it^{\prime}}%
^{4}\right)  \right]  ^{1/2}\left[  E\left[  \varepsilon_{it}^{2}\left(
\frac{1}{\left(  T^{-1}\boldsymbol{\varepsilon}_{i\circ}^{\prime}%
\mathbf{M}_{F}\boldsymbol{\varepsilon}_{i\circ}\right)  ^{1/2}}-1\right)
^{2}\right]  \right]  ^{1/2},
\]
which is $O\left(  T^{-1/2}\right)  $ based on $E\left(  \varepsilon
_{it^{\prime}}^{4}\right)  =O\left(  1\right)  $ and (\ref{Eegzi^2}). With
this result and (\ref{Eegzi}), it now follows that
\begin{equation}
D_{2,nT}=O\left(  T^{-1/2}\right)  . \label{egzi*v_2}%
\end{equation}
Similarly, it follows that
\begin{equation}
D_{3,nT}=O\left(  T^{-1/2}\right)  . \label{egzi*v_3}%
\end{equation}
Consider the fourth term of (\ref{egzi*v}), and using the serial independence
of $\varepsilon_{it}$,
\begin{align}
D_{4,nT}  &  =\frac{1}{Tn^{2}}\sum_{t=1}^{T}\sum_{t^{\prime}=1}^{T}\sum
_{i=1}^{n}\sum_{r\neq i}^{n}a_{in}^{2}E\left(  \varepsilon_{it}\varepsilon
_{it^{\prime}}\right)  E\left[  \left(  \frac{1}{\left(  T^{-1}%
\boldsymbol{\varepsilon}_{r\circ}^{\prime}\mathbf{M}_{F}%
\boldsymbol{\varepsilon}_{r\circ}\right)  ^{1/2}}-1\right)  ^{2}%
\varepsilon_{rt}\varepsilon_{rt^{\prime}}\right] \nonumber\\
&  =\frac{1}{Tn^{2}}\sum_{t=1}^{T}\sum_{i=1}^{n}\sum_{r\neq i}^{n}a_{in}%
^{2}E\left(  \varepsilon_{it}^{2}\right)  E\left[  \left(  \frac{1}{\left(
T^{-1}\boldsymbol{\varepsilon}_{r\circ}^{\prime}\mathbf{M}_{F}%
\boldsymbol{\varepsilon}_{r\circ}\right)  ^{1/2}}-1\right)  ^{2}%
\varepsilon_{rt}^{2}\right] \nonumber\\
&  =O\left(  \frac{1}{T}\right)  \label{egzi*v_4}%
\end{align}
in which the final equation holds given $E\left(  \varepsilon_{it}^{2}\right)
=1$ by part (a) of Assumption \ref{ass:errors} and (\ref{Eegzi^2}). To derive
the order of $D_{5,nT}$, using result (\ref{E7}) we have%
\begin{equation}
E\left[  \frac{\varepsilon_{it}^{2}}{\left(  T^{-1}\boldsymbol{\varepsilon
}_{i\circ}^{\prime}\mathbf{M}_{F}\boldsymbol{\varepsilon}_{i\circ}\right)
^{1/2}}-\varepsilon_{it}^{2}\right]  =O\left(  \frac{1}{T}\right)  .
\label{e*egzi}%
\end{equation}
Also since $\varepsilon_{it}^{2}\left(  \frac{1}{\left(  T^{-1}%
\boldsymbol{\varepsilon}_{i\circ}^{\prime}\mathbf{M}_{F}%
\boldsymbol{\varepsilon}_{i\circ}\right)  ^{1/2}}-1\right)  $ is independently
distributed across $i$, then \textbf{ }%
\begin{align}
D_{5,nT}  &  =\frac{1}{Tn^{2}}\sum_{t=1}^{T}\sum_{t^{\prime}=1}^{T}\sum
_{i=1}^{n}\sum_{j\neq i}^{n}a_{in}a_{jn}E\left[  \varepsilon_{it}^{2}\left(
\frac{1}{\left(  T^{-1}\boldsymbol{\varepsilon}_{i\circ}^{\prime}%
\mathbf{M}_{F}\boldsymbol{\varepsilon}_{i\circ}\right)  ^{1/2}}-1\right)
\right]  E\left[  \varepsilon_{jt^{\prime}}^{2}\left(  \frac{1}{\left(
T^{-1}\boldsymbol{\varepsilon}_{j\circ}^{\prime}\mathbf{M}_{F}%
\boldsymbol{\varepsilon}_{j\circ}\right)  ^{1/2}}-1\right)  \right]
\nonumber\\
&  =O\left(  \frac{1}{T}\right)  . \label{egzi*v_5}%
\end{align}
To establish the order of $D_{6,nT}$, using (\ref{E8}) and the serial
independence of $\varepsilon_{it}$ we have%
\[
E\left[  \frac{\varepsilon_{it}\varepsilon_{it^{\prime}}}{\left(
T^{-1}\boldsymbol{\varepsilon}_{i\circ}^{\prime}\mathbf{M}_{F}%
\boldsymbol{\varepsilon}_{i\circ}\right)  ^{1/2}}-\varepsilon_{it}%
\varepsilon_{it^{\prime}}\right]  =O\left(  \frac{1}{T}\right)  \text{, for
}t\neq t^{\prime}.
\]
Using this result and (\ref{e*egzi}) we have%
\begin{align}
D_{6,nT}  &  =\frac{1}{Tn^{2}}\sum_{t=1}^{T}\sum_{i=1}^{n}\sum_{j\neq i}%
^{n}a_{in}a_{jn}E\left[  \varepsilon_{it}^{2}\left(  \frac{1}{\left(
T^{-1}\boldsymbol{\varepsilon}_{i\circ}^{\prime}\mathbf{M}_{F}%
\boldsymbol{\varepsilon}_{i\circ}\right)  ^{1/2}}-1\right)  \right]  E\left[
\varepsilon_{jt}^{2}\left(  \frac{1}{\left(  T^{-1}\boldsymbol{\varepsilon
}_{j\circ}^{\prime}\mathbf{M}_{F}\boldsymbol{\varepsilon}_{j\circ}\right)
^{1/2}}-1\right)  \right] \nonumber\\
&  +\frac{1}{Tn^{2}}\sum_{t=1}^{T}\sum_{t^{\prime}\neq t}^{T}\sum_{i=1}%
^{n}\sum_{j\neq i}^{n}a_{in}a_{jn}E\left[  \varepsilon_{it}\varepsilon
_{it^{\prime}}\left(  \frac{1}{\left(  T^{-1}\boldsymbol{\varepsilon}_{i\circ
}^{\prime}\mathbf{M}_{F}\boldsymbol{\varepsilon}_{i\circ}\right)  ^{1/2}%
}-1\right)  \right]  \times\nonumber\\
&  E\left[  \varepsilon_{jt}\varepsilon_{jt^{\prime}}\left(  \frac{1}{\left(
T^{-1}\boldsymbol{\varepsilon}_{j\circ}^{\prime}\mathbf{M}_{F}%
\boldsymbol{\varepsilon}_{j\circ}\right)  ^{1/2}}-1\right)  \right]
\nonumber\\
&  =O\left(  \frac{1}{T}\right)  . \label{egzi*v_6}%
\end{align}
Hence, using (\ref{egzi*v_1}), (\ref{egzi*v_2}), (\ref{egzi*v_3}),
(\ref{egzi*v_4}), (\ref{egzi*v_5}) and (\ref{egzi*v_6}) in (\ref{egzi*v}) and
then in (\ref{Varr_nT}), we have $Var\left(  r_{1,nT}\right)  =O\left(
\frac{1}{\sqrt{T}}\right)  +O\left(  \frac{\sqrt{T}}{n}\right)  ,$ and by
Chebyshev's inequality $r_{1,nT}=O\left(  \frac{1}{\sqrt{T}}\right)
+O_{p}\left(  \frac{\sqrt{T}}{n}\right)  $ . For the second term of
$r_{nT}\left(  \lambda_{T}\right)  $, as $\xi_{t,n}\left(  \lambda_{T}\right)
=\frac{1}{\sqrt{n}}\sum_{i=1}^{n}a_{i,n}\varepsilon_{it}\left(  \lambda
_{T}\right)  $ where $\varepsilon_{it}\left(  \lambda_{T}\right)
=\varepsilon_{it}+\lambda_{T}b_{it}$ with $\lambda_{T}=c_{\lambda}T^{-1/2}$
and $b_{it}=\sum_{s=1}^{n}w_{is}\varepsilon_{st}$, then
\begin{align*}
\left\vert r_{2,nT}\left(  \lambda_{T}\right)  \right\vert  &  \leq\frac
{1}{\sqrt{T}}\sum_{t=1}^{T}\left\vert \left[  \xi_{t,n}\left(  \lambda
_{T}\right)  -\xi_{t,n}\right]  \upsilon_{t,nT}\left(  \lambda_{T}\right)
\right\vert \\
&  \leq\left(  \frac{1}{\sqrt{T}}\sum_{t=1}^{T}\left[  \xi_{t,n}\left(
\lambda_{T}\right)  -\xi_{t,n}\right]  \right)  \left(  \sup_{t}\left\vert
\upsilon_{t,nT}\left(  \lambda_{T}\right)  \right\vert \right) \\
&  \leq\left(  \frac{1}{T}\sum_{t=1}^{T}\left\vert \frac{c_{\lambda}}{\sqrt
{n}}\sum_{i=1}^{n}a_{i,n}b_{it}\right\vert \right)  \left(  \sup_{t}\left\vert
\upsilon_{t,nT}\left(  \lambda_{T}\right)  \right\vert \right)  .
\end{align*}
Since $\sup_{t}\left\vert \upsilon_{t,nT}\left(  \lambda_{T}\right)
\right\vert =O_{p}\left(  \frac{\ln\left(  nT\right)  }{\sqrt{T}}\right)  $ as
shown by (\ref{sup:v}), while $c_{\lambda}n^{-1/2}\sum_{i=1}^{n}a_{i,n}b_{it}$
is serially independent and
\begin{align}
E\left(  \frac{c_{\lambda}}{\sqrt{n}}\sum_{i=1}^{n}a_{i,n}b_{it}\right)  ^{2}
&  =\frac{c_{\lambda}^{2}}{n}\sum_{i=1}^{n}\sum_{j=1}^{n}a_{i,n}%
a_{j,n}\mathbf{w}_{i0}^{\prime}E\left(  \boldsymbol{\varepsilon}_{\circ
t}\boldsymbol{\varepsilon}_{\circ t}^{\prime}\right)  \mathbf{w}%
_{j0}\nonumber\\
&  =\frac{1}{n}\sum_{i=1}^{n}\sum_{j=1}^{n}a_{i,n}a_{j,n}\mathbf{w}%
_{i0}^{\prime}\mathbf{w}_{j0}=\frac{1}{n}\sum_{i=1}^{n}\sum_{j=1}^{n}%
a_{i,n}a_{j,n}\sum_{s=1}^{n}w_{is}w_{js}\nonumber\\
&  =\frac{1}{n}\sum_{s=1}^{n}\sum_{i=1}^{n}\sum_{j=1}^{n}a_{i,n}w_{is}%
a_{j,n}w_{js}=\frac{1}{n}\sum_{s=1}^{n}\left(  \sum_{i=1}^{n}a_{i,n}%
w_{is}\right)  ^{2}<C, \label{E:ab}%
\end{align}
then it follows $r_{2,nT}\left(  \lambda_{T}\right)  =O_{p}\left(  \frac
{\ln\left(  nT\right)  }{\sqrt{T}}\right)  =o_{p}\left(  1\right)  .$ Now
consider the third term of $r_{nT}\left(  \lambda_{T}\right)  $ and note%
\begin{align*}
\left\vert r_{3,nT}\left(  \lambda_{T}\right)  \right\vert  &  =\frac{1}%
{\sqrt{T}}\sum_{t=1}^{T}\xi_{t,n}\left(  \lambda_{T}\right)  \left[
\upsilon_{t,nT}\left(  \lambda_{T}\right)  -\upsilon_{t,nT}\right] \\
&  =\frac{1}{T}\sum_{t=1}^{T}\xi_{t,n}\left(  \lambda_{T}\right)  \left[
\frac{c_{\lambda}}{\sqrt{n}}\sum_{i=1}^{n}\left(  \frac{1}{\left(
T^{-1}\boldsymbol{\varepsilon}_{i\circ}^{\prime}\mathbf{M}_{F}%
\boldsymbol{\varepsilon}_{i\circ}\right)  ^{1/2}}-1\right)  b_{it}\right] \\
&  \leq\left(  \frac{1}{T}\sum_{t=1}^{T}\left\vert \xi_{t,n}\left(
\lambda_{T}\right)  \right\vert \right)  \sup_{t}\left\vert \frac{c_{\lambda}%
}{\sqrt{n}}\sum_{i=1}^{n}\left(  \frac{1}{\left(  T^{-1}%
\boldsymbol{\varepsilon}_{i\circ}^{\prime}\mathbf{M}_{F}%
\boldsymbol{\varepsilon}_{i\circ}\right)  ^{1/2}}-1\right)  b_{it}\right\vert
.
\end{align*}
Clearly, $\xi_{t,n}\left(  \lambda_{T}\right)  $ is serially independent and
we have shown $\xi_{t,n}\left(  \lambda_{T}\right)  =O_{p}\left(  1\right)  $
by (\ref{E:ksi2}). Since by part (a) of Assumption \ref{ass:errors} and Lemma
\ref{lm:sup(ub)}, $b_{it}$ is sub-exponential, then $\sup_{i,t}\left\vert
b_{it}\right\vert =O_{p}\left[  \ln\left(  nT\right)  \right]  $ and using
(\ref{ave:zMz1}) it follows that
\begin{align}
\sup_{t}\left\vert \frac{c_{\lambda}}{\sqrt{n}}\sum_{i=1}^{n}\left(  \frac
{1}{\left(  T^{-1}\boldsymbol{\varepsilon}_{i\circ}^{\prime}\mathbf{M}%
_{F}\boldsymbol{\varepsilon}_{i\circ}\right)  ^{1/2}}-1\right)  b_{it}%
\right\vert  &  \leq\left\vert \left(  \sup_{i,t}\left\vert b_{it}\right\vert
\right)  \left[  \frac{c_{\lambda}}{\sqrt{n}}\sum_{i=1}^{n}\left\vert \frac
{1}{\left(  T^{-1}\boldsymbol{\varepsilon}_{i\circ}^{\prime}\mathbf{M}%
_{F}\boldsymbol{\varepsilon}_{i\circ}\right)  ^{1/2}}-1\right\vert \right]
\right\vert \nonumber\\
&  =O_{p}\left(  \frac{\ln\left(  nT\right)  }{\sqrt{T}}\right)  .
\label{diff:vt,nT}%
\end{align}
Overall we have $r_{3,nT}\left(  \lambda_{T}\right)  =O_{p}\left(  \frac
{\ln\left(  nT\right)  }{\sqrt{T}}\right)  =o_{p}\left(  1\right)  $. Now
consider $r_{4,nT}\left(  \lambda_{T}\right)  $ and note that%
\begin{align*}
&  \frac{1}{\sqrt{T}}\sum_{t=1}^{T}\left[  \xi_{t,n}\left(  \lambda
_{T}\right)  -\xi_{t,n}\right]  \left[  \upsilon_{t,nT}\left(  \lambda
_{T}\right)  -\upsilon_{t,nT}\right] \\
&  \leq\left(  \frac{1}{\sqrt{T}}\sum_{t=1}^{T}\left[  \xi_{t,n}\left(
\lambda_{T}\right)  -\xi_{t,n}\right]  ^{2}\right)  ^{1/2}\left(  \frac
{1}{\sqrt{T}}\sum_{t=1}^{T}\left[  \upsilon_{t,nT}\left(  \lambda_{T}\right)
-\upsilon_{t,nT}\right]  ^{2}\right)  ^{1/2}\\
&  =\left[  \frac{1}{T\sqrt{T}}\sum_{t=1}^{T}\left\vert \frac{c_{\lambda}%
}{\sqrt{n}}\sum_{i=1}^{n}a_{i,n}b_{it}\right\vert ^{2}\right]  ^{1/2}\left[
\frac{1}{T\sqrt{T}}\sum_{t=1}^{T}\left[  \frac{1}{\sqrt{n}}\sum_{i=1}%
^{n}\left(  \frac{1}{\left(  T^{-1}\boldsymbol{\varepsilon}_{i\circ}^{\prime
}\mathbf{M}_{F}\boldsymbol{\varepsilon}_{i\circ}\right)  ^{1/2}}-1\right)
c_{\lambda}b_{it}\right]  ^{2}\right]  ^{1/2}\\
&  =o_{p}\left(  1\right)  ,
\end{align*}
where the final line holds by (\ref{E:ab}) and (\ref{diff:vt,nT}). Overall, $r_{nT}\left(  \lambda_{T}\right)  =o_{p}\left(  1\right)  $ and
hence $w_{nT}\left(  \lambda_{T}\right) =o_{p}\left(  1\right)
$.
\end{proof}

\section{Derivation of $\theta_{n}$ in terms of factor
strengths\label{section.thetan}}

Consider $\theta_{n}$ defined by (\ref{thetan}), and note that it can be
written as
\begin{align}
\theta_{n}  &  =1-\frac{1}{n}\sum_{i=1}^{n}a_{i,n}^{2}=1-\frac{1}{n}\sum
_{i=1}^{n}\left(  1-\sigma_{i}\boldsymbol{\varphi}_{n}^{\prime}%
\boldsymbol{\gamma}_{i}\right)  ^{2}\nonumber\\
&  =2\boldsymbol{\varphi}_{n}^{\prime}\left(  \frac{1}{n}\sum_{i=1}^{n}%
\sigma_{i}\boldsymbol{\gamma}_{i}\right)  -\boldsymbol{\varphi}_{n}^{\prime
}\left(  \frac{1}{n}\sum_{i=1}^{n}\sigma_{i}^{2}\boldsymbol{\gamma}%
_{i}\boldsymbol{\gamma}_{i}^{\prime}\right)  \boldsymbol{\varphi}_{n},
\label{thetaAPP}%
\end{align}
where $\boldsymbol{\varphi}_{n}=n^{-1}\sum_{i=1}^{n}\boldsymbol{\gamma}%
_{i}/\sigma_{i}$. Then%
\[
\left\vert \theta_{n}\right\vert \leq\sup_{i}(\sigma_{i}^{2})\left\Vert
\boldsymbol{\varphi}_{n}\right\Vert _{1}^{2}\left(  \frac{1}{n}\sum_{i=1}%
^{n}\left\Vert \boldsymbol{\gamma}_{i}\right\Vert _{1}^{2}\right)  +2\sup
_{i}(\sigma_{i})\left\Vert \boldsymbol{\varphi}_{n}\right\Vert _{1}\left(
\frac{1}{n}\sum_{i=1}^{n}\left\Vert \boldsymbol{\gamma}_{i}\right\Vert
_{1}\right)  .
\]
$\left\Vert \boldsymbol{\gamma}_{i}\right\Vert _{1}=$ $\sum_{j=1}^{m_{0}%
}\left\vert \gamma_{ij}\right\vert $, and%
\begin{equation}
\left\Vert \boldsymbol{\varphi}_{n}\right\Vert _{1}\leq\inf_{i}(\sigma
_{i})\left(  \frac{1}{n}\sum_{i=1}^{n}\left\Vert \boldsymbol{\gamma}%
_{i}\right\Vert _{1}\right)  =\inf_{i}(\sigma_{i})\left[  \sum_{j=1}^{m_{0}%
}\left(  \frac{1}{n}\sum_{i=1}^{n}\left\vert \gamma_{ij}\right\vert \right)
\right]  . \label{bound phin}%
\end{equation}
Since by assumption $\inf_{i}(\sigma_{i})>c>0$, and $\sup_{i}(\sigma_{i}%
^{2})<C<\infty$, then the order of $\left\vert \theta_{n}\right\vert $ is
determined by $\sum_{j=1}^{m_{0}}\left(  \frac{1}{n}\sum_{i=1}^{n}\left\vert
\gamma_{ij}\right\vert \right)  $, where $m_{0}$ is a fixed integer. Hence,
$\left\vert \theta_{n}\right\vert =\ominus\left(  n^{\alpha-1}\right)  $ as
required, where $\alpha=\max_{j}(\alpha_{j}),$ and $\alpha_{j}$ is defined by
$\sum_{i=1}^{n}\left\vert \gamma_{ij}\right\vert =\ominus\left(  n^{\alpha
_{j}}\right)  $. See (\ref{fs}).

{\footnotesize
\normalsize
}

\section{Juodis and Reese's CD$_{W+}$ test\label{section.JR}}

The CD$_{W+}$ test statistic, proposed by Juodis and Reese (2022, JR), is used
in our Monte Carlo experiments and is defined by $CD_{W+}=CD_{W}+\Delta_{nT},$
where $CD_{W}$ is the randomized component given by (\ref{CDW+}) and
$\Delta_{nT}$ is the so-called screening component defined by (\ref{Delta nT})
in the paper. Here we provide some theoretical insights regarding the size and
power of the CD$_{W+}$ test. First, for the CD$_{W+}$ test to have the correct
size under $H_{0}:\rho_{ij}=0$, for all $i\neq j$, the screening component
$\Delta_{nT}$ of the test given by (\ref{Delta nT}) must converge to zero as
$n$ and $T\rightarrow\infty$, jointly. To our knowledge, the conditions under
which this holds are not investigated by JR. Whilst it is beyond the scope of
the present paper to investigate the limiting properties of $\Delta_{nT}$ in
the case of a general factor model, using results presented in Bailey et al.
(2019) (BPS), we will provide sufficient conditions for $\Delta_{nT}%
\rightarrow_{p}0$ in the case of the simple null model given by $y_{it}%
=\mu_{i}+\sigma_{i}\varepsilon_{it}$. By the Cauchy-Schwarz inequality we
first note that for all $i\neq j$,%
\begin{align}
&  E\left[  \left\vert \hat{\rho}_{ij,T}\right\vert I\left(  \left\vert
\hat{\rho}_{ij,T}\right\vert >2\sqrt{\frac{\ln\left(  n\right)  }{T}}\right)
|\rho_{ij}=0\right] \nonumber\\
&  \leq\left[  E\left(  \left\vert \hat{\rho}_{ij,T}\right\vert ^{2}|\rho
_{ij}=0\right)  \right]  ^{1/2}\left[  P\left(  \left\vert \hat{\rho}%
_{ij,T}\right\vert >2\sqrt{\frac{\ln\left(  n\right)  }{T}}|\rho
_{ij}=0\right)  \right]  ^{1/2}, \label{JR0}%
\end{align}
where $\rho_{ij}=E\left(  \varepsilon_{it}\varepsilon_{jt}\right)  $. Hence%
\begin{align}
&  E\left(  \Delta_{nT}|\rho_{ij}=0,\text{ for all }i\neq j\right) \nonumber\\
&  \leq\frac{n(n-1)}{2}\sup_{i\neq j}\left[  E\left[  \left\vert \hat{\rho
}_{ij,T}\right\vert ^{2}|\rho_{ij}=0\right]  \right]  ^{1/2}\sup_{i\neq
j}\left[  P\left[  \left\vert \hat{\rho}_{ij,T}\right\vert >2\sqrt{\frac
{\ln\left(  n\right)  }{T}}|\rho_{ij}=0\right]  \right]  ^{1/2}. \label{JR1}%
\end{align}
Now using results $(9)$ and $(10)$ of BPS, we have%
\begin{equation}
E\left[  \left\vert \hat{\rho}_{ij,T}\right\vert ^{2}|\rho_{ij}=0\right]
=O\left(  \frac{1}{T}\right)  , \label{JR2}%
\end{equation}
and using result (A.4) in the online supplement of BPS, we also have%
\[
\sup_{i\neq j}P\left[  \left\vert \hat{\rho}_{ij,T}\right\vert >\frac
{C_{p}\left(  n,\delta\right)  }{\sqrt{T}}|\rho_{ij}=0\right]  =O\left(
e^{-\frac{1}{2}\frac{C_{p}^{2}\left(  n,\delta\right)  }{\varphi_{\max}}%
}\right)  +O\left(  T^{-\frac{\left(  s-1\right)  }{2}}\right)  ,
\]
where $C_{p}\left(  n,\delta\right)  =\Phi^{-1}\left(  1-\frac{p}{2n\delta
}\right)  ,$ $\ 0<p<1$, $\Phi^{-1}\left(  \cdot\right)  $ is the inverse of
the cumulative distribution of a standard normal variable, $\delta>0$,
$\varphi_{\max}=\sup_{i\neq j}E\left(  \varepsilon_{it}^{2}\varepsilon
_{jt}^{2}\right)  $, and $s$ is such that $\sup_{i\neq j}E\left\vert
\varepsilon_{it}\right\vert ^{2s}<C$, for some integer $s\geq3$ (see
Assumption 2 of BPS). Also using results in Lemma 2 in the online supplement
of BPS, we have%
\[
\lim_{n\rightarrow\infty}\frac{C_{p}^{2}\left(  n,\delta\right)  }{\ln\left(
n\right)  }=2\delta,\text{ and }e^{\frac{-\frac{1}{2}C_{p}^{2}\left(
n,\delta\right)  }{\varphi_{\max}}}=O\left(  n^{-\delta/\varphi_{\max}%
}\right)  .
\]
Therefore, $T^{-1/2}C_{p}\left(  n,\delta\right)  $, and $2\sqrt{T^{-1}%
\ln\left(  n\right)  }$ have the same limiting properties if we set $\delta
=2$. Overall, it then follows that%
\begin{equation}
\sup_{i\neq j}P\left(  \left\vert \hat{\rho}_{ij,T}\right\vert >2\sqrt
{\frac{\ln\left(  n\right)  }{T}}|\rho_{ij}=0\right)  =O\left(  n^{-\frac
{2}{\varphi_{\max}}}\right)  +O\left(  T^{-\frac{\left(  s-1\right)  }{2}%
}\right)  . \label{JR3}%
\end{equation}
Using (\ref{JR2}) and (\ref{JR3}) in (\ref{JR1}), we now have%
\begin{equation}
E\left(  \Delta_{nT}|\rho_{ij}=0,\text{ for all }i\neq j\right)  =O\left(
\frac{n^{2-\frac{1}{\varphi_{\max}}}}{\sqrt{T}}\right)  +O\left(
n^{2}T^{-\frac{\left(  s+1\right)  }{4}}\right)  . \label{JR4}%
\end{equation}
Therefore, $\Delta_{nT}\rightarrow_{p}0,$ if $n^{2}T^{-\frac{s+1}{4}%
}\rightarrow0$ and $T^{-1/2}n^{2-\frac{1}{\varphi_{\max}}}\rightarrow0$. It is
easily seen that both of these conditions will be met as $n$ and
$T\rightarrow\infty$ and $n=o\left(  \sqrt{T}\right)  $ if $\varepsilon_{it}$
is Gaussian, since under Gaussian errors, $\varphi_{\max}=1$ and $s$ can be
taken to be sufficiently large. But, in general the expansion rate of $T$
relative to $n$ required to ensure $\Delta_{nT}\rightarrow_{p}0\,$will also
depend on the degree to which $E\left(  \varepsilon_{it}^{2}\varepsilon
_{jt}^{2}\right)  $ exceeds unity. For example, if $\varepsilon_{it}$ has a
multivariate $t$-distribution with degrees of freedom $v>4$, then letting
$T=n^{d}$, $d$ $>0,$ and using results in Lemma 5 of BPS's online supplement,
we have%
\[
\varphi_{\max}=\sup_{i\neq j}E\left(  \varepsilon_{it}^{2}\varepsilon_{jt}%
^{2}|\rho_{ij}=0\right)  =\frac{v-2}{v-4}.
\]
Hence, $E\left(  \Delta_{nT}|\rho_{ij}=0,\text{ for all }i\neq j\right)  $
defined by (\ref{JR4}) tends to $0$ if $n^{2-\frac{v-4}{v-2}-d/2}%
\rightarrow0,$ or if $d>\frac{2v}{v-2}$. Assumption 1 of JR requires
$E\left\vert \varepsilon_{it}\right\vert ^{8+\epsilon}<C$, for some small
positive $\epsilon$, and for this to be satisfied in the case of
$t$-distributed errors we need $v>9$, which yields $d>2$ when $v=10$,
requiring $T$ to rise faster than $n$.

Finally, for the CD$_{W^{+}}$ test to have power it is also necessary to show
that $\Delta_{nT}$ diverges in $n$ and $T$ sufficiently fast under alternative
hypotheses of interest, namely spatial or network dependence. In Section
\ref{mc}, we provide some Monte Carlo evidence on this issue. We find the
CD$_{W^{+}}$ test lacks power against network alternatives, in turn suggesting
that $\Delta_{nT}$ need not diverge sufficiently fast under such alternatives.
Our Monte Carlo experiments also show that the CD$_{W^{+}}$ test tends to
over-reject when $n>>T$ and the errors are chi-squared distributed.

\section{Simulation results\label{section.s3}}

This section provides simulation results for the experiments discussed in
Section \ref{mc} of the main paper. Tables \ref{table.chi2.1f} to
\ref{table.chi2.x2f} report the results for the DGPs with serially independent
errors. Tables \ref{table.v.n.1f} to \ref{table.v.chi2.x2f} report the results
for the DGPs with serially correlated errors using variance adjustment. Tables
\ref{table.ardl.n.1f} to \ref{table.ardl.chi2.x2f} report the results for the
DGPs with serially correlated errors using ARDL adjustment. Figures
\ref{fig_psf} to \ref{fig_xpsf_n} display the simulated power functions of the
CD$^{\ast}$ test in the case of pure factor and panel regression DGPs with
serially independent errors for different $n=200,500$ and $T=100,500$
combinations. To highlight our theoretical result that the power of the
CD$^{\ast}$ test is primarily governed by $T$, each figure is split into two
parts, part (a) which gives the power functions for ($T=100, n=200$) and
($T=100,n=500$) and part (b) which gives the power functions for
($T=500,n=200$) and ($T=500, n=500$). These figures clearly show that when $T$
is fixed, increasing $n$ does not alter the shape of the power function, but
when $T$ is increased the power rises quite sharply as predicted by the
theory. See Theorem \ref{Theorem:CDs} in the paper.

\begin{sidewaystable}
\caption{Size and power of tests of error cross-sectional dependence for the latent factor model with one factor $(m_0=1)$ and serially independent non-Gaussian errors}%
\label{table.chi2.1f}
\scriptsize
\begin{center}
\begin{tabular}{cr|rrr|rrr|rrr|rrr|rrr|rrr}
$\hat{m}=1$& \multicolumn{19}{c}{}\\
\hline\hline
\multicolumn{1}{l}{}                               &      & \multicolumn{9}{c|}{Size ($H_o:\lambda=0$)}                                                                               & \multicolumn{9}{c}{Power ($H_1:\lambda=0.25$)}                                                                            \\\cline{3-20}
\multicolumn{1}{l}{}                               &      & \multicolumn{3}{c|}{$\alpha=1$} & \multicolumn{3}{c|}{$\alpha=2/3$} & \multicolumn{3}{c|}{$\alpha=1/2$} & \multicolumn{3}{c|}{$\alpha=1$} & \multicolumn{3}{c|}{$\alpha=2/3$} & \multicolumn{3}{c}{$\alpha=1/2$} \\\cline{3-20}
\multicolumn{1}{c}{Tests}                           & $n\setminus T$     & 100       & 200      & 500       & 100        & 200       & 500       & 100        & 200        & 500      & 100       & 200       & 500      & 100        & 200       & 500       & 100        & 200       & 500       \\\hline
\multirow{4}{*}{$CD$}        & 100  & 64.5 & 86.4 & 97.7  & 5.3  & 8.9  & 23.1 & 6.1  & 6.6  & 8.0 & 23.5 & 37.1 & 57.5 & 69.0 & 87.0 & 97.8  & 80.4 & 94.3 & 99.4  \\
& 200  & 68.8 & 93.1 & 99.6  & 5.2  & 7.6  & 13.5 & 5.7  & 5.3  & 5.5 & 16.0 & 29.2 & 50.0 & 78.1 & 93.3 & 100.0 & 86.7 & 97.3 & 100.0 \\
& 500  & 69.1 & 94.0 & 100.0 & 4.8  & 5.8  & 9.8  & 5.8  & 6.0  & 5.4 & 11.1 & 23.8 & 46.5 & 84.5 & 97.5 & 100.0 & 90.2 & 98.8 & 100.0 \\
& 1000 & 69.4 & 95.4 & 100.0 & 4.9  & 5.0  & 6.9  & 6.0  & 5.0  & 5.1 & 10.0 & 19.1 & 43.8 & 86.3 & 97.9 & 100.0 & 91.8 & 99.4 & 100.0 \\\hline
\multirow{4}{*}{$CD^\ast$} & 100  & 5.1  & 5.4  & 5.6   & 5.2  & 5.0  & 5.8  & 7.1  & 6.2  & 5.8 & 58.8 & 82.4 & 98.8 & 86.2 & 98.4 & 100.0 & 88.4 & 99.0 & 100.0 \\
& 200  & 3.8  & 4.4  & 4.4   & 5.8  & 5.4  & 5.1  & 5.5  & 5.3  & 5.0 & 58.4 & 81.4 & 99.3 & 86.7 & 98.7 & 100.0 & 89.6 & 99.1 & 100.0 \\
& 500  & 5.4  & 5.6  & 5.1   & 5.3  & 5.5  & 5.3  & 6.0  & 6.2  & 5.5 & 58.6 & 83.5 & 99.5 & 88.4 & 99.1 & 100.0 & 91.1 & 99.3 & 100.0 \\
& 1000 & 5.4  & 5.2  & 4.6   & 5.3  & 5.4  & 5.3  & 6.4  & 5.0  & 4.9 & 60.6 & 84.4 & 99.4 & 88.8 & 99.0 & 100.0 & 92.4 & 99.4 & 100.0 \\\hline
\multirow{4}{*}{$CD_{W+}$} & 100  & 8.3  & 5.9  & 6.0   & 7.4  & 6.0  & 5.4  & 7.9  & 6.2  & 7.4 & 9.8  & 11.7 & 55.1 & 10.0 & 12.4 & 66.7  & 11.4 & 12.2 & 70.9  \\
& 200  & 8.6  & 5.8  & 6.0   & 9.4  & 7.5  & 5.6  & 9.7  & 6.2  & 5.3 & 10.8 & 11.7 & 66.9 & 11.4 & 14.6 & 72.1  & 12.9 & 12.5 & 71.4  \\
& 500  & 17.3 & 10.1 & 6.1   & 16.7 & 10.9 & 5.1  & 17.6 & 9.8  & 5.3 & 18.3 & 17.7 & 79.8 & 18.1 & 17.6 & 82.4  & 19.3 & 16.7 & 82.9  \\
& 1000 & 27.8 & 14.9 & 6.1   & 28.5 & 16.3 & 6.2  & 27.0 & 16.6 & 6.4 & 27.8 & 21.2 & 85.4 & 28.1 & 23.6 & 87.1  & 27.1 & 24.6 & 85.9
\\\hline
\multicolumn{20}{c}{}\\
$\hat{m}=2$& \multicolumn{19}{c}{}\\
\hline\hline
\multicolumn{1}{l}{}                               &      & \multicolumn{9}{c|}{Size ($H_o:\lambda=0$)}                                                                               & \multicolumn{9}{c}{Power ($H_1:\lambda=0.25$)}                                                                               \\\cline{3-20}
\multicolumn{1}{l}{}                               &      & \multicolumn{3}{c|}{$\alpha=1$} & \multicolumn{3}{c|}{$\alpha=2/3$} & \multicolumn{3}{c|}{$\alpha=1/2$} & \multicolumn{3}{c|}{$\alpha=1$} & \multicolumn{3}{c|}{$\alpha=2/3$} & \multicolumn{3}{c}{$\alpha=1/2$} \\\cline{3-20}
\multicolumn{1}{c}{Tests}                           & $n\setminus T$     & 100       & 200      & 500       & 100        & 200       & 500       & 100        & 200        & 500      & 100       & 200       & 500      & 100        & 200       & 500       & 100        & 200       & 500       \\\hline
\multirow{4}{*}{$CD$}        & 100  & 65.8 & 87.0 & 97.6  & 5.9  & 10.4 & 22.8 & 6.3  & 6.9  & 9.4 & 26.5 & 44.2 & 62.5 & 58.9 & 75.0 & 86.1  & 70.2 & 84.5 & 91.1  \\
& 200  & 68.2 & 93.2 & 99.6  & 5.2  & 6.9  & 13.9 & 5.2  & 5.1  & 5.9 & 18.5 & 32.6 & 55.7 & 73.2 & 90.1 & 98.8  & 83.1 & 95.2 & 99.6  \\
& 500  & 68.8 & 93.7 & 100.0 & 4.7  & 5.1  & 9.5  & 5.6  & 6.1  & 5.8 & 11.4 & 24.9 & 50.2 & 82.0 & 96.9 & 100.0 & 88.4 & 98.6 & 100.0 \\
& 1000 & 68.4 & 95.3 & 100.0 & 5.1  & 5.1  & 6.5  & 6.1  & 5.5  & 5.1 & 10.3 & 20.1 & 45.9 & 85.4 & 97.9 & 100.0 & 91.2 & 99.2 & 100.0 \\\hline
\multirow{4}{*}{$CD^\ast$} & 100  & 6.2  & 5.9  & 7.2   & 5.9  & 5.5  & 5.9  & 7.1  & 6.7  & 5.8 & 57.5 & 81.7 & 98.3 & 83.8 & 97.8 & 100.0 & 86.4 & 98.6 & 100.0 \\
& 200  & 4.7  & 4.6  & 5.1   & 6.1  & 5.7  & 5.8  & 6.2  & 5.3  & 5.3 & 58.8 & 81.3 & 99.1 & 85.1 & 97.9 & 100.0 & 88.7 & 99.2 & 100.0 \\
& 500  & 5.4  & 5.7  & 5.6   & 5.4  & 5.8  & 5.2  & 6.3  & 6.5  & 5.4 & 58.4 & 82.6 & 99.5 & 87.5 & 99.0 & 100.0 & 90.8 & 99.1 & 100.0 \\
& 1000 & 5.5  & 5.3  & 4.4   & 5.9  & 5.9  & 5.5  & 6.2  & 5.5  & 4.7 & 60.8 & 83.8 & 99.2 & 88.4 & 98.9 & 100.0 & 91.9 & 99.5 & 100.0 \\\hline
\multirow{4}{*}{$CD_{W+}$} & 100  & 6.4  & 6.4  & 6.1   & 6.5  & 5.8  & 5.6  & 7.0  & 5.6  & 8.1 & 7.8  & 10.9 & 37.6 & 9.1  & 10.2 & 42.1  & 8.6  & 9.5  & 50.1  \\
& 200  & 7.3  & 5.9  & 6.6   & 10.1 & 6.3  & 5.0  & 9.8  & 5.5  & 5.6 & 9.3  & 10.5 & 52.0 & 11.3 & 10.1 & 54.9  & 11.1 & 9.8  & 54.8  \\
& 500  & 15.8 & 8.3  & 5.8   & 13.8 & 9.0  & 5.9  & 15.1 & 8.4  & 5.9 & 16.8 & 15.0 & 73.3 & 15.5 & 16.5 & 73.3  & 17.2 & 15.0 & 74.7  \\
& 1000 & 27.1 & 15.0 & 6.1   & 27.4 & 15.9 & 6.6  & 27.6 & 16.4 & 6.7 & 26.4 & 21.5 & 81.3 & 26.9 & 22.6 & 82.1  & 27.9 & 23.5 & 82.4
\\\hline
\end{tabular}
\end{center}
\par
\textit{Notes}: The DGP is given by (\ref{dgpy}) with $\beta_{i1}=\beta_{i2}=0$ and contains a single latent factor with different factor strengths, $\alpha=1,$ $2/3$, and $1/2$. $\lambda$ denotes the spatial autocorrelation coefficient of the error term defined in (\ref{error}). $m_0$ is the true number of factors and $\hat{m}$ is the number of selected PCs used to compute the different CD statistics.
$CD$ denotes the standard test of error cross-sectional dependence defined by (\ref{cd}), $CD^\ast$ is the bias-corrected version defined by (\ref{CD*a}), and $CD_{W+}$ is the power-enhanced randomized version defined by (\ref{CDW+}).
\end{sidewaystable}

\begin{sidewaystable}
\scriptsize
\caption{Size and power of tests of error cross-sectional dependence for the latent factor model with two factors $(m_0=2)$ and serially independent non-Gaussian errors}
\label{table.chi2.2f}
\begin{center}
\begin{tabular}{cr|rrr|rrr|rrr|rrr|rrr|rrr}
$\hat{m}=2$& \multicolumn{19}{c}{}\\
\hline\hline
\multicolumn{1}{l}{}                               &      & \multicolumn{9}{c|}{Size ($H_o:\lambda=0$)}                                                                               & \multicolumn{9}{c}{Power ($H_1:\lambda=0.25$)}                                                                                \\\cline{3-20}
\multicolumn{1}{l}{}                           &    &    \multicolumn{3}{c|}{$\alpha_1=1,\alpha_2=1$} & \multicolumn{3}{c|}{$\alpha_1=1,\alpha_2=2/3$} & \multicolumn{3}{c|}{$\alpha_1=2/3,\alpha_2=1/2$} & \multicolumn{3}{c|}{$\alpha_1=1,\alpha_2=1$} & \multicolumn{3}{c|}{$\alpha_1=1,\alpha_2=2/3$} & \multicolumn{3}{c}{$\alpha_1=2/3,\alpha_2=1/2$}  \\\cline{3-20}
\multicolumn{1}{c}{Tests}                           & $n\setminus T$     & 100       & 200      & 500       & 100        & 200       & 500       & 100        & 200        & 500      & 100       & 200       & 500      & 100        & 200       & 500       & 100        & 200       & 500       \\\hline
\multirow{4}{*}{$CD$}        & 100  & 99.9  & 100.0 & 100.0 & 98.4 & 99.9  & 100.0 & 9.9  & 17.3 & 43.1 & 99.1  & 99.9  & 100.0 & 88.9 & 97.5 & 99.3  & 55.8 & 64.5 & 79.5  \\
& 200  & 100.0 & 100.0 & 100.0 & 99.2 & 100.0 & 100.0 & 7.6  & 8.9  & 25.3 & 99.8  & 100.0 & 100.0 & 92.2 & 99.1 & 100.0 & 69.8 & 85.8 & 98.3  \\
& 500  & 100.0 & 100.0 & 100.0 & 99.5 & 100.0 & 100.0 & 7.4  & 5.3  & 10.8 & 100.0 & 100.0 & 100.0 & 92.0 & 99.8 & 100.0 & 82.7 & 95.9 & 99.9  \\
& 1000 & 100.0 & 100.0 & 100.0 & 99.7 & 100.0 & 100.0 & 7.6  & 5.1  & 7.9  & 100.0 & 100.0 & 100.0 & 93.5 & 99.8 & 100.0 & 87.4 & 97.9 & 100.0 \\\hline
\multirow{4}{*}{$CD^\ast$} & 100  & 5.3   & 5.5   & 5.0   & 5.9  & 4.8   & 5.1   & 10.8 & 6.7  & 5.4  & 22.8  & 35.1  & 61.3  & 34.1 & 51.8 & 79.6  & 83.5 & 97.4 & 100.0 \\
& 200  & 5.4   & 4.7   & 4.1   & 6.5  & 5.1   & 5.8   & 9.0  & 5.9  & 5.5  & 20.8  & 33.8  & 62.6  & 34.3 & 49.5 & 80.7  & 84.9 & 98.5 & 100.0 \\
& 500  & 5.6   & 6.1   & 5.9   & 6.3  & 5.1   & 5.2   & 9.0  & 5.9  & 5.0  & 22.0  & 35.1  & 63.6  & 33.8 & 53.8 & 83.8  & 88.9 & 98.9 & 100.0 \\
& 1000 & 5.8   & 4.7   & 4.7   & 4.8  & 5.1   & 5.9   & 8.6  & 5.7  & 4.5  & 23.4  & 35.1  & 64.5  & 35.9 & 53.3 & 84.1  & 90.4 & 99.3 & 100.0 \\\hline
\multirow{4}{*}{$CD_{W+}$} & 100  & 7.1   & 6.7   & 7.0   & 7.3  & 5.9   & 7.8   & 8.2  & 6.1  & 10.8 & 9.4   & 11.8  & 43.5  & 8.9  & 10.5 & 52.6  & 11.4 & 13.0 & 70.0  \\
& 200  & 9.1   & 7.0   & 5.9   & 8.8  & 6.5   & 5.1   & 9.6  & 6.8  & 5.7  & 10.6  & 13.1  & 58.1  & 10.5 & 11.6 & 60.6  & 12.9 & 12.7 & 70.8  \\
& 500  & 15.4  & 9.3   & 5.8   & 16.0 & 8.6   & 5.1   & 17.5 & 9.2  & 6.1  & 17.3  & 16.5  & 76.8  & 17.5 & 15.4 & 75.8  & 19.1 & 16.4 & 80.1  \\
& 1000 & 26.9  & 18.2  & 5.6   & 29.0 & 16.0  & 5.9   & 30.3 & 17.7 & 6.7  & 27.4  & 24.0  & 82.6  & 27.6 & 22.2 & 83.6  & 30.0 & 26.3 & 85.6
\\\hline
\multicolumn{20}{c}{}\\
$\hat{m}=4$& \multicolumn{19}{c}{}\\
\hline\hline
\multicolumn{1}{l}{}                               &      & \multicolumn{9}{c|}{Size ($H_o:\lambda=0$)}                                                                               & \multicolumn{9}{c}{Power ($H_1:\lambda=0.25$)}                                                                                \\\cline{3-20}
\multicolumn{1}{l}{}                               &      & \multicolumn{3}{c|}{$\alpha_1=1,\alpha_2=1$} & \multicolumn{3}{c|}{$\alpha_1=1,\alpha_2=2/3$} & \multicolumn{3}{c|}{$\alpha_1=2/3,\alpha_2=1/2$} & \multicolumn{3}{c|}{$\alpha_1=1,\alpha_2=1$} & \multicolumn{3}{c|}{$\alpha_1=1,\alpha_2=2/3$} & \multicolumn{3}{c}{$\alpha_1=2/3,\alpha_2=1/2$} \\\cline{3-20}
\multicolumn{1}{c}{Tests}                           & $n\setminus T$     & 100       & 200      & 500       & 100        & 200       & 500       & 100        & 200        & 500      & 100       & 200       & 500      & 100        & 200       & 500       & 100        & 200       & 500       \\\hline
\multirow{4}{*}{$CD$}        & 100  & 99.8  & 100.0 & 100.0 & 98.4 & 99.9  & 100.0 & 9.3  & 16.4 & 44.1 & 99.3  & 99.9  & 100.0 & 91.5 & 98.4 & 99.6  & 38.6 & 43.6 & 57.3  \\
& 200  & 100.0 & 100.0 & 100.0 & 99.4 & 100.0 & 100.0 & 6.3  & 8.5  & 27.1 & 99.8  & 100.0 & 100.0 & 93.1 & 99.2 & 100.0 & 60.3 & 73.9 & 88.8  \\
& 500  & 100.0 & 100.0 & 100.0 & 99.5 & 100.0 & 100.0 & 7.2  & 5.1  & 11.4 & 100.0 & 100.0 & 100.0 & 92.8 & 99.8 & 100.0 & 79.6 & 93.3 & 99.7  \\
& 1000 & 100.0 & 100.0 & 100.0 & 99.8 & 100.0 & 100.0 & 7.7  & 5.3  & 8.6  & 100.0 & 100.0 & 100.0 & 93.9 & 99.8 & 100.0 & 85.7 & 96.8 & 100.0 \\\hline
\multirow{4}{*}{$CD^\ast$} & 100  & 7.4   & 8.4   & 14.0  & 7.5  & 6.2   & 8.0   & 10.5 & 8.8  & 8.0  & 28.5  & 43.3  & 75.0  & 35.7 & 53.8 & 81.7  & 80.0 & 95.3 & 100.0 \\
& 200  & 6.2   & 5.8   & 7.2   & 7.2  & 5.6   & 6.7   & 8.4  & 6.6  & 7.4  & 24.1  & 37.2  & 69.4  & 34.5 & 50.5 & 82.1  & 82.2 & 98.0 & 100.0 \\
& 500  & 7.0   & 6.1   & 6.9   & 6.3  & 5.3   & 5.5   & 9.1  & 5.9  & 5.1  & 24.2  & 37.5  & 65.5  & 34.2 & 52.7 & 84.1  & 87.8 & 98.5 & 100.0 \\
& 1000 & 5.9   & 4.7   & 4.4   & 5.4  & 5.1   & 6.2   & 9.2  & 5.4  & 5.0  & 24.8  & 37.6  & 65.1  & 36.9 & 52.9 & 84.4  & 89.3 & 99.2 & 100.0 \\\hline
\multirow{4}{*}{$CD_{W+}$} & 100  & 6.5   & 7.8   & 26.9  & 6.6  & 6.9   & 10.7  & 6.3  & 6.6  & 9.7  & 7.8   & 10.1  & 51.2  & 7.1  & 9.9  & 44.2  & 7.9  & 9.2  & 37.8  \\
& 200  & 8.6   & 5.5   & 6.8   & 9.0  & 7.1   & 5.7   & 8.3  & 6.5  & 5.3  & 8.8   & 9.7   & 39.5  & 9.6  & 12.6 & 50.0  & 10.1 & 10.1 & 38.3  \\
& 500  & 15.5  & 8.7   & 5.1   & 14.9 & 9.0   & 5.8   & 14.8 & 9.5  & 6.2  & 15.9  & 13.7  & 63.9  & 14.9 & 13.9 & 69.8  & 16.8 & 14.1 & 64.9  \\
& 1000 & 29.8  & 14.1  & 7.0   & 27.1 & 16.1  & 6.1   & 28.1 & 14.7 & 5.5  & 27.5  & 20.5  & 76.4  & 26.6 & 23.2 & 80.5  & 28.0 & 21.7 & 78.5
\\\hline
\end{tabular}
\end{center}
\par
\textit{Notes}: The DGP is given by (\ref{dgpy}) with $\beta_{i1}=\beta_{i2}=0$ and contains two latent factors with different factor strengths, $(\alpha_{1},\alpha_{2})=(1,1)$, $(1,2/3)$, and $(2/3,1/2)$. $\lambda$ denotes the spatial autocorrelation coefficient of the error term defined in (\ref{error}). $m_0$ is the true number of factors and $\hat{m}$ is the number of selected PCs used to compute the different CD statistics.
$CD$ denotes the
standard test of error cross-sectional dependence defined by (\ref{cd}), $CD^\ast$ is the bias-corrected version defined by (\ref{CD*a}), and $CD_{W+}$ is the power-enhanced randomized version defined by (\ref{CDW+}).
\end{sidewaystable}

\begin{sidewaystable}
\scriptsize
\caption{Size and power of tests of error cross-sectional dependence for the panel regression model with one latent factor $(m_0=1)$ and serially independent non-Gaussian errors}
\label{table.chi2.x1f}
\begin{center}
\begin{tabular}{cr|rrr|rrr|rrr|rrr|rrr|rrr}
$\hat{m}=1$& \multicolumn{19}{c}{}\\
\hline\hline
\multicolumn{1}{l}{}                               &      & \multicolumn{9}{c|}{Size ($H_o:\lambda=0$)}                                                                               & \multicolumn{9}{c}{Power ($H_1:\lambda=0.25$)}                                                                                \\\cline{3-20}
\multicolumn{1}{l}{}                               &      & \multicolumn{3}{c|}{$\alpha=1$} & \multicolumn{3}{c|}{$\alpha=2/3$} & \multicolumn{3}{c|}{$\alpha=1/2$} & \multicolumn{3}{c|}{$\alpha=1$} & \multicolumn{3}{c|}{$\alpha=2/3$} & \multicolumn{3}{c}{$\alpha=1/2$} \\\cline{3-20}
\multicolumn{1}{c}{Tests}                           & $n\setminus T$     & 100       & 200      & 500       & 100        & 200       & 500       & 100        & 200        & 500      & 100       & 200       & 500      & 100        & 200       & 500       & 100        & 200       & 500       \\\hline
\multirow{4}{*}{$CD$}        & 100  & 65.7 & 87.9 & 98.4  & 6.1  & 9.3  & 20.3 & 6.9  & 6.6  & 9.0 & 26.1 & 38.0 & 56.0 & 70.9 & 87.5 & 97.5  & 81.8 & 95.2 & 99.4  \\
& 200  & 67.6 & 91.4 & 99.6  & 5.8  & 7.2  & 13.0 & 5.5  & 6.0  & 6.3 & 16.6 & 30.8 & 49.6 & 80.0 & 93.9 & 99.9  & 86.6 & 98.5 & 100.0 \\
& 500  & 67.4 & 94.9 & 99.9  & 5.3  & 5.0  & 7.4  & 6.9  & 5.8  & 4.9 & 12.0 & 22.3 & 45.6 & 85.3 & 97.4 & 100.0 & 91.3 & 98.9 & 100.0 \\
& 1000 & 68.7 & 95.3 & 100.0 & 6.3  & 4.9  & 7.0  & 6.7  & 6.3  & 5.0 & 10.0 & 19.3 & 44.1 & 86.1 & 98.1 & 100.0 & 90.9 & 99.3 & 100.0 \\\hline
\multirow{4}{*}{$CD^\ast$} & 100  & 5.3  & 4.9  & 5.1   & 5.9  & 5.8  & 5.9  & 7.6  & 6.8  & 6.1 & 57.0 & 82.4 & 98.4 & 86.0 & 98.2 & 100.0 & 89.1 & 98.8 & 100.0 \\
& 200  & 6.0  & 5.4  & 5.1   & 6.2  & 6.1  & 5.3  & 6.3  & 6.8  & 5.5 & 58.5 & 83.0 & 99.2 & 87.6 & 98.8 & 100.0 & 89.6 & 99.2 & 100.0 \\
& 500  & 5.5  & 5.1  & 4.5   & 6.2  & 5.0  & 4.5  & 7.3  & 6.1  & 4.9 & 60.2 & 83.2 & 99.5 & 90.0 & 98.8 & 100.0 & 92.2 & 99.1 & 100.0 \\
& 1000 & 5.5  & 5.4  & 4.8   & 7.0  & 5.6  & 5.3  & 6.7  & 6.5  & 5.4 & 58.4 & 84.3 & 99.2 & 88.9 & 99.2 & 100.0 & 91.3 & 99.2 & 100.0 \\\hline
\multirow{4}{*}{$CD_{W+}$} & 100  & 7.5  & 7.0  & 5.3   & 7.0  & 5.3  & 5.9  & 7.2  & 6.7  & 7.5 & 9.0  & 11.7 & 53.8 & 8.5  & 11.1 & 64.9  & 9.6  & 13.8 & 69.1  \\
& 200  & 9.0  & 7.6  & 4.5   & 8.6  & 7.8  & 5.1  & 8.8  & 7.0  & 5.9 & 10.3 & 13.4 & 64.9 & 10.7 & 14.1 & 70.4  & 11.6 & 16.0 & 71.9  \\
& 500  & 15.9 & 8.9  & 6.1   & 15.8 & 9.5  & 5.8  & 16.1 & 9.7  & 5.7 & 16.6 & 15.8 & 77.5 & 17.8 & 15.7 & 80.4  & 18.2 & 15.8 & 81.0  \\
& 1000 & 26.0 & 14.9 & 6.0   & 26.1 & 16.9 & 5.8  & 26.8 & 14.4 & 6.1 & 24.5 & 22.6 & 84.5 & 25.7 & 23.5 & 85.2  & 25.4 & 21.3 & 86.0
\\\hline
\multicolumn{20}{c}{}\\
$\hat{m}=2$& \multicolumn{19}{c}{}\\
\hline\hline
\multicolumn{1}{l}{}                               &      & \multicolumn{9}{c|}{Size ($H_o:\lambda=0$)}                                                                               & \multicolumn{9}{c}{Power ($H_1:\lambda=0.25$)}                                                                               \\\cline{3-20}
\multicolumn{1}{l}{}                               &      & \multicolumn{3}{c|}{$\alpha=1$} & \multicolumn{3}{c|}{$\alpha=2/3$} & \multicolumn{3}{c|}{$\alpha=1/2$} & \multicolumn{3}{c|}{$\alpha=1$} & \multicolumn{3}{c|}{$\alpha=2/3$} & \multicolumn{3}{c}{$\alpha=1/2$} \\\cline{3-20}
\multicolumn{1}{c}{Tests}                           & $n\setminus T$     & 100       & 200      & 500       & 100        & 200       & 500       & 100        & 200        & 500      & 100       & 200       & 500      & 100        & 200       & 500       & 100        & 200       & 500       \\\hline
\multirow{4}{*}{$CD$}        & 100  & 67.2 & 88.4 & 98.4  & 5.9  & 9.2  & 21.3 & 7.0  & 6.0  & 10.0 & 29.7 & 45.1 & 64.4 & 60.3 & 73.3 & 83.5  & 71.4 & 85.2 & 90.6  \\
& 200  & 67.3 & 90.6 & 99.8  & 5.5  & 6.9  & 13.0 & 6.1  & 5.3  & 6.5  & 18.1 & 35.1 & 55.5 & 72.9 & 88.9 & 98.7  & 80.8 & 95.8 & 99.4  \\
& 500  & 67.3 & 94.7 & 100.0 & 6.1  & 4.8  & 7.6  & 6.8  & 6.3  & 4.6  & 12.6 & 23.4 & 48.8 & 83.9 & 96.5 & 99.9  & 89.1 & 98.3 & 100.0 \\
& 1000 & 68.4 & 95.0 & 100.0 & 6.2  & 4.7  & 7.3  & 7.0  & 6.4  & 5.1  & 10.0 & 19.4 & 45.9 & 84.8 & 97.9 & 100.0 & 89.6 & 98.9 & 100.0 \\\hline
\multirow{4}{*}{$CD^\ast$} & 100  & 6.1  & 5.9  & 6.5   & 6.2  & 6.3  & 6.4  & 7.7  & 7.0  & 6.9  & 56.1 & 81.2 & 98.0 & 84.9 & 97.8 & 100.0 & 86.8 & 98.3 & 100.0 \\
& 200  & 6.8  & 5.8  & 5.8   & 6.3  & 6.1  & 5.1  & 6.9  & 6.3  & 5.8  & 58.3 & 82.9 & 99.1 & 85.5 & 98.5 & 100.0 & 87.7 & 99.0 & 100.0 \\
& 500  & 5.7  & 5.0  & 5.1   & 6.8  & 5.1  & 4.1  & 7.2  & 6.5  & 4.9  & 59.6 & 83.5 & 99.6 & 89.4 & 98.7 & 100.0 & 90.9 & 99.0 & 100.0 \\
& 1000 & 5.5  & 5.4  & 5.0   & 7.2  & 5.7  & 5.6  & 7.0  & 6.8  & 5.6  & 57.8 & 84.5 & 99.2 & 87.3 & 99.0 & 100.0 & 90.6 & 99.1 & 100.0 \\\hline
\multirow{4}{*}{$CD_{W+}$} & 100  & 6.5  & 6.4  & 7.1   & 6.3  & 5.8  & 6.6  & 7.3  & 5.9  & 8.5  & 8.9  & 10.0 & 38.2 & 8.1  & 8.6  & 40.8  & 8.8  & 10.0 & 48.9  \\
& 200  & 8.6  & 6.5  & 5.2   & 9.0  & 6.9  & 5.8  & 8.8  & 6.5  & 6.0  & 10.2 & 10.7 & 50.6 & 11.4 & 11.3 & 54.8  & 9.5  & 11.9 & 56.2  \\
& 500  & 14.1 & 9.5  & 5.0   & 15.5 & 10.0 & 5.6  & 14.9 & 7.7  & 4.9  & 15.8 & 13.4 & 70.7 & 15.8 & 15.4 & 72.8  & 15.9 & 12.8 & 73.1  \\
& 1000 & 26.0 & 15.3 & 6.9   & 24.0 & 14.5 & 6.4  & 25.1 & 14.5 & 6.0  & 25.0 & 20.8 & 79.3 & 24.1 & 21.7 & 81.0  & 25.1 & 20.5 & 81.1  \\\hline
\end{tabular}
\end{center}
\par
\textit{Notes}: The DGP is given by (\ref{dgpy}) with $\beta_{i1}$ and $\beta_{i2}$ both generated from normal distribution, and contains a single latent factor with different factor strengths, $\alpha=1,$ $2/3$, and $1/2$. $\lambda$ denotes the spatial autocorrelation coefficient of the error term defined in (\ref{error}). $m_0$ is the true number of factors and $\hat{m}$ is the number of selected PCs used to compute the different CD statistics.
$CD$ denotes the
standard test of error cross-sectional dependence defined by (\ref{cd}), $CD^\ast$ is the bias-corrected version defined by (\ref{CD*a}), and $CD_{W+}$ is the power-enhanced randomized version defined by (\ref{CDW+}).
\end{sidewaystable}

\begin{sidewaystable}
\scriptsize
\caption{Size and power of tests of error cross-sectional dependence for the panel regression model with two latent factors $(m_0=2)$ and serially independent non-Gaussian errors}
\label{table.chi2.x2f}
\begin{center}
\begin{tabular}{cr|rrr|rrr|rrr|rrr|rrr|rrr}
$\hat{m}=2$& \multicolumn{19}{c}{}\\
\hline\hline
\multicolumn{1}{l}{}                               &      & \multicolumn{9}{c|}{Size ($H_o:\lambda=0$)}                                                                               & \multicolumn{9}{c}{Power ($H_1:\lambda=0.25$)}                                                                                \\\cline{3-20}
\multicolumn{1}{l}{}                           &    &    \multicolumn{3}{c|}{$\alpha_1=1,\alpha_2=1$} & \multicolumn{3}{c|}{$\alpha_1=1,\alpha_2=2/3$} & \multicolumn{3}{c|}{$\alpha_1=2/3,\alpha_2=1/2$} & \multicolumn{3}{c|}{$\alpha_1=1,\alpha_2=1$} & \multicolumn{3}{c|}{$\alpha_1=1,\alpha_2=2/3$} & \multicolumn{3}{c}{$\alpha_1=2/3,\alpha_2=1/2$}  \\\cline{3-20}
\multicolumn{1}{c}{Tests}                           & $n/T$     & 100       & 200      & 500       & 100        & 200       & 500       & 100        & 200        & 500      & 100       & 200       & 500      & 100        & 200       & 500       & 100        & 200       & 500       \\\hline
\multirow{4}{*}{$CD$}        & 100  & 100.0 & 100.0 & 100.0 & 97.9 & 100.0 & 100.0 & 8.6  & 13.7 & 39.7 & 98.9  & 100.0 & 100.0 & 89.5 & 97.3 & 99.6  & 57.7 & 68.5 & 80.9  \\
& 200  & 100.0 & 100.0 & 100.0 & 99.0 & 100.0 & 100.0 & 6.6  & 7.4  & 22.9 & 99.7  & 100.0 & 100.0 & 91.4 & 99.1 & 100.0 & 72.8 & 87.4 & 98.3  \\
& 500  & 100.0 & 100.0 & 100.0 & 99.5 & 100.0 & 100.0 & 7.0  & 6.0  & 11.8 & 100.0 & 100.0 & 100.0 & 92.6 & 99.6 & 100.0 & 83.0 & 95.9 & 99.9  \\
& 1000 & 100.0 & 100.0 & 100.0 & 99.6 & 100.0 & 100.0 & 6.6  & 5.9  & 7.6  & 100.0 & 100.0 & 100.0 & 93.0 & 99.9 & 100.0 & 85.7 & 97.5 & 100.0 \\\hline
\multirow{4}{*}{$CD^\ast$} & 100  & 5.3   & 5.4   & 5.4   & 6.9  & 5.7   & 4.9   & 9.3  & 6.3  & 6.2  & 23.0  & 34.5  & 63.1  & 32.6 & 49.6 & 78.3  & 83.7 & 97.7 & 100.0 \\
& 200  & 5.2   & 5.1   & 5.1   & 6.3  & 4.7   & 5.5   & 8.2  & 6.7  & 5.4  & 21.2  & 35.9  & 64.3  & 34.5 & 50.2 & 81.1  & 87.6 & 98.5 & 100.0 \\
& 500  & 5.9   & 4.7   & 5.8   & 5.9  & 5.3   & 4.8   & 9.2  & 5.9  & 5.9  & 21.9  & 37.6  & 64.9  & 36.2 & 53.4 & 83.9  & 88.4 & 99.0 & 100.0 \\
& 1000 & 6.8   & 4.8   & 5.5   & 6.5  & 4.5   & 5.0   & 7.6  & 6.4  & 5.0  & 23.4  & 33.6  & 63.2  & 34.2 & 52.7 & 86.4  & 88.9 & 98.9 & 100.0 \\\hline
\multirow{4}{*}{$CD_{W+}$} & 100  & 7.1   & 6.4   & 5.6   & 7.1  & 5.8   & 6.8   & 7.2  & 5.6  & 9.9  & 8.4   & 10.5  & 42.2  & 9.0  & 9.9  & 50.6  & 10.7 & 13.1 & 67.6  \\
& 200  & 10.1  & 7.0   & 5.1   & 9.9  & 7.5   & 4.9   & 9.5  & 5.9  & 6.1  & 10.8  & 12.1  & 58.1  & 10.9 & 12.4 & 58.8  & 12.4 & 12.0 & 69.0  \\
& 500  & 16.2  & 10.4  & 6.1   & 15.8 & 9.0   & 5.4   & 16.3 & 9.8  & 5.8  & 16.5  & 15.2  & 74.4  & 16.0 & 14.2 & 74.8  & 17.7 & 17.6 & 79.8  \\
& 1000 & 26.6  & 14.7  & 6.2   & 28.1 & 18.0  & 5.2   & 28.1 & 15.0 & 5.7  & 25.4  & 21.5  & 84.1  & 27.5 & 24.2 & 82.8  & 27.4 & 22.2 & 85.1
\\\hline
\multicolumn{20}{c}{}\\
$\hat{m}=4$& \multicolumn{19}{c}{}\\
\hline\hline
\multicolumn{1}{l}{}                               &      & \multicolumn{9}{c|}{Size ($H_o:\lambda=0$)}                                                                               & \multicolumn{9}{c}{Power ($H_1:\lambda=0.25$)}                                                                                \\\cline{3-20}
\multicolumn{1}{l}{}                               &      & \multicolumn{3}{c|}{$\alpha_1=1,\alpha_2=1$} & \multicolumn{3}{c|}{$\alpha_1=1,\alpha_2=2/3$} & \multicolumn{3}{c|}{$\alpha_1=2/3,\alpha_2=1/2$} & \multicolumn{3}{c|}{$\alpha_1=1,\alpha_2=1$} & \multicolumn{3}{c|}{$\alpha_1=1,\alpha_2=2/3$} & \multicolumn{3}{c}{$\alpha_1=2/3,\alpha_2=1/2$} \\\cline{3-20}
\multicolumn{1}{c}{Tests}                           & $n/T$     & 100       & 200      & 500       & 100        & 200       & 500       & 100        & 200        & 500      & 100       & 200       & 500      & 100        & 200       & 500       & 100        & 200       & 500       \\\hline
\multirow{4}{*}{$CD$}        & 100  & 99.9  & 100.0 & 100.0 & 98.4 & 99.9  & 100.0 & 8.1  & 15.3 & 41.7 & 99.2  & 100.0 & 100.0 & 91.6 & 98.1 & 99.5  & 38.4 & 45.2 & 54.9  \\
& 200  & 100.0 & 100.0 & 100.0 & 99.1 & 100.0 & 100.0 & 6.1  & 8.6  & 23.5 & 99.7  & 100.0 & 100.0 & 92.6 & 99.2 & 100.0 & 61.2 & 76.4 & 88.2  \\
& 500  & 100.0 & 100.0 & 100.0 & 99.5 & 100.0 & 100.0 & 7.4  & 5.6  & 11.6 & 100.0 & 100.0 & 100.0 & 92.8 & 99.8 & 100.0 & 77.9 & 94.0 & 99.9  \\
& 1000 & 100.0 & 100.0 & 100.0 & 99.7 & 100.0 & 100.0 & 6.8  & 6.1  & 7.4  & 100.0 & 100.0 & 100.0 & 93.0 & 99.9 & 100.0 & 83.2 & 97.2 & 100.0 \\\hline
\multirow{4}{*}{$CD^\ast$} & 100  & 8.2   & 8.9   & 15.0  & 8.9  & 8.6   & 11.0  & 9.8  & 7.9  & 9.0  & 29.3  & 41.4  & 75.1  & 34.7 & 52.6 & 82.0  & 80.0 & 95.7 & 100.0 \\
& 200  & 5.8   & 6.0   & 7.1   & 6.9  & 5.5   & 6.9   & 8.9  & 8.1  & 7.0  & 25.1  & 39.8  & 70.7  & 35.3 & 51.0 & 82.4  & 84.1 & 97.7 & 100.0 \\
& 500  & 6.1   & 5.2   & 5.8   & 6.5  & 5.4   & 5.4   & 9.6  & 6.7  & 6.2  & 24.3  & 39.4  & 66.7  & 36.3 & 53.9 & 84.4  & 85.8 & 98.3 & 100.0 \\
& 1000 & 7.0   & 4.8   & 5.6   & 7.0  & 5.0   & 5.4   & 7.9  & 6.6  & 5.1  & 25.1  & 35.2  & 64.8  & 35.4 & 53.2 & 86.7  & 86.7 & 98.7 & 100.0 \\\hline
\multirow{4}{*}{$CD_{W+}$} & 100  & 7.3   & 7.6   & 26.2  & 6.8  & 6.8   & 15.5  & 5.5  & 6.4  & 12.0 & 8.9   & 11.1  & 48.9  & 7.6  & 9.6  & 43.3  & 6.2  & 8.6  & 38.0  \\
& 200  & 8.2   & 6.8   & 5.9   & 9.0  & 7.3   & 5.8   & 9.5  & 5.9  & 5.7  & 9.6   & 10.1  & 38.7  & 8.7  & 11.7 & 47.8  & 10.4 & 9.7  & 38.0  \\
& 500  & 14.5  & 9.6   & 4.9   & 15.1 & 8.8   & 6.6   & 13.3 & 8.1  & 5.3  & 15.8  & 13.6  & 61.3  & 15.6 & 13.8 & 68.8  & 14.6 & 12.5 & 62.7  \\
& 1000 & 26.3  & 13.2  & 6.9   & 25.2 & 14.4  & 6.3   & 24.7 & 16.0 & 5.6  & 25.1  & 19.8  & 75.1  & 26.4 & 21.2 & 77.8  & 23.4 & 21.5 & 77.5  \\\hline
\end{tabular}
\end{center}
\par
\textit{Notes}: The DGP is given by (\ref{dgpy}) with $\beta_{i1}$ and $\beta_{i2}$ both generated from normal distribution, and contains two latent factors with different factor strengths, $(\alpha_{1},\alpha_{2})=(1,1)$, $(1,2/3)$, and $(2/3,1/2)$. $\lambda$ denotes the spatial autocorrelation coefficient of the error term defined in (\ref{error}). $m_0$ is the true number of factors and $\hat{m}$ is the number of selected PCs used to compute the different CD statistics.
$CD$ denotes the
standard test of error cross-sectional dependence defined by (\ref{cd}), $CD^\ast$ is the bias-corrected version defined by (\ref{CD*a}), and $CD_{W+}$ is the power-enhanced randomized version defined by (\ref{CDW+}).
\end{sidewaystable}


\begin{sidewaystable}
\scriptsize
\caption{Size and power of variance adjusted tests of error cross-sectional dependence for the latent factor model with one factor $(m_0=1)$ and serially correlated Gaussian errors}
\label{table.v.n.1f}
\begin{center}
\begin{tabular}{cr|rrr|rrr|rrr|rrr|rrr|rrr}
$\hat{m}=1$& \multicolumn{19}{c}{}\\
\hline\hline
\multicolumn{1}{l}{}                               &      & \multicolumn{9}{c|}{Size ($H_o:\lambda=0$)}                                                                               & \multicolumn{9}{c}{Power ($H_1:\lambda=0.25$)}                                                                                \\\cline{3-20}
\multicolumn{1}{l}{}                               &      & \multicolumn{3}{c|}{$\alpha=1$} & \multicolumn{3}{c|}{$\alpha=2/3$} & \multicolumn{3}{c|}{$\alpha=1/2$} & \multicolumn{3}{c|}{$\alpha=1$} & \multicolumn{3}{c|}{$\alpha=2/3$} & \multicolumn{3}{c}{$\alpha=1/2$} \\\cline{3-20}
\multicolumn{1}{c}{Tests}                           & $n\setminus T$     & 100       & 200      & 500       & 100        & 200       & 500       & 100        & 200        & 500      & 100       & 200       & 500      & 100        & 200       & 500       & 100        & 200       & 500       \\\hline
\multirow{4}{*}{Variance
adjusted $CD$}                        & 100  & 41.0 & 71.4 & 93.8 & 5.3  & 6.6  & 12.9 & 6.7  & 6.5  & 5.9  & 12.8 & 26.1 & 44.1  & 59.8 & 74.9 & 91.8  & 67.5 & 85.9 & 98.3   \\
& 200  & 37.0 & 75.8 & 97.2 & 5.2  & 5.0  & 7.5  & 6.0  & 5.9  & 5.5  & 7.7  & 18.8 & 34.8  & 65.7 & 83.0 & 97.7  & 74.5 & 90.1 & 99.9   \\
& 500  & 33.5 & 76.6 & 98.6 & 6.7  & 4.8  & 5.7  & 6.5  & 5.4  & 5.6  & 5.3  & 14.4 & 31.9  & 74.2 & 90.1 & 99.3  & 73.0 & 94.2 & 99.7   \\
& 1000 & 28.8 & 76.6 & 99.6 & 10.1 & 5.9  & 4.5  & 6.9  & 6.6  & 4.2  & 4.4  & 10.4 & 27.4  & 77.9 & 92.3 & 100.0 & 78.7 & 93.7 & 100.0  \\\hline
\multirow{4}{*}{Variance adjusted
$CD^\ast$} & 100  & 5.6  & 5.3  & 5.6  & 6.2  & 6.2  & 6.3  & 6.8  & 6.7  & 4.9  & 40.9 & 61.7 & 89.3  & 69.4 & 88.9 & 99.7  & 73.2 & 92.4 & 99.8   \\
& 200  & 5.0  & 4.7  & 5.9  & 6.5  & 4.4  & 6.1  & 6.6  & 6.5  & 5.4  & 45.6 & 66.3 & 92.9  & 72.8 & 90.6 & 100.0 & 75.8 & 93.4 & 100.0  \\
& 500  & 6.3  & 5.0  & 5.4  & 7.2  & 5.5  & 4.5  & 6.8  & 5.8  & 5.6  & 46.2 & 63.5 & 93.5  & 77.1 & 93.7 & 99.9  & 73.7 & 94.8 & 99.8   \\
& 1000 & 6.7  & 5.2  & 5.1  & 11.4 & 6.3  & 4.0  & 6.9  & 6.4  & 4.4  & 49.9 & 66.7 & 94.2  & 79.9 & 93.4 & 100.0 & 78.8 & 94.2 & 100.0  \\\hline
\multirow{4}{*}{Variance adjusted $CD_{W+}$}                  & 100  & 16.0 & 15.5 & 12.8 & 16.0 & 14.5 & 14.9 & 15.7 & 16.9 & 17.0 & 22.8 & 39.1 & 86.8  & 23.6 & 38.8 & 93.9  & 26.5 & 43.3 & 94.4   \\
& 200  & 20.2 & 25.3 & 19.8 & 22.8 & 26.1 & 20.6 & 20.5 & 26.8 & 22.4 & 29.1 & 55.8 & 98.6  & 29.6 & 57.0 & 99.8  & 35.7 & 59.1 & 99.5   \\
& 500  & 30.4 & 54.2 & 51.5 & 30.5 & 54.7 & 51.1 & 40.1 & 58.0 & 51.7 & 35.1 & 79.6 & 100.0 & 34.9 & 78.9 & 100.0 & 50.9 & 82.4 & 100.0  \\
& 1000 & 36.6 & 80.8 & 85.5 & 39.5 & 78.3 & 87.0 & 53.9 & 82.6 & 86.9 & 42.5 & 92.6 & 100.0 & 42.9 & 92.2 & 100.0 & 58.1 & 94.6 & 100.0
\\\hline
\multicolumn{20}{c}{}\\
$\hat{m}=2$& \multicolumn{19}{c}{}\\
\hline\hline
\multicolumn{1}{l}{}                               &      & \multicolumn{9}{c|}{Size ($H_o:\lambda=0$)}                                                                               & \multicolumn{9}{c}{Power ($H_1:\lambda=0.25$)}                                                                                \\\cline{3-20}
\multicolumn{1}{l}{}                               &      & \multicolumn{3}{c|}{$\alpha=1$} & \multicolumn{3}{c|}{$\alpha=2/3$} & \multicolumn{3}{c|}{$\alpha=1/2$} & \multicolumn{3}{c|}{$\alpha=1$} & \multicolumn{3}{c|}{$\alpha=2/3$} & \multicolumn{3}{c}{$\alpha=1/2$}\\\cline{3-20}
\multicolumn{1}{c}{Tests}                           & $n\setminus T$     & 100       & 200      & 500       & 100        & 200       & 500       & 100        & 200        & 500      & 100       & 200       & 500      & 100        & 200       & 500       & 100        & 200       & 500       \\\hline
\multirow{4}{*}{Variance
adjusted $CD$}                        & 100  & 41.7 & 72.2 & 93.3 & 4.2  & 5.6  & 12.8 & 6.1  & 6.3  & 6.1  & 15.5 & 28.1 & 50.0  & 52.4 & 66.5 & 81.2  & 61.4 & 79.6 & 89.2   \\
& 200  & 38.2 & 75.3 & 96.9 & 4.8  & 5.0  & 7.4  & 6.4  & 5.2  & 5.4  & 8.1  & 19.4 & 37.4  & 63.2 & 80.7 & 96.3  & 71.4 & 88.2 & 99.1   \\
& 500  & 34.1 & 77.0 & 98.7 & 7.3  & 4.6  & 5.6  & 6.8  & 5.2  & 5.6  & 5.1  & 15.5 & 33.9  & 73.5 & 88.6 & 99.1  & 71.3 & 93.3 & 99.7   \\
& 1000 & 29.5 & 77.0 & 99.6 & 11.2 & 5.8  & 3.9  & 7.0  & 6.5  & 4.4  & 4.5  & 10.4 & 28.5  & 77.8 & 90.9 & 100.0 & 77.2 & 93.2 & 100.0  \\\hline
\multirow{4}{*}{Variance adjusted
$CD^\ast$} & 100  & 4.9  & 4.8  & 5.7  & 5.9  & 6.7  & 5.8  & 7.0  & 7.4  & 6.2  & 41.1 & 62.3 & 88.8  & 69.1 & 89.6 & 99.5  & 71.6 & 92.0 & 99.9   \\
& 200  & 5.6  & 4.2  & 6.1  & 6.8  & 4.9  & 7.2  & 7.1  & 6.0  & 5.5  & 45.5 & 65.7 & 93.6  & 71.5 & 90.4 & 99.7  & 75.3 & 93.2 & 100.0  \\
& 500  & 5.6  & 5.5  & 5.3  & 8.2  & 5.1  & 5.1  & 7.0  & 5.3  & 5.8  & 45.2 & 63.4 & 93.9  & 76.7 & 93.7 & 99.9  & 73.1 & 94.0 & 99.8   \\
& 1000 & 6.9  & 5.7  & 5.2  & 11.7 & 7.7  & 4.1  & 7.1  & 6.7  & 4.3  & 50.2 & 67.3 & 94.4  & 80.3 & 93.5 & 100.0 & 77.8 & 94.1 & 100.0  \\\hline
\multirow{4}{*}{Variance adjusted $CD_{W+}$}                  & 100  & 13.4 & 15.6 & 14.1 & 14.9 & 14.3 & 13.4 & 15.0 & 15.6 & 16.9 & 19.3 & 34.4 & 80.7  & 19.8 & 32.9 & 84.6  & 19.6 & 35.0 & 86.1   \\
& 200  & 18.3 & 22.9 & 20.2 & 19.9 & 25.2 & 18.3 & 17.4 & 27.0 & 21.8 & 26.3 & 47.6 & 97.9  & 24.5 & 49.3 & 98.5  & 27.5 & 50.0 & 98.9   \\
& 500  & 27.3 & 52.1 & 49.7 & 27.0 & 54.1 & 47.9 & 33.7 & 55.6 & 51.0 & 31.3 & 77.7 & 100.0 & 33.4 & 78.3 & 100.0 & 40.7 & 80.5 & 100.0  \\
& 1000 & 36.3 & 80.2 & 82.9 & 36.5 & 78.7 & 83.5 & 50.1 & 84.4 & 83.8 & 40.6 & 93.5 & 100.0 & 42.0 & 92.2 & 100.0 & 54.4 & 95.6 & 100.0  \\\hline
\end{tabular}
\end{center}
\par
\textit{Notes}: The DGP is given by (\ref{dgpy}) with $\beta_{i1}=\beta_{i2}=0$ and contains a single latent factor with different factor strengths, $\alpha=1,$ $2/3$, and $1/2$. $\lambda$ denotes the spatial autocorrelation coefficient of the error term defined in (\ref{error}). $m_0$ is the true number of factors and $\hat{m}$ is the number of selected PCs used to compute the different CD statistics.
$CD$ denotes the
standard test of error cross-sectional dependence defined by (\ref{cd}), $CD^\ast$ is the bias-corrected version defined by (\ref{CD*a}), and $CD_{W+}$ is the power-enhanced randomized version defined by (\ref{CDW+}).
\end{sidewaystable}

\begin{sidewaystable}
\scriptsize
\caption{Size and power of variance adjusted tests of error cross-sectional dependence for the latent factor model with two factors $(m_0=2)$ and serially correlated Gaussian errors}
\label{table.v.n.2f}
\begin{center}
\begin{tabular}{cr|rrr|rrr|rrr|rrr|rrr|rrr}
$\hat{m}=2$& \multicolumn{19}{c}{}\\
\hline\hline
\multicolumn{1}{l}{}                               &      & \multicolumn{9}{c|}{Size ($H_o:\lambda=0$)}                                                                               & \multicolumn{9}{c}{Power ($H_1:\lambda=0.25$)}                                                                                \\\cline{3-20}
\multicolumn{1}{l}{}                           &    &    \multicolumn{3}{c|}{$\alpha_1=1,\alpha_2=1$} & \multicolumn{3}{c|}{$\alpha_1=1,\alpha_2=2/3$} & \multicolumn{3}{c|}{$\alpha_1=2/3,\alpha_2=1/2$} & \multicolumn{3}{c|}{$\alpha_1=1,\alpha_2=1$} & \multicolumn{3}{c|}{$\alpha_1=1,\alpha_2=2/3$} & \multicolumn{3}{c}{$\alpha_1=2/3,\alpha_2=1/2$}   \\\cline{3-20}
\multicolumn{1}{c}{Tests}                           & $n\setminus T$     & 100       & 200      & 500       & 100        & 200       & 500       & 100        & 200        & 500      & 100       & 200       & 500      & 100        & 200       & 500       & 100        & 200       & 500       \\\hline
\multirow{4}{*}{Variance adjusted $CD$}    & 100  & 98.2 & 100.0 & 100.0 & 87.5 & 98.7  & 100.0 & 8.4  & 6.4  & 22.4 & 92.9 & 99.1  & 100.0 & 67.6 & 90.6 & 97.8  & 54.0 & 61.0 & 69.8  \\
& 200  & 99.0 & 100.0 & 100.0 & 88.7 & 99.8  & 100.0 & 9.0  & 5.2  & 9.8  & 95.4 & 99.7  & 100.0 & 66.6 & 94.6 & 99.4  & 64.8 & 77.8 & 94.0  \\
& 500  & 99.2 & 100.0 & 100.0 & 88.2 & 100.0 & 100.0 & 12.2 & 8.2  & 7.2  & 95.1 & 100.0 & 100.0 & 63.4 & 96.8 & 100.0 & 80.8 & 90.8 & 99.2  \\
& 1000 & 98.0 & 100.0 & 100.0 & 84.6 & 99.9  & 100.0 & 13.4 & 8.4  & 4.8  & 92.6 & 100.0 & 100.0 & 52.4 & 95.4 & 100.0 & 81.6 & 93.0 & 99.8  \\\hline
\multirow{4}{*}{Variance adjusted $CD^\ast$} & 100  & 5.0  & 5.0   & 4.5   & 5.9  & 5.4   & 4.4   & 12.2 & 5.8  & 6.0  & 17.4 & 22.0  & 40.3  & 27.0 & 34.8 & 57.8  & 71.6 & 91.4 & 99.6  \\
& 200  & 5.9  & 4.2   & 4.8   & 7.1  & 5.7   & 4.2   & 10.8 & 6.6  & 6.8  & 19.2 & 24.1  & 42.1  & 27.8 & 35.8 & 64.6  & 75.0 & 91.0 & 99.8  \\
& 500  & 8.4  & 4.8   & 4.3   & 9.0  & 6.2   & 4.5   & 14.2 & 9.6  & 6.6  & 25.3 & 27.1  & 45.3  & 30.0 & 41.4 & 71.2  & 83.8 & 93.0 & 99.8  \\
& 1000 & 13.7 & 4.7   & 5.1   & 12.5 & 7.7   & 5.6   & 14.4 & 11.6 & 6.6  & 33.6 & 28.4  & 46.4  & 42.6 & 41.2 & 65.8  & 83.6 & 94.8 & 100.0 \\\hline
\multirow{4}{*}{Variance adjusted $CD_{W+}$} & 100  & 17.2 & 16.5  & 14.1  & 16.8 & 18.7  & 20.9  & 16.8 & 17.8 & 23.6 & 23.3 & 37.8  & 85.9  & 21.2 & 33.4 & 76.6  & 26.0 & 43.6 & 95.4  \\
& 200  & 20.9 & 29.2  & 24.2  & 19.4 & 27.6  & 24.6  & 28.2 & 30.6 & 26.2 & 28.0 & 53.0  & 98.3  & 21.2 & 39.2 & 90.8  & 35.4 & 58.2 & 99.6  \\
& 500  & 32.4 & 54.9  & 51.7  & 29.1 & 55.4  & 52.7  & 45.4 & 63.2 & 55.4 & 38.1 & 79.6  & 100.0 & 27.6 & 51.6 & 99.4  & 50.4 & 80.8 & 100.0 \\
& 1000 & 39.4 & 80.9  & 85.9  & 38.5 & 81.3  & 87.7  & 59.8 & 87.6 & 88.0 & 43.8 & 94.0  & 100.0 & 36.4 & 79.0 & 100.0 & 60.2 & 96.4 & 100.0
\\\hline
\multicolumn{20}{c}{}\\
$\hat{m}=4$& \multicolumn{19}{c}{}\\
\hline\hline
\multicolumn{1}{l}{}                               &      & \multicolumn{9}{c|}{Size ($H_o:\lambda=0$)}                                                                               & \multicolumn{9}{c}{Power ($H_1:\lambda=0.25$)}                                                                                \\\cline{3-20}
\multicolumn{1}{l}{}                           &    &    \multicolumn{3}{c|}{$\alpha_1=1,\alpha_2=1$} & \multicolumn{3}{c|}{$\alpha_1=1,\alpha_2=2/3$} & \multicolumn{3}{c|}{$\alpha_1=2/3,\alpha_2=1/2$} & \multicolumn{3}{c|}{$\alpha_1=1,\alpha_2=1$} & \multicolumn{3}{c|}{$\alpha_1=1,\alpha_2=2/3$} & \multicolumn{3}{c}{$\alpha_1=2/3,\alpha_2=1/2$}  \\\cline{3-20}
\multicolumn{1}{c}{Tests}                           & $n\setminus T$     & 100       & 200      & 500       & 100        & 200       & 500       & 100        & 200        & 500      & 100       & 200       & 500      & 100        & 200       & 500       & 100        & 200       & 500       \\\hline
\multirow{4}{*}{Variance adjusted   $CD$}    & 100  & 98.7 & 100.0 & 100.0 & 88.7 & 98.8  & 100.0 & 5.8  & 6.2  & 20.4 & 95.4 & 99.7  & 99.9  & 73.2 & 94.0 & 99.4  & 44.4 & 46.6 & 46.2  \\
& 200  & 99.5 & 100.0 & 100.0 & 91.2 & 99.8  & 100.0 & 6.6  & 5.2  & 12.2 & 96.3 & 99.9  & 100.0 & 69.4 & 95.2 & 99.6  & 59.8 & 69.6 & 87.6  \\
& 500  & 99.2 & 100.0 & 100.0 & 88.4 & 100.0 & 100.0 & 11.4 & 7.6  & 7.6  & 95.5 & 100.0 & 100.0 & 61.4 & 97.6 & 100.0 & 77.4 & 87.6 & 98.4  \\
& 1000 & 98.1 & 100.0 & 100.0 & 83.5 & 99.7  & 100.0 & 12.2 & 8.2  & 5.8  & 92.4 & 100.0 & 100.0 & 51.6 & 96.6 & 100.0 & 78.8 & 91.8 & 99.8  \\\hline
\multirow{4}{*}{Variance adjusted $CD^\ast$} & 100  & 5.6  & 6.8   & 9.3   & 7.8  & 7.0   & 6.6   & 9.2  & 6.2  & 8.8  & 20.7 & 27.5  & 51.6  & 28.4 & 43.4 & 65.8  & 71.0 & 89.4 & 99.8  \\
& 200  & 5.5  & 5.8   & 5.7   & 7.1  & 6.9   & 5.1   & 10.2 & 7.2  & 8.4  & 20.3 & 26.9  & 48.1  & 29.4 & 38.2 & 67.2  & 73.8 & 90.6 & 100.0 \\
& 500  & 8.9  & 5.4   & 4.1   & 9.3  & 6.6   & 5.4   & 13.0 & 10.8 & 6.4  & 25.9 & 28.0  & 46.6  & 34.0 & 43.2 & 71.6  & 82.6 & 92.8 & 99.6  \\
& 1000 & 14.3 & 4.6   & 5.1   & 13.5 & 7.6   & 5.5   & 14.2 & 10.8 & 7.2  & 34.9 & 29.1  & 47.9  & 42.2 & 43.6 & 66.0  & 81.0 & 93.4 & 100.0 \\\hline
\multirow{4}{*}{Variance adjusted $CD_{W+}$} & 100  & 13.4 & 17.7  & 34.6  & 13.2 & 16.7  & 32.5  & 14.0 & 20.6 & 25.0 & 17.4 & 30.5  & 82.4  & 18.2 & 30.2 & 80.6  & 18.2 & 34.2 & 80.8  \\
& 200  & 18.8 & 26.3  & 23.4  & 20.6 & 24.0  & 25.8  & 22.8 & 25.6 & 22.4 & 22.7 & 45.5  & 93.6  & 23.6 & 43.0 & 95.6  & 27.6 & 45.2 & 95.4  \\
& 500  & 29.6 & 53.6  & 48.4  & 28.9 & 52.6  & 49.9  & 32.6 & 54.4 & 53.4 & 33.9 & 76.1  & 100.0 & 34.0 & 77.0 & 99.8  & 35.6 & 77.4 & 100.0 \\
& 1000 & 38.9 & 79.3  & 84.4  & 40.9 & 79.7  & 82.9  & 45.6 & 83.0 & 81.2 & 43.0 & 92.6  & 100.0 & 45.6 & 91.4 & 100.0 & 49.8 & 93.0 & 100.0 \\\hline
\end{tabular}
\end{center}
\par
\textit{Notes}: The DGP is given by (\ref{dgpy}) with $\beta_{i1}=\beta_{i2}=0$, and contains two latent factors with different factor strengths, $(\alpha_{1},\alpha_{2})=(1,1)$, $(1,2/3)$, and $(2/3,1/2)$. $\lambda$ denotes the spatial autocorrelation coefficient of the error term defined in (\ref{error}). $m_0$ is the true number of factors and $\hat{m}$ is the number of selected PCs used to compute the different CD statistics.
$CD$ denotes the
standard test of error cross-sectional dependence defined by (\ref{cd}), $CD^\ast$ is the bias-corrected version defined by (\ref{CD*a}), and $CD_{W+}$ is the power-enhanced randomized version defined by (\ref{CDW+}).
\end{sidewaystable}

\begin{sidewaystable}
\scriptsize
\caption{Size and power of variance adjusted tests of error cross-sectional dependence for the panel regression model with one latent factor $(m_0=1)$ and serially correlated Gaussian errors}
\label{table.v.n.x1f}
\begin{center}
\begin{tabular}{cr|rrr|rrr|rrr|rrr|rrr|rrr}
$\hat{m}=1$& \multicolumn{19}{c}{}\\
\hline\hline
\multicolumn{1}{l}{}                               &      & \multicolumn{9}{c|}{Size ($H_o:\lambda=0$)}                                                                               & \multicolumn{9}{c}{Power ($H_1:\lambda=0.25$)}                                                                                \\\cline{3-20}
\multicolumn{1}{l}{}                               &      & \multicolumn{3}{c|}{$\alpha=1$} & \multicolumn{3}{c|}{$\alpha=2/3$} & \multicolumn{3}{c|}{$\alpha=1/2$} & \multicolumn{3}{c|}{$\alpha=1$} & \multicolumn{3}{c|}{$\alpha=2/3$} & \multicolumn{3}{c}{$\alpha=1/2$} \\\cline{3-20}
\multicolumn{1}{c}{Tests}                           & $n\setminus T$     & 100       & 200      & 500       & 100        & 200       & 500       & 100        & 200        & 500      & 100       & 200       & 500      & 100        & 200       & 500       & 100        & 200       & 500       \\\hline
\multirow{4}{*}{Variance
adjusted $CD$}                        & 100  & 44.1 & 73.2 & 94.5 & 8.0  & 5.6  & 11.5 & 6.4  & 5.6  & 6.6  & 14.5 & 24.7 & 43.0  & 62.8 & 78.4 & 93.3  & 69.0 & 85.3 & 97.5   \\
& 200  & 39.6 & 74.0 & 97.0 & 5.2  & 5.4  & 8.3  & 6.1  & 7.1  & 4.6  & 9.3  & 18.9 & 36.3  & 68.1 & 87.3 & 98.1  & 74.6 & 91.7 & 99.9   \\
& 500  & 32.9 & 77.8 & 96.3 & 8.9  & 5.1  & 5.0  & 8.7  & 6.0  & 4.9  & 5.3  & 13.4 & 29.4  & 75.4 & 90.8 & 99.6  & 77.8 & 94.5 & 99.8   \\
& 1000 & 27.7 & 76.2 & 77.2 & 8.4  & 6.4  & 5.3  & 7.3  & 5.3  & 4.0  & 5.0  & 9.3  & 27.9  & 80.0 & 93.1 & 99.7  & 79.0 & 95.4 & 100.0  \\\hline
\multirow{4}{*}{Variance adjusted
$CD^\ast$} & 100  & 5.0  & 4.1  & 5.4  & 8.6  & 5.9  & 5.8  & 7.0  & 6.2  & 5.4  & 40.8 & 62.1 & 88.3  & 71.7 & 90.9 & 99.7  & 74.6 & 91.5 & 99.7   \\
& 200  & 5.3  & 6.6  & 5.0  & 5.9  & 6.5  & 5.8  & 6.4  & 7.0  & 5.1  & 42.5 & 64.7 & 91.8  & 74.6 & 93.1 & 99.9  & 76.9 & 93.6 & 100.0  \\
& 500  & 5.7  & 5.1  & 4.9  & 9.5  & 6.3  & 4.8  & 8.8  & 6.7  & 5.0  & 48.1 & 65.1 & 93.8  & 78.1 & 93.6 & 99.9  & 78.4 & 95.3 & 99.9   \\
& 1000 & 8.9  & 4.0  & 4.2  & 9.2  & 6.8  & 5.1  & 7.3  & 5.2  & 4.2  & 51.7 & 66.8 & 93.5  & 81.4 & 94.9 & 100.0 & 79.3 & 95.7 & 100.0  \\\hline
\multirow{4}{*}{Variance adjusted $CD_{W+}$}                  & 100  & 14.4 & 16.1 & 12.5 & 14.7 & 15.8 & 15.0 & 13.6 & 16.2 & 16.1 & 19.6 & 35.6 & 88.5  & 20.3 & 39.3 & 93.8  & 19.1 & 40.5 & 94.1   \\
& 200  & 16.1 & 25.2 & 20.6 & 18.2 & 25.1 & 22.8 & 20.5 & 26.9 & 19.8 & 22.6 & 50.7 & 99.4  & 23.8 & 51.8 & 99.3  & 26.4 & 54.6 & 99.6   \\
& 500  & 26.6 & 50.1 & 50.7 & 26.8 & 48.8 & 52.6 & 26.0 & 54.1 & 50.8 & 31.0 & 75.4 & 100.0 & 31.4 & 75.9 & 100.0 & 30.0 & 78.5 & 100.0  \\
& 1000 & 29.0 & 75.5 & 76.0 & 34.1 & 74.3 & 85.6 & 34.4 & 78.2 & 82.1 & 35.2 & 90.0 & 100.0 & 36.1 & 90.1 & 100.0 & 37.9 & 93.4 & 100.0 \\\hline
\multicolumn{20}{c}{}\\
$\hat{m}=2$& \multicolumn{19}{c}{}\\
\hline\hline
\multicolumn{1}{l}{}                               &      & \multicolumn{9}{c|}{Size ($H_o:\lambda=0$)}                                                                               & \multicolumn{9}{c}{Power ($H_1:\lambda=0.25$)}                                                                                \\\cline{3-20}
\multicolumn{1}{l}{}                               &      & \multicolumn{3}{c|}{$\alpha=1$} & \multicolumn{3}{c|}{$\alpha=2/3$} & \multicolumn{3}{c|}{$\alpha=1/2$} & \multicolumn{3}{c|}{$\alpha=1$} & \multicolumn{3}{c|}{$\alpha=2/3$} & \multicolumn{3}{c}{$\alpha=1/2$}\\\cline{3-20}
\multicolumn{1}{c}{Tests}                           & $n\setminus T$     & 100       & 200      & 500       & 100        & 200       & 500       & 100        & 200        & 500      & 100       & 200       & 500      & 100        & 200       & 500       & 100        & 200       & 500       \\\hline
\multirow{4}{*}{Variance
adjusted $CD$}                        & 100  & 43.4 & 72.9 & 93.8 & 7.2  & 5.7  & 11.3 & 6.0  & 5.9  & 7.3  & 15.6 & 28.3 & 51.4  & 52.8 & 67.8 & 82.6  & 60.8 & 77.5 & 88.4   \\
& 200  & 40.1 & 73.7 & 96.9 & 4.7  & 5.8  & 8.0  & 5.4  & 6.8  & 4.7  & 10.4 & 20.3 & 38.6  & 62.6 & 84.0 & 96.3  & 68.4 & 87.4 & 98.9   \\
& 500  & 34.4 & 77.7 & 96.3 & 7.5  & 5.0  & 5.1  & 7.8  & 5.5  & 5.2  & 5.1  & 14.0 & 31.5  & 71.3 & 89.6 & 99.4  & 74.2 & 93.9 & 99.7   \\
& 1000 & 27.1 & 75.2 & 75.9 & 7.1  & 5.3  & 4.8  & 6.0  & 5.8  & 3.8  & 5.5  & 10.1 & 29.5  & 76.2 & 92.6 & 99.7  & 74.3 & 95.3 & 100.0  \\\hline
\multirow{4}{*}{Variance adjusted
$CD^\ast$} & 100  & 5.5  & 4.4  & 5.7  & 8.4  & 6.5  & 6.6  & 6.9  & 7.1  & 7.0  & 40.4 & 63.3 & 90.4  & 67.9 & 90.3 & 99.9  & 72.7 & 92.0 & 99.9   \\
& 200  & 5.2  & 6.0  & 5.8  & 6.0  & 6.9  & 5.3  & 6.0  & 6.9  & 5.3  & 41.4 & 64.5 & 92.4  & 72.1 & 92.6 & 100.0 & 74.6 & 91.9 & 100.0  \\
& 500  & 6.0  & 5.3  & 5.1  & 8.3  & 6.4  & 5.7  & 7.8  & 5.7  & 5.3  & 47.7 & 65.4 & 94.6  & 75.3 & 93.2 & 99.9  & 75.7 & 95.4 & 100.0  \\
& 1000 & 9.1  & 4.3  & 4.4  & 7.5  & 6.0  & 5.6  & 6.2  & 5.8  & 4.1  & 51.2 & 66.5 & 93.7  & 78.0 & 94.8 & 99.9  & 75.5 & 95.7 & 100.0  \\\hline
\multirow{4}{*}{Variance adjusted $CD_{W+}$}                  & 100  & 11.6 & 15.2 & 14.8 & 12.5 & 15.7 & 15.5 & 14.5 & 14.7 & 14.7 & 16.4 & 30.5 & 80.4  & 17.6 & 31.3 & 84.5  & 18.4 & 31.1 & 85.0   \\
& 200  & 19.3 & 22.3 & 19.7 & 18.3 & 21.9 & 20.8 & 15.3 & 23.8 & 22.6 & 21.8 & 44.1 & 97.3  & 24.0 & 46.0 & 97.6  & 20.4 & 48.8 & 97.0   \\
& 500  & 23.8 & 46.9 & 50.5 & 25.7 & 47.0 & 46.7 & 26.9 & 50.6 & 52.3 & 28.7 & 71.3 & 100.0 & 30.0 & 70.3 & 100.0 & 31.8 & 75.9 & 100.0  \\
& 1000 & 30.4 & 75.9 & 76.0 & 31.2 & 74.4 & 83.5 & 32.0 & 74.4 & 83.8 & 32.8 & 91.6 & 100.0 & 33.2 & 90.0 & 100.0 & 36.2 & 88.4 & 100.0  \\\hline
\end{tabular}
\end{center}
\par
\textit{Notes}: The DGP is given by (\ref{dgpy}) with $\beta_{i1}$ and $\beta_{i2}$ both generated from normal distribution, and contains a single latent factor with different factor strengths, $\alpha=1,$ $2/3$, and $1/2$. $\lambda$ denotes the spatial autocorrelation coefficient of the error term defined in (\ref{error}). $m_0$ is the true number of factors and $\hat{m}$ is the number of selected PCs used to compute the different CD statistics.
$CD$ denotes the
standard test of error cross-sectional dependence defined by (\ref{cd}), $CD^\ast$ is the bias-corrected version defined by (\ref{CD*a}), and $CD_{W+}$ is the power-enhanced randomized version defined by (\ref{CDW+}).
\end{sidewaystable}

\begin{sidewaystable}
\scriptsize
\caption{Size and power of variance adjusted tests of error cross-sectional dependence for the panel regression model with two latent factors $(m_0=2)$ and serially correlated Gaussian errors}
\label{table.v.n.x2f}
\begin{center}
\begin{tabular}{cr|rrr|rrr|rrr|rrr|rrr|rrr}
$\hat{m}=2$& \multicolumn{19}{c}{}\\
\hline\hline
\multicolumn{1}{l}{}                               &      & \multicolumn{9}{c|}{Size ($H_o:\lambda=0$)}                                                                               & \multicolumn{9}{c}{Power ($H_1:\lambda=0.25$)}                                                                                \\\cline{3-20}
\multicolumn{1}{l}{}                           &    &    \multicolumn{3}{c|}{$\alpha_1=1,\alpha_2=1$} & \multicolumn{3}{c|}{$\alpha_1=1,\alpha_2=2/3$} & \multicolumn{3}{c|}{$\alpha_1=2/3,\alpha_2=1/2$} & \multicolumn{3}{c|}{$\alpha_1=1,\alpha_2=1$} & \multicolumn{3}{c|}{$\alpha_1=1,\alpha_2=2/3$} & \multicolumn{3}{c}{$\alpha_1=2/3,\alpha_2=1/2$}  \\\cline{3-20}
\multicolumn{1}{c}{Tests}                           & $n\setminus T$     & 100       & 200      & 500       & 100        & 200       & 500       & 100        & 200        & 500      & 100       & 200       & 500      & 100        & 200       & 500       & 100        & 200       & 500       \\\hline
\multirow{4}{*}{Variance   adjusted $CD$}      & 100  & 98.6 & 100.0 & 100.0 & 89.0 & 98.6 & 100.0 & 6.3  & 8.1  & 18.1 & 93.2 & 99.5  & 99.9  & 68.4 & 90.3 & 98.6  & 54.7 & 65.3 & 76.1  \\
& 200  & 99.6 & 100.0 & 100.0 & 91.0 & 99.5 & 100.0 & 8.6  & 6.2  & 9.9  & 94.6 & 100.0 & 100.0 & 68.9 & 92.8 & 99.6  & 67.1 & 79.0 & 95.2  \\
& 500  & 99.2 & 100.0 & 100.0 & 88.3 & 99.9 & 100.0 & 7.5  & 6.1  & 6.2  & 95.3 & 100.0 & 100.0 & 61.8 & 95.7 & 100.0 & 75.1 & 91.6 & 98.9  \\
& 1000 & 97.4 & 100.0 & 100.0 & 84.5 & 99.8 & 100.0 & 11.3 & 7.3  & 3.5  & 91.2 & 100.0 & 100.0 & 55.5 & 96.1 & 100.0 & 81.1 & 90.8 & 100.0 \\\hline
\multirow{4}{*}{Variance adjusted   $CD^\ast$} & 100  & 4.4  & 4.4   & 5.6   & 6.4  & 4.9  & 4.2   & 9.0  & 9.0  & 6.0  & 18.0 & 26.0  & 37.6  & 26.4 & 34.6 & 57.6  & 71.8 & 90.2 & 99.4  \\
& 200  & 6.0  & 3.2   & 4.8   & 6.2  & 5.6  & 4.3   & 10.3 & 7.2  & 6.0  & 19.1 & 22.0  & 43.8  & 27.6 & 37.3 & 60.9  & 75.7 & 92.6 & 99.8  \\
& 500  & 8.4  & 5.2   & 4.2   & 7.6  & 5.4  & 5.2   & 9.0  & 8.6  & 6.0  & 25.2 & 26.3  & 44.3  & 34.0 & 42.3 & 65.2  & 78.8 & 95.5 & 99.9  \\
& 1000 & 15.0 & 4.8   & 5.6   & 13.6 & 6.6  & 5.7   & 12.5 & 8.9  & 5.3  & 35.6 & 27.4  & 45.4  & 40.3 & 43.5 & 68.0  & 82.9 & 92.0 & 100.0 \\\hline
\multirow{4}{*}{Variance adjusted $CD_{W+}$}   & 100  & 14.6 & 14.6  & 15.6  & 15.7 & 15.8 & 16.6  & 17.3 & 19.2 & 22.8 & 19.0 & 36.1  & 86.3  & 20.6 & 36.1 & 88.5  & 23.9 & 43.0 & 95.5  \\
& 200  & 19.2 & 23.6  & 23.8  & 18.9 & 25.1 & 22.8  & 23.8 & 28.6 & 25.9 & 24.4 & 50.3  & 98.3  & 26.0 & 50.5 & 98.8  & 30.3 & 56.6 & 99.6  \\
& 500  & 28.6 & 52.8  & 48.2  & 26.8 & 50.8 & 48.7  & 38.2 & 53.3 & 50.1 & 33.2 & 76.8  & 100.0 & 32.8 & 75.0 & 100.0 & 44.1 & 79.6 & 100.0 \\
& 1000 & 36.4 & 76.4  & 84.0  & 33.0 & 74.3 & 84.7  & 49.0 & 81.4 & 83.9 & 40.0 & 89.9  & 100.0 & 36.8 & 89.0 & 100.0 & 51.2 & 84.0 & 100.0 \\\hline
\multicolumn{20}{c}{}\\
$\hat{m}=4$& \multicolumn{19}{c}{}\\
\hline\hline
\multicolumn{1}{l}{}                               &      & \multicolumn{9}{c|}{Size ($H_o:\lambda=0$)}                                                                               & \multicolumn{9}{c}{Power ($H_1:\lambda=0.25$)}                                                                                \\\cline{3-20}
\multicolumn{1}{l}{}                           &    &    \multicolumn{3}{c|}{$\alpha_1=1,\alpha_2=1$} & \multicolumn{3}{c|}{$\alpha_1=1,\alpha_2=2/3$} & \multicolumn{3}{c|}{$\alpha_1=2/3,\alpha_2=1/2$} & \multicolumn{3}{c|}{$\alpha_1=1,\alpha_2=1$} & \multicolumn{3}{c|}{$\alpha_1=1,\alpha_2=2/3$} & \multicolumn{3}{c}{$\alpha_1=2/3,\alpha_2=1/2$}  \\\cline{3-20}
\multicolumn{1}{c}{Tests}                           & $n\setminus T$     & 100       & 200      & 500       & 100        & 200       & 500       & 100        & 200        & 500      & 100       & 200       & 500      & 100        & 200       & 500       & 100        & 200       & 500       \\\hline
\multirow{4}{*}{Variance   adjusted $CD$}      & 100  & 98.6 & 100.0 & 100.0 & 90.3 & 98.7 & 100.0 & 6.8  & 8.2  & 18.0 & 94.3 & 99.5  & 99.9  & 73.5 & 93.0 & 99.0  & 39.4 & 44.4 & 50.2  \\
& 200  & 99.0 & 100.0 & 100.0 & 92.2 & 99.6 & 100.0 & 5.4  & 4.7  & 12.2 & 96.5 & 100.0 & 100.0 & 72.0 & 94.7 & 99.6  & 58.0 & 73.5 & 85.7  \\
& 500  & 99.6 & 100.0 & 100.0 & 89.7 & 99.8 & 100.0 & 6.1  & 5.7  & 5.7  & 94.8 & 100.0 & 100.0 & 64.5 & 96.3 & 100.0 & 72.8 & 88.6 & 99.2  \\
& 1000 & 96.6 & 100.0 & 100.0 & 84.1 & 99.8 & 100.0 & 10.0 & 7.3  & 3.6  & 91.2 & 100.0 & 100.0 & 55.3 & 96.7 & 100.0 & 78.0 & 88.9 & 100.0 \\\hline
\multirow{4}{*}{Variance adjusted   $CD^\ast$} & 100  & 6.0  & 5.8   & 8.0   & 6.6  & 7.4  & 8.0   & 9.9  & 9.5  & 7.8  & 20.0 & 30.3  & 49.9  & 27.5 & 40.2 & 65.4  & 67.9 & 88.3 & 99.3  \\
& 200  & 6.6  & 3.6   & 5.2   & 7.4  & 6.5  & 5.6   & 7.8  & 7.5  & 6.3  & 19.8 & 24.2  & 49.0  & 27.8 & 39.3 & 67.6  & 72.5 & 90.8 & 99.8  \\
& 500  & 9.4  & 5.6   & 4.4   & 7.8  & 5.5  & 5.1   & 7.7  & 7.6  & 5.8  & 26.2 & 27.6  & 46.6  & 33.3 & 42.7 & 67.5  & 78.3 & 93.7 & 99.9  \\
& 1000 & 16.0 & 4.6   & 5.4   & 12.4 & 6.5  & 5.4   & 10.2 & 9.3  & 5.5  & 35.5 & 27.9  & 47.8  & 42.1 & 42.8 & 68.7  & 80.1 & 91.0 & 100.0 \\\hline
\multirow{4}{*}{Variance adjusted $CD_{W+}$}   & 100  & 13.8 & 17.2  & 31.8  & 10.8 & 16.0 & 30.2  & 13.5 & 17.9 & 23.8 & 16.2 & 26.3  & 81.2  & 14.2 & 26.9 & 78.1  & 16.8 & 31.6 & 78.4  \\
& 200  & 13.2 & 22.6  & 23.8  & 15.4 & 18.4 & 24.6  & 18.6 & 24.0 & 22.9 & 19.5 & 40.6  & 94.6  & 20.2 & 37.9 & 93.7  & 22.4 & 44.9 & 95.4  \\
& 500  & 27.0 & 46.8  & 46.8  & 24.5 & 46.3 & 49.6  & 28.5 & 50.4 & 49.5 & 30.3 & 70.1  & 99.9  & 27.3 & 69.5 & 100.0 & 32.4 & 72.9 & 100.0 \\
& 1000 & 32.6 & 75.2  & 82.4  & 33.3 & 73.5 & 81.0  & 38.5 & 78.9 & 82.4 & 36.9 & 88.4  & 100.0 & 35.8 & 87.2 & 100.0 & 44.7 & 79.9 & 100.0 \\\hline
\end{tabular}
\end{center}
\par
\textit{Notes}: The DGP is given by (\ref{dgpy}) with $\beta_{i1}$ and $\beta_{i2}$ both generated from normal distribution, and contains two latent factors with different factor strengths, $(\alpha_{1},\alpha_{2})=(1,1)$, $(1,2/3)$, and $(2/3,1/2)$. $\lambda$ denotes the spatial autocorrelation coefficient of the error term defined in (\ref{error}). $m_0$ is the true number of factors and $\hat{m}$ is the number of selected PCs used to compute the different CD statistics.
$CD$ denotes the
standard test of error cross-sectional dependence defined by (\ref{cd}), $CD^\ast$ is the bias-corrected version defined by (\ref{CD*a}), and $CD_{W+}$ is the power-enhanced randomized version defined by (\ref{CDW+}).
\end{sidewaystable}


\begin{sidewaystable}
\scriptsize
\caption{Size and power of variance adjusted tests of error cross-sectional dependence for the latent factor model with one factor $(m_0=1)$ and serially correlated non-Gaussian errors}
\label{table.v.chi2.1f}
\begin{center}
\begin{tabular}{cr|rrr|rrr|rrr|rrr|rrr|rrr}
$\hat{m}=1$& \multicolumn{19}{c}{}\\
\hline\hline
\multicolumn{1}{l}{}                               &      & \multicolumn{9}{c|}{Size ($H_o:\lambda=0$)}                                                                               & \multicolumn{9}{c}{Power ($H_1:\lambda=0.25$)}                                                                                \\\cline{3-20}
\multicolumn{1}{l}{}                               &      & \multicolumn{3}{c|}{$\alpha=1$} & \multicolumn{3}{c|}{$\alpha=2/3$} & \multicolumn{3}{c|}{$\alpha=1/2$} & \multicolumn{3}{c|}{$\alpha=1$} & \multicolumn{3}{c|}{$\alpha=2/3$} & \multicolumn{3}{c}{$\alpha=1/2$} \\\cline{3-20}
\multicolumn{1}{c}{Tests}                           & $n\setminus T$     & 100       & 200      & 500       & 100        & 200       & 500       & 100        & 200        & 500      & 100       & 200       & 500      & 100        & 200       & 500       & 100        & 200       & 500       \\\hline
\multirow{4}{*}{Variance adjusted $CD$}                       & 100  & 41.6  & 70.9  & 93.4 & 4.6  & 7.0   & 10.6 & 6.4  & 5.8   & 6.3  & 14.1  & 26.1  & 41.6  & 59.7  & 76.6  & 94.1  & 66.4  & 87.0  & 97.4   \\
& 200  & 38.4  & 73.2  & 96.2 & 6.0  & 5.7   & 7.1  & 6.3  & 5.6   & 5.2  & 9.3   & 16.4  & 34.6  & 67.9  & 85.3  & 98.8  & 72.2  & 91.8  & 99.8   \\
& 500  & 36.3  & 76.5  & 98.8 & 7.2  & 6.4   & 6.3  & 7.3  & 7.2   & 5.5  & 5.4   & 12.1  & 31.0  & 74.7  & 89.6  & 100.0 & 76.0  & 93.8  & 100.0  \\
& 1000 & 29.4  & 75.6  & 99.5 & 8.9  & 6.6   & 4.5  & 6.8  & 6.7   & 4.0  & 5.3   & 9.5   & 26.8  & 78.2  & 92.5  & 99.9  & 78.6  & 91.8  & 100.0  \\\hline
\multirow{4}{*}{Variance adjusted $CD^\ast$} & 100  & 5.4   & 5.0   & 3.3  & 5.8  & 5.7   & 4.2  & 6.7  & 5.5   & 5.3  & 41.0  & 61.3  & 91.7  & 70.4  & 89.8  & 99.9  & 72.1  & 93.1  & 99.9   \\
& 200  & 4.5   & 6.2   & 6.2  & 6.9  & 6.6   & 5.7  & 6.6  & 6.3   & 6.7  & 42.6  & 66.0  & 92.9  & 74.7  & 92.3  & 99.8  & 74.2  & 94.3  & 99.9   \\
& 500  & 5.1   & 5.2   & 4.7  & 8.4  & 6.8   & 4.6  & 7.5  & 7.2   & 5.6  & 44.1  & 66.1  & 93.8  & 76.8  & 93.1  & 100.0 & 76.9  & 94.1  & 100.0  \\
& 1000 & 7.6   & 5.9   & 4.0  & 9.3  & 6.9   & 4.1  & 6.8  & 6.8   & 4.2  & 48.9  & 67.9  & 93.8  & 79.5  & 94.8  & 100.0 & 79.1  & 92.0  & 100.0  \\\hline
\multirow{4}{*}{Variance adjusted $CD_{W+}$}                  & 100  & 25.3  & 19.5  & 14.1 & 30.1 & 23.4  & 15.6 & 28.2 & 23.0  & 20.1 & 37.7  & 47.5  & 89.8  & 42.2  & 55.2  & 94.5  & 42.6  & 58.4  & 96.2   \\
& 200  & 57.4  & 50.8  & 26.6 & 57.1 & 52.4  & 29.1 & 60.2 & 56.4  & 30.1 & 66.3  & 81.6  & 99.4  & 64.6  & 81.2  & 99.9  & 69.7  & 85.2  & 99.6   \\
& 500  & 95.6  & 96.5  & 75.9 & 95.0 & 95.3  & 77.8 & 95.7 & 96.3  & 79.8 & 96.4  & 99.5  & 100.0 & 96.4  & 99.5  & 100.0 & 97.4  & 99.9  & 100.0  \\
& 1000 & 100.0 & 100.0 & 99.3 & 99.9 & 100.0 & 99.5 & 99.9 & 100.0 & 99.6 & 100.0 & 100.0 & 100.0 & 100.0 & 100.0 & 100.0 & 100.0 & 100.0 & 100.0
\\\hline
\multicolumn{20}{c}{}\\
$\hat{m}=2$& \multicolumn{19}{c}{}\\
\hline\hline
\multicolumn{1}{l}{}                               &      & \multicolumn{9}{c|}{Size ($H_o:\lambda=0$)}                                                                               & \multicolumn{9}{c}{Power ($H_1:\lambda=0.25$)}                                                                                \\\cline{3-20}
\multicolumn{1}{l}{}                               &      & \multicolumn{3}{c|}{$\alpha=1$} & \multicolumn{3}{c|}{$\alpha=2/3$} & \multicolumn{3}{c|}{$\alpha=1/2$} & \multicolumn{3}{c|}{$\alpha=1$} & \multicolumn{3}{c|}{$\alpha=2/3$} & \multicolumn{3}{c}{$\alpha=1/2$} \\\cline{3-20}
\multicolumn{1}{c}{Tests}                           & $n\setminus T$     & 100       & 200      & 500       & 100        & 200       & 500       & 100        & 200        & 500      & 100       & 200       & 500      & 100        & 200       & 500       & 100        & 200       & 500       \\\hline
\multirow{4}{*}{Variance adjusted $CD$}                         & 100  & 42.9 & 71.5 & 93.4 & 6.0   & 6.4   & 10.7 & 5.3  & 5.2   & 7.2  & 14.1 & 29.8  & 46.6  & 53.4 & 66.7  & 82.2  & 60.4 & 78.0  & 88.5   \\
& 200  & 38.2 & 74.2 & 96.6 & 6.3   & 5.0   & 7.9  & 6.1  & 6.0   & 6.2  & 9.0  & 18.1  & 40.0  & 63.9 & 81.4  & 97.0  & 69.5 & 90.2  & 99.4   \\
& 500  & 35.9 & 76.2 & 98.8 & 7.6   & 6.0   & 6.3  & 7.9  & 7.3   & 6.0  & 5.4  & 11.7  & 33.6  & 73.0 & 88.3  & 99.8  & 75.8 & 92.7  & 100.0  \\
& 1000 & 28.4 & 75.2 & 99.4 & 9.3   & 6.4   & 4.2  & 7.1  & 6.8   & 4.1  & 4.9  & 9.1   & 27.4  & 77.6 & 91.8  & 99.9  & 76.9 & 91.6  & 100.0  \\\hline
\multirow{4}{*}{Variance adjusted $CD^\ast$}  & 100  & 5.5  & 5.4  & 4.2  & 7.2   & 6.2   & 5.2  & 6.5  & 6.0   & 6.4  & 40.5 & 62.9  & 92.4  & 69.9 & 89.7  & 99.7  & 72.1 & 93.3  & 99.9   \\
& 200  & 4.7  & 6.3  & 6.2  & 7.0   & 6.4   & 5.6  & 6.9  & 7.1   & 7.6  & 42.8 & 65.6  & 92.7  & 75.5 & 91.8  & 99.6  & 73.7 & 93.6  & 99.8   \\
& 500  & 5.6  & 5.2  & 5.2  & 9.2   & 7.1   & 4.8  & 8.3  & 7.4   & 5.7  & 45.0 & 65.7  & 94.3  & 77.9 & 92.3  & 100.0 & 76.7 & 93.8  & 100.0  \\
& 1000 & 8.9  & 6.0  & 4.5  & 9.9   & 7.0   & 4.3  & 7.2  & 6.8   & 4.3  & 50.2 & 67.6  & 93.5  & 80.5 & 93.9  & 100.0 & 77.4 & 91.8  & 100.0  \\\hline
\multirow{4}{*}{Variance adjusted $CD_{W+}$}                  & 100  & 21.2 & 17.9 & 17.3 & 20.9  & 21.0  & 16.2 & 25.2 & 21.1  & 19.8 & 30.0 & 39.0  & 81.3  & 31.8 & 45.3  & 84.7  & 35.0 & 44.4  & 88.8   \\
& 200  & 47.0 & 43.0 & 26.0 & 48.4  & 44.3  & 27.7 & 48.3 & 46.9  & 29.6 & 56.7 & 72.8  & 98.5  & 57.1 & 73.2  & 99.2  & 58.2 & 75.6  & 98.9   \\
& 500  & 93.1 & 92.9 & 75.0 & 93.0  & 92.6  & 72.4 & 94.1 & 94.5  & 75.2 & 95.0 & 98.9  & 100.0 & 95.1 & 98.6  & 100.0 & 94.9 & 99.2  & 100.0  \\
& 1000 & 99.8 & 99.9 & 99.2 & 100.0 & 100.0 & 99.1 & 99.9 & 100.0 & 99.6 & 99.9 & 100.0 & 100.0 & 99.9 & 100.0 & 100.0 & 99.9 & 100.0 & 100.0   \\\hline
\end{tabular}
\end{center}
\par
\textit{Notes}: The DGP is given by (\ref{dgpy}) with $\beta_{i1}=\beta_{i2}=0$ and contains a single latent factor with different factor strengths, $\alpha=1,$ $2/3$, and $1/2$. $\lambda$ denotes the spatial autocorrelation coefficient of the error term defined in (\ref{error}). $m_0$ is the true number of factors and $\hat{m}$ is the number of selected PCs used to compute the different CD statistics.
$CD$ denotes the
standard test of error cross-sectional dependence defined by (\ref{cd}), $CD^\ast$ is the bias-corrected version defined by (\ref{CD*a}), and $CD_{W+}$ is the power-enhanced randomized version defined by (\ref{CDW+}).
\end{sidewaystable}

\begin{sidewaystable}
\scriptsize
\caption{Size and power of variance adjusted tests of error cross-sectional dependence for the latent factor model with two factors $(m_0=2)$ and serially correlated non-Gaussian errors}
\label{table.v.chi2.2f}
\begin{center}
\begin{tabular}{cr|rrr|rrr|rrr|rrr|rrr|rrr}
$\hat{m}=2$& \multicolumn{19}{c}{}\\
\hline\hline
\multicolumn{1}{l}{}                               &      & \multicolumn{9}{c|}{Size ($H_o:\lambda=0$)}                                                                               & \multicolumn{9}{c}{Power ($H_1:\lambda=0.25$)}                                                                                \\\cline{3-20}
\multicolumn{1}{l}{}                           &    &    \multicolumn{3}{c|}{$\alpha_1=1,\alpha_2=1$} & \multicolumn{3}{c|}{$\alpha_1=1,\alpha_2=2/3$} & \multicolumn{3}{c|}{$\alpha_1=2/3,\alpha_2=1/2$} & \multicolumn{3}{c|}{$\alpha_1=1,\alpha_2=1$} & \multicolumn{3}{c|}{$\alpha_1=1,\alpha_2=2/3$} & \multicolumn{3}{c}{$\alpha_1=2/3,\alpha_2=1/2$}   \\\cline{3-20}
\multicolumn{1}{c}{Tests}                           & $n\setminus T$     & 100       & 200      & 500       & 100        & 200       & 500       & 100        & 200        & 500      & 100       & 200       & 500      & 100        & 200       & 500       & 100        & 200       & 500       \\\hline
\multirow{4}{*}{Variance adjusted   $CD$}    & 100  & 98.6 & 100.0 & 100.0 & 89.0 & 98.6 & 100.0 & 6.3  & 8.1  & 18.1 & 93.2 & 99.5  & 99.9  & 68.4 & 90.3 & 98.6  & 54.7 & 65.3 & 76.1  \\
& 200  & 99.6 & 100.0 & 100.0 & 91.0 & 99.5 & 100.0 & 8.6  & 6.2  & 9.9  & 94.6 & 100.0 & 100.0 & 68.9 & 92.8 & 99.6  & 67.1 & 79.0 & 95.2  \\
& 500  & 99.2 & 100.0 & 100.0 & 88.3 & 99.9 & 100.0 & 7.5  & 6.1  & 6.2  & 95.3 & 100.0 & 100.0 & 61.8 & 95.7 & 100.0 & 75.1 & 91.6 & 98.9  \\
& 1000 & 97.4 & 100.0 & 100.0 & 84.5 & 99.8 & 100.0 & 11.3 & 7.3  & 3.5  & 91.2 & 100.0 & 100.0 & 55.5 & 96.1 & 100.0 & 81.1 & 90.8 & 100.0 \\\hline
\multirow{4}{*}{Variance adjusted $CD^\ast$} & 100  & 4.4  & 4.4   & 5.6   & 6.4  & 4.9  & 4.2   & 9.0  & 9.0  & 6.0  & 18.0 & 26.0  & 37.6  & 26.4 & 34.6 & 57.6  & 71.8 & 90.2 & 99.4  \\
& 200  & 6.0  & 3.2   & 4.8   & 6.2  & 5.6  & 4.3   & 10.3 & 7.2  & 6.0  & 19.1 & 22.0  & 43.8  & 27.6 & 37.3 & 60.9  & 75.7 & 92.6 & 99.8  \\
& 500  & 8.4  & 5.2   & 4.2   & 7.6  & 5.4  & 5.2   & 9.0  & 8.6  & 6.0  & 25.2 & 26.3  & 44.3  & 34.0 & 42.3 & 65.2  & 78.8 & 95.5 & 99.9  \\
& 1000 & 15.0 & 4.8   & 5.6   & 13.6 & 6.6  & 5.7   & 12.5 & 8.9  & 5.3  & 35.6 & 27.4  & 45.4  & 40.3 & 43.5 & 68.0  & 82.9 & 92.0 & 100.0 \\\hline
\multirow{4}{*}{Variance adjusted $CD_{W+}$} & 100  & 14.6 & 14.6  & 15.6  & 15.7 & 15.8 & 16.6  & 17.3 & 19.2 & 22.8 & 19.0 & 36.1  & 86.3  & 20.6 & 36.1 & 88.5  & 23.9 & 43.0 & 95.5  \\
& 200  & 19.2 & 23.6  & 23.8  & 18.9 & 25.1 & 22.8  & 23.8 & 28.6 & 25.9 & 24.4 & 50.3  & 98.3  & 26.0 & 50.5 & 98.8  & 30.3 & 56.6 & 99.6  \\
& 500  & 28.6 & 52.8  & 48.2  & 26.8 & 50.8 & 48.7  & 38.2 & 53.3 & 50.1 & 33.2 & 76.8  & 100.0 & 32.8 & 75.0 & 100.0 & 44.1 & 79.6 & 100.0 \\
& 1000 & 36.4 & 76.4  & 84.0  & 33.0 & 74.3 & 84.7  & 49.0 & 81.4 & 83.9 & 40.0 & 89.9  & 100.0 & 36.8 & 89.0 & 100.0 & 51.2 & 84.0 & 100.0\\\hline
\multicolumn{20}{c}{}\\
$\hat{m}=4$& \multicolumn{19}{c}{}\\
\hline\hline
\multicolumn{1}{l}{}                               &      & \multicolumn{9}{c|}{Size ($H_o:\lambda=0$)}                                                                               & \multicolumn{9}{c}{Power ($H_1:\lambda=0.25$)}                                                                                \\\cline{3-20}
\multicolumn{1}{l}{}                           &    &    \multicolumn{3}{c|}{$\alpha_1=1,\alpha_2=1$} & \multicolumn{3}{c|}{$\alpha_1=1,\alpha_2=2/3$} & \multicolumn{3}{c|}{$\alpha_1=2/3,\alpha_2=1/2$} & \multicolumn{3}{c|}{$\alpha_1=1,\alpha_2=1$} & \multicolumn{3}{c|}{$\alpha_1=1,\alpha_2=2/3$} & \multicolumn{3}{c}{$\alpha_1=2/3,\alpha_2=1/2$}   \\\cline{3-20}
\multicolumn{1}{c}{Tests}                           & $n\setminus T$     & 100       & 200      & 500       & 100        & 200       & 500       & 100        & 200        & 500      & 100       & 200       & 500      & 100        & 200       & 500       & 100        & 200       & 500       \\\hline
\multirow{4}{*}{Variance adjusted   $CD$}    & 100  & 98.6 & 100.0 & 100.0 & 90.3 & 98.7 & 100.0 & 6.8  & 8.2  & 18.0 & 94.3 & 99.5  & 99.9  & 73.5 & 93.0 & 99.0  & 39.4 & 44.4 & 50.2  \\
& 200  & 99.0 & 100.0 & 100.0 & 92.2 & 99.6 & 100.0 & 5.4  & 4.7  & 12.2 & 96.5 & 100.0 & 100.0 & 72.0 & 94.7 & 99.6  & 58.0 & 73.5 & 85.7  \\
& 500  & 99.6 & 100.0 & 100.0 & 89.7 & 99.8 & 100.0 & 6.1  & 5.7  & 5.7  & 94.8 & 100.0 & 100.0 & 64.5 & 96.3 & 100.0 & 72.8 & 88.6 & 99.2  \\
& 1000 & 96.6 & 100.0 & 100.0 & 84.1 & 99.8 & 100.0 & 10.0 & 7.3  & 3.6  & 91.2 & 100.0 & 100.0 & 55.3 & 96.7 & 100.0 & 78.0 & 88.9 & 100.0 \\\hline
\multirow{4}{*}{Variance adjusted $CD^\ast$} & 100  & 6.0  & 5.8   & 8.0   & 6.6  & 7.4  & 8.0   & 9.9  & 9.5  & 7.8  & 20.0 & 30.3  & 49.9  & 27.5 & 40.2 & 65.4  & 67.9 & 88.3 & 99.3  \\
& 200  & 6.6  & 3.6   & 5.2   & 7.4  & 6.5  & 5.6   & 7.8  & 7.5  & 6.3  & 19.8 & 24.2  & 49.0  & 27.8 & 39.3 & 67.6  & 72.5 & 90.8 & 99.8  \\
& 500  & 9.4  & 5.6   & 4.4   & 7.8  & 5.5  & 5.1   & 7.7  & 7.6  & 5.8  & 26.2 & 27.6  & 46.6  & 33.3 & 42.7 & 67.5  & 78.3 & 93.7 & 99.9  \\
& 1000 & 16.0 & 4.6   & 5.4   & 12.4 & 6.5  & 5.4   & 10.2 & 9.3  & 5.5  & 35.5 & 27.9  & 46.5  & 42.1 & 42.8 & 68.7  & 80.1 & 91.0 & 100.0 \\\hline
\multirow{4}{*}{Variance adjusted $CD_{W+}$} & 100  & 13.8 & 17.2  & 31.8  & 10.8 & 16.0 & 30.2  & 13.5 & 17.9 & 23.8 & 16.2 & 26.3  & 81.2  & 14.2 & 26.9 & 78.1  & 16.8 & 31.6 & 78.4  \\
& 200  & 13.2 & 22.6  & 23.8  & 15.4 & 18.4 & 24.6  & 18.6 & 24.0 & 22.9 & 19.5 & 40.6  & 94.6  & 20.2 & 37.9 & 93.7  & 22.4 & 44.9 & 95.4  \\
& 500  & 27.0 & 46.8  & 46.8  & 24.5 & 46.3 & 49.6  & 28.5 & 50.4 & 49.5 & 30.3 & 70.1  & 99.9  & 27.3 & 69.5 & 100.0 & 32.4 & 72.9 & 100.0 \\
& 1000 & 32.6 & 75.2  & 82.4  & 33.3 & 73.5 & 81.0  & 38.5 & 78.9 & 82.4 & 36.9 & 88.4  & 100.0 & 35.8 & 87.2 & 100.0 & 44.7 & 79.9 & 100.0\\\hline
\end{tabular}
\end{center}
\par
\textit{Notes}: The DGP is given by (\ref{dgpy}) with $\beta_{i1}=\beta_{i2}=0$, and contains two latent factors with different factor strengths, $(\alpha_{1},\alpha_{2})=(1,1)$, $(1,2/3)$, and $(2/3,1/2)$. $\lambda$ denotes the spatial autocorrelation coefficient of the error term defined in (\ref{error}). $m_0$ is the true number of factors and $\hat{m}$ is the number of selected PCs used to compute the different CD statistics.
$CD$ denotes the
standard test of error cross-sectional dependence defined by (\ref{cd}), $CD^\ast$ is the bias-corrected version defined by (\ref{CD*a}), and $CD_{W+}$ is the power-enhanced randomized version defined by (\ref{CDW+}).
\end{sidewaystable}

\begin{sidewaystable}
\scriptsize
\caption{Size and power of variance adjusted tests of error cross-sectional dependence for the panel regression model with one latent factor $(m_0=1)$ and serially correlated non-Gaussian errors}
\label{table.v.chi2.x1f}
\begin{center}
\begin{tabular}{cr|rrr|rrr|rrr|rrr|rrr|rrr}
$\hat{m}=1$& \multicolumn{19}{c}{}\\
\hline\hline
\multicolumn{1}{l}{}                               &      & \multicolumn{9}{c|}{Size ($H_o:\lambda=0$)}                                                                               & \multicolumn{9}{c}{Power ($H_1:\lambda=0.25$)}                                                                                \\\cline{3-20}
\multicolumn{1}{l}{}                               &      & \multicolumn{3}{c|}{$\alpha=1$} & \multicolumn{3}{c|}{$\alpha=2/3$} & \multicolumn{3}{c|}{$\alpha=1/2$} & \multicolumn{3}{c|}{$\alpha=1$} & \multicolumn{3}{c|}{$\alpha=2/3$} & \multicolumn{3}{c}{$\alpha=1/2$}\\\cline{3-20}
\multicolumn{1}{c}{Tests}                           & $n\setminus T$     & 100       & 200      & 500       & 100        & 200       & 500       & 100        & 200        & 500      & 100       & 200       & 500      & 100        & 200       & 500       & 100        & 200       & 500       \\\hline
\multirow{4}{*}{Variance adjusted $CD$}                    & 100  & 40.0 & 70.6  & 93.3  & 6.1  & 7.0   & 13.3 & 7.2  & 6.0   & 7.3  & 15.9 & 24.7  & 44.2  & 60.5 & 79.6 & 92.6  & 70.4 & 86.6  & 98.2   \\
& 200  & 38.4 & 71.7  & 96.5  & 6.4  & 5.3   & 7.9  & 6.2  & 5.1   & 5.0  & 7.4  & 16.5  & 36.7  & 69.3 & 85.4 & 99.1  & 73.8 & 92.3  & 99.9   \\
& 500  & 32.6 & 75.5  & 99.2  & 8.0  & 6.3   & 5.9  & 9.4  & 7.0   & 3.8  & 6.3  & 13.4  & 29.9  & 76.7 & 91.3 & 99.7  & 76.9 & 94.7  & 100.0  \\
& 1000 & 28.3 & 69.0  & 69.0  & 8.3  & 4.8   & 4.5  & 7.1  & 5.4   & 4.8  & 5.5  & 9.1   & 27.1  & 80.5 & 93.9 & 100.0 & 78.0 & 96.0  & 100.0  \\\hline
\multirow{4}{*}{Variance adjusted $CD^\ast$} & 100  & 5.6  & 5.8   & 4.2   & 7.6  & 6.5   & 5.5  & 8.0  & 7.0   & 6.4  & 42.8 & 62.4  & 89.9  & 72.2 & 91.0 & 99.9  & 75.9 & 93.5  & 100.0  \\
& 200  & 4.6  & 5.1   & 6.1   & 8.1  & 5.7   & 5.6  & 6.7  & 5.6   & 5.4  & 43.1 & 64.9  & 92.6  & 75.8 & 91.7 & 99.9  & 75.6 & 94.7  & 100.0  \\
& 500  & 5.8  & 5.2   & 5.4   & 8.5  & 7.2   & 4.7  & 9.7  & 7.3   & 3.6  & 46.2 & 67.6  & 93.3  & 79.2 & 94.3 & 100.0 & 77.5 & 95.4  & 100.0  \\
& 1000 & 7.2  & 5.9   & 5.9   & 9.2  & 5.5   & 5.1  & 7.2  & 5.6   & 5.9  & 51.0 & 65.7  & 95.4  & 82.3 & 96.4 & 100.0 & 78.4 & 96.2  & 100.0  \\\hline
\multirow{4}{*}{Variance adjusted $CD_{W+}$}                 & 100  & 22.5 & 17.1  & 13.5  & 24.5 & 22.2  & 19.1 & 21.7 & 21.9  & 17.8 & 33.0 & 43.8  & 90.2  & 34.4 & 52.4 & 95.8  & 35.7 & 51.9  & 96.4   \\
& 200  & 48.2 & 44.6  & 27.9  & 51.9 & 48.0  & 26.4 & 52.7 & 48.0  & 28.8 & 61.0 & 74.4  & 99.3  & 64.1 & 77.8 & 100.0 & 62.9 & 78.2  & 99.9   \\
& 500  & 93.1 & 93.5  & 73.9  & 91.7 & 94.7  & 76.6 & 91.7 & 94.2  & 77.4 & 94.1 & 98.5  & 100.0 & 93.9 & 99.2 & 100.0 & 94.3 & 99.1  & 100.0  \\
& 1000 & 98.1 & 100.0 & 100.0 & 99.7 & 100.0 & 99.5 & 99.9 & 100.0 & 99.5 & 99.9 & 100.0 & 100.0 & 99.9 & 99.9 & 100.0 & 99.7 & 100.0 & 100.0
\\\hline
\multicolumn{20}{c}{}\\
$\hat{m}=2$& \multicolumn{19}{c}{}\\
\hline\hline
\multicolumn{1}{l}{}                               &      & \multicolumn{9}{c|}{Size ($H_o:\lambda=0$)}                                                                               & \multicolumn{9}{c}{Power ($H_1:\lambda=0.25$)}                                                                                \\\cline{3-20}
\multicolumn{1}{l}{}                               &      & \multicolumn{3}{c|}{$\alpha=1$} & \multicolumn{3}{c|}{$\alpha=2/3$} & \multicolumn{3}{c|}{$\alpha=1/2$} & \multicolumn{3}{c|}{$\alpha=1$} & \multicolumn{3}{c|}{$\alpha=2/3$} & \multicolumn{3}{c}{$\alpha=1/2$} \\\cline{3-20}
\multicolumn{1}{c}{Tests}                           & $n\setminus T$     & 100       & 200      & 500       & 100        & 200       & 500       & 100        & 200        & 500      & 100       & 200       & 500      & 100        & 200       & 500       & 100        & 200       & 500       \\\hline
\multirow{4}{*}{Variance adjusted $CD$}                       & 100  & 41.2 & 70.6  & 93.1  & 4.4  & 5.5   & 13.6 & 6.8  & 5.6   & 7.3  & 15.5 & 27.2  & 50.2  & 51.8 & 69.9  & 79.2  & 57.7 & 77.0  & 86.9   \\
& 200  & 39.2 & 72.4  & 97.0  & 5.7  & 4.6   & 6.2  & 5.2  & 3.9   & 5.3  & 8.7  & 17.9  & 38.0  & 63.2 & 81.6  & 96.7  & 67.9 & 89.8  & 99.2   \\
& 500  & 32.0 & 75.6  & 99.3  & 6.7  & 6.1   & 5.9  & 7.4  & 6.4   & 3.3  & 6.3  & 12.8  & 31.1  & 71.9 & 89.6  & 99.6  & 72.3 & 94.1  & 100.0  \\
& 1000 & 28.9 & 69.1  & 69.1  & 7.2  & 5.0   & 4.5  & 7.0  & 6.1   & 5.2  & 5.2  & 9.5   & 28.7  & 76.6 & 94.7  & 99.9  & 73.2 & 95.3  & 100.0  \\\hline
\multirow{4}{*}{Variance adjusted $CD^\ast$} & 100  & 5.4  & 4.9   & 5.2   & 6.7  & 6.7   & 6.8  & 8.6  & 7.2   & 7.0  & 40.5 & 61.4  & 89.5  & 69.1 & 89.6  & 99.8  & 70.6 & 92.3  & 99.9   \\
& 200  & 4.6  & 4.2   & 6.7   & 6.8  & 4.2   & 5.6  & 6.3  & 4.8   & 6.3  & 41.4 & 64.8  & 92.7  & 72.8 & 91.5  & 99.9  & 73.0 & 94.4  & 100.0  \\
& 500  & 6.1  & 5.2   & 4.7   & 7.6  & 6.9   & 4.9  & 7.8  & 6.8   & 4.0  & 48.5 & 66.5  & 93.6  & 75.0 & 92.6  & 100.0 & 73.3 & 95.3  & 100.0  \\
& 1000 & 8.1  & 6.3   & 6.3   & 8.0  & 5.3   & 4.4  & 7.1  & 6.0   & 5.5  & 51.0 & 65.4  & 95.2  & 79.4 & 95.5  & 100.0 & 78.4 & 96.2  & 100.0  \\\hline
\multirow{4}{*}{Variance adjusted $CD_{W+}$}                  & 100  & 21.0 & 16.9  & 14.7  & 17.4 & 18.9  & 16.2 & 19.1 & 17.8  & 17.8 & 25.6 & 36.7  & 80.4  & 26.0 & 41.3  & 86.5  & 28.8 & 40.3  & 86.3   \\
& 200  & 42.2 & 40.4  & 26.5  & 45.7 & 40.8  & 25.7 & 44.4 & 39.2  & 25.1 & 50.7 & 64.5  & 98.2  & 55.0 & 68.7  & 99.2  & 53.8 & 69.4  & 98.5   \\
& 500  & 88.5 & 91.7  & 72.7  & 89.2 & 92.3  & 74.3 & 86.6 & 92.0  & 73.7 & 89.8 & 98.2  & 100.0 & 91.2 & 97.8  & 100.0 & 89.9 & 98.3  & 100.0  \\
& 1000 & 98.0 & 100.0 & 100.0 & 99.7 & 100.0 & 99.3 & 99.4 & 100.0 & 99.0 & 99.3 & 100.0 & 100.0 & 99.4 & 100.0 & 100.0 & 99.7 & 100.0 & 100.0    \\\hline
\end{tabular}
\end{center}
\par
\textit{Notes}: The DGP is given by (\ref{dgpy}) with $\beta_{i1}$ and $\beta_{i2}$ both generated from normal distribution, and contains a single latent factor with different factor strengths, $\alpha=1,$ $2/3$, and $1/2$. $\lambda$ denotes the spatial autocorrelation coefficient of the error term defined in (\ref{error}). $m_0$ is the true number of factors and $\hat{m}$ is the number of selected PCs used to compute the different CD statistics.
$CD$ denotes the
standard test of error cross-sectional dependence defined by (\ref{cd}), $CD^\ast$ is the bias-corrected version defined by (\ref{CD*a}), and $CD_{W+}$ is the power-enhanced randomized version defined by (\ref{CDW+}).
\end{sidewaystable}

\begin{sidewaystable}
\scriptsize
\caption{Size and power of variance adjusted tests of error cross-sectional dependence for the panel regression model with two latent factors $(m_0=2)$ and serially correlated non-Gaussian errors}
\label{table.v.chi2.x2f}
\begin{center}
\begin{tabular}{cr|rrr|rrr|rrr|rrr|rrr|rrr}
$\hat{m}=2$& \multicolumn{19}{c}{}\\
\hline\hline
\multicolumn{1}{l}{}                               &      & \multicolumn{9}{c|}{Size ($H_o:\lambda=0$)}                                                                               & \multicolumn{9}{c}{Power ($H_1:\lambda=0.25$)}                                                                                \\\cline{3-20}
\multicolumn{1}{l}{}                           &    &    \multicolumn{3}{c|}{$\alpha_1=1,\alpha_2=1$} & \multicolumn{3}{c|}{$\alpha_1=1,\alpha_2=2/3$} & \multicolumn{3}{c|}{$\alpha_1=2/3,\alpha_2=1/2$} & \multicolumn{3}{c|}{$\alpha_1=1,\alpha_2=1$} & \multicolumn{3}{c|}{$\alpha_1=1,\alpha_2=2/3$} & \multicolumn{3}{c}{$\alpha_1=2/3,\alpha_2=1/2$}  \\\cline{3-20}
\multicolumn{1}{c}{Tests}                           & $n\setminus T$     & 100       & 200      & 500       & 100        & 200       & 500       & 100        & 200        & 500      & 100       & 200       & 500      & 100        & 200       & 500       & 100        & 200       & 500       \\\hline
\multirow{4}{*}{Variance   adjusted $CD$}      & 100  & 97.9 & 100.0 & 100.0 & 87.6 & 99.0  & 100.0 & 8.2  & 6.7   & 20.8  & 91.5 & 99.2  & 100.0 & 69.1 & 90.7  & 98.4  & 53.1 & 65.7 & 73.2  \\
& 200  & 98.4 & 100.0 & 100.0 & 88.1 & 100.0 & 100.0 & 7.4  & 6.0   & 11.6  & 93.2 & 99.9  & 100.0 & 67.6 & 93.1  & 99.7  & 66.7 & 79.6 & 93.0  \\
& 500  & 98.4 & 100.0 & 100.0 & 88.2 & 100.0 & 100.0 & 10.7 & 7.5   & 5.6   & 93.9 & 100.0 & 100.0 & 62.3 & 95.8  & 100.0 & 75.1 & 88.9 & 98.2  \\
& 1000 & 98.1 & 100.0 & 100.0 & 82.3 & 100.0 & 100.0 & 11.0 & 8.7   & 5.0   & 92.7 & 100.0 & 100.0 & 54.3 & 96.1  & 100.0 & 75.8 & 88.7 & 100.0 \\\hline
\multirow{4}{*}{Variance adjusted   $CD^\ast$} & 100  & 6.3  & 4.5   & 5.2   & 6.5  & 5.3   & 4.3   & 9.8  & 7.7   & 6.7   & 16.6 & 21.9  & 36.8  & 24.3 & 35.3  & 54.4  & 71.6 & 91.0 & 99.4  \\
& 200  & 5.9  & 4.1   & 5.1   & 7.3  & 4.9   & 4.1   & 10.0 & 8.2   & 5.6   & 20.7 & 24.9  & 40.4  & 27.1 & 34.8  & 63.7  & 75.3 & 91.9 & 100.0 \\
& 500  & 8.9  & 6.1   & 4.4   & 9.0  & 5.4   & 5.1   & 11.9 & 9.6   & 5.2   & 25.7 & 26.8  & 44.5  & 32.6 & 38.9  & 65.7  & 79.0 & 92.7 & 99.0  \\
& 1000 & 14.2 & 4.9   & 5.4   & 13.5 & 6.9   & 5.6   & 12.0 & 10.5  & 6.2   & 33.5 & 29.1  & 44.1  & 41.2 & 42.0  & 66.4  & 78.3 & 91.2 & 100.0 \\\hline
\multirow{4}{*}{Variance adjusted $CD_{W+}$}   & 100  & 21.9 & 19.4  & 14.0  & 22.1 & 19.2  & 18.9  & 25.6 & 24.6  & 25.0  & 28.2 & 40.6  & 83.1  & 27.9 & 40.3  & 90.3  & 37.2 & 53.9 & 95.8  \\
& 200  & 45.7 & 43.7  & 24.3  & 46.9 & 43.3  & 28.2  & 54.7 & 50.6  & 30.1  & 54.7 & 72.3  & 98.5  & 52.7 & 74.1  & 99.2  & 64.8 & 79.8 & 99.7  \\
& 500  & 90.1 & 94.7  & 74.3  & 89.1 & 93.4  & 74.0  & 93.2 & 95.8  & 77.5  & 92.0 & 98.1  & 100.0 & 91.4 & 98.7  & 100.0 & 94.1 & 98.8 & 99.9  \\
& 1000 & 99.9 & 100.0 & 100.0 & 99.7 & 100.0 & 100.0 & 99.7 & 100.0 & 100.0 & 99.5 & 100.0 & 100.0 & 99.8 & 100.0 & 100.0 & 88.5 & 96.4 & 100.0
\\\hline
\multicolumn{20}{c}{}\\
$\hat{m}=4$& \multicolumn{19}{c}{}\\
\hline\hline
\multicolumn{1}{l}{}                               &      & \multicolumn{9}{c|}{Size ($H_o:\lambda=0$)}                                                                               & \multicolumn{9}{c}{Power ($H_1:\lambda=0.25$)}                                                                                \\\cline{3-20}
\multicolumn{1}{l}{}                           &    &    \multicolumn{3}{c|}{$\alpha_1=1,\alpha_2=1$} & \multicolumn{3}{c|}{$\alpha_1=1,\alpha_2=2/3$} & \multicolumn{3}{c|}{$\alpha_1=2/3,\alpha_2=1/2$} & \multicolumn{3}{c|}{$\alpha_1=1,\alpha_2=1$} & \multicolumn{3}{c|}{$\alpha_1=1,\alpha_2=2/3$} & \multicolumn{3}{c}{$\alpha_1=2/3,\alpha_2=1/2$}  \\\cline{3-20}
\multicolumn{1}{c}{Tests}                           & $n\setminus T$     & 100       & 200      & 500       & 100        & 200       & 500       & 100        & 200        & 500      & 100       & 200       & 500      & 100        & 200       & 500       & 100        & 200       & 500       \\\hline
\multirow{4}{*}{Variance   adjusted $CD$}      & 100  & 97.8 & 100.0 & 100.0 & 88.9 & 99.3  & 100.0 & 6.8  & 6.2   & 20.9  & 93.3 & 99.4  & 100.0 & 73.2 & 93.4  & 99.3  & 40.1 & 45.4 & 49.5  \\
& 200  & 98.6 & 100.0 & 100.0 & 89.2 & 100.0 & 100.0 & 6.8  & 5.6   & 11.3  & 94.8 & 99.9  & 100.0 & 70.5 & 94.8  & 99.5  & 57.3 & 71.1 & 87.2  \\
& 500  & 98.4 & 100.0 & 100.0 & 87.9 & 100.0 & 100.0 & 8.5  & 6.8   & 5.1   & 94.3 & 100.0 & 100.0 & 62.0 & 96.4  & 100.0 & 71.3 & 85.8 & 97.9  \\
& 1000 & 97.5 & 100.0 & 100.0 & 81.1 & 99.8  & 100.0 & 10.1 & 7.9   & 5.1   & 91.5 & 100.0 & 100.0 & 51.8 & 96.2  & 100.0 & 73.8 & 87.3 & 100.0 \\\hline
\multirow{4}{*}{Variance adjusted   $CD^\ast$} & 100  & 6.7  & 6.7   & 7.1   & 7.4  & 7.1   & 8.0   & 9.1  & 9.0   & 8.8   & 22.9 & 29.4  & 47.0  & 28.3 & 40.5  & 64.6  & 67.9 & 89.1 & 99.2  \\
& 200  & 7.7  & 5.7   & 5.7   & 6.9  & 6.1   & 5.2   & 9.1  & 7.8   & 4.9   & 22.3 & 28.0  & 46.0  & 27.8 & 37.7  & 66.7  & 72.5 & 91.3 & 99.9  \\
& 500  & 9.0  & 7.0   & 5.2   & 9.9  & 6.0   & 6.5   & 10.5 & 8.8   & 6.4   & 26.8 & 28.6  & 45.4  & 34.9 & 41.1  & 69.3  & 77.5 & 90.4 & 98.5  \\
& 1000 & 16.0 & 6.5   & 5.2   & 15.1 & 6.6   & 5.3   & 11.4 & 9.2   & 6.8   & 36.1 & 29.0  & 44.5  & 43.0 & 44.4  & 66.5  & 75.1 & 89.3 & 100.0 \\\hline
\multirow{4}{*}{Variance adjusted $CD_{W+}$}   & 100  & 18.6 & 17.3  & 33.9  & 14.4 & 16.6  & 28.9  & 16.7 & 18.8  & 24.6  & 23.1 & 31.5  & 82.3  & 19.6 & 31.6  & 81.7  & 21.6 & 36.5 & 80.3  \\
& 200  & 35.2 & 36.8  & 23.9  & 34.9 & 35.0  & 26.5  & 38.3 & 37.4  & 26.3  & 39.4 & 57.4  & 95.8  & 41.4 & 58.8  & 96.5  & 49.0 & 63.9 & 96.5  \\
& 500  & 83.6 & 88.1  & 69.0  & 84.6 & 89.7  & 70.5  & 84.8 & 92.4  & 71.6  & 87.3 & 96.6  & 100.0 & 86.6 & 97.3  & 100.0 & 88.8 & 96.8 & 99.5  \\
& 1000 & 99.4 & 99.9  & 100.0 & 99.6 & 100.0 & 100.0 & 99.4 & 100.0 & 100.0 & 99.3 & 100.0 & 100.0 & 99.2 & 100.0 & 100.0 & 87.8 & 96.4 & 100.0 \\\hline
\end{tabular}
\end{center}
\par
\textit{Notes}: The DGP is given by (\ref{dgpy}) with $\beta_{i1}$ and $\beta_{i2}$ both generated from normal distribution, and contains two latent factors with different factor strengths, $(\alpha_{1},\alpha_{2})=(1,1)$, $(1,2/3)$, and $(2/3,1/2)$. $\lambda$ denotes the spatial autocorrelation coefficient of the error term defined in (\ref{error}). $m_0$ is the true number of factors and $\hat{m}$ is the number of selected PCs used to compute the different CD statistics.
$CD$ denotes the
standard test of error cross-sectional dependence defined by (\ref{cd}), $CD^\ast$ is the bias-corrected version defined by (\ref{CD*a}), and $CD_{W+}$ is the power-enhanced randomized version defined by (\ref{CDW+}).
\end{sidewaystable}


\begin{sidewaystable}
\scriptsize
\caption{Size and power of ARDL adjusted tests of error cross-sectional dependence for the latent factor model with one factor $(m_0=1)$ and serially correlated Gaussian errors}
\label{table.ardl.n.1f}
\begin{center}
\begin{tabular}{cr|rrr|rrr|rrr|rrr|rrr|rrr}
$\hat{m}=2$& \multicolumn{19}{c}{}\\
\hline\hline
\multicolumn{1}{l}{}                               &      & \multicolumn{9}{c|}{Size ($H_o:\lambda=0$)}                                                                               & \multicolumn{9}{c}{Power ($H_1:\lambda=0.25$)}                                                                                \\\cline{3-20}
\multicolumn{1}{l}{}                               &      & \multicolumn{3}{c|}{$\alpha=1$} & \multicolumn{3}{c|}{$\alpha=2/3$} & \multicolumn{3}{c|}{$\alpha=1/2$} & \multicolumn{3}{c|}{$\alpha=1$} & \multicolumn{3}{c|}{$\alpha=2/3$} & \multicolumn{3}{c}{$\alpha=1/2$} \\\cline{3-20}
\multicolumn{1}{c}{Tests}                           & $n\setminus T$     & 100       & 200      & 500       & 100        & 200       & 500       & 100        & 200        & 500      & 100       & 200       & 500      & 100        & 200       & 500       & 100        & 200       & 500       \\\hline
\multirow{4}{*}{ARDL
adjusted $CD$}                        & 100  & 68.9 & 91.3 & 98.7  & 6.5 & 9.0 & 20.8 & 6.2 & 5.7 & 9.4 & 28.6 & 45.2 & 67.0 & 54.3 & 70.5 & 82.6  & 68.3 & 80.6 & 87.8   \\                                                               & 200  & 70.6 & 94.4 & 99.8  & 5.1 & 7.0 & 14.3 & 5.7 & 5.5 & 5.7 & 19.4 & 35.4 & 61.0 & 69.3 & 88.8 & 98.1  & 81.4 & 95.7 & 99.4   \\                                                               & 500  & 71.2 & 96.5 & 100.0 & 4.5 & 5.1 & 7.4  & 7.3 & 6.0 & 5.1 & 12.8 & 27.3 & 55.6 & 79.9 & 95.9 & 100.0 & 88.6 & 98.0 & 100.0  \\                                                               & 1000 & 73.0 & 97.2 & 100.0 & 5.9 & 5.4 & 5.6  & 7.2 & 5.1 & 5.2 & 10.1 & 22.6 & 52.1 & 81.7 & 96.7 & 100.0 & 87.4 & 99.2 & 100.0  \\\hline
\multirow{4}{*}{ARDL adjusted
$CD^\ast$} & 100  & 5.0  & 5.3  & 5.9   & 6.8 & 6.3 & 5.8  & 7.4 & 6.9 & 7.1 & 57.1 & 83.9 & 99.1 & 82.3 & 98.0 & 100.0 & 85.0 & 98.5 & 100.0  \\                                                               & 200  & 5.5  & 5.4  & 5.4   & 5.8 & 5.9 & 4.7  & 6.3 & 6.3 & 5.5 & 59.1 & 84.9 & 99.4 & 84.6 & 98.3 & 100.0 & 88.7 & 99.2 & 100.0  \\                                                               & 500  & 5.0  & 5.5  & 5.3   & 5.5 & 5.3 & 5.5  & 7.4 & 6.7 & 5.2 & 60.7 & 85.5 & 99.5 & 86.0 & 98.6 & 100.0 & 90.4 & 98.9 & 100.0  \\                                                               & 1000 & 5.4  & 4.5  & 4.6   & 6.1 & 6.1 & 4.9  & 7.7 & 5.3 & 5.1 & 61.0 & 85.6 & 99.8 & 85.1 & 98.3 & 100.0 & 88.6 & 99.4 & 100.0  \\\hline
\multirow{4}{*}{ARDL adjusted $CD_{W+}$}                  & 100  & 5.6  & 5.4  & 7.2   & 5.5 & 5.6 & 5.2  & 6.5 & 5.3 & 7.5 & 7.0  & 6.5  & 28.9 & 6.3  & 7.9  & 34.0  & 6.7  & 8.1  & 31.8   \\                                                               & 200  & 5.0  & 6.1  & 5.4   & 4.4 & 5.3 & 5.7  & 4.2 & 5.0 & 5.5 & 6.0  & 6.5  & 36.2 & 5.1  & 7.1  & 41.7  & 4.7  & 7.0  & 38.3   \\                                                               & 500  & 6.3  & 5.1  & 4.8   & 5.4 & 4.6 & 5.1  & 4.3 & 5.2 & 5.7 & 5.5  & 6.3  & 45.7 & 5.6  & 5.7  & 46.8  & 5.1  & 6.6  & 46.1   \\                                                               & 1000 & 4.7  & 6.1  & 5.0   & 5.6 & 4.6 & 4.3  & 4.3 & 5.5 & 5.5 & 4.4  & 5.7  & 47.0 & 5.9  & 5.6  & 45.7  & 4.7  & 5.6  & 47.3
\\\hline
\multicolumn{20}{c}{}\\
$\hat{m}=4$& \multicolumn{19}{c}{}\\
\hline\hline
\multicolumn{1}{l}{}                               &      & \multicolumn{9}{c|}{Size ($H_o:\lambda=0$)}                                                                               & \multicolumn{9}{c}{Power ($H_1:\lambda=0.25$)}                                                                                \\\cline{3-20}
\multicolumn{1}{l}{}                               &      & \multicolumn{3}{c|}{$\alpha=1$} & \multicolumn{3}{c|}{$\alpha=2/3$} & \multicolumn{3}{c|}{$\alpha=1/2$} & \multicolumn{3}{c|}{$\alpha=1$} & \multicolumn{3}{c|}{$\alpha=2/3$} & \multicolumn{3}{c}{$\alpha=1/2$}\\\cline{3-20}
\multicolumn{1}{c}{Tests}                           & $n\setminus T$     & 100       & 200      & 500       & 100        & 200       & 500       & 100        & 200        & 500      & 100       & 200       & 500      & 100        & 200       & 500       & 100        & 200       & 500       \\\hline
\multirow{4}{*}{ARDL
adjusted $CD$}                        & 100  & 68.9 & 91.6 & 98.6  & 6.2 & 9.2 & 22.6 & 6.9 & 5.8 & 10.2 & 33.7 & 56.9 & 78.1 & 38.4 & 47.3 & 62.0  & 49.3 & 58.9 & 69.0   \\                                                               & 200  & 71.5 & 94.0 & 99.7  & 5.4 & 7.1 & 13.9 & 5.5 & 5.4 & 6.0  & 22.3 & 42.7 & 69.3 & 58.9 & 79.1 & 91.5  & 72.2 & 90.4 & 95.5   \\                                                               & 500  & 70.8 & 96.9 & 100.0 & 5.0 & 5.8 & 7.6  & 7.3 & 5.9 & 4.8  & 14.7 & 30.8 & 60.9 & 75.9 & 94.7 & 99.9  & 84.6 & 97.6 & 100.0  \\                                                               & 1000 & 72.4 & 97.5 & 100.0 & 5.9 & 5.9 & 6.0  & 6.9 & 5.3 & 5.7  & 10.8 & 24.7 & 54.6 & 78.6 & 95.8 & 100.0 & 84.1 & 98.6 & 100.0  \\\hline
\multirow{4}{*}{ARDL adjusted
$CD^\ast$} & 100  & 7.0  & 7.1  & 12.0  & 7.9 & 8.1 & 10.3 & 8.9 & 9.7 & 12.0 & 58.0 & 84.7 & 99.5 & 80.3 & 97.2 & 100.0 & 80.8 & 98.1 & 100.0  \\                                                               & 200  & 6.4  & 5.8  & 6.6   & 6.4 & 6.5 & 5.7  & 6.7 & 6.8 & 6.8  & 57.5 & 84.2 & 99.3 & 82.1 & 97.8 & 100.0 & 85.8 & 98.8 & 100.0  \\                                                               & 500  & 5.4  & 5.9  & 5.2   & 5.3 & 5.3 & 5.5  & 7.9 & 6.7 & 5.2  & 58.6 & 85.2 & 99.6 & 85.0 & 98.4 & 100.0 & 88.1 & 98.4 & 100.0  \\                                                               & 1000 & 4.9  & 4.9  & 4.7   & 6.3 & 6.1 & 5.3  & 7.3 & 5.3 & 5.6  & 59.9 & 85.0 & 99.8 & 82.8 & 97.9 & 100.0 & 85.7 & 99.0 & 100.0  \\\hline
\multirow{4}{*}{ARDL adjusted $CD_{W+}$}                  & 100  & 5.9  & 5.9  & 8.2   & 5.5 & 6.1 & 7.4  & 5.7 & 6.2 & 7.8  & 6.8  & 7.0  & 20.5 & 5.8  & 7.4  & 18.9  & 5.6  & 8.0  & 18.6   \\                                                               & 200  & 5.7  & 5.1  & 5.2   & 4.8 & 5.1 & 4.8  & 5.0 & 5.7 & 5.1  & 5.9  & 6.1  & 20.2 & 6.0  & 5.1  & 22.1  & 5.1  & 5.5  & 19.8   \\                                                               & 500  & 4.7  & 4.5  & 5.7   & 5.7 & 5.4 & 4.7  & 5.4 & 5.9 & 5.9  & 5.1  & 4.9  & 31.8 & 5.8  & 5.1  & 33.2  & 5.6  & 6.1  & 33.7   \\                                                               & 1000 & 6.7  & 4.9  & 4.5   & 5.7 & 4.4 & 4.6  & 6.5 & 5.1 & 5.5  & 5.8  & 5.8  & 40.2 & 6.4  & 4.8  & 40.7  & 6.2  & 5.9  & 38.2 \\\hline
\end{tabular}
\end{center}
\par
\textit{Notes}: The DGP is given by (\ref{dgpy}) with $\beta_{i1}=\beta_{i2}=0$ and contains a single latent factor with different factor strengths, $\alpha=1,$ $2/3$, and $1/2$. $\lambda$ denotes the spatial autocorrelation coefficient of the error term defined in (\ref{error}). $m_0$ is the true number of factors and $\hat{m}$ is the number of selected PCs used to compute the different CD statistics.
$CD$ denotes the
standard test of error cross-sectional dependence defined by (\ref{cd}), $CD^\ast$ is the bias-corrected version defined by (\ref{CD*a}), and $CD_{W+}$ is the power-enhanced randomized version defined by (\ref{CDW+}).
\end{sidewaystable}

\begin{sidewaystable}
\scriptsize
\caption{Size and power of ARDL adjusted tests of error cross-sectional dependence for the latent factor model with two factors $(m_0=2)$ and serially correlated Gaussian errors}
\label{table.ardl.n.2f}
\begin{center}
\begin{tabular}{cr|rrr|rrr|rrr|rrr|rrr|rrr}
$\hat{m}=4$& \multicolumn{19}{c}{}\\
\hline\hline
\multicolumn{1}{l}{}                               &      & \multicolumn{9}{c|}{Size ($H_o:\lambda=0$)}                                                                               & \multicolumn{9}{c}{Power ($H_1:\lambda=0.25$)}                                                                                \\\cline{3-20}
\multicolumn{1}{l}{}                           &    &    \multicolumn{3}{c|}{$\alpha_1=1,\alpha_2=1$} & \multicolumn{3}{c|}{$\alpha_1=1,\alpha_2=2/3$} & \multicolumn{3}{c|}{$\alpha_1=2/3,\alpha_2=1/2$} & \multicolumn{3}{c|}{$\alpha_1=1,\alpha_2=1$} & \multicolumn{3}{c|}{$\alpha_1=1,\alpha_2=2/3$} & \multicolumn{3}{c}{$\alpha_1=2/3,\alpha_2=1/2$}   \\\cline{3-20}
\multicolumn{1}{c}{Tests}                           & $n\setminus T$     & 100       & 200      & 500       & 100        & 200       & 500       & 100        & 200        & 500      & 100       & 200       & 500      & 100        & 200       & 500       & 100        & 200       & 500       \\\hline
\multirow{4}{*}{ARDL
adjusted $CD$}                        & 100  & 100.0 & 100.0 & 100.0 & 98.9 & 99.9  & 100.0 & 7.4  & 11.9 & 38.0 & 99.6  & 100.0 & 100.0 & 93.6 & 99.0 & 99.9  & 37.3 & 39.4 & 50.1   \\                                                               & 200  & 100.0 & 100.0 & 100.0 & 99.5 & 100.0 & 100.0 & 6.9  & 8.8  & 21.8 & 99.8  & 100.0 & 100.0 & 94.3 & 99.7 & 100.0 & 59.4 & 72.1 & 86.3   \\                                                               & 500  & 100.0 & 100.0 & 100.0 & 99.9 & 100.0 & 100.0 & 7.0  & 5.7  & 10.7 & 99.9  & 100.0 & 100.0 & 93.3 & 99.9 & 100.0 & 76.3 & 93.4 & 99.8   \\                                                               & 1000 & 100.0 & 100.0 & 100.0 & 99.8 & 100.0 & 100.0 & 7.1  & 5.5  & 7.5  & 100.0 & 100.0 & 100.0 & 93.8 & 99.9 & 100.0 & 81.9 & 96.8 & 100.0  \\\hline
\multirow{4}{*}{ARDL adjusted
$CD^\ast$} & 100  & 8.0   & 9.3   & 14.5  & 7.8  & 7.2   & 10.5  & 10.2 & 9.9  & 9.9  & 31.8  & 53.7  & 87.9  & 39.8 & 62.2 & 93.0  & 81.2 & 96.7 & 100.0  \\                                                               & 200  & 5.7   & 7.2   & 7.9   & 6.6  & 6.2   & 7.1   & 8.1  & 7.6  & 6.3  & 30.2  & 50.0  & 86.0  & 38.2 & 60.0 & 92.3  & 81.7 & 97.7 & 100.0  \\                                                               & 500  & 6.1   & 6.1   & 6.5   & 7.0  & 6.2   & 5.0   & 8.6  & 6.1  & 4.7  & 29.7  & 49.5  & 83.5  & 40.5 & 61.4 & 92.4  & 85.3 & 98.3 & 100.0  \\                                                               & 1000 & 5.5   & 6.0   & 5.9   & 6.3  & 6.1   & 5.0   & 8.2  & 7.2  & 6.4  & 28.3  & 48.1  & 81.8  & 40.4 & 60.9 & 92.7  & 86.4 & 98.7 & 100.0  \\\hline
\multirow{4}{*}{ARDL adjusted $CD_{W+}$}                  & 100  & 6.3   & 6.9   & 8.5   & 5.4  & 6.2   & 10.2  & 5.8  & 5.6  & 9.2  & 6.1   & 8.7   & 21.1  & 8.4  & 14.4 & 75.6  & 6.8  & 8.1  & 25.8   \\                                                               & 200  & 5.3   & 5.9   & 6.5   & 5.5  & 6.0   & 5.1   & 6.0  & 5.5  & 6.3  & 5.5   & 7.0   & 25.3  & 6.7  & 8.0  & 25.5  & 6.4  & 6.3  & 26.9   \\                                                               & 500  & 5.4   & 4.0   & 4.4   & 5.3  & 4.7   & 4.4   & 5.9  & 4.7  & 6.2  & 5.6   & 5.1   & 36.1  & 5.8  & 6.8  & 23.2  & 6.0  & 4.9  & 35.9   \\                                                               & 1000 & 5.4   & 5.3   & 5.0   & 5.1  & 5.9   & 5.1   & 5.2  & 5.4  & 6.3  & 5.6   & 5.1   & 40.7  & 5.7  & 6.2  & 33.6  & 5.4  & 6.1  & 40.6  \\\hline
\multicolumn{20}{c}{}\\
$\hat{m}=6$& \multicolumn{19}{c}{}\\
\hline\hline
\multicolumn{1}{l}{}                               &      & \multicolumn{9}{c|}{Size ($H_o:\lambda=0$)}                                                                               & \multicolumn{9}{c}{Power ($H_1:\lambda=0.25$)}                                                                                \\\cline{3-20}
\multicolumn{1}{l}{}                           &    &    \multicolumn{3}{c|}{$\alpha_1=1,\alpha_2=1$} & \multicolumn{3}{c|}{$\alpha_1=1,\alpha_2=2/3$} & \multicolumn{3}{c|}{$\alpha_1=2/3,\alpha_2=1/2$} & \multicolumn{3}{c|}{$\alpha_1=1,\alpha_2=1$} & \multicolumn{3}{c|}{$\alpha_1=1,\alpha_2=2/3$} & \multicolumn{3}{c}{$\alpha_1=2/3,\alpha_2=1/2$}   \\\cline{3-20}
\multicolumn{1}{c}{Tests}                           & $n\setminus T$     & 100       & 200      & 500       & 100        & 200       & 500       & 100        & 200        & 500      & 100       & 200       & 500      & 100        & 200       & 500       & 100        & 200       & 500       \\\hline
\multirow{4}{*}{ARDL
adjusted $CD$}                        & 100  & 100.0 & 100.0 & 100.0 & 98.8 & 99.9  & 100.0 & 7.4  & 12.4 & 38.6 & 99.7  & 100.0 & 100.0 & 95.3 & 99.3  & 99.9  & 27.0 & 29.7 & 50.3   \\                                                               & 200  & 100.0 & 100.0 & 100.0 & 99.6 & 100.0 & 100.0 & 6.5  & 8.9  & 22.4 & 100.0 & 100.0 & 100.0 & 95.0 & 99.8  & 100.0 & 47.8 & 60.6 & 70.8   \\                                                               & 500  & 100.0 & 100.0 & 100.0 & 99.9 & 100.0 & 100.0 & 6.5  & 5.1  & 10.5 & 100.0 & 100.0 & 100.0 & 94.2 & 99.9  & 100.0 & 72.9 & 91.1 & 99.6   \\                                                               & 1000 & 100.0 & 100.0 & 100.0 & 99.8 & 100.0 & 100.0 & 7.1  & 5.5  & 7.7  & 100.0 & 100.0 & 100.0 & 94.0 & 100.0 & 100.0 & 78.0 & 95.9 & 100.0  \\\hline
\multirow{4}{*}{ARDL adjusted
$CD^\ast$} & 100  & 11.4  & 15.6  & 30.3  & 9.8  & 12.5  & 24.1  & 12.3 & 13.2 & 17.3 & 37.0  & 62.5  & 93.1  & 42.7 & 69.2  & 96.0  & 78.5 & 95.5 & 100.0  \\                                                               & 200  & 7.5   & 9.1   & 12.7  & 6.9  & 7.2   & 10.5  & 8.8  & 7.9  & 8.6  & 32.1  & 53.4  & 89.8  & 38.8 & 61.5  & 94.3  & 78.6 & 97.0 & 100.0  \\                                                               & 500  & 6.6   & 6.6   & 7.4   & 7.1  & 6.8   & 5.6   & 8.3  & 5.9  & 5.6  & 29.8  & 50.2  & 84.6  & 39.5 & 60.7  & 92.9  & 82.9 & 98.0 & 100.0  \\                                                               & 1000 & 6.0   & 6.4   & 6.1   & 6.8  & 5.8   & 4.9   & 8.2  & 7.0  & 6.4  & 28.3  & 47.0  & 81.6  & 40.0 & 60.7  & 92.5  & 83.6 & 98.7 & 100.0  \\\hline
\multirow{4}{*}{ARDL adjusted $CD_{W+}$}                  & 100  & 5.5   & 7.4   & 16.2  & 6.1  & 6.9   & 16.5  & 5.5  & 6.3  & 10.4 & 5.9   & 8.4   & 26.2  & 6.7  & 8.1   & 27.8  & 5.8  & 8.3  & 21.5   \\                                                               & 200  & 5.5   & 5.5   & 6.5   & 5.9  & 5.5   & 6.7   & 5.4  & 5.7  & 6.3  & 6.2   & 6.5   & 17.8  & 6.8  & 6.4   & 16.5  & 6.3  & 6.6  & 14.3   \\                                                               & 500  & 5.6   & 4.9   & 5.1   & 6.2  & 4.7   & 5.1   & 6.3  & 6.6  & 4.7  & 5.1   & 6.3   & 27.9  & 6.6  & 4.9   & 27.8  & 6.5  & 6.4  & 26.6   \\                                                               & 1000 & 5.2   & 5.4   & 5.3   & 5.4  & 5.0   & 4.8   & 6.3  & 5.5  & 5.7  & 5.5   & 6.0   & 35.9  & 5.5  & 5.3   & 34.8  & 6.4  & 6.8  & 36.6   \\\hline
\end{tabular}
\end{center}
\par
\textit{Notes}: The DGP is given by (\ref{dgpy}) with $\beta_{i1}=\beta_{i2}=0$, and contains two latent factors with different factor strengths, $(\alpha_{1},\alpha_{2})=(1,1)$, $(1,2/3)$, and $(2/3,1/2)$. $\lambda$ denotes the spatial autocorrelation coefficient of the error term defined in (\ref{error}). $m_0$ is the true number of factors and $\hat{m}$ is the number of selected PCs used to compute the different CD statistics.
$CD$ denotes the
standard test of error cross-sectional dependence defined by (\ref{cd}), $CD^\ast$ is the bias-corrected version defined by (\ref{CD*a}), and $CD_{W+}$ is the power-enhanced randomized version defined by (\ref{CDW+}).
\end{sidewaystable}

\clearpage

\begin{sidewaystable}
\scriptsize
\caption{Size and power of ARDL adjusted tests of error cross-sectional dependence for the panel regression model with one latent factor $(m_0=1)$ and serially correlated Gaussian errors}
\label{table.ardl.n.x1f}
\begin{center}
\begin{tabular}{cr|rrr|rrr|rrr|rrr|rrr|rrr}
$\hat{m}=2$& \multicolumn{19}{c}{}\\
\hline\hline
\multicolumn{1}{l}{}                               &      & \multicolumn{9}{c|}{Size ($H_o:\lambda=0$)}                                                                               & \multicolumn{9}{c}{Power ($H_1:\lambda=0.25$)}                                                                                \\\cline{3-20}
\multicolumn{1}{l}{}                               &      & \multicolumn{3}{c|}{$\alpha=1$} & \multicolumn{3}{c|}{$\alpha=2/3$} & \multicolumn{3}{c|}{$\alpha=1/2$} & \multicolumn{3}{c|}{$\alpha=1$} & \multicolumn{3}{c|}{$\alpha=2/3$} & \multicolumn{3}{c}{$\alpha=1/2$} \\\cline{3-20}
\multicolumn{1}{c}{Tests}                           & $n\setminus T$     & 100       & 200      & 500       & 100        & 200       & 500       & 100        & 200        & 500      & 100       & 200       & 500      & 100        & 200       & 500       & 100        & 200       & 500       \\\hline
\multirow{4}{*}{ARDL
adjusted $CD$}                        & 100  & 69.1 & 91.0 & 98.9  & 7.0 & 9.2 & 22.1 & 8.1 & 7.4 & 9.4 & 29.5 & 49.6 & 67.5 & 58.3 & 70.8 & 82.5  & 66.6 & 79.9 & 86.2   \\                                                               & 200  & 69.6 & 93.5 & 99.9  & 7.0 & 6.6 & 11.9 & 7.4 & 6.8 & 5.6 & 18.5 & 34.7 & 60.2 & 71.6 & 89.6 & 98.4  & 82.2 & 94.9 & 99.5   \\                                                               & 500  & 67.2 & 96.4 & 100.0 & 7.1 & 4.6 & 7.9  & 8.0 & 6.6 & 5.1 & 11.4 & 25.4 & 53.4 & 82.8 & 96.6 & 100.0 & 87.2 & 98.7 & 100.0  \\                                                               & 1000 & 69.0 & 96.8 & 100.0 & 7.6 & 6.0 & 6.2  & 8.0 & 6.7 & 5.1 & 10.5 & 22.2 & 51.2 & 85.5 & 98.3 & 100.0 & 87.3 & 99.1 & 100.0  \\\hline
\multirow{4}{*}{ARDL adjusted
$CD^\ast$} & 100  & 5.7  & 6.1  & 7.8   & 7.9 & 7.6 & 6.9  & 9.4 & 8.6 & 7.4 & 55.8 & 81.6 & 98.9 & 82.6 & 98.1 & 100.0 & 85.3 & 99.0 & 100.0  \\                                                               & 200  & 5.1  & 5.7  & 5.1   & 8.3 & 6.7 & 5.3  & 8.4 & 8.3 & 5.4 & 58.7 & 84.3 & 99.6 & 85.9 & 98.4 & 100.0 & 88.4 & 98.9 & 100.0  \\                                                               & 500  & 6.5  & 6.2  & 5.2   & 8.6 & 6.5 & 5.8  & 8.0 & 7.4 & 5.7 & 61.8 & 85.6 & 99.6 & 87.6 & 99.0 & 100.0 & 89.2 & 99.1 & 100.0  \\                                                               & 1000 & 6.2  & 4.7  & 4.6   & 8.3 & 6.0 & 6.4  & 8.2 & 7.0 & 5.9 & 60.1 & 85.3 & 99.5 & 88.3 & 99.0 & 100.0 & 88.3 & 99.3 & 100.0  \\\hline
\multirow{4}{*}{ARDL adjusted $CD_{W+}$}                  & 100  & 5.4  & 6.4  & 6.4   & 5.3 & 5.9 & 5.9  & 5.6 & 5.5 & 6.3 & 5.4  & 8.0  & 29.1 & 6.7  & 7.8  & 33.7  & 6.8  & 8.5  & 32.7   \\                                                               & 200  & 6.5  & 5.9  & 5.3   & 6.6 & 6.4 & 5.0  & 5.3 & 5.7 & 5.5 & 7.0  & 7.4  & 34.7 & 7.3  & 8.2  & 40.7  & 5.9  & 6.4  & 39.0   \\                                                               & 500  & 5.8  & 5.6  & 4.7   & 5.4 & 6.2 & 5.1  & 5.1 & 5.4 & 4.4 & 6.0  & 5.8  & 44.3 & 5.5  & 6.6  & 47.5  & 5.3  & 6.0  & 46.4   \\                                                               & 1000 & 6.7  & 5.6  & 4.9   & 5.9 & 5.3 & 5.8  & 6.7 & 6.0 & 5.2 & 6.5  & 5.5  & 42.0 & 5.5  & 6.3  & 48.2  & 6.7  & 6.8  & 47.9
\\\hline
\multicolumn{20}{c}{}\\
$\hat{m}=4$& \multicolumn{19}{c}{}\\
\hline\hline
\multicolumn{1}{l}{}                               &      & \multicolumn{9}{c|}{Size ($H_o:\lambda=0$)}                                                                               & \multicolumn{9}{c}{Power ($H_1:\lambda=0.25$)}                                                                                \\\cline{3-20}
\multicolumn{1}{l}{}                               &      & \multicolumn{3}{c|}{$\alpha=1$} & \multicolumn{3}{c|}{$\alpha=2/3$} & \multicolumn{3}{c|}{$\alpha=1/2$} & \multicolumn{3}{c|}{$\alpha=1$} & \multicolumn{3}{c|}{$\alpha=2/3$} & \multicolumn{3}{c}{$\alpha=1/2$}\\\cline{3-20}
\multicolumn{1}{c}{Tests}                           & $n\setminus T$     & 100       & 200      & 500       & 100        & 200       & 500       & 100        & 200        & 500      & 100       & 200       & 500      & 100        & 200       & 500       & 100        & 200       & 500       \\\hline
\multirow{4}{*}{ARDL
adjusted $CD$}                        & 100  & 69.6 & 90.3 & 99.2  & 6.7 & 10.4 & 23.7 & 7.1 & 7.3  & 10.8 & 38.9 & 60.2 & 80.9 & 38.7 & 47.8 & 60.3  & 44.8 & 58.5 & 65.8   \\                                                               & 200  & 69.5 & 93.8 & 100.0 & 6.7 & 6.9  & 11.6 & 7.1 & 5.9  & 6.8  & 23.4 & 42.0 & 70.1 & 58.1 & 77.6 & 91.6  & 68.5 & 87.9 & 95.8   \\                                                               & 500  & 67.9 & 96.0 & 100.0 & 6.3 & 4.7  & 7.8  & 6.9 & 6.6  & 5.1  & 14.1 & 28.0 & 58.8 & 74.9 & 95.0 & 100.0 & 80.3 & 97.9 & 100.0  \\                                                               & 1000 & 68.5 & 96.6 & 100.0 & 7.0 & 5.7  & 6.1  & 8.0 & 6.5  & 5.8  & 12.2 & 24.8 & 54.7 & 77.9 & 97.3 & 100.0 & 81.7 & 98.8 & 100.0  \\\hline
\multirow{4}{*}{ARDL adjusted
$CD_\ast$} & 100  & 6.7  & 8.3  & 12.0  & 9.0 & 10.3 & 11.9 & 9.3 & 10.5 & 12.6 & 55.9 & 83.0 & 99.4 & 77.2 & 96.7 & 99.9  & 80.3 & 97.9 & 100.0  \\                                                               & 200  & 5.4  & 6.6  & 6.6   & 8.1 & 7.6  & 6.4  & 8.6 & 8.8  & 8.1  & 57.4 & 84.1 & 99.5 & 81.5 & 97.5 & 100.0 & 85.1 & 98.3 & 100.0  \\                                                               & 500  & 6.6  & 7.0  & 5.4   & 7.1 & 6.2  & 6.1  & 7.6 & 7.6  & 6.2  & 58.9 & 84.8 & 99.7 & 83.6 & 98.7 & 100.0 & 85.3 & 98.8 & 100.0  \\                                                               & 1000 & 6.6  & 5.4  & 5.1   & 7.6 & 6.1  & 6.5  & 8.2 & 7.1  & 6.2  & 59.1 & 84.2 & 99.6 & 82.6 & 98.7 & 100.0 & 83.7 & 98.9 & 100.0  \\\hline
\multirow{4}{*}{ARDL adjusted $CD_{W+}$}                  & 100  & 7.0  & 6.5  & 9.5   & 5.8 & 4.2  & 7.3  & 6.6 & 6.6  & 7.6  & 7.1  & 8.3  & 20.5 & 6.4  & 6.1  & 17.8  & 6.5  & 8.1  & 17.8   \\                                                               & 200  & 5.3  & 5.9  & 4.9   & 6.5 & 5.5  & 5.5  & 6.3 & 5.7  & 5.4  & 5.7  & 6.6  & 21.6 & 7.0  & 6.7  & 20.5  & 7.0  & 6.8  & 20.0   \\                                                               & 500  & 5.5  & 5.1  & 5.3   & 4.8 & 5.7  & 4.9  & 6.4 & 5.5  & 5.1  & 5.6  & 6.1  & 30.0 & 5.8  & 5.5  & 34.4  & 6.5  & 6.5  & 32.1   \\                                                               & 1000 & 6.0  & 5.6  & 5.5   & 5.8 & 5.4  & 6.3  & 5.6 & 5.0  & 5.5  & 6.3  & 6.5  & 39.7 & 5.4  & 5.0  & 40.3  & 5.8  & 5.1  & 38.5   \\\hline
\end{tabular}
\end{center}
\par
\textit{Notes}: The DGP is given by (\ref{dgpy}) with $\beta_{i1}$ and $\beta_{i2}$ both generated from normal distribution, and contains a single latent factor with different factor strengths, $\alpha=1,$ $2/3$, and $1/2$. $\lambda$ denotes the spatial autocorrelation coefficient of the error term defined in (\ref{error}). $m_0$ is the true number of factors and $\hat{m}$ is the number of selected PCs used to compute the different CD statistics.
$CD$ denotes the
standard test of error cross-sectional dependence defined by (\ref{cd}), $CD^\ast$ is the bias-corrected version defined by (\ref{CD*a}), and $CD_{W+}$ is the power-enhanced randomized version defined by (\ref{CDW+}).
\end{sidewaystable}

\begin{sidewaystable}
\scriptsize
\caption{Size and power of ARDL adjusted tests of error cross-sectional dependence for the panel regression model with two latent factors $(m_0=2)$ and serially correlated Gaussian errors}
\label{table.ardl.n.x2f}
\begin{center}
\begin{tabular}{cr|rrr|rrr|rrr|rrr|rrr|rrr}
$\hat{m}=4$& \multicolumn{19}{c}{}\\
\hline\hline
\multicolumn{1}{l}{}                               &      & \multicolumn{9}{c|}{Size ($H_o:\lambda=0$)}                                                                               & \multicolumn{9}{c}{Power ($H_1:\lambda=0.25$)}                                                                                \\\cline{3-20}
\multicolumn{1}{l}{}                           &    &    \multicolumn{3}{c|}{$\alpha_1=1,\alpha_2=1$} & \multicolumn{3}{c|}{$\alpha_1=1,\alpha_2=2/3$} & \multicolumn{3}{c|}{$\alpha_1=2/3,\alpha_2=1/2$} & \multicolumn{3}{c|}{$\alpha_1=1,\alpha_2=1$} & \multicolumn{3}{c|}{$\alpha_1=1,\alpha_2=2/3$} & \multicolumn{3}{c}{$\alpha_1=2/3,\alpha_2=1/2$}  \\\cline{3-20}
\multicolumn{1}{c}{Tests}                           & $n\setminus T$     & 100       & 200      & 500       & 100        & 200       & 500       & 100        & 200        & 500      & 100       & 200       & 500      & 100        & 200       & 500       & 100        & 200       & 500       \\\hline
\multirow{4}{*}{ARDL
adjusted $CD$}                        & 100  & 100.0 & 100.0 & 100.0 & 98.4  & 100.0 & 100.0 & 8.8  & 12.0 & 37.8 & 99.4  & 100.0 & 100.0 & 92.7 & 98.7 & 100.0 & 35.9 & 41.8 & 49.9   \\                                                               & 200  & 100.0 & 100.0 & 100.0 & 99.3  & 100.0 & 100.0 & 7.7  & 8.4  & 19.6 & 99.8  & 100.0 & 100.0 & 92.3 & 99.6 & 100.0 & 60.6 & 74.2 & 87.6   \\                                                               & 500  & 100.0 & 100.0 & 100.0 & 99.6  & 100.0 & 100.0 & 6.8  & 6.1  & 8.3  & 99.9  & 100.0 & 100.0 & 92.2 & 99.8 & 100.0 & 78.0 & 95.1 & 99.7   \\                                                               & 1000 & 100.0 & 100.0 & 100.0 & 100.0 & 100.0 & 100.0 & 8.3  & 6.0  & 6.9  & 100.0 & 100.0 & 100.0 & 91.6 & 99.9 & 100.0 & 82.0 & 97.1 & 100.0  \\\hline
\multirow{4}{*}{ARDL adjusted
$CD^\ast$} & 100  & 7.7   & 7.7   & 11.9  & 7.8   & 8.8   & 11.6  & 11.9 & 11.4 & 12.4 & 29.3  & 51.8  & 84.9  & 37.2 & 64.2 & 92.2  & 77.8 & 96.4 & 100.0  \\                                                               & 200  & 6.1   & 7.0   & 6.0   & 6.6   & 6.3   & 6.1   & 11.6 & 9.3  & 6.7  & 29.3  & 48.8  & 81.9  & 40.1 & 61.5 & 91.4  & 83.4 & 97.7 & 100.0  \\                                                               & 500  & 7.2   & 5.7   & 6.0   & 7.4   & 5.7   & 4.7   & 10.1 & 8.6  & 6.6  & 29.6  & 41.9  & 78.6  & 40.0 & 60.9 & 92.4  & 86.4 & 98.6 & 100.0  \\                                                               & 1000 & 5.9   & 5.8   & 5.2   & 7.3   & 5.6   & 5.5   & 9.0  & 8.4  & 6.5  & 27.1  & 41.6  & 76.6  & 42.6 & 62.1 & 91.9  & 86.5 & 98.4 & 100.0  \\\hline
\multirow{4}{*}{ARDL adjusted $CD_{W+}$}                  & 100  & 6.4   & 6.8   & 10.2  & 5.7   & 6.9   & 9.4   & 5.7  & 6.3  & 9.1  & 8.3   & 9.1   & 25.2  & 6.4  & 7.2  & 24.9  & 7.0  & 7.7  & 24.9   \\                                                               & 200  & 5.8   & 5.9   & 6.0   & 5.1   & 5.9   & 5.3   & 5.8  & 4.7  & 5.5  & 6.2   & 6.8   & 24.3  & 5.6  & 6.9  & 24.1  & 5.5  & 5.5  & 25.1   \\                                                               & 500  & 5.9   & 5.8   & 5.5   & 6.4   & 5.5   & 5.6   & 5.8  & 5.7  & 5.5  & 5.9   & 7.0   & 34.3  & 6.9  & 5.9  & 36.0  & 5.6  & 5.7  & 37.1   \\                                                               & 1000 & 5.9   & 5.2   & 4.9   & 6.5   & 5.6   & 6.2   & 5.5  & 5.6  & 5.7  & 6.2   & 6.1   & 38.7  & 6.6  & 6.7  & 38.2  & 5.9  & 6.1  & 40.1
\\\hline
\multicolumn{20}{c}{}\\
$\hat{m}=6$& \multicolumn{19}{c}{}\\
\hline\hline
\multicolumn{1}{l}{}                               &      & \multicolumn{9}{c|}{Size ($H_o:\lambda=0$)}                                                                               & \multicolumn{9}{c}{Power ($H_1:\lambda=0.25$)}                                                                                \\\cline{3-20}
\multicolumn{1}{l}{}                           &    &    \multicolumn{3}{c|}{$\alpha_1=1,\alpha_2=1$} & \multicolumn{3}{c|}{$\alpha_1=1,\alpha_2=2/3$} & \multicolumn{3}{c|}{$\alpha_1=2/3,\alpha_2=1/2$} & \multicolumn{3}{c|}{$\alpha_1=1,\alpha_2=1$} & \multicolumn{3}{c|}{$\alpha_1=1,\alpha_2=2/3$} & \multicolumn{3}{c}{$\alpha_1=2/3,\alpha_2=1/2$}   \\\cline{3-20}
\multicolumn{1}{c}{Tests}                           & $n\setminus T$     & 100       & 200      & 500       & 100        & 200       & 500       & 100        & 200        & 500      & 100       & 200       & 500      & 100        & 200       & 500       & 100        & 200       & 500       \\\hline
\multirow{4}{*}{ARDL
adjusted $CD$}                        & 100  & 100.0 & 100.0 & 100.0 & 98.5  & 100.0 & 100.0 & 9.0  & 13.0 & 38.7 & 99.8  & 100.0 & 100.0 & 93.5 & 99.3 & 100.0 & 23.9 & 32.2 & 48.7   \\                                                               & 200  & 100.0 & 100.0 & 100.0 & 99.4  & 100.0 & 100.0 & 8.3  & 8.5  & 20.2 & 99.8  & 100.0 & 100.0 & 94.0 & 99.8 & 100.0 & 48.2 & 62.1 & 73.2   \\                                                               & 500  & 100.0 & 100.0 & 100.0 & 99.7  & 100.0 & 100.0 & 7.3  & 5.7  & 8.4  & 100.0 & 100.0 & 100.0 & 92.4 & 99.8 & 100.0 & 70.7 & 91.5 & 99.5   \\                                                               & 1000 & 100.0 & 100.0 & 100.0 & 100.0 & 100.0 & 100.0 & 6.9  & 6.1  & 7.0  & 100.0 & 100.0 & 100.0 & 92.6 & 99.9 & 100.0 & 75.9 & 96.2 & 100.0  \\\hline
\multirow{4}{*}{ARDL adjusted
$CD^\ast$} & 100  & 10.1  & 13.8  & 28.7  & 10.5  & 14.2  & 26.9  & 13.1 & 15.3 & 19.1 & 32.0  & 59.4  & 91.6  & 41.9 & 69.9 & 96.2  & 73.7 & 95.8 & 100.0  \\                                                               & 200  & 7.4   & 8.0   & 9.7   & 8.0   & 8.0   & 9.0   & 11.4 & 10.4 & 9.3  & 30.1  & 50.9  & 86.9  & 40.4 & 62.9 & 93.7  & 78.4 & 96.2 & 100.0  \\                                                               & 500  & 7.8   & 5.7   & 6.8   & 8.0   & 6.0   & 5.4   & 9.5  & 8.7  & 7.3  & 29.3  & 42.4  & 79.4  & 39.8 & 61.2 & 93.2  & 82.1 & 98.3 & 100.0  \\                                                               & 1000 & 6.3   & 6.2   & 5.1   & 8.3   & 6.5   & 5.6   & 8.4  & 7.7  & 6.4  & 27.7  & 42.4  & 77.8  & 41.2 & 63.1 & 91.1  & 82.1 & 98.2 & 100.0  \\\hline
\multirow{4}{*}{ARDL adjusted $CD_{W+}$}                  & 100  & 6.3   & 6.8   & 22.5  & 6.9   & 7.2   & 16.9  & 5.5  & 7.1  & 10.7 & 6.9   & 8.6   & 32.1  & 7.9  & 8.4  & 28.0  & 6.7  & 8.9  & 20.8   \\                                                               & 200  & 5.9   & 5.9   & 5.9   & 6.5   & 5.5   & 6.4   & 6.6  & 5.3  & 5.6  & 6.2   & 7.0   & 16.8  & 6.7  & 5.9  & 15.8  & 6.9  & 6.2  & 16.0   \\                                                               & 500  & 7.4   & 5.8   & 5.5   & 6.7   & 5.8   & 5.7   & 5.7  & 4.9  & 4.2  & 6.9   & 6.3   & 28.4  & 7.1  & 6.5  & 26.4  & 6.5  & 5.7  & 25.5   \\                                                               & 1000 & 5.9   & 5.9   & 5.2   & 6.4   & 6.3   & 5.3   & 6.1  & 5.8  & 5.3  & 6.1   & 6.5   & 33.1  & 6.2  & 6.8  & 32.9  & 6.6  & 5.9  & 34.4   \\\hline
\end{tabular}
\end{center}
\par
\textit{Notes}: The DGP is given by (\ref{dgpy}) with $\beta_{i1}$ and $\beta_{i2}$ both generated from normal distribution, and contains two latent factors with different factor strengths, $(\alpha_{1},\alpha_{2})=(1,1)$, $(1,2/3)$, and $(2/3,1/2)$. $\lambda$ denotes the spatial autocorrelation coefficient of the error term defined in (\ref{error}). $m_0$ is the true number of factors and $\hat{m}$ is the number of selected PCs used to compute the different CD statistics.
$CD$ denotes the
standard test of error cross-sectional dependence defined by (\ref{cd}), $CD^\ast$ is the bias-corrected version defined by (\ref{CD*a}), and $CD_{W+}$ is the power-enhanced randomized version defined by (\ref{CDW+}).
\end{sidewaystable}

\begin{sidewaystable}
\scriptsize
\caption{Size and power of ARDL adjusted tests of error cross-sectional dependence for the latent factor model with one factor $(m_0=1)$ and serially correlated non-Gaussian errors}
\label{table.ardl.chi2.1f}
\begin{center}
\begin{tabular}{cr|rrr|rrr|rrr|rrr|rrr|rrr}
$\hat{m}=2$& \multicolumn{19}{c}{}\\
\hline\hline
\multicolumn{1}{l}{}                               &      & \multicolumn{9}{c|}{Size ($H_o:\lambda=0$)}                                                                               & \multicolumn{9}{c}{Power ($H_1:\lambda=0.25$)}                                                                                \\\cline{3-20}
\multicolumn{1}{l}{}                               &      & \multicolumn{3}{c|}{$\alpha=1$} & \multicolumn{3}{c|}{$\alpha=2/3$} & \multicolumn{3}{c|}{$\alpha=1/2$} & \multicolumn{3}{c|}{$\alpha=1$} & \multicolumn{3}{c|}{$\alpha=2/3$} & \multicolumn{3}{c}{$\alpha=1/2$} \\\cline{3-20}
\multicolumn{1}{c}{Tests}                           & $n\setminus T$     & 100       & 200      & 500       & 100        & 200       & 500       & 100        & 200        & 500      & 100       & 200       & 500      & 100        & 200       & 500       & 100        & 200       & 500       \\\hline
\multirow{4}{*}{ARDL
adjusted $CD$}                        & 100  & 62.2 & 86.3 & 97.2  & 6.7  & 9.7  & 23.6 & 6.6  & 6.4  & 10.0 & 28.5 & 46.4 & 67.0 & 55.4 & 69.2 & 81.4  & 67.0 & 79.1 & 86.4   \\      & 200  & 66.2 & 91.6 & 96.4  & 4.9  & 7.0  & 14.2 & 5.1  & 5.4  & 6.9  & 19.2 & 34.4 & 58.8 & 70.1 & 88.1 & 98.1  & 80.1 & 95.0 & 99.2   \\                                                         & 500  & 66.0 & 88.8 & 97.3  & 5.6  & 5.8  & 9.5  & 7.3  & 5.9  & 5.9  & 12.2 & 27.9 & 52.8 & 81.3 & 95.8 & 100.0 & 88.9 & 98.5 & 100.0  \\                                                           & 1000 & 63.1 & 96.8 & 100.0 & 6.4  & 5.7  & 7.4  & 7.2  & 5.2  & 5.0  & 9.6  & 21.7 & 50.4 & 82.4 & 97.0 & 100.0 & 87.6 & 99.3 & 100.0  \\\hline
\multirow{4}{*}{ARDL adjusted
$CD^\ast$} & 100  & 5.5  & 5.9  & 5.7   & 7.5  & 6.2  & 7.6  & 8.6  & 7.6  & 6.9  & 58.8 & 84.5 & 99.1 & 81.7 & 97.4 & 100.0 & 85.9 & 98.3 & 100.0  \\                                                               & 200  & 5.5  & 5.4  & 6.3   & 5.6  & 5.5  & 5.7  & 6.1  & 6.3  & 6.2  & 59.6 & 86.1 & 99.3 & 85.1 & 98.2 & 100.0 & 87.9 & 99.1 & 100.0  \\                                                               & 500  & 5.8  & 5.3  & 5.8   & 6.4  & 5.5  & 5.6  & 7.6  & 6.6  & 5.8  & 61.7 & 85.5 & 99.4 & 86.4 & 98.6 & 100.0 & 90.4 & 99.0 & 100.0  \\                                                               & 1000 & 5.5  & 6.2  & 5.4   & 7.0  & 6.0  & 4.8  & 7.4  & 5.6  & 5.8  & 62.3 & 85.8 & 99.5 & 86.4 & 98.8 & 100.0 & 88.7 & 99.6 & 100.0  \\\hline
\multirow{4}{*}{ARDL adjusted $CD_{W+}$}                  & 100  & 6.5  & 7.0  & 6.4   & 6.2  & 5.1  & 5.7  & 5.9  & 5.7  & 7.3  & 7.9  & 10.9 & 33.0 & 7.0  & 9.2  & 37.8  & 8.0  & 9.2  & 35.1   \\                                                               & 200  & 8.5  & 6.5  & 6.1   & 8.5  & 6.5  & 5.4  & 8.1  & 5.5  & 5.1  & 10.4 & 10.7 & 47.1 & 9.6  & 11.5 & 53.4  & 9.2  & 10.4 & 51.9   \\                                                               & 500  & 13.5 & 10.9 & 6.9   & 15.9 & 8.1  & 5.3  & 14.1 & 7.8  & 5.0  & 16.9 & 16.0 & 67.8 & 16.4 & 14.9 & 72.1  & 15.8 & 13.6 & 68.7   \\                                                               & 1000 & 24.3 & 13.9 & 5.9   & 27.4 & 14.2 & 7.0  & 26.0 & 16.4 & 6.5  & 25.2 & 19.9 & 81.7 & 26.3 & 20.7 & 81.4  & 24.7 & 23.8 & 82.0
\\\hline
\multicolumn{20}{c}{}\\
$\hat{m}=4$& \multicolumn{19}{c}{}\\
\hline\hline
\multicolumn{1}{l}{}                               &      & \multicolumn{9}{c|}{Size ($H_o:\lambda=0$)}                                                                               & \multicolumn{9}{c}{Power ($H_1:\lambda=0.25$)}                                                                                \\\cline{3-20}
\multicolumn{1}{l}{}                               &      & \multicolumn{3}{c|}{$\alpha=1$} & \multicolumn{3}{c|}{$\alpha=2/3$} & \multicolumn{3}{c|}{$\alpha=1/2$} & \multicolumn{3}{c|}{$\alpha=1$} & \multicolumn{3}{c|}{$\alpha=2/3$} & \multicolumn{3}{c}{$\alpha=1/2$} \\\cline{3-20}
\multicolumn{1}{c}{Tests}                           & $n\setminus T$     & 100       & 200      & 500       & 100        & 200       & 500       & 100        & 200        & 500      & 100       & 200       & 500      & 100        & 200       & 500       & 100        & 200       & 500       \\\hline
\multirow{4}{*}{ARDL
adjusted $CD$}                        & 100  & 63.6 & 86.2 & 97.5  & 6.7  & 9.5  & 24.6 & 7.2  & 6.7  & 10.7 & 35.2 & 60.2 & 79.8 & 38.6 & 46.8 & 61.5  & 50.3 & 59.0 & 66.0   \\                                                               & 200  & 66.5 & 90.9 & 96.3  & 5.4  & 7.6  & 14.5 & 4.8  & 5.5  & 6.7  & 20.8 & 41.9 & 67.9 & 60.0 & 78.7 & 90.7  & 70.1 & 87.4 & 94.8   \\                                                               & 500  & 64.4 & 88.4 & 97.3  & 6.1  & 5.6  & 9.2  & 6.9  & 6.3  & 5.8  & 13.8 & 29.5 & 58.7 & 77.5 & 94.7 & 99.9  & 85.0 & 97.9 & 100.0  \\                                                               & 1000 & 62.0 & 96.8 & 100.0 & 6.5  & 5.4  & 7.8  & 7.0  & 5.7  & 5.8  & 10.1 & 23.5 & 53.8 & 80.6 & 97.2 & 100.0 & 86.4 & 99.0 & 100.0  \\\hline
\multirow{4}{*}{ARDL adjusted
$CD^\ast$} & 100  & 7.6  & 7.0  & 10.7  & 9.8  & 8.5  & 11.3 & 9.0  & 8.3  & 8.6  & 58.0 & 83.8 & 99.7 & 78.9 & 96.9 & 100.0 & 83.2 & 97.3 & 100.0  \\                                                               & 200  & 7.2  & 6.6  & 7.1   & 6.5  & 6.3  & 7.1  & 6.1  & 6.4  & 6.6  & 59.9 & 86.1 & 99.1 & 81.7 & 97.7 & 100.0 & 85.0 & 98.7 & 100.0  \\                                                               & 500  & 6.4  & 5.3  & 5.7   & 7.0  & 6.4  & 5.6  & 7.2  & 7.0  & 5.7  & 61.3 & 85.6 & 99.6 & 85.5 & 98.3 & 100.0 & 87.9 & 98.8 & 100.0  \\                                                               & 1000 & 6.5  & 6.1  & 5.3   & 7.0  & 5.9  & 5.1  & 7.4  & 6.2  & 6.0  & 61.6 & 85.6 & 99.5 & 85.6 & 98.7 & 100.0 & 88.3 & 99.3 & 100.0  \\\hline
\multirow{4}{*}{ARDL adjusted $CD_{W+}$}                  & 100  & 6.7  & 6.1  & 7.4   & 6.7  & 7.0  & 7.4  & 6.3  & 5.8  & 8.0  & 8.2  & 8.3  & 22.0 & 7.6  & 8.0  & 18.7  & 6.5  & 7.3  & 18.5   \\                                                               & 200  & 6.8  & 5.9  & 6.3   & 6.3  & 6.3  & 6.9  & 7.0  & 6.1  & 5.3  & 7.6  & 7.5  & 28.2 & 7.7  & 9.7  & 29.5  & 8.6  & 7.6  & 26.2   \\                                                               & 500  & 13.5 & 10.4 & 6.1   & 14.1 & 8.4  & 6.4  & 14.0 & 9.3  & 5.7  & 15.4 & 13.4 & 52.7 & 15.5 & 12.3 & 54.7  & 14.3 & 13.2 & 57.5   \\                                                               & 1000 & 22.8 & 13.4 & 5.3   & 24.3 & 14.9 & 6.3  & 23.3 & 14.5 & 6.2  & 24.0 & 18.6 & 73.0 & 24.0 & 20.0 & 73.2  & 23.8 & 19.9 & 71.9 \\\hline
\end{tabular}
\end{center}
\par
\textit{Notes}: The DGP is given by (\ref{dgpy}) with $\beta_{i1}=\beta_{i2}=0$ and contains a single latent factor with different factor strengths, $\alpha=1,$ $2/3$, and $1/2$. $\lambda$ denotes the spatial autocorrelation coefficient of the error term defined in (\ref{error}). $m_0$ is the true number of factors and $\hat{m}$ is the number of selected PCs used to compute the different CD statistics.
$CD$ denotes the
standard test of error cross-sectional dependence defined by (\ref{cd}), $CD^\ast$ is the bias-corrected version defined by (\ref{CD*a}), and $CD_{W+}$ is the power-enhanced randomized version defined by (\ref{CDW+}).
\end{sidewaystable}

\begin{sidewaystable}
\scriptsize
\caption{Size and power of ARDL adjusted tests of error cross-sectional dependence for the latent factor model with two factors $(m_0=2)$ and serially correlated non-Gaussian errors}
\label{table.ardl.chi2.2f}
\begin{center}
\begin{tabular}{cr|rrr|rrr|rrr|rrr|rrr|rrr}
$\hat{m}=4$& \multicolumn{19}{c}{}\\
\hline\hline
\multicolumn{1}{l}{}                               &      & \multicolumn{9}{c|}{Size ($H_o:\lambda=0$)}                                                                               & \multicolumn{9}{c}{Power ($H_1:\lambda=0.25$)}                                                                                \\\cline{3-20}
\multicolumn{1}{l}{}                           &    &    \multicolumn{3}{c|}{$\alpha_1=1,\alpha_2=1$} & \multicolumn{3}{c|}{$\alpha_1=1,\alpha_2=2/3$} & \multicolumn{3}{c|}{$\alpha_1=2/3,\alpha_2=1/2$} & \multicolumn{3}{c|}{$\alpha_1=1,\alpha_2=1$} & \multicolumn{3}{c|}{$\alpha_1=1,\alpha_2=2/3$} & \multicolumn{3}{c}{$\alpha_1=2/3,\alpha_2=1/2$}  \\\cline{3-20}
\multicolumn{1}{c}{Tests}                           & $n\setminus T$     & 100       & 200      & 500       & 100        & 200       & 500       & 100        & 200        & 500      & 100       & 200       & 500      & 100        & 200       & 500       & 100        & 200       & 500       \\\hline
\multirow{4}{*}{ARDL
adjusted $CD$}                        & 100  & 100.0 & 100.0 & 100.0 & 98.6 & 100.0 & 100.0 & 8.3  & 14.6 & 39.4 & 99.6  & 100.0 & 100.0 & 92.4 & 98.4  & 99.9  & 36.6 & 41.7 & 51.0   \\                                                               & 200  & 100.0 & 100.0 & 100.0 & 99.4 & 100.0 & 100.0 & 5.9  & 8.6  & 24.1 & 99.8  & 100.0 & 100.0 & 92.6 & 99.5  & 100.0 & 58.4 & 70.3 & 84.7   \\                                                               & 500  & 100.0 & 100.0 & 100.0 & 99.7 & 100.0 & 100.0 & 7.5  & 5.2  & 11.5 & 99.9  & 100.0 & 100.0 & 93.2 & 99.7  & 100.0 & 78.0 & 93.3 & 99.8   \\                                                               & 1000 & 100.0 & 100.0 & 100.0 & 99.7 & 100.0 & 100.0 & 7.4  & 6.1  & 7.7  & 100.0 & 100.0 & 100.0 & 91.9 & 100.0 & 100.0 & 82.9 & 96.4 & 100.0  \\\hline
\multirow{4}{*}{ARDL adjusted
$CD^\ast$} & 100  & 9.0   & 8.9   & 14.1  & 8.0  & 8.9   & 12.1  & 10.8 & 9.0  & 10.4 & 32.8  & 54.9  & 87.1  & 41.0 & 64.1  & 92.6  & 79.4 & 95.8 & 100.0  \\                                                               & 200  & 6.5   & 7.1   & 8.0   & 6.4  & 6.9   & 7.1   & 8.2  & 7.4  & 7.0  & 29.8  & 49.2  & 86.8  & 39.1 & 59.7  & 91.1  & 82.5 & 97.5 & 100.0  \\                                                               & 500  & 6.4   & 6.1   & 6.2   & 6.6  & 6.0   & 5.6   & 9.6  & 7.9  & 6.5  & 31.3  & 53.1  & 85.0  & 42.4 & 62.4  & 91.8  & 86.3 & 98.6 & 100.0  \\                                                               & 1000 & 6.2   & 5.7   & 5.1   & 7.1  & 6.2   & 4.6   & 8.9  & 6.5  & 4.8  & 30.5  & 46.8  & 82.2  & 41.9 & 61.9  & 92.4  & 87.5 & 98.3 & 100.0  \\\hline
\multirow{4}{*}{ARDL adjusted $CD_{W+}$}                  & 100  & 6.9   & 7.4   & 8.1   & 6.0  & 6.4   & 12.1  & 6.1  & 6.4  & 10.4 & 7.9   & 9.7   & 25.0  & 6.4  & 8.3   & 29.4  & 6.7  & 8.8  & 26.4   \\                                                               & 200  & 8.1   & 6.3   & 5.2   & 7.8  & 6.9   & 6.0   & 7.8  & 6.7  & 6.4  & 9.7   & 9.9   & 31.8  & 9.4  & 9.0   & 32.6  & 9.0  & 8.3  & 34.0   \\                                                               & 500  & 13.0  & 7.6   & 5.2   & 15.0 & 8.6   & 5.7   & 14.7 & 8.4  & 5.5  & 13.6  & 11.9  & 59.4  & 15.4 & 12.0  & 56.7  & 14.6 & 13.8 & 59.3   \\                                                               & 1000 & 25.8  & 13.7  & 6.0   & 26.9 & 13.0  & 6.2   & 24.1 & 14.6 & 6.3  & 24.2  & 19.9  & 74.9  & 25.5 & 18.6  & 72.6  & 23.8 & 21.4 & 76.4  \\\hline
\multicolumn{20}{c}{}\\
$\hat{m}=6$& \multicolumn{19}{c}{}\\
\hline\hline
\multicolumn{1}{l}{}                               &      & \multicolumn{9}{c|}{Size ($H_o:\lambda=0$)}                                                                               & \multicolumn{9}{c}{Power ($H_1:\lambda=0.25$)}                                                                                \\\cline{3-20}
\multicolumn{1}{l}{}                           &    &    \multicolumn{3}{c|}{$\alpha_1=1,\alpha_2=1$} & \multicolumn{3}{c|}{$\alpha_1=1,\alpha_2=2/3$} & \multicolumn{3}{c|}{$\alpha_1=2/3,\alpha_2=1/2$} & \multicolumn{3}{c|}{$\alpha_1=1,\alpha_2=1$} & \multicolumn{3}{c|}{$\alpha_1=1,\alpha_2=2/3$} & \multicolumn{3}{c}{$\alpha_1=2/3,\alpha_2=1/2$}   \\\cline{3-20}
\multicolumn{1}{c}{Tests}                           & $n\setminus T$     & 100       & 200      & 500       & 100        & 200       & 500       & 100        & 200        & 500      & 100       & 200       & 500      & 100        & 200       & 500       & 100        & 200       & 500       \\\hline
\multirow{4}{*}{ARDL
adjusted $CD$}                        & 100  & 99.9  & 100.0 & 100.0 & 98.5 & 100.0 & 100.0 & 8.2  & 15.0 & 39.4 & 99.5  & 100.0 & 100.0 & 94.6 & 99.2  & 99.9  & 26.1 & 32.5 & 51.7   \\                                                               & 200  & 100.0 & 100.0 & 100.0 & 99.2 & 100.0 & 100.0 & 5.8  & 8.5  & 24.3 & 99.9  & 100.0 & 100.0 & 94.5 & 99.7  & 100.0 & 49.3 & 60.5 & 67.7   \\                                                               & 500  & 100.0 & 100.0 & 100.0 & 99.5 & 100.0 & 100.0 & 7.6  & 5.7  & 10.9 & 100.0 & 100.0 & 100.0 & 92.8 & 99.7  & 100.0 & 75.4 & 90.9 & 99.6   \\                                                               & 1000 & 100.0 & 100.0 & 100.0 & 99.6 & 100.0 & 100.0 & 8.0  & 5.9  & 7.4  & 100.0 & 100.0 & 100.0 & 91.9 & 100.0 & 100.0 & 80.4 & 95.3 & 100.0  \\\hline
\multirow{4}{*}{ARDL adjusted
$CD^\ast$} & 100  & 13.1  & 14.5  & 30.0  & 11.9 & 14.3  & 25.1  & 12.7 & 12.2 & 17.2 & 38.6  & 61.7  & 93.3  & 45.5 & 69.5  & 96.2  & 76.1 & 95.8 & 100.0  \\                                                               & 200  & 7.6   & 8.4   & 11.8  & 7.9  & 8.1   & 10.4  & 8.8  & 9.0  & 9.1  & 32.9  & 54.6  & 90.5  & 41.3 & 62.5  & 93.8  & 80.7 & 97.1 & 100.0  \\                                                               & 500  & 6.8   & 6.8   & 6.7   & 6.9  & 5.9   & 6.0   & 10.0 & 8.3  & 6.5  & 31.9  & 52.8  & 84.4  & 43.6 & 62.7  & 92.5  & 85.4 & 98.3 & 100.0  \\                                                               & 1000 & 7.1   & 6.2   & 5.7   & 8.3  & 7.0   & 4.9   & 9.3  & 6.6  & 4.9  & 33.0  & 47.4  & 82.4  & 43.0 & 62.8  & 92.8  & 86.0 & 98.2 & 100.0  \\\hline
\multirow{4}{*}{ARDL adjusted $CD_{W+}$}                  & 100  & 6.0   & 6.9   & 15.6  & 6.4  & 6.3   & 17.1  & 6.8  & 6.1  & 10.6 & 6.6   & 7.6   & 27.7  & 6.8  & 8.2   & 29.0  & 7.4  & 8.3  & 19.2   \\                                                               & 200  & 7.7   & 6.3   & 7.1   & 8.0  & 5.7   & 5.5   & 7.4  & 5.6  & 5.7  & 7.7   & 8.6   & 19.1  & 7.9  & 7.5   & 20.5  & 8.1  & 7.7  & 18.3   \\                                                               & 500  & 12.1  & 8.2   & 5.8   & 13.9 & 8.4   & 6.3   & 12.9 & 8.4  & 5.8  & 14.2  & 11.9  & 44.1  & 14.1 & 11.5  & 44.6  & 13.8 & 12.1 & 44.9   \\                                                               & 1000 & 24.0  & 12.8  & 5.9   & 26.5 & 14.4  & 5.5   & 23.2 & 13.5 & 6.9  & 24.2  & 18.2  & 67.7  & 24.9 & 19.9  & 64.9  & 24.2 & 18.5 & 67.7 \\\hline
\end{tabular}
\end{center}
\par
\textit{Notes}: The DGP is given by (\ref{dgpy}) with $\beta_{i1}=\beta_{i2}=0$, and contains two latent factors with different factor strengths, $(\alpha_{1},\alpha_{2})=(1,1)$, $(1,2/3)$, and $(2/3,1/2)$. $\lambda$ denotes the spatial autocorrelation coefficient of the error term defined in (\ref{error}). $m_0$ is the true number of factors and $\hat{m}$ is the number of selected PCs used to compute the different CD statistics.
$CD$ denotes the
standard test of error cross-sectional dependence defined by (\ref{cd}), $CD^\ast$ is the bias-corrected version defined by (\ref{CD*a}), and $CD_{W+}$ is the power-enhanced randomized version defined by (\ref{CDW+}).
\end{sidewaystable}

\begin{sidewaystable}
\scriptsize
\caption{Size and power of ARDL adjusted tests of error cross-sectional dependence for the panel regression model with one latent factor $(m_0=1)$ and serially correlated non-Gaussian errors}
\label{table.ardl.chi2.x1f}
\begin{center}
\begin{tabular}{cr|rrr|rrr|rrr|rrr|rrr|rrr}
$\hat{m}=2$& \multicolumn{19}{c}{}\\
\hline\hline
\multicolumn{1}{l}{}                               &      & \multicolumn{9}{c|}{Size ($H_o:\lambda=0$)}                                                                               & \multicolumn{9}{c}{Power ($H_1:\lambda=0.25$)}                                                                                \\\cline{3-20}
\multicolumn{1}{l}{}                               &      & \multicolumn{3}{c|}{$\alpha=1$} & \multicolumn{3}{c|}{$\alpha=2/3$} & \multicolumn{3}{c|}{$\alpha=1/2$} & \multicolumn{3}{c|}{$\alpha=1$} & \multicolumn{3}{c|}{$\alpha=2/3$} & \multicolumn{3}{c}{$\alpha=1/2$}\\\cline{3-20}
\multicolumn{1}{c}{Tests}                           & $n\setminus T$     & 100       & 200      & 500       & 100        & 200       & 500       & 100        & 200        & 500      & 100       & 200       & 500      & 100        & 200       & 500       & 100        & 200       & 500       \\\hline
\multirow{4}{*}{ARDL
adjusted $CD$}                        & 100  & 66.7 & 90.9 & 98.9  & 7.2  & 10.1 & 20.4 & 7.9  & 7.0  & 9.9 & 28.5 & 49.2 & 68.2 & 59.1 & 70.9 & 82.7  & 69.2 & 78.1 & 86.2   \\                                                               & 200  & 68.7 & 93.2 & 99.9  & 7.5  & 5.8  & 13.6 & 8.6  & 5.5  & 6.4 & 18.7 & 36.4 & 58.7 & 71.8 & 90.2 & 98.1  & 81.0 & 95.7 & 99.3   \\                                                            & 500  & 68.1 & 96.0 & 100.0 & 8.2  & 7.7  & 8.9  & 8.0  & 5.6  & 5.8 & 11.1 & 25.0 & 53.0 & 82.6 & 96.1 & 100.0 & 87.3 & 98.9 & 100.0  \\                                                               & 1000 & 66.6 & 95.9 & 100.0 & 7.2  & 5.6  & 5.2  & 7.9  & 6.9  & 4.6 & 10.9 & 21.6 & 51.7 & 84.8 & 97.9 & 100.0 & 89.5 & 98.8 & 100.0  \\\hline
\multirow{4}{*}{ARDL adjusted
$CD^\ast$} & 100  & 5.8  & 6.3  & 7.4   & 8.6  & 8.2  & 6.8  & 8.8  & 7.2  & 7.4 & 57.8 & 84.2 & 99.5 & 83.6 & 97.6 & 100.0 & 87.2 & 97.9 & 100.0  \\                                                               & 200  & 6.8  & 6.7  & 5.5   & 9.4  & 6.9  & 7.0  & 9.9  & 7.2  & 6.5 & 60.0 & 85.2 & 99.3 & 84.9 & 98.5 & 100.0 & 87.9 & 99.2 & 100.0  \\                                                               & 500  & 6.1  & 5.5  & 5.1   & 9.5  & 7.5  & 6.5  & 8.5  & 6.7  & 6.1 & 60.3 & 84.7 & 99.6 & 87.5 & 99.1 & 100.0 & 89.5 & 99.4 & 100.0  \\                                                               & 1000 & 6.2  & 5.7  & 5.0   & 8.5  & 6.8  & 5.1  & 8.2  & 7.4  & 4.8 & 60.6 & 85.7 & 99.6 & 88.0 & 99.1 & 100.0 & 90.6 & 98.9 & 100.0  \\\hline
\multirow{4}{*}{ARDL adjusted $CD_{W+}$}                  & 100  & 7.6  & 6.3  & 6.2   & 6.9  & 6.4  & 7.0  & 6.7  & 6.2  & 6.4 & 8.4  & 9.9  & 32.9 & 8.8  & 11.4 & 40.0  & 8.3  & 10.1 & 37.5   \\                                                               & 200  & 9.0  & 5.9  & 5.1   & 8.9  & 6.0  & 6.0  & 9.2  & 6.7  & 5.7 & 9.8  & 9.9  & 46.7 & 10.3 & 10.7 & 52.8  & 10.4 & 11.4 & 52.3   \\                                                               & 500  & 15.0 & 9.1  & 5.3   & 13.2 & 9.3  & 5.5  & 12.8 & 9.7  & 5.3 & 15.0 & 14.1 & 67.8 & 15.3 & 14.5 & 71.5  & 13.8 & 15.2 & 71.7   \\                                                               & 1000 & 25.9 & 14.0 & 6.2   & 23.8 & 15.6 & 5.9  & 22.8 & 14.2 & 5.8 & 24.1 & 19.4 & 79.9 & 23.3 & 22.3 & 79.3  & 22.9 & 20.7 & 78.0
\\\hline
\multicolumn{20}{c}{}\\
$\hat{m}=4$& \multicolumn{19}{c}{}\\
\hline\hline
\multicolumn{1}{l}{}                               &      & \multicolumn{9}{c|}{Size ($H_o:\lambda=0$)}                                                                               & \multicolumn{9}{c}{Power ($H_1:\lambda=0.25$)}                                                                                \\\cline{3-20}
\multicolumn{1}{l}{}                               &      & \multicolumn{3}{c|}{$\alpha=1$} & \multicolumn{3}{c|}{$\alpha=2/3$} & \multicolumn{3}{c|}{$\alpha=1/2$} & \multicolumn{3}{c|}{$\alpha=1$} & \multicolumn{3}{c|}{$\alpha=2/3$} & \multicolumn{3}{c}{$\alpha=1/2$} \\\cline{3-20}
\multicolumn{1}{c}{Tests}                           & $n\setminus T$     & 100       & 200      & 500       & 100        & 200       & 500       & 100        & 200        & 500      & 100       & 200       & 500      & 100        & 200       & 500       & 100        & 200       & 500       \\\hline
\multirow{4}{*}{ARDL
adjusted $CD$}                        & 100  & 68.1 & 91.0 & 98.7  & 6.9  & 10.2 & 22.4 & 7.8  & 7.2  & 10.4 & 37.8 & 61.2 & 80.5 & 38.1 & 46.9 & 60.4  & 49.4 & 57.4 & 66.1   \\                                                               & 200  & 67.9 & 93.2 & 99.8  & 6.7  & 6.6  & 13.0 & 8.4  & 6.0  & 7.4  & 23.0 & 43.1 & 69.7 & 59.0 & 78.4 & 89.8  & 68.2 & 88.6 & 94.3   \\                                                               & 500  & 67.8 & 95.8 & 100.0 & 8.0  & 7.0  & 9.0  & 7.4  & 6.7  & 5.7  & 12.8 & 29.3 & 57.8 & 75.1 & 93.8 & 99.9  & 80.8 & 97.8 & 100.0  \\                                                               & 1000 & 66.1 & 95.9 & 100.0 & 6.6  & 5.6  & 5.0  & 7.2  & 6.6  & 4.7  & 11.1 & 23.0 & 55.2 & 79.2 & 96.8 & 100.0 & 85.0 & 98.1 & 100.0  \\\hline
\multirow{4}{*}{ARDL adjusted
$CD^\ast$} & 100  & 7.6  & 8.0  & 11.8  & 10.0 & 9.8  & 10.3 & 10.7 & 10.0 & 12.9 & 56.4 & 84.6 & 99.7 & 79.9 & 96.6 & 100.0 & 82.0 & 97.0 & 100.0  \\                                                               & 200  & 7.8  & 7.2  & 6.6   & 8.4  & 7.1  & 8.5  & 9.8  & 8.1  & 8.6  & 57.4 & 84.8 & 99.4 & 80.7 & 97.5 & 100.0 & 84.3 & 98.4 & 100.0  \\                                                               & 500  & 6.8  & 5.6  & 5.8   & 9.1  & 7.7  & 6.9  & 7.9  & 7.8  & 6.5  & 59.9 & 84.4 & 99.5 & 84.0 & 98.6 & 100.0 & 85.5 & 99.2 & 100.0  \\                                                               & 1000 & 6.9  & 5.9  & 5.5   & 7.6  & 6.5  & 5.9  & 7.7  & 7.6  & 5.3  & 60.1 & 85.8 & 99.7 & 84.2 & 98.7 & 100.0 & 87.0 & 98.8 & 100.0  \\\hline
\multirow{4}{*}{ARDL adjusted $CD_{W+}$}                  & 100  & 7.0  & 6.7  & 7.4   & 6.3  & 6.1  & 7.8  & 6.6  & 6.6  & 8.3  & 7.8  & 7.6  & 21.8 & 8.2  & 8.7  & 20.2  & 8.1  & 8.0  & 20.7   \\                                                               & 200  & 7.3  & 5.9  & 5.9   & 8.3  & 5.8  & 6.3  & 7.9  & 5.8  & 6.0  & 8.6  & 8.5  & 27.2 & 9.0  & 8.6  & 29.3  & 8.7  & 8.2  & 25.5   \\                                                               & 500  & 13.1 & 9.8  & 5.1   & 14.4 & 7.8  & 6.2  & 11.9 & 8.7  & 5.5  & 14.4 & 13.6 & 53.8 & 14.6 & 12.0 & 55.3  & 13.7 & 12.2 & 54.4   \\                                                               & 1000 & 23.9 & 14.3 & 6.5   & 24.1 & 11.9 & 5.7  & 22.6 & 13.3 & 6.0  & 23.0 & 19.3 & 71.8 & 21.8 & 19.4 & 71.6  & 21.7 & 18.7 & 71.0  \\\hline
\end{tabular}
\end{center}
\par
\textit{Notes}: The DGP is given by (\ref{dgpy}) with $\beta_{i1}$ and $\beta_{i2}$ both generated from normal distribution, and contains a single latent factor with different factor strengths, $\alpha=1,$ $2/3$, and $1/2$. $\lambda$ denotes the spatial autocorrelation coefficient of the error term defined in (\ref{error}). $m_0$ is the true number of factors and $\hat{m}$ is the number of selected PCs used to compute the different CD statistics.
$CD$ denotes the
standard test of error cross-sectional dependence defined by (\ref{cd}), $CD^\ast$ is the bias-corrected version defined by (\ref{CD*a}), and $CD_{W+}$ is the power-enhanced randomized version defined by (\ref{CDW+}).
\end{sidewaystable}

\begin{sidewaystable}
\scriptsize
\caption{Size and power of ARDL adjusted tests of error cross-sectional dependence for the panel regression model with two latent factors $(m_0=2)$ and serially correlated non-Gaussian errors}
\label{table.ardl.chi2.x2f}
\begin{center}
\begin{tabular}{cr|rrr|rrr|rrr|rrr|rrr|rrr}
$\hat{m}=4$& \multicolumn{19}{c}{}\\
\hline\hline
\multicolumn{1}{l}{}                               &      & \multicolumn{9}{c|}{Size ($H_o:\lambda=0$)}                                                                               & \multicolumn{9}{c}{Power ($H_1:\lambda=0.25$)}                                                                                \\\cline{3-20}
\multicolumn{1}{l}{}                           &    &    \multicolumn{3}{c|}{$\alpha_1=1,\alpha_2=1$} & \multicolumn{3}{c|}{$\alpha_1=1,\alpha_2=2/3$} & \multicolumn{3}{c|}{$\alpha_1=2/3,\alpha_2=1/2$} & \multicolumn{3}{c|}{$\alpha_1=1,\alpha_2=1$} & \multicolumn{3}{c|}{$\alpha_1=1,\alpha_2=2/3$} & \multicolumn{3}{c}{$\alpha_1=2/3,\alpha_2=1/2$}  \\\cline{3-20}
\multicolumn{1}{c}{Tests}                           & $n\setminus T$     & 100       & 200      & 500       & 100        & 200       & 500       & 100        & 200        & 500      & 100       & 200       & 500      & 100        & 200       & 500       & 100        & 200       & 500       \\\hline
\multirow{4}{*}{ARDL
adjusted $CD$}                        & 100  & 99.9  & 100.0 & 100.0 & 98.4 & 100.0 & 100.0 & 8.2  & 12.3 & 38.5 & 99.4 & 100.0 & 100.0 & 91.5 & 98.6 & 100.0 & 36.2 & 43.2 & 50.4   \\                                                               & 200  & 100.0 & 100.0 & 100.0 & 99.3 & 100.0 & 100.0 & 8.3  & 7.0  & 21.5 & 99.6 & 100.0 & 100.0 & 92.2 & 99.7 & 100.0 & 59.8 & 75.4 & 85.4   \\                                                               & 500  & 100.0 & 100.0 & 100.0 & 99.6 & 100.0 & 100.0 & 7.5  & 7.0  & 10.7 & 99.9 & 100.0 & 100.0 & 91.5 & 99.5 & 100.0 & 77.3 & 94.8 & 99.8   \\                                                               & 1000 & 100.0 & 100.0 & 100.0 & 99.6 & 100.0 & 100.0 & 9.7  & 6.1  & 7.1  & 99.9 & 100.0 & 100.0 & 91.5 & 99.9 & 100.0 & 82.1 & 97.9 & 100.0  \\\hline
\multirow{4}{*}{ARDL adjusted
$CD^\ast$} & 100  & 8.4   & 8.4   & 12.0  & 8.3  & 9.4   & 9.6   & 11.2 & 12.1 & 9.3  & 33.3 & 50.8  & 83.4  & 40.5 & 64.4 & 92.1  & 79.4 & 96.5 & 100.0  \\                                                               & 200  & 6.4   & 6.8   & 6.5   & 7.5  & 7.8   & 6.3   & 11.5 & 9.2  & 8.1  & 29.1 & 46.2  & 80.7  & 39.8 & 60.7 & 91.8  & 82.6 & 98.4 & 100.0  \\                                                               & 500  & 7.0   & 6.7   & 5.8   & 6.6  & 5.6   & 5.4   & 10.1 & 9.2  & 8.2  & 30.3 & 45.0  & 80.1  & 43.5 & 61.9 & 91.9  & 85.9 & 98.3 & 100.0  \\                                                               & 1000 & 6.4   & 5.0   & 5.8   & 7.5  & 6.6   & 5.8   & 10.5 & 9.0  & 7.2  & 29.3 & 44.2  & 78.7  & 42.4 & 64.2 & 90.6  & 86.6 & 99.2 & 100.0  \\\hline
\multirow{4}{*}{ARDL adjusted $CD_{W+}$}                  & 100  & 7.2   & 7.2   & 10.3  & 7.3  & 6.6   & 10.2  & 6.1  & 6.4  & 9.9  & 7.9  & 9.4   & 27.9  & 7.6  & 8.6  & 29.4  & 7.7  & 10.2 & 28.1   \\                                                               & 200  & 8.5   & 6.2   & 6.3   & 7.6  & 7.0   & 5.9   & 7.6  & 5.9  & 5.9  & 9.8  & 9.3   & 32.8  & 10.0 & 9.5  & 32.0  & 8.0  & 9.0  & 34.0   \\                                                               & 500  & 12.1  & 8.8   & 5.6   & 12.2 & 8.2   & 5.8   & 12.8 & 9.0  & 4.8  & 13.9 & 13.2  & 59.5  & 14.7 & 12.5 & 57.6  & 14.5 & 14.0 & 60.7   \\                                                               & 1000 & 24.6  & 15.6  & 5.9   & 25.0 & 12.8  & 6.6   & 23.5 & 12.6 & 6.5  & 23.7 & 22.2  & 73.1  & 24.2 & 19.1 & 73.0  & 23.3 & 18.9 & 75.6
\\\hline
\multicolumn{20}{c}{}\\
$\hat{m}=6$& \multicolumn{19}{c}{}\\
\hline\hline
\multicolumn{1}{l}{}                               &      & \multicolumn{9}{c|}{Size ($H_o:\lambda=0$)}                                                                               & \multicolumn{9}{c}{Power ($H_1:\lambda=0.25$)}                                                                                \\\cline{3-20}
\multicolumn{1}{l}{}                           &    &    \multicolumn{3}{c|}{$\alpha_1=1,\alpha_2=1$} & \multicolumn{3}{c|}{$\alpha_1=1,\alpha_2=2/3$} & \multicolumn{3}{c|}{$\alpha_1=2/3,\alpha_2=1/2$} & \multicolumn{3}{c|}{$\alpha_1=1,\alpha_2=1$} & \multicolumn{3}{c|}{$\alpha_1=1,\alpha_2=2/3$} & \multicolumn{3}{c}{$\alpha_1=2/3,\alpha_2=1/2$}   \\\cline{3-20}
\multicolumn{1}{c}{Tests}                           & $n\setminus T$     & 100       & 200      & 500       & 100        & 200       & 500       & 100        & 200        & 500      & 100       & 200       & 500      & 100        & 200       & 500       & 100        & 200       & 500       \\\hline
\multirow{4}{*}{ARDL
adjusted $CD$}                        & 100  & 100.0 & 100.0 & 100.0 & 97.9 & 100.0 & 100.0 & 8.1  & 13.2 & 40.7 & 99.6 & 100.0 & 100.0 & 93.9 & 99.3 & 100.0 & 24.7 & 32.6 & 50.5   \\                                                               & 200  & 100.0 & 100.0 & 100.0 & 99.1 & 100.0 & 100.0 & 7.6  & 7.7  & 22.0 & 99.9 & 100.0 & 100.0 & 92.9 & 99.7 & 100.0 & 49.2 & 61.7 & 71.7   \\                                                               & 500  & 100.0 & 100.0 & 100.0 & 99.4 & 100.0 & 100.0 & 8.6  & 6.5  & 11.3 & 99.9 & 100.0 & 100.0 & 92.5 & 99.6 & 100.0 & 70.8 & 91.6 & 99.2   \\                                                               & 1000 & 100.0 & 100.0 & 100.0 & 99.5 & 100.0 & 100.0 & 8.9  & 6.1  & 6.8  & 99.9 & 100.0 & 100.0 & 91.4 & 99.9 & 100.0 & 76.6 & 96.6 & 100.0  \\\hline
\multirow{4}{*}{ARDL adjusted
$CD^\ast$} & 100  & 12.0  & 14.0  & 29.5  & 11.0 & 14.9  & 25.8  & 12.3 & 16.3 & 16.7 & 39.0 & 59.8  & 91.6  & 44.9 & 70.5 & 96.1  & 74.4 & 95.5 & 99.9   \\                                                               & 200  & 7.3   & 8.4   & 9.6   & 8.8  & 8.6   & 9.2   & 11.5 & 10.3 & 9.5  & 32.5 & 49.6  & 85.1  & 42.4 & 64.2 & 94.0  & 78.1 & 97.1 & 100.0  \\                                                               & 500  & 7.7   & 7.1   & 6.5   & 6.9  & 6.4   & 5.5   & 10.3 & 9.5  & 8.3  & 31.9 & 46.9  & 80.5  & 44.4 & 62.5 & 92.7  & 82.5 & 98.0 & 100.0  \\                                                               & 1000 & 7.0   & 5.3   & 6.5   & 7.9  & 6.7   & 5.9   & 9.8  & 8.3  & 7.1  & 30.5 & 45.1  & 79.5  & 43.5 & 63.5 & 91.1  & 83.5 & 98.4 & 100.0  \\\hline
\multirow{4}{*}{ARDL adjusted $CD_{W+}$}                  & 100  & 6.5   & 7.4   & 21.7  & 7.0  & 7.1   & 17.3  & 7.1  & 7.7  & 10.4 & 6.9  & 9.4   & 34.3  & 7.1  & 8.0  & 32.4  & 8.0  & 9.1  & 22.7   \\                                                               & 200  & 6.9   & 5.8   & 6.8   & 7.6  & 6.4   & 7.0   & 7.4  & 6.3  & 6.2  & 8.0  & 9.2   & 21.9  & 8.0  & 7.8  & 20.1  & 8.6  & 7.7  & 19.6   \\                                                               & 500  & 12.9  & 7.5   & 6.3   & 13.2 & 8.3   & 5.7   & 12.6 & 9.2  & 6.1  & 14.5 & 11.8  & 45.6  & 12.6 & 10.8 & 46.9  & 13.7 & 12.4 & 44.9   \\                                                               & 1000 & 23.5  & 14.9  & 5.6   & 24.8 & 12.7  & 6.3   & 23.4 & 12.5 & 6.2  & 23.5 & 19.7  & 66.5  & 24.6 & 18.1 & 65.0  & 22.5 & 18.5 & 64.7    \\\hline
\end{tabular}
\end{center}
\par
\textit{Notes}: The DGP is given by (\ref{dgpy}) with $\beta_{i1}$ and $\beta_{i2}$ both generated from normal distribution, and contains two latent factors with different factor strengths, $(\alpha_{1},\alpha_{2})=(1,1)$, $(1,2/3)$, and $(2/3,1/2)$. $\lambda$ denotes the spatial autocorrelation coefficient of the error term defined in (\ref{error}). $m_0$ is the true number of factors and $\hat{m}$ is the number of selected PCs used to compute the different CD statistics.
$CD$ denotes the
standard test of error cross-sectional dependence defined by (\ref{cd}), $CD^\ast$ is the bias-corrected version defined by (\ref{CD*a}), and $CD_{W+}$ is the power-enhanced randomized version defined by (\ref{CDW+}).
\end{sidewaystable}

\clearpage

\bigskip%
\begin{figure}[ptb]%
\centering
\includegraphics[
height=3in,
width=6.8761in
]%
{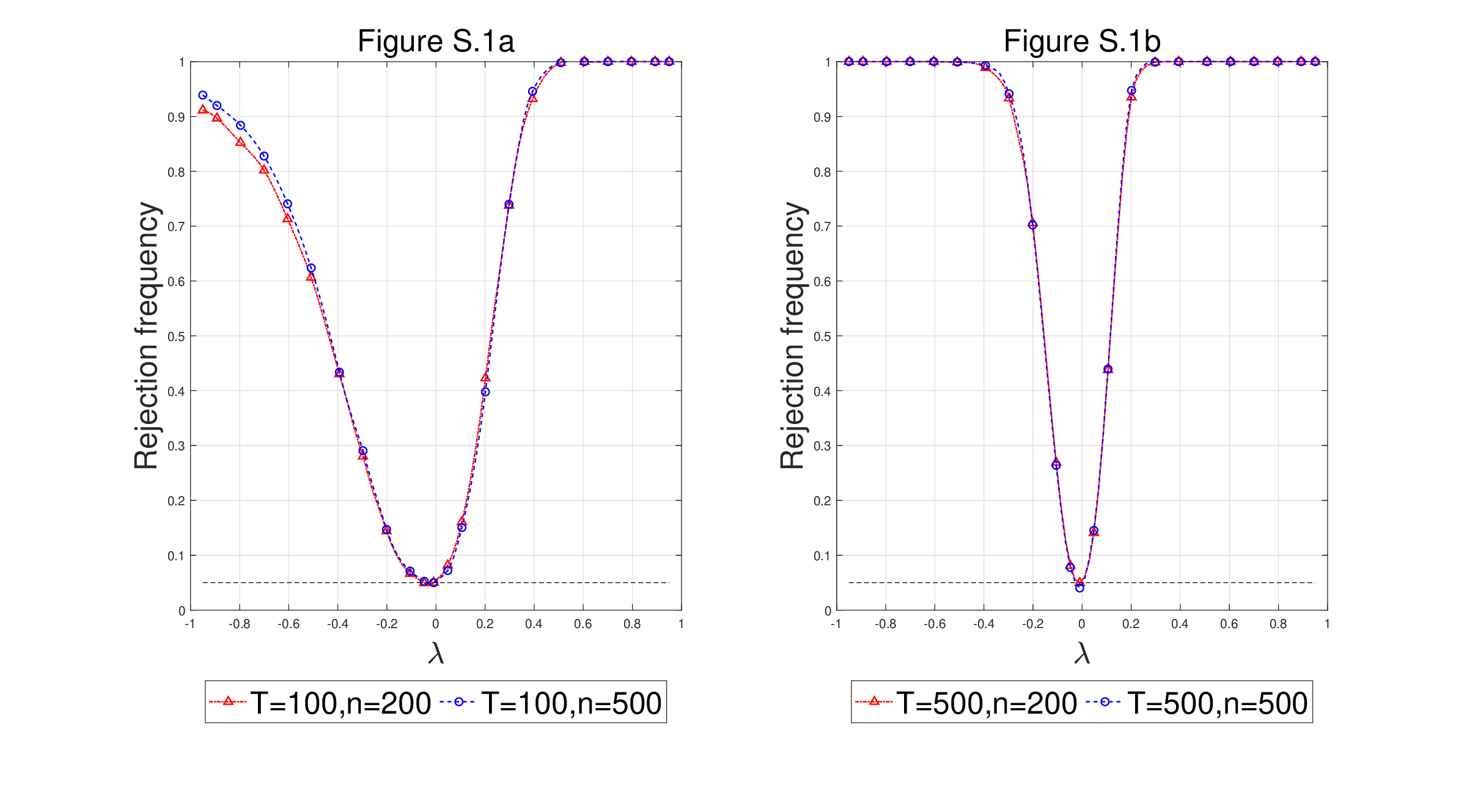}%
\caption{Empirical power functions of the CD$^{\text{*}}$ test against spatial
alternatives for the pure latent factor model with one latent strong factor
and serially independent Gaussian errors, for different sample sizes.}%
\label{fig_psf}%
\end{figure}
\begin{figure}[ptb]%
\centering
\includegraphics[
height=3in,
width=6.8761in
]%
{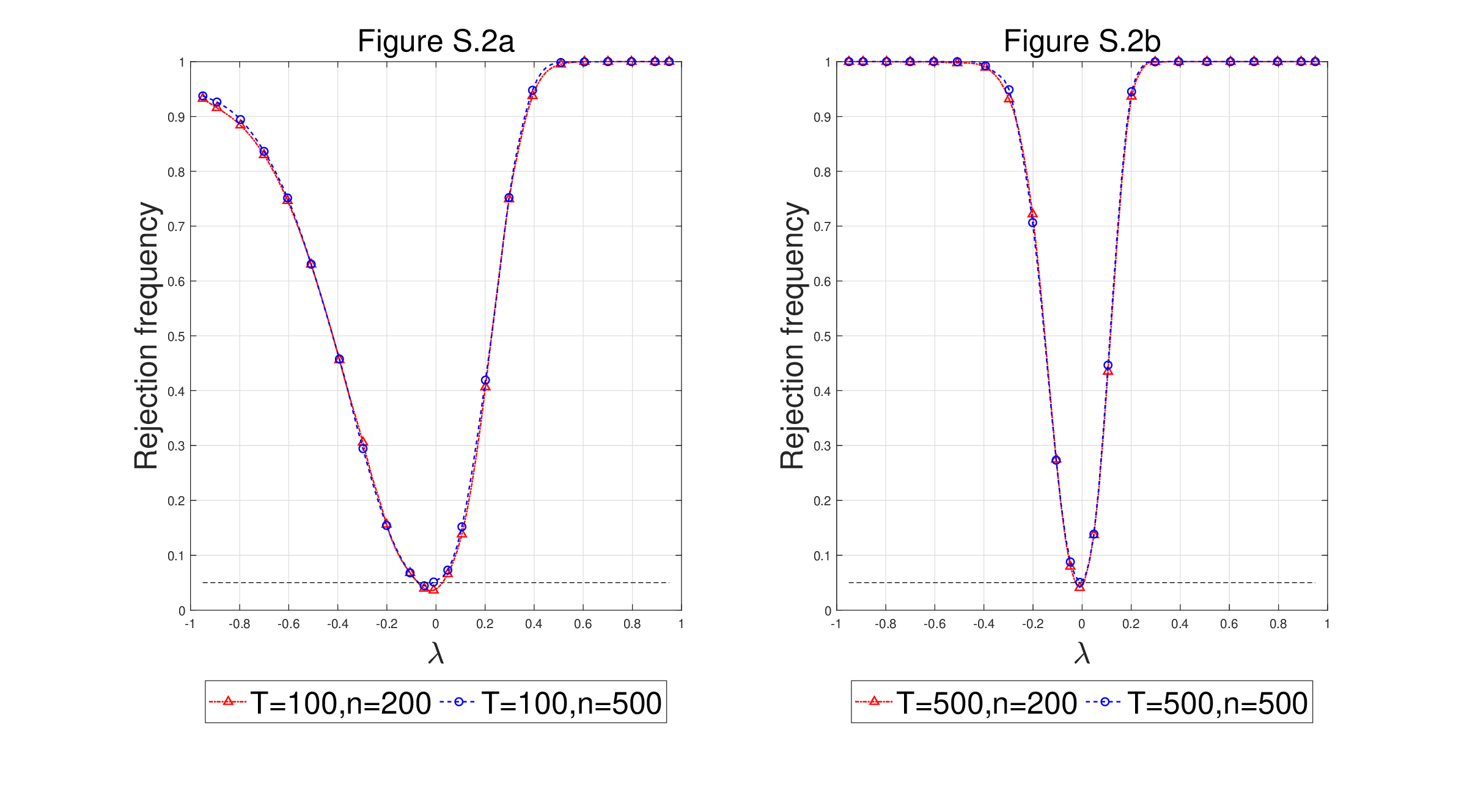}%
\caption{Empirical power functions of the CD$^{\text{*}}$ test against spatial
alternatives for the panel regression model with one latent strong factor and
serially independent Gaussian errors, for different sample sizes.}%
\end{figure}

\bigskip\bigskip\bigskip%
\begin{figure}[ptb]%
\centering
\includegraphics[
height=3in,
width=6.8761in
]%
{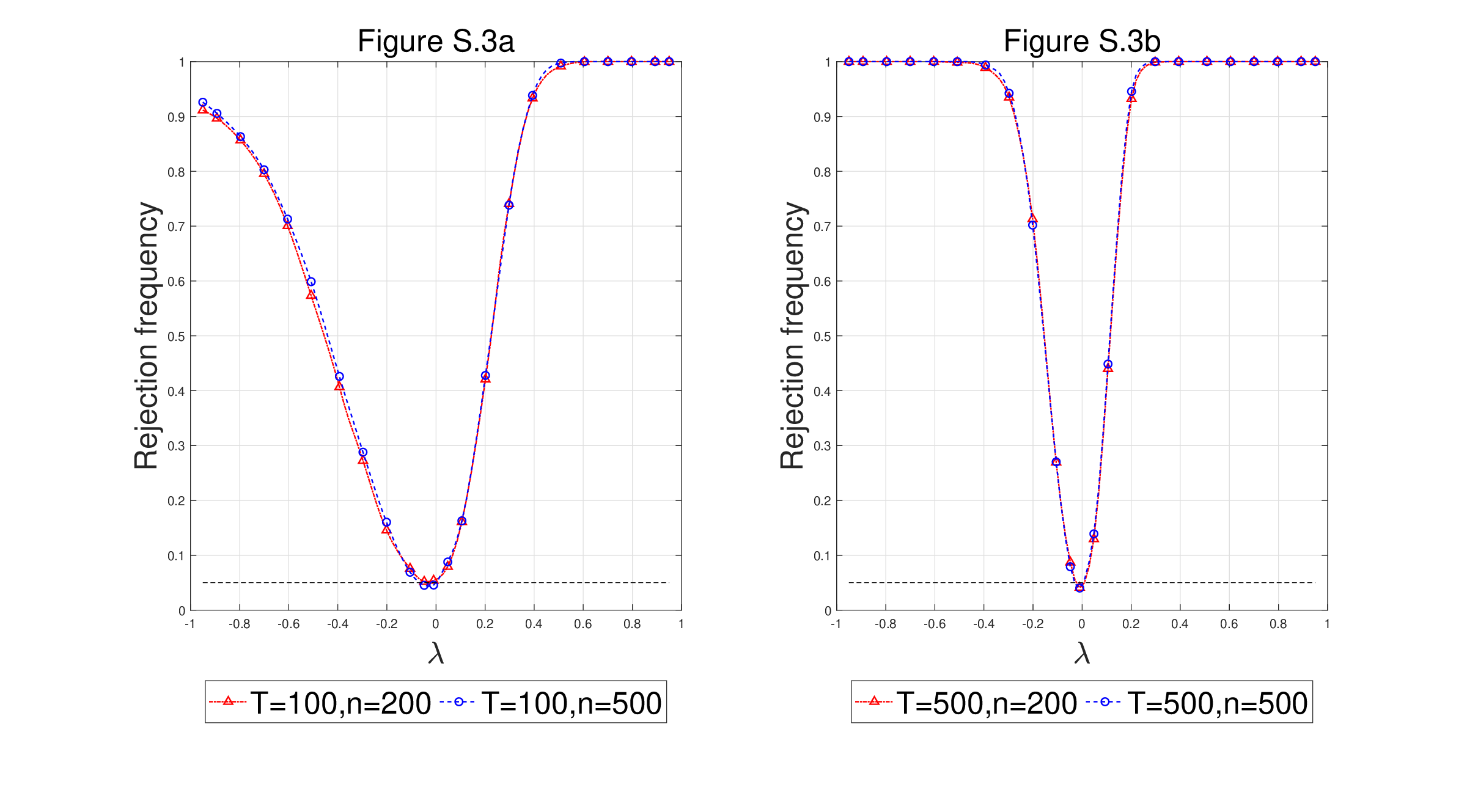}%
\caption{Empirical power functions of the CD$^{\text{*}}$ test against spatial
alternatives for the pure latent factor model with one latent strong factor
and serially independent non-Gaussian errors, for different sample sizes.}%
\end{figure}

\bigskip%

\begin{figure}[ptb]%
\centering
\includegraphics[
height=3in,
width=6.8761in
]%
{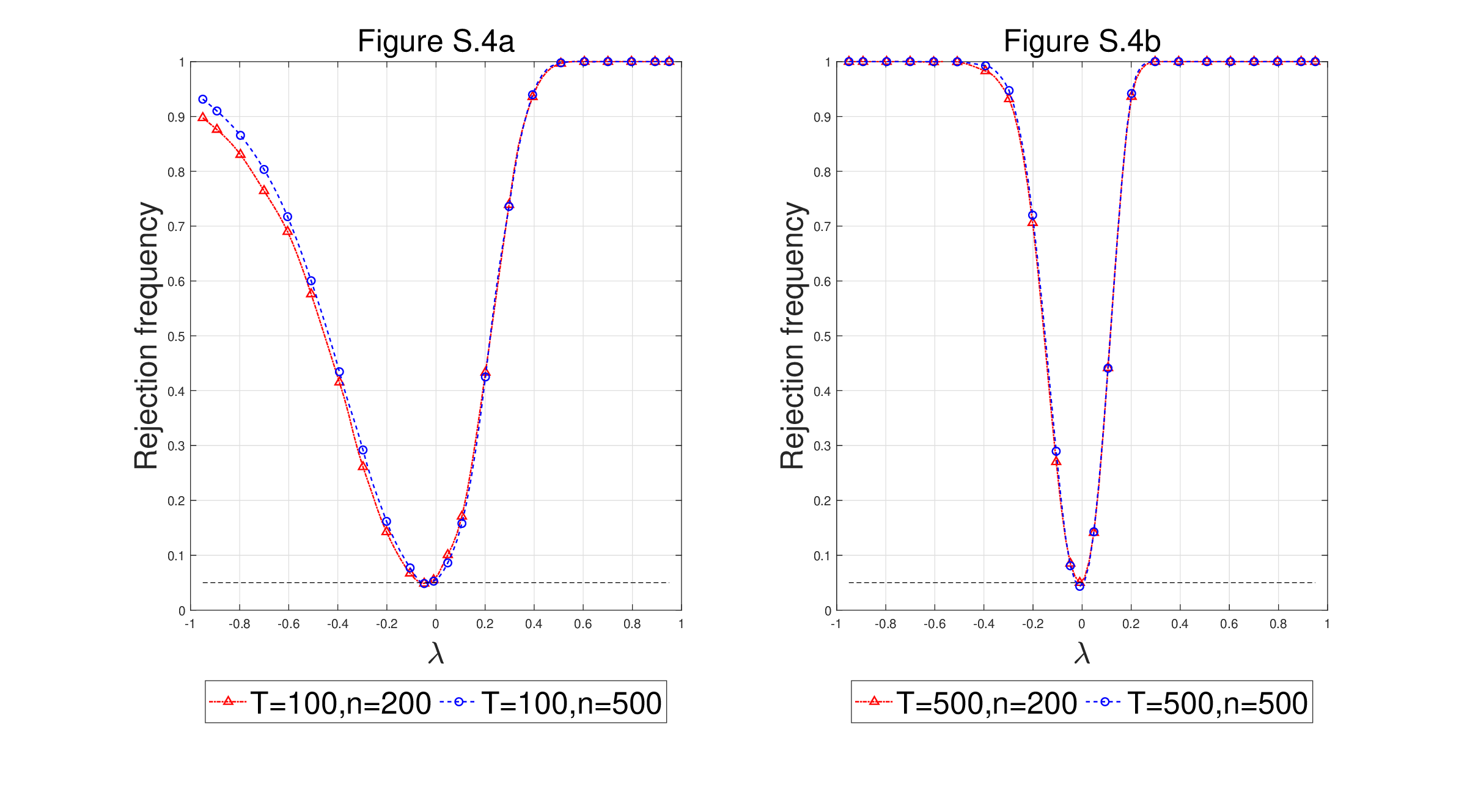}%
\caption{Empirical power functions of the CD$^{\text{*}}$ test against spatial
alternatives for the panel regression model with one latent strong factor and
serially independent non-Gaussian errors, for different sample sizes.}%
\label{fig_xpsf_n}%
\end{figure}

\clearpage\noindent\textbf{References} \newline\vbox{}

\noindent Bai, J. (2003). Inferential theory for factor models of large
dimensions. \textit{Econometrica}, 71(1):135-171.

\bigskip

\noindent Bailey, N., Pesaran, M. H., and Smith, L. V. (2019). A multiple
testing approach to the regularisation of large sample correlation matrices.
\textit{Journal of Econometrics}, 208(2):507--534.

\bigskip

\noindent Chudik, A., Kapetanios, G., \& Pesaran, M. H. (2018). A one
covariate at a time, multiple testing approach to variable selection in
high-dimensional linear regression models. \textit{Econometrica}, 86(4):1479-1512.

\bigskip

\noindent\noindent Juodis, A. and Reese, S. (2022). The incidental parameters
problem in testing for remaining cross-section correlation. \textit{Journal of
Business }$\&$\textit{ Economic Statistics}, 40(3):1191-1203.

\bigskip

\noindent Lieberman, O. (1994). A Laplace approximation to the moments of a
ratio of quadratic forms. \textit{Biometrika}, 81(4):681-690.

\bigskip

\noindent Pesaran, M. H. and Yamagata, T. (2024). Testing for alpha in linear
factor pricing models with a large number of securities. \textit{Journal of
Financial Econometrics}, 22(2):407-460.

\end{document}